\documentclass[reqno,11pt,letterpaper]{amsart}

\usepackage{lipsum}
\usepackage{amsmath}
\usepackage{amssymb}
\usepackage{amsthm}
\usepackage{mathrsfs}
\usepackage{accents}
\usepackage{calc}
\usepackage{arydshln}
\usepackage{upgreek}
\usepackage{slashed}
\usepackage{xifthen}
\usepackage{graphicx}
\usepackage{longtable}
\usepackage[inline]{enumitem}

\usepackage{xcolor}
\definecolor{winered}{rgb}{0.6,0,0}
\definecolor{lessblue}{rgb}{0,0,0.7}

\usepackage[pdftex,colorlinks=true,linkcolor=winered,citecolor=lessblue,urlcolor=lessblue,breaklinks=true,bookmarksopen=true]{hyperref}

\makeatletter
\newcommand{\myitem}[3]{\item[#2]\def\@currentlabel{#3}\label{#1}}
\makeatother

\usepackage{titletoc}

\makeatletter

\def\@tocline#1#2#3#4#5#6#7{
\begingroup
  \par
    \parindent\z@ \leftskip#3 \relax \advance\leftskip\@tempdima\relax
                  \rightskip\@pnumwidth plus 4em \parfillskip-\@pnumwidth
    \ifcase #1 
       \vskip 0.6em \hskip 0em 
       \or
       \or \hskip 0em 
       \or \hskip 1em 
    \fi%
    #6
    \nobreak\relax{\leavevmode\leaders\hbox{\,.}\hfill}
    \hbox to\@pnumwidth {\@tocpagenum{#7}}
  \par
\endgroup
}

 \def\l@section{\@tocline{0}{0pt}{0pc}{}{}}

\renewcommand{\tocsection}[3]{%
  \indentlabel{\@ifnotempty{#2}{ 
    \ignorespaces\bfseries{#2. #3}}}
  \indentlabel{\@ifempty{#2}{\ignorespaces\bfseries{#3}}{}} 
    \vspace{1.5pt}}

\renewcommand{\tocsubsection}[3]{%
  \indentlabel{\@ifnotempty{#2}{
    \ignorespaces#2. #3}}
  \indentlabel{\@ifempty{#2}{\ignorespaces #3}{}}
    \vspace{1.5pt}}

\renewcommand{\tocsubsubsection}[3]{%
  \indentlabel{\@ifnotempty{#2}{
    \ignorespaces#2. #3}}
  \indentlabel{\@ifempty{#2}{\ignorespaces #3}{}}
    \vspace{1.5pt}}

\makeatother

\makeatletter
\def\@nomenstarted{0}
\newlength{\@nomenoldtabcolsep}

\newcommand{\nomenstart}
  {%
    \def\@nomenstarted{1}%
    \setlength{\@nomenoldtabcolsep}{\tabcolsep}%
    \setlength{\tabcolsep}{3.5pt}%
    \begin{longtable}{p{0.11\textwidth} p{0.86\textwidth}}
  }

\newcommand{\nomenitem}[2]{%
    \ifcase\@nomenstarted%
      \or 
      \or \\ 
    \fi%
    #1\,{\leavevmode\leaders\hbox{\,.}\hfill} & #2%
    \def\@nomenstarted{2}%
  }%
\newcommand{\nomenend}
  {\\%
      \end{longtable}%
      \setlength{\tabcolsep}{\@nomenoldtabcolsep}%
      \def\@nomenstarted{0}%
  }
\makeatother

\makeatletter
\newcommand{\BIG}{\bBigg@{3.5}}
\newcommand{\vast}{\bBigg@{4}}
\newcommand{\Vast}{\bBigg@{5}}
\newcommand{\VAST}[1]{\bBigg@{#1}}
\makeatother

\allowdisplaybreaks

\numberwithin{equation}{section}
\numberwithin{figure}{section}
\newtheorem{thm}{Theorem}[section]

\newtheorem{prop}[thm]{Proposition}
\newtheorem{lemma}[thm]{Lemma}
\newtheorem{cor}[thm]{Corollary}
\newtheorem{conj}[thm]{Conjecture}

\newtheorem*{thm*}{Theorem}
\newtheorem*{prop*}{Proposition}
\newtheorem*{cor*}{Corollary}
\newtheorem*{conj*}{Conjecture}

\theoremstyle{definition}
\newtheorem{definition}[thm]{Definition}
\newtheorem{notation}[thm]{Notation}

\theoremstyle{remark}
\newtheorem{rmk}[thm]{Remark}
\newtheorem{example}[thm]{Example}

\makeatletter
\newcommand{\fakephantomsection}{%
  \Hy@MakeCurrentHref{\@currenvir.\the\Hy@linkcounter}
  \Hy@raisedlink{\hyper@anchorstart{\@currentHref}\hyper@anchorend}%
  \Hy@GlobalStepCount\Hy@linkcounter%
}
\makeatother

\newcommand{\mc}{\mathcal}
\newcommand{\cA}{\mc A}

\newcommand{\cC}{\mc C}
\newcommand{\cD}{\mc D}
\newcommand{\cE}{\mc E}
\newcommand{\cF}{\mc F}
\newcommand{\cG}{\mc G}

\newcommand{\cK}{\mc K}
\newcommand{\cL}{\mc L}
\newcommand{\cM}{\mc M}
\newcommand{\cN}{\mc N}
\newcommand{\cO}{\mc O}

\newcommand{\cS}{\mc S}

\newcommand{\cU}{\mc U}
\newcommand{\cV}{\mc V}

\newcommand{\ms}{\mathscr}

\newcommand{\sD}{\ms D}

\newcommand{\sE}{\ms E}

\newcommand{\sR}{\ms R}

\newcommand{\C}{\mathbb{C}}
\newcommand{\N}{\mathbb{N}}
\newcommand{\R}{\mathbb{R}}
\newcommand{\Z}{\mathbb{Z}}

\newcommand{\Sph}{\mathbb{S}}

\newcommand{\sfe}{\mathsf{e}}

\newcommand{\sfs}{\mathsf{s}}

\newcommand{\sfG}{\mathsf{G}}

\newcommand{\bfw}{\mathbf{w}}

\newcommand{\bfB}{\mathbf{B}}

\newcommand{\fa}{\mathfrak{a}}

\newcommand{\fg}{\mathfrak{g}}

\newcommand{\fm}{\mathfrak{m}}
\newcommand{\fp}{\mathfrak{p}}
\newcommand{\fq}{\mathfrak{q}}

\newcommand{\ft}{\mathfrak{t}}

\newcommand{\sld}{\slashed{\dd}{}}
\newcommand{\slg}{\slashed{g}{}}

\newcommand{\slGamma}{\slashed{\Gamma}{}}
\newcommand{\sldelta}{\slashed{\delta}{}}
\newcommand{\slDelta}{\slashed{\Delta}{}}
\newcommand{\slnabla}{\slashed{\nabla}{}}

\newcommand{\slstar}{\slashed{\star}}
\newcommand{\slomega}{\slashed{\omega}{}}
\newcommand{\sltr}{\operatorname{\slashed\tr}}

\newcommand{\Err}{{\mathrm{Err}}{}}

\newcommand{\scal}{\mathsf{S}}
\newcommand{\scalspace}{\mathbf{S}}
\newcommand{\vect}{\mathsf{V}}

\newcommand{\vectspace}{\mathbf{V}}

\newcommand{\vol}{\operatorname{vol}}
\newcommand{\ran}{\operatorname{ran}}
\newcommand{\ann}{\operatorname{ann}}

\newcommand{\codim}{\operatorname{codim}}

\newcommand{\Hom}{\operatorname{Hom}}
\renewcommand{\Re}{\operatorname{Re}}

\newcommand{\Id}{\operatorname{Id}}

\newcommand{\mathspan}{\operatorname{span}}
\newcommand{\supp}{\operatorname{supp}}

\newcommand{\tr}{\operatorname{tr}}

\newcommand{\diag}{\operatorname{diag}}
\newcommand{\Res}{\operatorname{Res}}

\newcommand{\Ups}{\Upsilon}

\newcommand{\eps}{\epsilon}
\newcommand{\ftrans}{\;\!\wh{\ }\;\!}

\newcommand{\hra}{\hookrightarrow}
\newcommand{\la}{\langle}

\newcommand{\ol}{\overline}
\newcommand{\pa}{\partial}
\newcommand{\dd}{{\mathrm d}}
\newcommand{\ra}{\rangle}

\newcommand{\ul}[1]{\underline{#1}{}}

\newcommand{\wh}{\widehat}
\newcommand{\wt}{\widetilde}
\newcommand{\xra}{\xrightarrow}

\newcommand{\myrightleftarrows}[1]{\mathrel{\substack{\xrightarrow{\rule{#1}{0cm}} \\[-.7ex] \xleftarrow{\rule{#1}{0cm}}}}}

\newcommand{\ubar}[1]{\underaccent{\bar}#1}
\newcommand{\pfstep}[1]{$\bullet$\ \underline{\textit{#1}}}
\newcommand{\pfsubstep}[2]{{\bf#1}\ \textit{#2}}

\newcommand{\bop}{{\mathrm{b}}}

\newcommand{\sop}{{\mathrm{s}}}
\newcommand{\seop}{{\mathrm{se}}}

\newcommand{\scop}{{\mathrm{sc}}}

\newcommand{\scl}{{\mathrm{sc}}}

\newcommand{\eop}{{\mathrm{e}}}
\newcommand{\tbop}{{3\mathrm{b}}}
\newcommand{\tscop}{{3\mathrm{sc}}}

\newcommand{\ff}{\mathrm{ff}}

\newcommand{\sface}{{\mathrm{sf}}}

\newcommand{\rms}{{\mathrm{s}}}
\newcommand{\rmv}{{\mathrm{v}}}

\newcommand{\cp}{{\mathrm{c}}}

\newcommand{\Diff}{\mathrm{Diff}}

\newcommand{\Vb}{\cV_\bop}
\newcommand{\Ve}{\cV_\eop}

\newcommand{\Vs}{\cV_\sop}
\newcommand{\Vse}{\cV_\seop}

\newcommand{\Diffb}{\Diff_\bop}
\newcommand{\Diffe}{\Diff_\eop}

\newcommand{\Diffse}{\Diff_\seop}

\newcommand{\Vtb}{\cV_\tbop}
\newcommand{\Vtsc}{\cV_{3\scl}}
\newcommand{\Difftsc}{\Diff_\tscop}
\newcommand{\Difftb}{\Diff_\tbop}

\newcommand{\Vsc}{\cV_\scop}
\newcommand{\Diffsc}{\Diff_\scop}

\newcommand{\Omegasc}{{}^{\scop}\Omega}

\newcommand{\Tse}{{}^\seop T}
\newcommand{\Tsc}{{}^{\scop}T}

\newcommand{\Ttsc}{{}^{\tscop}T}

\newcommand{\half}{{\tfrac{1}{2}}}

\newcommand{\loc}{{\mathrm{loc}}}
\newcommand{\CI}{\cC^\infty}
\newcommand{\CIdot}{\dot\cC^\infty}

\newcommand{\CIc}{\cC^\infty_\cp}

\newcommand{\Hb}{H_{\bop}}

\newcommand{\Hbext}{\bar H_{\bop}}

\newcommand{\phg}{{\mathrm{phg}}}

\newcommand{\Riem}{\mathrm{Riem}}
\newcommand{\Ric}{\mathrm{Ric}}
\newcommand{\Ein}{\mathrm{Ein}}
\newcommand{\bhm}{\fm}
\newcommand{\bha}{\fa}

\newcommand{\openbigpmatrix}[1]
  {%
    \def\@bigpmatrixsize{#1}%
    \addtolength{\arraycolsep}{-#1}%
    \begin{pmatrix}%
  }
\newcommand{\closebigpmatrix}
  {%
    \end{pmatrix}%
    \addtolength{\arraycolsep}{\@bigpmatrixsize}%
  }

\newenvironment{reduce}
 {\hbox\bgroup\scriptsize$\displaystyle}
 {$\egroup}
\newenvironment{reduce2}
 {\hbox\bgroup\tiny$\displaystyle}
 {$\egroup}

\newlength{\enummargin}

\newcommand{\usref}[1]{{\upshape\ref{#1}}}

\DeclareGraphicsExtensions{.mps}

\makeatletter
\newcommand*{\fwbw}[1]{\expandafter\@fwbw\csname c@#1\endcsname}
\newcommand*{\@fwbw}[1]{\ifcase #1 \or {\rm fw}\or {\rm bw}\fi}
\AddEnumerateCounter{\fwbw}{\@fwbw}
\makeatother

\begin{document}

\title{Gluing small black holes along timelike geodesics I: formal solution}

\date{\today. Original version: June 12, 2023}

\begin{abstract}
  Given a smooth globally hyperbolic $(3+1)$-dimensional spacetime satisfying the Einstein vacuum equations (possibly with cosmological constant) and an inextendible timelike geodesic, we construct a family of metrics depending on a small parameter $\eps>0$ with the following properties. (1) They solve the Einstein vacuum equations modulo $\cO(\eps^\infty)$. (2) Away from the geodesic they tend to the original metric as $\eps\to 0$. (3) Their $\eps^{-1}$-rescalings near every point of the geodesic tend to a fixed subextremal Kerr metric. Our result applies on all spacetimes with noncompact Cauchy hypersurfaces, and also on spacetimes without nontrivial Killing vector fields in a neighborhood of a point on the geodesic. If $(M,g)$ is a neighborhood of the domain of outer communications of subextremal or extremal Kerr(--anti de~Sitter) spacetime, our metrics model extreme mass ratio mergers if we choose the timelike geodesic to cross the event horizon.

  The metrics which we construct here depend on $\eps$ and the (rescaled) coordinates on the original spacetime in a log-smooth fashion. This in particular justifies the formal perturbation theoretic setup in work of Gralla--Wald on gravitational self-force in the case of small black holes.
\end{abstract}

\subjclass[2010]{Primary: 83C05, 35B25. Secondary: 83C57, 35C20, 35B40}

\author{Peter Hintz}
\address{Department of Mathematics, ETH Z\"urich, R\"amistrasse 101, 8092 Z\"urich, Switzerland}
\email{peter.hintz@math.ethz.ch}

\maketitle

\setlength{\parskip}{0.00pt}
\tableofcontents
\setlength{\parskip}{0.05in}

\section{Introduction}
\label{SI}

The Einstein vacuum equations (with cosmological constant $\Lambda\in\R$) for a Lorentzian metric $g$ (with signature $(-,+,+,+)$) on a $(3+1)$-dimensional manifold $M$ (assumed to be connected) read
\begin{equation}
\label{EqIEin}
  \Ric(g) - \Lambda g = 0
\end{equation}
where $\Ric$ is the Ricci curvature. Equivalently, $\Ein(g)+\Lambda g=0$ where $\Ein(g)=\Ric(g)-\frac12 R_g g$ is the Einstein tensor (with $R_g$ being the scalar curvature). We assume that $(M,g)$ is globally hyperbolic. The aim of this paper is to construct approximate (in a sense which we make precise below) solutions $g_\eps$ of this equation (which in local coordinates is a quasilinear second order partial differential equation for the coefficients $g_{\mu\nu}$ of the metric $g$) which are obtained from $g$ by gluing a small Kerr black hole \cite{KerrKerr} along a timelike geodesic $\cC\subset M$. The metric $\hat g_{\bhm,\bha}$ of a subextremal Kerr black hole depends on two parameters, $\bhm>0$ (mass) and $\bha\in\R^3$ (specific angular momentum) with $|\bha|<\bhm$. We recall that $\hat g_{\bhm,\bha}$ is
\begin{itemize}
\item defined on $\R_{\hat t}\times\{\hat x\in\R^3\colon|\hat x|>\bhm\}$;
\item stationary, i.e.\ time translations $(\hat t,\hat x)\mapsto(\hat t+c,\hat x)$ are isometries;
\item axisymmetric when $\bha\neq 0$, with the axis of symmetry (rotation axis of the black hole) given by $\frac{\bha}{|\bha|}$, or rotationally symmetric when $\bha=0$ (which gives the Schwarzschild metric \cite{SchwarzschildPaper});
\item asymptotically flat, i.e.\ $\hat g_{\bhm,\bha}=-\dd\hat t^2+\dd\hat x^2+\cO(|\hat x|^{-1})$ tends to the Minkowski metric $-\dd\hat t^2+\dd\hat x^2$ as $|\hat x|\to\infty$;
\item a solution of the Einstein vacuum equations $\Ric(\hat g_{\bhm,\bha})=0$.
\end{itemize}
See Definition~\ref{DefGK}. We may arrange that the $\hat t$-level sets are spacelike (see Lemma~\ref{LemmaGKCoords}).

\begin{thm}[Main result]
\label{ThmI}
  Let $(M,g)$ be a globally hyperbolic spacetime solving~\eqref{EqIEin}. Let $\fp\in M$, let $v\in T_\fp M$ be a future timelike unit vector, and denote by $\cC\subset M$ the maximal geodesic with initial conditions $\fp,v$. Let $\bhm>0$ and $\bha\in T_\fp M$, $\bha\perp v$, $|\bha|<\bhm$. In Fermi normal coordinates\footnote{In such coordinates, $\cC$ is given by $I\times\{0\}$ where $I\subseteq\R$, and $g$ is equal to the Minkowski metric $-\dd t^2+\dd x^2$ up to $\cO(|x|^2)$ errors. Furthermore, the curves $s\mapsto(t,s x)$ for constant $t,x$ are geodesics. Coordinates with these properties are uniquely determined up to constant shifts of $t$ and $x\mapsto A x$ where $A\in O(3)$ is $t$-independent.} $(t,x)\in I\times\R^3$, $I\subseteq\R$, around $\cC$, identify $\bha$ with a vector in $\R^3$. Fix a Cauchy hypersurface $X$ of $(M,g)$ with $X\cap\cC=\{\fp\}$ which is orthogonal\footnote{We require this orthogonality only for notational convenience: it ensures that $t$-level sets are spacelike also near the small Kerr black hole, cf.\ \eqref{EqIKerrConv}, given the form of metric we use here. See also Lemma~\ref{LemmaGKMod}.} to $\cC$. Assume that either
  \begin{enumerate}
  \myitem{ItIOGeneric}{\rm (I)}{I} in the domain of dependence of a connected open neighborhood $\cU^\circ\subset X$ of $\fp$ there do not exist any nontrivial Killing vector fields for $g$; or
  \myitem{ItIOKdS}{\rm (II)}{II} $(M,g)$ is a neighborhood of the domain of outer communications of a Kerr (when $\Lambda=0$), Kerr--de~Sitter (when $\Lambda>0$) or Kerr--anti de~Sitter (when $\Lambda<0$) black hole which is subextremal or extremal, in which case we fix $\cU^\circ\subset X$ to be a connected open set containing $\fp$ as well as a point in the black hole interior of $(M,g)$.
  \end{enumerate}
  For $\eps\in(0,1)$, let $M_\eps=M\setminus\{(t,x)\colon|x|<\eps\bhm\}$. Then there exists a family $(g_\eps)_{\eps\in(0,1)}$, where $g_\eps$ is a smooth symmetric 2-tensor on $M_\eps$, with the following properties.
  \begin{enumerate}
  \item\label{ItIAway}{\rm (Away from $\cC$: close to $g$.)} We have convergence $g_\eps|_{M\setminus\cC}\to g|_{M\setminus\cC}$ in the smooth topology, i.e.\ locally uniformly with all derivatives. That is, if $z\in\R^4$ denotes local coordinates on the closure $\bar U$ of a precompact open set $U\subset M$ with $\bar U\cap\cC=\emptyset$, then the metric coefficients $(g_\eps)_{\mu\nu}(z)=g_\eps(z)(\pa_{z^\mu},\pa_{z^\nu})$ converge to $g_{\mu\nu}(z)$ as $\eps\searrow 0$, together with all derivatives. More precisely, $(\eps,z)\mapsto(g_\eps)_{\mu\nu}(z)$ is log-smooth at $\eps=0$ in $[0,1)_\eps\times\R^4_z$, i.e.\ it has a full generalized Taylor expansion at $\eps=0$ into terms $\eps^m(\log\eps)^k a_{m,k}(z)$ where $m,k\in\N_0$ (with $k=0$ when $m=0$) and $a_{m,k}$ is smooth.
  \item\label{ItINear}{\rm (Near $\cC$: close to a small Kerr black hole.)} In Fermi normal coordinates, we have\footnote{The coefficients of $\hat g_{\bhm,\bha}$ are independent of $\hat t$ by stationarity.}
    \begin{equation}
    \label{EqIKerrConv}
      (g_\eps)_{\mu\nu}(t,\eps\hat x) \to (\hat g_{\bhm,\bha})_{\hat\mu\hat\nu}(\hat x)
    \end{equation}
    locally uniformly with all derivatives on $I_t\times\R^3_{\hat x}$. Here, $\mu,\nu=0,\ldots,3$ are indices for $z=(t,x)$, and $\hat\mu,\hat\nu$ are indices for the corresponding components of $\hat z=(\hat t,\hat x)$. More precisely, $(g_\eps)_{\mu\nu}(t,\eps\hat x)$ is log-smooth at $\eps=0$ inside $[0,1)_\eps\times I_t\times\R^3_{\hat x}$ and equals $(\hat g_{\bhm,\bha})_{\hat\mu\hat\nu}(\hat x)$ up to $\cO(\eps^2)$ errors.
    \item\label{ItITransition}{\rm (Transition region.)} The coefficients $(g_\eps)_{\mu\nu}$, as functions of
    \begin{equation}
    \label{EqICoordTransition}
      t\in I,\qquad
      \rho_\circ:=\frac{\eps}{|x|}\geq 0,\qquad
      \hat\rho:=|x|\geq 0,\qquad
      \omega=\frac{x}{|x|},
    \end{equation}
    are continuous down to $\rho_\circ=0$ and $\hat\rho=0$, with boundary values
    \[
      (g_\eps)_{\mu\nu}(t,\rho_\circ,\hat\rho,\omega) = \begin{cases} g_{\mu\nu}(t,\hat\rho\omega), & \rho_\circ=0, \\ (\hat g_{\bhm,\bha})_{\hat\mu\hat\nu}(\rho_\circ^{-1}\omega), & \hat\rho=0. \end{cases}
    \]
    More precisely, $(g_\eps)_{\mu\nu}$, defined on a neighborhood of $I\times\{0\}\times\{0\}\times\Sph^2$ in $I_t\times[0,1)_{\rho_\circ}\times[0,1)_{\hat\rho}\times\Sph^2_\omega$, is log-smooth at $\rho_\circ=0$ and at $\hat\rho=0$, and equals $(\hat g_{\bhm,\bha})_{\hat\mu\hat\nu}(\rho_\circ^{-1}\omega)$ up to $\cO(\hat\rho^2)$ errors.
    \item\label{ItIFormal}{\rm (Formal solution at $\eps=0$.)} The family $(g_\eps)_{\eps\in(0,1)}$ is a \emph{formal solution} of the Einstein vacuum equations in the following sense. Let $V\subset M$ be a precompact open set, and let $\eps(\ol{V})>0$ be such that $g_\eps$ is a Lorentzian metric on $V\cap M_\eps$ for $\eps\in(0,\eps(\ol{V}))$.\footnote{The existence of such an $\eps(\ol{V})$ is a consequence of parts~\eqref{ItIAway}--\eqref{ItITransition}.} Then
      \begin{equation}
      \label{EqIErr}
        \Err_\eps := \Ric(g_\eps) - \Lambda g_\eps = \cO(\eps^\infty),
      \end{equation}
      i.e.\ the components of $\Err_\eps$ in local coordinates on $M$ near $\ol{V}$ (restricted to $M_\eps$) are bounded by $C_N\eps^N$ for all $N\in\N$.
  \item\label{ItIFormalCauchy}{\rm (Formal solution at a Cauchy hypersurface.)} The error $\Err_\eps$ vanishes to infinite order at $X\cap M_\eps$ for all $\eps$; more precisely, in the notation of part~\eqref{ItIFormal}, it vanishes to infinite order at $X\cap M_\eps\cap V$ for $\eps\in(0,\eps(\ol{V}))$, where $V\subset M$ is an arbitrary precompact open set.
  \item\label{ItISupp}{\rm (Support.)} We have $g_\eps=g$ outside the domain of influence (with respect to $g$) of a compact subset of $\cU^\circ$.
  \end{enumerate}
\end{thm}

See Theorem~\ref{ThmM} for setting~\eqref{ItIOGeneric} and Theorem~\ref{ThmXGlue} for setting~\eqref{ItIOKdS}. We can also consider a third setting (see Remark~\ref{RmkXNc}):

\begin{thm}[Main result: third setting]
\label{ThmI2}
  In the notation of Theorem~\usref{ThmI}, let $X$ be a Cauchy hypersurface of $(M,g)$, and suppose $X$ is \emph{noncompact}. Let $V\subset M$ be any precompact open set. Then there exists a family $(g_\eps)_{\eps\in(0,1)}$ of symmetric 2-tensors on $M_\eps\cap\bar V$ so that the conclusions~\eqref{ItIAway}--\eqref{ItIFormalCauchy} hold on $\bar V$.
\end{thm}

We can interpret~\eqref{EqIKerrConv} as follows: using the scaling property\footnote{This follows from the form~\eqref{EqGKBL} of the Kerr metric, specifically from the fact that pullback under the scaling map $(\hat t,\hat x)\mapsto(\eps\hat t,\eps\hat x)$, $\eps>0$, sends $\hat g_{\eps\bhm,\eps\bha}$ to $\eps^2\hat g_{\bhm,\bha}$.}
\[
  (\hat g_{\bhm,\bha})_{\hat\mu\hat\nu}(x/\eps) = (\hat g_{\eps\bhm,\eps\bha})_{\hat\mu\hat\nu}(x),
\]
the metric $g_\eps(t,x)$ is, for $|x|\lesssim\eps$, close to the metric $\hat g_{\eps\bhm,\eps\bha}(x)$ of a Kerr black hole with mass $\eps\bhm$ and specific angular momentum $\eps\bha$. See Figure~\ref{FigI}, and also Figure~\ref{FigI2} below. We furthermore obtain rough bounds on the exponents of the logarithms appearing in the generalized Taylor expansions. Namely, in part~\eqref{ItIAway} only $\eps^0$, $\eps^1$, $\eps^2\log\eps$, $\eps^2$, and $\eps^m(\log\eps)^k$ with $m\geq 3$ appear, and in part~\eqref{ItINear} only $\eps^0,\eps^2,\eps^m(\log\eps)^k$ with $m\geq 3$ (and similarly in part~\eqref{ItITransition} regarding the expansion at $\rho_\circ=0$, resp.\ $\hat\rho=0$).

\begin{figure}[!ht]
\centering
\includegraphics{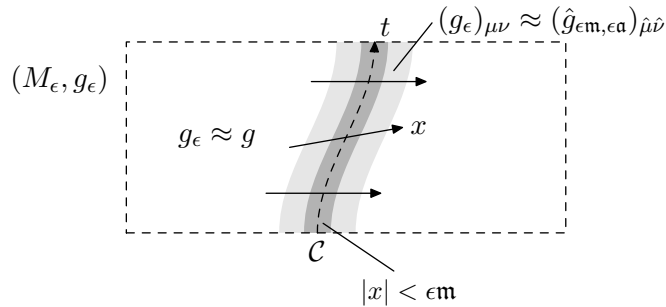}
\caption{Illustration of the metric $g_\eps$ from Theorem~\ref{ThmI} for some small positive $\eps>0$: $g_\eps$ is close to $g$ away from $\cC$, and near all points of $\cC$ close to the metric of a small Kerr black hole with mass $\eps\bhm$ and specific angular momentum $\eps\bha$. We cut out a ball $|x|<\eps\bhm$ in the interior of the small black hole.}
\label{FigI}
\end{figure}

\begin{rmk}[Black hole mergers]
\label{RmkIMerger}
  If in the Kerr--de~Sitter setting~\eqref{ItIOKdS} in Theorem~\ref{ThmI} we take $\cC$ to cross the event horizon of the KdS black hole in finite proper time, the metric $g_\eps$ describes an \emph{extreme mass ratio merger}:\footnote{We carefully distinguish this from \emph{extreme mass ratio inspirals} (EMRIs). In an EMRI, the small black hole moves (in the limit $\eps\searrow 0$) along a \emph{bound orbit} in an ambient Kerr spacetime. The loss of energy due to gravitational waves causes the parameters of the orbit to change and ultimately the small black hole to merge with the ambient one on time scales $\sim\eps^{-1}$. Theorem~\ref{ThmI} on the other hand does not give uniform control beyond time scales $\sim 1$.}   the merger of a mass $\eps$ Kerr black hole with a unit mass Kerr--de~Sitter black hole. Once the small black hole is a fixed distance $2\eta>0$ past the event horizon of the unit mass KdS black hole, one can, for small $\eps>0$, consider the initial data of $g_\eps$ on a suitable spacelike hypersurface in an $\eta$-neighborhood of the domain of outer communications of the unit mass black hole: these are $\eps$-close to the initial data of the unit mass KdS black hole. The future evolution of these data (which satisfy the constraints only modulo $\cO(\eps^\infty)$ errors due to~\eqref{EqIErr}) can be controlled using the robust stability result \cite{HintzVasyKdSStability}; see~\S\ref{SX} for details; see also Figure~\ref{FigXGlue}. If $(M,g)$ is a Kerr black hole, then an application of the existing stability results \cite{KlainermanSzeftelGCM1,KlainermanSzeftelGCM2,DafermosHolzegelRodnianskiTaylorSchwarzschild,KlainermanSzeftelKerr,GiorgiKlainermanSzeftelStability,ShenGCMKerr} would require not merely formal but true solutions; cf.\ Conjecture~\ref{ConjITrue} below.
\end{rmk}

In future work, we hope to correct the formal solution $g_\eps$ to a true solution of~\eqref{EqIEin}:

\begin{conj}[True solution]
\label{ConjITrue}
  Fix a precompact open set $V\subset M$. Then for some small $\eps_0>0$, there exists a smooth tensor $h_\eps$ on $\ol{V}\cap M_\eps$, $0<\eps<\eps_0$ which vanishes to infinite order at $\eps=0$ and outside the domain of influence of a compact subset of $\cU^\circ$, and which has the property that $g_\eps+h_\eps$ is a \emph{true solution} of the Einstein vacuum equations, i.e.\ $\Ric(g_\eps+h_\eps)-\Lambda(g_\eps+h_\eps)=0$ for all $\eps\in(0,\eps_0]$.
\end{conj}

This conjecture implies that the formal solutions constructed by Theorem~\ref{ThmI} describe the interaction of the small Kerr black hole with the ambient spacetime $(M,g)$ to all orders in $\eps$.

Returning to Theorem~\ref{ThmI}, part~\eqref{ItINear} is a strengthening of the following statement. Fix $t_0\in I$, corresponding to a point $(t_0,0)$ in $\cC$, and introduce the `fast' coordinates
\begin{equation}
\label{EqICoordFast}
  \hat t=\frac{t-t_0}{\eps},\qquad
  \hat x=\frac{x}{\eps}
\end{equation}
near it: these are time and space coordinates for an observer on an $\eps^{-1}$ rescaling of $(M,g)$. Then $(g_\eps)_{\mu\nu}(t_0+\eps\hat t,\eps\hat x)\to(\hat g_{\bhm,\bha})_{\hat\mu\hat\nu}(\hat x)$ as $\eps\searrow 0$, locally uniformly and with all derivatives in $(\hat t,\hat x)\in\R\times\R^3$. Since $\pa_{\hat\mu}=\eps\pa_\mu$, this means that
\begin{equation}
\label{EqIConvKerr}
  \eps^{-2}g_\eps|_{(t_0+\eps\hat t,\eps\hat x)}(\pa_{\hat\mu},\pa_{\hat\nu}) \to (\hat g_{\bhm,\bha})_{\hat\mu\hat\nu}|_{\hat x}(\pa_{\hat\mu},\pa_{\hat\nu}).
\end{equation}
We may interpret this as follows. If $g$ is a solution of~\eqref{EqIEin}, then since $\Ric(\eps^{-2}g)=\Ric(g)$, the rescaling $\eps^{-2}g$ solves~\eqref{EqIEin} with cosmological constant $\eps^2\Lambda$. The coefficients of $\eps^{-2}g$ with respect to `fast' coordinates $\hat t=\frac{t-t_0}{\eps}$, $\hat x=\frac{x}{\eps}$ near a point $(t,x)=(t_0,0)$ in $M$, similarly to~\eqref{EqICoordFast}, are equal to the coefficients of $g$ with respect to `slow' coordinates $t,x$. At $(t,x)=(t_0,0)$ itself, $g$ is equal to the Minkowski metric if we choose the coordinates $t,x$ appropriately. In this sense, $\eps^{-2}g$ tends to the flat Minkowski metric in a $\cO(\eps)$-neighborhood of $(t_0,0)$ as $\eps\searrow 0$. Thus, the convergence~\eqref{EqIConvKerr} means that the local limit of $\eps^{-2}g_\eps$ at every point on $\cC$ is not the Minkowski metric, but the (asymptotically flat) Kerr metric. Under the smoothness conditions on $g_\eps$ in part~\eqref{ItITransition} of Theorem~\ref{ThmI}, the limit in~\eqref{EqIConvKerr} is necessarily an asymptotically flat metric which, given~\eqref{EqIErr}, is moreover Ricci-flat.

A further feature of~\eqref{EqIConvKerr} is that the limiting (Kerr) metric is \emph{stationary}. The stationarity of the local limit follows more generally for families $g_\eps$ for which, on the set
\begin{equation}
\label{EqICoordLocal}
  [0,1)_\eps\times I_t\times\R^3_{\hat x},
\end{equation}
the metric $(g_\eps)_{\mu\nu}$ depends only on the \emph{slow} time variable $t$. (The smoothness of $g_\eps$ on~\eqref{EqICoordLocal} means that the small black hole evolves in a \emph{quasistationary}, or \emph{adiabatic}, manner.) Therefore, the regularity properties of $g_\eps$ force the local limits to be stationary, asymptotically flat Ricci-flat spacetimes, and thus by necessity (at least conjecturally) Kerr spacetimes; see Remark~\ref{RmkGKWhyKerr} for further details. This explains why, in the vacuum setting under study here, one can only possibly glue Kerr black holes into $(M,g)$. In this sense, Theorem~\ref{ThmI} is the simplest possible result of its type.\footnote{One can likely perform a similar construction in the setting of the Einstein--Maxwell equations by gluing in Kerr--Newman black holes. Outside the (electro)vacuum regime however, even just the existence of small stationary bodies that one could attempt to glue in is a highly nontrivial issue; see for example \cite{ReinVlasovEinstein,JabiriEinsteinVlasovStatic}.}

\begin{prop}[Necessary conditions for gluing]
\label{PropINec}
  In the notation of Theorem~\usref{ThmI}, suppose that $\cC\subset M$ is an inextendible timelike curve,\footnote{There is still a notion of Fermi normal coordinates; see Lemma~\ref{LemmaGLFermi}.} furthermore $g_\eps$ is a family of metrics on $M_\eps$ which is log-smooth as in points~\eqref{ItINear}--\eqref{ItITransition} but where the subextremal Kerr parameters $\bhm,\bha$ are allowed to depend on $t\in I$ (i.e.\ on the point in $\cC$), with $\bhm=\bhm(t)$ not identically $0$; finally, assume that~\eqref{EqIErr} holds. Then $\cC$ is a geodesic (thus proving the geodesic hypothesis in our setting), $\bhm$ is constant, and $\bha$ is parallel along $\cC$.
\end{prop}

See Proposition~\ref{PropAhNec}, and also \S\ref{SsAcGW}. Thus, in the quasistationary setting, one cannot possibly prove a more general result than Theorem~\ref{ThmI}.

Finally, we remark that $\cO(\hat\rho^2)$ nature of the corrections to $\hat g_{\bhm,\bha}$ in part~\eqref{ItITransition} of Theorem~\ref{ThmI} is optimal in that the nonvanishing of the Riemann curvature tensor of $(M,g)$ at the point $t=t_0$ on the geodesic $\cC$ induces nontrivial $\hat\rho^2$ correction terms at $t=t_0$. Indeed, the quadratic terms of $g$ at $\cC=x^{-1}(0)$ in Fermi normal coordinates are given in terms of components of the Riemann curvature tensor; see Lemma~\ref{LemmaFMc1Metric}.

\subsection{Context and prior work}
\label{SsICx}

The presence of two regimes, as described in parts~\eqref{ItIAway} and \eqref{ItINear} of Theorem~\ref{ThmI}, has for a long time been a prominent feature of studies in the physics literature on the motion of small bodies, whose mass is a small parameter $\eps>0$, in curved spacetimes $(M,g)$ satisfying Einstein's field equations. The starting point was work by Burke \cite{BurkeGravRadiationDamping} who was the first to apply the method of \emph{matched asymptotic expansions} to general relativity.

One key objective in such studies is to determine the motion of the small body: to leading order as $\eps\to 0$, it must move along a geodesic $\cC$ of $(M,g)$ (see \cite{EinsteinInfeldHoffmanGeodesics,ThomasGeodesicHypothesis,TaubGeodesicHypothesis} for early contributions, and \cite{EhlersGerochMotion,GrallaWaldSelfForce} for more recent works), and one is interested in $\cO(\eps)$ and higher order corrections to geodesic motion arising from \emph{gravitational self-force}: the interaction of the small body with the gravitational field which is generated by it and interacts with the ambient spacetime. (The notion of a limit of a family of spacetimes $(M_\eps,g_\eps)$ was studied by Geroch \cite{GerochLimits}.) In the self-force problem, every formula for the correction of geodesic motion is necessarily gauge-dependent. One particular such result is the MiSaTaQuWa equation \cite{MinoSasakiTanakaRadiationReaction,QuinnWaldRadiationReaction} in harmonic (or Lorenz) gauge. We do not obtain new results regarding the problem of gravitational self-force here. Rather, in our construction of $g_\eps$, which does not involve a fixed choice of gauge, we re-center the small black hole in the first few steps of the construction (and could re-center it to all orders in $\eps$ if we so desired). For a detailed introduction and comprehensive literature review on the topic of gravitational self-force, we refer the reader to \cite{PoissonPoundVegaPointParticles,PoundWardellBHPerturbationSelfForce}. The description of extreme mass ratio inspirals (EMRIs) is discussed in the review article \cite{BarackPoundSelfForce}.

The method of matched asymptotic expansions assumes the existence of two expansions of $g_\eps$: a `near-field expansion', which in our notation is a (generalized) Taylor expansion on
\[
  [0,1)_\eps\times I_t\times\R^3_{\hat x},\qquad \hat x=\frac{x}{\eps},
\]
at $\eps=0$, with coefficients that are regular in the rescaled spatial coordinates and adiabatic (i.e.\ they only depend on the `slow' time variable $t$ of the ambient spacetime); and a `far-field expansion', i.e.\ an expansion on $[0,1)_\eps\times(M\setminus\cC)$. The two expansions are matched in an intermediate (or buffer) region (where $|\hat x|\gg\eps^{-1}$ but $|x|\ll 1$, e.g.\ where $|x|\sim\eps^{1/2}$), which produces boundary conditions for the terms in each expansion. This method was subsequently applied by D'Eath \cite{DEathSmallBHDynamics} to study corrections to the parameters of a small Kerr black hole in the same setting that we study here; in particular, he already argued for the validity of Proposition~\ref{PropINec}. D'Eath's work was extended by Kates \cite{KatesSmallBodyMotion} and Thorne--Hartle \cite{ThorneHartleSmallBHMotion} who studied general small bodies and computed leading order corrections for mass and spin.

Gralla and Wald \cite{GrallaWaldSelfForce} introduced a clean perspective on the self-force problem which discards the intermediate region in favor of a joint smoothness requirement of the metric coefficients which corresponds \emph{exactly} to smoothness in the coordinates~\eqref{EqICoordTransition} here. In this paper, we shall relate this requirement to smoothness on a manifold with corners which encodes the parameter $\eps$ and all spacetime manifolds $M_\eps$, $\eps\in(0,1)$, in a single object, denoted $\wt M$ below; see~\S\ref{SsIPf}. (We will also revisit their argument for the necessity that $\cC$ is a geodesic; see~\S\ref{SsAcGW}.)

The following is a list of further novel features of our approach and result which are related to the above physics literature.

\begin{itemize}
\item We demonstrate that families of spacetimes satisfying the \emph{assumptions} made in the aforementioned works indeed \emph{exist}---at least on the level of formal solutions as in~\eqref{EqIErr} (which is sufficient for the arguments in all of those works to go through unchanged, as they only require the validity of the field equations up to $\cO(\eps^N)$ errors, with $N\leq\infty$ depending on the problem under study).
\item We produce the metrics $g_\eps$ in a constructive manner, proceeding in an order-by-order fashion where correction terms are computed as solutions of the linearized field equations with sources, in turn at the original manifold $M$ (but with singular boundary conditions at $\cC$) and on the small Kerr black hole spacetime (with asymptotic boundary conditions at spacelike infinity).
\item Unlike previous works, our method \emph{does not involve any fixed choice of gauge}. Instead, the terms of the (generalized) Taylor expansion of $g_\eps$ at the interface $\hat\rho=\rho_\circ=0$ between the near- and far-field regimes are constructed in an essentially gauge-free manner, and only the solutions of linearized field equations with `trivial' forcing terms (which are essentially supported entirely in either the far-field or near-field regime) involve (rather arbitrary) gauge choices which can be chosen at each step of the construction individually. The main point is that one can solve the sourced linearized field equations $D_g(\Ein+\Lambda)h=f$ (under a genericity assumption on $g$) whenever $f$ is divergence-free; and this condition on $f$ arises naturally from the second Bianchi identity in the iterative construction (where $f$ is equal to $(\Ein+\Lambda)(g+\textrm{[correction terms]})$) of $g_\eps$.
\item The construction in Taylor series in the near-field regime requires a \emph{modulation} of the Kerr black hole parameters to compensate for the failure of solvability for the linearized Einstein vacuum equations around the Kerr solution with (divergence-free) sources when restricting to spaces of stationary tensors.
\end{itemize}
See~\S\ref{SsIPf} for further details.

\medskip

From a mathematical perspective, we have two main goals in the present paper:
\begin{enumerate}
\item\label{ItIMBH} contribute to the theory of many-black-hole spacetimes;
\item\label{ItIGluing} study gluing problems for the Einstein equations.
\end{enumerate}

Regarding~\eqref{ItIMBH}, we recall that the only known explicit solutions of Einstein's field equations describing several black holes are the asymptotically flat Majumdar--Papapetrou \cite{MajumdarSolution,PapapetrouSolution} solutions of the Einstein--Maxwell equations and the related Kastor--Tra\-schen solutions \cite{KastorTraschenManyBH} with positive cosmological constant. In previous work \cite{HintzGluedS}, the author constructed de~Sitter spacetimes in which exact Kerr--de~Sitter black holes are glued into neighborhoods of points at the future conformal boundary. Furthermore, Chru\'sciel--Mazzeo \cite{ChruscielMazzeoManyBH} showed that certain classes of asymptotically flat many-black-hole initial data (constructed in \cite{ChruscielDelayMapping}) evolve into spacetimes with the property that for many asymptotically hyperboloidal slicings of the spacetime (up to some finite retarded time) the apparent horizon has several connected components.

More is known regarding initial data. Recall here that the initial data of a Lorentzian metric $g$ on a $(3+1)$-dimensional spacetime $(M,g)$ at a spacelike hypersurface $X$ are the first and second fundamental form of $X$, respectively; we denote them
\[
  \gamma, k \in \CI(X;S^2 T^*X).
\]
When $(M,g)$ solves the Einstein vacuum equations~\eqref{EqIEin}, the pair $(\gamma,k)$ is a solution of the \emph{constraint equations}
\begin{equation}
\label{EqIConstraints}
  R_\gamma - |k|_\gamma^2 + (\tr_\gamma k)^2 - 2\Lambda = 0,\qquad
  \delta_\gamma k + \dd\tr_\gamma k = 0,
\end{equation}
where $R_\gamma$ is the scalar curvature of $\gamma$, and $(\delta_\gamma k)_\mu=-k_{\mu\nu;}{}^\nu$ is the (negative) divergence operator. Conversely, every solution of the constraint equations gives rise to a unique (up to isometries) maximal globally hyperbolic spacetime attaining $\gamma,k$ as its initial data \cite{ChoquetBruhatLocalEinstein,ChoquetBruhatGerochMGHD,RingstromBook}. Brill and Lindquist \cite{BrillLindquist,LindquistIVP} as well as Misner \cite{MisnerGeometrostatics} constructed explicit (and rigid) time-symmetric solutions of~\eqref{EqIConstraints} (i.e.\ $k=0$) describing multiple black holes in the sense that the initial data contain multiple minimal 2-spheres. Corvino \cite{CorvinoScalar} introduced a flexible gluing method, based on the underdetermined elliptic nature of (the linearization of) the constraint equations. This has found many applications \cite{ChruscielDelayMapping,CarlottoSchoenData,AndersonCorvinoPasqualottoMultiLoc}; see also the review article \cite{CarlottoConstraints}.

The paper \cite{HintzGlueID}, which also uses Corvino's technique \cite{CorvinoScalar}, constructs initial data by gluing in any asymptotically flat data set (satisfying~\eqref{EqIConstraints} with $\Lambda=0$ on the complement of a compact subset of $\R^3$, and with $\gamma$ tending to the Euclidean metric and $k$ to $0$ at infinity) into the neighborhood of a point in a given data set, much like Theorem~\ref{ThmI} but in the elliptic setting of~\eqref{EqIConstraints} instead of in the hyperbolic setting of~\eqref{EqIEin}. Our remarks following Theorem~\ref{ThmI} above imply that the evolution of the glued initial data $(\gamma_\eps,k_\eps)$ constructed in \cite{HintzGlueID} \emph{cannot} be adiabatic \emph{unless} one glues in Kerr initial data; and even if the data $(\hat\gamma,\hat k)$ in \cite[Theorem~1.1]{HintzGlueID} are those of a subextremal Kerr black hole, the requirement in the present paper that the spacetime metric be adiabatic \emph{to all orders in $\eps$} imposes requirements on $(\gamma_\eps,k_\eps)$ \emph{at all orders}.\footnote{This is the reason why the conjecture in \cite[\S1.4]{HintzGlueID}---which is our motivation for Theorem~\ref{ThmI} and Conjecture~\ref{ConjITrue}---required the family $(\gamma_\eps,k_\eps)$ to be a `suitable family'.} Conversely, the initial data of $g_\eps$ at $X\cap M_\eps$ describe a Kerr initial data set glued into the initial data of $(M,g)$ at $X$. In this sense, Theorem~\ref{ThmI} reproves the formal result \cite[Proposition~5.4]{HintzGlueID} in the special case of Kerr data.\footnote{However, our proof of Theorem~\ref{ThmI} relies on some of the results proved in \cite{HintzGlueID}, specifically the solvability theory for the linearized constraints on $X$ with control on supports, see Proposition~\ref{PropAcConstr} and \cite{ChruscielDelayMapping}. Part~\eqref{ItIFormalCauchy} of Theorem~\ref{ThmI} relies on the nonlinear analysis in~\cite{HintzGlueID}. We further remark that the present construction directly produces log-smooth total families; this was not the case with the construction in \cite{HintzGlueID}, although the latter can be modified to give log-smooth total families, as demonstrated in \cite{HintzUnDet}.}

Initial data gluing has recently been developed for the \emph{characteristic} initial value problem of~\eqref{EqIEin}, with initial data given on null hypersurfaces \cite{AretakisCzimekRodnianskiGluing,CzimekRodnianskiGluing,ChruscielCongCharGluing}. Recent works by Kehle--Unger \cite{KehleUngerThirdLaw,KehleUngerHorizonGluing} demonstrate the effectiveness of characteristic gluing techniques for gluing event horizons of different spacetimes (Minkowski and black hole spacetimes).

\medskip

Concerning point~\eqref{ItIGluing} above, the present paper appears to be the first work on spacetime gluing for the Einstein vacuum equations which is not chiefly based on some version of initial data gluing (i.e.\ gluing for the constraint equations).\footnote{The key step in the construction of \cite{HintzGluedS} is the solution of a linear divergence equation (relative to a Riemannian metric) related to the constraint equations on the conformal boundary of a de~Sitter type spacetime \cite{FriedrichDeSitterPastSimple}. The only point where the construction in \cite{HintzGluedS} involves a hyperbolic PDE is the solution of a (quasilinear) gauge-fixed Einstein equation in a final step; this step is straightforward however, since the error term solved away there is supported away from the Kerr--de~Sitter black holes which are glued in, and thus the solution of the PDE has support only in an asymptotically de~Sitter type region, far from the glued-in black holes.} We mention however Yang's work \cite{YangGeodesicHypo} (building on Stuart's earlier \cite{StuartGeodesicsEinstein}) on the construction of (true, not merely formal) solutions in a toy model describing the motion of a very small amplitude and $\eps$-rescaled stable nonlinear Klein--Gordon soliton which is glued along a timelike geodesic in a given spacetime $(M,g)$ as a solution of the Einstein--scalar field system, with the scalar field potential scaled in a way that matches the scaling of the amplitude of the soliton; in \cite{YangGeodesicHypo}, the singularly perturbed spacetime metric is $\cC^1$-close to $g$, even near the geodesic.

The literature on gluing or singular perturbation methods for other hyperbolic evolution equations has largely been concerned with semilinear PDE. Results include the existence of multi-soliton solutions for the nonlinear Schr\"odinger equation \cite{MerleCritNLSBlowup,MartelMerleMultiSolitonNLS} and for generalized Korteweg--de Vries equations \cite{MartelGenKdVMultiSoliton}; the proofs evolve approximate solutions backwards in time and use compactness arguments relying on uniform estimates to extract the desired solutions. This strategy was extended to multi-soliton constructions involving exponentially unstable solitons in \cite{CoteMartelMerlegKdVNLSMultiSoliton,MartelMerleNLW5MultiSoliton,JendrejTwoBubble,JendrejMartelNLWMultiBubble}. Further gluing, multi-soliton, or multi-bubble results include \cite{DaviladelPinoMussoWeiVortex,DaviladelPinoMussoWeiHelices} for the Euler equations, \cite{CoteMunozNLKGMultiSoliton,BellazziniGhimentiLeCozNLKGMultiSoliton,CoteMartelNLKGTravelling} for the nonlinear Klein--Gordon equation, and \cite{MingRoussetTzvetkovWaterWavesMulti} for the water waves system.

\subsection{Elements of the proof}
\label{SsIPf}

As in \cite{HintzGlueID} and related works on gluing problems in the elliptic category such as \cite{SchroersSingerInstantons,KottkeSingerMonopoles}, we adopt a geometric singular analysis perspective and phrase Theorem~\ref{ThmI} as a singular perturbation problem. We construct the family of metrics $\wt g:=(g_\eps)_{\eps\in(0,1)}$ as a single smooth section of the bundle of vertical symmetric 2-tensors on $(0,1)_\eps\times M$ (i.e.\ they annihilate $\pa_\eps$) which is log-smooth on an appropriate (partial) compactification $\wt M$ of $(0,1)_\eps\times M$.

In this introduction, we shall largely work in Fermi normal coordinates $t\in I$, $x\in\R^3$ near $\cC$; this is sufficient to describe all aspects of the analysis except for the very far field behavior of $g_\eps$. We denote by
\begin{equation}
\label{EqIPfhatx}
  \hat x = \frac{x}{\eps}
\end{equation}
the rescaled (`fast') spatial coordinates near the small Kerr black hole. Then $\wt M$ should contain the `far field' $[0,1)_\eps\times(M\setminus\cC)$ (containing the chart $[0,1)_\eps\times I_t\times(\R^3_x\setminus\{0\})$) and the `near field' $[0,1)_\eps\times I_t\times\R^3_{\hat x}$ as smooth submanifolds. We may glue these two charts together over $\eps>0$, $x\neq 0$; but their respective boundary hypersurfaces at $\eps=0$ are disjoint, and there are curves which remain bounded in $M$ and along which $\eps\to 0$ but which do not have a limit.\footnote{A simple example is $(\eps,t,x)=(\eps,t_0,\eps^{1/2})$ where $\eps\in(0,1)$ while $t_0\in I$ is fixed. In terms of~\eqref{EqIPfhatx}, this is $(\eps,t,\hat x)=(\eps,t_0,\eps^{-1/2})$.} To remedy this failure of compactness, we include in $\wt M$ also a coordinate chart
\begin{equation}
\label{EqIPfChart}
  I_t \times [0,1)_{\rho_\circ} \times [0,1)_{\hat\rho} \times \Sph^2,
\end{equation}
glued together with the previous two charts via
\begin{equation}
\label{EqIPfCoordCorner}
  t,\qquad
  \rho_\circ=\frac{\eps}{|x|}=\frac{1}{|\hat x|},\qquad
  \hat\rho=|x|=\eps|\hat x|,\qquad
  \omega=\frac{x}{|x|},
\end{equation}
on the common domains of definition. Invariantly, $\wt M$ is the \emph{blow-up} of $\wt M':=[0,1)_\eps\times M$ at $\{0\}\times\cC$, denoted $\wt M=[\wt M';\{0\}\times\cC]$; see \cite{MelroseDiffOnMwc}. Thus, $\wt M$ has two boundary hypersurfaces:
\begin{enumerate}
\item the \emph{front face} $\hat M=I_t\times\ol{\R^3_{\hat x}}$, where $\ol{\R^3_{\hat x}}=\R^3\sqcup\Sph^2$ is the radial compactification of $\R^3_{\hat x}$ in which $\Sph^2=\{|\hat x|^{-1}=0\}$ is attached as the sphere at infinity;
\item the lift of the original spacetime $M_\circ=I_t\times[0,\infty)_{\hat\rho}\times\Sph^2_\omega$, which is obtained from $M$ by replacing the curve $\cC$ with its spherical normal bundle (which can be thought of as an infinitesimal tube $I_t\times\{0\}\times\Sph^2_\omega$ around $\cC$). There is a smooth \emph{blow-down map}
  \[
    \upbeta_\circ\colon M_\circ\to M,\qquad \upbeta_\circ(t,\hat\rho,\omega) \mapsto (t,\hat\rho\omega) \in M,
  \]
  which over $\hat\rho>0$ is a diffeomorphism $\{\hat\rho>0\}\to M\setminus\cC$; the preimage of $\cC$ is a bundle of 2-spheres over $\cC$.
\end{enumerate}
The front face $\hat M$ is the total space of a fibration $\hat M\to\cC\cong I_t$. Its fibers $\hat M_t$, $t\in I$, are copies of $\ol{\R^3}$; these should be thought of as compactifications of the quotients of a local stationary spacetime manifold (namely, $T_p M\cong\R^4$ at $p\in\cC$) by the time translation action (which is the translation action by $T_p\cC\cong\R\times\{0\}$). Thus, $\hat M$ accurately captures adiabatic behavior. See Figure~\ref{FigI2}.

\begin{figure}[!ht]
\centering
\includegraphics{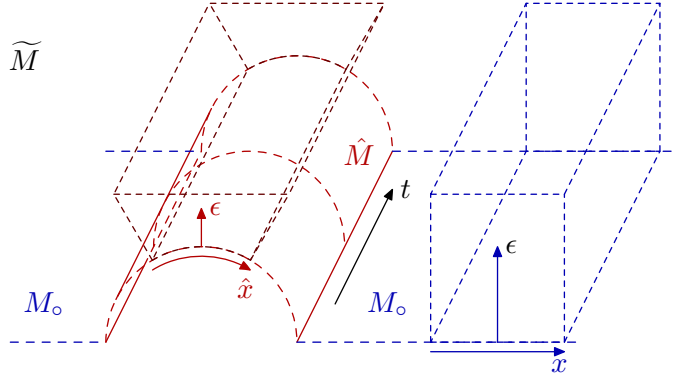}
\caption{Illustration of the total space $\wt M$. Shown are: the portion of the near field $[0,1)_\eps\times I_t\times\R^3_{\hat x}$ where $|\hat x|\leq 1$, the boundary hypersurface $\hat M$, and three fibers of $\hat M$ (in red); a portion of the far field $[0,1)_\eps\times I_t\times(\R^3_x\setminus\{0\})$ where $|x|\geq 1$, and the boundary hypersurface $M_\circ$ (in blue); and some local coordinates.}
\label{FigI2}
\end{figure}

One may hope to insert solutions of a PDE on $\R^3_{\hat x}$ (e.g.\ the Kerr metric restricted to $\hat t=0$) into the fibers of $\hat M$ which at $|\hat x|=\infty$ in the fiber over $p\in\cC$ match with a given solution on $M$ at $p$ (e.g.\ the metric $g$, which at $p$ is the Minkowski metric in Fermi normal coordinates), and to subsequently correct the resulting `zeroth order gluing' by higher order (i.e.\ vanishing at $M_\circ$ and $\hat M$) corrections. Carefully note now that $\wt M$ does \emph{not} contain an intermediate (or buffer) region; thus, to find such corrections, one must solve away error terms at $M_\circ$, resp.\ $\hat M$ with full asymptotic control at the other boundary hypersurface, i.e.\ at the boundary $\pa M_\circ=M_\circ\cap\hat M=\pa\hat M$, with the error terms, and thus also the solutions, typically featuring singular behavior (e.g.\ terms involving $\hat\rho^k(\log\hat\rho)^m$) at the boundary.

We first illustrate how to carry out such a procedure in a linear toy model in~\S\ref{SssIPfM} before turning to the setting of Theorem~\ref{ThmI} in~\S\S\ref{SssIPfB}--\ref{SssIPfN}.

\subsubsection{A model hyperbolic singular perturbation problem}
\label{SssIPfM}

We work on the Minkowski spacetime $(M,g)=(\R\times\R^3,-\dd t^2+\dd x^2)$. Denote by $u=u(t,x)$ a smooth solution of the linear scalar wave equation
\[
  \Box_g u = (-D_t^2+\Delta_x)u = \biggl(\pa_t^2-\sum_{j=1}^3 \pa_j^2\biggr)u=0,\qquad D = \frac{1}{i}\partial.
\]
We define a singular perturbation of $\Box_g$ by\footnote{The assumptions on $V$ here are made for the sake of maximal simplicity. One can treat potentials with smooth dependence on $t$ with only notational modifications, and similarly also potentials with inverse cubic (or faster) decay. Many potentials with inverse quadratic decay can be handled as well with more substantial modifications, including altered exponents in polyhomogeneous expansions.}
\[
  P_\eps := \Box_g + \eps^{-2}V\Bigl(\frac{x}{\eps}\Bigr),\qquad 0\leq V\in\CIc(\R^3_{\hat x}),\quad \eps\in(0,1).
\]
We wish to perturb $u$ to a formal solution $u_\eps$ of the equation $P_\eps u_\eps=0$; that is, we want to find a log-smooth function $\wt u$ on $\wt M=[[0,1)\times M;\{0\}\times\cC]$, where $\cC=\R_t\times\{0\}\subset M$, so that
\begin{subequations}
\begin{equation}
\label{EqIPfMwtu}
  \wt u|_{M_\circ} = \upbeta_\circ^*u\qquad (\text{i.e.}\ \wt u|_{(\eps,t,x)=(0,t,x)}=u(t,x),\ x\neq 0),\qquad
  |P\wt u| \leq C_N\eps^N\ \forall\,N,
\end{equation}
where $\wt u$, resp.\ $P\wt u$ is defined on an $\eps$-level set by $u_\eps$, resp.\ $P_\eps u_\eps$.

The first task is to determine the local limit of $\wt u$ at $\hat M$ to be glued in. This task is trivial in the black hole gluing problem, where we wish to glue in the known Kerr metric; by contrast, in the current linear model problem, we artificially perturbed the partial differential operator from $\Box_g$ to $P_\eps$, and thus must find the stationary solution of the local model PDE on the fiber $\hat M_t$ of $\hat M$ over $(t,0)\in\cC$ which matches $u(t,0)$ (i.e.\ a constant for each fixed value of $t$) at $|\hat x|=\infty$. In the coordinates $t$ and $\hat x=\frac{x}{\eps}$, we have
\[
  \eps^2 P_\eps = (\eps\pa_t)^2 + \bigl(\Delta_{\hat x}+V(\hat x)\bigr),
\]
which acts on smooth functions of $\eps,t,\hat x$ (i.e.\ they only depend on the \emph{slow} time $t$) as $\hat P(0):=\Delta_{\hat x}+V(\hat x)$ at $\eps=0$. We then use:\footnote{A clean proof uses b-analytic techniques: the operator $\Delta+V\colon\Hb^{2,\alpha}\to\Hb^{0,\alpha+2}$, for $\alpha\in(-\frac32,-\frac12)$, is a compact perturbation of the invertible operator $\Delta\colon\Hb^{2,\alpha}\to\Hb^{0,\alpha+2}$, where $\Hb^{k,\beta}$ consists of all functions on $\R^3$ which lie in $\la x\ra^{-\beta}L^2(\R^3)$ upon application of up to $k$ of the vector fields $\pa_j$, $x^j\pa_\ell$. Since $V\geq 0$, the nullspace of $\Delta+V$ is trivial. Thus, $\hat u_{(0)}=1-(\Delta+V)^{-1}V$, and the stated regularity then follows in view of $\Delta\hat u_{(0)}=-V\hat u_{(0)}$ from elementary properties of $\Delta$. Alternatively, one can write $\hat u_{(0)}=1+v$ and find $v$ solving $(I+\Delta^{-1}V)v=-\Delta^{-1}V$ where $\Delta^{-1}=\frac{1}{4\pi|\cdot|}*$; the operator $I+\Delta^{-1}V$ is compact on $\la x\ra^\beta L^2$ for $\beta>\frac12$, and thus the Fredholm alternative applies and implies its invertibility when in addition $\beta<\frac32$.}

\begin{lemma}[Stationary solution]
\label{LemmaIPfM}
  There exists a unique solution $\hat u_{(0)}\in\CI(\R^3_{\hat x})$ of
  \[
    \hat P(0)\hat u_{(0)}=(\Delta+V)\hat u_{(0)}=0,\qquad
    \hat u_{(0)}\to 1\ \ \text{as}\ \ |\hat x|\to\infty.
  \]
  Moreover, $\hat u_{(0)}\in\CI(\ol{\R^3_{\hat x}})$, i.e.\ in $|\hat x|>0$, $\hat u_{(0)}$ is a smooth function of $|\hat x|^{-1}$ and $\frac{\hat x}{|\hat x|}$.
\end{lemma}

We now supplement the requirements~\eqref{EqIPfMwtu} by
\begin{equation}
\label{EqIPfMwtu2}
  \wt u|_{\hat M_t} = u(t,0)\hat u_{(0)}\qquad (\text{i.e.}\ \wt u|_{(\eps,t,\hat x)=(0,t,\hat x)} = u(t,0)\hat u_{(0)}(\hat x),\ t\in\R,\ \hat x\in\R^3).
\end{equation}
\end{subequations}
The requirements on $\wt u$ at $M_\circ$ and $\hat M$ in~\eqref{EqIPfMwtu} and \eqref{EqIPfMwtu2} are consistent at the corner $M_\circ\cap\hat M$.

\begin{rmk}[$P$ in terms of fast variables]
\label{RmkIPfPFast}
  If we fix $t_0\in\R$ and set $\hat t=\frac{t-t_0}{\eps}$, then
  \[
    \eps^2 P_\eps = -D_{\hat t}^2 + \Delta_{\hat x} + V(\hat x)
  \]
  is a wave operator; $\hat P(0)$ is its spectral family at frequency $0$.
\end{rmk}

\begin{rmk}[Other toy models]
\label{RmkIPfOther}
  A related singular perturbation problem which one can study using the procedure described below is $\Box_{g_\eps}u_\eps=0$ where $g_\eps=(g_\eps)_{\mu\nu}\,\dd z^\mu\,\dd z^\nu$, $z=(t,x)$, is a singular perturbation of the Minkowski metric $g$ in that $(g_\eps)_{\mu\nu}\in\CI(\wt M)$ restricts to $g_{\mu\nu}$ at $M_\circ$ and to $\hat g_{\mu\nu}$ at $\hat M$ where $\hat g$ is a stationary and asymptotically (as $|\hat x|\to\infty$) Minkowski metric for which $\pa_{\hat t}$ and $\dd\hat t$ are timelike (one may also allow $\hat g$ to depend smoothly on the parameter $t$). One demands that $\wt u=(u_\eps)_{\eps\in(0,1)}$ restricts to $u$ at $M_\circ$ and to the constant $u(t,0)$ at $\hat M_t$. One can also study singular perturbations in nonlinear settings, such as $\Box_{g_\eps}u_\eps=u_\eps^2$ when one is given $u$ satisfying $\Box_g u=u^2$. We leave the details to the interested reader.
\end{rmk}

\textbf{Step 1. Naive gluing.} Let $\wt u_0\in\CI(\wt M)$ be any function which is equal to $\upbeta_\circ^*u$ at $M_\circ$ and equal to $u(t,0)\hat u_{(0)}(\hat x)$ on $\hat M$.\footnote{In the coordinates~\eqref{EqIPfCoordCorner}, we can for example take $\wt u_0(t,\rho_\circ,\hat\rho,\omega)=u(t,\hat\rho\omega)+u(t,0)\hat u_{(0)}(\rho_\circ^{-1}\omega)-u(t,0)$; another possibility is to set $\wt u_0(\eps,t,x)=u(t,x)\hat u_{(0)}(x/\eps)$. We stress again that this is a \emph{sharp} gluing of $\hat u_{(0)}=\hat u_{(0)}(\hat x)$ and $u=u(t,x)$ in that there is no transition region, and we do not use any cutoffs like $\chi(|x|/\sqrt\eps)$ for transitioning between the near- and far-field regimes.} Since the restrictions of $\wt u_0$ to both boundary hypersurfaces $\eps^{-1}(0)=\hat M\cup M_\circ$ of $\wt M$ satisfy the desired PDE, the error term
\[
  \Err_0 := P\wt u_0
\]
vanishes to leading order at $\hat M\cup M_\circ$. More precisely, let
\[
  \hat\rho=(\eps^2+|x|^2)^{1/2},\qquad
  \rho_\circ=\frac{\eps}{(\eps^2+|x|^2)^{1/2}}
\]
denote defining functions of $\hat M$ and $M_\circ$: they vanish only at $\hat M$, resp.\ $M_\circ$, with nonvanishing differentials there. Then:

\begin{lemma}[Mapping properties]
\label{LemmaIPfMap}
  Let $k,\ell\in\R$. Then $P\colon\rho_\circ^k\hat\rho^\ell\CI(\wt M)\to\rho_\circ^k\hat\rho^{\ell-2}\CI(\wt M)$. Furthermore, if $\wt u\in\rho_\circ^k\CI(\wt M)$, then $(\eps^2 P\wt u)|_{\hat M}=\hat P(0)(\wt u|_{\hat M})$; and if $\wt u\in\hat\rho^\ell\CI(\wt M)$, then $(P\wt u)|_{M_\circ}=\Box_g(\wt u|_{M_\circ})$.
\end{lemma}
\begin{proof}
  Near the manifold interior $(M_\circ)^\circ$, and indeed where $|x|>\delta>0$ and thus weights in $\hat\rho$ can be dropped while $\rho_\circ$ can be replaced by $\eps$, this follows from the smoothness of the coefficients of the operator $P$---which in such a region equals $\Box_g$ when $\eps$ is sufficiently small. Near $\hat M^\circ$, and indeed where $|\hat x|<R<\infty$ and thus weights in $\rho_\circ$ can be dropped while $\hat\rho$ can be replaced by $\eps$, this follows from $P=\pa_t^2+\eps^{-2}(\Delta_{\hat x}+V(\hat x))$. We leave the calculation near the corner in the coordinates~\eqref{EqIPfCoordCorner} to the reader.
\end{proof}

By the choice of $\wt u_0$, the error $\Err_0\in\hat\rho^{-2}\CI(\wt M)$ has an extra order of decay at both $\hat M$ and $M_\circ$, so
\begin{equation}
\label{EqIPfErr0Worse}
  \Err_0 \in \rho_\circ\hat\rho^{-1}\CI(\wt M).
\end{equation}
If $\pa_x u(t,0)=0$, then we can choose $\wt u_0$ so that $\Err_0$ has an extra order of decay at $\hat M$, so
\begin{equation}
\label{EqIPfErr0Better}
  \Err_0 \in \rho_\circ\CI(\wt M);
\end{equation}
indeed, this holds for $\wt u_{(0)}(\eps,t,x)=u(t,x)\hat u_{(0)}(x/\eps)$ by direct computation. The plan is to add correction terms to $\wt u_0$ to improve the error term; we do this in turns at $M_\circ$ and $\hat M$. In what follows, we shall assume that the error satisfies~\eqref{EqIPfErr0Better}, as an analogue of this will hold in the black hole gluing setting. If, in the present toy model discussion, we instead only have~\eqref{EqIPfErr0Worse}, one merely needs to interchange Steps 2 and 3 below (and adjust the overall powers of $\eps$ there).

\textbf{Step 2. Solving away the error at $M_\circ$.} Given~\eqref{EqIPfErr0Better}, we want to find $h=h(t,x)$ so that\footnote{One may want to cut $h$ off to a neighborhood of $M_\circ$ in $\wt M$ in order to emphasize that $\eps h$ is a correction term \emph{at $M_\circ$}. We do not do this in this sketch for notational brevity.}
\[
  P(\wt u_0+\eps h) = \Err_0 + \eps P(h) \in \rho_\circ^2\CI(\wt M)
\]
vanishes to one order more at $M_\circ$ than $\Err_0$. By Lemma~\ref{LemmaIPfMap}, this is equivalent to
\begin{equation}
\label{EqIPfMcirc}
  \Box_g h = f_0 := -(\eps^{-1}\Err_0)|_{M_\circ} \in \hat\rho^{-1}\CI(M_\circ) = r^{-1}\CI(\R_t\times[0,\infty)_r\times\Sph^2),\qquad r=|x|.
\end{equation}
This can be viewed as a wave equation on $(M,g)$ with a source term that is singular at $\cC$ (albeit in a highly structured, \emph{polyhomogeneous conormal} sense). We claim that there exists a solution $h\in\hat\rho\CI(M_\circ)$ (ignoring the possible necessity of logarithmic terms for simplicity of presentation). Indeed, one can find $h$ in a two-step procedure.
\begin{enumerate}
\item First, one solves $\Box_g h_\sharp=(\pa_t^2-\pa_r^2-\frac{2}{r}\pa_r+r^{-2}\Delta_{\Sph^2})h_\sharp=f_0$ \emph{in Taylor series at $r=0$}; this is accomplished by making the ansatz $h_\sharp(t,r,\omega)\sim\sum_{j\geq 1}r^j h_{\sharp,j}(t,\omega)$ (ignoring the possibility that $(\log r)^k$ factors may be needed) and solving iteratively for $h_{\sharp,1}$, $h_{\sharp,2}$, $\ldots$.
\item Second, one finds a correction term $h_\flat$ satisfying $\Box_g h_\flat=f_0-\Box_g h_\sharp$; the source term here is smooth on $M_\circ$, but now vanishes to infinite order at $x=0$ and can thus be regarded as a smooth function on $M$. One can find $h_\flat$ by solving an initial value problem on $(M,g)$ with arbitrarily chosen (smooth) initial data at $t=0$, say.
\end{enumerate}

Thus $h=h_\sharp+\upbeta_\circ^*h_\flat$ solves~\eqref{EqIPfMcirc}, and we set $\wt u^1:=\wt u_0+\eps h$. The additional term here satisfies $\eps h\in\rho_\circ\hat\rho^2\CI(\wt M)$.

\begin{rmk}[Solving for $h_\sharp$, I: iterative procedure]
\label{RmkIPfhsharp1}
  The equation for $h_\sharp$ in the first step can be rewritten as $(-(r\pa_r)^2-r\pa_r+\Delta_{\Sph^2})h_\sharp=r^2 f_0-r^2\pa_t^2 h_\sharp$. Due to the smoothness requirement in the `slow' variable $t$, this is an equation which at each step in the construction of the Taylor series of $h_\sharp$ depends only \emph{parametrically} on $t$. Moreover, the operator on the left acts on $r^\lambda v(\omega)$ as $r^\lambda$ times
  \[
    N(r^2\wh{\Box_g}(0),\lambda)v,\qquad
    N(r^2\wh{\Box_g}(0),\lambda):=-\lambda^2-\lambda+\Delta_{\Sph^2}.
  \]
  Here, $N(r^2\wh{\Box_g}(0),\lambda)$ is the \emph{indicial family} of the zero energy operator $\wh{\Box_g}(0)$ of $\Box_g$. See~\S\ref{SssBgE} regarding the structural reasons for this behavior. For example, the equation for $h_{\sharp,1}$ reads $(-2+\Delta_{\Sph^2})h_{\sharp,1}(t,\omega)=(r f_0)|_{(t,r,\omega)=(t,0,\omega)}$, which is solvable if the right hand side is orthogonal to $l=1$ spherical harmonics (otherwise $\log r$ factors are needed in $h_\sharp$); the computation of $h_{\sharp,j}$, $j\geq 2$, involves $t$-derivatives of $f_0$ and $h_{\sharp,k}$ for $k\leq j-1$.
\end{rmk}

\begin{rmk}[Solving for $h_\sharp$, II: non-characteristic nature of $\cC$]
\label{RmkIPfhsharp2}
  The construction of $h_\sharp$ is a special case of solving (pseudo)differential equations $P u=f\bmod\CI$ where $f$ is a (polyhomogeneous) conormal symbol at a submanifold $\cC$, with $P$ non-characteristic at $N^*\cC$. Cf.\ \cite[Theorem~18.2.12]{HormanderAnalysisPDE3} for the principal symbol statement: this prompts one to invert the restriction of the symbol of $\Box_g$, given by $(t,x,\tau,\xi)\mapsto-\tau^2+|\xi|^2$, to $N^*\cC\setminus o=\{(t,x,\tau,\xi)\colon x=0,\ \tau=0,\ \xi\neq 0\}$. The procedure in Remark~\ref{RmkIPfhsharp1} is a `physical space' version of this.
\end{rmk}

\begin{rmk}[Issues with initial value problems: I]
\label{RmkIPfIVP1}
  If one were to solve~\eqref{EqIPfMcirc} immediately via an initial value problem, with initial data at $t=0$, say, the solution $h$ would typically be singular along the future and past light cones emanating from the point $\{t=0\}\cap\{r=0\}$ of intersection of the Cauchy surface and the curve $\cC$. As a concrete example, the solution of $\Box_g h=r^{-1}$, which is of the form~\eqref{EqIPfMcirc}, with initial data $(h,\pa_t h)|_{t=0}=(0,0)$, is
  \[
    h(t,r)=\begin{cases} t-r/2, & r<t, \\ t^2 / (2 r), & r>t. \end{cases}
  \]
  This fails to be $\cC^2$ at $r=t$. By contrast, we have $\Box_g(-\frac{r}{2})=r^{-1}$, and so $h=-\frac{r}{2}$ is the type of `good' solution produced by the above two-step procedure; in other words, obtaining the `good' solution from an initial value problem would require fine-tuning the initial conditions to match those of $h_\sharp$ modulo $\upbeta_\circ^*\CI(M)$.
\end{rmk}

\textbf{Step 3. Solving away the error at $\hat M$.} By construction, the approximate solution $\wt u^1\in\CI(\wt M)$ has error
\[
  \Err^1 = P(\wt u^1) \in \rho_\circ^2\CI(\wt M).
\]
We now solve away the restriction of $\Err^1$ to $\hat M$,
\[
  f^1 := -\Err^1|_{\hat M} \in \rho_\circ^2\CI(\hat M) = \CI\bigl(\R_t; \rho_\circ^2\CI(\ol{\R^3_{\hat x}})\bigr),\qquad \rho_\circ=\la\hat x\ra^{-1}.
\]
That is, we wish to find $h=h(t,\hat x)$ so that\footnote{As before, we do not explicitly cut off $h$ to a neighborhood of $\hat M$, for notational brevity.} $P(\wt u^1+\eps^2 h)=\Err^1+\eps^2 P(h)$ vanishes to one order more at $\hat M$ than $\Err_0$. By Lemma~\ref{LemmaIPfMap}, this requires solving
\begin{equation}
\label{EqIPfhatPEq}
  \hat P(0)h(t,\hat x) = f^1(t,\hat x)
\end{equation}
\emph{parametrically} in $t$. This is \emph{not} a wave equation in the fast variables (cf.\ Remark~\ref{RmkIPfPFast}); rather, the fact that only the zero energy operator $\hat P(0)$ appears here is due to our requirement that the solution $\wt u$ we seek be adiabatic. One can solve~\eqref{EqIPfhatPEq} in a two-step procedure.
\begin{enumerate}
\item First, one solves for $h(t,\hat x)$ in Taylor series at $|\hat x|^{-1}=0$ (in inverse polar coordinates $|\hat x|^{-1}$, $\omega$). This involves the same indicial family as the one mentioned in Remark~\ref{RmkIPfhsharp1}.
\item Second, one applies $\hat P(0)^{-1}=(\Delta+V)^{-1}$ to the remaining rapidly decaying error term. This produces a correction with a full asymptotic expansion as $|\hat x|\to\infty$ (here concretely a smooth function of $|\hat x|^{-1}$ and $\omega$ which vanishes at $|\hat x|^{-1}=0$).
\end{enumerate}

Ignoring the possibility of logarithmic terms, we obtain a solution $h\in\CI(\R_t;\CI(\ol{\R^3_{\hat x}}))$ of~\eqref{EqIPfhatPEq}. This gives
\[
  \wt u_1 := \wt u^1 + \eps^2 h,\qquad \Err_1 := P(\wt u_1) \in \rho_\circ^2\hat\rho\CI(\wt M).
\]

\begin{rmk}[Issues with initial value problems: II]
\label{RmkIPfIVP2}
  Consider equation~\eqref{EqIPfhatPEq} for $t$ near $0\in\R$ in the coordinates $\hat t=\frac{t}{\eps}$. Clearly, we \emph{cannot} solve this from the perspective of an initial value problem, as $h$ is (up to addition of $t$-dependent multiples of $\hat u_{(0)}$) uniquely determined if we restrict its growth as $|\hat x|\to\infty$ to be sublinear (so that $\eps^2 h$ does not affect the earlier correction term at $M_\circ$). If one were to attempt to solve the problem~\eqref{EqIPfMwtu}--\eqref{EqIPfMwtu2} via an initial value problem, with trivial initial data at $t=0$, say, as in Remark~\ref{RmkIPfIVP1}, then an adiabatic solution $h$ would not exist. Instead, the equation to be solved in the region $|\hat t|\lesssim 1$ would have to be the wave equation
  \begin{equation}
  \label{EqIPfIVP2}
    (-D_{\hat t}^2+\Delta_{\hat x}+V(\hat x))h(\hat t,\hat x) = f^1(0,\hat x),\qquad (h,\pa_{\hat t}h)|_{\hat t=0}=(0,0),
  \end{equation}
  the solution of which is \emph{not} stationary (unless the initial data were chosen so as to be consistent with~\eqref{EqIPfhatPEq}). The solution of~\eqref{EqIPfIVP2} typically has a nontrivial radiation field at null infinity in $\R_{\hat t}\times\R^3_{\hat x}$, causing oscillations on the scale $\hat t-|\hat x|\sim 1$, i.e.\ $t-|x|\sim\eps$, and hence a singularity to emerge out of $\cC$ along the light cone emanating from $(t,x)=(0,0)$.
\end{rmk}

\textbf{Step 4. Iteration; formal solution.} One continues solving away error terms in turns at $M_\circ$ and $\hat M$, obtaining correction terms which vanish to successively higher orders at $\{\eps=0\}=M_\circ\cup\hat M$. Note that after one full step, the decay of the error is improved at both boundary hypersurfaces $M_\circ$ and $\hat M$ by one order; thus, the linear equations one needs to solve in this iteration are always \emph{the same} (i.e.\ as in Steps 2 and 3 above) as far as the growth/decay of the right hand sides are concerned (modulo the possibility of logarithmic factors). Taking an asymptotic sum of $\wt u_0$ and these correction terms produces the desired formal solution $\wt u\in\CI(\wt M)$ (which is, really, only log-smooth) of the singular perturbation problem~\eqref{EqIPfMwtu}--\eqref{EqIPfMwtu2}. (Mirroring Conjecture~\ref{ConjITrue}, obtaining a true solution requires solving $P\wt v=-P\wt u\in\CIdot(\wt M)=\eps^\infty\CI(\wt M)$ with $\wt v\in\CIdot(\wt M)$; accomplishing this requires different arguments which will be discussed elsewhere.)

\subsubsection{Basic setup for black hole gluing}
\label{SssIPfB}

We now turn to the setting of the black hole gluing problem solved formally by Theorem~\ref{ThmI}. Working in the chart~\eqref{EqIPfChart} near the codimension $2$ corner of $\wt M$, we seek $\wt g=(g_\eps)_{\eps\in(0,1)}$ in the form
\[
  \wt g=\wt g(t,\rho_\circ,\hat\rho,\omega)_{\mu\nu}\,\dd z^\mu\,\dd z^\nu,\qquad z=(t,x).
\]
This matches the Gralla--Wald setup \cite{GrallaWaldSelfForce}, with the minor caveat that we need to allow for the presence of (powers of) logarithms (i.e.\ $\log\rho_\circ$ and $\log\hat\rho$) at $M_\circ$ and $\hat M$ in the lower order terms of the expansion of $\wt g$. At $\rho_\circ=0$, we demand that $\wt g$ be equal to $g_{\mu\nu}\,\dd z^\mu\,\dd z^\nu$ at the point $(t,x)=(t,\hat\rho\omega)$. As $\hat\rho\to 0$, this converges to the Minkowski metric, which thus becomes the boundary condition at infinity of the restriction of $\wt g$ to $\hat\rho=0$. This indeed holds for the Kerr metric $\wt g|_{\hat\rho=0}=(\hat g_{\bhm,\bha})_{\mu\nu}(\rho_\circ^{-1}\omega)\,\dd z^\mu\,\dd z^\nu$.

We face additional difficulties compared to the toy model considered in~\S\ref{SssIPfM}.
\begin{enumerate}
\item The equation $\Ric(\wt g)-\Lambda\wt g=0$ which we wish to formally solve is \emph{nonlinear}. A partial relief is the fact that the correction terms which we need to add to a naive gluing $\wt g_0$ are solutions of \emph{linear} equations, concretely of
  \begin{alignat}{2}
  \label{EqIPfBLinCirc}
    (D_g\Ric-\Lambda)h &= f&\quad&\text{(on $M_\circ$)}, \\
  \label{EqIPfBLinHat}
    \wh{D_{\hat g_b}\Ric}(0)h &= f&\quad&\text{(on $\hat M$)},
  \end{alignat}
  where $D_g\Ric$ is the linearization of the Ricci curvature operator, further $\hat g_b$ (with $b=(\bhm,\bha)$) is the metric of the small Kerr black hole, and $\wh{D_{\hat g_b}\Ric}(0)$ is the restriction of $D_{\hat g_b}\Ric$ to stationary symmetric 2-tensors on $\R_{\hat t}\times\R^3_{\hat x}$. Nonetheless, there are nonlinear interactions between various Taylor coefficients in the construction of the formal solution whose treatment requires some care.
\item The linear equation~\eqref{EqIPfBLinCirc} is not hyperbolic and~\eqref{EqIPfBLinHat} is not elliptic, and neither equation is solvable for general $f$. However, the error terms $f$ which arise in the construction are always leading order terms of the error $\Ric(\wt g_j)-\Lambda\wt g_j$ from a previous step of the construction, and thus in view of the second Bianchi identity lie in the kernel of $\delta_{\wt g_j}\sfG_{\wt g_j}$ (where $\sfG_g:=I-\frac12 g\tr_g$ is the \emph{trace reversal} operator), so at $M_\circ$, resp.\ $\hat M$ in the kernel of $\delta_g\sfG_g$, resp.\ $\delta_{\hat g_b}\sfG_{\hat g_b}$. This extra information on $f$ is sufficient in the settings considered in Theorem~\ref{ThmI} for the solvability of~\eqref{EqIPfBLinCirc}; however, equation~\eqref{EqIPfBLinHat} has a nontrivial cokernel even within the kernel of $\delta_{\hat g_b}\sfG_{\hat g_b}$.
\item\label{ItIPfBTaylor} The linear operator $D_g\Ric-\Lambda$ is everywhere characteristic, so even just solving error terms away in Taylor series at $\pa M_\circ$ is a nontrivial task.
\end{enumerate}

The solvability issues arising from the lack of hyperbolicity of~\eqref{EqIPfBLinCirc} can be avoided if one passes to a gauge-fixed version of the Einstein vacuum equations. But since we seek a formal solution of the Einstein vacuum equations themselves, one would need to ensure that the solutions of the corresponding gauge-fixed version of~\eqref{EqIPfBLinCirc} satisfy the linearized gauge conditions---which is equivalent to arranging the validity of gauge conditions at a Cauchy hypersurface (by the usual argument involving the linearized second Bianchi identity). However, we already observed in Remark~\ref{RmkIPfIVP1} the inadequacy of initial value formulations for the purpose of adiabatic gluing. This forces us to free ourselves from any particular choice of gauge, even though we are still allowed to use gauge conditions at various substeps of the construction as long as they are consistent with an adiabatic construction; this will for example allow us to tackle point~\eqref{ItIPfBTaylor} above (see~\S\ref{SssIPfF}). A similar comment applies to the problem of solving away error terms at $\hat M$, cf.\ Remark~\ref{RmkIPfIVP2}.

The starting point of the construction is a naive gluing $\wt g_0=(\wt g_0)_{\mu\nu}\,\dd z^\mu\,\dd z^\nu$ where $(\wt g_0)_{\mu\nu}\in\CI(\wt M)$, with boundary values
\begin{subequations}
\begin{align}
\label{EqIPfBValueMcirc}
  (\wt g_0)_{\mu\nu}|_{(\eps,t,x)=(0,t,x)}&=g_{\mu\nu}|_{(t,x)}, \\
\label{EqIPfBValueMhat}
  (\wt g_0)_{\mu\nu}|_{(\eps,t,\hat x)=(0,t,\hat x)}&=(\hat g_b)_{\hat\mu\hat\nu}|_{\hat x}.
\end{align}
\end{subequations}
Since $g$ and $\hat g_b$ solve the field equations, the tensor $\wt g_0$ satisfies the field equations to leading order at $M_\circ$ and $\hat M$, which means\footnote{That is, the coefficients in the coordinates $z=(t,x)$ on $M$ lie in $\rho_\circ\CI(\wt M)$. We remind the reader that $\Ric(\wt g_0)$ is defined on an $\eps$-level set as the Ricci curvature of the restriction of $\wt g_0$ to this level set; likewise for other geometric quantities and operators on $\wt M$.}\ \footnote{For general Lorentzian metrics $\wt g_0\in\CI(\wt M)$, one only has $\Err_0\in\hat\rho^{-2}\CI(\wt M)$. The gain of \emph{two} orders at $\hat M$ in~\eqref{EqIPfBErr0} holds for a careful choice of $\wt g_0$, and is due to the \emph{quadratic} nature of the difference of $g$ and the Minkowski metric in Fermi normal coordinates.}
\begin{equation}
\label{EqIPfBErr0}
  \Err_0 := \Ric(\wt g_0) - \Lambda\wt g_0 \in \rho_\circ\CI(\wt M).
\end{equation}

\subsubsection{Far field: linearized Einstein equations with sources}
\label{SssIPfF}

We wish to add to $\wt g_0$ a tensor $\eps h$ to solve away the error~\eqref{EqIPfBErr0} to leading order at $M_\circ$. This leads to the linear equation
\begin{equation}
\label{EqIPfFEq}
  (D_g\Ric-\Lambda)h=f_0 := -(\eps^{-1}\Err_0)|_{M_\circ} \in \hat\rho^{-1}\CI(M_\circ).
\end{equation}
Solvability of this equation requires $f_0$ to solve the linearized equations of motion, i.e.\ $\delta_g\sfG_g f_0=0$. Crucially, they hold \emph{automatically} due to the second Bianchi identity for $\wt g_0$, which reads $\delta_{\wt g_0}\sfG_{\wt g_0}\Err_0=0$.

The solvability of equation~\eqref{EqIPfFEq} is discussed in~\S\ref{SAc}. In brief, we first find $h_\sharp$ which solves~\eqref{EqIPfFEq} formally at $r=0$. Since $D_g\Ric-\Lambda$ is characteristic at $N^*\cC$ (cf.\ Remark~\ref{RmkIPfhsharp2}), this is a non-trivial task. We accomplish it as follows: as in Remark~\ref{RmkIPfhsharp1}, one first needs to solve
\begin{equation}
\label{EqIPfFIndicial}
  N\bigl(r^2\wh{D_{\ubar g}\Ric}(0),\lambda\bigr)h_{\sharp,1} \Bigl( = r^{-\lambda+2}D_{\ubar g}\Ric\bigl(r^\lambda h_{\sharp,1}(\omega)\bigr) \Bigr) = f_{0,-1}
\end{equation}
with smooth parametric dependence on $t\in\R$; here $\ubar g=-\dd t^2+\dd x^2$ is the Minkowski metric, which $g$ is equal to at $\cC$, and $f_{0,-1}$ is the $r^{-1}$ coefficient of $f_0$, while $h_{\sharp,1}$ is the sought-after $r^1$ coefficient of $h_\sharp$. The integrability condition $\delta_g\sfG_g f_0=0$ implies the analogous condition $N(r\wh{\delta_{\ubar g}\sfG_{\ubar g}}(0),-1)f_{0,-1}=0$. The solvability (and uniqueness) theory of~\eqref{EqIPfFIndicial}, which is an equation on spacetime symmetric 2-tensors on Minkowski space restricted to $(0,\infty)_r$ times a coordinate 2-sphere which are quasi-homogeneous in $r$, is studied in detail in~\S\ref{SRic}, the main results being Theorem~\ref{ThmRicUniq} and \ref{ThmRicSolv}. The upshot is that a formal solution of~\eqref{EqIPfFEq} exists, with precise control also on the logarithmic factors $(\log r)^k$ appearing in its generalized Taylor expansion at $r=0$.

The second step is to solve away the remaining error by solving
\[
  (D_g\Ric-\Lambda)h_\flat = f_0 - (D_g\Ric-\Lambda)h_\sharp.
\]
The right hand side now vanishes to infinite order at $\cC$ and thus is smooth on $M$; moreover, it still lies in the kernel of $\delta_g\sfG_g$. However, this by itself does not suffice for the existence of a solution $h_\flat$ with controlled support unless $(M,g)$ does not admit nontrivial Killing vector fields; see \cite{HintzLinEin}.\footnote{The control of supports---specifically, ensuring that correction terms vanish near spacelike infinity---is mainly of importance when performing gluing constructions on asymptotically flat spacetimes, with combinations with stability results in mind. See Remark~\ref{RmkXDomain}.} For the proof, one first solves the linearized constraint equations at $X$ using results going back to \cite{ChruscielDelayMapping}; this uses the assumptions on $\cU^\circ$ in Theorem~\ref{ThmI}, which ensure that the initial data for $h_\flat$ can be chosen to have support contained in $\cU^\circ$. (In the setting of Theorem~\ref{ThmI2}, we give up control of the initial data for $h_\flat$---and indeed allow for arbitrary growth at infinity in $X$ of the initial data---in return for the ability to unconditionally solve the linearized constraints.) Subsequently, one finds $h_\flat$ as the solution of a gauge-fixed version of the linearized Einstein equations; the support property in part~\eqref{ItISupp} of Theorem~\ref{ThmI} follows by finite speed of propagation. (Since at this point we work only with smooth tensors on $M$, gauge-fixing is consistent with the adiabatic nature of the gluing problem at $\cC$.)

The tensor $h=h_\sharp+\upbeta_\circ^*h_\flat$ solves~\eqref{EqIPfFEq}, and thus the log-smooth tensor
\begin{equation}
\label{EqIPfFg1}
  \wt g^1 := \wt g_0 + \eps h
\end{equation}
on $\wt M$ satisfies (ignoring\footnote{For completeness, we mention that in early stages of the construction one does need to keep careful track of leading order logarithmic factors, as e.g.\ in~\S\ref{SssFhM1Mod}.} logarithmic terms)
\[
  \Err^1 := \Ric(\wt g^1) - \Lambda\wt g^1 \in \rho_\circ^2\CI(\wt M).
\]
See~\S\ref{SsFMc1} for details.

\subsubsection{Near field: modulation of black hole parameters}
\label{SssIPfN}

Corresponding to Step 3 in~\S\ref{SssIPfM}, we seek an adiabatic tensor $h=h(t,\hat x)_{\mu\nu}\,\dd z^\mu\,\dd z^\nu$ so that $\wt g^1+\eps^2 h$ solves the Einstein vacuum equations to one order more at $\hat M$ than $\wt g^1$; this leads to the equation\footnote{As already noted by D'Eath \cite{DEathSmallBHDynamics}, `[W]e have QS [quasistationary] internal perturbations because gravitational waves only need a time of order $M$' (in present notation: $\eps$) `to cross the black hole, whereas the background is changing on a time scale of order $1$. Thus the small black hole can adjust its gravitational field on what it feels to be a long time scale in order to cope with the tidal field of the background.'}
\begin{equation}
\label{EqIPfNEq}
  \wh{D_{\hat g_b}\Ric}(0)h = f^1,
\end{equation}
where $f^1=-\Err^1|_{\hat M}\in\rho_\circ^2\CI(\hat M)$. The necessary condition for solvability, $\delta_{\hat g_b}\sfG_{\hat g_b}f^1=0$, follows from the second Bianchi identity for $\wt g^1$. By first solving~\eqref{EqIPfNEq} to infinite order at $|\hat x|^{-1}=0$, one can reduce to the case that $f^1$ has rapid decay as $|\hat x|\to\infty$.

For each fixed $t\in\R$, we are faced with the equation $D_{\hat g_b}\Ric(h)=f^1$ where $f^1=f^1(\hat x)$---dropping the $t$-dependence---, and we seek a stationary solution $h=h(\hat x)$. \emph{Pure gauge} tensors $h=\delta_{\hat g_b}^*\omega$, where $\omega$ is an arbitrary 1-form, solve the homogeneous equation; and when $\omega$ is a spatial translation or rotation, such $h$ have good (i.e.\ $\cO(|\hat x|^{-1})$ or better) decay as $\rho_\circ=\la\hat x\ra^{-1}\to 0$. \emph{Linearized Kerr metrics} $\hat g'_b(\dot b)=\frac{\dd}{\dd s}\hat g_{b+s\dot b}|_{s=0}$, $\dot b=(\dot\bhm,\dot\bha)\in\R\times\R^3$, provide further homogeneous solutions with good decay. There is a 7-dimensional space of homogeneous solutions comprised of certain such tensors.\footnote{Rotations are either Killing vector fields, or can be rewritten as changes of the axis of rotation of the Kerr black hole. Thus, there are $3$ translational and $4$ black hole parameter degrees of freedom.}

Dually, asymptotic symmetries (temporal and spatial translations, and rotations) give rise to a 7-dimensional space of \emph{dual pure gauge} tensors $\sfG_{\hat g_b}\delta_{\hat g_b}^*\omega$ in the kernel of the adjoint of $\wh{D_{\hat g_b}\Ric}(0)$; see~\S\ref{SsAh0} for details. This gives a 7-dimensional cokernel (i.e.\ obstruction space for the solvability) of~\eqref{EqIPfNEq},\footnote{We do not directly phrase this as a Fredholm index $0$ statement for $\wh{D_{\hat g_b}\Ric}(0)$. But for the proof we do use that a \emph{gauge-fixed version} of the linearized Ricci curvature operator, at zero frequency, can be regarded as a Fredholm index $0$ operator between suitable function spaces, as shown in \cite[Theorem~4.3]{HaefnerHintzVasyKerr}.} given by integration (over the spatial manifold $\R^3_{\hat x}\setminus\{|\hat x|<\bhm\}$ of Kerr) against these dual pure gauge tensors.

As is usual in geometric gluing problems, the basic idea is to avoid this cokernel via a modulation procedure which takes advantage of the flexibility we have in inserting the small black hole. The particular form which this modulation takes depends on the order of vanishing of the error term at $\hat M$. The basic observation is the following: given an adiabatic tensor $h=h(t,\hat x)=h_{\hat\mu\hat\nu}(t_0+\eps\hat t,\hat x)\,\dd\hat z^\mu\,\dd\hat z^\nu$, where $\hat z=(\hat t,\hat x)$ with $\hat t=\frac{t-t_0}{\eps}$ and $\hat x=\frac{x}{\eps}$, we can expand it in Taylor series in $\eps$ to obtain (for bounded $\hat z$)
\[
  h_0(\hat x) + \eps\hat t h_1(\hat x) + \eps^2\frac{\hat t^2}{2}h_2(\hat x) + \cO(\eps^3),\qquad
  h_j(\hat x) := \pa_t^j h(t_0,\hat x).
\]
While to leading order at $\eps=0$ this is stationary, the coefficient of $\eps^j$, $j\geq 1$, is a polynomial of degree $j$ in $\hat t$. Thus,
\[
  D_{\hat g_b}\Ric(h)=\wh{D_{\hat g_b}\Ric}(0)h_0 + \eps D_{\hat g_b}\Ric(\hat t h_1) + \eps^2 D_{\hat g_b}\Ric\Bigl(\frac{\hat t^2}{2}h_2\Bigr) + \cO(\eps^3).
\]
Roughly speaking, this means that a correction term $\eps^k h$ to the family of spacetime metrics produces a correction term to the output of $\Ric-\Lambda$ at order\footnote{The factor of $\eps^{-2}$ arises from switching from the coordinates $\hat z$ back to $z=(t,x)$.} $\eps^{-2}\eps^{k+j}=\eps^{k+j-2}$ which is of the form $\eps^{k+j-2}D_{\hat g_b}\Ric(\frac{\hat t^j}{j!}h_j)$. This can move the $\cO(\eps^{k+j-2})$ error term one is trying to solve away off the cokernel of $\wh{D_{\hat g_b}\Ric}(0)$.

\textbf{Modulation at the first step.} In~\eqref{EqIPfNEq}, we use a variant of this observation. We revisit the definition of $\wt g_0$: instead of gluing the Kerr black hole into $\hat M$ via~\eqref{EqIPfBValueMhat}, we shall set
\[
  (\wt g_0)_{\mu\nu}|_{(\eps,t,\hat x)=(0,t,\hat x)} = (\hat g_b)_{\hat\mu\hat\nu}|_{\hat x+\hat c(t)}
\]
for a function $\hat c\in\CI(\R;\R^3)$ which we need to determine. Expanding in Taylor series around $t=t_0$ as above, and assuming that $\hat c(t_0)=0$ for notational simplicity, this is
\[
  \hat g_b|_{\hat x} + \eps\hat t \cL_{\hat c'(t_0)\cdot\pa_{\hat x}}\hat g_b + \eps^2\frac{\hat t^2}{2}\bigl(\cL_{\hat c''(t_0)\cdot\pa_{\hat x}}\hat g_b + \cL_{\hat c'(t_0)\cdot\pa_{\hat x}}^2\hat g_b\bigr) + \cO(\eps^3).
\]
We can add to this a further adiabatic $\cO(\eps)$ term so that the total coefficient of $\eps^1$ of the resulting family $\wt g_{\hat c}$ is the Lie derivative of $\hat g_b$ along the Lorentz boost $\hat t\hat c'(t_0)\cdot\pa_{\hat x}+(\hat c'(t_0)\cdot\hat x)\pa_{\hat t}$. This implies that $\Err_{0,\hat c}:=\Ric(\wt g_{\hat c})-\Lambda\wt g_{\hat c}$ is of the same class as $\Err_0$ in~\eqref{EqIPfBErr0}. (Also, the leading order term of $\Err_{0,\hat c}$ at $M_\circ$ is the same as that of $\Err_0$, which means that the correction step at $M_\circ$ is unaffected by the presence of $\hat c$.) The new leading order term $f_{\hat c}^1$, defined analogously to~\eqref{EqIPfNEq}, \emph{is} sensitive to $\hat c(t)$, and indeed gets changed by\footnote{The argument really involves an additional term which is linear in $\hat t$ and $\hat c''(t_0)$, arising from the aforementioned further $\cO(\eps)$ correction term. We also omit further terms arising from the $\cO(\eps)$ terms of $\wt g_{\hat c}$ through the quadratic terms in the Einstein equations.}
\begin{equation}
\label{EqIPfNCorr}
  D_{\hat g_b}\Ric\Bigl(\frac{\hat t^2}{2}\cL_{\hat c''(t_0)\cdot\pa_{\hat x}}\hat g_b\Bigr).
\end{equation}
The inner product of this term with the dual pure gauge solutions related to the three generators of spatial translations can be made to be equal to any desired three values for a suitable choice of $\hat c''(t_0)$. In this manner, one can eliminate $3$ out of $7$ obstructions to the solvability of~\eqref{EqIPfNEq} by solving a linear second order ODE for $\hat c$.\footnote{The fact that quadratic-in-$\hat t$ correction terms are required to eliminate parts of the cokernel is closely related to the fact that the resolvent family for the linearization of a gauge-fixed version of the Einstein vacuum equations around Kerr has a second order pole at zero frequency \cite{HaefnerHintzVasyKerr}.}

The remaining $4$ obstructions can be eliminated by modulating the black hole parameters at order $\eps$: if instead of $\hat g_b$ one uses $\hat g_b+\eps\hat g'_b(\dot b(t))=\hat g_b+\eps\hat g_b'(\dot b(t_0))+\eps^2\hat t\hat g'_b(\dot b'(t_0))+\cO(\eps^3)$, with $\dot b(t)=(\dot\bhm(t),\dot\bha(t))$ to be determined, one produces a further correction term, in addition to~\eqref{EqIPfNCorr}, given by
\[
  D_{\hat g_b}\Ric\bigl(\hat t\hat g'_b(\dot b(t_0))\bigr).
\]
This can be made to integrate against the dual pure gauge solutions related to time translations, resp.\ generators of spatial rotations to yield any 4-tuple of numbers if one chooses $\dot\bhm'(t_0)$, resp.\ $\dot\bha'(t_0)$ suitably. For representation theoretic reasons, at this first correction step at $\hat M$ only those choices of $\dot b$ are needed which correspond to infinitesimal rotations of the rotation axis; these are pure gauge solutions (see Lemma~\ref{LemmaFhM1AxisLie} and Corollary~\ref{CorFhM1CokerParam}). See Theorem~\ref{ThmAhKCoker} and Proposition~\ref{PropFhM1} for details.

In summary, through adiabatic translations and sub-leading order pure gauge changes of the black hole parameters we can move the leading order error at $\hat M$ in~\eqref{EqIPfNEq} into the range of $\wh{D_{\hat g_b}\Ric}(0)$; see Theorem~\ref{ThmAhPhgSolv}. (For the proof of this theorem, we use a gauge-fixing procedure and apply results from \cite{HaefnerHintzVasyKerr,AnderssonHaefnerWhitingMode} on the solvability properties of the zero energy operator of the gauge-fixed linearized Ricci curvature operator.) Since the Einstein vacuum equations are diffeomorphism-covariant, one can pull back the resulting family of metrics $\wt g^1+\eps^2 h$ along a suitable diffeomorphism to re-center the center of mass and axis of rotation of the small Kerr black holes.

\textbf{Modulation at later steps.} At later stages of the construction of $\wt g$ in Theorem~\ref{ThmI}, we modulate the center of mass and black hole parameters at higher orders in $\eps$. For instance, if we add to the family $\wt g^k$ of spacetime metrics after the $k$-th step, $k\geq 2$, a correction term $\eps^{k-1}\cL_{\hat c(t)\cdot\pa_{\hat x}}\hat g_b$ or $\eps^k\hat g'_b(\dot b(t))$, we can eliminate the cokernel for the size $\eps^{-2}\eps^{k+1}=\eps^{k-1}$ leading order term of the error $\Ric(\wt g^k)-\Lambda\wt g^k$. For small $k$, the details are somewhat involved, as one needs to take into account quadratic and cubic\footnote{For instance, a $\cO(\eps^1)$ modulation of the center of mass of the small black hole, required to eliminate the cokernel when solving away a $\cO(\eps^{-2}\eps^{1+2})=\cO(\eps)$ error term at $\hat M$, produces also a $\cO(\eps^{-2}\cdot(\eps^1)^3)=\cO(\eps)$ term via cubic self-interaction, and further $\cO(\eps)$ terms via quadratic interactions with additional $\cO(\eps^2)$ correction terms.} nonlinearities of the Einstein vacuum equations, though only rather elementary structural information about these terms suffices. This is carried out in~\S\S\ref{SsFhM2} and \ref{SsFhM3}.

We also remark that during early stages of the construction, we keep careful track of the logarithmic terms at $\hat M$ and $M_\circ$ as well as of the algebraic structure of certain terms in the generalized Taylor expansion of the error terms at $\hat M$. This is needed to ensure the equality of $\wt g$ and $\hat g_b$ modulo $\cO(\eps^2)$ errors at $\hat M$ in part~\eqref{ItINear} of Theorem~\ref{ThmI}, which plays an important role in the construction as it significantly reduces the number of nonlinear interaction terms one needs to keep track of.

\subsubsection{Formal solution at a Cauchy hypersurface}
\label{SssIPfC}

Part~\eqref{ItIFormalCauchy} of Theorem~\ref{ThmI} is proved in~\S\ref{SIVP}. If $g_\eps$ solves $\Ric(g_\eps)-\Lambda g_\eps=\cO(\eps^\infty)$, then also the constraint equations are valid at $X\cap M_\eps$ modulo $\cO(\eps^\infty)$. The first step is to correct the first and second fundamental forms of $g_\eps$ at $X\cap M_\eps$ by tensors of size $\cO(\eps^\infty)$ so that the constraint equations are satisfied exactly. We accomplish this by adapting a contraction mapping type argument from \cite{HintzGlueID}. In a second step, we construct the Taylor series of the metric tensor at $X\cap M_\eps$ by expressing it in a $(3+1)$-splitting with fixed lapse and shift.

\subsection{Outline of the paper}
\label{SsIO}

\begin{itemize}
\item We begin in~\S\ref{SBg} with a review of notions from geometric singular analysis which are used throughout the paper, in particular blow-ups, b- and scattering structures and their (reduced) 3-body analogues, as well as polyhomogeneity.
\item In~\S\ref{SG}, we describe in detail the manifold $\wt M$, already introduced in~\S\ref{SssIPfM} above, on which the gluing construction will take place. The family $\wt g$ of metrics is a section of a smooth vector bundle on $\wt M$ (with local trivializations induced by lifts of coordinates $z=(t,x)$ on $M$) which we study in some detail. In particular, we explain the sense in which smooth sections of this bundle induce stationary metrics on the fibers of $\hat M$ over each point of the geodesic $\cC\subset M$; and we analyze the structural properties of differential operators and geometric quantities associated with such metrics.
\item In~\S\ref{SE}, we recall aspects of initial value problems and gauge-fixing for the Einstein vacuum equations and their linearizations, as well as the structural properties of these equations on $\wt M$.
\item In~\S\ref{SM}, we introduce and state the main result of this paper in full detail; see Theorem~\ref{ThmM}.
\item As explained in Remark~\ref{RmkIPfhsharp1} and~\S\ref{SssIPfF}, the analysis of the linearized field equations at $M_\circ$ and $\hat M$ utilizes the construction of formal solutions at $\pa M_\circ=M_\circ\cap\hat M=\pa\hat M$, which requires a detailed analysis of the linearized field equations on Minkowski space acting on tensors which are quasi-homogeneous with respect to spatial dilations. This is the content of~\S\ref{SRic}, following the computation of the explicit form of various geometric operators on Minkowski space in~\S\ref{SMk}.
\item The linear theory in the far field, or more precisely on $M_\circ$, as sketched in~\S\ref{SssIPfF}, is developed in~\S\ref{SAc}.
\item The linear theory in the near field, or more precisely on $\hat M$, as sketched in~\S\ref{SssIPfN}, is developed in~\S\ref{SAh}; this includes a precise description of the cokernel of the linearization of the Einstein vacuum equations around Kerr at zero frequency.
\item The heart of the paper is the construction of a formal solution of the gluing problem at $\eps=0$ in~\S\ref{SF}; an outline of the detailed construction is given there.
\item The formal solution at the Cauchy surface $X\cap M_\eps$, as discussed in~\S\ref{SssIPfC}, is constructed in~\S\ref{SIVP}.
\item Up to this point, all proofs take place in setting~\eqref{ItIOGeneric} of Theorem~\ref{ThmI}. The minor modifications required to handle the case of Kerr(--de~Sitter) background spacetimes, i.e.\ setting~\eqref{ItIOKdS}, and the application to the construction of extreme mass ratio mergers, is described in~\S\ref{SX}.
\end{itemize}

\section{Background on geometric singular analysis}
\label{SBg}

An $n$-dimensional manifold $M$ with corners is diffeomorphic to $[0,\infty)^k\times\R^{n-k}$ in the neighborhood of a point $p\in M$, where $k\in\{0,\ldots,n\}$ depends on $p$. The boundary hypersurfaces of $M$ are the closures of the connected components of the set of $p\in M$ for which $k=1$; following \cite{MelroseDiffOnMwc}, we require all boundary hypersurfaces to be embedded submanifolds. If one can take $k=0,1$ for all $p\in M$, then $M$ is a manifold with boundary. A \emph{boundary defining function} of a boundary hypersurface $H\subset M$ is a smooth function $\rho\in\CI(M)$ so that $\rho\geq 0$ on $M$, $\rho^{-1}(0)=H$, and $\dd\rho\neq 0$ on $H$. Any two boundary defining functions of the same boundary hypersurface are smooth nonzero multiples of each other. When working in an open subset $U\subset M$, a local boundary defining function of $H$ is a function $\rho\in\CI(U)$ satisfying these conditions on $U$.

An important example of a manifold with boundary is the radial compactification of $\R^n$, defined by
\[
  \ol{\R^n} := \Bigl( \R^n \sqcup \bigl( [0,\infty)_\rho\times\Sph^{n-1}_\omega\bigr) \Bigr) / \sim,\qquad \R^n\setminus\{0\} \ni x=r\omega \sim (r^{-1},\omega),
\]
where $r=|x|$ and $\omega=\frac{x}{|x|}$ are standard polar coordinates on $\R^n$. A boundary defining function of the sphere at infinity $\pa\ol{\R^n}=\rho^{-1}(0)\cong\Sph^{n-1}$ is $\la x\ra^{-1}=(1+|x|^2)^{-1/2}$; a local boundary defining function in $x\neq 0$ is $|x|^{-1}$. Note that the space $\CI(\ol{\R^n})$ of smooth functions on $\ol{\R^n}$ consists of all smooth functions $u$ on $\R^n$ which, when expressed in terms of $|x|^{-1}$ and $\frac{x}{|x|}$, are smooth down to $|x|^{-1}=0$; this means that they have Taylor expansions at infinity,
\[
  u \sim \sum_{j\geq 0} |x|^{-j}u_j\Bigl(\frac{x}{|x|}\Bigr),\qquad |x|\to\infty,\qquad u_j\in\CI(\Sph^{n-1}),
\]
meaning that the difference of $u$ and the truncation of the sum to $j\leq J$ is smooth and vanishes to order $J$ at $|x|^{-1}=0$.

The procedure of (real) blow-up produces a manifold with corners if one is given a manifold with corners $M$ and a \emph{p-submanifold} $S\subset M$: this is a submanifold so that at each point $p\in S$ there exists a coordinate chart $[0,\infty)^k\times\R^{n-k}$ on $M$ so that $S$ is given by the vanishing of a subset of the collection of local coordinates. Namely, the blow-up of $M$ along $S$ is
\[
  [M;S] := (M\setminus S) \sqcup S N^+S
\]
as a set, where $N^+S=T^+_S M/T^+S$ is the inward pointing normal bundle (with $T^+_q M$, for $q\in M$, consisting of all non-strictly inward pointing tangent vectors), and the inward pointing spherical normal bundle $S N^+S=(N^+S\setminus o)/\R^+$ is the quotient of the complement of the zero section $o\subset N^+S$ by the dilation action in the fibers. This can be given a smooth structure by declaring polar coordinates around $S$ to be smooth down to the polar coordinate origin. The \emph{blow-down} map
\[
  \upbeta \colon [M;S]\to M
\]
is defined to be the identity map on $M\setminus S$ and the base projection $S N^+S\to S$ on the \emph{front face} $S N^+S$.

We make this concrete in local coordinates, so $S=\{x^1=\ldots=x^j=0,\ y^1=\ldots=y^m=0\}\subset M=[0,\infty)_x^k\times\R^{n-k}_y$, where $0\leq j\leq k$ and $0\leq m\leq n-k$. If $j=0$ (thus $S$ is not contained in a boundary hypersurface), then
\begin{align*}
  &[M;S] = [0,\infty)_x^k \times \R_{(y^{m+1},\ldots,y^{n-k})}^{n-k-m} \times [0,\infty)_R \times \Sph_\omega^{m-1}, \\
  &\hspace{8em} R:=\biggl(\sum_{i=1}^m(y^i)^2\biggr)^{1/2},\ \omega:=\frac{(y^1,\ldots,y^m)}{R},
\end{align*}
with $S N^+S=R^{-1}(0)$. If $j\geq 1$, so $S$ is a boundary p-submanifold, we have
\[
  [M;S] = [0,\infty)_{(x^{j+1},\ldots,x^k)}^{k-j} \times \R_{(y^{m+1},\ldots,y^{n-k})}^{n-k-m} \times [0,\infty)_R \times \Sph^{j+m-1}_j,
\]
where $R=(\sum_{i=1}^j(x^i)^2+\sum_{i=1}^m(y^i)^2)^{1/2}$ and
\[
  \frac{(x^1,\ldots,x^j,y^1,\ldots,y^m)}{R}\in\Sph^{j+m-1}_j:=\{(\xi^1,\ldots,\xi^j,\eta^1,\ldots,\eta^m) \in \Sph^{j+m-1},\ \xi^1,\ldots,\xi^j\geq 0\}.
\]
The blow-down map is given by the product of the identity map in the coordinates $x^{j+1}$, $\ldots$, $x^k$, $y^{m+1}$, $\ldots$, $y^{n-k}$, and the polar coordinate map $(R,\omega)\mapsto R\omega$ in the remaining variables.

If $T\subset M$ is a submanifold, then the \emph{lift} $\upbeta^*T$ of $T$ to $[M;S]$ is defined to be $\upbeta^{-1}(T)$ when $T\cap S=\emptyset$, and the closure of $\upbeta^{-1}(T\setminus S)$ in $[M;S]$ otherwise. If $S,T\subset M$ are p-submanifolds so that the lift of $T$ to $[M;S]$ is again a p-submanifold, one can form the iterated blow-up $[M;S;T]:=[[M;S];\upbeta^*T]$. In the case that $S$ is a p-submanifold of $M$ and $T\subset S$ is also given by the vanishing of a subset of local coordinates in which already $S$ is of this form, then this condition is satisfied for $S,T$ and also for $T,S$; and the two iterated blow-ups $[M;S;T]$ and $[M;T;S]$ are naturally diffeomorphic (i.e.\ the identity map on $M\setminus S$ extends to a diffeomorphism of these two manifolds with corners).

As a special case, let $M$ denote a smooth $n$-dimensional manifold without boundary, and let $\cC\subset M$ be a closed p-submanifold of codimension $k$. Consider $\wt M=[[0,1)\times M;\{0\}\times\cC]$. Then the front face $\hat M\subset\wt M$ is a fiber bundle over $\cC$ with typical fiber $\ol{\R^k}$. In fact, there is a natural diffeomorphism $\hat M\cong\ol{N}\cC$ where $\ol{N}\cC$ is the fiber-wise radial compactification of the normal bundle $N\cC=T_\cC M/T\cC$; indeed, given $p\in\cC$ and a representative $V\in T_p M$ of an element of $N\cC$, let $\gamma\colon(-1,1)\to M$ denote a smooth curve with $\gamma(0)=p$, $\gamma'(0)=V$; then we can map $V$ to $\lim_{\eps\searrow 0}(\eps,\gamma(\eps))\in\wt M$. A local coordinate calculation shows that this extends by continuity to the claimed diffeomorphism. The lift of $\{0\}\times M$ is equal to $[M;\cC]$.

\subsection{Lie algebras of vector fields}
\label{SsBgLie}

On a manifold with corners $M$, the space $\Vb(M)$ of \emph{b-vector fields} \cite{MelroseMendozaB,MelroseAPS,GrieserBasics} consists of all smooth vector fields $V\in\cV(M):=\CI(M;T M)$ which are tangent to all boundary hypersurfaces. In local coordinates $[0,\infty)_x^k\times\R^{n-k}_y$, these are linear combinations, with smooth coefficients, of the vector fields $x^i\pa_{x^i}$ ($i=1,\ldots,k$) and $\pa_{y^j}$ ($j=1,\ldots,n-k$). If $M$ is a manifold with boundary, and if $\rho\in\CI(M)$ is a boundary defining function, then $\Vsc(M):=\rho\Vb(M)=\{\rho V\colon V\in\Vb(M)\}$ is the space of \emph{scattering vector fields} \cite{MelroseEuclideanSpectralTheory}. Both $\Vb(M)$ and $\Vsc(M)$ are Lie algebras with respect to the commutator of vector fields, and therefore we have graded algebras
\[
  \Diffb(M) = \bigoplus_{m\in\N_0} \Diffb^m(M),\qquad
  \Diffsc(M) = \bigoplus_{m\in\N_0} \Diffsc^m(M)
\]
of differential operators which are locally finite sums of up to $m$-fold compositions of b- and scattering vector fields, respectively.

From now on we only consider the case that $M$ is a manifold with boundary. In a local coordinate chart $[0,\infty)_\rho\times\R^{n-1}_y$, an element $P\in\Diffb^m(M)$ is of the form
\[
  P = \sum_{j+|\alpha|\leq m}a_{j\alpha}(\rho,y)(\rho\pa_\rho)^j \pa_y^\alpha,\qquad a_{j\alpha}\in\CI([0,\infty)\times\R^{n-1}).
\]
Its \emph{normal operator} at $\pa M=\rho^{-1}(0)$ is given by restricting the coefficients to $\rho=0$, giving
\[
  N(P) := \sum_{j+|\alpha|\leq m} a_{j\alpha}(0,y)(\rho\pa_\rho)^j \pa_y^\alpha \in \Diff_{\bop,\rm I}^m([0,\infty)\times\R^{n-1}),
\]
where the subscript `$I$' records the invariance of $N(P)$ under dilations $(\rho,y)\mapsto(\mu\rho,y)$, $\mu>0$. The action of $N(P)$ on functions of the form $\rho^\lambda v(y)$ is given by the \emph{indicial family}
\[
  N(P,\lambda) := \sum_{j+|\alpha|\leq m} a_{j\alpha}(0,y)\lambda^j \pa_y^\alpha \in \Diff^m(\R^{n-1}).
\]
Globally, one can define $N(P)\in\Diff_{\bop,\rm I}^m([0,\infty)\times\pa M)$ if one fixes a collar neighborhood of $M$, and then $N(P,\lambda)\in\Diff^m(\pa M)$.

We next recall that $\Vsc(M)$ is the space of smooth sections of the \emph{scattering tangent bundle} $\Tsc M$, with local frame $\rho^2\pa_\rho$, $\rho\pa_{y^j}$ ($j=1,\ldots,n-1$) ; the dual bundle is the \emph{scattering cotangent bundle} $\Tsc^*M$, with local frame $\frac{\dd\rho}{\rho^2}$, $\frac{\dd y^j}{\rho}$ ($j=1,\ldots,n-1$). When $M=\ol{\R^n}$, then a computation in projective coordinates shows that $\Vsc(\ol{\R^n})$ is spanned over $\CI(\ol{\R^n})$ by the standard coordinate vector fields $\pa_{x^1},\ldots,\pa_{x^n}$; thus, these form a \emph{global} (i.e.\ down to $\pa\ol{\R^n}$) smooth frame of $\Tsc\ol{\R^n}$, and the differentials $\dd x^1,\ldots,\dd x^n$ form a global smooth frame of $\Tsc^*\ol{\R^n}$. The Euclidean metric $\sum_{j=1}^n(\dd x^j)^2$ is thus an example of a smooth positive definite section of $S^2\,\Tsc^*\ol{\R^n}$ (also called a Riemannian scattering metric).

If $p\in\pa M$, then the blow-up $[M;\{p\}]$ is a manifold with corners. Following Vasy \cite{VasyThreeBody}, we define the Lie algebra $\Vtsc([M;\{p\}])$ of \emph{3-body-scattering vector fields} (or \emph{3sc-vector fields}) as the $\CI([M;\{p\}])$-span of the space of lifts of elements of $\Vsc(M)$ to $[M;\{p\}]$. This generalizes in a straightforward manner to the case that one blows up several distinct points in $\pa M$. The case of interest in the present paper will be $M=\ol{\R^n}$ where $\R^n=\R_t\times\R_x^{n-1}$, and we blow up the `north' and `south poles' $\{N,S\}=\pa\ol\R\times\{0\}\subset\pa\ol{\R^n}$; in the case $n=4$, a subset of
\begin{equation}
\label{EqBgBlowupNS}
  [\ol{\R^n};\{N,S\}]
\end{equation}
carries the Kerr metric as a smooth (and stationary) Lorentzian 3sc-metric, i.e.\ a Lorentzian signature section of $S^2\,\Ttsc^*[\ol{\R^n};\{N,S\}]$. The main point is that the two front faces of $[\ol{\R^n};\{N,S\}]$ are diffeomorphic to $\ol{\R^{n-1}_x}$, and thus linear combinations of the second symmetric tensor products of $\dd t,\dd x^1,\ldots,\dd x^{n-1}$ with $\CI(\ol{\R^{n-1}_x})$-coefficients are smooth symmetric 3sc-2-tensors. See Figure~\ref{FigBg3sc}.

\begin{figure}[!ht]
\centering
\includegraphics{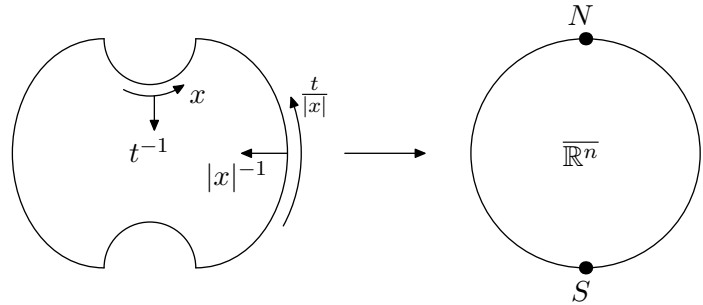}
\caption{Illustration of $[\ol{\R^n};\{N,S\}]$ and the blow-down map to $\ol{\R^n}$, together with some local coordinates.}
\label{FigBg3sc}
\end{figure}

We further recall from \cite{Hintz3b} the Lie algebra of \emph{3-body/b-vector fields} (or \emph{3b-vector fields}): if $\rho_\sface$ is a defining function of the lift of $\pa M$ to $[M;\{p\}]$, then this is defined as
\[
  \Vtb([M;\{p\}]):=\rho_\sface^{-1}\Vtsc([M;\{p\}]).
\]
Associated with these Lie algebras, we have graded algebras of differential operators, which we denote $\Difftsc([M;\{p\}])$ and $\Difftb([M;\{p\}])$. As an important special case, consider again $[\ol{\R_t\times\R_x^{n-1}};\{N,S\}]$ from~\eqref{EqBgBlowupNS}, and let us introduce inverse polar coordinates $\rho=|x|^{-1}$, $\omega=\frac{x}{|x|}\in\Sph^{n-2}$ in $x\neq 0$. Then elements of $\Difftb^m([\ol{\R^n};\{N,S\}])$ which are moreover invariant under translations in $t$ are, in $x\neq 0$, of the form
\[
  P = \sum_{j+k+|\alpha|\leq m} a_{j k\alpha}(x) (\rho^{-1}\pa_t)^j (\rho\pa_\rho)^k \pa_\omega^\alpha,
\]
where $a_{j k\alpha}\in\CI(\ol{\R^{n-1}})$. (The wave operator with respect to the Kerr metric, or indeed any stationary 3sc-metric, is of this form upon multiplication by $r^2$.) Formally passing to the Fourier transform in $t$, i.e.\ replacing $\pa_t$ by $-i\sigma$ where $\sigma\in\R$, produces the spectral family
\[
  \hat P(\sigma) = \sum_{j+k+|\alpha|\leq m} a_{j k\alpha}(x) (-i\sigma\rho^{-1})^j (\rho\pa_\rho)^k \pa_\omega^\alpha,
\]
which for $\sigma=0$, resp.\ $\sigma\neq 0$ gives an element of $\Diffb^m(\ol{\R^{n-1}})$, resp.\ $\rho^{-m}\Diffsc^m(\ol{\R^{n-1}})$. We refer the reader to \cite{Hintz3b,HintzNonstat} for further information on 3b-operators and their relationships with scattering or 3-body-scattering geometries.

Lastly, we recall from \cite{MazzeoEdge} the Lie algebra of \emph{edge vector fields}, defined on a manifold $M$ with boundary whose boundary hypersurface $\pa M$ is the total space of a fibration $Z-\pa M\to Y$; to wit, $\Ve(M)$ consists of all smooth vector fields on $M$ which are tangent to the fibers of the fibration. (In particular, $\Ve(M)\subset\Vb(M)$.) This situation arises when blowing up a p-submanifold inside a smooth manifold without boundary: the front face fibers over the p-submanifold, with the typical fiber being a sphere. See also~\S\ref{SssBgE}. Similarly to before, there exists an associated graded algebra $\Diffe^m(M)$ of edge differential operators.

For all algebras $D$ of differential operators introduced so far, one can also consider weighted versions: if $w$ is a weight, i.e.\ the product of (real) powers of boundary defining functions, one can define $w D:=\{w A\colon A\in D\}$. If $w'$ is another weight, then $w D\circ w' D\subset(w w')D$. (This follows from the fact that $w^{-1}V(w)$ is smooth, including at the boundary, when $w$ is a weight and $V$ is a b-vector field.)

\subsubsection{Edge operators arising from blow-ups}
\label{SssBgE}

Suppose $M$ is a smooth $(n+1)$-dimensional manifold without boundary, and let $\cC\subset M$ denote a closed 1-dimensional submanifold. Let $E,F\to M$ denote smooth vector bundles. We will encounter the following situation: we are given a smooth coefficient operator
\[
  L\in\Diff^m(M;E,F)
\]
for which we wish to solve $L u=f$ in $M\setminus\cC$ where $f$ is conormal or polyhomogeneous at $\cC$; and we wish for $u$ to be conormal or polyhomogeneous as well. Under the assumption that $L$ is non-characteristic at $N^*\cC$, this is a classical problem which can e.g.\ be solved using the symbol calculus for conormal distributions \cite[\S18]{HormanderAnalysisPDE3}; in our application, however, $L$ will be the linearized Einstein operator which does not satisfy this assumption.

We instead proceed as follows: write $M_\circ=[M;\cC]$ for the real blow-up, and $\upbeta_\circ\colon M_\circ\to M$ for the blow-down map. The front face $\pa M_\circ$ of $M_\circ$ is the total space of the fibration $\Sph^{n-1}-\pa M_\circ\to\cC$ given by the blow-down map restricted to the front face. Let $r\in\CI(M_\circ)$ denote a defining function of the front face $\pa M_\circ$.

\begin{lemma}[Edge operator via blow-up]
\label{LemmaBgE}
  The lift of $r^m L$ to $M_\circ$ satisfies
  \[
    L_\eop := \upbeta_\circ^*(r^m L)\in\Diffe^m(M_\circ;\upbeta_\circ^*E,\upbeta_\circ^*F).
  \]
\end{lemma}
\begin{proof}
  In local coordinates $(t,x)$ along $\cC$, with $\cC=x^{-1}(0)$, and in local trivializations of $E,F$, the operator $r^m L$ is a sum of terms of the form
  \begin{equation}
  \label{EqBgETerm}
    r^m a_{j\beta}(t,x) D_t^j D_x^\beta = r^{m-j-|\beta|}a_{j\beta}(t,x)(r D_t)^j r^{|\beta|}D_x^\beta,\qquad j+|\beta|\leq m,
  \end{equation}
  where $a_{j\beta}$ is a smooth matrix-valued function. It then remains to note that the vector fields $r\pa_t$, $r\pa_{x^j}$ ($j=1,\ldots,n$) are a local frame for $\Ve(M_\circ)$. Indeed, in the region $x^1\gtrsim|x|$, we can use local coordinates $t$, $x^1$, $\hat x^j=\frac{x^j}{x^1}$ ($j=2,\ldots,n$), and $r$ is a smooth positive multiple of $x^1$. The vector fields $x^1\pa_t$, $x^1\pa_{x^1}$, $x^1\pa_{x^j}$ take the form $x^1\pa_t$, $x^1\pa_{x^1}-\sum_{j=2}^n\hat x^j\pa_{\hat x^j}$, $\pa_{\hat x^j}$; the latter vector fields indeed span the space of edge vector fields in our chart, since the fibers of $\pa M_\circ$ over $x^1=0$ are the level sets of $t$.
\end{proof}

When considering the action of $\upbeta_\circ^*(r^m L)$ on conormal distributions at $\pa M_\circ$, it is more appropriate to regard this operator as a b-differential operator. We proceed to compute its normal operator at $\pa M_\circ$. If $\pi\colon N\cC\to\cC$ is the base projection and $o\subset N\cC$ is the zero section, write\footnote{We remark that $[N\cC;o]$ is naturally diffeomorphic to ${}^+N(\pa M_\circ)$.}
\[
  {}^\vee\cV_{\bop,\rm I}([N\cC;o]) \subset \Vb([N\cC;o])
\]
for the Lie subalgebra of vertical b-vector fields (i.e.\ they lie in $\ker\pi_*$) which are dilation-invariant in the fibers of $[N\cC;o]$. (Locally identifying $[N\cC;o]=\R_t\times[0,\infty)_r\times\Sph^{n-1}$, these vector fields are $a(t,\frac{x}{|x|})r D_x$ where $a\in\CI(\R_t\times\Sph^{n-1})$; i.e.\ $r D_t$ from~\eqref{EqBgETerm} is absent.) Write ${}^\vee\Diff_{\bop,\rm I}^m([N\cC;o])$ for the corresponding space of $m$-th order differential operators.

\begin{lemma}[b-normal operator]
\label{LemmaBgEb}
  In the notation of Lemma~\usref{LemmaBgE}, the b-normal operator of $\upbeta_\circ^*(r^m L)$ at $\pa M_\circ$, which we shall denote $\wh{L_\eop}(0)$, satisfies
  \[
    \wh{L_\eop}(0) \in {}^\vee\Diff_{\bop,\rm I}^m([N\cC;o];\pi^*E|_\cC,\pi^*F|_\cC).
  \]
  Moreover, the restriction of the \emph{operator} $\wh{L_\eop}(0)$ to a fiber $[N_p\cC;o]$, $p\in\cC$, only depends on the restriction of the \emph{principal symbol} of $L$ to $\ann T_p\cC\subset T_p^*M$.
\end{lemma}
\begin{proof}
  The term~\eqref{EqBgETerm} has vanishing coefficients at $r=0$ as a b-operator unless $j=0$ and $|\beta|=m$; note that the terms with $j=0$ and $|\beta|=m$ only involve differentiation in $(r,\omega)$ (which are also the coordinates on the fibers $N\cC$) but not in $t$. This implies the claim.
\end{proof}

Thus, if we identify a collar neighborhood of $\cC$ with a neighborhood of the zero section in $N\cC$, and identify $E$ and $F$ in such a neighborhood with the pullback along $\pi$ of their restrictions to $\cC$, then
\[
  \upbeta_\circ^*(r^m L) - \hat\chi\wh{L_\eop}(0)\hat\chi \in r\Diffb^m(M_\circ;\upbeta_\circ^*E,\upbeta_\circ^*F),
\]
where $\hat\chi\in\CI(M_\circ)$ is identically $1$ near $\pa M_\circ$ and supported in the collar neighborhood of $\pa M_\circ$.

\begin{example}[Wave operator]
\label{ExBgEb}
  If $\cC=\{x=0\}\subset\R_t\times\R^n_x$, and $L=-D_t^2+\sum_{j=1}^n D_{x^j}^2=-D_t^2+D_r^2-\frac{i(n-1)}{r}D_r+r^{-2}\slDelta$ is the wave operator on Minkowski space, then $\wh{L_\eop}(0)=r^2 D_r^2-i(n-1)r D_r+\slDelta$ is $r^2$ times the Laplacian on $\R^n$, and indeed a smooth (in $t\in\R$) family of dilation-invariant operators on $[\ol{\R^n};\{0\}]$. (Note here that $[N\cC;o]=\R_t\times[\ol{\R^n};\{0\}]$.)
\end{example}

\begin{rmk}[Pullback bundle]
\label{RmkBgEbPullback}
  When computing the form of the operator $\wh{L_\eop}(0)$ in concrete applications, it is useful to note that one may work in bundle splittings of $\pi^*E$ induced not merely by splittings of $E|_\cC$, but of $(\upbeta_\circ^*E)|_{\pa M_\circ}$. This is due to the fact that in terms of the projection $\pi_\ff\colon N\cC\to S N\cC=\pa M_\circ$, we can factor $\pi=\upbeta_\circ\circ\pi_\ff$, and therefore
  \[
    \pi^*E|_\cC = \pi_\ff^*\bigl((\upbeta_\circ^*E)|_{\pa M_\circ}\bigr).
  \]
\end{rmk}

\subsection{Conormality and polyhomogeneity; boundary pairing}
\label{SsBgPhg}

Let $X$ denote a manifold with boundary, and let $\rho\in\CI(X)$ denote a boundary defining function. For $\alpha\in\R$, we then denote by
\[
  \cA^\alpha(X) = \{ u\in\rho^\alpha L^\infty_\loc(X) \colon P u\in\rho^\alpha L^\infty_\loc(X)\ \forall\,P\in\Diffb(X) \}
\]
the space of \emph{conormal functions} with weight $\alpha$. (Crucially, the local uniform boundedness holds \emph{up to} the boundary $\pa X$. The subscript `$\loc$' can be dropped when $X$ is compact.) Its elements are smooth on $X^\circ$, but become singular in a controlled fashion at $\pa X$. A typical element of $\cA^\alpha(X)$ is the function $\rho^\alpha$. Spaces of conormal functions can be defined in a completely analogous manner also on manifolds with corners.

Next, we recall that an \emph{index set} is a subset $\cE\subset\C\times\N_0$ so that $(z,k)\in\cE$ implies $(z+1,k)\in\cE$ and also $(z,k-1)\in\cE$ when $k\geq 1$, and so that for all $C\in\R$ there only exist finitely many elements $(z,k)\in\cE$ with $\Re z<C$. We use special notation for important examples, namely
\begin{equation}
\label{EqBgPhgIndexSets}
  (z,k) := \{(z+j,l)\colon j\in\N_0,\ l\leq k\}, \qquad
  (z,k)_+ := \{(z+j,l)\colon j\in\N_0,\ l\leq k+j\}.
\end{equation}
We moreover write
\[
  (z,*)
\]
for an index set which is contained in $(z+\N_0)\times\N_0$, but which we otherwise do not specify explicitly. We write $\Re\cE>\alpha$ if $\Re z>\alpha$ for all $(z,k)\in\cE$. Given index sets $\cE$ and $\cF$, we write $\cE+\cF=\{(z+w,k+l)\colon(z,k)\in\cE,\ (w,l)\in\cF\}$. In the special case $\cF=\cE$, we write $\cE+\cE=:2\cE$, and inductively $j\cE:=(j-1)\cE+\cE$. If moreover $\Re\cE>0$, then we define the \emph{nonlinear closure} of $\cE$ by
\begin{equation}
\label{EqBgPhgIndexSets2}
  \cE_\times := \bigcup_{j\in\N} j\cE;
\end{equation}
it is the smallest index set which contains $\cE$ and for every finite collection $(z_1,k_1)$, $\ldots$, $(z_N,k_N)$ also contains $(\sum_{i=1}^N z_i,\sum_{i=1}^N k_i)$. Note that for $z,k\in\N_0$, we have $(z,k)_{+\times}=((z,k)_+)_\times=(z,k)_+$ when $k\leq z$; otherwise $(z,k)_{+\times}\supsetneq(z,k)_+$. 

On a manifold $X$ with boundary, and with $\rho\in\CI(X)$ denoting a boundary defining function which moreover satisfies $0\leq\rho<\frac12$, we define
\[
  \cA_\phg^\cE(X)
\]
to consist of all smooth functions $u$ on $X^\circ$ which in a collar neighborhood $[0,\infty)_\rho\times\pa X$ of $\pa X$ are asymptotic sums
\[
  u(\rho,y) \sim \sum_{(z,k)\in\cE} \rho^z(\log\rho)^k u_{(z,k)}(y),\qquad u_{(z,k)}\in\CI(\pa X),
\]
meaning that for all $C\in\R$, the difference of $u$ and the finite sum obtained by restricting to $(z,k)\in\cE$ with $\Re z\leq C$ lies in $\cA^C([0,\infty)\times\pa X)$. (This space is independent of the choice of $\rho$ and the collar neighborhood, and it is a module over $\CI(X)$.) When $\cE$ is nonlinearly closed, then $\cA_\phg^\cE(X)$ is an algebra under pointwise multiplication.

Consider next a manifold $M$ with corners, for concreteness $M=[0,\frac12)_{\rho_1}\times[0,\frac12)_{\rho_2}\times X$. Given two index sets $\cE_1,\cE_2\subset\C\times\N_0$, we then define $\cA_\phg^{\cE_1,\cE_2}(M)$ to consist of smooth functions $u$ on $M^\circ$ which are polyhomogeneous at $H_j=\rho_j^{-1}(0)$ with index set $\cE_j$. That is, at $H_1$, the function $u$ is an asymptotic sum
\[
  u(\rho_1,\rho_2,y) \sim \sum_{(z,k)\in\cE_1} \rho_1^z(\log\rho_1)^k u_{(z,k)}(\rho_2,y),\quad \rho_1\searrow 0,\qquad u_{(z,k)}\in\cA_\phg^{\cE_2}([0,\tfrac12)\times X),
\]
which now means that the difference of $u$ and the truncation of the sum to $\Re z\leq C$ lies in $\cA^{C,\alpha_2}(M)$ for all $C$, where $\alpha_2\in\R$ is any $C$-independent constant for which $\Re\cE_2>C$; and an analogous expansion holds at the other boundary hypersurface $H_2$. In this manner, elements of $\cA_\phg^{\cE_1,\cE_2}$ have \emph{joint asymptotic expansions}, or \emph{full compound asymptotics}, into terms $\rho_1^z\rho_2^w(\log\rho_1)^k(\log\rho_2)^l$ where $(z,k)\in\cE_1$ and $(w,l)\in\cE_2$. See \cite{MelroseDiffOnMwc}, \cite[\S2A]{MazzeoEdge}, and \cite{MelrosePushfwd}.

Computations of kernels and cokernels of b-operators on a manifold with boundary often involve \emph{boundary pairings}; the following result will be used frequently for this purpose in~\S\ref{SsAh0}.

\begin{lemma}[Boundary pairing computation]
\label{LemmaBgBdyPair}
  Let $X$ be a manifold with compact boundary $\pa X$; denote by $x\in\CI(X)$ a boundary defining function, and identify a collar neighborhood of $\pa X\subset X$ with a neighborhood of $\{0\}\times\pa X$ inside of $[0,\infty)_x\times\pa X$. Let $\chi\in\CI([0,\infty))$ be equal to $0$ on $[0,1]$ and equal to $1$ near $\infty$, and set $\chi_\eps(x)=\chi(\frac{x}{\eps})$ for $0<\eps\ll 1$. Fix a smooth density on $X^\circ$ which near $\pa X$ is equal to $\mu=x^{-w}|\frac{\dd x}{x}\nu(x)|$ where $w\in\R$, and $0<\nu$ is a smooth density on $\pa X$ which depends smoothly on $x$. Fix $\mu_\pa=x^{-w}|\frac{\dd x}{x}\nu(0)|$ as the density on $[0,\infty)\times\pa X$. Let $L\in x^{-\alpha}\Diffb^m(X)$, write $N(L)=x^{-\alpha}N(x^\alpha L)\in\Diff_{\bop,\rm I}^m([0,\infty)\times\pa X)$, and denote the indicial family of the normal operator of $x^\alpha L$ by $N(x^\alpha L,\lambda)\in\Diff^m(\pa X)$, $\lambda\in\C$. Let $u,v^*$ be such that $L u=0$, $L^*v^*=0$, and suppose that
  \begin{align*}
    u &\in \cA_\phg^{(z,k)}(X) + \cA^{\Re z+\delta}(X), \\
    v^* &\in x^{-\bar z+\alpha+w}\CI(X) + \cA^{-\Re z+\alpha+w+\delta}(X),
  \end{align*}
  where $z\in\C$, $k\in\N_0$, and $\delta>0$. Write the leading order terms of $u$, resp.\ $v^*$ as
  \[
    u_\pa = \sum_{j=0}^k \frac{1}{j!}x^z(\log x)^j u_j,\qquad
    v^*_\pa := x^{-z+\alpha+w}v^*_0,
  \]
  respectively, where $u_j,v^*_0\in\CI(\pa X)$. Set $\tilde u_\pa(\lambda)=\sum_{j=0}^k (\lambda-z)^{-j-1}u_j$. Then
  \begin{equation}
  \label{EqBgBdyPair}
  \begin{split}
    \lim_{\eps\searrow 0} \la [L,\chi_\eps]u, v^* \ra_{L^2(X,\mu)} &= \la [N(L),\chi]u_\pa, v^*_\pa \ra_{L^2([0,\infty)\times\pa X,\mu_\pa)} \\
      &= \big\la \bigl(N(x^\alpha L,\lambda)\tilde u_\pa(\lambda)\bigr)|_{\lambda=z}, v^*_0 \big\ra_{L^2(\pa X,\nu(0))} \\
      \Bigl(\ \text{in the case $k=0$:}\ &= \la \pa_\lambda N(x^\alpha L,z)u_0,v_0^*\ra_{L^2(\pa X,\nu(0))}\ \Bigr).
  \end{split}
  \end{equation}
  The same holds true \emph{mutatis mutandis} when $L$ acts between sections of vector bundles $E,F$ which are equipped with nondegenerate (but not necessarily positive definite) fiber inner products.
\end{lemma}
\begin{proof}
  Replacing $L$ and $\mu$ by $x^{-w}L$ and $x^w\mu$ gives the same setup, but with $w=0$ and a different value for $\alpha$. Similarly, replacing $L$, $v^*$ by $x^\alpha L$, $x^{-\alpha}v^*$ gives the same setup, but now also with $\alpha=0$. Finally, we may replace $L$ by $x^{-z}L x^z$ and $u$, $v^*$ by $x^{-z}u$, $x^{\bar z}v^*$. Altogether, we may thus assume that $\alpha=w=z=0$.

  Note that $N(L)u_\pa=0$ and $N(L^*)v^*_\pa=0$. To prove the first equality in~\eqref{EqBgBdyPair}, note that replacing $L$, $u$, $v^*$ by their respective leading order terms $N(L)$, $u_\pa$, $v^*_\pa$ produces vanishing errors in the limit $\eps\searrow 0$. The second expression in~\eqref{EqBgBdyPair} on the other hand is unchanged if we pass from $\chi$ to another cutoff $\tilde\chi\in\CI([0,\infty))$ which is $0$ near $1$ and $1$ near $\infty$, for integration by parts in $\la [N(L),\chi-\tilde\chi]u_\pa,v_\pa^*\ra=\la N(L)((\chi-\tilde\chi)u_\pa),v_\pa^*\ra$ does not produce any boundary terms since $\chi-\tilde\chi\in\CIc((0,\infty))$.

  In order to prove the second equality in~\eqref{EqBgBdyPair}, note that
  \[
    u_\pa(x)=\Res_{\lambda=0}\bigl(x^\lambda\tilde u_\pa(\lambda)\bigr)=\frac{1}{2\pi i}\oint_0 x^\lambda\tilde u_\pa(\lambda)\,\dd\lambda,
  \]
  where we integrate along a small circle around $0$. We conclude that
  \[
    0 = N(L)u_\pa = \frac{1}{2\pi i}\oint_0 x^\lambda N(L,\lambda)\tilde u_\pa(\lambda)\,\dd\lambda,
  \]
  i.e.\ $N(L,\lambda)\tilde u_\pa(\lambda)$ is holomorphic. (Conversely, the holomorphicity of $N(L,\lambda)\tilde u_\pa(\lambda)$ implies $N(L)u_\pa=0$.) Therefore, the second line of~\eqref{EqBgBdyPair} is well-defined. Writing
  \[
    N(L) = \sum_{j=0}^m L_j(x\pa_x)^j,\qquad L_j\in\Diff^{m-j}(\pa X),
  \]
  we have $N(L,\lambda)=\sum_{j=0}^m L_j\lambda^j$, and we can then compute
  \begin{align*}
    &\la [N(L),\chi]u_\pa,v_\pa^* \ra \\
    &\quad = \frac{1}{2\pi i}\oint_0 \int_0^\infty \big\la [N(L),\chi](x^\lambda\tilde u_\pa(\lambda)), v_0^*\big\ra_{L^2(\pa X)}\,\frac{\dd x}{x}\,\dd\lambda \\
    &\quad = \frac{1}{2\pi i}\oint_0 \int_0^\infty \sum_{j=0}^m \bigg\la L_j \sum_{i=0}^{j-1} (x\pa_x)^i x\chi'(x) (x\pa_x)^{j-i-1} (x^\lambda\tilde u_\pa(\lambda)), v_0^* \bigg\ra_{L^2(\pa X)}\,\frac{\dd x}{x}\,\dd\lambda.
  \end{align*}
  Integration by parts of $(x\pa_x)^i$ produces $0$ (since $x\pa_x v_0^*=0$) unless $i=0$, so this is further equal to
  \begin{align*}
    &\frac{1}{2\pi i}\oint_0 \int_0^\infty x\chi'(x) \sum_{j=0}^m \Big\la L_j (x\pa_x)^{j-1}(x^\lambda\tilde u_\pa(\lambda)),v_0^*\Big\ra_{L^2(\pa X)}\,\frac{\dd x}{x}\,\dd\lambda \\
    &\quad = \int_0^\infty \frac{1}{2\pi i} \oint_0 x\chi'(x) \big\la x^\lambda \lambda^{-1}N(L,\lambda)\tilde u_\pa(\lambda), v_0^*\big\ra_{L^2(\pa X)}\,\dd\lambda\,\frac{\dd x}{x}.
  \end{align*}
  But $N(L,\lambda)\tilde u_\pa(\lambda)=f_0+\cO(\lambda)$ is holomorphic, so
  \[
    \frac{1}{2\pi i}\oint_0 x^\lambda\lambda^{-1}N(L,\lambda)\tilde u_\pa(\lambda)\,\dd\lambda = f_0 = \bigl(N(L,\lambda)\tilde u_\pa(\lambda)\bigr)|_{\lambda=0}.
  \]
  Using $\int_0^\infty x\chi'(x)\,\frac{\dd x}{x}=1$ finally proves~\eqref{EqBgBdyPair}.
\end{proof}

\section{Structure and geometry of the total gluing spacetime}
\label{SG}

Denote by $M$ an open $(n+1)$-dimensional manifold. Denote by $c\colon\R\to M$ an embedding whose image $\cC=c(\R)\subset M$ is a closed 1-dimensional submanifold. (An example to keep in mind is $M=(-1,1)\times\R^n$ and $\cC=\{(t,x_0)\colon t\in(-1,1)\}$.) We shall construct singular deformations, depending on a small parameter $\eps>0$, of a Lorentzian metric on $M$ by working on a resolution of an $(n+2)$-dimensional space $\wt M'$ which fibers over $[0,1)_\eps$ with typical fiber $M$. We immediately fix a trivialization
\[
  \wt M' = [0,1)_\eps \times M.
\]

\begin{definition}[Total gluing spacetime; tangent bundle]
\label{DefGTot}
  The \emph{total gluing spacetime} for $(M,\cC)$ is the resolution
  \begin{equation}
  \label{EqGTot}
    \wt M := [\wt M';\{0\}\times\cC] = \bigl[\,[0,1)\times M; \{0\}\times\cC\,\bigr]
  \end{equation}
  of $\wt M'$. The blow-down map is denoted $\wt\upbeta\colon\wt M\to\wt M'$. We denote by
  \[
    M_\circ=\wt\upbeta^*\bigl(\eps^{-1}(0)\bigr),
    \qquad
    \hat M=\wt\upbeta^*(\{0\}\times\cC)
  \]
  the lift of $\{0\}\times M$ and the front face, respectively. The restrictions of $\wt\upbeta$ to $M_\circ$ and $\hat M$ are denoted $\upbeta_\circ\colon M_\circ\to M$ and $\hat\upbeta\colon\hat M\to\cC$, respectively. The fiber of $\wt M$ over $\eps\in(0,1)$ is denoted $\wt M_\eps$.\footnote{The preimage of $0$ under $\eps\colon\wt M\to[0,1)$ is the union $M_\circ\cup\hat M$.} Moreover, the fiber of $\hat M$ over a point $p\in\cC$ is denoted
  \[
    \hat M_p:=\wt\upbeta^*(\{0\}\times\{p\}) \subset \hat M.
  \]
  We denote by $\wt T\wt M'\to\wt M'$ the vertical tangent bundle, i.e.\ the bundle of tangent vectors which are tangent to the fibers of $\wt M'\to[0,1)$, and by $\wt T\wt M\to\wt M$ the pullback of $\wt T\wt M'\to\wt M'$ along $\wt\upbeta$. Finally, we write $\wt\cV(\wt M):=\CI(\wt M;\wt T\wt M)$.
\end{definition}

Thus, $M_\circ=[M;\cC]$ is a manifold with boundary, and $\hat M=\ol{N}\cC$ is the radially compactified normal bundle of $\cC$, which is a bundle of closed $n$-balls over $\cC$. (See the last paragraph before~\S\ref{SsBgLie}.) See Figure~\ref{FigGTot}. Directly from the definition, we have $\wt T_{M_\circ}\wt M=\upbeta_\circ^*(T M_\circ)$. Definition~\ref{DefGTot} is completely analogous to \cite[Definition~3.1]{HintzGlueID}; however, the tangent bundle $\wt T\wt M$ has different features over the front face $\hat M$ due to the 1-dimensional nature of the submanifold being blown up in~\eqref{EqGTot}, as we will see in~\S\ref{SsGff}.

\begin{figure}[!ht]
\centering
\includegraphics{FigGTot}
\caption{\textit{On the left:} the total space $\wt M$ and its boundary hypersurfaces $\hat M$ and $M_\circ$. \textit{On the right:} the product space $\wt M'=[0,1)\times M$ and the inextendible curve $\cC\subset M\cong\{0\}\times M\subset\wt M'$. Also indicated are the blow-down map $\wt\upbeta$ and its restriction $\upbeta_\circ$ to $M_\circ$, as well as a fiber $\hat M_p$ of $\hat M$ (on the left) over the base point $p\in\cC$ (on the right). The blow-down map $\hat\upbeta\colon\hat M\to\cC$ is not shown here.}
\label{FigGTot}
\end{figure}

We write
\begin{equation}
\label{EqGBdfs}
  \hat\rho\in\CI(\wt M),\qquad \rho_\circ\in\CI(\wt M)
\end{equation}
for defining functions of $\hat M$ and $M_\circ$, respectively; we shall also use this notation for local defining functions (i.e.\ defining functions of $\hat M\cap\cU$ and $M_\circ\cap\cU$ defined over an open subset $\cU\subset\wt M$ depending on the context). For local coordinate computations near $\cC$ on $M$, we shall use
\begin{subequations}
\begin{equation}
\label{EqGLocCoord}
  (t,x),\qquad t\in\R,\quad x\in\R^n,
\end{equation}
with $c(t)=(t,0)$. These coordinates are valid for $|x|<r_0(t)$ where $0<r_0\in\CI(\R)$. Local coordinates near the interior $\hat M^\circ$ of $\hat M$ are then
\begin{equation}
\label{EqGLocCoordHat}
  (\eps,t,\hat x),\qquad \hat x:=\frac{x}{\eps}.
\end{equation}
Near the corner $\hat M\cap M_\circ$, we can use
\begin{equation}
\label{EqGLocCoordCorner}
  (t,\hat\rho,\rho_\circ,\omega),\qquad \hat\rho=|x|\in[0,r_0(t)),\quad \rho_\circ=\frac{\eps}{|x|},\quad \omega:=\frac{x}{|x|}\in\Sph^{n-1}.
\end{equation}
Projective coordinates are computationally more convenient at times; if we write $x=(x^1,x')$, then in the region where $x^1\gtrsim|x'|$, we may use
\begin{equation}
\label{EqGLocCoordProj}
  (t,\hat\rho,\rho_\circ,\hat x'),\qquad \hat\rho=x^1,\quad\rho_\circ=\frac{\eps}{x^1},\quad \hat x'=\frac{x'}{x^1}\in\R^{n-1}.
\end{equation}
\end{subequations}
Examples of local defining functions near $\hat M$ are
\[
  \hat\rho = (\eps^2+|x|^2)^{1/2} = \eps\la\hat x\ra,\qquad
  \rho_\circ = \frac{\eps}{(\eps^2+|x|^2)^{1/2}} = \la\hat x\ra^{-1}.
\]

We record the following analogue of \cite[Lemma~3.4]{HintzGlueID}:

\begin{lemma}[Relationships between parameterized spaces]
\label{LemmaGRel}
  The identity map $\wt M'\to\wt M'$ lifts to a diffeomorphism
  \[
    [\wt M;[0,1)_\eps\times\cC] \xra{\cong} \bigl[ [0,1)_\eps\times M_\circ; \{0\}\times\pa M_\circ \bigr].
  \]
  Using the above coordinates $(t,x)$ near $\{0\}\times\cC$, the identity map $\wt M'\to\wt M'$ also lifts to a diffeomorphism
  \begin{equation}
  \label{EqGRelHat}
    \wt M\cap\{|x|<r_0(t)\} \xra{\cong} \bigl[ [0,1)_\eps\times\hat M; \{0\}\times\pa\hat M \bigr] \cap \{ (\eps,t,\hat x) \colon |\hat x|<\eps^{-1}r_0(t) \}.
  \end{equation}
\end{lemma}
\begin{proof}
  In local coordinates, the second diffeomorphism is a smoothly parameterized (by $t\in\R$) version of the second diffeomorphism in \cite[Lemma~3.4]{HintzGlueID}. The first diffeomorphism can be obtained as in the reference upon replacing $X,\wt X,X_\circ,\{\fp\}$ there by $M,\wt M,M_\circ,\cC$, respectively.
\end{proof}

We denote by
\begin{equation}
\label{EqGRelCutoffs}
  \hat\chi,\ \chi_\circ \in \CI(\wt M)
\end{equation}
two cutoff functions, with $\hat\chi$, resp.\ $\chi_\circ$ identically $1$ near, and supported in, a collar neighborhood of $\hat M$, resp.\ $M_\circ$, and indeed so that $|x|<r_0(t)$ on $\supp\hat\chi$.

\subsection{The front face \texorpdfstring{$\hat M$}{of the total spacetime}; families of stationary tensors}
\label{SsGff}

Consider $p\in\cC$ and a point $q\in\hat M_p^\circ=N_p\cC=T_p M/T_p\cC$. Given an element $V\in\wt T_q\wt M=T_p M$, note that we can regard $V\in T_p M$ as a translation-invariant vector field $V'\in\cV(T_p M)$ on $T_p M$ itself by means of the canonical isomorphism $T_z(T_p M)\cong T_p M$ for all $z\in T_p M$. Now, points in $N_p\cC$ (such as $q$) are the same as orbits of the translation action of $T_p\cC$ on $T_p M$; thus, we may restrict $V'$ to $q\subset T_p M$. We have defined an isomorphism
\begin{equation}
\label{EqGffExt}
  \wt T_q\wt M \ni V \mapsto \sfe(V) \in \{\text{constant maps}\ q\to T_q(T_p M)\},\qquad q\in N_p\cC=T_p M/T_p\cC.
\end{equation}
See Figure~\ref{FigGffExt}.

\begin{figure}[!ht]
\centering
\includegraphics{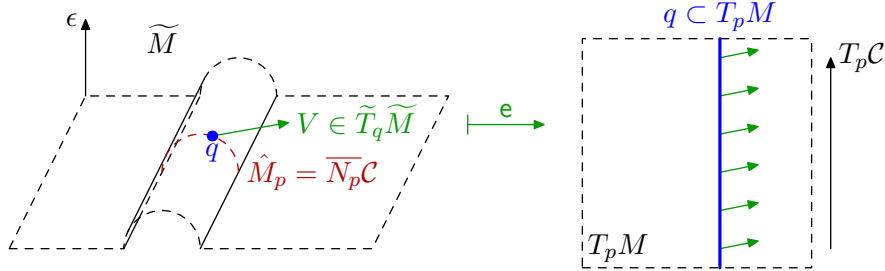}
\caption{Illustration of the map~\eqref{EqGffExt}.}
\label{FigGffExt}
\end{figure}

Given a section $V$ of $\wt T\wt M$ over $\hat M_p^\circ$, we can combine the images of $V|_q$, $q\in\hat M_p^\circ$, into a single element $\sfe(V)\in\CI_{\rm I}(T_p M;T(T_p M))\subset\cV(T_p M)$, where $\CI_{\rm I}(T_p M;T(T_p M))$ is the subspace of sections of $T(T_p M)\to T_p M$ which are constant along the fibers of $T_p M\to N_p\cC$. Equivalently, elements of $\CI_{\rm I}(T_p M;T(T_p M))$ are precisely the stationary ones in that they are invariant (hence the subscript `$\rm I$') under the translation action of $T_p\cC$ on $T_p M$. These maps $V\mapsto\sfe(V)$ in turn can be combined into a single isomorphism\footnote{We write $T^{a,b}M=(\bigotimes^a T M)\otimes(\bigotimes^b T^*M)$, similarly for tensor powers of other (tangent) bundles and their duals.}
\begin{equation}
\label{EqGffStatExt}
  \CI(\hat M^\circ;\wt T_{\hat M}^{a,b}\wt M) \ni V \mapsto \sfe(V) \in \CI_{\rm I}(T_\cC M;{}^\vee T^{a,b}(T_\cC M)),
\end{equation}
initially for $(a,b)=(1,0)$ and then also for general $a,b\in\N_0$. The notation here is as follows: the space on the right is the space of smooth sections which are stationary when restricted to each $T_p M$, $p\in\cC$; and ${}^\vee T(T_\cC M)\to T_\cC M$ denotes the vertical tangent bundle of the bundle $T_\cC M\to\cC$ (which is thus of rank $n+1$); that is, for $z\in T_\cC M$---i.e.\ $z\in T_p M$ where $p\in\cC$---we have $({}^\vee T(T_\cC M))_z=T_z(T_p M)\cong T_p M$.

We make this concrete in local coordinates~\eqref{EqGLocCoord}--\eqref{EqGLocCoordHat}. Write
\begin{equation}
\label{EqGffHatCoords}
  \hat t=\dd t(-)\in\R,\qquad \hat x=\dd x(-)\in\R^n
\end{equation}
for the induced linear coordinates on the fibers of $T_\cC M\to\cC$; also, $\hat t$ is a linear coordinate on the fibers of $T\cC\to\cC$, and $\hat x$ is a linear coordinate system on the fibers of $N\cC\to\cC$ since $T\cC=\R\pa_t$ and $\dd x(\pa_t)=0$. Note then that the point $q=(t,\hat x)\in\hat M_p$, whose base point in $(t,x)$-coordinates is $p=\hat\upbeta(q)=(t,0)$, is the equivalence class of $\hat x\pa_x$ in $T_p M/\R\pa_t$, whose coordinates in $T_p M$ are thus equal to $\hat x$ indeed; this justifies the notation. Moreover, the lifts of $\pa_t,\pa_{x^j}\in T_p M$ to $T(T_p M)$ are $\pa_{\hat t},\pa_{\hat x^j}$. Therefore, the map $\sfe$ on vector fields is
\begin{equation}
\label{EqGffStatExtVF}
  \sfe \colon \Bigl((t,\hat x)\mapsto a(t,\hat x)\pa_z\Bigr) \mapsto \Bigl( (t,\hat x)\mapsto a(t,\hat x)\pa_{\hat z}\Bigr),
\end{equation}
where $z=t$ or $x^j$ and $\hat z=\hat t$ or $\hat x^j$. On symmetric 2-tensors,
\begin{equation}
\label{EqGffStatExtS2}
\begin{split}
  \sfe \colon &\Bigl((t,\hat x)\mapsto g_{0 0}(t,\hat x)\,\dd t^2 + 2 g_{0 j}(t,\hat x)\,\dd t\,\dd x^j + g_{i j}(t,\hat x)\,\dd x^i\,\dd x^j\Bigr) \\
  &\qquad \mapsto \Bigl( (t,\hat x) \mapsto g_{0 0}(t,\hat x)\,\dd\hat t^2 + 2 g_{0 j}(t,\hat x)\,\dd\hat t\,\dd\hat x^j + g_{i j}(t,\hat x)\,\dd\hat x^i\,\dd\hat x^j\Bigr).
\end{split}
\end{equation}
The right hand side is a smooth family, parameterized by $t\in\R$, of stationary symmetric 2-tensors on $\R^{1+n}_{\hat t,\hat x}$.

In order to obtain a uniform description of the isomorphism~\eqref{EqGffStatExt} on $\hat M$, we first introduce:

\begin{definition}[Bundle of stationary spacetimes; tangent bundles]
\label{DefGffBdl}
  Define the fiber bundle
  \[
    \breve T_\cC M = \bigsqcup_{p\in\cC} \{p\} \times \breve T_p M \to\cC,\qquad \breve T_p M:=[\ol{T_p}M;\pa\ol{T_p}\cC],
  \]
  with base $\cC$ and typical fiber $[\ol{\R\times\R^n};\pa\ol\R\times\{0\}]$. We write
  \[
    {}^{\tscop,\vee}T(\breve T_\cC M)\to\breve T_\cC M,\qquad {}^\vee\cV_\tscop(\breve T_\cC M):=\CI(\breve T_\cC M;{}^{\tscop,\vee}T(\breve T_\cC M)),
  \]
  for the vertical 3sc-tangent bundle and the space of its smooth sections: the fiber of ${}^{\tscop,\vee}T(\breve T_\cC M)$ over a point $z\in\breve T_\cC M$ lying over $p\in\cC$ is $\Ttsc(\breve T_p M)$, which is the pullback along $\breve T_p M\to\ol{T_p}M$ of $\Tsc(\ol{T_p}M)$ (which has as a smooth frame the vector fields $\pa_{\hat t}$, $\pa_{\hat x^j}$, $j=1,\ldots,n$, in the coordinates used in~\eqref{EqGffStatExtVF}).
\end{definition}

The manifold interior of $\breve T_p M$ is $T_p M$; and for $z\in T_p M$, we have ${}^{\tscop,\vee}T_z(\breve T_\cC M)={}^\vee T_z(\breve T_\cC M)=T_z(T_p M)\cong T_p M$. This means that the restriction of ${}^{\tscop,\vee}T(\breve T_\cC M)$ to the manifold interior $T_\cC M$ of $\breve T_\cC M$ is equal to the bundle ${}^\vee T(T_\cC M)$ featuring in~\eqref{EqGffStatExt}.

Each fiber of $\breve T_\cC M$ carries a translation action by $T_p\cC$. (The closure of an orbit of this action is either a copy of $\ol\R$ or a single point in $\pa\breve T_p M$.) Moreover, the projection $T_p M\to N_p\cC$ extends to a smooth submersion
\begin{equation}
\label{EqGffBrevePi}
  \breve\pi \colon \breve T_p M\to\ol{N_p}\cC=\hat M_p.
\end{equation}
See Figure~\ref{FigGffBdl}. This is an instance of the following result:

\begin{lemma}[Quotient spaces and compactifications]
\label{LemmaGff3bProj}
  Let $V$ be a finite-dimensional real vector space and $W\subset V$ a linear subspace. Then the projection $V\to V/W$ extends by continuity from the interior to a smooth fibration
  \begin{equation}
  \label{EqGff3bProj}
    [\ol{V};\pa\ol{W}] \to \ol{V/W}.
  \end{equation}
  The preimage of a boundary defining function of $\ol{V/W}$ is a defining function of the lift of $\pa\ol V$.
\end{lemma}
\begin{proof}
  Extending a basis of $W$ to a basis of $V$, we may assume $V=\R^{m+n}_{t,x}$ and $W=\R^m_t\times\{0\}$. We verify the claim only near the codimension 2 corner of $[\ol V;\pa\ol W]$, and leave the remainder of the verification to the reader. It suffices to show that the map~\eqref{EqGff3bProj} is a submersion. Since smooth coordinates on $\ol V$ near $(|t|,x)=(\infty,0)$ are $\frac{t}{|t|}$ (when $m\geq 2$), $\frac{1}{|t|}$, and $\frac{x}{|t|}$, smooth coordinates near the corner of $[\ol V;\pa\ol W]$ are $\rho_\ff=\frac{|x|}{t}$, $\rho_\sface=\frac{1/|t|}{|x|/|t|}=\frac{1}{|x|}$, $\omega_W=\frac{t}{|t|}$ (when $m\geq 2$), and $\omega_V=\frac{x/|t|}{|x|/|t|}=\frac{x}{|x|}$. On the other hand, we can identify $V/W\cong\R^n_x$, with smooth coordinates near $|x|=\infty$ given by $|x|^{-1}$ and $\omega=\frac{x}{|x|}$. The map~\eqref{EqGff3bProj} is thus $(\rho_\ff,\rho_\sface,\omega_W,\omega_V)\mapsto(\rho_\sface,\omega_V)$; this is indeed a smooth submersion.
\end{proof}

\begin{figure}[!ht]
\centering
\includegraphics{FigGffBdl}
\caption{Illustration of a fiber $\breve T_p M$ of $\breve T_\cC M$, of translation orbits of $T_p\cC$ (red, dashed), and of the projection $\breve\pi\colon\breve T_p M\to\ol{N_p}\cC$.}
\label{FigGffBdl}
\end{figure}

\begin{lemma}[Stationary extension]
\label{LemmaGffStatExt}
  The map~\eqref{EqGffStatExt} for $(a,b)=(1,0)$ (i.e.\ on vector fields) is the restriction to $\hat M^\circ$ of the isomorphism
  \begin{equation}
  \label{EqGffCptExt}
    \sfe \colon \CI(\hat M;\wt T_{\hat M}\wt M) \xra{\cong} {}^\vee\cV_{\tscop,\rm I}(\breve T_\cC M) = \CI_{\rm I}(\breve T_\cC M;{}^{\tscop,\vee}T(\breve T_\cC M)),
  \end{equation}
  where the subscript `$\rm I$' denotes invariance under the $T_p\cC$-translation action; analogously for tensors of type $(a,b)$. It induces a short exact sequence
  \[
    0 \to \hat\rho\wt\cV(\wt M) \hra \wt\cV(\wt M) \xra{\sfe} {}^\vee\cV_{\tscop,\rm I}(\breve T_\cC M) \to 0.
  \]
\end{lemma}
\begin{proof}
  This follows from~\eqref{EqGffStatExtVF}, and from Lemma~\ref{LemmaGff3bProj} which shows that the $\hat x$-coordinates, resp.\ inverse polar coordinates $|\hat x|^{-1}$, $\frac{\hat x}{|\hat x|}$ are smooth coordinates on $\breve T_p M$ in $|\hat x|\lesssim 1$, resp.\ $|\hat x|\gtrsim 1$.
\end{proof}

The manifold $\breve T_p M$ on which sections of $\wt T\wt M$ over $\hat M_p$ can be regarded as stationary vector fields can be given an interpretation directly on $\wt M$. For this purpose, we blow up $\hat M_p$ to get
\begin{equation}
\label{EqGffBlowup}
  [\wt M;\hat M_p]=[[0,1)\times M;\{0\}\times\cC;\{0\}\times\{p\}] \cong [[0,1)\times M;\{0\}\times\{p\};\{0\}\times\cC].
\end{equation}
The front face is $[\ol{T_p}M;\pa\ol{T_p}\cC]=\breve T_p M$; see Figure~\ref{FigGffBlowup}. The interior of the front face carries \emph{both} the rescaled spatial variables $\hat x$ \emph{and} the rescaled (`fast') temporal variable $\hat t_p:=\frac{t-t_0}{\eps}$ where $p=c(t_0)$ (and where $\hat t_p$ can be further identified with $\hat t$). In this manner, working near the front face of $[\wt M;\hat M_p]$ enables one to understand how, say, a metric $\wt g\in\CI(\wt M;S^2\wt T^*\wt M)$ deviates from its stationary model $\sfe(\wt g|_{\hat M_p})$ as the parameter $\eps$ increases from $0$ to positive values.

\begin{figure}[!ht]
\centering
\includegraphics{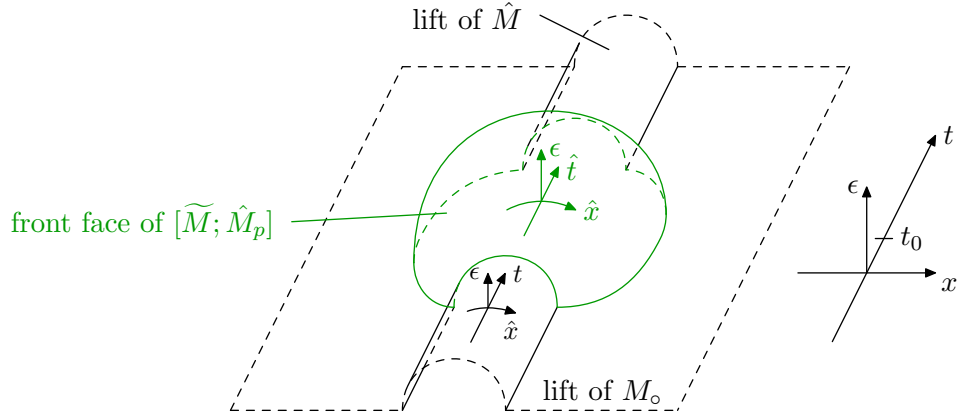}
\caption{Illustration of~\eqref{EqGffBlowup}, together with some local coordinate systems.}
\label{FigGffBlowup}
\end{figure}

\begin{rmk}[Multiplication by $\eps$]
\label{RmkGffMultEps}
  Carefully note that the lift of $V\in\cV(M)$ to an $\eps$-independent vector field on $\wt M'$ and then to a vector field on $\wt M$ is singular at $\hat M$ unless $V|_\cC\in T\cC$, as follows from the presence of the singular factor $\eps^{-1}$ in $\pa_x=\eps^{-1}\pa_{\hat x}$. (More generally, $V\in\wt\cV(\wt M)$, regarded as a smooth vector field on $\{\eps>0\}$, is singular at $\hat M$ unless $V(q)\in\wt\upbeta^*(T_p\cC)$ for all $p\in\cC$, $q\in\hat M_p$.) On the other hand, we saw above that $\sfe(\pa_x)=\pa_{\hat x}$. However, the map $\sfe$ on vector fields $V$ (sections of $\wt T_{\hat M}\wt M$) is not quite given by multiplication by $\eps$ (i.e.\ smooth extension off $\hat M$ as a section of $\wt T\wt M$, multiplication by $\eps$, and restriction back to $\hat M$ as a vector field) since $\eps\pa_t=0$ at $\hat M$; this should be contrasted with \cite[Lemma~3.2, Definition~3.3]{HintzGlueID}. This can be remedied by instead restricting $\eps V$ to the front face $\breve T_p M$ of $[\wt M;\hat M_p]$. In this manner, multiplication by $\eps$ induces an isomorphism between $\wt T_q\wt M$, $q\in\hat M_p$, and translation-invariant 3sc-vector fields on $[\ol{T_p}M;\pa\ol{T_p}\cC]$ defined over the translation orbit $q\subset T_p M$. Since $\breve T_\cC M$ is exactly the bundle of all $\breve T_p M$, we conclude that~\eqref{EqGffCptExt} is given by multiplication by $\eps$ for $(a,b)=(1,0)$, and by multiplication by $\eps^{a-b}$ in general. This now matches \cite[Definition~3.3]{HintzGlueID}.
\end{rmk}

\subsection{Hypersurfaces transversal to \texorpdfstring{$\cC$}{the gluing curve}}
\label{SsGHyp}

The study of evolution equations on $\wt M$ (i.e.\ on $\wt M_\eps$ for all small $\eps>0$ at once) requires the choice of Cauchy hypersurfaces. We recall from \cite[\S3]{HintzGlueID}:

\begin{definition}[Total gluing space for initial data]
\label{DefGHypTot}
  Let $X$ be a smooth open $n$-dimensional manifold, and let $\fp\in X$. Then we set $\wt X'=[0,1)_\eps\times X$ and
  \[
    \wt X = [\wt X';\{(0,\fp)\}],
  \]
  with boundary hypersurfaces denoted $\hat X$ (front face) and $X_\circ$ (lift of $\{0\}\times X$). The blow-down map is $\upbeta_{\wt X}\colon\wt X\to\wt X'$, with restrictions $\upbeta_{\hat X}\colon\hat X\to\{\fp\}$ and $\upbeta_{X_\circ}\colon X_\circ=[X;\{\fp\}]\to X$. The fiber of $\wt X$ over $\eps\in(0,1)$ is denoted $\wt X_\eps$. The bundle $\wt T\wt X'\to\wt X'$ is the vertical tangent bundle, and $\wt T\wt X\to\wt X$ is its pullback along $\upbeta_{\wt X}$; we write $\wt\cV(\wt X)=\CI(\wt X;\wt T\wt X)$.
\end{definition}

If $X\subset M$ is a smooth hypersurface which is transversal to $\cC$ and intersects $\cC$ only once in the point $\fp=c(t_0)\in\cC$, then the inclusion map $[0,1)\times X\hra[0,1)\times M$ lifts to an embedding $\wt X\hra\wt M$ of $\wt X$ as a smooth hypersurface, with $\wt X\cap\hat M=\hat M_\fp$ and $\wt X\cap M_\circ=X_\circ=\upbeta_\circ^*X$. It is important to retain more precise information about $\wt X$ near $\hat M$: to wit, $X$ defines a hypersurface $T_\fp X\subset T_\fp M$.

\begin{rmk}[Geometry of $\wt X\subset\wt M$]
  The lift of $\wt X$ to $[\wt M;\hat M_\fp]$ intersects the front face $\breve T_\fp M$ in a smooth hypersurface, namely the radial compactification $\ol{T_\fp}X$ of $T_\fp X$. See Figure~\ref{FigGHyp}. In this perspective, the lift of $\wt X$ is a Cauchy hypersurface both for wave evolution near $\fp$ in the fast ($T_\fp M$-)time scale and away from $\fp$ in the slow ($M$-)time scale.
  \begin{figure}[!ht]
  \centering
  \includegraphics{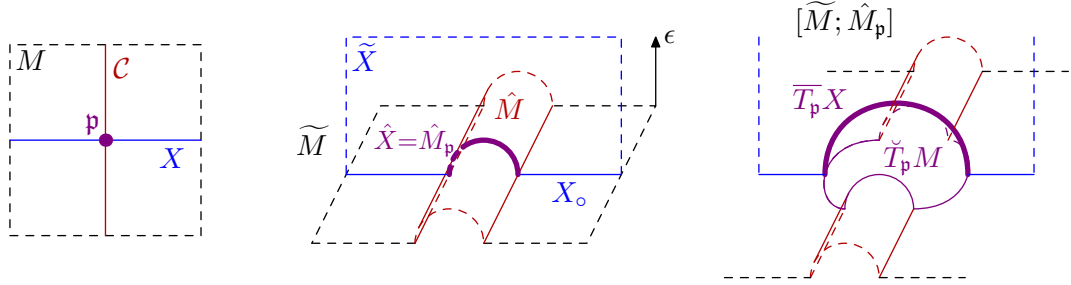}
  \caption{\textit{On the left:} the manifold $M$, the hypersurface $X$, and the curve $\cC$ (transversal to $X$). \textit{In the middle:} the total gluing space for initial data $\wt X$ as a hypersurface inside the total gluing spacetime $\wt M$. \textit{On the right:} lift of $\wt X$ to the blow-up of $\wt M$ at $\hat M_\fp$.}
  \label{FigGHyp}
  \end{figure}
\end{rmk}

We also recall the isomorphism of tensor bundles
\begin{equation}
\label{EqGHypScaling}
  \sfs \colon \wt T_q^{a,b}\wt X \xra{\cong} \Tsc_q^{a,b}\hat X,\qquad q\in\hat X,
\end{equation}
which is defined as multiplication (of a smooth extension) by $\eps^{a-b}$ (followed by restriction) in \cite[Definition~3.3]{HintzGlueID}. More in line with the construction in~\S\ref{SsGff}, we can define the map~\eqref{EqGHypScaling} for $(a,b)=(1,0)$ as follows: lift $V\in\wt T_q\wt X=T_\fp X$, $q\in\hat X=\ol{T_\fp}X$, to a translation-invariant vector field $V'$ on $T_\fp X$ (which is thus a scattering vector field on $\ol{T_\fp}X$) and evaluate this lift at $q$; the map $V\mapsto V'|_q$ thus defined is precisely $\sfs$. In local coordinates $x\in\R^n$ on $X$ and $\hat x=\dd x(-)$ (or equivalently $\hat x=\frac{x}{\eps}$) on $\hat X^\circ$, this map takes $\pa_{x^j}\mapsto\pa_{\hat x^j}$. Since $\wt T\wt X\hra\wt T_{\wt X}\wt M$, we have
\[
  \sfs(V)|_q = \sfe(V)|_{q_0},\qquad V\in\wt T_q\wt X,\ q\in\hat M_\fp^\circ,
\]
where $q_0=q\cap T_\fp X$ (with $q\in\hat M_\fp^\circ=N_\fp\cC=T_\fp M/T_\fp\cC\cong T_\fp X$ identified with the corresponding $T_\fp\cC$-orbit in $T_\fp M$).

\subsection{Vector fields}
\label{SsGVf}

Recall that smooth elements of the space $\wt\cV(\wt M)$ are singular at $\hat M$ when regarded as vector fields on $\wt M$; moreover, $\wt\cV(\wt M)$ is not a Lie algebra, the issue being the irregularity of the coefficients, which lie in $\CI(\wt M)$, with respect to (lifts of) vector fields on $M$ (i.e.\ $\pa_t,\pa_{x^j}$ in local coordinates). The description of differential operators related to geometric structures on $\wt T\wt M$ thus requires the usage of a different class of vector fields. As we shall see in~\S\ref{SsGL}, the appropriate class is the following.

\begin{definition}[se-vector fields]
\label{DefGVf}
  The space $\Vse(\wt M)$ of se-vector fields is defined as
  \[
    \Vse(\wt M) = \{ V\in\Vb(\wt M) \colon V\ \text{is vertical and tangent to the fibers of $\hat M$} \}.
  \]
  Here, $V$ being vertical means that $\dd\eps(V)=0$.
\end{definition}

\begin{rmk}[Terminology]
\label{RmkGVf}
  The total space $\wt M$ is the single surgery space associated with $M$ and $\cC$, as defined in \cite{MazzeoMelroseSurgery}. The Lie algebra of vector fields used in \cite{MazzeoMelroseSurgery} is
  \[
    \Vs(\wt M)=\{V\in\Vb(\wt M)\colon V\ \text{is vertical} \};
  \]
  its elements are \emph{surgery vector fields}. Thus, $\Vse(\wt M)\subset\Vs(\wt M)$ is the subspace (and indeed Lie subalgebra, as we argue below) consisting of those vector fields which are in addition of edge type \cite{MazzeoEdge} at $\hat M$ (which is the total space of a fibration $\ol{\R^n}-\hat M\to\cC$), hence the terminology `surgery-edge', or `se' for short.
\end{rmk}

Elements of $\Vse(\wt M)$ are smooth (in $\eps\in(0,1)$) families of smooth vector fields on $M$ which degenerate in a specific manner as $\eps\searrow 0$. In local coordinates~\eqref{EqGLocCoord}--\eqref{EqGLocCoordHat}, elements of $\Vse(\wt M)$ take the following form: away from $x=0$, they are smooth (in $\eps,t,x$) linear combinations of $\pa_t$, $\pa_{x^j}$, and near $\hat M^\circ$, they are smooth (in $\eps,t,\hat x$) linear combinations of $\eps\pa_t$, $\pa_{\hat x^j}$. Globally, in $(\eps,t,x)$-coordinates, they are smooth (on $\wt M$) linear combinations of $\hat\rho\pa_t$, $\hat\rho\pa_{x^j}$ (see~\eqref{EqGBdfs}). We verify this near the corner $\hat M\cap M_\circ$ using the coordinates~\eqref{EqGLocCoordProj}: the fibers of $\hat M$ are the level sets of $t$ at $\hat\rho=0$, and the claim follows from the fact that $\hat\rho\pa_t$, $\hat\rho\pa_{x^1}=\hat\rho\pa_{\hat\rho}-\rho_\circ\pa_{\rho_\circ}-\hat x'\pa_{\hat x'}$ and $\hat\rho\pa_{x^j}=\pa_{\hat x^{\prime j}}$, $j=2,\ldots,n$, are indeed tangent to $\rho_\circ=0$ and to the fibers of $\hat\rho=0$ (i.e.\ to the $t$-level sets), and linearly independent (as se-vector fields). We deduce in particular that
\begin{equation}
\label{EqGVfTilde}
  \Vse(\wt M) = \hat\rho\wt\cV(\wt M).
\end{equation}
Since $\wt\cV(\wt M)$ is spanned over $\CI(\wt M)$ by lifts of smooth sections of $\wt T\wt M'\to\wt M'$, i.e.\ by smooth families (in $\eps\in[0,1)$) of smooth vector fields on $M$, we infer that
\begin{equation}
\label{EqGVfMapPullback}
  \Vse(\wt M) \ni V \colon \wt\upbeta^*\CI(\wt M') \to \hat\rho\CI(\wt M).
\end{equation}
In fact, this mapping property characterizes se-vector fields in the space of vertical b-vector fields on $\wt M$. The space $\Vse(\wt M)$ is the space of smooth sections of the se-tangent bundle
\[
  \Tse\wt M \to \wt M,
\]
local frames of which are the explicit generators above. (By~\eqref{EqGVfTilde}, we have $\Tse\wt M=\hat\rho\wt T\wt M$.)

Since $\Vb(\wt M)$ is a Lie algebra, and since tangency to submanifolds is preserved under vector field commutators, we deduce that $\Vse(\wt M)$ is a Lie algebra. The corresponding space of $m$-th order differential operators is denoted $\Diffse^m(\wt M)$.

An element $P\in\Diffse^m(\wt M)$ has two normal operators which describe their leading order behavior at $\hat M$ and $M_\circ$.\footnote{A third `normal operator' is the principal symbol of $P$, which captures the leading order behavior at high frequencies. This plays no role in the present paper however.} Near $\hat M^\circ$, we first consider the case of se-vector fields. Combining~\eqref{EqGVfTilde} with~\eqref{EqGffCptExt}, we map $V\in\Vse(\wt M)$ into $\sfe(\eps^{-1}V)$; since $\eps^{-1}V\in\eps^{-1}\Vse(\wt M)=\rho_\circ^{-1}\wt\cV(\wt M)$, with $\rho_\circ$ a boundary defining function of each fiber $\ol{N_p}\cC$ of $\hat M$, the final part of Lemma~\ref{LemmaGff3bProj} gives
\[
  \sfe(\eps^{-1}V) \in \rho_\circ^{-1}{}^\vee\cV_{\tscop,\rm I}(\breve T_\cC M) = {}^\vee\cV_{\tbop,\rm I}(\breve T_\cC M);
\]
this is a smooth family of translation-invariant 3b-vector fields on the fibers $\breve T_p M$ of $\breve T_\cC M$. Upon setting $N_{\hat M}(V):=\sfe(\eps^{-1}V)$, we thus obtain a short exact sequence
\begin{equation}
\label{EqGVfSEShatM}
  0 \to \hat\rho\Vse(\wt M) \hra \Vse(\wt M) \xra{N_{\hat M}} {}^\vee\cV_{\tbop,\rm I}(\breve T_\cC M) \to 0.
\end{equation}

The normal operator
\[
  N_{M_\circ}(P)\in\Diffe^m(M_\circ)
\]
at $M_\circ$ is given by restriction as in the case of b-normal operators: $N_{M_\circ}(P)u=(P\tilde u)|_{M_\circ}$ where $\tilde u\in\CI(\wt M)$ is an (arbitrary) extension of $u\in\CI(M_\circ)$. The map $N_{M_\circ}$ is the multiplicative extension of the third arrow in
\[
  0\to\rho_\circ\Vse(\wt M)\hra\Vse(\wt M)\to\Ve(M_\circ)\to 0,
\]
where $\Ve(M_\circ)$ is the space of edge vector fields on $M_\circ=[M;\cC]$, i.e.\ those b-vector fields which are tangent to the fibers of the restriction of the blow-down map $\upbeta_\circ\colon M_\circ\to M$ to $\pa M_\circ$; this short exact sequence can be checked using the local coordinate descriptions above. (We remark that the fibers of $\upbeta_\circ|_{\pa M_\circ}$ are precisely the intersections of the fibers of $\hat M$ with $M_\circ$.) To summarize:

\begin{definition}[Normal operators of se-differential operators]
\label{DefGVfNorm}
  The multiplicative extension of~\eqref{EqGVfSEShatM}, resp.\ the restriction to $M_\circ$, gives rise to a surjective homomorphism
  \[
    N_{\hat M} \colon \Diffse^m(\wt M) \to {}^\vee\Diff_{\tbop,\rm I}^m(\breve T_\cC M),\qquad\text{resp.}\qquad
    N_{M_\circ} \colon \Diffse^m(\wt M) \to \Diffe^m(M_\circ),
  \]
  with kernel $\hat\rho\Diffse^m(\wt M)$, resp.\ $\rho_\circ\Diffse^m(\wt M)$. For $p\in\cC$, we write $N_{\hat M_p}\colon\Diffse^m(\wt M)\to\Diff_{\tbop,\rm I}^m(\breve T_p M)$ for the restriction of $N_{\hat M}$ to the fiber over $p$. (Here, the subscript `$I$' restricts to the space of operators which are invariant under $T\cC$-translations.)
\end{definition}

\begin{rmk}[Normal operator at $\hat M$ and restriction]
\label{RmkGVfNormHatRestr}
  In the context of~\eqref{RmkGffMultEps}, we can obtain $N_{\hat M_p}(V)$ as the restriction of $V\in\Vse(\wt M)$ to the front face $\breve T_p M\subset[\wt M;\hat M_p]$. The $\hat M_p$-normal operator of $P\in\Diffse^m(\wt M)$ is thus equal to the restriction of the lift of $P$ to $[\wt M;\hat M_p]$ to the front face $\breve T_p M$. (This can be seen explicitly from the local coordinate computations in~\eqref{EqGVfNhatMCoord}.)
\end{rmk}

In view of the translation-invariant nature of $N_{\hat M}$ on each fiber $\hat M_p$, we may pass to the Fourier transform along the fibers of the $T_p\cC$-action on $\breve T_p M$. To do this, we first fix $t\in\CI(\wt M)$ so that $\dd t\neq 0$ on $T\cC$. The only information about $t$ we need in the sequel is its differential $\dd t|_{T_\cC M}\in\CI(\cC;T_\cC^* M)$. We may then parameterize the orbits of the $T_p\cC$-translation action in the interior $T_p M$ of $\breve T_p M$ using the function
\[
  \hat t_p:=\dd t|_{T_p M}(-).
\]
Such a choice of `time function' also allows us to identify $N_p\cC$ with the transversal $\hat t_p^{-1}(0)=\ker\dd t\subset T_p M$ to the $T_p\cC$-action. Therefore, we may define the spectral family of $P\in\Diffse^m(\wt M)$ at $\hat M$ by
\begin{equation}
\label{EqGVfNhatM}
\begin{split}
  &\bigl(N_{\hat M}(P,\hat\sigma)u\bigr)(p,\hat x) := \bigl(e^{i\hat t_p\hat\sigma}N_{\hat M_p}(P)(e^{-i\hat t_p\hat\sigma}u)\bigr)(\hat x), \\
  &\hspace{6em} u\in\CI(\hat M^\circ),\ p\in\cC,\ \hat x\in N_p\cC\cong\hat t_p^{-1}(0),\ \hat\sigma\in\R.
\end{split}
\end{equation}
We shall also write
\[
  N_{\hat M_p}(P,\hat\sigma)=N_{\hat M}(P,\hat\sigma)(p,-).
\]
Note that for $P\in\Diffse^m(\wt M)$, this is a polynomial of degree $m$ in $\hat\sigma$. Conversely, we have
\begin{equation}
\label{EqGVfhatMSpec}
  N_{\hat M_p}(P) = \sum_{j=0}^m \frac{1}{j!}\pa_{\hat\sigma}^j N_{\hat M_p}(P,0)(-D_{\hat t_p})^j,
\end{equation}
since the spectral families of both sides are equal to $N_{\hat M_p}(P,\hat\sigma)$.

\begin{example}[Explicit computations]
\label{ExGVfExplicit}
  To make these constructions concrete, consider the se-vector fields $\hat\rho\pa_t$ and $\hat\rho\pa_{x^j}$ where $\hat\rho=(\eps^2+|x|^2)^{1/2}$. Put $\rho_\circ=\eps\hat\rho^{-1}=\la\hat x\ra^{-1}$. Then
  \[
    N_{M_\circ}(\hat\rho\pa_t)=|x|\pa_t,\qquad
    N_{M_\circ}(\hat\rho\pa_{x^j})=|x|\pa_{x^j};
  \]
  these are (a spanning set over $\CI(M_\circ)$ of the space of) edge vector fields on $M_\circ$. Using the coordinates $\hat t=\dd t(-)$, $\hat x^j=\dd x^j(-)$ on $T_\cC M$, we moreover have
  \begin{equation}
  \label{EqGVfNhatMCoord}
    N_{\hat M}(\hat\rho\pa_t) = \sfe(\rho_\circ^{-1}\pa_t) = \la\hat x\ra\pa_{\hat t},\qquad
    N_{\hat M}(\hat\rho\pa_{x^j}) = \la\hat x\ra\pa_{\hat x^j},
  \end{equation}
  which one should regard as smooth families (in $t$) of stationary vector fields on $\R_{\hat t}\times\R^n_{\hat x}$; these are 3b-vector fields on $[\ol{\R\times\R^n};\pa\ol\R\times\{0\}]$. Finally, for $\hat\sigma\in\R$,
  \[
    N_{\hat M}(\hat\rho\pa_t,\hat\sigma) = -i\hat\sigma\la\hat x\ra,\qquad
    N_{\hat M}(\hat\rho\pa_{x^j},\hat\sigma) = \la\hat x\ra\pa_{\hat x^j}.
  \]
  In the presence of smooth coefficients, we have, for example,
  \[
    N_{\hat M}(a(\eps,t,\hat x)\hat\rho\pa_t,\hat\sigma)=-i a(0,t,\hat x)\hat\sigma\la\hat x\ra;
  \]
  note the distinction of the slow ($t$) and fast ($\hat t$ and $\hat\sigma$) time (and frequency) scales.
\end{example}

Regarding $N_{\hat M}(P)$ as an operator on the compactification $\ol N\cC=\hat M$, the discussion of the structure of the spectral family of a 3b-differential operator in \cite[\S4.1]{Hintz3b} (or the explicit computations in Example~\ref{ExGVfExplicit}), applied here with smooth parametric dependence on $p\in\cC$, implies
\begin{align}
\label{EqGVfhatMSpec0}
  \wh P(0) := N_{\hat M}(P,0) &\in {}^\vee\Diffb^m(\hat M), \\
\label{EqGVfhatMSpecn0}
  \wh P(\hat\sigma) := N_{\hat M}(P,\hat\sigma) &\in \rho_\circ^{-m}\,{}^\vee\Diffsc^m(\hat M),\qquad \hat\sigma\neq 0, \\
\label{EqGVfhatMSpecn0Diff}
  \pa_{\hat\sigma}^j\wh P(0) &\in \rho_\circ^{-j}\,{}^\vee\Diffb^{m-j}(\hat M).
\end{align}
That is, these operators are b-, resp.\ weighted scattering operators on each fiber $\hat M_p$, $p\in\cC$. We finally note that the zero energy operator family $N_{\hat M}(P,0)$ is equal to the b-normal operator
\begin{equation}
\label{EqGVfhatMbNorm}
  \wh P(0) = N_{\hat M}(P,0) = N_{\hat M}(P)
\end{equation}
of $P$ (regarded as a b-differential operator $P\in\Diffb^m(\wt M)$, using that $\Vse(\wt M)\subset\Vb(\wt M)$) at $\hat M$; this can be seen from~\eqref{EqGVfNhatM} as a consequence of the fact that the lift of $\CI(\hat M^\circ)$ to $\breve T\cC$ is precisely the space of smooth translation-invariant functions on $(\breve T_\cC M)^\circ=T_\cC M$. In particular, $\wh P(0)$ is independent of the choice of $t$. The operators~\eqref{EqGVfhatMSpecn0} on the other hand do depend on the choice of $t$.

\subsection{Lorentzian metrics}
\label{SsGL}

We now assume that $M$ is equipped with a smooth Lorentzian metric $g\in\CI(M;S^2 T^*M)$ of signature $(-,+,\ldots,+)$; moreover, we shall assume that
\[
  \cC\ \text{is timelike}.
\]
Still requiring $\cC$ to be closed and the image of an embedding of $\R$, we now denote by $c\colon I\to M$ an arc-length parameterization of $\cC$, defined on some maximal interval $I\subset\R$ containing $0$. Moreover, we assume that $M$ is time-oriented and $c'$ is future timelike. The existence of \emph{Fermi normal coordinates} around $\cC$ is standard; we include a proof for completeness:

\begin{lemma}[Fermi normal coordinates]
\label{LemmaGLFermi}
  There exists a smooth coordinate system $(t,x)$ in a neighborhood of $\cC$ so that $c(t)=(t,0)$, the curves $s\mapsto(t_0,s x_0)$ are geodesics for all $(t_0,x_0)$, and
  \begin{equation}
  \label{EqGLFermi}
    g|_{(t,x)} = \bigl(-1-2\Gamma_{j 0 0}(t,0)x^j\bigr)\,\dd t^2 + \sum_{j=1}^n (\dd x^j)^2 + \cO(|x|^2).
  \end{equation}
  Here, $\Gamma_{\lambda\mu\nu}=\frac12(\pa_\mu g_{\lambda\nu}+\pa_\nu g_{\lambda\mu}-\pa_\lambda g_{\mu\nu})$ (with $\pa_0:=\pa_t$ and $\pa_j=\pa_{x^j}$, $j=1,\ldots,n$) denotes the Christoffel symbols of the first kind, and $\cO(|x|^2)$ denotes a symmetric 2-tensor on $M$ all of whose coefficients vanish quadratically at $\cC$. Moreover, if one fixes the tangent vectors $\pa_{x^1},\pa_{x^2},\pa_{x^3}$ at one point in $\cC$ (where they are an orthonormal basis of $(T\cC)^\perp$), then every other coordinate system $(t',x')$ with these properties satisfies $t'=t+a$, $x'=x$ for some $a\in\R$.
\end{lemma}
\begin{proof}
  Complete $c'(0)$ to an orthonormal basis $c'(0),V_1(0),\ldots,V_N(0)\in T_{c(0)}M$. We continue this to a smooth orthonormal frame $c'(t),V_1(t),\ldots,V_N(t)\in T_{c(t)}M$, $t\in\R$, in such a manner that
  \begin{equation}
  \label{EqGLFermiTransp}
    \nabla_{c'(t)}V_i(t) \parallel c'(t).
  \end{equation}
  If $\cC$ is a geodesic, we may simply define $V_i(t)$ via parallel transport. In order to arrange~\eqref{EqGLFermiTransp} for general timelike curves $\cC$, first pick an arbitrary smooth orthonormal frame $W_1(t),\ldots,W_N(t)$ of $c'(t)^\perp$, and let
  \begin{equation}
  \label{EqGLFermiS}
    S_{i j}(t) := g(\nabla_{c'(t)}W_i(t),W_j(t)) = -S_{j i}(t).
  \end{equation}
  We then seek $(A_{\ell i}(t))\in\CI(I;O(n))$ so that for $V_i(t)=\sum A_{i\ell}(t)W_\ell(t)$ defined relative to Fermi normal coordinates, we have
  \[
    0 = g(\nabla_{c'}V_i,V_j) = ( A S A^T + A' A^T )_{i j} = \bigl(A(S+A^{-1}A')A^T\bigr)_{i j}\,.
  \]
  But if we set $A(0)=I$, then the solution of $A'(t)=-A(t)S(t)$ defines a smooth family of matrices along $\cC$ for which $J(t):=A^T(t)A(t)$ satisfies $J(0)=I$ and $J'=-S^T J-J S$; the unique solution of this ODE is $J(t)=I$ in view of~\eqref{EqGLFermiS}, and hence $A$ is orthogonal.

  The desired coordinate chart is
  \[
    (t,x) \mapsto \exp_{c(t)}\Biggl(\,\sum_{j=1}^n x^j V_j(t)\,\Biggr),
  \]
  which has invertible differential at $(t,0)$ and is thus a diffeomorphism onto its image for $(t,x)$ near $I\times\{0\}$. Note indeed that $g(t,0)=\diag(-1,1,\ldots,1)$ for all $t$. Moreover, for any fixed $t$ and for all $x\in\R^n$, the curves $s\mapsto(t,s x)$ are geodesics passing through $(t,0)$; by the geodesic equation $\ddot z^\mu(s)+\Gamma^\mu_{\kappa\lambda}|_{z(s)}\dot z^\kappa(s)\dot z^\lambda(s)=0$ (where $z=(t,x)$), this implies $\Gamma_{i j}^\mu(t,0)=0$ for $1\leq i,j\leq n$ (spatial indices) and $0\leq\mu\leq n$ (all indices), so $\Gamma_{\mu i j}(t,0)=0$. For $\mu=k\in\{1,\ldots,n\}$, this implies $\pa_i g_{j k}=\Gamma_{k i j}+\Gamma_{j i k}=0$. For $\mu=0$ on the other hand, we get for $1\leq i,j\leq n$ the equation $0=2\Gamma_{0 i j}=\pa_i g_{0 j}+\pa_j g_{0 i}$, while~\eqref{EqGLFermiTransp} implies $0=2 g(\nabla_0\pa_i,\pa_j)=2\Gamma_{j 0 i}=\pa_i g_{0 j}-\pa_j g_{0 i}$ at $(t,0)$; together, this gives $\pa_i g_{0 j}=0$. Since $-2\Gamma_{j 0 0}(t,0)=\pa_j g_{0 0}(t,0)$, the proof is complete.
\end{proof}

Regarding $g$ as an $\eps$-independent metric on $\wt M'$, the pullback $\wt\upbeta^*g\in\CI(\wt M;S^2\wt T^*\wt M)$ is a Lorentzian section of $S^2\wt T^*\wt M$; and in the coordinates~\eqref{EqGffHatCoords}, we have
\begin{equation}
\label{EqGLFermiMet}
  \sfe(\wt\upbeta^*g) = -\dd\hat t^2+\dd\hat x^2
\end{equation}
on $T_p M$ for all $p\in\cC$. Our gluing problem amounts to modifying $\wt\beta^*g$ (within the space of formal solutions of the Einstein vacuum equations) so that its restriction to $\hat M$ is the metric of a Kerr black hole, while the restriction to $M_\circ$ remains equal to $\upbeta_\circ^*g$.

\begin{definition}[Total family]
\label{DefGLTot}
  Let $g\in\CI(M;S^2 T^*M)$. Let $\hat K^\circ=\bigcup_{p\in\cC}\hat K_p^\circ\subset\hat M$ be a relatively open subset whose closure $\hat K=\bigcup_{p\in\cC}\hat K_p$ is disjoint from $\pa\hat M=\hat M\cap M_\circ$ and has connected complement in $\hat M$.\footnote{In our application, $\hat K_p$ will be a closed ball which, in $\hat x$-coordinates, has a fixed radius, and a center depending smoothly on $p$.} Let $\wt K=\{(\eps,t,x)\colon (t,\hat x)=(t,\eps x)\in\hat K\}$ be an extension of $\hat K$ to $\wt M$, where $(t,x)\in\R\times\R^n$ are local coordinates near $\cC$ with $\cC=\{x=0\}$. Let $\hat\cE,\cE\subset\C\times\N_0$ denote two nonlinearly closed index sets with $\Re\hat\cE,\Re\cE>0$. Then a Lorentzian signature $(-,+,\ldots,+)$ section $\wt g$ of $S^2\wt T^*\wt M$ over $\wt M\setminus\wt K^\circ$ is called a \emph{$(\hat\cE,\cE)$-smooth total family (relative to $(M,\cC,g)$)} if
  \begin{equation}
  \label{EqGLTot}
    \wt g = \wt\upbeta^*g + \wt g_{(1)},\qquad \wt g_{(1)}\in\cA_\phg^{\N_0\cup\hat\cE,\cE}(\wt M\setminus\wt K^\circ;S^2\wt T^*\wt M),
  \end{equation}
  with the index sets referring to $\hat M$ and $M_\circ$ (in this order). The \emph{$M_\circ$-model of $\wt g$} is $g$, and the \emph{$\hat M$-model of $\wt g$} is $\hat g:=\sfe(\wt g|_{\hat M})$. Moreover, we write $\hat g_p=\hat g|_{\hat M_p}=\sfe(\wt g|_{\hat M_p})$ for $p\in\cC$ and call this the \emph{$\hat M_p$-model of $\wt g$}.
\end{definition}

\begin{notation}[Lifts of tensors on $M$]
\label{NotGLLifts}
  In~\eqref{EqGLTot}, we regard $g\in\CI(M;S^2 T^*M)$ as an $\eps$-independent element of $\CI(\wt M';S^2\wt T^*\wt M')$, which we then pull back to $\wt M$ via $\wt\upbeta$. We shall use this notation also in the sequel for lifts of smooth functions, tensors, and differential operators on $M$.
\end{notation}

As a simple example, $\wt g=g$, as an $\eps$-independent tensor, is a $(\emptyset,\emptyset)$-smooth total family whose $\hat M_p$-model, with respect to Fermi normal coordinates $(t,x)$ and $\hat t=\dd t(-)$, $\hat x=\dd x(-)$, is the Minkowski metric $-\dd\hat t^2+\dd\hat x^2$ for all $p\in\cC$.

Definition~\ref{DefGLTot} is analogous to \cite[Definition~4.17]{HintzGlueID}. Note that $g$ is uniquely determined by $\wt g$ via $\upbeta_\circ^*g=\wt g|_{M_\circ}$. The $\hat M$-model $\wt g|_{\hat M}\in\cA_\phg^{\N_0\cup\cE}(\hat M\setminus\hat K^\circ;S^2\wt T^*\wt M)$ is Lorentzian, and hence $\hat g_p$ is a stationary Lorentzian metric on $T_p M$ with smooth dependence on $p\in\cC$. More precisely, the model of $\wt\upbeta^*g$ at $\hat M$ is
\begin{equation}
\label{EqGLMink}
  \hat\eta := \sfe\bigl(\hat\upbeta^*(g|_\cC)\bigr) \in \CI_{\rm I}(\breve T_\cC M;S^2\,{}^{\tscop,\vee}T^*\breve T_\cC M),
\end{equation}
which is a family of stationary Lorentzian metrics; in the coordinates~\eqref{EqGffHatCoords} associated with Fermi normal coordinates around $\cC$, we have $\hat\eta_p=-\dd\hat t^2+\dd\hat x^2$ for all $p\in\cC$ (cf.\ \eqref{EqGLFermiMet}); therefore,
\begin{equation}
\label{EqGLAF}
  \hat g-\hat\eta \in \cA_{\phg,\rm I}^\cE\bigl(\breve T_\cC M\setminus\breve\pi^{-1}(\hat K^\circ);S^2\,{}^{\tscop,\vee}T^*(\breve T_\cC M)\bigr),
\end{equation}
where $\breve\pi$ was defined in~\eqref{EqGffBrevePi} (so $\breve\pi^{-1}(\hat K^\circ)=\ol{\R_{\hat t}}\times\hat K^\circ$ in local coordinates); in this sense, $\hat g$ is a family of asymptotically flat metrics. The notation on the right in~\eqref{EqGLAF} means that the index set at the lift of $\ol{T_\cC}M$ to $\breve T_\cC M$ is $\cE$, and that $\hat g-\hat\eta$ is stationary on each fiber of $\breve T_\cC M$ (which implies, but is stronger than, $\hat g-\hat\eta$ having index set $\N_0$ at the front face of $\breve T_\cC M$).

Conversely, suppose that $\wt g$ of the form~\eqref{EqGLTot}, and suppose that the restriction of $\wt g$ to $M_\circ$ is a Lorentzian metric $g$ and that the $\hat M_p$-model of $\wt g$ is Lorentzian for all $p\in\cC$. We then claim that for every precompact open set $V\subset M$ there exists an $\eps(\ol{V})>0$ so that $\wt g$ is Lorentzian on
\begin{equation}
\label{EqGLDod}
  \wt\upbeta^{-1}\bigl([0,\eps(\ol{V}))\times V\bigr)\subset\wt M.
\end{equation}
Indeed, note that $\wt g$ is automatically Lorentzian in an open neighborhood of $\hat M\cup M_\circ$; since any such neighborhood contains $\wt\upbeta^{-1}(\{0\}\times\ol{V})=(\hat M\cap\hat\upbeta^{-1}(\ol{V}\cap\cC))\cup\upbeta_\circ^{-1}(\ol{V})$, the existence of $\eps(\ol{V})$ follows from the compactness of $\ol{V}$. If we take the union of the sets~\eqref{EqGLDod} over all precompact $V\subset M$, we obtain an open neighborhood of $\hat M\cup M_\circ$ on which $\wt g$ is Lorentzian. (Conversely, every open neighborhood of $\hat M\cup M_\circ$ contains the set~\eqref{EqGLDod} for any precompact $V\subset M$ and some $\eps(\ol{V})>0$.)

\begin{notation}[Definitions on neighborhoods of $\hat M\cup M_\circ$]
\label{NotGLDef}
  By a mild abuse of notation, we shall write $\CI(\wt M)$, $\Diff(\wt M)$, etc.\ for functions, differential operators, etc.\ which are defined on an open neighborhood of $\hat M\cup M_\circ$.
\end{notation}

We proceed to describe geometric objects and operators associated with $\wt g$, thus in particular explaining how the se-structures discussed in~\S\ref{SsGVf} arise. We write $\dd_{\wt M}$ for the fiberwise exterior derivative which is given by the usual exterior derivative on each fiber $\wt M_\eps\cong M$ of $\wt M$ over $(0,1)_\eps$. We similarly write $\nabla^{\wt g}$ for the fiberwise Levi-Civita connection. On $\breve T_\cC M$, we may similarly define the fiberwise exterior derivative $\dd_{\breve T_\cC M}$ (which restricts to the interior $T_p M$ of $\breve T_p M$, $p\in\cC$, to the exterior derivative on $T_p M$) and fiberwise connection $\nabla^{\hat g}$ with respect to the $\hat M$-model $\hat g$ of $\wt g$.

\begin{lemma}[Exterior derivative]
\label{LemmaGLExtDer}
  The exterior derivative on $k$-forms (here meaning: sections of $\Lambda^k\wt T^*\wt M$) satisfies (using Notation~\usref{NotGLLifts})
  \begin{equation}
  \label{EqGLExtDer}
    \dd_{\wt M} \in \wt\upbeta^*\Diff^1(M;\Lambda^k T^*M,\Lambda^{k+1}T^*M) \subset \hat\rho^{-1}\Diffse^1(\wt M;\Lambda^k\wt T^*\wt M,\Lambda^{k+1}\wt T^*\wt M).
  \end{equation}
  For the normal operators, we have
  \begin{equation}
  \label{EqGLExtDerNorm}
    \sfe\circ N_{\hat M}(\eps\dd_{\wt M})\circ\sfe^{-1} = \dd_{\breve T_\cC M},\qquad
    N_{M_\circ}(\dd_{\wt M}) = \upbeta_\circ^*\dd,
  \end{equation}
  where $\dd$ is the exterior derivative on $M$ and
  \[
    \upbeta_\circ^*\dd \in \upbeta_\circ^*\bigl(\Diff^1(M;\Lambda^k T^*M,\Lambda^{k+1}T^*M)\bigr) \subset \hat\rho^{-1}\Diffe^1(M_\circ;\upbeta_\circ^*\Lambda^k T^*M,\upbeta_\circ^*\Lambda^{k+1} T^*M)
  \]
  is its lift to $M_\circ$.
\end{lemma}
\begin{proof}
  The membership~\eqref{EqGLExtDer} follows from the fact that elements of $\wt T\wt M$ (such as the local coordinate derivatives $\pa_t$, $\pa_{x^j}$ on $M$ lifted to $\wt M$) lie in $\hat\rho^{-1}\Vse(\wt M)$. Alternatively, one can use the coordinate-free formula for the exterior derivative and use that for $V,W\in\wt\cV(\wt M)=\hat\rho^{-1}\Vse(\wt M)$, we have $[V,W]\in\hat\rho^{-2}\Vse(\wt M)=\hat\rho^{-1}\wt\cV(\wt M)$.

  Consider next the action of $\eps\dd_{\wt M}$ on functions $u$, which is
  \[
    \eps\dd_{\wt M}u = (\eps\pa_t u)\,\dd t + (\eps\pa_{x^j}u)\,\dd x^j.
  \]
  Since $N_{\hat M}(\eps\pa_z)=\sfe(\pa_z)=\pa_{\hat z}$ for $z=t,x^j$ and $\hat z=\hat t,\hat x^j$ (in the notation~\eqref{EqGffHatCoords}), and since $\sfe(\dd z)=\dd\hat z$, we have verified the first equality in~\eqref{EqGLExtDerNorm} on functions. The verification on $k$-forms is analogous. The second equality in~\eqref{EqGLExtDerNorm} is clear.
\end{proof}

\begin{lemma}[Covariant derivative]
\label{LemmaGLNabla}
  Let $\wt g$ be a $(\hat\cE,\cE)$-smooth total family relative to $(M,\cC,g)$, with $\hat M$-model $\hat g$. Let $a,b\in\N_0$. Then, using Notation~\usref{NotGLDef},
  \begin{align*}
    \nabla^{\wt g} &\in \wt\upbeta^*\bigl(\Diff^1(M;\wt T^{a,b}M,\wt T^{a,b+1}M)\bigr) + \cA_\phg^{(\N_0\cup\hat\cE)-1,\cE}\Diffse^1(\wt M\setminus\wt K^\circ;\wt T^{a,b}\wt M,\wt T^{a,b+1}\wt M) \\
      &\subset \cA_\phg^{(\N_0\cup\hat\cE)-1,\N_0\cup\cE}\Diffse^1(\wt M\setminus\wt K^\circ;\wt T^{a,b}\wt M,\wt T^{a,b+1}\wt M).
  \end{align*}
  The normal operators are given by
  \[
    \sfe\circ N_{\hat M}(\eps\nabla^{\wt g})\circ\sfe^{-1}=\nabla^{\hat g},\qquad
    N_{M_\circ}(\nabla^{\wt g})=\upbeta_\circ^*(\nabla^g).
  \]
\end{lemma}
\begin{proof}
  It suffices to prove this for $(a,b)=(0,0),(1,0)$. (This implies the result for $(a,b)=(0,1)$, and the general case then follows from the Leibniz rule.) For $(a,b)=(0,0)$, the gradient $\nabla^{\wt g}$ is the composition of\footnote{Here $\upbeta^*\CI=\upbeta^*\CI(M)$ is a subset of $\cA_\phg^{\N_0,\N_0}(\wt M)$.}
  \[
    \wt g^{-1} \in \bigl(\wt\upbeta^*\CI+\cA_\phg^{\N_0\cup\hat\cE,\cE}\bigr)\bigl(\wt M\setminus\wt K^\circ;\Hom(\wt T^*\wt M,\wt T\wt M)\bigr)
  \]
  (using that $\hat\cE,\cE$ are nonlinearly closed) with $\dd_{\wt M}$; the result in this case thus follows from Lemma~\ref{LemmaGLExtDer}. For $(a,b)=(1,0)$, one may compute the Christoffel symbols $\Gamma(\wt g)_{\mu\nu}^\kappa$ of $\wt g$ in local coordinates $z=(t,x)$: since $\pa_\lambda\wt g_{\mu\nu}\in\wt\upbeta^*\CI+\cA_\phg^{(\N_0\cup\hat\cE)-1,\cE}$, we get $\Gamma(\wt g)_{\mu\nu}^\kappa\in\wt\upbeta^*\CI+\cA_\phg^{(\N_0\cup\hat\cE)-1,\cE}$. The claim then follows from $\nabla^{\wt g}_\mu(a^\nu\pa_\nu)=(\pa_\mu a^\kappa+a^\nu\Gamma_{\mu\nu}^\kappa(\wt g))\pa_\kappa$, since $\pa_\mu$ lies in $\hat\rho^{-1}\Vse(\wt M)$ as a differential operator and is smooth as a section of $\wt T\wt M$. The $M_\circ$-normal operator selects the coefficients of class $\wt\upbeta^*\CI$ (which come solely from $g$), while for the $\hat M$-normal operator we note that the restriction of $\eps\pa_{z^\lambda}\wt g_{\mu\nu}=\pa_{\hat z^\lambda}\wt g_{\mu\nu}$ to $\hat M_p$ is $\pa_{\hat z^\lambda}\hat g(\pa_{\hat z^\mu},\pa_{\hat z^\nu})$. (We leave a coordinate-free proof using the Koszul formula to the reader.) 
\end{proof}

By the properties of the $\hat M$-normal operator map, or by direct computation on $\hat M$, the operators $\dd_{\breve T_\cC M}$ and $\nabla^{\hat g}$ are elements of $\rho_\circ{}^\vee\Difftb^1(\hat M)$ acting between the appropriate tensor powers of ${}^{\tscop,\vee}T(\breve T_\cC M)$. This is an instance of the well-known fact that geometric operators associated with \emph{scattering} metrics are \emph{weighted b}-operators, see e.g.\ \cite[Theorem~6.8]{VasyMinicourse} and \cite[\S4]{HintzUnDet}.

\begin{cor}[Curvature]
\label{CorGLCurvature}
  Let $\wt g$ be a $(\hat\cE,\cE)$-smooth family relative to $(M,\cC,g)$, with $\hat M$-model $\hat g$. Then the Riemann curvature tensor $\Riem(\wt g)(X,Y)Z=([\nabla^{\wt g}_X,\nabla^{\wt g}_Y]-\nabla^{\wt g}_{[X,Y]})Z$ satisfies
  \begin{align*}
    &\Riem(\wt g) \in \wt\upbeta^*\CI(M;\wt T^{1,3}M)+\cA_\phg^{(\N_0\cup\hat\cE)-2,\cE}(\wt M\setminus\wt K^\circ;\wt T^{1,3}\wt M) \\
      &\hspace{1.3em}\subset \cA_\phg^{(\N_0\cup\hat\cE)-2,\N_0\cup\cE}(\wt M\setminus\wt K^\circ;\wt T^{1,3}\wt M), \\
    &\sfe\bigl((\eps^2\Riem(\wt g))|_{\hat M}\bigr) = \Riem(\hat g),\qquad
      \Riem(\wt g)|_{M_\circ}=\upbeta_\circ^*(\Riem(g)).
  \end{align*}
  The Ricci tensor similarly satisfies
  \begin{align*}
    &\Ric(\wt g)\in\wt\upbeta^*\CI(M;S^2\wt T^*M)+\cA_\phg^{(\N_0\cup\hat\cE)-2,\cE}(\wt M\setminus\wt K^\circ;S^2\wt T^*\wt M), \\
      &\hspace{1.3em}\subset \cA_\phg^{(\N_0\cup\hat\cE)-2,\N_0\cup\cE}(\wt M\setminus\wt K^\circ;\wt T^{1,3}\wt M), \\
    &\sfe\bigl((\eps^2\Ric(\wt g))|_{\hat M}\bigr) = \Ric(\hat g),\qquad
    \Ric(\wt g)|_{M_\circ}=\upbeta_\circ^*(\Ric(g));
  \end{align*}
  analogously for the scalar curvature.
\end{cor}

The following result connects the notion of total spacetime family to the corresponding notion on the level of initial data sets for the Einstein equations, cf.\ \cite[Definition~4.17]{HintzGlueID}:

\begin{cor}[Initial data]
\label{CorGLInitial}
  Let $\wt g$ be a $(\hat\cE,\cE)$-smooth total family relative to $(M,\cC,g)$ on $\wt M\setminus\wt K^\circ$ in the notation of Definition~\usref{DefGLTot}. Let $X\subset M$ be a smooth spacelike hypersurface of $(M,g)$ with $X\cap\cC=\{\fp\}$. Suppose that $T_\fp X\setminus\hat K_\fp^\circ$ is spacelike for the $\hat M_\fp$-model $\hat g_\fp$. Let $\cK\subset M$ be compact. Then:
  \begin{enumerate}
  \item\label{ItGLInitialSpace} $\wt X\setminus\wt K^\circ\subset\wt M$ is spacelike for $\wt g$ on $\wt\upbeta^{-1}([0,\eps_0)\times\cK)$ when $\eps_0>0$ is sufficiently small;
  \item\label{ItGLInitialMem} the initial data $(\wt\gamma,\wt k)$ of $\wt g$ at $\wt X\setminus\wt K^\circ$ (i.e.\ the first and second fundamental form of $\wt X_\eps\setminus\wt K^\circ\subset\wt M_\eps$, $\eps\in(0,1)$) take the form
    \[
      \wt\gamma = \wt\upbeta^*\gamma + \wt\gamma_{(1)},\qquad
      \wt k = \wt\upbeta^*k + \wt k_{(1)},
    \]
    where $(\gamma,k)$ are the initial data of $X$ in $(M,g)$, and where
    \[
      \wt\gamma_{(1)} \in \cA_\phg^{\N_0\cup\hat\cE,\cE}(\wt X\setminus\wt K^\circ;S^2\wt T^*\wt X),\qquad
      \wt k_{(1)} \in \cA_\phg^{(\N_0\cup\hat\cE)-1,\cE}(\wt X\setminus\wt K^\circ;S^2\wt T^*\wt X).
    \]
    Moreover, $(\hat\gamma,\hat k):=(\sfs(\wt\gamma|_{\hat X}),\sfs(\eps\wt k|_{\hat X}))$ are the initial data of $\hat g_\fp$ at $T_\fp X\setminus\hat K_\fp^\circ$.
  \end{enumerate}
\end{cor}

We shall also call $(\hat\gamma,\hat k)$ the initial data of $\wt g$ at $\hat M_\fp$. In the terminology of \cite[Definition~4.18]{HintzGlueID}, part~\eqref{ItGLInitialMem} states that the pair $(\wt\gamma,\wt k)$ is a $(\hat\cE,\cE)$-smooth total family with boundary data $(\hat\gamma,\hat k)$ and $(\gamma,k)$.

\begin{proof}[Proof of Corollary~\usref{CorGLInitial}]
  We verify the claims near the codimension $2$ corner of $\wt M$. We work in local coordinates $z=(t,x)$ (which need not be Fermi normal coordinates) on $M$ with respect to which $X=t^{-1}(0)$ and $\cC=x^{-1}(0)$, and use coordinates $t\in\R$, $\hat\rho=|x|\geq 0$, $\rho_\circ=\frac{\eps}{|x|}\geq 0$, $\omega\in\Sph^{n-1}$ near $\hat M\cap M_\circ$. Then the dual metric $\wt g^{-1}$ takes the form
  \[
    \wt g^{-1}(t,\hat\rho,\rho_\circ,\omega) = \bigl(g^{\mu\nu}(t,\hat\rho\omega) + g_{(1)}^{\mu\nu}(t,\hat\rho,\rho_\circ,\omega)\bigr)\pa_{z^\mu}\otimes_s\pa_{z^\nu},
  \]
  where $g_{(1)}^{\mu\nu}$ is $\cE$-smooth at $\rho_\circ=0$ and $(\N_0\cup\hat\cE)$-smooth at $\hat\rho=0$; in particular it is continuous and vanishes at $\rho_\circ=0$. The spacelike nature of $X$ is equivalent to $g^{0 0}(t,\hat\rho\omega)<0$. Writing $\hat z=(\hat t,\hat x)=(\dd t(-),\dd x(-))$, the $\hat M_t$-model of $\wt g^{-1}$ is
  \[
    \bigl(g^{\mu\nu}(t,0)+g_{(1)}^{\mu\nu}(t,0,\rho_\circ,\omega)\bigr) \pa_{\hat z^\mu}\otimes_s\pa_{\hat z^\nu},
  \]
  and the spacelike nature of $T_\fp X=\hat t^{-1}(0)$ is equivalent to $g^{0 0}(t,0)+g_{(1)}^{0 0}(t,0,\rho_\circ,\omega)<0$. Taken together, we conclude that $\wt g^{\mu\nu}(t,\hat\rho,\rho_\circ,\omega)<0$ when $(t,\hat\rho,\rho_\circ,\omega)$ is a sufficiently small neighborhood of $\wt X\cap(\hat M\cup M_\circ)$.

  Turning to part~\eqref{ItGLInitialMem}, the statements about the induced metric $\wt\gamma$ are clear. Regarding the second fundamental form, we note that the normal vector field $\wt g^{-1}(\dd t,-)$ to $\wt X$ is of class
  \begin{equation}
  \label{EqGLInitialClass}
    (\wt\upbeta^*\CI+\cA_\phg^{\N_0\cup\hat\cE,\cE})(\wt X\setminus\wt K^\circ;\wt T\wt M),
  \end{equation}
  and its squared norm near $\hat X\cup X_\circ$ is strictly negative. Therefore, the unit normal is an element of~\eqref{EqGLInitialClass} as well. By Lemma~\ref{LemmaGLNabla}, its covariant derivative along elements of $\wt T\wt M$ lies in $\wt\upbeta^*\CI+\cA_\phg^{(\N_0\cup\hat\cE)-1,\cE}$; therefore $\wt k$ is of the stated form. The identification of the boundary values of $\wt k$ follows at $M_\circ$ from these arguments, and at $\hat M$ as usual by passing to the coordinates $(\hat t,\hat x)=(\frac{t}{\eps},\frac{x}{\eps})$.
\end{proof}

\subsection{Kerr metrics as front face models}
\label{SsGK}

In this section, we work with
\[
  n=3
\]
spatial dimensions. (Thus, $M$ is 4-dimensional, and $\wt M$ is 5-dimensional.) First, we recall (see \cite{KerrKerr,BoyerLindquistKerr}):

\begin{definition}[Kerr metric]
\label{DefGK}
  Let $\bhm>0$ and $a\in\R$ be \emph{subextremal} Kerr parameters, i.e.\ $|a|<\bhm$. Write $\hat r_{\bhm,a}:=\bhm+\sqrt{\bhm^2-a^2}\in(\bhm,2\bhm]$ for the radius of the event horizon. Then the Kerr metric $\hat g_{\bhm,a}$ in Boyer--Lindquist coordinates is the Lorentzian metric
  \begin{equation}
  \label{EqGKBL}
  \begin{split}
    \hat g_{\bhm,a} &= -\frac{\mu(\hat r)}{\varrho^2(\hat r,\theta)}(\dd\hat t_{\rm BL}-a\sin^2\theta\,\dd\phi)^2 + \varrho^2(\hat r,\theta)\Bigl(\frac{\dd\hat r^2}{\mu(\hat r)}+\dd\theta^2\Bigr) \\
    &\quad\hspace{3em} + \frac{\sin^2\theta}{\varrho^2(\hat r,\theta)}\bigl((\hat r^2+a^2)\dd\phi-a\,\dd\hat t_{\rm BL}\bigr)^2, \\
    \mu(\hat r)&=\hat r^2-2\bhm\hat r+a^2,\qquad
      \varrho^2(\hat r,\theta)=\hat r^2+a^2\cos^2\theta.
  \end{split}
  \end{equation}
  on the manifold $\R_{\hat t_{\rm BL}}\times(\hat r_{\bhm,a},\infty)_{\hat r}\times\Sph^2_{\theta,\phi}$.
\end{definition}

This solves $\Ric(\hat g_{\bhm,a})=0$. Moreover, $\pa_{\hat t_{\rm BL}}$ is the unique Killing vector field which is asymptotically (as $\hat r\to\infty$) future timelike and has squared length approaching $-1$. A useful form of $\hat g_{\bhm,a}$, obtained by using the standard metric on $\Sph^2$, $\slg=\dd\theta^2+\sin^2\theta\,\dd\phi^2$, to rewrite the $\dd\theta^2$ term, is
\begin{equation}
\label{EqGKBL2}
\begin{split}
  \hat g_{\bhm,a} &= -\Bigl(1-\frac{2\bhm\hat r}{\varrho^2}\Bigr)\dd\hat t_{\rm BL}^2 + \Bigl(1+\frac{2\bhm\hat r-a^2\sin^2\theta}{\mu}\Bigr)\dd\hat r^2+\varrho^2\slg \\
    &\qquad + \Bigl(1+\frac{2\bhm\hat r}{\varrho^2}\Bigr)(a\sin^2\theta\,\dd\phi)^2 - \frac{4 a\bhm\hat r\sin^2\theta}{\varrho^2}\dd\hat t_{\rm BL}\,\dd\phi.
\end{split}
\end{equation}

\begin{rmk}[Importance of Kerr]
\label{RmkGKWhyKerr}
  Suppose $\wt g$ is a time-oriented $(\hat\cE,\cE)$-smooth total family on $\wt M\setminus\wt K^\circ$ relative to $(M,\cC,g)$, and suppose that
  \[
    \Ric(\wt g) - \Lambda\wt g=0.
  \]
  By Corollary~\ref{CorGLCurvature}, this implies $0=\sfe^{-1}((\eps^2\Ric(\wt g)-\eps^2\Lambda\wt g)|_{\hat M})=\Ric(\hat g)$; that is, for each $p\in\cC$, the metric $\hat g_p\in\CI(T_p M\setminus\hat K_p^\circ;S^2 T^*(T_p M))$ is a stationary and asymptotically flat solution of $\Ric(\hat g_p)=0$. According to the black hole uniqueness conjecture, see e.g.\ \cite[Conjecture~3.4]{ChruscielCostaHeuslerStationaryBH}, this (together with additional hypotheses concerning the extent of $T_p M\setminus\hat K_p^\circ$) forces $\hat g_p$ to be isometric to the metric of a Kerr black hole.
\end{rmk}

Near the event horizon, the Boyer--Lindquist coordinate singularity can be removed by passing to new coordinates:

\begin{lemma}[Smooth coordinates across the future event horizon]
\label{LemmaGKCoords}
  Given subextremal Kerr parameters $\bhm,a$, denote by $T,\Phi\colon(\hat r_{\bhm,a},\infty)\to\R$ a pair of smooth functions with
  \[
    T'(\hat r) = -\frac{\hat r^2+a^2}{\mu(\hat r)}+\tilde T(\hat r),\qquad
    \Phi'(\hat r) = -\frac{a}{\mu(\hat r)} + \tilde\Phi(\hat r),
  \]
  where $\tilde T,\tilde\Phi$ are analytic on $[0,4\bhm]$. Let $\hat t=\hat t_{\rm BL}-T(\hat r)$ and $\phi_*=\phi-\Phi(\hat r)$. Then $\hat g_{\bhm,a}$ extends analytically from $\hat r>\hat r_{\bhm,a}$ to a metric on $\R_{\hat t}\times(\hat r^-_{\bhm,a},\infty)_{\hat r}\times\Sph^2_{\theta,\phi_*}$ where $\hat r^-_{\bhm,a}=\bhm-\sqrt{\bhm^2-a^2}\in[0,\bhm)$. Moreover, we may choose $\tilde T,\tilde\Phi$ so that the following additional conditions are satisfied:
  \begin{enumerate}
  \item $T(\hat r)=0$ and $\Phi(\hat r)=0$ for large $\hat r$;
  \item $\dd\hat t$ is everywhere (past) timelike;
  \item for $\bhm,a$ which are close to fixed subextremal parameters $\bhm_0,a_0$, the functions $\tilde T$ and $\tilde\Phi$ depend smoothly on $\bhm,a$.
  \end{enumerate}
\end{lemma}
\begin{proof}
  This is a standard and straightforward computation; see e.g.\ \cite[\S3.1]{HintzKdSMS}.
\end{proof}

\begin{definition}[Kerr model]
\label{DefGKModel}
  Fix Kerr parameters $\bhm>0$, $\bha\in\R^3$ which are subextremal (in the sense that $\bhm,a:=|\bha|$ are subextremal), and define $\hat t,\phi_*$ as in Lemma~\ref{LemmaGKCoords}. Fix the spacetime manifold
  \[
    \hat M_{\bhm,\bha}^\circ := \R_{\hat t} \times \hat X^\circ_{\bhm,\bha},\qquad \hat X^\circ_{\bhm,\bha}=\R^3\setminus(\hat K_{\bhm,\bha})^\circ,
  \]
  where $\hat K_{\bhm,\bha}=\hat K_{\bhm,\bha}^0$ with $\hat K_{\bhm,\bha}^\delta:=\{\hat x\in\R^3\colon\hat r=|\hat x|\leq\bhm-\delta\}$ for $\delta\in(\bhm-\hat r_{\bhm,a},\bhm-\hat r_{\bhm,a}^-)=(-\sqrt{\bhm^2-a^2},\sqrt{\bhm^2-a^2})$. Define $\hat g_{\bhm,a}$ as a metric on $\hat M_{\bhm,\bha}^\circ$ by identifying $\theta\in(0,\pi)$ and $\phi_*\in(0,2\pi)$ with the standard polar coordinates on $\R^3_{\hat x}$ in $\hat r\geq\bhm$. Let $R\in SO(3,1)$ be a rotation in the spatial variables which maps $(\hat t,0,0,a)^T\mapsto(\hat t,\bha)^T$. We then write\footnote{The Lorentz transformation $R$ is unique up to pre-composition by a rotation around the $z$-axis. Since such rotations are isometries for $\hat g_{\bhm,a}$, the metric $\hat g_{\bhm,\bha}$ is well-defined.}
  \[
    \hat g_{\bhm,\bha}=R_*\hat g_{\bhm,a}
  \]
  Setting $ b=(\bhm,\bha)$, we also write these objects as
  \[
    \hat g_b,\ \hat M_b^\circ,\ \hat X_b^\circ,\ \hat K_b.
  \]
  We finally denote by $\hat{\ubar g}=-\dd\hat t^2+\dd\hat x^2$ the Minkowski metric on the same spacetime manifold.
\end{definition}

Thus, $\hat g_{\bhm,\bha}$ is stationary in that $\pa_{\hat t}$ is a Killing vector field which is asymptotically timelike. Moreover, the components of $\hat g_{\bhm,\bha}$ and $\hat{\ubar g}$ in the coordinates $(\hat t,\hat x)$ satisfy
\begin{equation}
\label{EqGKAsymp}
  \hat g_{\bhm,\bha} - \hat{\ubar g} \in \hat r^{-1}\CI(\ol{\R^3}\setminus(\hat K_{\bhm,\bha}^\delta)^\circ),\qquad
  \hat g_{\bhm,\bha} - \hat g_{\bhm,0} \in \hat r^{-2}\CI(\ol{\R^3}\setminus(\hat K_{\bhm,\bha}^\delta)^\circ)
\end{equation}
for all $\delta\in[0,\bhm-\hat r_{\bhm,a}^-)$, as can easily be verified using the formula~\eqref{EqGKBL2}; see also equation~\eqref{EqGKLot} below.

In order to efficiently keep track of asymptotic expansions, we now compactify:

\begin{definition}[Compactifications related to the Kerr spacetime manifold]
\label{DefGKCpt}
  In the notation of Definition~\usref{DefGKModel}, we write $\hat X_b\subset\ol{\R^3}\setminus\hat K_b^\circ$ for the closure $\hat X_b=\hat X_b^\circ\cup\pa\ol{\R^3}$ of $\hat X_b^\circ$ inside the radial compactification of $\R^3$. We furthermore set
  \[
    \hat M_b := \bigl[\ol{\R_{\hat t}\times\R^3_{\hat x}}; \pa\ol{\R_{\hat t}}\,\bigr] \setminus \hat\pi^{-1}(\hat K_b^\circ),
  \]
  where $\hat\pi\colon[\ol{\R\times\R^3};\pa\ol{\R}]\to\ol{\R^3}$ is the lift of the projection $(\hat t,\hat x)\mapsto\hat x$.
\end{definition}

See Figure~\ref{FigGKCpt}. Therefore, $\hat g_{\bhm,\bha},\hat{\ubar g}\in\CI_{\rm I}(\hat M_{\bhm,\bha};S^2\,\Ttsc^*\hat M_{\bhm,\bha})$ (the subscript `$\rm I$' restricting to translation-invariant sections). By Lemma~\ref{LemmaGff3bProj}, we can identify $\CI_{\rm I}(\hat M_{\bhm,\bha})\cong\CI(\hat X_{\bhm,\bha})$. The membership~\eqref{EqGKAsymp} is then equivalent to
\begin{equation}
\label{EqGKCpt3sc}
  \hat g_{\bhm,\bha}-\hat{\ubar g}\in\rho_\circ\CI(\hat X_{\bhm,\bha};S^2\,\Ttsc^*_{\hat X_{\bhm,\bha}}\hat M_{\bhm,\bha}),\qquad
  \hat g_{\bhm,\bha}-\hat g_{\bhm,0}\in\rho_\circ^2\CI,\qquad \rho_\circ:=\la\hat x\ra^{-1}.
\end{equation}

\begin{figure}[!ht]
\centering
\includegraphics{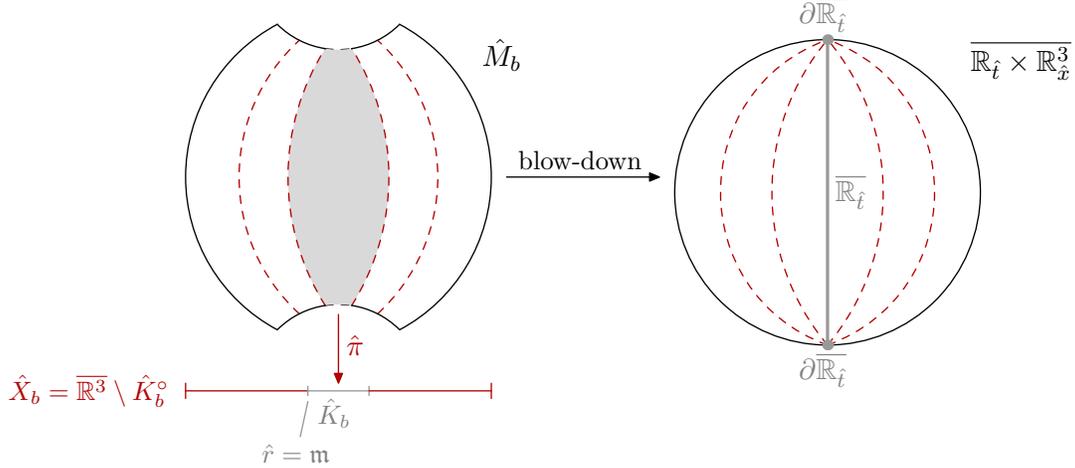}
\caption{Illustration of Definition~\ref{DefGKCpt}. \textit{On the left:} the compactification $\hat M_b$ of the Kerr spacetime manifold, and the projection map to the spatial manifold $\hat X_b$. \textit{On the right:} the radial compactification of $\R_{\hat t}\times\R^3_{\hat x}$.}
\label{FigGKCpt}
\end{figure}

We include the parameter $\delta>0$ in Definition~\ref{DefGKModel} in order to enable us to regard $\hat g_b$ for $b=(\bhm,\bha)$ near fixed subextremal Kerr parameters $b_0=(\bhm_0,\bha_0)$ as a smooth family of metrics on the fixed spacetime manifold $\hat M_{b_0}^\circ\subset\R_{\hat t}\times(\R^3\setminus(\hat K_b^\delta)^\circ)$; this was shown in the Kerr--de~Sitter case in \cite[Proposition~3.5]{HintzVasyKdSStability}. We recall the relevant calculations in the region $\hat r\gg\bhm$ where $(\hat t,\hat r,\theta,\phi_*)$ are equal to the Boyer--Lindquist coordinates $(\hat t_{\rm BL},\hat r,\theta,\phi)$. There, the metric $\hat g_{\bhm,\bha}$ is then given by~\eqref{EqGKBL2} but with the spherical coordinates defined relative to the axis of rotation $\hat\bha=\frac{\bha}{|\bha|}$ (when $\bha\neq 0$) instead of $(0,0,1)^T\in\R^3$. Using Euclidean coordinates $\hat x=(\hat r\sin\theta\cos\phi,\hat r\sin\theta\sin\phi,\hat r\cos\theta)$ and the Euclidean inner product, we can then write $\hat r^2=|\hat x|^2$, $a^2=|\bha|^2$, $a^2\cos^2\theta=(\bha\cdot\frac{\hat x}{|\hat x|})^2$, and
\[
  \varrho^2=|\hat x|^2+\Bigl(\bha\cdot\frac{\hat x}{|\hat x|}\Bigr)^2,\qquad
  \hat r a\sin^2\theta\,\dd\phi=\Bigl(\bha\times\frac{\hat x}{|\hat x|}\Bigr)\cdot\dd\hat x,\qquad
  a^2\sin^2\theta=|\bha|^2-\Bigl(\bha\cdot\frac{\hat x}{|\hat x|}\Bigr)^2.
\]
Thus, all terms in~\eqref{EqGKBL2} depend smoothly on $\bha$ indeed. We record the leading order behavior
\begin{equation}
\label{EqGKLot}
\begin{split}
  \hat g_{\bhm,\bha} &\equiv \hat{\ubar g} + \hat r^{-1}\bigl[2\bhm(\dd\hat t^2+\dd\hat r^2)\bigr] \\
    &\qquad + \hat r^{-2}\biggl[ \Bigl(4\bhm^2-|\bha|^2+\Bigl(\bha\cdot\frac{\hat x}{|\hat x|}\Bigr)^2\Bigr)\,\dd\hat r^2 + \Bigl(\bha\cdot\frac{\hat x}{|\hat x|}\Bigr)^2\hat r^2\slg + \Bigl(\Bigl(\bha\times\frac{\hat x}{|\hat x|}\Bigr)\cdot\dd\hat x\Bigr)^2 \\
    &\qquad\hspace{6em} - 4\bhm\,\dd\hat t \otimes_s \Bigl(\Bigl(\bha\times\frac{\hat x}{|\hat x|}\Bigr)\cdot\dd\hat x\Bigr)\biggr] \\
    &\qquad \bmod \hat r^{-3}\CI\bigl(\hat X_{\hat\bhm,\hat\bha};S^2\,\Ttsc^*_{\hat X_{\hat\bhm,\hat\bha}}\hat M_{\hat\bhm,\hat\bha}\bigr).
\end{split}
\end{equation}

\begin{definition}[Linearized Kerr metrics]
\label{DefGKLin}
  Given subextremal Kerr parameters $b=(\bhm,\bha)$, and given $\dot b=(\dot\bhm,\dot\bha)\in\R\times\R^3$, we define
  \[
    \hat g'_b(\dot b) := \frac{\dd}{\dd s}\hat g_{b+s\dot b}\Big|_{s=0}
  \]
  as a symmetric 2-tensor on $\hat M_b^\circ$.
\end{definition}

In view of~\eqref{EqGKCpt3sc}, and again identifying $\CI_{\rm I}(\hat M_b)\cong\CI(\hat X_b)$, we have (with $\rho_\circ=\la\hat x\ra^{-1}$)
\begin{equation}
\label{EqGKLinSize}
\begin{split}
  \hat g'_b(\dot b) &\in \rho_\circ\CI(\hat X_b;S^2\,\Ttsc^*_{\hat X_b}\hat M_b), \\
  \hat g'_b(0,\dot\bha) &\in \rho_\circ^2\CI(\hat X_b;S^2\,\Ttsc^*_{\hat X_b}\hat M_b).
\end{split}
\end{equation}
Furthermore, linearizing the equation $\Ric(\hat g_b)=0$ in $b$, we obtain
\begin{equation}
\label{EqGKLinRic}
  D_{\hat g_b}\Ric(\hat g'_b(\dot b))=0\qquad\forall\,\dot b\in\R\times\R^3.
\end{equation}

\subsubsection{Gluing in a family of Kerr metrics}

We proceed to describe how to glue a family of Kerr metrics into the front face $\hat M$ in the setting of~\S\ref{SsGL}. Recalling the Minkowski metric
\begin{equation}
\label{EqGKMink}
  \hat\eta_p=\sfe(\hat\upbeta^*g_p)
\end{equation}
on $T_p M$ from~\eqref{EqGLMink}, define
\begin{equation}
\label{EqGKX}
  \hat X_p^\circ := (T_p\cC)^{\perp_{\hat\eta_p}} \subset T_p M.
\end{equation}
Denote by
\begin{equation}
\label{EqGKT}
  T\in\cV(T_p M)
\end{equation}
the unique generator of the $T_p\cC$-translation action which is future timelike and has squared length $-1$. (Thus, $T=\pa_{\hat t}$ in the coordinates~\eqref{EqGffHatCoords} relative to Fermi normal coordinates along $\cC$.) By choosing an orthonormal basis of $\hat X_p^\circ$ to define coordinates $\hat x\in\R^3$, and moreover defining the linear coordinate
\[
  \hat t := -\hat\eta(T,-)
\]
on $T_p M=T_p\cC\oplus\hat X_p^\circ$, we obtain an isometric isomorphism
\[
  \Phi_p \colon (\R^{1+3}_{\hat t,\hat x},\hat{\ubar g}) \to (T_p M,\hat\eta_p).
\]
Since $T\hat x=0$ by construction, the pushforward $(\Phi_p)_*\hat g_{\bhm,\bha}$ is a stationary (with respect to $T_p\cC$) metric on $T_p M$ which asymptotes to $\hat\eta_p$, and for which $\hat X_p^\circ$ is a spacelike Cauchy surface. More precisely,
\begin{equation}
\label{EqGKPhiAF}
\begin{split}
  (\Phi_p)_*\hat g_{\bhm,\bha} &\in \CI_{\rm I}\bigl(\breve T_p M\setminus\breve\pi^{-1}(\Phi_p(\hat K_{\bhm,\bha}^\circ));S^2\,{}^{\tscop,\vee}T^*(\breve T_p M)\bigr), \\
  (\Phi_p)_*\hat g_{\bhm,\bha}-\hat\eta_p&\in\cA_{\phg,\rm I}^{\N_0+1}\bigl(\breve T_p M\setminus\breve\pi^{-1}(\Phi_p(\hat K_{\bhm,\bha}^\circ));S^2\,{}^{\tscop,\vee}T^*(\breve T_p M)\bigr),
\end{split}
\end{equation}
cf.\ \eqref{EqGLAF} and~\eqref{EqGKCpt3sc}, as follows from the fact that $\hat\eta_p=\Phi_*\hat{\ubar g}$. In~\eqref{EqGKPhiAF}, we wrote $\Phi_p$ also for the map $\R^3_{\hat x}\to\hat X_p^\circ$ given by restriction of $\Phi_p$ to $\hat t=0$.

\begin{lemma}[Naive gluing]
\label{LemmaGKNaive}
  Let $M,\cC,g$ be as in~\textnormal{\S}\usref{SsGL}, and fix $0<\bhm\in\CI(\cC)$, $\bha\in\CI(\cC;\R^3)$ where $|\bha(p)|<\bhm(p)$ for all $p\in\cC$; set $b=(\bhm,\bha)$.\footnote{In our gluing construction, we will operate under the assumptions that $\cC$ is a geodesic and $b$ is constant. We state the more general result here so that we may prove the \emph{necessity} of these assumptions for the existence of a total family with specified $M_\circ$- and $\hat M$-models satisfying the Einstein vacuum equations; see~\S\S\ref{SsAcGW} and \ref{SsAhNec}. We could further generalize the naive gluing here by considering maps $\Phi_p(\hat t,\hat x)=(\hat t,A(p)\hat x+\hat c(p))$ where $\hat A\in\CI(\cC;O(3))$ and $\hat c\in\CI(\cC;\R^3)$. It will turn out that $\hat A$ is necessarily constant, see Proposition~\ref{PropAhNec}, while the flexibility of shifting the center of mass of the $\hat M$-model metrics via $\hat c$ is not needed; in any case this could equivalently be implemented via pulling back along $(\eps,t,x)\mapsto(\eps,t,x-\eps\hat c(t))$ near $\cC$ (glued via a partition of unity to the identity map on $\wt M$ away from $\hat M$).} Fix Fermi normal coordinates $(t,x)$ along $\cC$, and write $r=|x|$ (Euclidean norm), $\omega=\frac{x}{|x|}\in\Sph^2$.
  \begin{enumerate}
  \item\label{ItGKNaive} There exists a $(\emptyset,\N_0+1)$-smooth total family $\wt g$ relative to $(M,\cC,g)$ so that for all $p\in\cC$, the $\hat M_p$-model of $\wt g$ is $(\Phi_p)_*\hat g_{b(p)}$, where
    \begin{equation}
    \label{EqGKNaivePhi}
      \Phi_p(\hat t,\hat x)=(\hat t,\hat x)
    \end{equation}
    in the coordinates $(\hat t,\hat x)$ on $\hat M_{b(p)}^\circ$ and the fiber-linear coordinates $\hat t=\dd t$ and $\hat x^j=\dd x^j$, $j=1,2,3$, on $T_\cC M$. More precisely, if we set $\hat K^\circ=\bigsqcup_{p\in\cC}\Phi_p(\hat K_{b(p)}^\circ)$ and define $\wt K\subset\wt M$ as in Definition~\usref{DefGLTot}, then $\wt g$ is defined on $\wt M\setminus\wt K^\circ$.
  \item\label{ItGKNaiveGeod} If $\cC$ is a geodesic, then the total family $\wt g$ in part~\eqref{ItGKNaive} can be defined so that $\wt g(t_0)-\sfe^{-1}((\Phi_{c(t_0)})_*\hat g_{b(c(t_0))})\in\hat\rho^2\CI(\{t=t_0\};S^2\wt T^*\wt M)$ for all $t_0\in I$, and so that
  \begin{equation}
  \label{EqGKNaiveEps}
    \bigl(\eps^{-1}(\wt g-\wt\upbeta^*g)\bigr)|_{M_\circ}\equiv\frac{2\bhm}{r}(\dd t^2+\dd r^2)\bmod\CIdot(M_\circ;\upbeta_\circ^*S^2 T^*M)
  \end{equation}
  near $\pa M_\circ$. Here, we extend the section $\sfe^{-1}(\Phi_{c(t_0)})_*\hat g_{b(c(t_0))}\in\CI(\hat M_{c(t_0)};S^2\wt T^*\wt M)$ to $\{t=t_0\}\subset\wt M$ to be constant in $\eps$ in the space on the right in~\eqref{EqGRelHat}.\footnote{Concretely, writing $\sfe^{-1}((\Phi_{c(t_0)})_*\hat g_{b(c(t_0))})=-\dd t^2+\dd r^2+r^2\slg+h(t_0,\frac{\eps}{r},\omega)$, this extension is simply given by the expression on the right hand side.}
  \end{enumerate}
\end{lemma}
\begin{proof}
  The prescription $\wt g|_{\hat M_p}=\sfe^{-1}((\Phi_p)_*\hat g_{b(p)})$ on $\hat M_p$, $p\in\cC$, defines a smooth section $\wt g|_{\hat M}$ of $S^2\wt T^*\wt M$ over $\hat M\setminus\hat K^\circ$. Now, the difference $\wt g|_{\hat M_p}-\hat\upbeta^*g_p$ vanishes at $\pa\hat M_p$ as a section of $S^2\,{}^{\tscop,\vee}T^*(\breve T_p M)$; indeed, note that its image under $\sfe$ is $(\Phi_p)_*\hat g_{b(p)}-\hat\eta_p$, and recall~\eqref{EqGKPhiAF}. We conclude that $\wt g|_{\hat M}$ and $\upbeta_\circ^*g$ are equal (as sections of $S^2\wt T^*\wt M$) at the corner $\hat M\cap M_\circ$ (and indeed they are equal to the Minkowski metric $-\dd t^2+\dd x^2$ there). Therefore, there indeed exists a $(\emptyset,\N_0+1)$-smooth total family $\wt g$ which over $\wt M\setminus\wt K^\circ$ equals $\wt g|_{\hat M}$ at $\hat M$ and $\upbeta_\circ^*g$ at $M_\circ$ (as a section of $S^2\wt T^*\wt M$).

  For part~\eqref{ItGKNaiveGeod}, we note that Lemma~\ref{LemmaGLFermi} gives
  \[
    g(t,r,\omega)=-\dd t^2+\dd r^2+r^2\slg+g',
  \]
  where $g'\in\CI(M;S^2 T^*M)$ vanishes quadratically at $x=0$. Moreover,
  \[
   \sfe^{-1}\bigl((\Phi_{c(t_0)})_*\hat g_{b(c(t_0))}\bigr) = -\dd t^2+\dd r^2+r^2\slg+h\Bigl(t_0,\frac{\eps}{r},\omega\Bigr)
 \]
 where $h$ is a smooth section of $S^2\wt T^*\wt M$ over $\hat M$ which vanishes simply at $\frac{\eps}{r}=0$, and which near $\pa\hat M$ and modulo $(\frac{\eps}{r})^2\CI=\rho_\circ^2\CI$ is equal to $\frac{2\bhm(t_0)}{r/\eps}(\dd t^2+\dd r^2)$. We may then take
  \[
    \wt g(t_0) = -\dd t^2+\dd r^2+r^2\slg + \wt\upbeta^*g'(t_0) + \hat\chi(t_0,x)h\Bigl(t_0,\frac{\eps}{r},\omega\Bigr),
  \]
  where $\hat\chi$ is a cutoff function to the Fermi normal coordinate chart which is $1$ near $x=0$. Since $\wt\upbeta^*f\subset\hat\rho^2\CI(\wt M)$ when $f\in\CI(M)$ vanishes quadratically at $\cC$, this has the required properties.
\end{proof}

\subsubsection{Spacelike hypersurfaces in \texorpdfstring{$M$}{the background spacetime}}

Suppose now we are given a spacelike hypersurface $X\subset M$ with $X\cap\cC=\{\fp\}$. Besides $\hat X_\fp^\circ=(T_\fp\cC)^\perp$ (determined by $\cC$ and $g$), we then have another hypersurface $T_\fp X\subset T_\fp M$ (determined by $X$) which is spacelike with respect to the Minkowski metric $\hat\eta_\fp$ from~\eqref{EqGKMink}.

\begin{definition}[Lorentz boosts]
\label{DefGKBoostDef}
  Fix as the generators $B_j\in\mathfrak{so}(1,3)$, $j=1,2,3$, of Lorentz boosts in the $\hat x^j$-direction on $\R_{\hat t}\times\R^3_{\hat x}$ the matrices $(B_j)_{p q}=\delta_{0 p}\delta_{j q}+\delta_{0 q}\delta_{j p}$, $0\leq p,q\leq 3$, and let $\bfB=(B_1,B_2,B_3)$. Given a vector $\hat\bfw\in\R^3$, we define by $L(\hat\bfw)=\exp(\hat\bfw\cdot\bfB)\in SO(1,3)$ the \emph{Lorentz boost with rapidity $\hat\bfw$}.
\end{definition}

Identifying $(T_\fp M,\hat\eta_\fp)\cong(\R^{1+3}_{\hat t,\hat x},\hat{\ubar g})$ (and thus $\hat X_\fp^\circ\cong\R^3_{\hat x}$), there exists a unique Lorentz boost $L(\hat\bfw)\colon T_\fp M\to T_\fp M$, where $\hat\bfw\in\hat X_\fp^\circ$, so that
\begin{equation}
\label{EqGKBoost}
  L(\hat\bfw)(\hat X_\fp^\circ) = T_\fp X.
\end{equation}
(Equivalently, $L(\hat\bfw)$ maps $T$ from~\eqref{EqGKT} to the future unit normal of $X$ at $\fp$.) While $T_\fp X$ may not be globally spacelike for a Kerr metric $\hat g_\fp=(\Phi_\fp)_*\hat g_{b(\fp)}$ (using~\eqref{EqGKNaivePhi}), it is spacelike near its boundary at infinity where $\hat g_\fp=\hat\eta_\fp+\cO(\rho_\circ)$. Therefore, the initial data $(\hat\gamma,\hat k)$ of $\hat g_\fp$ are well-defined on the complement of a large enough ball in $T_\fp X$, and there they are boosted Kerr initial data with rapidity $\hat\bfw$ determined by~\eqref{EqGKBoost}. See Figure~\ref{FigGKBoost}.

\begin{figure}[!ht]
\centering
\includegraphics{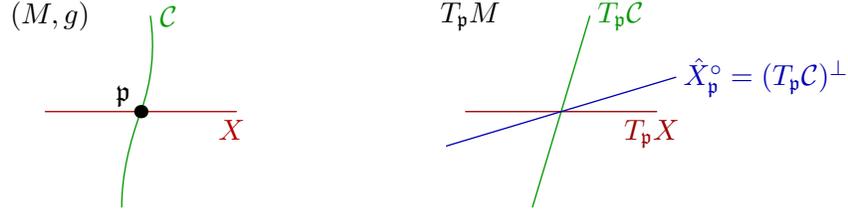}
\caption{\textit{On the left:} the timelike curve $\cC$ and the spacelike hypersurface $X$ inside of $(M,g)$. \textit{On the right:} geometry of $T_\fp M$: the hypersurfaces $\hat X_\fp^\circ$ and $T_\fp X$ (which are related via a Lorentz boost) are both spacelike with respect to the Minkowski metric $\hat\eta_\fp$ induced by $(M,g)$ at $\hat M_\fp$. With respect to a subextremal Kerr metric $(\Phi_\fp)_*\hat g_b$ on $T_\fp M$ (sans the black hole interior), the hypersurface $\hat X_\fp^\circ$ is still spacelike, whereas $T_\fp X$ is spacelike only in the complement of a large enough spatial coordinate ball.}
\label{FigGKBoost}
\end{figure}

There are two ways by which one can extend the lift of $\{0\}\times X$ to $\wt M$ to a spacelike (for small $\eps>0$) hypersurface in $\wt M$ equipped with a total family with respect to $g$ and with $\hat M_\fp$-model $\hat g_\fp$: either one modifies $\wt X\subset\wt M$ near a compact subset of $\hat M_\fp^\circ$ so that the intersection of its lift to $[\wt M;\hat M_\fp]$ (cf.\ the discussion around~\eqref{RmkGffMultEps}) with the front face becomes spacelike for $\hat g_\fp$; or one passes to a suitable modification $X_0$ of $X$ with the property that $\wt{X_0}$ is spacelike for such $\wt g$. We implement the latter approach:

\begin{lemma}[Modified spacelike hypersurface]
\label{LemmaGKMod}
  Let $M,\cC,g$ be as in~\textnormal{\S}\usref{SsGL}. Let $X\subset M$ be spacelike with $X\cap\cC=\{\fp\}$. Let $\cU\subset X$ be an open neighborhood of $\fp$. Then there exists a spacelike hypersurface $X_0\subset M$ so that $X_0\setminus\cU=X\setminus\cU$, while near $\fp$ the hypersurface $X_0$ is the level set of a Fermi normal coordinate function $t$; in particular, $T_\fp X_0=(T_\fp\cC)^\perp$ (with respect to $g|_\fp$). In particular, if $\wt g$ is a total family with $\hat M_\fp$-model given by a Kerr metric $(\Phi_\fp)_*\hat g_b$, and if $\cK\subset M$ is compact, then $\wt{X_0}\subset\wt M$ is a spacelike hypersurface in $(\wt M,\wt g)$ over $\wt\upbeta^{-1}([0,\eps_0)\times\cK)$ when $\eps_0>0$ is sufficiently small.
\end{lemma}

See Figure~\ref{FigGKMod}.

\begin{figure}[!ht]
\centering
\includegraphics{FigGKMod}
\caption{Illustration of Lemma~\ref{LemmaGKMod}: the spacelike hypersurface $X$ inside $(M,g)$ is modified near its intersection point $\fp$ with $\cC$ to a new spacelike hypersurface $X_0$ for which $T_\fp X_0=(T_\fp\cC)^\perp$ (and indeed $X_0$ is the level set of a Fermi normal coordinate $t$ near $\fp$), which is thus spacelike for Kerr metrics $(\Phi_\fp)_*\hat g_b$.}
\label{FigGKMod}
\end{figure}

\begin{proof}[Proof of Lemma~\usref{LemmaGKMod}]
  On $(1+1)$-dimensional Minkowski space with coordinates $(t,x)$, the idea is the following: we want to modify the spacelike hypersurface $t=v x$, $|v|<1$, by a spacelike hypersurface which is tangent to the $x$-axis at the origin. We define the latter to be equal to $t=0$ for $|x|<1$, the (nonlinear) interpolation $t=\frac{|x|-1}{|x|}\frac{R_0}{R_0-1}v x$ for $1\leq\pm x\leq R_0$, and $t=v x$ for $|x|>R_0$. In order for this to be spacelike, we need $|\frac{v R_0}{R_0-1}|<1$, which holds true for sufficiently large $R_0$: if $|v|<1-2\delta$, then $R_0=\frac{1-\delta}{\delta}$ (or any larger value) works with a $\delta$ to spare. Upon scaling the resulting hypersurface in space and time, we obtain an arbitrarily localized (near $(0,0)$) modification which, moreover, can be smoothed out to give a smooth hypersurface.

  To prove the Lemma, we use Fermi normal coordinates $(t,x)$ along $\cC$, with $(t,x)=(0,0)$ at $\fp$. The hypersurface $X$ is then the graph $\{(f(x),x)\colon |x|<\delta\}$ of a smooth function $f\colon\R^3\to\R$ with $f(0)=0$; here $\delta>0$ is small enough so that the Fermi normal coordinates are valid in an open set $\cU$ with $[-C,C]\times\{|x|\leq\delta\}\subset\cU$ where $C=\sup_{|x|\leq\delta}|f(x)|$. The hypersurface $\hat X_\fp^\circ\subset T_\fp M$ is, in local coordinates, the tangent space at $(0,0)$ of $\{0\}\times\R^3$. Moreover, the inverse metric is $g^{-1}=-\pa_t^2+\pa_x^2+\cO(|x|)$. Upon shrinking $\delta>0$ further, we may assume that $|f'(0)|<1-2\delta$. Let $R_0=\frac{1-\delta}{\delta}$ and $R_1\in(R_0,R_0+1)$. Let $\chi\in\CI([0,\infty))$ be a smoothed-out version of the function which is equal to $0$ on $[0,1]$, equal to $1$ on $[R_1,\infty)$, and given by $\frac{r-1}{r}\frac{R_1}{R_1-1}$ on $[1,R_1]$; we can arrange that $\chi(r)=0$ for $r\leq 1$, further $\chi(r)=1$ for $r\geq R_0+1$, and

  \[
    0\leq\chi+r\chi'=(r\chi)'\leq\frac{R_0}{R_0-1}\quad \text{for}\quad 1\leq r\leq R_0+1.
  \]
  For $0<\eta<\frac{\delta}{R_0+1}$ chosen momentarily, we define
  \[
    X_0=\Bigl\{\Bigl(\chi\Bigl(\frac{|x|}{\eta}\Bigr)f(x),x\Bigr)\colon|x|<\delta\Bigr\}\cup(X\setminus\{|x|\leq\delta\}).
  \]
  We claim that for sufficiently small $\eta>0$, the differential of $t-\chi(|x|/\eta)f(x)$ is timelike; we only need to check this for $\eta\leq|x|\leq\eta(R_0+1)$. Passing to rescaled coordinates $(\frac{t}{\eta},\frac{x}{\eta})$ and letting $\eta\searrow 0$, it suffices to check this in the case that $f(x)=x\cdot v$ is linear with $|v|<1$ (using the Euclidean inner product and norm), with $g^{-1}=-\pa_t^2+\pa_x^2$, and with the scaling parameter $\eta$ absent. That is, we need to check that $|\nabla(\chi(|x|)x\cdot v)|<1$ (using the Euclidean gradient and norm). But the gradient has norm
  \[
    \Bigl|\chi(|x|)v + |x|\chi'(|x|)\frac{x}{|x|}\Bigl(\frac{x}{|x|}\cdot v\Bigr)\Bigr| \leq \bigl(\chi(|x|)+|x|\chi'(|x|)\bigr)|v| < \frac{R_0}{R_0-1}(1-2\delta) = 1-\delta < 1,
  \]
  as required. The final conclusion of the Lemma is now a consequence of Corollary~\ref{CorGLInitial}\eqref{ItGLInitialSpace}.
\end{proof}

\section{The Einstein vacuum equations and their linearization}
\label{SE}

In this section, we discuss algebraic and analytic properties of the Einstein vacuum equations, with cosmological constant $\Lambda\in\R$,
\begin{equation}
\label{EqEEq}
  \Ric(g) - \Lambda g = 0,
\end{equation}
on a Lorentzian manifold $(M,g)$ and on the total gluing spacetime $\wt M$, defined as in Definition~\ref{DefGTot} with $\cC\subset M$ a smooth embedded timelike curve diffeomorphic to the real line.

\subsection{Constraint equations}
\label{SsECE}

Let $X$ be a 3-dimensional manifold. For a Riemannian metric $\gamma\in\CI(X;S^2 T^*X)$ and a symmetric 2-tensor $k\in\CI(X;S^2 T^*X)$, the \emph{constraints map} is
\begin{equation}
\label{EqMConstraints}
  P(\gamma,k;\Lambda) = \bigl(P_1(\gamma,k;\Lambda), P_2(\gamma,k)\bigr) := \bigl( R_\gamma - |k|_\gamma^2 + (\tr_\gamma k)^2 - 2\Lambda, \delta_\gamma k+\dd\tr_\gamma k\bigr),
\end{equation}
where we recall that $R_\gamma=\tr_\gamma\Ric(\gamma)$ is the scalar curvature and $\delta_\gamma$ is the negative divergence. For any Lorentzian metric $g$ on $M$ and its initial data $(\gamma,k)$ at a spacelike hypersurface $X\subset M$ with unit normal $\nu\in\CI(X;T_X M)$---i.e.\ $\gamma(V,W)=g(V,W)$ and $k(V,W)=g(\nabla_V\nu,W)$, $V,W\in T X$, are the first and second fundamental form of $X$, respectively---we have
\begin{equation}
\label{EqMConstraintsEin}
\begin{split}
  P_1(\gamma,k;\Lambda) &= 2\sfG_g(\Ric(g)-\Lambda g)(\nu,\nu) \in \CI(X), \\
  P_2(\gamma,k) &= -\sfG_g(\Ric(g)-\Lambda g)(\nu,-) \in \CI(X;T^*X),
\end{split}
\end{equation}
where $\sfG_g=I-\half g\tr_g$. (Note that $P_2$ is indeed independent of $\Lambda$.) In this sense, the constraint equations are equivalent to the validity of part of the Einstein vacuum equations at $X$. In particular, when $\Ric(g)-\Lambda g=0$, then $P(\gamma,k;\Lambda)=0$.

We say that $(X,\gamma,k)$ \emph{has no KIDs} (`Killing Initial Data sets') in $\cU^\circ$ if the kernel of the formal adjoint $(D_{(\gamma,k)}P)^*$ on $\CI(\cU;\ul\R\oplus T^*_\cU X)$ is trivial, where $\cU=\ol{\cU^\circ}$. (We remark that the linearization of $P$ is independent of $\Lambda$.) The operator $(D_{(\gamma,k)}P)^*$ is overdetermined elliptic, and its kernel is automatically finite-dimensional. Furthermore, elements of $\ker(D_{(\gamma,k)}P)^*$ are smooth, and they are determined by their restriction to any nonempty open subset of $\cU^\circ$. (See e.g.\ \cite[Lemma~4.3]{HintzGlueID}.) By \cite{MoncriefLinStabI} or \cite[Lemma~2.2]{FischerMarsdenMoncriefEinstein}, the space of KIDs for $(X,\gamma,k)$ in $\cU^\circ$ can be identified with the space of Killing vector fields in the maximal globally hyperbolic development of $(\cU^\circ,\gamma|_{\cU^\circ},k|_{\cU^\circ})$. In particular, if a spacetime $(M,g)$ does not admit any Killing vector fields in some globally hyperbolic subset with Cauchy surface $\cU^\circ$, then the initial data of $g$ at $\cU^\circ$ have no KIDs (and vice versa).

\subsection{Linearization and gauge fixing}
\label{SsELin}

The linearization of the Einstein vacuum operator $g\mapsto\Ric(g)-\Lambda g$ in $g$ is given by the second order differential operator
\begin{equation}
\label{EqELinDRic}
  D_g\Ric-\Lambda = \frac12\Box_g - \delta_g^*\delta_g\sfG_g + \sR_g - \Lambda,
\end{equation}
where $(\delta_g h)_\mu=-h_{\mu\nu;}{}^\nu$ and $(\delta_g\omega)_{\mu\nu}=\half(\omega_{\mu;\nu}+\omega_{\nu;\mu})$, and
\[
  (\sR_g h)_{\mu\nu} = R_{\kappa\mu\nu\lambda}h^{\kappa\lambda} + \frac12(\Ric(g)_{\mu\kappa}h_\nu{}^\kappa + \Ric(g)_{\kappa\nu}h_\mu{}^\kappa).
\]
We use the convention $R_{\kappa\mu\nu\lambda}=\la\pa_\kappa,\Riem(g)(\pa_\nu,\pa_\lambda)\pa_\mu\ra$ where $\Riem(g)(\pa_\nu,\pa_\lambda)=[\nabla_\nu,\nabla_\lambda]$. See for example~\cite[Equation~(2.4)]{GrahamLeeConformalEinstein}.

From now on, we shall assume that $g$ solves~\eqref{EqEEq}, and $(M,g)$ is globally hyperbolic. The linearized Einstein operator is analytically rather ill-behaved. The operator $D_g\Ric-\Lambda\in\Diff^2(M;S^2 T^*M)$ has an infinite-dimensional kernel which contains all \emph{pure gauge} metric perturbations, i.e.\ all symmetric 2-tensors which are of the form $\delta_g^*\omega$ where $\omega$ is any 1-form on $M$ (and $\omega$ may even be a distribution). The formal adjoint
\begin{equation}
\label{EqELinDRicAdj}
  (D_g\Ric-\Lambda)^* = \sfG_g\circ(D_g\Ric-\Lambda)\circ\sfG_g
\end{equation}
likewise has an infinite-dimensional kernel which contains all \emph{dual pure gauge} perturbations $\sfG_g\delta_g^*\omega$, with $\omega$ again an arbitrary 1-form. Therefore, the cokernel of $D_g\Ric$ (acting on any reasonable space of tensors) is infinite-dimensional. Towards a characterization of the cokernel, recall the second Bianchi identity $\delta_g\sfG_g\Ric(g)=0$ (and thus also $\delta_g\sfG_g(\Ric(g)-\Lambda g)=0$), which holds for \emph{all} metrics $g$ regardless of the validity of the Einstein vacuum equations. If $g$ solves~\eqref{EqEEq}, then upon linearizing this identity in $g$ we also obtain the \emph{linearized second Bianchi identity}
\begin{equation}
\label{EqELin2ndBianchi}
  \delta_g\sfG_g\circ(D_g\Ric-\Lambda) = 0.
\end{equation}
Thus, symmetric 2-tensors in the range of $D_g\Ric-\Lambda$ necessarily lie in $\ker\delta_g\sfG_g$. The converse is not always true, and indeed generally fails when $(M,g)$ admits nontrivial Killing vectors, as demonstrated in \cite{HintzLinEin}.

A typical equation we need to solve in the course of our gluing construction is\footnote{While often, e.g.\ in linear stability problems, only considers the case $f=0$, we need to allow general $f$ here; see the discussion in~\S\ref{SssIPfF}.}
\begin{equation}
\label{EqELinNonFixed}
  (D_g\Ric-\Lambda)h=f,\qquad f\in\CI\cap\ker\delta_g\sfG_g.
\end{equation}
We explain the procedure for solving this equation, following \cite{HintzLinEin}, as we will encounter it in a variety of contexts later on. In order to pass to a hyperbolic equation admitting a well-posed forcing (or initial value) problem, one adds a gauge fixing term; we shall work here with generalized harmonic coordinate gauges, which take the form $\delta_g\sfG_g h-\theta=0$ where $\theta\in\CI(M;T^*M)$ is fixed but may be chosen arbitrarily. The gauge-fixed linearized Einstein equation is then
\begin{equation}
\label{EqELinFixed}
  (D_g\Ric-\Lambda)h + \delta_g^*(\delta_g\sfG_g h-\theta) = \Bigl(\frac12\Box_g + \sR_g - \Lambda\Bigr)h - \delta_g^*\theta = f.
\end{equation}
Upon specifying Cauchy data
\begin{equation}
\label{EqEhData}
  h_0=h|_X,\qquad
  h_1=\nabla_\nu h|_X
\end{equation}
for $h$ at some Cauchy hypersurface $X\subset M$ with future unit normal $\nu$, there exists a unique solution $h$ of~\eqref{EqELinFixed}. This solves~\eqref{EqELinNonFixed} provided the gauge 1-form $\eta:=\delta_g\sfG_g h-\theta$ vanishes. By virtue of~\eqref{EqELin2ndBianchi}, $\eta$ satisfies the decoupled equation
\[
  \Box^\cC_g \eta = 0,\qquad \Box^\cC_g := 2\delta_g\sfG_g\circ\delta_g^*,
\]
where $\Box^\cC_g$ is a wave operator on 1-forms (i.e.\ its principal symbol is scalar and equal to that of $\Box_g$). Thus, one needs to ensure that $\eta$ has trivial Cauchy data $\eta|_X,\nabla_\nu\eta|_X\in\CI(X;T^*_X M)$. (We stress that these are 1-forms on $M$ defined over $X$, and \emph{not} 1-forms on $X$.) But if $\eta|_X=0$, then the vanishing of $\nabla_\nu\eta|_X$ is easily seen to be equivalent to that of $(\sfG_g\delta_g^*\eta)(\nu,-)\in\CI(X;T^*_X M)$. Now, applying $\sfG_g$ to the PDE~\eqref{EqELinFixed} and plugging $\nu$ into the first slot of the resulting equation at $X$, we obtain
\[
  D_{(\gamma,k)}P(\dot\gamma,\dot k) + \bigl(\sfG_g\delta_g^*\eta\bigr)(\nu,-) = (\sfG_g f)(\nu,-).
\]
Here, $\dot\gamma$ and $\dot k$ are the linearized first and second fundamental form corresponding to the linearized Cauchy data $h_0$, $h_1$. Requiring the second term on the left to vanish is thus equivalent to a PDE for $(\dot\gamma,\dot k)$.

To summarize, in our approach to solving~\eqref{EqELinNonFixed}, we need to choose $h_0,h_1,\theta$ so that
\begin{align}
\label{EqEIDConstraints}
  D_{(\gamma,k)}P(\dot\gamma,\dot k)&=(\sfG_g f)(\nu,-), \\
\label{EqEIDeta}
  \eta=\delta_g\sfG_g h-\theta&=0
\end{align}
at $X$ where $h$ is related to $h_0,h_1$ via~\eqref{EqEhData}. \emph{If}\footnote{We stress that \emph{equation~\eqref{EqEIDConstraints} may not be solvable} when $(X,\gamma,k)$ has nontrivial KIDs.} one can solve the first equation, then one can pick $h_0,h_1$ which induce the data $\dot\gamma,\dot k$. Such $h_0,h_1$ always exist (see Remark~\ref{RmkEIDh0h1} below), but they are not unique. One may then take $\theta\in\CI(M;T^*M)$ to be any 1-form with value $\delta_g\sfG_g h$ at $X$; note here that the 1-form $(\delta_g\sfG_g h)|_X\in\CI(X;T^*_X M)$ only depends on $h_0,h_1$. This arranges~\eqref{EqEIDConstraints}. Having thus fixed $\theta,h_0,h_1$, the solution $h$ of the gauge-fixed equation~\eqref{EqELinFixed} satisfies $\eta=0$ on $M$ and thus also~\eqref{EqELinNonFixed}.

\begin{rmk}[Existence of Cauchy data inducing linearized initial data]
\label{RmkEIDh0h1}
  In order to build Cauchy data $(h_0,h_1)$ from $(\dot\gamma,\dot k)$, note first that on a small open neighborhood $\cM\subset\R\times X$ of $\{0\}\times X$, the map $\cM\ni(s,p)\mapsto\exp_p(s\nu)$ is a diffeomorphism onto a neighborhood of $X$. Splitting the tangent bundle of $M$ in this collar neighborhood into $\la\pa_s\ra\oplus T X$, the metric $g$ then takes the block form
  \[
    g(s,p) = \begin{pmatrix} -1 & 0 \\ 0 & \gamma(s) \end{pmatrix},\qquad \gamma(0)=\gamma,\quad \gamma'(0)=2 k.
  \]
  We may thus take
  \[
    h_0 = \begin{pmatrix} 0 & 0 \\ 0 & \dot\gamma \end{pmatrix},\qquad
    h_1 = \begin{pmatrix} 0 & 0 \\ 0 & 2\dot k \end{pmatrix}.
  \]
\end{rmk}

\subsection{Linearization on the total gluing spacetime}
\label{SsELT}

In the notation of Definition~\ref{DefGLTot}, let $\wt g$ denote a $(\hat\cE,\cE)$-smooth total family (relative to $M,\cC,g$) on the total gluing spacetime $\wt M\setminus\wt K^\circ$; we write $\hat g=(\hat g_p)_{p\in\cC}$ for its $\hat M$-model. Corollary~\ref{CorGLCurvature} implies that
\begin{equation}
\label{EqELTErr}
\begin{split}
  \Err &:= \Ric(\wt g)-\Lambda\wt g \in \wt\upbeta^*\CI(M;S^2\wt T^*M)+\cA_\phg^{(\N_0\cup\hat\cE)-2,\cE}(\wt M\setminus\wt K^\circ;S^2\wt T^*\wt M), \\
    &\quad \sfe^{-1}\bigl((\eps^2\Err)|_{\hat M}\bigr)=\Ric(\hat g),\qquad
    \Err|_{M_\circ}=\upbeta_\circ^*(\Ric(g)-\Lambda g).
\end{split}
\end{equation}
We fix cutoff functions $\hat\chi,\chi_\circ\in\CI(\wt M)$ to collar neighborhoods of $\hat M,M_\circ$ as in~\eqref{EqGRelCutoffs}.

\begin{lemma}[Linearization and its model operators]
\label{LemmaELTLin}
  The linearization of the Einstein vacuum operator $\fg\mapsto\Ric(\fg)-\Lambda\fg$ at $\wt g$ (defined on each fiber $\wt M_\eps$, $\eps>0$, of $\wt M$ as the linearization at $\wt g|_{\wt M_\eps}$) satisfies
  \begin{align*}
    D_{\wt g}\Ric-\Lambda &\in \wt\upbeta^*\bigl(\Diff^2(M;S^2 T^*M)\bigr) + \cA_\phg^{(\N_0\cup\hat\cE)-2,\cE}\Diffse^2(\wt M\setminus\wt K^\circ;S^2\wt T^*\wt M) \\
      &\subset \cA_\phg^{(\N_0\cup\hat\cE)-2,\N_0\cup\cE}\Diffse^2(\wt M\setminus\wt K^\circ;S^2\wt T^*\wt M).
  \end{align*}
  Its normal operators are
  \[
    \sfe\circ N_{\hat M}\bigl(\eps^2(D_{\wt g}\Ric-\Lambda)\bigr)\circ\sfe^{-1}=D_{\hat g}\Ric,\qquad
    N_{M_\circ}(D_{\wt g}\Ric-\Lambda) = \upbeta^*(D_g\Ric-\Lambda).
  \]
\end{lemma}
\begin{proof}
  The structure of $D_{\wt g}\Ric-\Lambda$ and the identification of its normal operators follows immediately from the formula~\eqref{EqELinDRic}, the expressions for the normal operators in Lemma~\ref{LemmaGLNabla} and Corollary~\ref{CorGLCurvature}, and the multiplicativity of the normal operator maps.
\end{proof}

We next study the extent to which the linearization provides an approximation to the nonlinear Einstein operator. The part of the following result concerning $M_\circ$ will suffice for our purposes; at $\hat M$ on the other hand we will need significantly more refined descriptions in~\S\ref{SF}.

\begin{prop}[Accuracy of the linearization]
\label{PropELTAcc}
  Let $\hat\cF',\cF'\subset\C\times\N_0$ denote two index sets with $\Re\hat\cF',\Re\cF'>0$, and let $\wt h\in\cA_\phg^{\hat\cF,\cF}(\wt M\setminus\wt K^\circ;S^2\wt T^*\wt M)$. Denote by $\hat\cF,\cF$ the nonlinear closures of $\hat\cF',\cF'$. Then
  \begin{equation}
  \label{EqELTAccLin}
    L := \bigl(\Ric(\wt g+\wt h)-\Lambda(\wt g+\wt h)\bigr) - \bigl(\Ric(\wt g)-\Lambda\wt g\bigr) \subset \cA_\phg^{\hat\cF+(\N_0\cup\hat\cE)-2,\cF+(\N_0\cup\cE)}(\wt M\setminus\wt K^\circ;S^2\wt T^*\wt M).
  \end{equation}
  To leading order at $\hat M$ and $M_\circ$, the tensor $L$ has the following description.
  \begin{enumerate}
  \item\label{ItELTAccMhat}{\rm (Accuracy near $M_\circ$.)} Let us regard $g$ as an $\eps$-independent section of $S^2\wt T^*\wt M'$ over $\wt M'$ in order to define $D_g\Ric-\Lambda$ as an $\eps$-independent differential operator on sections of $S^2\wt T^*\wt M'$, and lift it to $\wt M$. Then
    \begin{equation}
    \label{EqELTAccMcirc}
      L - \chi_\circ(D_g\Ric-\Lambda)(\chi_\circ\wt h) \in \cA_\phg^{\hat\cF+(\N_0\cup\hat\cE)-2,(\cF+\cE)\cup 2\cF}(\wt M\setminus\wt K^\circ;S^2\wt T^*\wt M).
    \end{equation}
  \item\label{ItELTAccMcirc}{\rm (Accuracy near $\hat M$.)} Extend the zero energy operator family $\wh{D_{\hat g}\Ric}(0)$ (which we recall is a vertical b-differential operator on $\hat M$) to an $\eps$-independent operator on $[0,1)_\eps\times\hat M$ and identify it on $\supp\hat\chi$ with an operator on $\wt M$ via the (coordinate-dependent) diffeomorphism~\eqref{EqGRelHat} from Lemma~\usref{LemmaGRel}. Then
    \begin{equation}
    \label{EqELTAccMhat}
      L - \eps^{-2}\hat\chi\sfe^{-1}\Bigl(\wh{D_{\hat g}\Ric}(0)\bigl(\hat\chi\sfe(\wt h)\bigr)\Bigr) \in \cA_\phg^{\hat\cF+((\N_0+1)\cup\hat\cE)-2,\cF+(\N_0\cup\cE)}(\wt M\setminus\wt K^\circ;S^2\wt T^*\wt M).
    \end{equation}
  \end{enumerate}
\end{prop}
\begin{proof}
  Since $L$ is polyhomogeneous by Corollary~\ref{CorGLCurvature}, it suffices to work near the interiors of $M_\circ$ and $\hat M$, i.e.\ away from the corner of $\wt M$, in order to determine the index sets in~\eqref{EqELTAccLin}--\eqref{EqELTAccMhat}.

  Let us first work near $(M_\circ)^\circ$, where $\wt g\equiv g\bmod\cA_\phg^\cE$, where we abbreviate $\cA_\phg^\cE:=\cA_\phg^\cE([0,1);\CI(M\setminus\cC;S^2 T^*M))$;\footnote{This is the same as the space $\cA_\phg^\cE(\wt M'\setminus([0,1)\times\cC);S^2\wt T^*\wt M')$.} and $\wt h\in\cA_\phg^\cF$. Then
  \[
    (\wt g+\wt h)^{-1}-\wt g^{-1}\equiv-\wt g^{-1}\wt h\wt g^{-1}\bmod\cA_\phg^{2\cF+(\N_0\cup\cE)}\subset\cA_\phg^{\cF+(\N_0\cup\cE)},
  \]
  since, in local coordinates, finite products of components of $\wt g^{-1}$ lie in $\cA_\phg^{\N_0\cup\cE}$, while $k$-fold products of components of $\wt h$, with $k\geq 2$, lie in $\cA_\phg^{k\cF}\subset\cA_\phg^{2\cF}$. This verifies the second index set in~\eqref{EqELTAccLin}, and moreover gives
  \begin{align*}
    L - (D_{\wt g}\Ric-\Lambda)(\wt h) &= \bigl(\Ric(\wt g+\wt h)-\Lambda(\wt g+\wt h)\bigr)-\bigl(\Ric(\wt g)-\Lambda\wt g\bigr) - (D_{\wt g}\Ric-\Lambda)(\wt h) \\
      &\in \cA_\phg^{2\cF+(\N_0\cup\cE)}.
  \end{align*}
  Note then further that $D_{\wt g}\Ric-D_g\Ric\in\cA_\phg^\cE([0,1);\Diff^2(M\setminus\cC;S^2 T^*M))$, which when evaluated on $\wt h$ gives an element of $\cA_\phg^{\cF+\cE}$. Since $(2\cF+(\N_0\cup\cE))\cup(\cF+\cE)=(2\cF\cup(2\cF+\cE))\cup(\cF+\cE)=(\cF+\cE)\cup 2\cF$, this verifies the second index set in~\eqref{EqELTAccMcirc}.

  Regarding the second index set of~\eqref{EqELTAccMhat}, we note that $\eps^{-2}\hat\chi\sfe^{-1}\circ\wh{D_{\hat g}\Ric}(0)\circ\hat\chi\sfe\in\eps^{-2}\rho_\circ^2\Diffse^2(\wt M\setminus\wt K^\circ;S^2\wt T^*\wt M)\subset\hat\rho^{-2}\Diffb^2$ maps $\wt h$ near $(M_\circ)^\circ$ into an element of $\cA_\phg^\cF$; it then remains to use~\eqref{EqELTAccLin} to conclude.

  We next turn to a neighborhood of $\hat M^\circ$ where we use the smooth coordinates $\eps\geq 0$, $t\in I$, and $\hat x=\frac{x}{\eps}\in\R^3$. It suffices to evaluate~\eqref{EqELTAccLin}--\eqref{EqELTAccMhat} at a single fiber $\hat M_p^\circ$, $p=c(t_0)$, of $\hat M^\circ$; we do this in by passing to the `fast time' coordinate $\hat t:=\frac{t-t_0}{\eps}$. Thus,
  \begin{equation}
  \label{EqELTAccFastCoord}
    \eps\geq 0,\qquad
    \hat t\in\R,\qquad
    \hat x\in\R^3
  \end{equation}
  are local coordinates near the interior of the front face of $[\wt M;\hat M_p]$, and the map $\sfe\colon\dd z^\mu\,\dd z^\nu\mapsto\dd\hat z^\mu\,\dd\hat z^\nu$, $z=(t,x)$, $\hat z=(\hat t,\hat x)$, is given by multiplication by $\eps^2$; cf.\ Remark~\ref{RmkGffMultEps}. We thus consider
  \[
    \eps^2\sfe L = \bigl(\Ric(\sfe\wt g+\sfe\wt h)-\eps^2\Lambda(\sfe\wt g+\sfe\wt h)\bigr) - \bigl(\Ric(\sfe\wt g)-\eps^2\Lambda\sfe\wt g\bigr).
  \]
  By Definition~\ref{DefGLTot}, we have
  \[
    \sfe\wt g\equiv\hat g_p\bmod\cA_\phg^{(\N_0+1)\cup\hat\cE}:=\cA_\phg^{(\N_0+1)\cup\hat\cE}\bigl([0,1)_\eps;\CI(\R^4\setminus(\R\times\hat K_p^\circ);S^2 T^*\R^4)\bigr)
  \]
  and $\sfe\wt h\in\cA_\phg^{\hat\cF}$. Carefully note also the improved regularity
  \begin{equation}
  \label{EqELTAccThatReg}
    \pa_{\hat t}^j(\sfe\wt h)\in\cA_\phg^{\hat\cF+j},\qquad j\in\N,
  \end{equation}
  of $\wt h$ in the fast time variable; this follows from $\pa_{\hat t}=\eps\pa_t$. Therefore,
  \[
    (\sfe\wt g+\sfe\wt h)^{-1}-\sfe\wt g^{-1}\equiv-\sfe \wt g^{-1}\wt h\wt g^{-1}\bmod\cA_\phg^{2\hat\cF+(\N_0\cup\hat\cE)} \subset \cA_\phg^{\hat\cF+(\N_0\cup\hat\cE)}
  \]
  and thus
  \[
    \eps^2\sfe L \in \cA_\phg^{\hat\cF+(\N_0\cup\hat\cE)},\qquad
    \eps^2\sfe L-(D_{\sfe\wt g}\Ric-\eps^2\Lambda)(\sfe\wt h)\in\cA_\phg^{2\hat\cF+(\N_0\cup\hat\cE)}
  \]
  similarly to above. But in view of~\eqref{EqELTAccThatReg},
  \[
    D_{\sfe\wt g}\Ric(\sfe\wt h) \equiv \wh{D_{\hat g_p}\Ric}(0)(\sfe\wt h) \bmod \cA_\phg^{\hat\cF+((\N_0+1)\cup\hat\cE)},
  \]
  and $\eps^2\Lambda(\sfe\wt h)\in\cA_\phg^{\hat\cF+2}$ is of even lower order. This gives~\eqref{EqELTAccMhat}.

  Finally, the verification of the $\hat M$-index set of~\eqref{EqELTAccMcirc} uses that $\chi_\circ(D_g\Ric-\Lambda)\chi_\circ\in\chi_\circ r^{-2}\Diffse^2(\wt M;S^2\wt T^*\wt M)$ maps $\wt h$ into $\cA_\phg^{\hat\cF-2}$ near $\hat M^\circ$.
\end{proof}

\section{Setup and statement of the main result}
\label{SM}

\begin{definition}[Gluing data]
\label{DefMData}
  \emph{Gluing data} are a tuple $(M,g,\fp,v,\bhm,\bha,\Lambda)$ with the following properties:
  \begin{enumerate}
  \item\label{ItMDataM} $(M,g)$ is a smooth open globally hyperbolic Lorentzian manifold which satisfies the Einstein vacuum equations with cosmological constant $\Lambda\in\R$, that is,
    \[
      \Ric(g) - \Lambda g = 0;
    \]
  \item\label{ItMDataGeod} $\fp\in M$, and $v\in T_\fp M$ is a future timelike unit vector;
  \item\label{ItMDataKerr} $\bhm>0$, $\bha\in T_\fp M$, $\bha\perp v$, $|\bha|<\bhm$ (so $\bhm,|\bha|$ are subextremal Kerr black hole parameters);
  \item\label{ItMDataKID} there exists a precompact connected open neighborhood $\cU_M^\circ\subset M$ of $\fp$ so that $(\cU_M^\circ,g|_{\cU_M^\circ})$ does not have any nontrivial Killing vector fields.
  \end{enumerate}
\end{definition}

(Kerr(--de~Sitter) spacetimes violate~\eqref{ItMDataKID}; we discuss the modifications required to handle this case in~\S\ref{SX}.) Given such gluing data, we shall fix the following further objects.
\begin{enumerate}
\setcounter{enumi}{4}
\item\label{ItMC} We write $c\colon I\to M$ (with $I\subseteq\R$ an open interval) for the maximally extended arc-length parameterized timelike geodesic with $c(0)=\fp$ and $c'(0)=v$, and we denote by $\cC=c(I)\subset M$ its image. We denote by $\wt M=[[0,1)\times M;\{0\}\times\cC]$ the total gluing spacetime of Definition~\ref{DefGTot}, with boundary hypersurfaces $\hat M$ and $M_\circ$.
\item\label{ItMFermi} We fix Fermi normal coordinates
  \begin{equation}
  \label{EqMFermi}
    (t,x) \in \R\times\R^3,\qquad t\in I,\quad r=|x|<r_0(t),
  \end{equation}
  along $\cC$, with $\fp=(0,0)$; we require $r_0\in\CI(I)$ to satisfy $0<r_0(t)<\frac12$ for all $t\in I$. In these coordinates, identify $\bha\in(T_\fp\cC)^\perp\subset T_\fp M$ with a vector in $\R^3$, denoted $\bha$ still.
\item\label{ItMX} By $X\subset M$ we denote a Cauchy hypersurface with $X\cap\cC=\{\fp\}$ and $T_\fp X\perp T_\fp\cC$, and we let $\cU^\circ\subset X$ denote a smoothly bounded precompact connected open neighborhood of $\fp$ so that the domain of dependence of $\cU^\circ$ contains $\cU_M^\circ$. (See Remarks~\ref{RmkMXChoice} and \ref{RmkMKIDs} below.)
\item\label{ItMFamily} We let $\wt g_0\in\CI(\wt M\setminus\wt K^\circ;S^2\wt T^*\wt M)$ denote a $(\emptyset,\N_0+1)$-smooth total family as produced by Lemma~\ref{LemmaGKNaive}\eqref{ItGKNaiveGeod}; here
  \[
    \wt K=\{(\eps,t,x)\colon|\eps x|\leq\bhm\}
  \]
  is as in Definition~\ref{DefGLTot} for the choice $\hat K_p=\Phi_p(\hat K_{\bhm,\bha})\subset\hat M_p$ (with $\Phi_p$ given by Lemma~\ref{LemmaGKNaive}, and $\hat K_{\bhm,\bha}=\{\hat x\in\R^3\colon|\hat x|\leq\bhm\}$ as in Definition~\ref{DefGKModel}). We write $\hat g_p$ for the $\hat M_p$-model of $\wt g$, and $\hat g=(\hat g_p)_{p\in\cC}$ for the $\hat M$-model of $\wt g$. Thus, when expressing $\hat g$ in the frame $\dd\hat t$, $\dd\hat x$ (where $\hat t=\dd t(-)$, $\hat x=\dd x(-)$) on $T_p M$, it is equal to the Kerr metric $\hat g_{\bhm,\bha}=\hat g_{\bhm,\bha}(\hat x;\dd\hat t,\dd\hat x)$ and thus $t$-independent.
\item\label{ItMData} We write $(\gamma,k)$ for the initial data of $(M,g)$ at $X$; that is, $\gamma\in\CI(X;S^2 T^*X)$ is the induced metric, and $k\in\CI(X;S^2 T^*X)$ is the second fundamental form of $X\subset M$.
\item\label{ItMDataTilde} We define the total gluing space $\wt X$ as a subset of $\wt M$ as in~\S\ref{SsGHyp}.
\end{enumerate}

\begin{rmk}[Choice of $X$]
\label{RmkMXChoice}
  By Lemma~\ref{LemmaGKMod}, we can always modify a given Cauchy hypersurface intersecting $\cC$ at $\fp$ in an arbitrarily small neighborhood of $\fp$ to a hypersurface satisfying the stated orthogonality condition; if we make this modification inside of a convex neighborhood of $\fp$, the modified hypersurface is guaranteed to be a Cauchy hypersurface still.
\end{rmk}

\begin{rmk}[Existence of $\cU^\circ$; KIDs]
\label{RmkMKIDs}
  As the set $\cU^\circ\subset X$ in item~\eqref{ItMX} above, we may take any smoothly bounded precompact connected open set containing the set $\cV\subset X$ of the intersection points with $X$ of all maximally extended causal curves emanating from a point in $\ol{\cU_M^\circ}$. Note that since $\ol{\cU_M^\circ}$ is compact and $(M,g)$ is globally hyperbolic, the set $\cV$ is compact. As recalled in~\S\ref{SsECE}, we then conclude that $(X,\gamma,k)$ has no KIDs in $\cU^\circ$. We also recall from \cite[Remark~4.13]{HintzGlueID} that for all sufficiently small $\delta>0$, the nonexistence of KIDs persists on the subset of $\cU^\circ$ consisting of all points with distance larger than $\delta>0$ from $\pa\ol{\cU^\circ}$.
\end{rmk}

\begin{thm}[Main theorem, precise version]
\label{ThmM}
  Given gluing data $(M,g,\fp,v,\bhm,\bha,\Lambda)$, define $\cC,X\subset M$ and $\wt K\subset\wt M$ as above. Then there exists a $(\hat\cE,\cE)$-smooth (with index sets $\hat\cE\subset(3+\N_0)\times\N_0$ and $\cE\subset(1,0)_+\cup((3+\N_0)\times\N_0)$) total family $\wt g$ on $\wt M\setminus\wt K^\circ$ (see Definition~\usref{DefGLTot}) with respect to $g$ (i.e.\ $\wt g|_{M_\circ}=\upbeta_\circ^*g$) and with $\hat M$-model equal to $\hat g$ (i.e.\ with $\hat M_p$-model equal to the Kerr metric $\hat g_b=\hat g_{\bhm,\bha}$ for all $p\in\cC$, as in point~\eqref{ItMFamily} above) so that
  \begin{enumerate}
  \item $\wt g$ is a \emph{formal solution} of the Einstein vacuum equations with cosmological constant $\Lambda$ at $\eps=0$ and at $\wt X$. That is, in a neighborhood of $(\hat M\setminus\wt K)\cup M_\circ$,
    \begin{equation}
    \label{EqM}
      \Ric(\wt g)-\Lambda\wt g
    \end{equation}
    is a smooth section of $S^2\wt T^*\wt M$ over $\wt M\setminus\wt K^\circ$ which vanishes to infinite order at $\hat M$, $M_\circ$, and $\wt X\subset\wt M$;
  \item $\wt g$ is equal to $g$ outside the domain of influence of a compact subset of $\cU^\circ$;
  \item $\sfe\wt g$ is equal to the Kerr metric $\hat g_b$ at $\hat M$ modulo quadratically vanishing error terms, in the sense that in Fermi normal coordinates around $\cC$, the components of
    \[
      \sfe\wt g(\eps,t,\hat x;\dd\hat t,\dd\hat x) - \hat g_b(\hat x;\dd\hat t,\dd\hat x)
    \]
    in the frame $\dd\hat t,\dd\hat x$ vanish quadratically at $\eps=0$.
  \end{enumerate}
\end{thm}

The set $\wt K$ which we excise lies inside the interior of the small Kerr black hole glued along $\hat M$, and is indeed disjoint from the event horizon at $\eps=0$. See Figure~\ref{FigM}.

\begin{figure}[!ht]
\centering
\includegraphics{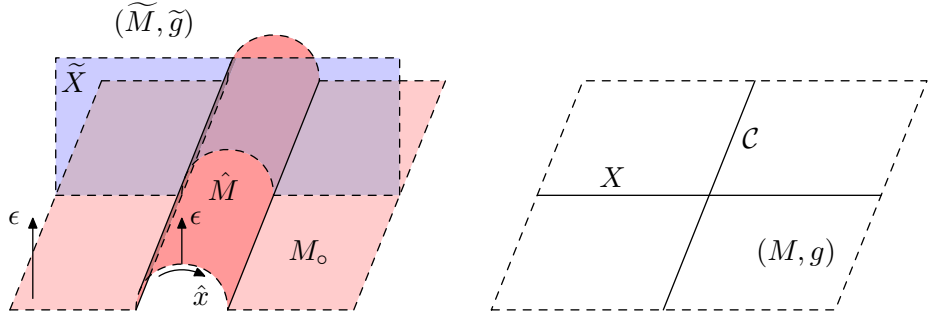}
\caption{Illustration of Theorem~\ref{ThmM}. \textit{On the left:} the total family $\wt g$ solves the Einstein vacuum equations~\eqref{EqM} to infinite order at $\eps=0$ (the union of $\hat M$ and $M_\circ$) and at $\wt X$. A set $\wt K=\{\hat r=|\hat x|\leq\bhm\}$ in the interior of the small Kerr black hole is excised (not shown here). \textit{On the right:} the background spacetime $(M,g)$, and timelike geodesic $\cC$ along which we want to glue in a small Kerr black hole, and a spacelike hypersurface $X$; the lift of $[0,1)_\eps\times X$ to $\wt M$ is $\wt X$ on the left.}
\label{FigM}
\end{figure}

As outlined in~\S\ref{SsIPf}, the proof of Theorem~\ref{ThmM} consists of two steps.

\begin{enumerate}
\myitem{ItMFormalI}{(I)}{I} Construction of a formal solution $\wt g_\infty$ at $\eps=0$; that is, $\wt g_\infty$ is a total family with respect to $g$ and with $\hat M$-model $\hat g$ which satisfies
  \[
    \Ric(\wt g_\infty)-\Lambda\wt g_\infty\in\eps^\infty\CI(\wt M\setminus\wt K^\circ;S^2\wt T^*\wt M).
  \]
  This forms the main part of the paper and is accomplished in~\S\S\ref{SRic}--\ref{SF} following some preliminary calculations in~\S\ref{SMk}; see Theorem~\ref{ThmFh}. 
\myitem{ItMFormalII}{(II)}{II} Correction of $\wt g_\infty$ to $\wt g=\wt g_\infty+\wt h$ where $\wt h$ vanishes to infinite order at $\eps=0$, and so that $\wt g$ satisfies the Einstein vacuum equations also to infinite order at $\wt X$. This step is the subject of~\S\ref{SIVP}; it is considerably easier and can be read independently. See Theorem~\ref{ThmIVP}.
\end{enumerate}

The combination of Theorems~\ref{ThmFh} and \ref{ThmIVP} proves Theorem~\ref{ThmM}.

\section{Geometric differential operators on Minkowski space}
\label{SMk}

In this section, we compute the explicit form of various geometric operators on Minkowski space which are related to the Einstein equations. A key concept which we recall here as well is the decomposition of tensors on $\Sph^2$ into scalar and vector type components.

\subsection{Spherical harmonics}
\label{SsMkY}

Geometric differential operators on the unit sphere $(\Sph^2,\slg)$ are denoted with a slash, so $\slg$ is the standard metric, $\sld$ is the exterior derivative, $\sldelta=\delta_\slg$ is the (negative) divergence, $\sldelta^*$ the symmetric gradient, $\sltr=\tr_\slg$ is the trace, and $\slDelta=\Delta_\slg$ the non-negative (tensor) Laplacian. For $l\in\N_0$, we define the $(2 l+1)$-dimensional space of spherical harmonics of degree $l$ on $\Sph^2$ by
\[
  \scalspace_l = \ker\bigl(\slDelta-l(l+1)\bigr) \subset \CI(\Sph^2).
\]

\begin{definition}[1-form and symmetric 2-tensor spherical harmonics]
\label{DefMkY}
  Let $l\in\N$.
  \begin{enumerate}
  \item{\rm (1-forms.)} We say that a 1-form $\omega\in\CI(\Sph^2;T^*\Sph^2)$ is of \emph{scalar type $l$} if $\omega=\sld\scal$ for some $\scal\in\scalspace_l$.\footnote{We exclude the case $l=0$ here, since $\scalspace_0$ consists of constants.} We say that $\omega$ is of \emph{vector type $l$} if $\omega=\slstar\sld\scal$, where $\slstar$ is the Hodge star operator on $(\Sph^2,\slg)$; we denote by
  \[
    \vectspace_l = \{\slstar\sld\scal\colon\scal\in\scalspace_l\} \subset \CI(\Sph^2;T^*\Sph^2)
  \]
  the space of all vector type $l$ 1-forms.
  \item{\rm (Symmetric 2-tensors.)} We say that $h\in\CI(\Sph^2;S^2 T^*\Sph^2)$ is of \emph{scalar type $l$} if $h=\sldelta_0^*\sld\scal+\scal'\slg$ with $\scal,\scal'\in\scalspace_l$; here $\sldelta_0^*=\sldelta^*+\half\slg\sldelta$ is the trace-free part of the symmetric gradient. We say that $h$ is of \emph{vector type $l$} if $h=\sldelta^*\vect$ for some $\vect\in\vectspace_l$.
  \end{enumerate}
\end{definition}

For brevity, we call elements of $\scalspace_l$ also \emph{$\rms l$ functions}; similarly, \emph{$\rmv l$ 1-forms} are 1-forms of vector type $l$, and likewise for scalar type 1-forms and for symmetric 2-tensors on $\Sph^2$.

We say that a function, 1-form, or tensor is of \emph{pure type} if it is of scalar type $l$ or vector type $l$ for some $l\in\N_0$. Recall that since $\sld$ has injective principal symbol, every $\omega\in\CI(\Sph^2;T^*\Sph^2)$ can be uniquely decomposed as $\omega=\sld u+v$ where $u\in\CI(\Sph^2)$ and $v\in\CI(\Sph^2;T^*\Sph^2)\cap\ker\sld^*$, i.e.\ $\sldelta v=0$; this means $\sld\slstar v=0$, and since $\Sph^2$ has trivial first cohomology, we can further write $\slstar v=\sld v'$ for $v'\in\CI(\Sph^2)$. (In total, $\omega=\sld u-\slstar\sld v'$.) Expanding $u,v'$ into spherical harmonics, we thus conclude that every smooth 1-form on $\Sph^2$ can be expanded into a (rapidly convergent) sum of 1-forms of pure type.

We have a similar decomposition of symmetric 2-tensors $h\in\CI(\Sph^2;S^2 T^*\Sph^2)$: using that $\sldelta_0^*$ is elliptic, and recalling that there do not exist nontrivial divergence- and trace-free symmetric 2-tensors on $\Sph^2$ \cite{HiguchiSpherical}, we can write $h=u\slg+\sldelta_0^*\omega$ with $u\in\CI(\Sph^2)$ and $\omega\in\CI(\Sph^2;T^*\Sph^2)$; decomposing $u$ and $\omega$ into their pure type components shows that $h$ can be written as a sum of pure type tensors as well. We remark that the case $l=1$ is special in that
\[
  \sldelta_0^*\sld\scal=0\quad\forall\,\scal\in\scalspace_1,\qquad
  \sldelta^*\vect=0\quad\forall\,\vect\in\vectspace_1.
\]
The second identity is a re-statement of the fact that the elements of $\vectspace_1$ are (dual to) rotations on $\Sph^2$ and thus are Killing 1-forms. The first identity only needs to be checked for $\scal$ equal to the height function on $\Sph^2\subset\R^3$, in which case $\slg^{-1}(\sld\scal,-)$ is related, via stereographic projection from the south pole, to the radial vector field on $\R^2$, which is a conformal Killing vector field indeed.

\begin{definition}[Projections onto pure types]
\label{DefMkYProj}
  Given a function, 1-form, or symmetric 2-tensor $q$ on $\Sph^2$, we denote by
  \[
    \pi_{\rms l}(q),\qquad \pi_{\rmv l}(q)
  \]
  its scalar type $l$, resp.\ vector type $l$ part.
\end{definition}

\begin{example}[Some projections]
\label{ExMkYProj}
  Regarding $\Sph^2$ as the unit sphere in $\R_x^3$, write $\omega^j=\frac{x^j}{|x|}$. If $u\in\CI(\Sph^2)$, then $\pi_{\rms 0}(u)=\frac{1}{4\pi}\int_{\Sph^2}u\,\dd\slg$ and $\pi_{\rms 1}(u)=\sum_{j=1}^3(\frac{3}{4\pi}\int_{\Sph^2} u\omega^j\,\dd\slg)\omega^j$, where we use that $\frac{1}{4\pi}\int_{\Sph^2}\omega^j\omega^l\,\dd\slg=\frac{1}{3}\delta^{j l}$ (Kronecker delta). For $\omega\in\CI(\Sph^2;T^*\Sph^2)$, we have
  \[
    \pi_{\rms 1}(\omega) = \sum_{j=1}^3 \left(\frac{1}{4\pi}\int_{\Sph^2} \la\omega,\sld\omega^j\ra\,\dd\slg\right) \frac32 \sld\omega^j;
  \]
  the factor $\frac32$ here is the reciprocal of $\frac{1}{4\pi}\la\sld\omega^j,\sld\omega^j\ra_{L^2(\Sph^2;T^*\Sph^2)}=\frac{1}{4\pi}\la\sldelta\sld\omega^j,\omega^j\ra_{L^2(\Sph^2)}=\frac23$.
\end{example}

We next discuss these notions in terms of the representation theory of $O(3)=SO(3)\times(\Z/2\Z)$. The representation of $SO(3)$ on $\scalspace_l$ via pullback of functions on $SO(3)$ is the unique (up to isomorphism) complex $(2 l+1)$-dimensional representation $\rho_l$. Note moreover that elements of $\scalspace_l$ are even, resp.\ odd with respect to the antipodal map $-I$ for even, resp.\ odd $l$. On the other hand, $-I$ reverses orientation, thus $(-I)\circ\slstar=-\slstar\circ(-I)$. Therefore,
\[
  \text{$O(3)$ acts (via pullback) on}\quad \begin{cases} \scalspace_l \\ \vectspace_l \end{cases}\ \text{as the representation}\quad\begin{cases} \rho_l \otimes (-1)^l, & l\geq 0, \\ \rho_l\otimes(-1)^{l+1}, & l\geq 1. \end{cases}
\]
We then recall that $\rho_l\otimes\rho_{l'}\cong\rho_{|l-l'|}\oplus\cdots\oplus\rho_{l+l'}$, and $\rho_l\otimes_s\rho_l\cong\rho_{2 l}\oplus\rho_{2 l-2}\oplus\cdots\oplus\rho_0$. For example, $\vectspace_1\otimes_s\vectspace_1\subset\CI(\Sph^2;T^*\Sph^2)$, as an $SO(3)$-representation, is isomorphic to $\rho_0\oplus\rho_2$; and since it is even under $-I$, we must have
\begin{equation}
\label{EqMkYV1V1}
  \vectspace_1\otimes_s\vectspace_1\cong\scalspace_0\oplus\scalspace_2
\end{equation}
as $O(3)$-representations. That is, if $\vect\in\vectspace_1$, then $\vect\otimes_s\vect$ is the sum of a scalar type $0$ symmetric 2-tensor (i.e.\ a constant multiple of $\slg$) and a scalar type $2$ symmetric 2-tensor. Explicitly, for $\vect=\pa_\phi^\flat=\sin^2\theta\,\dd\phi$, one finds $\vect\otimes_s\vect=\frac13\slg+(Y\slg+\sldelta_0^*\sld Y)$ where $Y=\half\sin^2\theta-\frac13\in\scalspace_2$. The `vector type $0$' representation $\rho_0\otimes(-1)$, given by $O(3)\ni A\mapsto(v\mapsto(\det A)v)$, $v\in\vectspace_0:=\C$, is not realized by tensors on $\Sph^2$. Thus, for example, while $\scalspace_1\otimes\vectspace_1\cong\vectspace_0\oplus\scalspace_1\oplus\vectspace_2$ as representations, the space $\{f\omega\colon f\in\scalspace_1,\ \omega\in\vectspace_1\}\subset\CI(\Sph^2;T^*\Sph^2)$ is isomorphic to $\scalspace_1\oplus\vectspace_2$. Explicitly, the map $\scalspace_1\otimes\vectspace_1\ni(f,\omega)\mapsto f\omega$ has 1-dimensional kernel spanned by $\sum \omega^j\otimes\slstar\sld\omega^j$ since $\sum\omega^j\slstar\sld\omega^j=\half\slstar\sld\sum(\omega^j)^2$ and $\sum(\omega^j)^2=1$.

\begin{lemma}[Identities for spherical harmonics] 
\label{LemmaMkYId}
  Let $l\in\N_0$, and $\scal\in\scalspace_l$. Then
  \begin{alignat}{3}
  \label{EqMkYIds1}
    \sldelta(\sld\scal) &= l(l+1)\scal, &\quad \sldelta^*(\sld\scal)&=\sldelta_0^*\sld\scal-\frac{l(l+1)}{2}\slg\scal, &\quad \slDelta(\sld\scal)&=(l(l+1)-1)\scal, \\
  \label{EqMkYIds2}
    \slDelta(\slg\scal) &= l(l+1)\slg\scal, &\quad \sldelta(\sldelta_0^*\sld\scal)&=\frac{l(l+1)-2}{2}\sld\scal, &\quad \slDelta(\sldelta_0^*\sld\scal)&=(l(l+1)-4)\sldelta_0^*\sld\scal.
  \end{alignat}
  For $l\in\N$ and $\vect\in\vectspace_l$, we have
  \begin{equation}
  \label{EqMkYIdv}
    \slDelta\vect = (l(l+1)-1)\vect,\quad
    \sldelta\sldelta^*\vect = \frac{l(l+1)-2}{2}\vect,\quad
    \slDelta\sldelta^*\vect = (l(l+1)-4)\sldelta^*\vect.
  \end{equation}
\end{lemma}
\begin{proof}
  The first two identities in~\eqref{EqMkYIds1} and the first identity in~\eqref{EqMkYIds2} follow directly from the definitions and $[\slDelta,\slg]=0$. The third identity in~\eqref{EqMkYIds1} follows from the fact that the Hodge Laplacian and the tensor Laplacian on 1-forms on a Riemannian manifold $(\cM,g)$ are related via $\Delta_g+\Ric(g)=\dd\delta_g+\delta_g\dd$; but $\Ric(\slg)=\slg$ acts as the identity operator on $\CI(\Sph^2;T^*\Sph^2)$, so $\slDelta\sld=(\sld\sldelta+\sldelta\sld-1)\sld=\sld(\slDelta-1)$ on functions. Since $[\slDelta,\slstar]=0$, this also implies the first identity in~\eqref{EqMkYIdv}.

  For a 1-form $\omega$ on a Riemannian manifold $(\cM,g)$, we have
  \begin{align*}
    (2\delta_g\delta_g^*\omega)_\mu &= -2(\delta_g^*\omega)_{\mu\nu;}{}^\nu = -\omega_{\mu;\nu}{}^\nu - \omega_{\nu;\mu}{}^\nu = -\omega_{\mu;\nu}{}^\nu-\omega_{\nu;}{}^\nu{}_\mu + g^{\nu\lambda}(\omega_{\nu;\lambda\mu}-\omega_{\nu;\mu\lambda}) \\
      &= \bigl((\Delta_g+\dd\delta_g)\omega\bigr)_\mu + g^{\nu\lambda}R^\rho{}_{\nu\lambda\mu}\omega_\rho = \bigl((\Delta_g+\dd\delta_g-\Ric(g))\omega\bigr)_\mu.
  \end{align*}
  On $(\Sph^2,\slg)$, this implies $\sldelta\sldelta^*=\half(\slDelta+\sld\sldelta-1)$, which in view of $\sldelta\vect=0$ proves the second identity in~\eqref{EqMkYIdv}. The second identity in~\eqref{EqMkYIds2} follows similarly from
  \[
    \sldelta(\sldelta_0^*\sld\scal) = \sldelta\sldelta^*\sld\scal + \sldelta\Bigl(\frac{l(l+1)}{2}\slg\scal\Bigr) = \frac12(2 l(l+1)-2)\sld\scal - \frac{l(l+1)}{2}\sld\scal.
  \]

  Finally, for $\omega\in\CI(M;T^*M)$ on an $n$-dimensional Riemannian manifold $(\cM,g)$ with constant sectional curvature $K$, so $R^\rho{}_{\mu\nu\lambda}=K(g_{\mu\lambda}\delta_\nu^\rho-g_{\mu\nu}\delta_\lambda^\rho)$, we compute
  \begin{align*}
    2(\Delta_g\delta_g^*\omega)_{\mu\nu} &= -g^{\kappa\lambda}(\omega_{\mu;\nu\kappa\lambda}+\omega_{\nu;\mu\kappa\lambda}) \\
      &= -g^{\kappa\lambda}\bigl(\omega_{\mu;\kappa\lambda\nu}+\omega_{\nu;\kappa\lambda\mu} + R^\rho{}_{\mu\nu\lambda}\omega_{\rho;\kappa}+R^\rho{}_{\kappa\nu\lambda}\omega_{\mu;\rho} + R^\rho{}_{\nu\mu\lambda}\omega_{\rho;\kappa}+R^\rho{}_{\kappa\mu\lambda}\omega_{\nu;\rho} \\
      &\qquad\qquad + ((R^\rho{}_{\mu\nu\kappa}+R^\rho{}_{\nu\mu\kappa})\omega_\rho)_{;\lambda}\bigr) \\
      &= 2(\delta_g^*\Delta_g\omega)_{\mu\nu} - K g^{\kappa\lambda}\bigl(g_{\mu\lambda}\omega_{\nu;\kappa}-g_{\mu\nu}\omega_{\lambda;\kappa} + g_{\kappa\lambda}\omega_{\mu;\nu}-g_{\kappa\nu}\omega_{\mu;\lambda} \\
      &\hspace{9.2em} + g_{\nu\lambda}\omega_{\mu;\kappa}-g_{\nu\mu}\omega_{\lambda;\kappa} + g_{\kappa\lambda}\omega_{\nu;\mu}-g_{\kappa\mu}\omega_{\nu;\lambda}\bigr) \\
      &\qquad - K\bigl(g_{\mu\kappa}\omega_{\nu;}{}^\kappa-g_{\mu\nu}\omega_{\kappa;}{}^\kappa + g_{\nu\kappa}\omega_{\mu;}{}^\kappa-g_{\nu\mu}\omega_{\kappa;}{}^\kappa\bigr) \\
      &= 2\Bigl(\bigl(\delta_g^*(\Delta_g-(n+1)K)-2 K g\delta_g\bigr)\omega\Bigr)_{\mu\nu}.
  \end{align*}
  On the 2-sphere, with $n=2$ and $K=1$, this gives
  \[
    \slDelta\sldelta^* = \sldelta^*(\slDelta-3)-2\slg\sldelta.
  \]
  Since $\sldelta\vect=0$, this implies the third identity in~\eqref{EqMkYIdv}, and after a short calculation also the last identity in~\eqref{EqMkYIds2}.
\end{proof}

Using the representation theory of $O(3)$, or by direct computation, one can check that if $E=\ul\R,T^*\Sph^2,S^2 T^*\Sph^2$ and $u,v\in\CI(\Sph^2;E)$ are of different pure types, then $\la u,v\ra_{L^2(\Sph^2;E)}=0$ where we use the volume density and fiber inner product induced by $\slg$ to define the $L^2$-inner product. As an example of an explicit check, we have for $\scal\in\scalspace_l$ and $\vect\in\vectspace_{l'}$, $l,l'\in\N$, the identity $\la\sldelta_0^*\sld\scal,\sldelta^*\vect\ra=\la\sldelta^*\sld\scal,\sldelta^*\vect\ra=\la\sld\scal,\sldelta\sldelta^*\vect\ra=\half(l'(l'+1)-2)\la\sld\scal,\vect\ra$ (using~\eqref{EqMkYIdv}), which vanishes after another integration by parts since $\sldelta\vect=0$.

\subsection{Operators on Minkowski space}
\label{SsMk}

We write $\ubar g=-\dd t^2+\dd x^2=-\dd t^2+\dd r^2+r^2\slg$ for the Minkowski metric on $\R_t\times\R^3_x$, and
\[
  \ubar\delta=\delta_{\ubar g},\qquad
  \ubar\delta^*=\delta^*_{\ubar g},\qquad
  \ul\tr=\tr_{\ubar g},\qquad
  \ubar\sfG=\sfG_{\ubar g}=I-\half\ubar g\ul\tr,\qquad
  \ubar\Box=\Box_{\ubar g}.
\]
Polar coordinates on $\R^3$ are denoted $r=|x|$, $\omega=\frac{x}{|x|}\in\Sph^2$; we write $\rho=r^{-1}$ for the inverse radial coordinate. In $r>0$, we introduce the double null coordinates\footnote{One can equally well perform all computations in $(t,r)$-coordinates, and in fact in the present paper some calculations would be slightly simplified. We use double null coordinates $(x^0,x^1)$ and the related coordinates $(t_*,r)$ here, as these are more commonly used in related works on spectral theory and microlocal analysis on asymptotically flat spacetimes.}
\[
  x^0 = t+r = t_*+2 r,\qquad
  x^1 = t-r =: t_*,
\]
in terms of which we have
\begin{alignat*}{2}
  \ubar g&= -\dd x^0\,\dd x^1 + r^2\slg &&
    = -\dd t_*^2 - 2\,\dd t_*\,\dd r + r^2\slg, \\
  \ubar g^{-1}&=-4\pa_0\otimes_s\pa_1 + r^{-2}\slg^{-1} &&
    = -2\pa_{t_*}\otimes_s\pa_r + \pa_r^2 + r^{-2}\slg^{-1}.
\end{alignat*}
We write $x^a$ ($a=2,3$) for coordinates on $\Sph^2$, and use the letters $a,b,c,d\in\{2,3\}$ for spherical indices. The Christoffel symbols of $\ubar g$ in the coordinates $x^0,x^1,x^a$ all vanish, with the exception of
\begin{subequations}
\begin{equation}
\label{EqMkChristoffel1}
  \ubar\Gamma^c_{0 b}=\half r^{-1}\delta_b^c, \quad
  \ubar\Gamma^c_{1 b}=-\half r^{-1}\delta_b^c, \quad
  \ubar\Gamma^0_{a b}=-r\slg_{a b}, \quad
  \ubar\Gamma^1_{a b}=r\slg_{a b}, \quad
  \ubar\Gamma^c_{a b}=\slGamma^c_{a b}.
\end{equation}
In particular, this gives
\begin{equation}
\label{EqMkChristoffel2}
  \ubar g^{\mu\nu}\ubar\Gamma^0_{\mu\nu}=-2 r^{-1},\quad
  \ubar g^{\mu\nu}\ubar\Gamma^1_{\mu\nu}=2 r^{-1},\quad
  \ubar g^{\mu\nu}\ubar\Gamma^c_{\mu\nu}=r^{-2}\slg^{a b}\slGamma^c_{a b}.
\end{equation}
\end{subequations}

We introduce the bundle splittings
\begin{equation}
\label{EqMkBundleSplit}
\begin{split}
  T^*\R^4 &= \la\dd x^0\ra \oplus \la\dd x^1\ra \oplus r T^*\Sph^2, \\
  S^2 T^*\R^4 &= \la(\dd x^0)^2\ra \oplus \la 2\dd x^0\,\dd x^1\ra \oplus (2\dd x^0\otimes_s r T^*\Sph^2) \\
    &\qquad \oplus \la(\dd x^1)^2\ra \oplus (2 \dd x^1\otimes_s r T^*\Sph^2) \oplus r^2 S^2 T^*\Sph^2.
\end{split}
\end{equation}
(That is, we write a covector $\omega$ at a point in $\R^4\setminus r^{-1}(0)$ as $\omega=\omega_0\,\dd x^0+\omega_1\,\dd x^1+r\slomega$ where $\slomega\in T^*\Sph^2$.) In these splittings, the fiber inner products on $T^*\R^4$, resp.\ $S^2 T^*\R^4$ induced by $\ubar g$ take the form
\begin{equation}
\label{EqMkMinkInner}
  \begin{pmatrix} 0 & -2 & 0 \\ -2 & 0 & 0 \\ 0 & 0 & \slg^{-1} \end{pmatrix},\quad\text{resp.}\quad
  \begin{pmatrix} 0 & 0 & 0 & 4 & 0 & 0 \\ 0 & 8 & 0 & 0 & 0 & 0 \\ 0 & 0 & 0 & 0 & -4\slg^{-1} & 0 \\ 4 & 0 & 0 & 0 & 0 & 0 \\ 0 & 0 & -4\slg^{-1} & 0 & 0 & 0 \\ 0 & 0 & 0 & 0 & 0 & \slg^{-1}_2 \end{pmatrix},
\end{equation}
where $\slg^{-1}_2$ is the fiber inner product on $S^2 T^*\Sph^2$. Tensor calculations in the splittings~\eqref{EqMkBundleSplit} require careful bookkeeping of $r$-weights.

\begin{definition}[$r$-weights]
\label{DefMkRweight}
  For $N\in\N$ and $\mu_1,\ldots,\mu_N\in\{0,1,2,3\}$, define
  \[
    s(\mu_1,\ldots,\mu_N) := \#\bigl\{i\in\{1,\ldots,N\} \colon \mu_i \in \{2,3\} \bigr\}.
  \]
  Given a tensor $T$ with components $T_{\mu_1\ldots\mu_N}{}^{\nu_1\ldots\nu_K}$, we then set
  \[
    T_{\ol{\mu_1}\ldots\ol{\mu_N}}{}^{\ol{\nu_1}\ldots\ol{\nu_K}} := r^{s(\nu_1,\ldots,\nu_K)-s(\mu_1,\ldots,\mu_N)}T_{\mu_1\ldots\mu_N}{}^{\nu_1\ldots\nu_K}.
  \]
  We similarly define the weighted Christoffel symbols $\ubar\Gamma_{\bar\mu\bar\nu}^{\bar\kappa} := r^{s(\kappa)-s(\mu,\nu)}\ubar\Gamma_{\mu\nu}^\kappa$, which are thus all zero except for
  \begin{equation}
  \label{EqMkRweightGamma}
    \ubar\Gamma^{\bar c}_{0\bar b} = \half r^{-1}\delta_b^c,\quad
    \ubar\Gamma^{\bar c}_{1\bar b} = -\half r^{-1}\delta_b^c,\quad
    \ubar\Gamma^0_{\bar a\bar b}=-r^{-1}\slg_{a b}, \quad
    \ubar\Gamma^1_{\bar a\bar b}=r^{-1}\slg_{a b}, \quad
    \ubar\Gamma^{\bar c}_{\bar a\bar b} = r^{-1}\slGamma^c_{a b}.
  \end{equation}
\end{definition}

For example, for $\omega=\omega_0\,\dd x^0+\omega_1\,\dd x^1+r\slomega$, we have $\omega_{\bar 0}=\omega_0$, $\omega_{\bar 1}=\omega_1$, and $\omega_{\bar a}=\slomega_a$.

\begin{lemma}[Form of geometric operators on Minkowski space]
\label{LemmaMk}
  We work in the bundle splittings~\eqref{EqMkBundleSplit} and in the coordinates $(t_*,\rho,\omega)=(t-|x|,\frac{1}{|x|},\frac{x}{|x|})$.
  \begin{enumerate}
  \item\label{ItMk1Ops} Acting on 1-forms, we have
    \begin{align*}
      \ubar\Box &= -2\pa_{t_*}\rho(\rho\pa_\rho-1)+\wh{\ubar\Box}(0),\qquad \wh{\ubar\Box}(0)=\rho^2\left(-(\rho\pa_\rho)^2+\rho\pa_\rho+\slDelta+\begin{pmatrix} 1 & -1 & -\sldelta \\ -1 & 1 & \sldelta \\ -2\sld & 2\sld & 1 \end{pmatrix} \right), \\
      \ubar\delta^* &= \begin{pmatrix} 0 & 0 & 0 \\ \half & 0 & 0 \\ 0 & 0 & 0 \\ 0 & 1 & 0 \\ 0 & 0 & \half \\ 0 & 0 & 0 \end{pmatrix}\pa_{t_*} + \wh{\ubar\delta^*}(0),\qquad \wh{\ubar\delta^*}(0)=\rho\begin{pmatrix} -\half\rho\pa_\rho & 0 & 0 \\ \tfrac14\rho\pa_\rho & -\tfrac14\rho\pa_\rho & 0 \\ \half\sld & 0 & -\tfrac14(\rho\pa_\rho+1) \\ 0 & \half\rho\pa_\rho & 0 \\ 0 & \half\sld & \tfrac14(\rho\pa_\rho+1) \\
      \slg & -\slg & \sldelta^* \end{pmatrix}.
    \end{align*}
    Here, $\slDelta$ is the block-diagonal operator all of whose diagonal entries are equal to the respective (tensor) Laplacian on the standard 2-sphere.
  \item\label{ItMk2Ops} Acting on symmetric 2-tensors, we have $\ul\tr=(0,-4,0,0,0,\sltr)$,
    \begin{align*}
      \ubar\Box &= -2\pa_{t_*}\rho(\rho\pa_\rho-1) + \wh{\ubar\Box}(0), \\ 
        \wh{\ubar\Box}(0)&= \rho^2\left(-(\rho\pa_\rho)^2+\rho\pa_\rho+\slDelta+\begin{pmatrix} 2 & -2 & -2\sldelta & 0 & 0 & -\half\sltr \\ -1 & 2 & \sldelta & -1 & -\sldelta & \half\sltr \\ -2\sld & 2\sld & 3 & 0 & -2 & -\sldelta \\ 0 & -2 & 0 & 2 & 2\sldelta & -\half\sltr \\ 0 & -2\sld & -2 & 2\sld & 3 & \sldelta \\ -2\slg & 4\slg & -4\sldelta^* & -2\slg & 4\sldelta^* & 2 \end{pmatrix} \right), \\
      \ubar\sfG &= \begin{pmatrix} 1 & 0 & 0 & 0 & 0 & 0 \\ 0 & 0 & 0 & 0 & 0 & \tfrac14\sltr \\ 0 & 0 & 1 & 0 & 0 & 0 \\ 0 & 0 & 0 & 1 & 0 & 0 \\ 0 & 0 & 0 & 0 & 1 & 0 \\ 0 & 2\slg & 0 & 0 & 0 & \sfG_\slg \end{pmatrix}, \\
      \ubar\delta &= \begin{pmatrix} 2 & 0 & 0 & 0 & 0 & 0 \\ 0 & 2 & 0 & 0 & 0 & 0 \\ 0 & 0 & 2 & 0 & 0 & 0 \end{pmatrix} \pa_{t_*} + \wh{\ubar\delta}(0), \\
      \wh{\ubar\delta}(0) &= \rho\begin{pmatrix} \rho\pa_\rho-2 & -\rho\pa_\rho+2 & \sldelta & 0 & 0 & \half\sltr \\ 0 & \rho\pa_\rho-2 & 0 & -\rho\pa_\rho+2 & \sldelta & -\half\sltr \\ 0 & 0 & \rho\pa_\rho-3 & 0 & -\rho\pa_\rho+3 & \sldelta \end{pmatrix}.
    \end{align*}
  \end{enumerate}
\end{lemma}
\begin{proof}
  The Minkowski metric and dual metric take the form $\ubar g=(0,-\half,0,0,0,\slg)^T$ and $\ubar g^{-1}=(0,-2,0,0,0,\slg^{-1})$ (using the dual splitting). This implies $\ul\tr=(0,-4,0,0,0,\sltr)$. Moreover, the passage from $(x^0,x^1)$-coordinates to $(t_*,r)$- or $(t_*,\rho)$-coordinates is given by
  \begin{equation}
  \label{EqMkpa01}
    \pa_0 = \half\pa_r = -\half\rho^2\pa_\rho,\qquad
    \pa_1 = \pa_{t_*} - \half\pa_r = \pa_{t_*} + \half\rho^2\pa_\rho.
  \end{equation}
  On functions, the expressions~\eqref{EqMkpa01} and \eqref{EqMkChristoffel1}--\eqref{EqMkChristoffel2} thus give
  \begin{equation}
  \label{EqMkBox0}
  \begin{split}
    \ubar\Box&=-\ubar g^{\mu\nu}(\pa_\mu\pa_\nu-\ubar\Gamma^\kappa_{\mu\nu}\pa_\kappa) = 4\pa_0\pa_1-2 r^{-1}\pa_0+2 r^{-1}\pa_1 + r^{-2}\slDelta \\
      &= -2\rho\pa_{t_*}(\rho\pa_\rho-1) + \rho^2\bigl(-(\rho\pa_\rho)^2+\rho\pa_\rho+\slDelta\bigr).
  \end{split}
  \end{equation}
  
  We have $(\ubar\delta^*\omega)_{\bar\mu\bar\nu}=r^{-s(\mu,\nu)}\half(\pa_\mu(r^{s(\nu)}\omega_{\bar\nu})+\pa_\nu(r^{s(\mu)}\omega_{\bar\mu}))-\ubar\Gamma_{\bar\mu\bar\nu}^{\bar\kappa}\omega_{\bar\kappa}$, and therefore
  \[
    \ubar\delta^* = \begin{pmatrix} \pa_0 & 0 & 0 \\ \half\pa_1 & \half\pa_0 & 0 \\ \half r^{-1}\sld & 0 & \half(r^{-1}\pa_0 r-r^{-1}) \\ 0 & \pa_1 & 0 \\ 0 & \half r^{-1}\sld & \half(r^{-1}\pa_1 r+r^{-1}) \\ r^{-1}\slg & -r^{-1}\slg & r^{-1}\sldelta^* \end{pmatrix},
  \]
  which produces the stated expression upon using~\eqref{EqMkpa01}. The divergence operator on symmetric 2-tensors can be computed as the formal adjoint of $\ubar\delta^*$ with respect to the $L^2$ inner product with volume density $r^2|\dd\ubar x^0\dd\ubar x^1\dd\slg|$ and fiber inner products~\eqref{EqMkMinkInner}. This gives
  \[
    \ubar\delta=\begin{pmatrix} 2 r^{-2}\pa_1 r^2 & 2 r^{-2}\pa_0 r^2 & r^{-1}\sldelta & 0 & 0 & \half r^{-1}\sltr \\ 0 & 2 r^{-2}\pa_1 r^2 & 0 & 2 r^{-2}\pa_0 r^2 & r^{-1}\sldelta & -\half r^{-1}\sltr \\ 0 & 0 & 2(r^{-1}\pa_1 r-r^{-1}) & 0 & 2(r^{-1}\pa_0 r+r^{-1}) & r^{-1}\sldelta \end{pmatrix},
  \]
  and thus the stated expression upon plugging in~\eqref{EqMkpa01}.
  
  The wave operator on 1-forms can be computed by recalling that for any metric $g$ one has $\Box_g=2\delta_g\sfG_g\delta_g^*+\Ric(g)$; using that $\Ric(\ubar g)=0$ and $\Ric(\slg)=\slg$, one can thus compute $\ubar\Box=2\ubar\delta\ubar\sfG\ubar\delta^*$ using the already known expressions for $\ubar\delta$, $\ubar\sfG$, and $\ubar\delta^*$.
  
  We compute the action of the wave operator on a symmetric 2-tensor $h$ in two steps.

  \pfsubstep{Step 1.}{Covariant derivative of a symmetric 2-tensor.} By direct calculation using the formula
  \[
    h_{\bar\mu\bar\nu;\bar\kappa} = r^{-s(\mu,\nu,\kappa)}\pa_\kappa\bigl(r^{s(\mu,\nu)}h_{\bar\mu\bar\nu}\bigr) - \ubar\Gamma^{\bar\rho}_{\bar\mu\bar\kappa}h_{\bar\rho\bar\nu} - \ubar\Gamma^{\bar\rho}_{\bar\nu\bar\kappa}h_{\bar\mu\bar\rho},
  \]
  one finds that the components $h_{\bar\mu\bar\nu;\bar\kappa}$ of the covariant derivative $\ubar\nabla h=\nabla^{\ubar g}h$ are given by
    \begin{alignat*}{3}
      h_{0 0;0}&=\pa_0 h_{0 0}, &\quad h_{0 0;1}&=\pa_1 h_{0 0}, &\quad h_{0 0;\bar c}&=r^{-1}\pa_c h_{0 0}-r^{-1}h_{0\bar c}, \\
      h_{0 1;0}&=\pa_0 h_{0 1}, &\quad h_{0 1;1}&=\pa_1 h_{0 1}, &\quad h_{0 1;\bar c}&=r^{-1}\pa_c h_{0 1}+\half r^{-1}(h_{0\bar c}-h_{1\bar c}), \\
      h_{0\bar b;0}&=\pa_0 h_{0\bar b}, &\quad h_{0\bar b;1}&=\pa_1 h_{0\bar b}, &\quad h_{0\bar b;\bar c}&=r^{-1}\slnabla_c h_{0\bar b}+r^{-1}(h_{0 0}-h_{0 1})\slg_{b c}-\half r^{-1}h_{\bar b\bar c}, \\
      h_{1 1;0}&=\pa_0 h_{1 1}, &\quad h_{1 1;1}&=\pa_1 h_{1 1}, &\quad h_{1 1;\bar c}&=r^{-1}\pa_c h_{1 1}+r^{-1}h_{1\bar c}, \\
      h_{1\bar b;0}&=\pa_0 h_{1\bar b}, &\quad h_{1\bar b;1}&=\pa_1 h_{1\bar b}, &\quad h_{1\bar b;\bar c}&=r^{-1}\slnabla_c h_{1\bar b} + r^{-1}(h_{0 1}-h_{1 1})\slg_{b c}+\half r^{-1}h_{\bar b\bar c}, \\
      h_{\bar a\bar b;0}&=\pa_0 h_{\bar a\bar b}, &\quad h_{\bar a\bar b;1}&=\pa_1 h_{\bar a\bar b}, &\quad h_{\bar a\bar b;\bar c}&=r^{-1}\slnabla_c h_{\bar a\bar b}+r^{-1}(h_{0\bar a}-h_{1\bar a})\slg_{b c} + r^{-1}(h_{0\bar b}-h_{1\bar b})\slg_{a c}.
    \end{alignat*}
  
  \pfsubstep{Step 2.}{Wave operator.} In order to compute $\ubar\Box h$, we use the formula
  \begin{align}
    (\ubar\Box h)_{\bar\mu\bar\nu} &= -g^{\bar\kappa\bar\lambda}h_{\bar\mu\bar\nu;\bar\kappa\bar\lambda} \nonumber\\
      &= -g^{\bar\kappa\bar\lambda}\bigl[r^{-s(\mu,\nu,\kappa,\lambda)}\pa_\lambda\bigl(r^{s(\mu,\nu,\kappa)}h_{\bar\mu\bar\nu;\bar\kappa}\bigr) - \ubar\Gamma^{\bar\rho}_{\bar\kappa\bar\lambda} h_{\bar\mu\bar\nu;\bar\rho} - \ubar\Gamma^{\bar\rho}_{\bar\mu\bar\lambda}h_{\bar\rho\bar\nu;\bar\kappa} - \ubar\Gamma^{\bar\rho}_{\bar\nu\bar\lambda}h_{\bar\mu\bar\rho;\bar\kappa} \bigr] \nonumber\\
  \label{EqMkBox0Formula}
  \begin{split}
      &= 2 \Bigl( r^{-s(\mu,\nu)}\pa_0(r^{s(\mu,\nu)}h_{\bar\mu\bar\nu;1}) + r^{-s(\mu,\nu)}\pa_1(r^{s(\mu,\nu)}h_{\bar\mu\bar\nu;0}) - r^{-1}h_{\bar\mu\bar\nu;0} + r^{-1}h_{\bar\mu\bar\nu;1} \\
      &\qquad\qquad - \ubar\Gamma^{\bar\rho}_{1 \bar\mu}h_{\bar\rho\bar\nu;0}-\ubar\Gamma^{\bar\rho}_{0 \bar\mu}h_{\bar\rho\bar\nu;1} - \ubar\Gamma^{\bar\rho}_{1\bar\nu}h_{\bar\mu\bar\rho;0}-\ubar\Gamma^{\bar\rho}_{0\bar\nu}h_{\bar\mu\bar\rho;1}\Bigr) + (\cS h)_{\bar\mu\bar\nu},
  \end{split}
  \end{align}
  where we define
  \[
    (\cS h)_{\bar\mu\bar\nu} = - \slg^{c d}r^{-1}\pa_c h_{\bar\mu\bar\nu;\bar d} + \slg^{c d}(\ubar\Gamma^{\bar\rho}_{\bar\mu\bar c}h_{\bar\rho\bar\nu;\bar d}+\ubar\Gamma^{\bar\rho}_{\bar\nu\bar c}h_{\bar\mu\bar\rho;\bar d}+\ubar\Gamma^{\bar e}_{\bar c\bar d}h_{\bar\mu\bar\nu;\bar e}).
  \]
  Since $\ubar\Gamma^{\bar\rho}_{1\bar\mu}=-\half r^{-1}s(\mu)\delta_\mu^\rho=-\ubar\Gamma^{\bar\rho}_{0\bar\mu}$, the derivatives falling on $r^{s(\mu,\nu)}$ in the big parenthesis cancel the terms from the Christoffel symbols; the big parenthesis in~\eqref{EqMkBox0Formula} thus evaluates to
  \begin{equation}
  \label{EqMkBox201}
    (4\pa_0\pa_1-2 r^{-1}\pa_0+2 r^{-1}\pa_1)h_{\bar\mu\bar\nu}=-2\rho\pa_{t_*}(\rho\pa_\rho-1)+\rho^2\bigl(-(\rho\pa_\rho)^2+\rho\pa_\rho\bigr)h_{\bar\mu\bar\nu},
  \end{equation}
  cf.\ \eqref{EqMkBox0}. It remains to evaluate the action of $\cS$. Using~\eqref{EqMkRweightGamma}, we compute
  \begin{align*}
    (\cS h)_{0 0} &=-\slg^{c d}r^{-1}\slnabla_c h_{0 0;\bar d} + \slg^{c d}r^{-1}h_{0\bar c;\bar d}, \\
    (\cS h)_{0 1} &=-\slg^{c d}r^{-1}\slnabla_c h_{0 1;\bar d} + \half\slg^{c d}r^{-1}(h_{1\bar c;\bar d}-h_{0\bar c;\bar d}), \\
    (\cS h)_{0\bar b} &=-\slg^{c d}r^{-1}\slnabla_c h_{0\bar b;\bar d} + r^{-1}(h_{0 1;\bar b}-h_{0 0;\bar b})+\half\slg^{c d}r^{-1}h_{\bar b\bar c;\bar d}, \\
    (\cS h)_{1 1} &=-\slg^{c d}r^{-1}\slnabla_c h_{1 1;\bar d} - \slg^{c d}r^{-1}h_{1\bar c;\bar d}, \\
    (\cS h)_{1\bar b} &=-\slg^{c d}r^{-1}\slnabla_c h_{1\bar b;\bar d} + r^{-1}(h_{1 1;\bar b}-h_{0 1;\bar b}) - \half\slg^{c d}r^{-1}h_{\bar b\bar c;\bar d}, \\
    (\cS h)_{\bar a\bar b} &=-\slg^{c d}r^{-1}\slnabla_c h_{\bar a\bar b;\bar d} + r^{-1}(h_{1\bar a;\bar b}+h_{1\bar b;\bar a}-h_{0\bar a;\bar b}-h_{0\bar b;\bar a}).
  \end{align*}
  Using the expressions for $h_{\bar\mu\bar\nu;\bar\kappa}$ given above, and denoting spherical indices by `$\bullet$', this gives
  \begin{align*}
    (r^2\cS h)_{0 0} &= \slDelta h_{0 0} + 2(h_{0 0}-h_{0 1}) - 2\sldelta h_{0\bullet} - \half\sltr h_{\bullet\bullet}, \\
    (r^2\cS h)_{0 1} &= \slDelta h_{0 1} - h_{0 0}+2 h_{0 1}-h_{1 1} + \sldelta h_{0\bullet} - \sldelta h_{1\bullet} + \half\sltr h_{\bullet\bullet}, \\
    (r^2\cS h)_{0\bullet} &= \slDelta h_{0\bullet} - 2\,\sld h_{0 0} + 2\,\sld h_{0 1} + 3 h_{0\bullet} - 2 h_{1\bullet} - \sldelta h_{\bullet\bullet}, \\
    (r^2\cS h)_{1 1} &= \slDelta h_{1 1} - 2 h_{0 1}+2 h_{1 1} + 2\sldelta h_{1\bullet} - \half\sltr h_{\bullet\bullet}, \\
    (r^2\cS h)_{1\bullet} &= \slDelta h_{1\bullet} - 2\,\sld h_{0 1}+2\,\sld h_{1 1}-2 h_{0\bullet}+3 h_{1\bullet}+\sldelta h_{\bullet\bullet}, \\
    (r^2\cS h)_{\bullet\bullet} &= \slDelta h_{\bullet\bullet} + (-2 h_{0 0}+4 h_{0 1}-2 h_{1 1})\slg - 4\sldelta^*h_{0\bullet} + 4\sldelta^*h_{1\bullet}+2 h_{\bullet\bullet}.
  \end{align*}
  Together with~\eqref{EqMkBox201}, this verifies the form of $\ubar\Box$ on symmetric 2-tensors. The proof of Lemma~\ref{LemmaMk} is complete.
\end{proof}

\begin{cor}[Indicial family of linearized Ricci at zero frequency]
\label{CorMkDiffDRic}
  Regarding $\R^3=t^{-1}(0)\subset\R^4$, define $\wh{D_{\ubar g}\Ric}(0)\in\Diff^2(\R^3;S^2 T^*_{\R^3}\R^4)$ to be the restriction of $D_{\ubar g}\Ric$ to $t$-translation invariant symmetric 2-tensors on $\R^4$. In the coordinates $\rho=r^{-1}=|x|^{-1}$, $\omega=\frac{x}{|x|}$, and in the bundle splitting~\eqref{EqMkYSplit}, the operator $\rho^{-2}\wh{D_{\ubar g}\Ric}(0)$ is then dilation-invariant on $[0,\infty)_\rho\times\Sph^2_\omega$, and its indicial family\footnote{This is a family of second order differential operators on $\Sph^2$ acting on sections of $\ul\R\oplus\ul\R\oplus T^*\Sph^2\oplus\ul\R\oplus T^*\Sph^2\oplus S^2 T^*\Sph^2$. Conceptually more accurately, it is a family of operators on $\pa(\ol{\R^3}\setminus\{0\})=\pa\ol{\R^3}$ acting on sections of the restriction to $\pa\ol{\R^3}$ of the bundle on $\ol{\R^3}\setminus\{0\}$ defined by continuous extension of~\eqref{EqMkBundleSplit}.} $N(\rho^{-2}\wh{D_{\ubar g}\Ric}(0),\lambda)=\rho^{-\lambda-2}\wh{D_{\ubar g}\Ric}(0)\rho^\lambda$ (i.e.\ defined with respect to $\rho$) satisfies
  \begin{align*}
    &N\bigl(\rho^{-2}\wh{D_{\ubar g}\Ric}(0),\lambda\bigr) \\
    &\ =\begin{reduce2}\begin{pmatrix}
              \frac12\slDelta & \lambda & \frac12(\lambda{-}1)\sldelta & 0 & 0 & {-}\frac18\lambda(\lambda{-}1)\sltr \\
              {-}\frac14\lambda(\lambda{-}1) & {-}\frac12\lambda(\lambda{+}1){+}\frac12\slDelta & {-}\frac14(\lambda{-}1)\sldelta & {-}\frac14\lambda(\lambda{-}1) & \frac14(\lambda{-}1)\sldelta & \frac18\lambda(\lambda{-}1)\sltr \\
              {-}\frac12\lambda\sld & {-}\frac12(\lambda{+}2)\sld & {-}\frac14(\lambda^2{-}\lambda{+}2){+}\frac12\sldelta\sld & 0 & {-}\frac14(\lambda{-}2)(\lambda{+}1) & \frac14\lambda(\sldelta{+}\sld\sltr) \\
              0 & \lambda & 0 & \frac12\slDelta & {-}\frac12(\lambda{-}1)\sldelta & {-}\frac18\lambda(\lambda{-}1)\sltr \\
              0 & \frac12(\lambda{+}2)\sld & {-}\frac14(\lambda{-}2)(\lambda{+}1) & \frac12\lambda\sld & {-}\frac14(\lambda^2{-}\lambda{+}2){+}\frac12\sldelta\sld & {-}\frac14\lambda(\sldelta{+}\sld\sltr) \\
              {-}(\lambda{-}1)\slg & {-}2\slg{+}2\sldelta^*\sld & {-}(\lambda{-}1)\sldelta^*{-}\slg\sldelta & {-}(\lambda{-}1)\slg & (\lambda{-}1)\sldelta^*{+}\slg\sldelta & {-}\frac12(\lambda{+}1)(\lambda{-}2){+}\frac12\slDelta \\ &&&&& {+}\frac12(\lambda{-}2)\slg\sltr{-}\sldelta^*\sldelta\sfG_\slg
            \end{pmatrix}\end{reduce2}, \\
    &\pa_\lambda N\bigl(\rho^{-2}\wh{D_{\ubar g}\Ric}(0),\lambda\bigr) \\
    &\ =\begin{pmatrix}
        0 & 1 & \half\sldelta & 0 & 0 & {-}\tfrac18(2\lambda{-}1)\sltr \\
        {-}\tfrac14(2\lambda{-}1) & {-}\half(2\lambda{+}1) & {-}\tfrac14\sldelta & {-}\tfrac14(2\lambda{-}1) & \tfrac14\sldelta & \tfrac18(2\lambda{-}1)\sltr \\
        {-}\half\sld & {-}\half\sld & {-}\tfrac14(2\lambda{-}1) & 0 & {-}\tfrac14(2\lambda{-}1) & \tfrac14(\sldelta{+}\sld\sltr) \\
        0 & 1 & 0 & 0 & {-}\half\sldelta & {-}\tfrac18(2\lambda{-}1)\sltr \\
        0 & \half\sld & {-}\tfrac14(2\lambda{-}1) & \half\sld & {-}\tfrac14(2\lambda{-}1) & {-}\tfrac14(\sldelta{+}\sld\sltr) \\
        {-}\slg & 0 & {-}\sldelta^* & {-}\slg & \sldelta^* & {-}\half(2\lambda{-}1{-}\slg\sltr)
      \end{pmatrix}, \\
    &\pa_\lambda N\bigl(\rho^{-1}[D_{\ubar g}\Ric,t_*]\ftrans(0),\lambda\bigr) \\
    &\ =\begin{pmatrix}
              0 & 0 & 0 & 0 & 0 & 0 \\
              -1 & -1 & 0 & 0 & 0 & \tfrac14\sltr \\
              0 & 0 & -\half & 0 & 0 & 0 \\
              0 & 0 & 0 & 0 & 0 & -\half\sltr \\
              0 & 0 & -1 & 0 & -\half & 0 \\
              0 & 0 & 0 & 0 & 0 & -1
            \end{pmatrix}.
  \end{align*}
\end{cor}
\begin{proof}
  One obtains $\wh{D_{\ubar g}\Ric}(0)$ from $D_{\ubar g}\Ric$ (expressed in $t_*,\rho,\omega$ coordinates) by dropping all derivatives in $t_*$. We then have
  \[
    N\bigl(\rho^{-2}\wh{D_{\ubar g}\Ric}(0),\lambda\bigr) = \frac12 N\bigl(\rho^{-2}\wh{\ubar\Box}(0),\lambda\bigr) - N\bigl(\rho^{-1}\wh{\ubar\delta^*}(0),\lambda+1\bigr) N\bigl(\rho^{-1}\wh{\ubar\delta}(0),\lambda\bigr) \ubar\sfG.
  \]
  The first half of the Lemma then follows by differentiation in $\lambda$ and evaluation using the expressions in Lemma~\ref{LemmaMk}. To prove the second half, we note that
  \begin{equation}
  \label{EqMkDiffDRicComm}
    [D_{\ubar g}\Ric,t_*]\ftrans(0) = \frac12\wh{[\ubar\Box,t_*]}(0) - [\ubar\delta^*,t_*]\wh{\ubar\delta}(0)\ubar\sfG - \wh{\ubar\delta^*}(0)[\ubar\delta,t_*]\ubar\sfG.
  \end{equation}
  Since $[\ubar\delta,t_*]$ is homogeneous of degree $0$, this implies
  \begin{align*}
    &\pa_\lambda N\bigl(\rho^{-1}[D_{\ubar g}\Ric,t_*]\ftrans(0),\lambda\bigr) \\
    &\quad = \frac12\pa_\lambda N\bigl(\rho^{-1}\wh{[\ubar\Box,t_*]}(0),\lambda\bigr) - [\ubar\delta^*,t_*]\circ\pa_\lambda N(\rho^{-1}\wh{\ubar\delta}(0),\lambda)\circ\ubar\sfG \\
    &\quad \hspace{11.4em} - \pa_\lambda N(\rho^{-1}\wh{\ubar\delta^*}(0),\lambda)\circ[\ubar\delta,t_*]\circ\ubar\sfG.
  \end{align*}
  A computation using Lemma~\ref{LemmaMk} gives the stated result.
\end{proof}

Using polar coordinates $x=r\omega$, $r>0$, $\omega\in\Sph^2$, we can decompose tensors on $\R^3\setminus\{0\}$ into spherical harmonics. In the case of functions $u=u(r,\omega)$, this amounts to projecting $u(r,-)$ to $\scalspace_l$ for each $r>0$ and $l\in\N_0$; in this manner, we can write $u\in\CI(\R^3\setminus\{0\})$ as a (rapidly converging) series $\sum_{l\in\N_0} u_l(r)$ where $u_l\in\CI((0,\infty);\scalspace_l)$. In the case of 1-forms and symmetric 2-tensors, we split
\begin{equation}
\label{EqMkYSplit}
  T^*\R^3 = \la\dd r\ra \oplus r T^*\Sph^2,\qquad
  S^2 T^*\R^3 = \la\dd r^2\ra \oplus (2\dd r\otimes_s r T^*\Sph^2) \oplus r^2 S^2 T^*\Sph^2
\end{equation}
over $r>0$. (We note that the splittings~\eqref{EqMkYSplit} extend to smooth splittings of $\Tsc^*\ol{\R^3}$ and $S^2\,\Tsc^*\ol{\R^3}$ away from $r=0$.) A 1-form $\omega\in\CI(\R^3\setminus\{0\};T^*\R^3)$ is then given by $\omega=(u,r\slomega)$ where $u\in\CI(\R^3\setminus\{0\})$ and $\slomega\in\CI((0,\infty);\CI(\Sph^2;T^*\Sph^2))$, and we may split $u$ and similarly $\slomega$ into its pure type components as above. The symmetric 2-tensor case is completely analogous.

Further generalizing this to time-translation-invariant sections of $T^*\R^4$ and $S^2 T^*\R^4$ in the bundle splitting~\eqref{EqMkBundleSplit}, we can split the restriction to a single coordinate 2-sphere (say $\{t=0,|x|=1\}$) into pure types as follows:
\begin{subequations}
\begin{enumerate}
\item scalar type $0$: for $a,b,c,d\in\C$,
  \begin{equation}
  \label{EqMkYSplits0}
    \text{1-forms $(a,b,0)^T$,\qquad symmetric 2-tensors $(a,b,0,c,0,d\slg)^T$}
  \end{equation}
  (i.e.\ $a\,\dd x^0+b\,\dd x^1$ and $a(\dd x^0)^2+2 b\,\dd x^0\,\dd x^1+c(\dd x^1)^2+d\cdot r^2\slg$);
\item scalar type $1$:  for $a,b,c,d,e,f\in\C$ and $\scal\in\scalspace_1$,
  \begin{equation}
  \label{EqMkYSplits1}
    \text{1-forms $(a\scal,b\scal,c\,\sld\scal)$,\qquad symmetric 2-tensors $(a\scal,b\scal,c\,\sld\scal,d\scal,e\,\sld\scal,f\scal\slg)$;}
  \end{equation}
\item scalar type $l\geq 2$: for $a,b,c,d,e,f,h\in\C$ and $\scal\in\scalspace_l$,
  \begin{equation}
  \label{EqMkYSplitsl}
    \text{1-forms $(a\scal,b\scal,c\,\sld\scal)$,\qquad symmetric 2-tensors $(a\scal,b\scal,c\,\sld\scal,d\scal,e\,\sld\scal,f\scal\slg+h\sldelta_0^*\sld\scal)$;}
  \end{equation}
\item vector type $1$:  for $a,b\in\C$ and $\vect\in\vectspace_1$,
  \begin{equation}
  \label{EqMkYSplitv1}
    \text{1-forms $(0,0,a\vect)^T$,\qquad symmetric 2-tensors $(0,0,a\vect,0,b\vect,0)^T$;}
  \end{equation}
\item vector type $l\geq 2$: for $a,b,c\in\C$ and $\vect\in\vectspace_l$,
  \begin{equation}
  \label{EqMkYSplitvl}
    \text{1-forms $(0,0,a\vect)^T$,\qquad symmetric 2-tensors $(0,0,a\vect,0,b\vect,c\sldelta^*\vect)^T$.}
  \end{equation}
\end{enumerate}
\end{subequations}

\begin{rmk}[Action on pure type tensors]
\label{RmkMkYAction}
  Importantly, the operators in Lemma~\ref{LemmaMk} and Corollary~\ref{CorMkDiffDRic} not only preserve pure types; they also map pure type tensors of the above forms \emph{with fixed} $\scal$ or $\vect$ into pure type tensors of the same type \emph{with the same $\scal$, $\vect$}. Therefore, for example, the restriction of $\wh{\ubar\delta^*}(0)$ to scalar type $l\geq 2$ tensors is given by a $7\times 3$ matrix of differential operators in $r$ (or $\rho=r^{-1}$), formally obtained from Lemma~\ref{LemmaMk} by replacing $\sld$ by $1$, $\slg$ by $(1,0)^T$, and $\sldelta^*$ by $(-\half l(l+1),1)^T$.
\end{rmk}

We shall use such decompositions mostly to analyze the behavior of tensors on $\R^3$ (arising as stationary tensors on $\R^4$) as $r\to\infty$ or $r\to 0$.

\begin{definition}[Leading order type]
\label{DefMkYType}
  Let $\rho=r^{-1}$ on $\R^3\setminus\{0\}$. Let $E\to\ol{\R^3}\setminus\{0\}$ be the pullback along the projection $(\rho,\omega)\mapsto\omega$ of a direct sum of the bundles $\ul\R,T^*\Sph^2,S^2 T^*\Sph^2\to\Sph^2$. Let $u\in\CI(\ol{\R^3}\setminus\{0\};E)$. Then we say that $u$ is of scalar (resp.\ vector) type $l$ modulo $\rho^k\CI$ if there exists $u_0\in\CI(\ol{\R^3}\setminus\{0\};E)$ which on each coordinate sphere is of scalar (resp.\ vector) type $l$ and for which $u-u_0\in\rho^k\CI$. If this holds for $k=1$, we say that $u$ is of scalar (resp.\ vector) type \emph{to leading order at infinity}.
\end{definition}

\section{Linearized Ricci curvature operator on Minkowski space}
\label{SRic}

We continue using the notation $\ubar g=-\dd t^2+\dd r^2+r^2\slg$ from~\S\ref{SsMk}. In order to prepare solving the linearized Einstein equations with right hand sides which are polyhomogeneous at $\pa M_\circ$ or $\pa\hat M$ (cf.\ \S\ref{SssIPfF}), we now study the mapping properties of the indicial family
\begin{equation}
\label{EqRic}
  N\bigl(r^2\wh{D_{\ubar g}\Ric}(0),\lambda\bigr),\qquad \lambda\in\C,
\end{equation}
defined with respect to $r$, i.e.\ $N(r^2\wh{D_{\ubar g}\Ric}(0),\lambda)=r^{2-\lambda}\wh{D_{\ubar g}\Ric}(0)r^\lambda$ acting on stationary (in $t$) and dilation-invariant (in $r$ relative to the splitting~\eqref{EqMkBundleSplit}) symmetric 2-tensors.

\begin{rmk}[Indicial family with respect to $\rho=r^{-1}$]
\label{RmkRicRho}
  If one defines the indicial family using $\rho$ (as done in Corollary~\ref{CorMkDiffDRic}), then the parameter $\lambda$ gets replaced by $-\lambda$.
\end{rmk}

In the bundle splittings~\eqref{EqMkBundleSplit}, the operator $N\bigl(r^2\wh{D_{\ubar g}\Ric}(0),\lambda)$ is of class $\Diff^2(\Sph^2;\ul\R\oplus\ul\R\oplus T^*\Sph^2\oplus\ul\R\oplus T^*\Sph^2\oplus S^2 T^*\Sph^2)$. More accurately, it is the indicial family of $r^2\wh{D_{\ubar g}\Ric}(0)$ at the front face $\ff\cong\Sph^2$ of $[\R^3;\{0\}]$, and thus acts on sections of the pullback of $S^2 T^*_{(0,0)}\R^4\to\{0\}\subset\R^3$ along the blow-down map $\ubar\upbeta\colon[\R^3;\{0\}]\to\R^3$. We mostly consider its action on pure type tensors, and write
\[
  N_\bullet(r^2\wh{D_{\ubar g}\Ric}(0),\lambda)
\]
for the restriction of $N(r^2\wh{D_{\ubar g}\Ric}(0),\lambda)$ to tensors of pure type $\bullet\in\{\rms l,\rmv l\}$; expressed in terms of the bases~\eqref{EqMkYSplits0}--\eqref{EqMkYSplitvl}, this is a square matrix whose coefficients are quadratic functions of $\lambda\in\C$. We similarly write $N_\bullet(r\wh{\ubar\delta^*}(0),\lambda)$ and $N_\bullet(r\wh{\ubar\delta\ubar\sfG}(0),\lambda)$. Finally, we write
\[
  \CI_\bullet(\ff;\ubar\upbeta^*T^*_{(0,0)}\R^4) \subset \CI(\ff;\ubar\upbeta^*T^*_{(0,0)}\R^4)
\]
for the (finite-dimensional) subspace consisting of pure type $\bullet$ 1-forms; similarly for symmetric 2-tensors.

Concretely, we shall characterize the extent to which the range of~\eqref{EqRic} is smaller than the kernel of $N(r\wh{\ubar\delta\ubar\sfG}(0),\lambda-2)$ (cf.\ the linearized second Bianchi identity $\ubar\delta\ubar\sfG\circ D_{\ubar g}\Ric=0$), or the kernel is bigger than the range of $N(r\wh{\ubar\delta^*}(0),\lambda+1)$ (cf.\ the identity $D_{\ubar g}\Ric\circ\ubar\delta^*=0$). In other words, we will study the extent to which the complex
\begin{equation}
\label{EqRicSeq}
\begin{split}
  0 &\to \CI(\ff;\ubar\upbeta^*T^*_{(0,0)}\R^4) \xra{N(r\wh{\ubar\delta^*}(0),\lambda+1)} \CI(\ff;\ubar\upbeta^*S^2 T^*_{(0,0)}\R^4) \\
    &\hspace{2em} \xra{N(r^2\wh{D_{\ubar g}\Ric}(0),\lambda)} \CI(\ff;\ubar\upbeta^*S^2 T^*_{(0,0)}\R^4) \xra{N(r\wh{\ubar\delta\ubar\sfG}(0),\lambda-2)} \CI(\ff;\ubar\upbeta^*T^*_{(0,0)}\R^4) \to 0
\end{split}
\end{equation}
is exact, and precisely characterize the failure of exactness when it occurs.

For notational brevity, we shall henceforth write $\sfG=\ubar\sfG$ and
\begin{equation}
\label{EqRicNotation}
  N(D\Ric,\lambda) := N(r^2\wh{D_{\ubar g}\Ric}(0),\lambda),\quad
  N(\delta^*,\lambda) := N(r\wh{\ubar\delta^*}(0),\lambda),\quad
  N(\delta\sfG,\lambda) := N(r\wh{\ubar\delta\ubar\sfG}(0),\lambda).
\end{equation}
Note that the adjoint of $D_{\ubar g}\Ric\circ\ubar\delta^*=0$ is $0=\ubar\delta\circ\ubar\sfG D_{\ubar g}\Ric\,\ubar\sfG=\ubar\delta\ubar\sfG\circ D_{\ubar g}\Ric\,\ubar\sfG$. Thus, the first two (nontrivial) arrows in~\eqref{EqRicSeq} are essentially the formal adjoints of the final two arrows. This implies a certain duality between the two middle homology groups. Moreover, we shall see that $\ker N(D\Ric,\lambda)/\ran N(\delta^*,\lambda+1)$ can be nontrivial, which means that we can have $D_{\ubar g}\Ric(r^\lambda h_0)=0$ (with $h_0\in\CI(\ubar\upbeta^*S^2 T^*_{(0,0)}\R^4)$) but $r^\lambda h_0\neq\ubar\delta^*(r^{\lambda+1}\omega_0)$. One may then worry that one can have $D_{\ubar g}\Ric(r^\lambda(\log r)h_0+r^\lambda h_1)=0$, and yet the argument is not pure gauge to leading order, and so on, which would imply that quasi-homogeneous linearized metric perturbations can have arbitrarily high powers of $\log r$, even after quotienting out by pure gauge solutions. However, already the existence of $h_1$ turns out to force $h_0$ to lie in the range of $N(\delta^*,\lambda+1)$. We proceed to discuss the underlying linear algebra of such `restricted kernels' (and the dual `generalized ranges') in an abstract setting in~\S\ref{SsLA}.

In~\S\ref{SsRicMS}, we then turn to the mode stability of the Minkowski metric at zero energy on homogeneous (with respect to spatial dilations) tensors in $r>0$. The dual problem of solving the equation $N(D\Ric,\lambda)h=f$ for homogeneous $f\in\ker N(\delta\sfG,\lambda-2)$ is considered in~\S\ref{SsRicR}. These two sections consider individual pure types, and thus only require finite-dimensional linear algebra; in~\S\ref{SsRicLarge}, we study all but finitely many pure types (depending on the value of $\lambda\in\C$) in one go by means of a gauge-fixed linearized Einstein operator.

We then combine these ingredients in~\S\ref{SsRicSolv} to give a complete description of the nullspace, modulo pure gauge, of $D_{\ubar g}\Ric$ on quasi-homogeneous tensors in $r>0$, and of its range inside of $\ker\ubar\delta\ubar\sfG$.

\subsection{Linear algebra of restricted kernels and extended ranges}
\label{SsLA}

Consider two fi\-nite-\-di\-men\-sion\-al complex vector spaces $V,W$, each equipped with a nondegenerate (but not necessarily positive definite) sesquilinear inner product, and a holomorphic family of linear maps
\[
  A(\lambda) \colon V\to W,\qquad \lambda\in\Omega,
\]
where $\Omega\subseteq\C$ is open and nonempty. For $j\in\N_0$ and $\lambda_0\in\Omega$, define
\begin{equation}
\label{EqLASpace}
\begin{split}
  K_j(A,\lambda_0) &:= \biggl\{ \omega(\lambda)=\sum_{k=0}^j (\lambda-\lambda_0)^{-k-1}\omega_k \colon \omega_k\in V,\ A(\lambda)\omega(\lambda)\ \text{is holomorphic at}\ \lambda_0 \biggr\}, \\
  K_{[j]}(A,\lambda_0) &:= \left\{ \omega_j \colon \exists\,\omega(\lambda)\in K_j(A,\lambda_0)\ \text{with leading term}\ (\lambda-\lambda_0)^{-j-1}\omega_j \right\}, \\
  R_{[j]}(A,\lambda_0) &:= \biggl\{ (A(\lambda)\omega(\lambda))|_{\lambda=\lambda_0} \colon \omega(\lambda)=\sum_{k=0}^j (\lambda-\lambda_0)^{-k}\omega_k,\ \omega_k\in V, \\
    &\quad\hspace{12em} A(\lambda)\omega(\lambda)\ \text{is holomorphic at}\ \lambda_0 \biggr\}.
\end{split}
\end{equation}
Define $K^*_j(A^*,\lambda_0)$, $K^*_{[j]}(A^*,\lambda_0)$, $R^*_{[j]}(A^*,\lambda_0)$ analogously by replacing $A(\lambda)$ by $A(\lambda)^*$ and $\lambda-\lambda_0$, `holomorphic' by $\bar\lambda-\ol{\lambda_0}$, `anti-holomorphic'. We note that
\begin{align*}
  &\ran A(\lambda_0) = R_{[0]}(A,\lambda_0) \subseteq R_{[1]}(A,\lambda_0) \subseteq\cdots, \\
  \cdots \subseteq\ &K_{[1]}(A,\lambda_0) \subseteq K_{[0]}(A,\lambda_0) = \ker A(\lambda_0),
\end{align*}
hence the nomenclature of `generalized range' and `restricted kernel'.

\begin{lemma}[Duality]
\label{LemmaLADuality}
  For $\lambda_0\in\C$ and for all $j\in\N_0$, we have $K_{[j]}(A,\lambda_0)=R^*_{[j]}(A^*,\lambda_0)^\perp$ and $K^*_{[j]}(A^*,\lambda_0)=R_{[j]}(A,\lambda_0)^\perp$.
\end{lemma}
\begin{proof}
  This is elementary for $j=0$. We thus consider $j\geq 1$. We only explicitly discuss the proof of $K_{[j]}(A,\lambda_0)=R^*_{[j]}(A^*,\lambda_0)^\perp$, since the proof of the second part is completely analogous. We give two proofs of this result. 

  \pfstep{First proof.} Let $\omega(\lambda)=\sum_{k=0}^j (\lambda-\lambda_0)^{-k-1}\omega_j\in K_j(A,\lambda_0)$, so $\omega_k\in K_{[j]}(A,\lambda_0)$, and let $h^*(\lambda)=\sum_{k=0}^j (\bar\lambda-\ol{\lambda_0})^{-k}h_k^*$ be such that $\omega^*:=(A(\lambda)^*h^*(\lambda))|_{\lambda=\lambda_0}\in R^*_{[j]}(A^*,\lambda_0)$. Then
  \[
    \la\omega_k,\omega^*\ra = \lim_{\lambda\to\lambda_0} \la (\lambda-\lambda_0)^{j+1}\omega(\lambda), A(\lambda)^*h^*(\lambda)\ra = \lim_{\lambda\to\lambda_0} (\lambda-\lambda_0)\la A(\lambda)\omega(\lambda),(\bar\lambda-\ol{\lambda_0})^j h^*(\lambda)\ra
  \]
  vanishes since the pairing remains bounded as $\lambda\to\lambda_0$. This establishes the inclusion $K_{[j]}(A,\lambda_0)\subseteq R^*_{[j]}(A^*,\lambda_0)^\perp$.

  It remains to show that $R^*_{[j]}(A^*,\lambda_0)^\perp\subseteq K_{[j]}(A,\lambda_0)$. Arguing by induction (simultaneously also for the second part of the Lemma), suppose that $\omega\in R^*_{[j]}(A^*,\lambda_0)^\perp$. Then $\omega\in R^*_{[j-1]}(A^*,\lambda_0)^\perp=K_{[j-1]}(A,\lambda_0)$ by the inductive hypothesis, so there exists $\omega(\lambda)=\sum_{k=0}^{j-1}(\lambda-\lambda_0)^{-k-1}\omega_k\in K_{j-1}(A,\lambda_0)$ with $\omega_{j-1}=\omega$; and for all $h^*(\lambda)=\sum_{k=0}^j(\bar\lambda-\ol{\lambda_0})^{-k}h^*_k$ for which $A(\lambda)^*h^*(\lambda)$ is anti-holomorphic at $\lambda=\lambda_0$, we have
  \begin{align*}
    0 &= \lim_{\lambda\to\lambda_0} \la\omega, A(\lambda)^*h^*(\lambda) \ra \\
      &= \lim_{\lambda\to\lambda_0} \big\la (\lambda-\lambda_0)^j\omega(\lambda),A(\lambda)^*h^*(\lambda)\big\ra \\
      &= \lim_{\lambda\to\lambda_0} \big\la A(\lambda)\omega(\lambda), (\bar\lambda-\ol{\lambda_0})^j h^*(\lambda) \big\ra \\
      &= \big\la (A(\lambda)\omega(\lambda))|_{\lambda=\lambda_0}, h_j^* \big\ra.
  \end{align*}
  The space of all leading order terms $h_j^*$ of such $h^*(\lambda)$ is $K^*_{[j-1]}(A^*,\lambda_0)$. Thus, using the inductive hypothesis (for the second part of the Lemma), we deduce that
  \[
    (A(\lambda)\omega(\lambda))|_{\lambda=\lambda_0}\in K^*_{[j-1]}(A^*,\lambda_0)^\perp=R_{[j-1]}(A,\lambda_0)
  \]
  can be written as $(A(\lambda)\omega^\flat(\lambda))|_{\lambda=\lambda_0}$ where $\omega^\flat(\lambda)=\sum_{k=0}^{j-1}(\lambda-\lambda_0)^{-k}\omega_k^\flat$. This means that $A(\lambda)(\omega(\lambda)-\omega_\flat(\lambda))$ vanishes at $\lambda=\lambda_0$, and therefore
  \begin{equation}
  \label{EqLA2HDiff}
  \begin{split}
    &(\lambda-\lambda_0)^{-1}\omega(\lambda) - (\lambda-\lambda_0)^{-1}\omega^\flat(\lambda) \\
    &\qquad = (\lambda-\lambda_0)^{-j-1}\omega + \sum_{k=0}^{j-2} (\lambda-\lambda_0)^{-k-2}\omega_k - \sum_{k=0}^{j-1} (\lambda-\lambda_0)^{-k-1}\omega_k^\flat
  \end{split}
  \end{equation}
  gets mapped by $A(\lambda)$ into a function which is holomorphic at $\lambda=\lambda_0$. This means that~\eqref{EqLA2HDiff} lies in $K_j(A,\lambda_0)$, and therefore $\omega\in K_{[j]}(A,\lambda_0)$. This completes the proof.

  \pfstep{Second proof.} If $T(\lambda),S(\lambda)$ are invertible linear maps which are defined and holomorphic near $\lambda=\lambda_0$, then $K_j(A,\lambda_0)=\{S(\lambda)\omega(\lambda)\colon \omega\in K_j(T A S,\lambda_0)\}$ and therefore
  \[
    K_{[j]}(A,\lambda_0) = S(\lambda_0) K_{[j]}(T A S,\lambda_0).
  \]
  Similarly, $R^*_{[j]}(A^*,\lambda)=(S(\lambda_0)^*)^{-1}R^*_{[j]}(S^*A^*T^*,\lambda_0)$. Therefore, it suffices to prove the Lemma for $T A S$ in place of $A$. Fixing orthonormal bases $\{e_1,\ldots,e_N\}$ of $V$ and $\{f_1,\ldots,f_{N'}\}$ of $W$, we can then choose $T,S$ so that $T A S$ is in Smith normal form: in the case that $N'=\dim W\geq N=\dim V$, this means that
  \[
    T(\lambda)A(\lambda)S(\lambda) = \begin{pmatrix} D(\lambda) \\ 0 \end{pmatrix},\qquad
    D(\lambda)=\diag((\lambda-\lambda_0)^{j_1},\ldots,(\lambda-\lambda_0)^{j_N}),
  \]
  where $0\leq j_1\leq\cdots\leq j_N$, and thus $S(\lambda)^*A(\lambda)^*T(\lambda)^*=(D(\lambda)^*,0)$. Therefore,
  \[
    K_{[j]}(T A S,\lambda_0)=\mathspan\{ e_k\colon j_k>j\},\qquad
    R^*_{[j]}((T A S)^*,\lambda_0)=\mathspan\{ e_k \colon j_k\leq j \}.
  \]
  The case that $N'\leq N$ is similar. This implies the claim.
\end{proof}

Let now $Z$ be a further finite-dimensional complex vector space equipped with a nondegenerate sesquilinear inner product. Suppose
\[
  A(\lambda) \colon V\to W,\qquad
  B(\lambda) \colon W\to Z,\qquad \lambda\in\Omega,
\]
where $\Omega\subseteq\C$ is open and nonempty, are holomorphic families of linear maps with $B(\lambda)\circ A(\lambda)=0$ for all $\lambda\in\Omega$. In other words, we have a family of complexes
\[
  V \mathrel{\mathop{\myrightleftarrows{1cm}}^{A(\lambda)}_{A(\lambda)^*}} W \mathrel{\mathop{\myrightleftarrows{1cm}}^{B(\lambda)}_{B(\lambda)^*}} Z,\qquad \lambda\in\Omega,
\]
with the top row holomorphic and the bottom row antiholomorphic.

\begin{prop}[Nondegenerate pairings]
\label{PropLA2}
  Let $i,j\in\N_0$, $\lambda_0\in\Omega$. Then $R_{[j]}(A,\lambda_0)\subset K_{[i]}(B,\lambda_0)$ and $R^*_{[i]}(B^*,\lambda_0)\subset K^*_{[j]}(A^*,\lambda_0)$. Furthermore, the inner product $K_{[i]}(B,\lambda_0)\times K^*_{[j]}(A^*,\lambda_0)\ni (h,h^*)\mapsto\la h,h^*\ra$ induces a nondegenerate sesquilinear pairing
  \begin{equation}
  \label{EqLA2Pair}
    \bigl( K_{[i]}(B,\lambda_0) / R_{[j]}(A,\lambda_0) \bigr) \times \bigl( K^*_{[j]}(A^*,\lambda_0) / R^*_{[i]}(B^*,\lambda_0) \bigr) \to \C.
  \end{equation}
\end{prop}

The special case $i=j=0$ is the non-degeneracy of the pairing
\[
  \bigl(\ker B(\lambda)/\ran A(\lambda)\bigr) \times \bigl(\ker A(\lambda)^* / \ran B(\lambda)^* \bigr) \to \C,
\]
which is elementary.

\begin{rmk}[Interpretation]
\label{RmkLAInt}
  If $A,B$ depend polynomially on $\lambda$, then $K_{[0]}(B,\lambda_0)$ is equal to the space of all $h\in W$ so that the $W$-valued function $r\mapsto r^{\lambda_0}h$ on $(0,\infty)_r$ satisfies $B(r\pa_r)(r^{\lambda_0}h)=0$. The stronger membership $h\in K_{[1]}(B,\lambda_0)$ implies that there exists $h_1\in W$ so that
  \begin{equation}
  \label{EqLAInt}
    B(r\pa_r)(r^{\lambda_0}(\log r)h+r^{\lambda_0}h_1)=0;
  \end{equation}
  and so on. For $i=1$, $j=0$, the first factor in~\eqref{EqLA2Pair} is trivial if and only if all $\log r$ leading order terms (such as $h$ in~\eqref{EqLAInt}) of quasi-homogeneous (of degree $\lambda_0$) elements in $\ker B(r\pa_r)$ lie in the range of $A(\lambda_0)$. In this case, the second factor is then also trivial, which means that every $h^*\in W$ with $A^*(r\pa_r)(r^{\ol{\lambda_0}}h^*)=0$ can be written as $h^*=B(r\pa_r)^*(r^{\ol{\lambda_0}}(\log r)f_0^*+r^{\ol{\lambda_0}}f_1^*)$ for some $f_0^*\in\ker B(\lambda_0)^*$ and $f_1^*\in Z$. (Here, $B(r\pa_r)^*$ is defined using the b-density $|\frac{\dd r}{r}|$ on $(0,\infty)$.)
\end{rmk}

\begin{proof}[Proof of Proposition~\usref{PropLA2}]
  If $\omega(\lambda)=\sum_{k=0}^j (\lambda-\lambda_0)^{-k}\omega_k$ is such that $A(\lambda)\omega(\lambda)$ is holomorphic, then the identity
  \[
    B(\lambda)((\lambda-\lambda_0)^{-i-1}A(\lambda)\omega(\lambda))=B(\lambda)A(\lambda)((\lambda-\lambda_0)^{-i-1}\omega(\lambda))=0
  \]
  shows that $(A(\lambda)\omega(\lambda))|_{\lambda_0}\in R_{[j]}(A,\lambda_0)$ lies in $K_{[i]}(B,\lambda_0)$. Similarly, one shows that the second quotient space in~\eqref{EqLA2Pair} is well-defined. The non-degeneracy follows from Lemma~\ref{LemmaLADuality}.
\end{proof}

Besides the case $i=0$, $j=0$, we only need the case $i=1$, $j=0$ (or $i=0$, $j=1$) in our application. We rephrase this as follows (see also Remark~\ref{RmkLAInt}).

\begin{cor}[Nondegenerate pairings, special case]
\label{CorLASpecial}
  The inner product on $W$ induces a nondegenerate sesquilinear pairing between the spaces
  \begin{equation}
  \label{EqLASpecial1}
    \{ h\in \ker B(\lambda_0)  \colon B'(\lambda_0)h \in \ran B(\lambda_0) \} / \ran A(\lambda_0)
  \end{equation}
  and
  \begin{equation}
  \label{EqLASpecial2}
    \ker A(\lambda_0)^* / \{ B(\lambda_0)^*f_0^* + B'(\lambda_0)^*f_1^* \colon f_0^*\in Z,\ f_1^*\in\ker B(\lambda_0)^* \}.
  \end{equation}
\end{cor}
\begin{proof}
  Note that if $h\in\ker B(\lambda_0)$, then $v\in W$ solves $B'(\lambda_0)h=B(\lambda_0)v$ if and only if $B(\lambda)((\lambda-\lambda_0)^{-2}h-(\lambda-\lambda_0)^{-1}v)$ is holomorphic at $\lambda_0$; thus, such a vector $v$ exists if and only if $h\in K_{[1]}(B,\lambda_0)$. Thus,~\eqref{EqLASpecial1} is the space $K_{[1]}(B,\lambda_0)/R_{[0]}(A,\lambda_0)$.

  Similarly, if given $f_0^*,f_1^*\in Z$ we define $f^*$ by $f^*(\lambda)=(\bar\lambda-\ol{\lambda_0})^{-1} f_1^*+f_0^*$, then
  \[
    B(\lambda)^*f^*(\lambda)=(\bar\lambda-\ol{\lambda_0})^{-1}B(\lambda_0)^*f_1^* + \bigl(B'(\lambda_0)^*f_1^* + B(\lambda_0)^* f_0^*\bigr) + \cO(|\bar\lambda-\ol{\lambda_0}|)
  \]
  is anti-holomorphic at $\lambda=\lambda_0$ if and only if $f_1^*\in\ker B(\lambda_0)^*$. Thus,~\eqref{EqLASpecial2} is equal to the space $K_{[0]}(A^*,\lambda_0)/R_{[1]}(B^*,\lambda_0)$. The claim now follows from Proposition~\ref{PropLA2}.
\end{proof}

\subsection{Mode stability at zero energy for the indicial family}
\label{SsRicMS}

We study the kernel of $N(D\Ric,\lambda)$ modulo pure gauge solutions, i.e.\ modulo the range of $N(\delta^*,\lambda+1)$. This is a variant on the classical problem of mode stability for the linearized Einstein vacuum equations (see e.g.\ \cite{ReggeWheelerSchwarzschild,KodamaIshibashiMaster}, \cite[\S8]{HaefnerHintzVasyKerr}), here at the Minkowski metric and at zero energy but with the caveat that we are considering stationary metric perturbations in $r>0$ which are homogeneous of degree $\lambda$ with respect to spatial dilations $(r,\omega)\mapsto(c r,\omega)$, $c>0$ (and thus typically singular at $r=0$). Since $N(D\Ric,\lambda)$ preserves pure type tensors in the strong sense explained in Remark~\ref{RmkMkYAction}, one can specialize this problem further and study the mode stability individually for each pure type.\footnote{The converse is not automatic; that is, suppose, say, $\lambda\in\C$ is such that $\ker N_\bullet(D\Ric,\lambda)=\ran N_\bullet(\delta^*,\lambda+1)$ for all $\bullet\in\{\rms l,\,\rmv l\}$. Then we can only conclude that $\ker N(D\Ric,\lambda)=\ran N(\delta^*,\lambda+1)$ if we restrict to the subspace of $\CI(\ff;\ubar\upbeta^*T^*_{(0,0)}\R^4)$ consisting of tensors with finite pure type support, i.e.\ only finitely many projections onto pure type tensors are nonzero. We show in~\S\ref{SsRicLarge}---necessarily using analytic tools---how to remove this restriction.} We identify $\ker N_{\rms 1}(D\Ric,\lambda)\subset\CI_{\rms 1}(\ff;\ubar\upbeta^*S^2 T^*_{(0,0)}\R^4)$ with the subspace of $\C^4$ which is the kernel of the $4\times 4$ matrix obtained by expressing $N(D\Ric,\lambda)$ in the splitting~\eqref{EqMkYSplits1} for any fixed $0\neq\scal\in\scalspace_1$ (see Remark~\ref{RmkMkYAction}); similarly for other pure types. Explicit expressions are given below, see e.g.\ \eqref{EqRicMSs1DRic}.

\begin{prop}[Mode stability at zero energy for the indicial family]
\label{PropRicMS}
  Let $\lambda\in\C$.
  \begin{enumerate}
  \item\label{ItRicMSs0}{\rm (Scalar type $0$.)} For $\lambda\neq -1,0$, we have $\ker N_{\rms 0}(D\Ric,\lambda)=\ran N_{\rms 0}(\delta^*,\lambda+1)$. For $\lambda=-1$, in the notation of~\eqref{EqLASpace}, the quotient $K_{\rms 0,[0]}(D\Ric,-1)/R_{\rms 0,[0]}(\delta^*,0)=\ker N_{\rms 0}(D\Ric,-1)/\ran N_{\rms 0}(\delta^*,0)$ is 2-dimensional, whereas (in terms of~\eqref{EqMkYSplits0})
    \begin{align*}
      &K_{\rms 0,[0]}(D\Ric,-1) / R_{\rms 0,[1]}(\delta^*,0) \\
      &\quad = \ker N_{\rms 0}(D\Ric,-1) / \bigl(\ran N_{\rms 0}(\delta^*,0) + \ran_{\ker N_{\rms 0}(\delta^*,0)} \pa_\lambda N_{\rms 0}(\delta^*,0)\bigr) \\
      &\quad= \mathspan\{ (1,0,1,0)^T \}
    \end{align*}
    is 1-dimensional (in coordinates spanned by the equivalence class of $r^{-1}((\dd x^0)^2+(\dd x^1)^2)=\frac{2}{r}(\dd t^2+\dd r^2)$). For $\lambda=0$ finally,
    \[
      K_{\rms 0,[0]}(D\Ric,0) / R_{\rms 0,[0]}(\delta^*,1) = \mathspan\{ (1,1,1,0)^T \}
    \]
    (in coordinates, $(1,1,1,0)^T$ corresponds to $4\,\dd t^2$).
  \item\label{ItRicMSs1}{\rm (Scalar type $1$.)} For $\lambda\neq -2,-1,1$, we have $\ker N_{\rms 1}(D\Ric,\lambda)=\ran N_{\rms 1}(\delta^*,\lambda+1)$. For the exceptional values of $\lambda$, we have, in terms of~\eqref{EqMkYSplits1}\footnote{As explained above,~\eqref{EqRicMSs1m2} means that in the splitting~\eqref{EqMkYSplits1}, the space of symmetric 2-tensors of the form $(-\scal,0,\dd\scal,-\scal,-\dd\scal,0)^T$, $\scal\in\scalspace_1$, projects isomorphically onto the quotient space.}
    \begin{align}
    \label{EqRicMSs1m2}
      K_{\rms 1,[0]}(D\Ric,-2) / R_{\rms 1,[0]}(\delta^*,-1) &= \mathspan\{ (-1,0,1,-1,-1,0)^T \}, \\
    \label{EqRicMSs11}
      K_{\rms 1,[0]}(D\Ric,1) / R_{\rms 1,[0]}(\delta^*,2) &= \mathspan\{ (0,0,0,2,-1,0)^T \};
    \end{align}
    finally,
    \begin{equation}
    \label{EqRicMSs1m1}
      K_{\rms 1,[0]}(D\Ric,-1)/R_{\rms 1,[0]}(\delta^*,0)=\mathspan\Bigl\{\pa_\lambda N_{\rms 1}(\delta^*,0)\Bigl(\frac12,-\frac12,1\Bigr)^T\Bigr\}
    \end{equation}
    (in coordinates $\ubar\delta^*((\log r)\dd(r\scal))=\frac{1}{r}\dd r\otimes_s\dd(r\scal)$, $\scal\in\scalspace_1$), so
    \begin{equation}
    \label{EqRicMSs1m1Triv}
    \begin{split}
      &K_{\rms 1,[0]}(D\Ric,-1) / R_{\rms 1,[1]}(\delta^*,0) \\
      &\quad = \ker N_{\rms 1}(D\Ric,-1) / \bigl(\ran N_{\rms 1}(\delta^*,0) + \ran_{\ker N_{\rms 1}(\delta^*,0)} \pa_\lambda N_{\rms 1}(\delta^*,0)\bigr) \\
      &\quad= \{0\}.
    \end{split}
    \end{equation}
  \item\label{ItRicMSsl}{\rm (Scalar type $l\geq 2$.)} For $\lambda\neq -l-1,l$, we have $\ker N_{\rms l}(D\Ric,\lambda)=\ran N_{\rms l}(\delta^*,\lambda+1)$. For $\lambda=-l-1,l$, we have
    \begin{align*}
      K_{\rms l,[0]}(D\Ric,\lambda) / R_{\rms l,[0]}(\delta^*,\lambda+1) = \mathspan\{ (1,0,0,1,0,2,0)^T \}.
    \end{align*}
  \item\label{ItRicMSv1}{\rm (Vector type $l=1$.)} For $\lambda\neq -2,0,1$, we have $\ker N_{\rmv 1}(D\Ric,\lambda)=\ran N_{\rmv 1}(\delta^*,\lambda+1)$. Furthermore,
    \[
      K_{\rmv 1,[0]}(D\Ric,\lambda) / R_{\rmv 1,[0]}(\delta^*,\lambda+1) = \mathspan\{ (1,1)^T \},\qquad \lambda=-2,1.
    \]
    (In coordinates, $(1,1)^T$ is $4 r^\lambda\,\dd t\otimes_s r\vect$, $\vect\in\vectspace_1$.) Finally, for $\lambda=0$, the space
    \begin{equation}
    \label{EqRicMSv10}
      K_{\rmv 1,[0]}(D\Ric,0)/R_{\rmv 1,[0]}(\delta^*,1) = \mathspan \bigl\{ \pa_\lambda N_{\rmv 1}(\delta^*,1)(1,-1)^T \bigr\}
    \end{equation}
    (in coordinates $\ubar\delta^*(2 r\cdot r\vect)$, $\vect\in\vectspace_1$) is 1-dimensional, whereas
    \begin{equation}
    \label{EqRicMSv10Triv}
    \begin{split}
      &K_{\rmv 1,[0]}(D\Ric,0) / R_{\rmv 1,[1]}(\delta^*,1) \\
      &\quad = \ker N_{\rmv 1}(D\Ric,0) / \bigl(\ran N_{\rmv 1}(\delta^*,1) + \ran_{\ker N_{\rmv 1}(\delta^*,1)} \pa_\lambda N_{\rmv 1}(\delta^*,1)\bigr) \\
      &\quad = \{0\}.
    \end{split}
    \end{equation}
  \item\label{ItRicMSvl}{\rm (Vector type $l\geq 2$.)} For $\lambda\neq -l-1,l$, we have $\ker N_{\rmv l}(D\Ric,\lambda)=\ran N_{\rmv l}(\delta^*,\lambda+1)$. Furthermore,
    \[
      K_{\rmv l,[0]}(D\Ric,\lambda) / R_{\rmv l,[0]}(\delta^*,\lambda+1) = \mathspan\{ (1,1,0)^T \},\qquad \lambda=-l-1,l.
    \]
  \end{enumerate}
\end{prop}

When the quotient of $\ker N_\bullet(D\Ric,\lambda)$ by $\ran N_\bullet(\delta^*,\lambda+1)$ is nontrivial, this indicates the \emph{possibility} that there exists a homogeneous metric perturbation which is not pure gauge; in the $\rms 0$, $\rms 1$, and $\rmv 1$ cases, the quotient can contain nontrivial elements which are nonetheless pure gauge if the gauge potential is permitted to contain an additional logarithmic term. But even if one quotients out by the generalized range of $N(\delta^*,-)$, there remain certain nontrivial quotient spaces for all pure types. Certain ones have simple interpretations which will play an important role in~\S\ref{SAh}:
\begin{itemize}
\item for $\rms 0$ and $\lambda=-1$: linearized mass perturbations---compare $\frac{2}{r}(\dd t^2+\dd r^2)$ in part~\eqref{ItRicMSs0} with the first line in expression~\eqref{EqGKLot};
\item for $\rms 1$ and $\lambda=-2$: deformation tensors of translations on a mass $\neq 0$ Schwarzschild spacetime (which are \emph{not} deformation tensors on Minkowski space however). This is the reason for the choice $(-1,0,1,-1,-1,0)^T$ of basis in~\eqref{EqRicMSs1m2}: relative to the scalar type $1$ function $\scal=\scal(\hat c)=\hat c\cdot\frac{x}{|x|}$, this (times $r^{-2}$) is the leading order part of $h_{(2,0),\hat c}=\frac12\cL_{\hat c\cdot\pa_{\hat x}}\hat g_{2,0}$ where $\hat g_{2,0}$ is the mass $2$ Schwarzschild metric (though expressed in $x$-coordinates here), cf.\ \eqref{EqAhKBoostLot} and \eqref{EqAhKCokerCOMh2}; and
\item for $\rmv 1$ and $\lambda=-2$: linearized angular momentum perturbations---compare the tensor $4 r^{-2}\,\dd t\otimes_s r\vect$ with the last line in~\eqref{EqGKLot} (which is the only term in large square parentheses which is of vector type $1$, cf.\ the discussion after~\eqref{EqAhKCokerKPf3}).
\end{itemize}

\begin{proof}[Proof of Proposition~\usref{PropRicMS}]
  Since $N_\bullet(D\Ric,\lambda)$ and $N_\bullet(\delta^*,\lambda+1)$ are (finite-dimensional) matrices, this can in be verified by direct computation. For most pure types, we shall proceed in a mildly conceptual fashion by mimicking some of the arguments in \cite{KodamaIshibashiMaster} as in \cite[\S8]{HaefnerHintzVasyKerr}. We use Lemmas~\ref{LemmaMkYId} and \ref{LemmaMk}, Corollary~\ref{CorMkDiffDRic}, and the splittings~\eqref{EqMkYSplits0}--\eqref{EqMkYSplitvl} to obtain the expressions for the matrices below. (Recall that $\lambda$ is the power of $r=\rho^{-1}$ here, rather than of $\rho$ in Corollary~\ref{CorMkDiffDRic}.)

  \pfstep{Part~\eqref{ItRicMSs0}: $\rms 0$ tensors.} We note the explicit expressions
  \begin{align}
    N_{\rms 0}(\delta^*,\lambda+1)&=\frac14\begin{pmatrix} 2(\lambda+1) & 0 \\ -\lambda-1 & \lambda+1 \\ 0 & -2(\lambda+1) \\ 4 & -4 \end{pmatrix}, \nonumber\\
  \label{EqRicMSs0DRic}
    N_{\rms 0}(D\Ric,\lambda) &= \frac14\begin{pmatrix} 0 & -4\lambda & 0 & -\lambda(\lambda+1) \\ -\lambda(\lambda+1) & -2\lambda(\lambda-1) & -\lambda(\lambda+1) & \lambda(\lambda+1) \\ 0 & -4\lambda & 0 & -\lambda(\lambda+1) \\ 4(\lambda+1) & -8 & 4(\lambda+1) & -2(\lambda+1)(\lambda+2) \end{pmatrix}, \nonumber\\
    N_{\rms 0}(\delta\sfG,\lambda-2)&=\frac12\begin{pmatrix} -2\lambda & 4 & 0 & \lambda \\ 0 & -4 & 2\lambda & -\lambda \end{pmatrix}
  \end{align}
  For $\lambda\neq 0$, $N_{\rms 0}(\delta\sfG,\lambda-2)$ is surjective and thus has 2-dimensional kernel; since $N_{\rms 0}(D\Ric,\lambda)$ maps a supplementary subspace of $\ker N_{\rms 0}(D\Ric,\lambda)$ injectively into this kernel, we conclude that $\dim\ran N_{\rms 0}(D\Ric,\lambda)=4-\dim\ker N_{\rms 0}(D\Ric,\lambda)\leq 2$. But for $\lambda\neq-1,0$, the dimension of $\ran N_{\rms 0}(D\Ric,\lambda)$ is \emph{at least} 2. Thus, it must equal to $2$. Since for $\lambda\neq -1$, $\ran N_{\rms 0}(\delta^*,\lambda+1)$ is a 2-dimensional subspace of $\ker N_{\rms 0}(D\Ric,\lambda)$, we conclude that for $\lambda\neq -1,0$ we have $\ker N_{\rms 0}(D\Ric,\lambda)=\ran N_{\rms 0}(\delta^*,\lambda+1)$.

  For $\lambda=-1$, we have $\ker N_{\rms 0}(\delta^*,\lambda+1)=\mathspan\{(1,1)^T\}$ (in coordinates, this is $\ubar\delta^*(2\dd t)=0$), so $\ran N_{\rms 0}(\delta^*,\lambda+1)|_{\lambda=-1}=\mathspan\{(0,0,0,1)^T\}$ is a 1-dimensional subspace of
  \begin{equation}
  \label{EqRicMSs0m1}
    N_{\rms 0}(D\Ric,-1)=\mathspan\{(1,0,0,0)^T,(0,0,1,0)^T,(0,0,0,1)^T\}.
  \end{equation}
  We note that the range of $\pa_\lambda N_{\rms 0}(\delta^*,\lambda+1)|_{\lambda=-1}\colon(1,1)^T\mapsto\half(1,0,-1,0)^T$ (in coordinates, this is the equation $\ubar\delta^*(2\log r\,\dd t)=2 r^{-1}\dd t\otimes_s\dd r=\half r^{-1}((\dd x^0)^2-(\dd x^1)^2)$) spans a 1-dimensional subspace in the quotient $\ker N_{\rms 0}(D\Ric,-1)/\ran N_{\rms 0}(\delta^*,0)$, and a complementary subspace is spanned by $(1,0,1,0)^T$.

  For $\lambda=0$ finally, the quotient of
  \[
    \ker N_{\rms 0}(D\Ric,\lambda)=\mathspan\{(1,0,0,1)^T,(-1,0,1,0)^T,(2,1,0,0)^T\}
  \]
  by $\ran N_{\rms 0}(\delta^*,\lambda+1)|_{\lambda=0}=\mathspan\{(2,-1,0,4)^T,(0,1,-2,-4)^T\}$ is spanned e.g.\ by the image of $(1,1,1,0)^T$ in the quotient.

  \pfstep{Part~\eqref{ItRicMSs1}: $\rms 1$ tensors.} For $\lambda\neq -1$, the annihilator of the range of
  \[
    N_{\rms 1}(\delta^*,\lambda+1) = \frac14\begin{pmatrix} 2(\lambda+1) & 0 & 0 \\ -\lambda-1 & \lambda+1 & 0 \\ 2 & 0 & \lambda \\ 0 & -2(\lambda+1) & 0 \\ 0 & 2 & -\lambda \\ 4 & -4 & -4 \end{pmatrix}
  \]
  is spanned by
  \begin{alignat*}{6}
    h^*_1&=(1&&,\ 2&&,\ 0&&,\ 1&&,\ 0&&,\ 0), \\
    h^*_2&=\Bigl(-\frac{2}{\lambda+1}&&,\ {-}\frac{2}{\lambda+1}&&,\ 1&&,\ 0&&,\ 1&&,\ 0\Bigr), \\
    h^*_3&=\Bigl(-\frac{1}{\lambda+1}&&,\ \frac{\lambda}{\lambda+1}&&,\ 1&&,\ 0&&,\ 0&&,\ \frac{\lambda}{4}\Bigr).
  \end{alignat*}
  Thus, for $h\in\CI_{\rms 1}(\ff;\ubar\upbeta^*S^2 T^*_{(0,0)}\R^4)$, the value of $N(D\Ric,\lambda)h$ only depends on the gauge-invariant quantities $(X,Y,Z):=(h^*_1(h),h^*_2(h),h^*_3(h))$. Since $(X,Y,Z)=(1,0,0)$, $(0,1,0)$, $(0,0,1)$ for $h=(0,0,0,1,0,0)^T$, $(0,0,0,0,1,0)^T$, $(-\lambda-1,0,0,\lambda+1,-2,0)^T$, respectively, one computes using
  \begin{equation}
  \label{EqRicMSs1DRic}
    N_{\rms 1}(D\Ric,\lambda) = \frac12
      \begin{reduce}
      \openbigpmatrix{2pt}
        2 & {-}2\lambda & {-}2(\lambda{+}1) & 0 & 0 & {-}\frac12\lambda(\lambda{+}1) \\
        {-}\frac12\lambda(\lambda{+}1) & {-}(\lambda{-}2)(\lambda{+}1) & \lambda{+}1 & {-}\frac12\lambda(\lambda{+}1) & {-}\lambda{-}1 & \frac12\lambda(\lambda{+}1) \\
        \lambda & \lambda{-}2 & {-}\frac12(\lambda^2{+}\lambda{+}2) & 0 & {-}\frac12(\lambda{-}1)(\lambda{+}2) & {-}\frac12\lambda \\
        0 & {-}2\lambda & 0 & 2 & 2(\lambda{+}1) & {-}\frac12\lambda(\lambda{+}1) \\
        0 & {-}\lambda{+}2 & {-}\frac12(\lambda{-}1)(\lambda{+}2) & {-}\lambda & {-}\frac12(\lambda^2{+}\lambda{+}2) & \frac12\lambda \\
        2(\lambda{+}1) & {-}8 & {-}2(\lambda{+}3) & 2(\lambda{+}1) & 2(\lambda{+}3) & {-}\lambda(\lambda{+}3)
      \closebigpmatrix
      \end{reduce}
  \end{equation}
  the matrix of $N(D\Ric,\lambda)$ in terms of $(X,Y,Z)$ (for the columns) and the splitting~\eqref{EqMkYSplits1} (for the rows) to be
  \[
    N_{\rms 1}(D\Ric,\lambda) =
      \frac12
      \begin{pmatrix}
        0 & 0 & -2(\lambda+1) \\
        -\frac12\lambda(\lambda+1) & -\lambda-1 & 2(\lambda+1) \\
        0 & -\frac12(\lambda-1)(\lambda+2) & -2 \\
        2 & 2(\lambda+1) & -2(\lambda+1) \\
        -\lambda & -\frac12(\lambda^2+\lambda+2) & 2 \\
        2(\lambda+1) & 2(\lambda+3) & -4(\lambda+3)
      \end{pmatrix}.
  \]
  Since $\lambda\neq -1$, the third column is linearly independent from the first two (cf.\ the $(1,3)$ entry). The second column is $\lambda+1$ times the first (cf.\ the fourth row) if and only if $\lambda=-2,1$. We conclude that for $\lambda\neq -2,-1,1$, $N_{\rms 1}(D\Ric,\lambda)h=0$ implies $(X,Y,Z)=0$ and thus $h\in\ran N_{\rms 1}(\delta^*,\lambda+1)$. For $\lambda=-2$, resp.\ $\lambda=1$, the kernel of $N_{\rms 1}(D\Ric,\lambda)$ modulo $\ran N_{\rms 1}(\delta^*,\lambda+1)$ is spanned by $h$ with $(X,Y,Z)=(1,1,0)$, resp.\ $(X,Y,Z)=(2,-1,0)$ e.g.\ $h=\half(1,0,-1,1,1,0)^T$, resp.\ $h=(0,0,0,2,-1,0)^T$.

  For $\lambda=-1$, $N_{\rms 1}(\delta^*,\lambda+1)$ has 1-dimensional kernel spanned by $(\frac12,-\frac12,1)^T$ (corresponding in coordinates to the fact that translation 1-forms $\dd(r\scal)$, $\scal\in\scalspace_1$, are Killing 1-forms on Minkowski space). One computes
  \begin{align*}
    \ran N_{\rms 1}(\delta^*,\lambda+1)|_{\lambda=-1} &= \mathspan\{ (0,0,1,0,0,2)^T, (0,0,1,0,1,0)^T \}, \\
    \ker N_{\rms 1}(D\Ric,\lambda)|_{\lambda=-1} &= \mathspan\{ (0,0,1,0,0,2)^T, (0,0,1,0,1,0)^T, (1,-1,1,1,-1,0)^T \},
  \end{align*}
  so the quotient space is 1-dimensional; but
  \[
    4\pa_\lambda N_{\rms 1}(\delta^*,\lambda+1)\Bigl(\frac12,-\frac12,1\Bigr)^T=(1,-1,1,1,-1,0)^T
  \]
  (in coordinates $\ubar\delta^*((\log r)\dd(r\scal))=\frac{1}{r}\dd r\otimes_s\dd(r\scal)$) spans this quotient.

  \pfstep{Part~\eqref{ItRicMSsl}: $\rms l$ tensors, $l\geq 2$.} Now
  \[
    N_{\rms l}(\delta^*,\lambda+1)
      =\frac14\begin{pmatrix}
         2(\lambda+1) & 0 & 0 \\
         -\lambda-1 & \lambda+1 & 0 \\
         2 & 0 & \lambda \\
         0 & -2(\lambda+1) & 0 \\
         0 & 2 & -\lambda \\
         4 & -4 & -2 l(l+1) \\
         0 & 0 & 4
       \end{pmatrix}
  \]
  is injective for all $\lambda\in\C$. The annihilator of its range is spanned by\footnote{While we do not attempt to exactly parallel the arguments in~\cite[\S8.1.1]{HaefnerHintzVasyKerr} here, we do remark that the choice of $h^*_4$ (which in particular tests the spherical pure trace and trace free parts of a symmetric 2-tensor) is inspired by the definition of the gauge-invariant quantity $J$ in~\cite[Equation~(8.8)]{HaefnerHintzVasyKerr}, and similarly the choice of $(h^*_1,h^*_2,h^*_3)$ (which tests for the non-spherical part of a symmetric 2-tensor) is inspired by $\wt F$ in the reference.}
  \begin{alignat*}{8}
    h^*_1&=\Bigl(1&&,\ 0&&,\ -\lambda-1&&,\ 0&&,\ 0&&,\ 0&&,\ \frac{\lambda(\lambda+1)}{4}\Bigr), \\
    h^*_2&=\Bigl(0&&,\ 1&&,\ \frac{\lambda+1}{2} &&,\ 0&&,\ -\frac12(\lambda+1)&&,\ 0&&,\ -\frac{\lambda(\lambda+1)}{4}\Bigr), \\
    h^*_3&=\Bigl(0&&,\ 0&&,\ 0&&,\ 1&&,\ \lambda+1&&,\ 0&&,\ \frac{\lambda(\lambda+1)}{4}\Bigr), \\
    h^*_4&=\Bigl(0&&,\ 0&&,\ -2&&,\ 0&&,\ 2&&,\ 1&&,\ \frac{l(l+1)}{2}+\lambda\Bigr).
  \end{alignat*}
  Given $h\in\CI_{\rms l}(\ff;\ubar\upbeta^*S^2 T^*_{(0,0)}\R^4)$, the vector $(X,Y,Z,J):=(h_1^*(h),h_2^*(h),h_3^*(h),h_4^*(h))$ is thus gauge-invariant. Noting that $(1,0,0,0,0,0,0)^T$, $(0,1,0,0,0,0,0)^T$, $(0,0,0,1,0,0,0)^T$, $(0,0,0,0,0,1,0)^T$ is a dual basis (i.e.\ the matrix of evaluations against $h_1^*,\ldots,h_4^*$ is the identity matrix), the action of
    \begin{align}
    \label{EqRicMSslDRic}
    &2 N_{\rms l}(D\Ric,\lambda) \\
    &\begin{reduce2}=
      \openbigpmatrix{3pt}
        l(l{+}1) & {-}2\lambda & {-}l(l{+}1)(\lambda{+}1) & 0 & 0 & {-}\frac12\lambda(\lambda{+}1) & 0 \\
        {-}\frac12\lambda(\lambda{+}1) & -(\lambda{-}l{-}1)(\lambda{+}l) & \frac12 l(l{+}1)(\lambda{+}1) & {-}\frac12\lambda(\lambda{+}1) & {-}\frac12 l(l{+}1)(\lambda{+}1) & \frac12\lambda(\lambda{+}1) & 0 \\
        \lambda & \lambda{-}2 & {-}\frac12(\lambda^2{+}\lambda{+}2) & 0 & {-}\frac12(\lambda{-}1)(\lambda{+}2) & {-}\frac{\lambda}{2} & {-}\frac14(l{-}1)(l{+}2)\lambda \\
        0 & {-}2\lambda & 0 & l(l{+}1) & l(l{+}1)(\lambda{+}1) & {-}\frac12\lambda(\lambda{+}1) & 0 \\
        0 & {-}\lambda{+}2 & {-}\frac12(\lambda{-}1)(\lambda{+}2) & {-}\lambda & {-}\frac12(\lambda^2{+}\lambda{+}2) & \frac{\lambda}{2} & \frac14(l{-}1)(l{+}2)\lambda \\
        2(\lambda{+}1) & {-}2(l^2{+}l{+}2) & {-}l(l{+}1)(\lambda{+}3) & 2(\lambda{+}1) & l(l{+}1)(\lambda{+}3) & {-}(\lambda{-}l{+}1)(\lambda{+}l{+}2) & \frac12(l{-}1)l(l{+}1)(l{+}2) \\
        0 & 4 & 2(\lambda{+}1) & 0 & {-}2(\lambda{+}1) & 0 & {-}\lambda(\lambda{+}1)
      \closebigpmatrix
      \end{reduce2} \nonumber
  \end{align}
  on $h$ is thus given, in terms of $(X,Y,Z,J)$ (for the columns) and the splitting~\eqref{EqMkYSplitsl} (for the rows), by the first, second, fourth, and sixth columns, so
  \[
    2 N_{\rms l}(D\Ric,\lambda) = 
      \openbigpmatrix{2pt}
        l(l{+}1) & {-}2\lambda & 0 & {-}\frac12\lambda(\lambda{+}1) \\
        {-}\frac12\lambda(\lambda{+}1) & -(\lambda{-}l{-}1)(\lambda{+}l) & {-}\frac12\lambda(\lambda{+}1) & \frac12\lambda(\lambda{+}1)  \\
        \lambda & \lambda{-}2 & 0 & {-}\frac{\lambda}{2} \\
        0 & {-}2\lambda & l(l{+}1) & {-}\frac12\lambda(\lambda{+}1) \\
        0 & {-}\lambda{+}2 & {-}\lambda & \frac{\lambda}{2} \\
        2(\lambda{+}1) & {-}2(l^2{+}l{+}2) & 2(\lambda{+}1) & {-}(\lambda{-}l{+}1)(\lambda{+}l{+}2) \\
        0 & 4 & 0 & 0
      \closebigpmatrix.
  \]
  We need to determine for which $\lambda\in\C$ this matrix has a nontrivial nullspace. Note that the second column is linearly independent from the span of the other three columns which we denote $r_1,r_3,r_4$. When $\lambda=0$, then $r_1$ is linearly independent of $\mathspan\{r_3,r_4\}$ (cf.\ the $(1,1)$ entry), and $r_3,r_4$ are linearly independent as well (cf.\ the fourth row). When $\lambda\neq 0$, then $a r_1+b r_3+2 c r_4=0$ implies $(b-c)\lambda=0$ (fifth row), so $b=c$, and furthermore $b l(l+1)=c\lambda(\lambda+1)$ (fourth row), so $\lambda=-l-1,l$. And indeed for $\lambda=-l-1,l$, tensors with $(X,Y,Z,J)=(1,0,1,2)$ (for either value of $\lambda$) lie in the kernel of $N_{\rms l}(D\Ric,\lambda)$.

  \pfstep{Part~\eqref{ItRicMSv1}: $\rmv 1$ tensors.} We have
  \begin{equation}
  \label{EqRicMSv1DRic}
    N_{\rmv 1}(\delta^*,\lambda+1) = \frac14 \begin{pmatrix} \lambda \\ -\lambda \end{pmatrix},\qquad
    N_{\rmv 1}(D\Ric,\lambda) = -\frac{(\lambda-1)(\lambda+2)}{4} \begin{pmatrix} 1 & 1 \\ 1 & 1 \end{pmatrix}.
  \end{equation}
  For $\lambda\neq -2,1$, $\ker N_{\rmv 1}(D\Ric,\lambda)$ is therefore 1-dimensional and spanned by $(1,-1)^T$, which for $\lambda\neq 0$ spans the range of $N_{\rmv 1}(\delta^*,\lambda+1)$. Turning to the exceptional values, note that for $\lambda=-2,1$, the quotient space $\ker N_{\rmv 1}(D\Ric,\lambda)/\ran N_{\rmv 1}(\delta^*,\lambda+1)$ is spanned by $\{(1,1)^T\}$. (Since $N_{\rmv 1}(\delta^*,\lambda+1)$ is injective, quotienting out by the generalized range gives the same space.) For $\lambda=0$, the quotient space is 1-dimensional (and spanned by $(1,-1)^T$), but in view of $\ran(\pa_\lambda N_{\rmv 1}(\delta^*,\lambda+1))=\mathspan\{(1,-1)^T\}$ the quotient by $R_{\rmv 1,[1]}(N_{\rmv 1}(\delta^*,\lambda+1))$ is trivial.

  \pfstep{Part~\eqref{ItRicMSvl}: $\rmv l$ tensors, $l\geq 2$.} Similarly to the scalar type $l\geq 2$ case, the operator
  \[
    N_{\rmv l}(\delta^*,\lambda+1) = \frac14 \begin{pmatrix} \lambda \\ -\lambda \\ 4 \end{pmatrix}
  \]
  is now always injective; the annihilator of its range is spanned by $h_1^*=(1,0,-\frac14\lambda)^T$, $h_2^*=(0,1,\frac14\lambda)$. A dual basis being $(1,0,0)$, $(0,1,0)$, we can compute the action of
  \begin{equation}
  \label{EqRicMSvlDRic}
    -4 N_{\rmv l}(D\Ric,\lambda) = \begin{pmatrix} {-}2(l^2{+}l{-}1){+}\lambda(\lambda{+}1) & (\lambda{-}1)(\lambda{+}2) & \frac12(l{-}1)(l{+}2)\lambda \\ (\lambda{-}1)(\lambda{+}2) & {-}2(l^2{+}l{-}1){+}\lambda(\lambda{+}1) & {-}\frac12(l{-}1)(l{+}2)\lambda \\ {-}4(\lambda{+}1) & 4(\lambda{+}1) & 2\lambda(\lambda{+}1) \end{pmatrix}
  \end{equation}
  on a tensor $h$ in terms of $(X,Y)=(h_1^*(h),h_2^*(h))$ (for the columns) to be given by
  \[
    -4 N_{\rmv l}(D\Ric,\lambda) = \begin{pmatrix} {-}2(l^2{+}l{-}1){+}\lambda(\lambda{+}1) & (\lambda{-}1)(\lambda{+}2) \\ (\lambda{-}1)(\lambda{+}2) & {-}2(l^2{+}l{-}1){+}\lambda(\lambda{+}1) \\ {-}4(\lambda{+}1) & 4(\lambda{+}1) \end{pmatrix}
  \]
  For $\lambda=-1$, this has trivial nullspace (using $l^2+l-1\neq\pm 1$). For $\lambda\neq -1$, the only possible nontrivial linear combination of the two columns which gives the zero vector is (up to scalar multiples) their sum (cf.\ the third row), i.e.\ $(X,Y)=(1,1)$ (corresponding e.g.\ to $h=(1,1,0)^T$). This indeed vanishes if $0=-2(l^2+l-1)+\lambda(\lambda+1)+(\lambda-1)(\lambda+2)=2(\lambda-l)(\lambda+l+1)$, which has the solutions $\lambda=-l-1,l$. This implies the claim.
\end{proof}

What this result does not yet treat is the generalized nullspace of $N(D\Ric,\lambda)$, i.e.\ the possibility and structure of solutions of $D_{\ubar g}\Ric(\sum_{j=0}^k r^\lambda (\log r)^j h_j)=0$ with $k\geq 1$. Note that for such solutions, the leading order term $h_k$ necessarily lies in $\ker N(D\Ric,\lambda)$. We will prove in Proposition~\ref{PropRicK} that leading order terms for $k\geq 1$ are always pure gauge; this is thus a result on the `semi-simplicity modulo pure gauge' for the (generalized) nullspace of $N(D\Ric,\lambda)$.

In this analysis of the generalized nullspace of $N(D\Ric,\lambda)$, we will use duality and pairing arguments. The fiber inner products are induced by the Minkowski metric. Making explicit the volume density used for the definition of adjoints (and recalling that the Minkowskian volume density is $r^2\dd r$ times the standard density on $\Sph^2$), we then note that
\begin{equation}
\label{EqRicAdj}
\begin{alignedat}{2}
  N(D\Ric,\lambda)^* &:= N(r^2\wh{D_{\ubar g}\Ric}(0),\lambda)^* &&= \bigl(r^{-\lambda} r^2 \wh{D_{\ubar g}\Ric}(0) r^\lambda\bigr)^{*,|\frac{\dd r}{r}|} \\
    &= r^3 \bigl(r^{-\lambda} r^2 \wh{D_{\ubar g}\Ric}(0) r^\lambda\bigr)^{*,|r^2\dd r|}r^{-3} &&= r^3 \bigl( r^{\bar\lambda} \wh{D_{\ubar g}\Ric}(0)^* r^2 r^{-\bar\lambda}\bigr) r^{-3} \\
    &= r^{1+\bar\lambda} r^2\wh{D_{\ubar g}\Ric}(0)^* r^{-1-\bar\lambda} &&= N\bigl(r^2 \ubar\sfG\wh{D_{\ubar g}\Ric}(0)\ubar\sfG, -1-\bar\lambda\bigr) \\
    &= \sfG\circ N(D\Ric,-1-\bar\lambda)\circ\sfG,
\end{alignedat}
\end{equation}
and similarly
\begin{equation}
\label{EqRicAdjDel}
\begin{alignedat}{2}
  N(\delta^*,\lambda+1)^* &= N(\delta,-3-\bar\lambda) &&= N(\delta\sfG,-3-\bar\lambda)\circ\sfG, \\
  N(\delta\sfG,\lambda-2)^* &= \sfG\circ N(\delta^*,-\bar\lambda) &&= N(\sfG\delta^*,-\bar\lambda).
\end{alignedat}
\end{equation}
Thus, Proposition~\ref{PropLA2} asserts that for every pure type $\bullet\in\{\rms l,\,\rmv l\}$, the pairings
\begin{equation}
\label{EqRicPair}
\begin{split}
  &\bigl( K_{\bullet,[j]}(\delta\sfG,\lambda-2) / R_{\bullet,[i]}(D\Ric,\lambda) \bigr) \times \bigl( K_{\bullet,[i]}(D\Ric,-1-\bar\lambda) / R_{\bullet,[j]}(\delta^*,-\bar\lambda) \bigr) \\
  &\qquad \ni (h,h^*) \mapsto \la h,\sfG h^*\ra \in \C
\end{split}
\end{equation}
are non-degenerate.

\begin{prop}[Restricted kernels are pure gauge]
\fakephantomsection
\label{PropRicK}
  \begin{enumerate}
  \item\label{ItRicKs0}{\rm (Scalar type $0$.)} We have
    \begin{subequations}
    \begin{align}
    \label{EqRicKs0m1}
    \begin{split}
      K_{\rms 0,[1]}(D\Ric,-1) &= \bigl\{ h\in\ker N_{\rms 0}(D\Ric,-1) \colon \pa_\lambda N_{\rms 0}(D\Ric,-1)h \in \ran N_{\rms 0}(D\Ric,0) \bigr\} \\
        &= R_{\rms 0,[1]}(\delta^*,0),
    \end{split} \\
    \label{EqRicKs00}
      K_{\rms 0,[1]}(D\Ric,0) &= R_{\rms 0,[0]}(\delta^*,1).
    \end{align}
    \end{subequations}
    Moreover,
    \begin{equation}
    \label{EqRicKs0m1R0}
      K_{\rms 0,[1]}(D\Ric,-1)/R_{\rms 0,[0]}(\delta^*,0) = \mathspan\{ \pa_\lambda N(\delta^*,0)(1,1)^T \}
    \end{equation}
    (in local coordinates, $\pa_\lambda N(\delta^*,0)(1,1)^T$ is $\ubar\delta^*(2(\log r)\dd t)=2 r^{-1}\,\dd r\otimes_s\dd t$).
  \item\label{ItRicKs1}{\rm (Scalar type $1$.)} We have
    \begin{subequations}
    \begin{align}
    \label{EqRicKs1m2}
      K_{\rms 1,[1]}(D\Ric,-2) &= R_{\rms 1,[0]}(\delta^*,-1), \\
    \label{EqRicKs11}
      K_{\rms 1,[1]}(D\Ric,1) &= R_{\rms 1,[0]}(\delta^*,2).
    \end{align}
    \end{subequations}
  \item\label{ItRicKsl}{\rm (Scalar type $l\geq 2$.)} We have $K_{\rms l,[1]}(D\Ric,\lambda)=R_{\rms l,[0]}(\delta^*,\lambda+1)$ for $\lambda=-l-1,l$.
  \item\label{ItRicKv1}{\rm (Vector type $1$.)} We have $K_{\rmv 1,[1]}(D\Ric,\lambda)=R_{\rmv l,[0]}(\delta^*,\lambda+1)$ for $\lambda=-2,1$.
  \item\label{ItRicKvl}{\rm (Vector type $l\geq 2$.)} We have $K_{\rmv l,[1]}(D\Ric,\lambda)=R_{\rmv l,[0]}(\delta^*,\lambda+1)$ for $\lambda=-l-1,l$.
  \end{enumerate}
\end{prop}
\begin{proof}
  Let $h\in K_{\bullet,[0]}(D\Ric,\lambda)=\ker N_\bullet(D\Ric,\lambda)$. In order to find a condition whether $h\in K_{\bullet,[1]}(D\Ric,\lambda)$, i.e.\ $\pa_\lambda N_\bullet(D\Ric,\lambda)h\in\ran N_\bullet(D\Ric,\lambda)$, note that
  \[
    \pa_\lambda N_\bullet(D\Ric,\lambda)h\in K_{\bullet,[0]}(\delta\sfG,\lambda-2)
  \]
  since on $\ker N_\bullet(D\Ric,\lambda)$ we have
  \[
    N_\bullet(\delta\sfG,\lambda-2)\circ\pa_\lambda N_\bullet(D\Ric,\lambda)=-\pa_\lambda N_\bullet(\delta\sfG,\lambda-2)\circ N_\bullet(D\Ric,\lambda)=0.
  \]
  Thus, we have a well-defined linear map
  \begin{equation}
  \label{EqRicKPair}
    \bigl(K_{\bullet,[0]}(D\Ric,-1-\bar\lambda) / R_{\bullet,[0]}(\delta^*,-\bar\lambda) \bigr) \ni h^* \mapsto \la \pa_\lambda N_\bullet(D\Ric,\lambda)h,\sfG h^*\ra.
  \end{equation}
  By~\eqref{EqRicPair}, this is the zero map if and only if $\pa_\lambda N_\bullet(D\Ric,\lambda)h\in R_{\bullet,[0]}(D\Ric,\lambda)$, i.e.\ $h\in K_{\bullet,[1]}(D\Ric,\lambda)$. Proposition~\ref{PropRicMS} describes the spaces of inputs $h^*$ in the pairing~\eqref{EqRicKPair}.

  Now,~\eqref{EqRicKPair} depends only on the equivalence class of $h$ in $K_{\bullet,[0]}(D\Ric,\lambda)/R_{\bullet,[N]}(\delta^*,\lambda+1)$ where $N\in\N_0$ is arbitrary: this follows from the fact that for $\omega(\mu+1)=\sum_{j=0}^N(\mu-\lambda)^{-j-1}\omega_j$ such that $N_\bullet(\delta^*,\mu+1)\omega(\mu+1)$ is holomorphic at $\mu=\lambda$ we have
  \begin{align*}
    &\big\la \pa_\mu N_\bullet(D\Ric,\mu) \bigl( N_\bullet(\delta^*,\mu+1)\omega(\mu+1) \bigr), \sfG h^*\big\ra|_{\mu=\lambda} \\
    &\quad = -\big\la N_\bullet(D\Ric,\lambda) \pa_\mu\bigl( N_\bullet(\delta^*,\mu+1)\omega(\mu+1) \bigr), \sfG h^*\big\ra|_{\mu=\lambda} \\
    &\quad = \big\la \pa_\lambda\bigl( N_\bullet(\delta^*,\lambda+1)\omega(\lambda+1) \bigr), \sfG N_\bullet(D\Ric,-1-\bar\lambda)h^*\big\ra|_{\mu=\lambda} \\
    &\quad = 0.
  \end{align*}
  We have thus shown that
  \begin{equation}
  \label{EqRicKChar}
  \begin{split}
    &K_{\bullet,[1]}(D\Ric,\lambda) / R_{\bullet,[N]}(\delta^*,\lambda+1) \\
    &\quad = \Bigl\{ [h] \in K_{\bullet,[0]}(D\Ric,\lambda) / R_{\bullet,[N]}(\delta^*,\lambda+1) \colon \la \pa_\lambda N_\bullet(D\Ric,\lambda)h,\sfG h^*\ra=0 \\
    &\quad \hspace{14em} \forall h^*\in K_{\bullet,[0]}(D\Ric,-1-\bar\lambda) / R_{\bullet,[0]}(\delta^*,-\bar\lambda) \Bigr\}.
  \end{split}
  \end{equation}
  We use this characterization to show that $K_{\bullet,[1]}(D\Ric,\lambda)=R_{\bullet,[N]}(\delta^*,\lambda+1)$ for various values of $\bullet$, $\lambda$, and $N$.

  \pfstep{Part~\eqref{ItRicKs0}: $\rms 0$ tensors.} Consider first~\eqref{EqRicKs0m1}, and $\lambda=-1$ in~\eqref{EqRicKChar}. By Proposition~\ref{PropRicMS}\eqref{ItRicMSs0} it suffices to check that for
  \[
    h:=(1,0,1,0)^T\in K_{\rms 0,[0]}(D\Ric,-1),\qquad
    h^*=(1,1,1,0)^T\in K_{\rms 0,[0]}(D\Ric,0),
  \]
  we have $\la\pa_\lambda N_{\rms 0}(D\Ric,-1)h,\sfG h^*\ra\neq 0$. The matrix of the fiber inner product on $\rms 0$ sections of $\ubar\upbeta^*S^2 T^*_{(0,0)}\R^4$ tensors can be read off from~\eqref{EqMkMinkInner} and is
  \[
    \begin{pmatrix} 0 & 0 & 4 & 0 \\ 0 & 8 & 0 & 0 \\ 4 & 0 & 0 & 0 \\ 0 & 0 & 0 & 2 \end{pmatrix}.
  \]
  Using the expression~\eqref{EqRicMSs0DRic}, we thus compute
  \begin{align*}
    &\big\la \pa_\lambda N_{\rms 0}(D\Ric,-1)(1,0,1,0)^T, \sfG(1,1,1,0)^T \big\ra_{L^2(\Sph^2)} \\
    &\qquad = \Big\la \Bigl(0,\frac12,0,2\Bigr)^T, (1,0,1,2)^T \Big\ra_{L^2(\Sph^2)} = 32\pi \neq 0.
  \end{align*}

  The statement~\eqref{EqRicKs00} follows from the same calculation: we now take $h=(1,1,1,0)^T$ to span $K_{\rms 0,[0]}(D\Ric,0)/R_{\rms 0,[0]}(\delta^*,1)$; for $h^*=(1,0,1,0)^T\in K_{\rms 0,[0]}(D\Ric,-1)$, we then have
  \[
    \la\pa_\lambda N_{\rms 0}(D\Ric,0)h,\sfG h^*\ra=\la\sfG h,\pa_\lambda N_{\rms 0}(D\Ric,-1)h^*\ra = 32\pi \neq 0.
  \]
  (That is, the roles of $h,h^*$ and $0,-1$ are reversed compared to the first computation.)

  To prove~\eqref{EqRicKs0m1R0}, note that $R_{\rms 0,[0]}(\delta^*,0)$ is a 1-codimensional subspace of $R_{\rms 0,[1]}(\delta^*,0)=K_{\rms 0,[1]}(D\Ric,-1)$, with the quotient spanned by $\pa_\lambda N(\delta^*,0)(1,1)^T$; see the computations following~\eqref{EqRicMSs0m1}.

  \pfstep{Part~\eqref{ItRicKs1}: $\rms 1$ tensors.} Since for $\lambda=-2$ we have $-1-\bar\lambda=1$, we calculate with $h=(-1,0,1,-1,-1,0)^T$ and $h^*=(0,0,0,2,-1,0)^T$ (see~\eqref{EqRicMSs1m2}--\eqref{EqRicMSs11}), and using~\eqref{EqRicMSs1DRic} and the fiber inner product on $\rms 1$ tensors (from~\eqref{EqMkMinkInner})
  \[
    \begin{pmatrix} 0 & 0 & 0 & 4 & 0 & 0 \\ 0 & 8 & 0 & 0 & 0 & 0 \\ 0 & 0 & 0 & 0 & -8 & 0 \\ 4 & 0 & 0 & 0 & 0 & 0 \\ 0 & 0 & -8 & 0 & 0 & 0 \\ 0 & 0 & 0 & 0 & 0 & 2 \end{pmatrix}
  \]
  the pairing
  \begin{align*}
    &\la\pa_\lambda N_{\rms 1}(D\Ric,-2)h,\sfG h^*\ra_{L^2(\Sph^2)} \\
    &\qquad = \Big\la \Bigl(-1,-\frac12,-\frac12,-1,\frac12,-4\Bigr)^T, (0,0,0,2,-1,0)^T \Big\ra_{L^2(\Sph^2)} = -48\pi \neq 0.
  \end{align*}
  This gives~\eqref{EqRicKs1m2}; and~\eqref{EqRicKs11} follows from the same calculation with $h,h^*$ and $-2,1$ interchanged.

  \pfstep{Part~\eqref{ItRicKsl}: $\rms l$ tensors, $l\geq 2$.} Using $h=(1,0,0,1,0,2,0)^T=h^*$ from Proposition~\ref{PropRicMS}\eqref{ItRicMSsl}, the expression~\eqref{EqRicMSslDRic}, and the fiber inner product (from~\eqref{EqMkMinkInner}, using~\eqref{EqMkYSplitsl} and Lemma~\ref{LemmaMkYId})
  \[
    \begin{pmatrix}
      0 & 0 & 0 & 4 & 0 & 0 & 0 \\
      0 & 8 & 0 & 0 & 0 & 0 & 0 \\
      0 & 0 & 0 & 0 & -4 l(l+1) & 0 & 0 \\
      4 & 0 & 0 & 0 & 0 & 0 & 0 \\
      0 & 0 & -4 l(l+1) & 0 & 0 & 0 & 0 \\
      0 & 0 & 0 & 0 & 0 & 2 & 0 \\
      0 & 0 & 0 & 0 & 0 & 0 & \frac12(l-1)l(l+1)(l+2)
    \end{pmatrix},
  \]
  we compute
  \begin{align*}
    &\la\pa_\lambda N_{\rms l}(D\Ric,-l-1)h,\sfG h^*\ra_{L^2(\Sph^2)} \\
    &\qquad = \Big\la \Bigl(l+\frac12,0,0,l+\frac12,0,2 l+1,0\Bigr)^T, (1,1,0,1,0,0,0)^T \Big\ra_{L^2(\Sph^2)} = 16(2 l+1)\pi \neq 0.
  \end{align*}
  This implies the claims in part~\eqref{ItRicKsl}.

  \pfstep{Part~\eqref{ItRicKv1}: $\rmv 1$ tensors.} For $l=1$ and $\lambda=-2$, we use Proposition~\ref{PropRicMS}\eqref{ItRicMSv1} and take $h=(1,1)^T=h^*$. The fiber inner product being
  \[
    \begin{pmatrix} 0 & -8 \\ -8 & 0 \end{pmatrix},
  \]
  we find, using~\eqref{EqRicMSv1DRic},
  \[
    \la \pa_\lambda N_{\rmv 1}(D\Ric,-2)h,\sfG h^*\ra_{L^2(\Sph^2)} = \Big\la \Bigl(\frac32,\frac32\Bigr)^T, (1,1)^T \Big\ra  = -96\pi \neq 0.
  \]
  The case $\lambda=1$ follows from the same calculation.

  \pfstep{Part~\eqref{ItRicKvl}: $\rmv l$ tensors, $l\geq 2$.} This is similar to the case $l=1$: now $h=(1,1,0)^T=h^*$, and using~\eqref{EqRicMSvlDRic} and the fiber inner product
  \[
    \begin{pmatrix} 0 & -4 l(l+1) & 0 \\ -4 l(l+1) & 0 & 0 \\ 0 & 0 & \frac12(l-1)l(l+1)(l+2) \end{pmatrix},
  \]
  we find
  \begin{align*}
    \la \pa_\lambda N_{\rmv l}(D\Ric,-l-1)h,\sfG h^*\ra_{L^2(\Sph^2)} &= \Big\la \Bigl(l+\frac12,l+\frac12,0\Bigr)^T, (1,1,0)^T \Big\ra_{L^2(\Sph^2)} \\
      &= -16 l(l+1)(2 l+1)\pi \neq 0.
  \end{align*}
  This completes the proof.
\end{proof}

\subsection{Generalized range}
\label{SsRicR}

Let $\bullet\in\{\rms l,\,\rmv l\}$. Having studied the (restricted) kernel of $N(D\Ric,\lambda)$, we are now interested in \emph{solving} the equation $N_\bullet(D\Ric,\lambda)h=f$, or more generally $\pa_\lambda N_\bullet(D\Ric,\lambda)h+N_\bullet(D\Ric,\lambda)h_1=f$ (necessarily with $h\in\ker N_\bullet(D\Ric,\lambda)$). A necessary condition for $f\in R_{\bullet,[i]}(D\Ric,\lambda)$ for any $i\in\N_0$ is $f\in\ker N_\bullet(\delta\sfG,\lambda-2)=K_{\bullet,[0]}(\delta\sfG,\lambda-2)$. The extent to which the converse fails can be read off from the non-degeneracy of the pairings~\eqref{EqRicPair} (with $j=0$).

\begin{prop}[Solvability for the indicial family]
\label{PropRicR}
  Let $\lambda\in\C$.
  \begin{enumerate}
  \item\label{ItRicRs0}{\rm (Scalar type $0$.)} For $\lambda\neq -1,0$, we have
    \[
      K_{\rms 0,[0]}(\delta\sfG,\lambda-2) = \ker N_{\rms 0}(\delta\sfG,\lambda-2) = \ran N_{\rms 0}(D\Ric,\lambda) = R_{\rms 0,[0]}(D\Ric,\lambda).
    \]
    Furthermore,
    \begin{align*}
      K_{\rms 0,[0]}(\delta\sfG,-3) &= R_{\rms 0,[1]}(D\Ric,-1) \\
        &= \{ \pa_\lambda N_{\rms 0}(D\Ric,-1)h + N_{\rms 0}(D\Ric,-1)h_1 \colon h\in\ker N_{\rms 0}(D\Ric,-1) \}.
    \end{align*}
    Moreover, we have an isomorphism\footnote{The right hand side is $\la \pa_\lambda N_{\rms 0}(\delta\sfG,-2)f,(\frac12,\frac12)^T\ra$ in terms of~\eqref{EqMkYSplits0}.}
    \begin{equation}
    \label{EqRicRs0Iso}
      K_{\rms 0,[0]}(\delta\sfG,-2) / R_{\rms 0,[1]}(D\Ric,0) \ni f \mapsto \la \pa_\lambda N(\delta\sfG,-2)f, \dd t \ra_{L^2(\Sph^2)} \in \C.
    \end{equation}
  \item\label{ItRicRs1}{\rm (Scalar type $1$.)} We have
    \begin{align*}
      K_{\rms 1,[0]}(\delta\sfG,\lambda-2) &= R_{\rms 1,[0]}(D\Ric,\lambda),\qquad \lambda\neq -2,0,1, \\
      K_{\rms 1,[0]}(\delta\sfG,\lambda-2) &= R_{\rms 1,[1]}(D\Ric,\lambda),\qquad \lambda=-2,1.
    \end{align*}
    Moreover, we have an isomorphism\footnote{The right hand side, for fixed $\scal\in\scalspace_1$, is $\la \pa_\lambda N_{\rms 1}(\delta\sfG,-2)f,(\frac12,-\frac12,1)^T\ra$ in terms of~\eqref{EqMkYSplits1}.}
    \begin{equation}
    \label{EqRicRs1Iso}
      K_{\rms 1,[0]}(\delta\sfG,-2) / R_{\rms 1,[0]}(D\Ric,0) \ni f \mapsto \Bigl(\scalspace_1\ni\scal \mapsto \la \pa_\lambda N(\delta\sfG,-2)f, \dd(r\scal) \ra_{L^2(\Sph^2)} \Bigr) \in \scalspace_1^*.
    \end{equation}
  \item\label{ItRicRsl}{\rm (Scalar type $l\geq 2$.)} We have
    \begin{align*}
      K_{\rms l,[0]}(\delta\sfG,\lambda-2) &= R_{\rms l,[0]}(D\Ric,\lambda),\qquad \lambda\neq-l-1,l, \\
      K_{\rms l,[0]}(\delta\sfG,\lambda-2) &= R_{\rms l,[1]}(D\Ric,\lambda),\qquad \lambda=-l-1,l.
    \end{align*}
  \item\label{ItRicRv1}{\rm (Vector type $1$.)} We have
    \begin{align*}
      K_{\rmv l,[0]}(\delta\sfG,\lambda-2) &= R_{\rmv l,[0]}(D\Ric,\lambda),\qquad \lambda\neq-2,1, \\
      K_{\rmv l,[0]}(\delta\sfG,\lambda-2) &= R_{\rmv l,[1]}(D\Ric,\lambda),\qquad \lambda=-2,1.
    \end{align*}
    Moreover, we have an isomorphism\footnote{The right hand side, for fixed $\vect\in\vectspace_1$, is $\la \pa_\lambda N_{\rmv 1}(\delta\sfG,1)f,(\frac12,-\frac12)^T\ra$ in terms of~\eqref{EqMkYSplits1}.}
    \begin{equation}
    \label{EqRicRv1Iso}
      K_{\rmv 1,[0]}(\delta\sfG,-3) / R_{\rmv 1,[0]}(D\Ric,-1) \ni f \mapsto \Bigl(\vectspace_1\ni\vect \mapsto \la \pa_\lambda N(\delta\sfG,-3)f, r^2\vect \ra_{L^2(\Sph^2)} \Bigr) \in \vectspace_1^*.
    \end{equation}
  \item\label{ItRicRvl}{\rm (Vector type $l\geq 2$.)} We have
    \begin{align*}
      K_{\rmv l,[0]}(\delta\sfG,\lambda-2) &= R_{\rmv l,[0]}(D\Ric,\lambda),\qquad \lambda\neq-l-1,l, \\
      K_{\rmv l,[0]}(\delta\sfG,\lambda-2) &= R_{\rmv l,[1]}(D\Ric,\lambda),\qquad \lambda=-l-1,l.
    \end{align*}
  \end{enumerate}
\end{prop}
\begin{proof}
  Those cases in Propositions~\ref{PropRicMS}--\ref{PropRicK} in which the (restricted) kernel $K_{\bullet,[i]}(D\Ric,\lambda)$, $i\in\N_0$, is equal to the range $R_{\bullet,[0]}(\delta^*,\lambda+1)$ imply, via~\eqref{EqRicPair} with $j=0$, most of the Proposition.

  The isomorphism~\eqref{EqRicRs0Iso} follows from~\eqref{EqRicKs0m1R0} using~\eqref{EqRicPair} (with $\lambda=0$, $i=1$, $j=0$) and the fact that $(\sfG\circ\pa_\lambda N(\delta^*,0))^*=\pa_\lambda N(\delta\sfG,-2)$. Similarly, the isomorphism~\eqref{EqRicRs1Iso} follows from~\eqref{EqRicMSs1m1} using~\eqref{EqRicPair} (with $\lambda=0$, $i=0$, $j=0$); and~\eqref{EqRicRv1Iso} follows from~\eqref{EqRicMSv10} using~\eqref{EqRicPair} (with $\lambda=-1$, $i=0$, $j=0$).
\end{proof}

In~\eqref{EqRicRs0Iso}, \eqref{EqRicRs1Iso}, and~\eqref{EqRicRv1Iso}, the generalized ranges $R_{\bullet,[i]}(D\Ric,-)$ have already stabilized, i.e.\ they remain unchanged if one increases $i$ further.\footnote{The argument is as follows. Let $j\in\N_0$. If there exists $i$ such that $R_{[j]}(D\Ric,\lambda)=K_{[i]}(\delta\sfG,\lambda-2)$, then $R_{[j']}(D\Ric,\lambda)=R_{[j]}(D\Ric,\lambda)$ for all $j'\geq j$ since $R_{[j']}(D\Ric,\lambda)\supset R_{[j]}(D\Ric,\lambda)$ is always contained in $K_{[i]}(\delta\sfG,\lambda-2)$ by Proposition~\ref{PropLA2}. Note now that by~\eqref{EqRicPair}, the assumption on $i$ here is equivalent to $R_{[i]}(\delta^*,-\bar\lambda)=K_{[j]}(D\Ric,-1-\bar\lambda)$. This holds in the $\rms 0$ case with $\lambda=0$, $i=1$, $j=1$ by~\eqref{EqRicKs0m1}, in the $\rms 1$ case with $\lambda=0$, $j=0$, $i=1$ by~\eqref{EqRicMSs1m1Triv}, and in the $\rmv 1$ case with $\lambda=-1$, $j=0$, $i=1$ by~\eqref{EqRicMSv10Triv}.} Thus, these isomorphisms capture the full extent to which $r^2\wh{D_{\ubar g}\Ric}(0)h=r^{\lambda-2}f$, $f\in\ker N_\bullet(\delta\sfG,\lambda-2)$, does not have a solution $h$ which is quasi-homogeneous of degree $\lambda$. Requiring $f$ to lie in the \emph{restricted} kernel $K_{\bullet,[1]}(\delta\sfG,\lambda-2)$ \emph{does} guarantee solvability:

\begin{lemma}[Solvability for restricted right hand sides]
\label{LemmaRicRRestr}
  We have
  \begin{align*}
    K_{\rms 0,[1]}(\delta\sfG,-2) &= R_{\rms 0,[1]}(D\Ric,0), \\
    K_{\rms 1,[1]}(\delta\sfG,-2) &= R_{\rms 1,[0]}(D\Ric,0), \\
    K_{\rmv 1,[1]}(\delta\sfG,-3) &= R_{\rmv 1,[0]}(D\Ric,-1).
  \end{align*}
\end{lemma}
\begin{proof}
  This follows by duality from the non-degeneracy of~\eqref{EqRicPair} and the equalities~\eqref{EqRicKs0m1}, \eqref{EqRicMSs1m1Triv}, and \eqref{EqRicMSv10Triv}.
\end{proof}

\subsection{Kernel and range for pure types with large \texorpdfstring{$l$}{l}}
\label{SsRicLarge}

If the equation $N(D\Ric,\lambda)h=f\in\ker N(\delta\sfG,\lambda)$, with $f$ smooth, can be solved for each pure type separately, the existence of a smooth solution $h$ cannot directly be obtained by summing the infinitely many individual pure type solutions constructed above due to possible convergence issues. Instead, we work with a gauge-fixed version of the linearized Einstein equation in this and the next section. Concretely, we consider
\begin{equation}
\label{EqRicLargeL}
  L := D_{\ubar g}\Ric + \ubar\delta^*\circ\ubar\delta\ubar\sfG,
\end{equation}
which by~\eqref{EqELinDRic} is equal to $\frac12\ubar\Box$ where $\ubar\Box$ is the tensor wave operator, i.e.\ the scalar wave operator on each coefficient of a symmetric 2-tensor in the bundle trivialization induced by the coordinates $t,x$. In polar coordinates $x=r\omega$, the zero energy operator of the scalar wave operator is the spatial Laplacian, $\wh{\ubar\Box}(0)=-\pa_r^2-\frac{2}{r}\pa_r+r^{-2}\slDelta$, and therefore
\[
  N\bigl(r^2\wh{\ubar\Box}(0),\lambda\bigr) = -\lambda(\lambda+1) + \slDelta.
\]
Restricted to scalar type $l$ functions, this is multiplication by $-\lambda(\lambda+1)+l(l+1)=-(\lambda-l)(\lambda+l+1)$, with roots $\lambda=-l-1,l$. In particular, if $\lambda$ is fixed, then the restriction $N_{\geq l}(r^2\wh{\ubar\Box}(0),\lambda)$ to functions whose scalar type $\leq l$ components vanish is injective when $|\lambda|>l+1$. But since $N(r^2\wh{\ubar\Box}(0),\lambda)$ has Fredholm index $0$ since $\slDelta$ is elliptic, this implies the invertibility of $N_{\geq l}(r^2\wh{\ubar\Box}(0),\lambda)$.

Note next that $\dd t$ is of scalar type $0$, and $\dd x^j$, $j=1,2,3$, is of scalar type $1$; and therefore $\dd z^\mu\otimes_s\dd z^\nu$ (where $z=(z^\mu)_{\mu=0,\ldots,3}=(t,x^1,x^2,x^3)$) is a sum of tensors of scalar or vector type at most $2$. This allows us to pass from the $\dd z^\mu\otimes_s\dd z^\nu$ bundle splitting of $S^2 T^*\R^4$ to the splitting~\eqref{EqMkBundleSplit}. We have thus shown:

\begin{lemma}[Invertibility of the indicial family of the gauge-fixed linearized Einstein operator]
\label{LemmaRicLarge}
  Fix $\lambda\in\C$. Then for large enough $l_0\in\N_0$, the restriction $N_{\geq l_0}(r^2\hat L(0),\lambda)$ of $N(r^2\hat L(0),\lambda)\in\Diff^2(\ff;\ubar\upbeta^*S^2 T^*_{(0,0)}\R^4)$ to the space of smooth tensors all of whose scalar type $l$ and vector type $l$ components vanish for $l<l_0$ is invertible; and also the indicial operators $N(r^2\wh{\ubar\Box}(0),\lambda\pm 1)$ of the tensor wave operator on 1-forms are invertible. (Concretely, one may take $l_0>|\lambda|+3$.)
\end{lemma}

\begin{cor}[Kernel modulo pure gauge; range]
\label{CorRicLarge}
  Fix $\lambda\in\C$, and let $l_0\in\N_0$ be as in Lemma~\usref{LemmaRicLarge} (e.g.\ $l_0>|\lambda|+3$). We use the notation~\eqref{EqRicNotation}.
  \begin{enumerate}
  \item\label{ItRicLargeKer}{\rm (Kernel modulo pure gauge.)} If $h\in\CI_{\geq l_0}(\ff;\ubar\upbeta^*S^2 T^*_{(0,0)}\R^4)$ (i.e.\ $h$ is smooth with vanishing scalar and vector type $l$ components for all $l<l_0$) satisfies $N(D\Ric,\lambda)h=0$, then there exists $\omega\in\CI_{\geq l_0}(\ff;\ubar\upbeta^*T^*_{(0,0)}\R^4)$ with $N(\delta^*,\lambda+1)\omega=h$.
  \item\label{ItRicLargeRan}{\rm (Range.)} If $f\in\CI_{\geq l_0}(\ff;\ubar\upbeta^*S^2 T^*_{(0,0)}\R^4)$ satisfies $N(\delta\sfG,\lambda-2)f=0$, then there exists $h\in\CI_{\geq l_0}(\ff;\ubar\upbeta^*S^2 T^*_{(0,0)}\R^4)$ with $N(D\Ric,\lambda)h=f$.
  \end{enumerate}
\end{cor}
\begin{proof}
  We begin with part~\eqref{ItRicLargeRan}. We may solve
  \[
    N(r^2\hat L(0),\lambda)h = f
  \]
  with $h\in\CI_{\geq l_0}$ since $N_{\geq l_0}(r^2\hat L(0),\lambda)$ is invertible. Applying $N(\delta\sfG,\lambda-2)$ to this equation and using the definition~\eqref{EqRicLargeL} of $L$ as well as $f\in\ker N(\delta\sfG,\lambda-2)$ implies $N(r^2\wh{\ubar\delta\ubar\sfG}(0)\circ\wh{\ubar\delta^*}(0),\lambda-1)\eta=0$ where $\eta:=N(\ubar\delta\ubar\sfG,\lambda)h$. But $\ubar\delta\ubar\sfG\circ\ubar\delta^*=\half\ubar\Box$ (tensor wave operator on 1-forms), and therefore $\eta=0$ by Lemma~\ref{LemmaRicLarge}. This implies $N(r^2\wh{D_{\ubar g}\Ric}(0)h)=f$, as desired.

  For part~\eqref{ItRicLargeKer}, we first claim that there exists $\omega\in\CI_{\geq l_0}(\ff;\ubar\upbeta^*S^2 T^*_{(0,0)}\R^4)$ so that $h':=h+N(\delta^*,\lambda+1)\omega$ satisfies the gauge condition $N(\delta\sfG,\lambda)(h+N(\delta^*,\lambda+1)\omega)=0$. To verify this, note that $\omega$ needs to satisfy the equation
  \[
    \frac12 N(r^2\wh{\ubar\Box}(0),\lambda+1)\omega = -N(\delta\sfG,\lambda)h \in \CI_{\geq l_0}(\ff;\ubar\upbeta^*T^*_{(0,0)}\R^4),
  \]
  which does have a solution. But then $N(r^2\hat L(0),\lambda)h'=0$ implies $h'=0$ by Lemma~\ref{LemmaRicLarge}. This gives $h=-N(\delta^*,\lambda+1)\omega$. The proof is complete.
\end{proof}

\subsection{Solvability and uniqueness at \texorpdfstring{$r=0$}{the spatial origin} for quasi-ho\-mo\-ge\-neous tensors}
\label{SsRicSolv}

We can now prove the two main theorems of this section. We write $\pi_\ff(r,\omega)=\omega$ for the projection $[\R^3;\{0\}]=[0,\infty)_r\times\Sph^2\to\Sph^2$. 

\begin{thm}[Quasi-homogeneous nullspace modulo pure gauge]
\label{ThmRicUniq}
  Let $z\in\C$, $k\in\N_0$. Consider a stationary solution
  \[
    h(r,\omega) = \sum_{j=0}^k \frac{1}{j!}r^z(\log r)^j\pi_\ff^*h_j(\omega),\qquad h_j\in\CI(\ff;\ubar\upbeta^*S^2 T^*_{(0,0)}\R^4),
  \]
  of the equation $\wh{D_{\ubar g}\Ric}(0)h=0$. Then there exists a stationary 1-form
  \[
    \omega=\sum_{j=0}^{k+k'}\frac{1}{j!}r^{z+1}(\log r)^j\pi_\ff^*\omega_j(\omega),\qquad \omega_j\in\CI(\ff;\ubar\upbeta^*T^*_{(0,0)}\R^4),
  \]
  so that $h=\ubar\delta^*\omega$ if $z\notin\Z$ in which case we can take $k'=0$. Otherwise:
  \begin{enumerate}
  \item if $z\in\{-l-1,l\}$ where $l\geq 1$: there exists $\omega$, with $k'=0$, so that
    \begin{equation}
    \label{EqRicUniqEq}
      h-\ubar\delta^*\omega=r^z\pi_\ff^* h',\qquad h'\in\CI(\ff;\ubar\upbeta^*S^2 T^*_{(0,0)}\R^4),
    \end{equation}
    with $h'\in\sum_{\bullet=\rms l,\rmv l}\ker N_\bullet(D\Ric,z)$;
  \item if $z=-1$, set $k'=1$ (unless $h_{\rms 0,k}=0$ and $h_{\rms 1,k}=0$, in which case $k'=0$ works); if $z=0$, set $k'=1$ (unless $h_{\rmv 1,k}=0$, in which case $k'=0$ works): then there exists $\omega$ so that~\eqref{EqRicUniqEq} holds with $h'\in\ker N_{\rms 0}(D\Ric,z)$.\footnote{Thus, if $h$ does not have any scalar type $0$ components, then we can find $\omega$ so that $h'=0$, so $h$ is pure gauge.}
  \end{enumerate}
  One can moreover choose the 1-forms $\omega_j$ to depend continuously on $(h_0,\ldots,h_k)$ for fixed $z\in\C$.
\end{thm}

Thus, for non-integer $z$, the metric perturbation $h$ is always pure gauge, whereas for integer $z$ one can always add a pure gauge term to $h$ so as to eliminate all terms involving $(\log r)^j$ with $j\geq 1$.

\begin{proof}[Proof of Theorem~\usref{ThmRicUniq}]
  \pfstep{Step 1. Analysis for individual pure types.}

  \pfsubstep{(1.i)}{$z\notin\Z$.} Fix a pure type $\bullet$. Note that $h_k\in K_{\bullet,[0]}(D\Ric,z)=R_{\bullet,[0]}(\delta^*,z+1)$ (using Proposition~\ref{PropRicMS}). Therefore, we can write $h_k=N(\delta^*,z+1)\omega_k$. The remaining error
  \begin{equation}
  \label{EqRicUniqErr}
    h - \ubar\delta^*\bigl((\log r)^k\omega_k\bigr) = \sum_{j=0}^{k-1}\frac{1}{j!}r^z(\log r)^j\pi_\ff^* h_j - [\ubar\delta^*,(\log r)^k](r^{z+1}\pi_\ff^*\omega_k)
  \end{equation}
  still lies in the kernel of $\wh{D_{\ubar g}\Ric}(0)$, but the largest power of $\log r$ is reduced by $1$. Iterating this argument until $k=0$ (in which case the remaining error vanishes) produces the desired gauge potential $\omega$.

  \pfsubstep{(1.ii)}{$z=-l-1,l$ with $l\geq 1$.} The same arguments apply to the pure gauge $\bullet\notin\{\rms l,\,\rmv l\}$ part of $h$ by Proposition~\ref{PropRicMS}. Consider next the $\rms l$ part $h_{\rms l}$ of $h$. If $k=0$, we may simply take $\omega=0$. For $k\geq 1$ on the other hand, note that $h_{\rms l,k}$ (the $\rms l$ part of $h_k$) lies in $K_{\rms l,[k]}(D\Ric,z)\subset K_{\rms l,[1]}(D\Ric,z)$; but $K_{\rms l,[1]}(D\Ric,z)=R_{\rms l,[0]}(\delta^*,z+1)$ by Proposition~\ref{PropRicK} (part~\eqref{ItRicKs1} for $l=1$, and part~\eqref{ItRicKsl} for $l\geq 2$). Therefore, we can write $h_{\rms l,k}=N(\delta^*,z+1)\omega_k$. The computation~\eqref{EqRicUniqErr} again applies and reduces the largest logarithmic exponent $k$ by $1$, until one reaches $k=0$ in which case the remaining error $h'$, which has no logarithmic terms anymore, can no longer be solved away. The same reasoning applies in the $\rmv l$ case, now using parts~\eqref{ItRicKv1} and \eqref{ItRicKvl} of Proposition~\ref{PropRicK}.

  \pfsubstep{(1.iii)}{$z=-1$.} Proposition~\ref{PropRicMS} shows that we only need to consider the scalar type $0$ and $1$ cases. Consider first the $\rms 0$ case. If $k=0$, we take $\omega=0$. For $k\geq 1$, we have $h_{\rms 0,k}\in K_{\rms 0,[1]}(D\Ric,-1)=R_{\rms 0,[1]}(\delta^*,0)$, where we use Proposition~\ref{PropRicK}\eqref{ItRicKs0}. Therefore, we can find $\omega_{k+1}\in\ker N(\delta^*,0)$ and $\omega_k$ so that
  \[
    r^{-1} h_{\rms 0,k}=\ubar\delta^*\bigl((\log r)\omega_{k+1}+\omega_k\bigr).
  \]
  Using $\ubar\delta^*\circ(\log r)^{k+1}=(\log r)^{k+1}\ubar\delta^*+(k+1)(\log r)^k[\ubar\delta^*,\log r]$, we thus have
  \begin{align*}
    &\frac{1}{k!}r^{-1}(\log r)^k h_{\rms 0,k} - \ubar\delta^*\Bigl( \frac{1}{(k+1)!}(\log r)^{k+1}\omega_{k+1} + \frac{1}{k!}(\log r)^k\omega_k \Bigr) \\
    &\quad = \frac{1}{k!}(\log r)^k \Bigl( r^{-1} h_{\rms 0,k} - \ubar\delta^*\bigl( (\log r)\omega_{k+1} \bigr) + \omega_k \bigr) \Bigr) - \frac{1}{k!}[\ubar\delta^*,(\log r)^k]\omega_k \\
    &\quad = - \frac{1}{k!}[\ubar\delta^*,(\log r)^k]\omega_k,
  \end{align*}
  which has one power of $\log r$ less than $h$. We can thus eliminate all logarithmic terms of $h$ until we are left with a stationary error term $h'$.

  In the $\rms 1$ case, the same arguments apply for \emph{all} $k\geq 0$ by virtue of~\eqref{EqRicMSs1m1Triv}. Thus, $h_{\rms 1}$ is pure gauge in this case.

  \pfsubstep{(1.iv)}{$z=0$.} We only need to consider the scalar type $0$ and vector type $1$ cases in view of Proposition~\ref{PropRicMS}. In the $\rms 0$ case, for $k\geq 1$ we can use~\eqref{EqRicKs00} to solve away all logarithmic terms (i.e.\ $k\geq 1$) as in step~\textbf{(1.i)} until we are left with a tensor $h'$ without logarithmic terms. In the $\rmv 1$ case on the other hand, we use~\eqref{EqRicMSv10Triv} to solve away \emph{all} terms using the same arguments as in step~\textbf{(1.iii)}.

  \pfstep{Step 2. Analysis for all pure types with large $l$ simultaneously.} Let $l_0\in\N_0$ be as in Corollary~\ref{CorRicLarge} with $\lambda=z$. We may replace $h$ by its projection $h_{\geq l_0}$ off all spaces of scalar and vector type $l<l_0$ tensors. Since $h_k\in\ker N(D\Ric,z)$, we can thus pick $\omega_k\in\CI_{\geq l_0}(\ff;\ubar\upbeta^* T^*_{(0,0)}\R^4)$ with $N(\delta^*,z+1)\omega_k=h_k$. Via~\eqref{EqRicUniqErr}, we can then eliminate the term $r^z(\log r)^k\pi_\ff^*h_k$ (at the expense of causing changes to lower order terms). Iterating until $k=0$ finishes the construction of $\omega$.

  Finally, the continuous dependence of $\omega_j$ on $(h_0,\ldots,h_k)$ is guaranteed for the projections to scalar and vector type $l\geq l_0$ tensors, whereas for the finitely many remaining pure types the finite-dimensionality of all function spaces involved implies, by linear algebra, that one can choose the pure type $\rms l$ or $\rmv l$, $l<l_0$, parts of $\omega_j$ to depend linearly (and thus automatically continuously) on the corresponding pure type parts of $(h_0,\ldots,h_k)$.
\end{proof}

\begin{thm}[Solvability with quasi-homogeneous forcing]
\label{ThmRicSolv}
  Let $z\in\C$, $k\in\N_0$. Consider the stationary tensor\footnote{This means that the components of $f(r,\omega)$ with respect to the bundle trivialization induced by the coordinates $t,x$ are homogeneous of degree $z$ with respect to dilations $(r,\omega)\mapsto(\lambda r,\omega)$ when $k=0$, and quasi-homogeneous when $k\geq 1$.}
  \[
    f(r,\omega) = \sum_{j=0}^k \frac{1}{j!} r^{z-2}(\log r)^j \pi_\ff^*f_j(\omega), \qquad f_j\in\CI(\ff;\ubar\upbeta^*S^2 T^*_{(0,0)}\R^4),
  \]
  on $\R^3\setminus\{0\}$. Write $f_{\bullet,j}$ for the pure type $\bullet\in\{\rms l,\,\rmv l\}$ part of $f_j$. Suppose that $\wh{\ubar\delta\ubar\sfG}(0)f=0$. Fix a cutoff function $\chi\in\CIc([0,\infty)_r)$ which equals $1$ near $r=0$. Consider the following possibilities.
  \begin{enumerate}
  \item $z\notin\Z$: set $k'=0$.
  \item\label{ItRicSolvl} $z\in\{-l-1,l\}$ where $l\geq 1$: set $k'=0$ if $f_{\rms l,k}=0$ and $f_{\rmv l,k}=0$, and $k'=1$ otherwise.
  \item\label{ItRicSolvm1} $z=-1$: writing $\la-,-\ra$ for $L^2$-pairings on $\R^3$ with volume density $r^2|\dd r\,\dd\slg|$ and fiber inner products induced by the Minkowski metric, assume that\footnote{An example is $\vect=\pa_\phi^\flat=\sin^2\theta\,\dd\phi$, so $r^2\vect=x^1\,\dd x^2-x^2\,\dd x^1$.}
    \begin{equation}
    \label{EqRicSolvv1Pair}
      \big\la [\ubar\delta\ubar\sfG,\chi]f,r^2\vect\big\ra=0 \quad \forall\,\vect\in\vectspace_1,
    \end{equation}
    Set $k'=0$ if $f_{\rms 0,k}=0$, and $k'=1$ otherwise.
  \item\label{ItRicSolv0} $z=0$: assume that
    \begin{align}
    \label{EqRicSolvs0Pair}
      \big\la [\ubar\delta\ubar\sfG,\chi]f,\dd t\big\ra&=0, \\
    \label{EqRicSolvs1Pair}
      \big\la [\ubar\delta\ubar\sfG,\chi]f,\dd(r\scal)\big\ra&=0 \quad \forall\,\scal\in\scalspace_1.
    \end{align}
    Set $k'=0$ if $f_{\rms 0,k}=0$, and $k'=1$ otherwise.
  \end{enumerate}
  Then there exist $h_j\in\CI(\ff;\ubar\upbeta^*S^2 T^*_{(0,0)}\R^4)$, $j=0,\ldots,k+k'$, so that, in $r>0$,
  \begin{equation}
  \label{EqRicSolv}
    \wh{D_{\ubar g}\Ric}(0)h=f,\qquad
    h(r,\omega)=\sum_{j=0}^{k+k'}\frac{1}{j!}r^z(\log r)^j\pi_\ff^*h_j(\omega),
  \end{equation}
  and so that $(h_0,\ldots,h_{k+k'})$ depends linearly and continuously on $(f_0,\ldots,f_k)$. When $k'=1$, then $h_{k+1}\in\ker N(D\Ric,z)$.
\end{thm}

The conditions~\eqref{EqRicSolvv1Pair}, \eqref{EqRicSolvs0Pair}, and \eqref{EqRicSolvs1Pair} only depend on $f_{\rmv 1}$, $f_{\rms 0}$, and $f_{\rms 1}$, respectively. Moreover, the pairings in~(7.36)--(7.38) are independent of the choice of $\chi$; indeed, the difference of any two such cutoffs is a function $\psi\in\CIc((0,\infty))$, and thus there are no boundary terms in the integration by parts computation $\la[\ubar\delta\ubar\sfG,\psi]f,\dd t\ra=\la\ubar\sfG \psi f,\ubar\delta^*\dd t\ra-\la \psi\ubar\delta\ubar\sfG f,\dd t\ra=0-0=0$ (similarly for the other two pairings).

\begin{proof}[Proof of Proposition~\usref{ThmRicSolv}]
  Let us write
  \[
    \tilde f(\lambda,\omega)=\sum_{j=0}^k(\lambda-z)^{-j-1}f_j(\omega),\qquad
    \tilde h(\lambda,\omega)=\sum_{j=0}^{k+k'}(\lambda-z)^{-j-1}h_j(\omega),
  \]
  so that $f(r,\omega)=\Res_{\lambda=z}(r^{\lambda-2}\tilde f(\lambda,\omega))$, while $h(r,\omega)=\Res_{\lambda=z}(r^\lambda\tilde h(\lambda,\omega))$ is the solution of~\eqref{EqRicSolv} which we seek. In terms of $\tilde f,\tilde h$, equation~\eqref{EqRicSolv} is equivalent to
  \begin{equation}
  \label{EqRicSolvInd}
    N(D\Ric,\lambda)\tilde h(\lambda) = \tilde f(\lambda) + {\rm hol.},
  \end{equation}
  where we write ${\rm hol.}$ for a $\lambda$-dependent tensor which is holomorphic at $\lambda=z$. This follows as in the proof of Lemma~\ref{LemmaBgBdyPair}. Indeed, multiplying~\eqref{EqRicSolvInd} by $r^{\lambda-2}$, integrating along a small circle around $\lambda=z$, and recalling our shorthand notation $N(D\Ric,\lambda)=N(r^2\wh{D_{\ubar g}\Ric}(0),\lambda)$ gives~\eqref{EqRicSolv}; conversely, we can pull the action of $r^2\wh{D_{\ubar g}\Ric}(0)$ on $h(r,\omega)=\frac{1}{2\pi i}\oint_z r^\lambda\tilde h(\lambda,\omega)\,\dd\lambda$ under the integral sign where it acts on $\tilde h$ via $N(D\Ric,\lambda)$, and~\eqref{EqRicSolvInd} follows.

  Similarly, we shall use that $\wh{\ubar\delta\ubar\sfG}(0)f=0$ is equivalent to
  \begin{equation}
  \label{EqRicSolvNdelG}
    N(\delta\sfG,\lambda-2)\tilde f(\lambda) = {\rm hol.}
  \end{equation}
  The conditions~\eqref{EqRicSolvv1Pair}, \eqref{EqRicSolvs0Pair}, and \eqref{EqRicSolvs1Pair} are moreover equivalent to
  \begin{align}
  \label{EqRicSolvv1Pair2}
      \big\la N(r\wh{\ubar\delta\ubar\sfG}(0),\lambda-2)\tilde f(\lambda),r\vect\big\ra_{L^2(\Sph^2)}\big|_{\lambda=-1}&=0 \quad \forall\,\vect\in\vectspace_1, \\
  \label{EqRicSolvs0Pair2}
      \big\la N(r\wh{\ubar\delta\ubar\sfG}(0),\lambda-2)\tilde f(\lambda),\dd t\ra_{L^2(\Sph^2)}\big|_{\lambda=0}&=0, \\
  \label{EqRicSolvs1Pair2}
      \big\la N(r\wh{\ubar\delta\ubar\sfG}(0),\lambda-2)\tilde f(\lambda),\dd(r\scal)\ra_{L^2(\Sph^2)}\big|_{\lambda=0}&=0 \quad \forall\,\scal\in\scalspace_1,
  \end{align}
  respectively; this is a special case of Lemma~\ref{LemmaBgBdyPair} (for $X=[\R^3;\{0\}]$, $w=-3$, $\alpha=0$, $L=r\wh{\ubar\delta\ubar\sfG}(0)$, and $z=-3,-2,-2$).

  \pfstep{Step 1. Individual solvability for pure types.} Fix a pure type $\bullet$, and replace $\tilde f$ by its pure type $\bullet$ part $\tilde f_\bullet$.

  \pfsubstep{(1.i)}{$z\notin\Z$.} Proposition~\ref{PropRicR} gives $K_{\bullet,[0]}(\delta\sfG,z-2)=R_{\bullet,[0]}(D\Ric,z)$. If $k=0$ and thus $f_0\in\ker N_\bullet(\delta\sfG,z-2)$, this produces $h_0$ so that~\eqref{EqRicSolvInd} holds. For general $k\in\N$, we have $f_k\in\ker N_\bullet(\delta\sfG,z-2)$; choosing $h_k$ with $N_\bullet(D\Ric,z)h_k=f_k$, we get
  \begin{equation}
  \label{EqRicSolvErr}
    N_\bullet(D\Ric,\lambda)\bigl((\lambda-z)^{-k-1}h_k\bigr) = \tilde f(\lambda) + \tilde e(\lambda),\qquad
    \tilde e(\lambda)=\sum_{j=0}^{k-1}(\lambda-z)^{-j-1}e_j,
  \end{equation}
  where thus $\tilde e(\lambda)$ is more regular by one power of $\lambda-z$; and since $N_\bullet(\delta\sfG,\lambda-2)$ annihilates the left hand side identically, we find using~\eqref{EqRicSolvNdelG} that
  \begin{equation}
  \label{EqRicSolvErr2}
    N_\bullet(\delta\sfG,\lambda-2)\tilde e(\lambda)\ \text{is holomorphic at}\ \lambda=z.
  \end{equation}
  Having thus reduced $k$ by $1$, the solvability of~\eqref{EqRicSolvInd} follows by induction on $k$.

  \pfsubstep{(1.ii)}{$z\in\{-l-1,l\}$, $l\geq 1$.} If $\bullet\neq\rms l,\rmv l$ (or even if $\bullet\in\{\rms l,\,\rmv l\}$ but $\tilde f_{\rms l}=0$ and $\tilde f_{\rmv l}=0$), then the arguments from step \textbf{(1.i)} apply without change. If $\bullet\in\{\rms l,\,\rmv l\}$ and $\tilde f_\bullet\neq 0$, then since $f_k\in\ker N_\bullet(\delta\sfG,z-2)$, Proposition~\ref{PropRicR} (part~\eqref{ItRicRs1}, resp.\ \eqref{ItRicRv1} for $l=1$, and part~\eqref{ItRicRsl}, resp.\ \eqref{ItRicRvl} for $l\geq 2$) produces $h_{k+1}\in\ker N_\bullet(D\Ric,z)$ and $h_k$ so that $f_k=\pa_\lambda N_\bullet(D\Ric,z)h_{k+1}+N_\bullet(D\Ric,z)h_k$; in other words,
  \[
    N_\bullet(D\Ric,\lambda)\bigl((\lambda-z)^{-k-2}h_{k+1} + (\lambda-z)^{-k-1}h_k \bigr) = \tilde f(\lambda) + \tilde e(\lambda)
  \]
  where, as in part \textbf{(1.i)}, $\tilde e$ has one power of $(\lambda-z)^{-1}$ less than $\tilde f$, and $N(\delta\sfG,\lambda-2)\tilde e(\lambda)$ is holomorphic at $\lambda=z$. An inductive argument finishes the proof also in this case.

  \pfsubstep{(1.iii)}{$z=-1$.} If $\bullet\neq\rms 0,\rmv 1$, the arguments from step \textbf{(1.i)} apply. For $\bullet=\rms 0$, the arguments from step \textbf{(1.ii)} apply in view of Proposition~\ref{PropRicR}\eqref{ItRicRs0}.

  For $\bullet=\rmv 1$ finally, consider first the case $k=0$. Since $f_0\in\ker N_{\rmv 1}(\delta\sfG,-3)$, we obtain $(N_{\rmv 1}(\delta\sfG,\lambda-2)\tilde f(\lambda))|_{\lambda=-1}=\pa_\lambda N_{\rmv 1}(\delta\sfG,-3)f_0$, and thus~\eqref{EqRicSolvv1Pair2} reads $\la\pa_\lambda N(\delta\sfG,-3)f_0,r\vect\ra=0$, $\vect\in\vectspace_1$. In view of the isomorphism~\eqref{EqRicRv1Iso}, this implies that we can write $f_0=N(D\Ric,-1)h_0$.

  For $k\geq 1$, the argument is different: we now have $f_k\in K_{\rmv 1,[k]}(\delta\sfG,-3)\subset K_{\rmv 1,[1]}(\delta\sfG,-3)$ and thus $f_k=N_{\rmv 1}(D\Ric,-1)h_k$ for some $h_k$ by Lemma~\ref{LemmaRicRRestr}. Thus, we again have~\eqref{EqRicSolvErr}--\eqref{EqRicSolvErr2} (with $z=-1$). Moreover, applying $N_{\rmv 1}(\delta\sfG,\lambda-2)$ to~\eqref{EqRicSolvErr} and taking the inner product with $r\vect$ gives
  \[
    0 = \big\la N_{\rmv 1}(\delta\sfG,\lambda-2)\tilde f(\lambda),r\vect\big\ra + \big\la N_{\rmv 1}(\delta\sfG,\lambda-2)\tilde e(\lambda),r\vect\big\ra,
  \]
  and thus
  \begin{equation}
  \label{EqRicSolvErrv1}
    \big\la N_{\rmv 1}(\delta\sfG,\lambda-2)\tilde e(\lambda),r\vect\big\ra\big|_{\lambda=-1} = 0.
  \end{equation}
  Therefore, $\tilde e$ has one power of $(\lambda+1)^{-1}$ less than $\tilde f$ while satisfying the same assumptions, and an inductive argument finishes the proof in this case. (If $k-1=0$, one applies the first part of the argument, whereas for $k-1\geq 1$ one repeats the second part.)

  \pfsubstep{(1.iv)}{$z=0$.} Since the scalar type $1$ and vector type $1$ cases in Proposition~\ref{PropRicR} (specifically \eqref{EqRicRs1Iso} and \eqref{EqRicRv1Iso}) and Lemma~\ref{LemmaRicRRestr} are completely analogous, up to replacing $z=0$ by $z=-1$, as far as the orders of (restricted) kernels and (generalized) ranges are concerned, the scalar type $1$ case follows by the same arguments as the vector type $1$ case.

  The treatment of the scalar type $0$ case is a combination of steps \textbf{(1.ii)} and \textbf{(1.iii)}. In the case $k=0$, the tensor $f_0\in\ker N(\delta\sfG,-2)$ lies in the kernel of the map~\eqref{EqRicRs0Iso} and thus can be written as $f_0=\pa_\lambda N_{\rms 0}(D\Ric,0)h_1+N_{\rms 0}(D\Ric,0)h_0$ where $h_1\in\ker N_{\rms 0}(D\Ric,0)$, which gives~\eqref{EqRicSolvInd} for $\tilde h(\lambda)=\lambda^{-2}h_1+\lambda^{-1}h_0$. On the other hand, if $k\geq 1$, then $f_k\in K_{\rms 0,[1]}(\delta\sfG,-2)=R_{\rms 0,[1]}(D\Ric,0)$ by Lemma~\ref{LemmaRicRRestr}, and thus we obtain $h_{k+1}\in\ker N_{\rms 0}(D\Ric,0)$ and $h_k$ so that
  \[
    N_{\rms 0}(D\Ric,\lambda)\bigl(\lambda^{-k-2}h_{k+1}+\lambda^{-k-1}h_k\bigr) = \tilde f(\lambda) + \tilde e(\lambda),
  \]
  where $\tilde e$ is as in~\eqref{EqRicSolvErr}--\eqref{EqRicSolvErr2}; and the arguments leading to~\eqref{EqRicSolvErrv1} apply \textit{mutatis mutandis} to give
  \[
    \big\la N_{\rms 0}(\delta\sfG,\lambda-2)\tilde e(\lambda),\dd t\big\ra\big|_{\lambda=0} = 0.
  \]
  By induction, we can solve away $\tilde e(\lambda)$ to finish the proof.

  \pfstep{Step 2. Simultaneous solvability for all pure types with large $l$.} With $z$ fixed, there exists $l_0\in\N_0$ so that Corollary~\ref{CorRicLarge} applies (with $\lambda=z$). Thus, for the projection $f_{\geq l_0,k}$ of $f_k$ off the space of scalar and vector type $l<l_0$ tensors, we have $N(\delta\sfG,z-2)f_{\geq l_0,k}=0$, whence there exists a solution $h_k\in\CI_{\geq l_0}(\ff;\ubar\upbeta^*S^2 T^*_{(0,0)}\R^4)$ of $N(D\Ric,z)h_k=f_{\geq l_0,k}$. Repeating the argument from step \textbf{(1.i)} then reduces the task to one where $k$ is reduced by $1$. Induction finishes the construction of $h$.

  The linear continuous dependence of $(h_0,\ldots,h_{k+k'})$ on $(f_0,\ldots,f_k)$ can be arranged by the same argument as in the proof of Theorem~\ref{ThmRicUniq}.
\end{proof}

\begin{rmk}[Necessity of solvability conditions]
\label{RmkAcBdy0Opm2}
  Proposition~\ref{PropRicR} and the comments following it imply that the conditions~\eqref{EqRicSolvv1Pair}--\eqref{EqRicSolvs1Pair} are necessary for the solvability of equation~\eqref{EqRicSolv} for $z=-1,0$ regardless of the value of $k'\in\N_0$. Tensors $f$ violating condition~\eqref{EqRicSolvs1Pair} (thus $z=0$) arise rather directly later on. An explicit example is the scalar type $1$ tensor $f=r^{-2}f_0$, $f_0=(-2,-\half,\half,-2,-\half,4)^T$ (in terms of~\eqref{EqMkYSplits1} with $0\neq\scal\in\scalspace_1$ fixed), for which one can check by direct computation that $f_0\in\ker N(\delta\sfG,-2)$ whereas the pairing~\eqref{EqRicSolvs1Pair} (for the same choice of $\scal$) evaluates to $24\pi\|\scal\|_{L^2(\Sph^2)}^2\neq 0$. See~\S\ref{SsAcGW} for the origin of this example.
\end{rmk}

\section{Linear analysis on \texorpdfstring{$M_\circ$}{the blown-up background spacetime}}
\label{SAc}

We use the setup and notation of~\S\ref{SM}. In this section, we solve the linearized Einstein vacuum equations with nontrivial right hand side $f\in\ker\delta_g\sfG_g$ on $M_\circ$. We shall only study the case that $f=\cO(|x|^{-2+\delta})$ for some $\delta>0$; in practice, we will in fact only encounter log-smooth $f$ (which thus have leading order behavior $|x|^{-1}(\log|x|)^k$). Importantly, Theorem~\ref{ThmRicSolv} is applicable to each term in the polyhomogeneous expansion of such $f$ at $x=0$ since for exponents $z\in\C$ with $z>0$, there are no further necessary conditions for solvability (since~\eqref{EqRicSolvv1Pair}--\eqref{EqRicSolvs1Pair} only enter for $z=0,-1$).

We only need to use here that $(M,g)$ is globally hyperbolic and satisfies $\Ric(g)-\Lambda g=0$, and that the initial data set $(X,\gamma,k)$ has no KIDs in the precompact connected smoothly bounded open neighborhood $\cU^\circ\subset X$ of $\fp\in X\cap\cC$. The curve $\cC=c(I)\subset M$ can be any smooth inextendible timelike curve, and we recall the blow-down map $\upbeta_\circ\colon M_\circ=[M;\cC]\to M$.

\begin{thm}[Solvability of the linearized Einstein vacuum equations at $M_\circ$]
\label{ThmAc}
  Let $\hat\cF\subset\C\times\N_0$ be an index set with $\Re\hat\cF>-2$. Set $\hat\cE=\{(z+j+2,l)\colon (z,k)\in\hat\cF,\ j\in\N_0,\ l\leq k+j+1\}$.\footnote{This index set is not sharp for those $(z,k)\in\hat\cF$ with $z\notin\Z$. Since in this paper all exponents $z$ will be integers, we content ourselves with the possibly oversized index set $\hat\cE$ here.} If $f\in\cA_\phg^{\hat\cF}(M_\circ;\upbeta_\circ^*S^2 T^*M)$, with $\supp f$ contained in the domain of influence of a compact subset of $\cU^\circ$, satisfies $\delta_g\sfG_g f=0$, then there exists
  \begin{equation}
  \label{EqAch}
    h = h_\sharp + \upbeta_\circ^*h_\flat,\qquad h_\sharp\in\cA_\phg^{\hat\cE}(M_\circ;\upbeta_\circ^*S^2 T^*M_\circ),\quad h_\flat\in\CI(M;S^2 T^*M),
  \end{equation}
  with the following properties:
  \begin{enumerate}
  \item\label{ItAcSolv} on $(M_\circ)^\circ$, we have
    \begin{equation}
    \label{EqAc}
      (D_g\Ric-\Lambda)h=f;
    \end{equation}
  \item\label{ItAcSupp} $h$ vanishes near $X\setminus\cU^\circ$, or equivalently $\supp h\cap\upbeta_\circ^*X\Subset\upbeta_\circ^*\cU^\circ$, and in fact $\supp h$ is contained in the domain of influence of a compact subset of $\cU^\circ$;
  \item\label{ItAchflat} if $\cC$ is a geodesic: $h_\flat$ vanishes quadratically at $\cC$, i.e.\ its coefficients in smooth coordinates on $M$ near $\cC$ vanish quadratically at $\cC$.
  \end{enumerate}
\end{thm}

\begin{rmk}[Weight at $M_\circ$]
\label{RmkAcWeight}
  Note that $f,h$ are locally integrable at $\cC$ with respect to the lift of a smooth positive density on $M$ to $M_\circ$ (such as $|\dd g|$), and thus they can be extended uniquely from $M\setminus\cC=(M_\circ)^\circ$ to $L^1_\loc$-distributions $E f,E h$ on $M$. The difference $(D_g\Ric-\Lambda)(E h)-E f$, which is supported at $\cC$, must vanish identically by homogeneity considerations, since both summands are $\lesssim|x|^{-2+\delta}$ near $x=0$ where $\delta\in(0,\min\Re\hat\cF+2)$. Since we always have the distributional equality $\delta_g\sfG_g(D_g\Ric-\Lambda)(E h)=0$ by the linearized second Bianchi identity, a necessary condition for the solvability of~\eqref{EqAc} on $M\setminus\cC$ (with $|h|\lesssim|x|^{-1+\delta}$ for some $\delta>0$) is $\delta_g\sfG_g E f=0$ (in the distributional sense on $M$); and this equation indeed holds not only in $(M_\circ)^\circ$ by assumption, but indeed globally by homogeneity considerations since $\delta_g\sfG_g E f$ is polyhomogeneous with degrees $>-3$ (which excludes $\delta$-distributions at $\cC$). If one drops the assumption $\Re\hat\cF>-2$, then one may have $\delta_g\sfG_g E f\neq 0$ even though $\delta_g\sfG_g f=0$. See also Remark~\ref{RmkAcBdy0Opm2} and~\S\ref{SsAcGW} below.
\end{rmk}

The proof of Theorem~\ref{ThmAc} will be given in~\S\S\ref{SsAcBdy}--\ref{SsAcTrue}. In~\S\ref{SsAcBdy}, we first find a formal solution $h_\sharp$ at $\pa M_\circ$, i.e.\ $h_\sharp$ satisfies~\eqref{EqAc} to infinite order at $\pa M_\circ$. We then correct the formal solution to a true solution by solving an initial value problem (with carefully chosen initial data) on the blow-down $M$ of $M_\circ$ in~\S\ref{SsAcTrue}. Constraints which our sharp solvability theory imposes on the conditions under which our gluing construction can succeed at all are discussed in~\S\ref{SsAcGW}.

\subsection{Formal solution near \texorpdfstring{$\pa M_\circ$}{the blown-up geodesic}}
\label{SsAcBdy}

For this part, we only need to assume that $\Ric(g)-\Lambda g$ vanishes to infinite order at $\cC$. Recall the Fermi normal coordinates $(t,x)$ near $\cC\subset M$ and parameterize $\cC$ by $c\colon t\mapsto(t,0)$. Write
\[
  \ubar g=-\dd t^2+\dd x^2
\]
for the Minkowski metric in these coordinates. (Thus, $g-\ubar g$ vanishes at $x=0$ by Lemma~\ref{LemmaGLFermi}.) Introduce $r=|x|$ and $\omega=\frac{x}{|x|}\in\Sph^2$ in order to pass to $M_\circ=[M;\cC]$. The coordinates $(t,r,\omega)$ are valid in a collar neighborhood
\begin{subequations}
\begin{equation}
\label{EqAcBdyN}
  \cN=\{(t,r,\omega)\colon t\in I,\ 0\leq r<r_0(t)\}
\end{equation}
of $\pa M_\circ$, where we recall that $r_0\in\CI(I)$ satisfies $0<r_0<\half$. Let
\begin{equation}
\label{EqAcBdyNCutoff}
  \hat\chi\in\CIc(\cN),\qquad \hat\chi=1\ \text{near}\ \pa\cN=r^{-1}(0)
\end{equation}
\end{subequations}
be a cutoff function. We identify $\cN$ with a collar neighborhood of the zero section in ${}^+N\cC$.

Since $D_g\Ric-\Lambda\in\Diff^2(M;S^2 T^*M)$ by~\eqref{EqELinDRic}, Lemma~\ref{LemmaBgE} gives $\upbeta_\circ^*(D_g\Ric-\Lambda)\in r^{-2}\Diffe^2(M_\circ;\upbeta_\circ^*S^2 T^*M)$, and Lemma~\ref{LemmaBgEb} allows us to compute the b-normal operator. To wit, the principal symbol of $D_g\Ric-\Lambda$ over a point $p\in\cC$ only depends on $g(p)$, which in our local coordinates is equal to the Minkowski metric $\ubar g$. Thus, the b-normal operator of $\upbeta_\circ^*r^2(D_g\Ric-\Lambda)$, restricted to a level set $t^{-1}(t_0)$, $t_0\in I$, is the same as that of $\upbeta_\circ^*r^2 D_{\ubar g}\Ric$. The latter is a vertical operator which, restricted to a fiber $t^{-1}(t_0)$, $t_0\in\R$, of ${}^+N\cC$, acts on sections of the pullback of $S^2 T^*_{c(t_0)}M\cong S^2 T^*_{(t_0,0)}(\R\times\R^3)$ along the projection
\[
  \pi\colon{}^+N\cC\to\cC.
\]
It is the zero energy operator of $r^2 D_{\ubar g}\Ric$, i.e.\ obtained from $r^2 D_{\ubar g}\Ric$, regarded as a differential operator on $\R_t\times\R^3_x$, by dropping all $\pa_t$ (or more precisely $r\pa_t$) derivatives. It is thus independent of $t_0\in I$. We denote this zero energy operator by
\begin{equation}
\label{EqAcBdy0EnOp}
  \wh{D_{\ubar g}\Ric}(0) \in \Diff^2(\R^3;S^2 T^*_{(t_0,0)}\R^4);
\end{equation}
it is the unique operator with $D_{\ubar g}\Ric(h)|_{t^{-1}(t_0)}=\wh{D_{\ubar g}\Ric}(0)(h|_{t^{-1}(t_0)})$ for all $t$-invariant $h$. Since $\ubar g$ is homogeneous of degree $2$ with respect to dilations $(t_0+t',x)\mapsto(t_0+\lambda t',\lambda x)$, the operator $D_{\ubar g}\Ric$ is homogeneous of degree $-2$, and thus so is $\wh{D_{\ubar g}\Ric}(0)$ with respect to dilations $x\mapsto\lambda x$. Therefore,
\begin{equation}
\label{EqAcBdyNormOp}
  \upbeta_\circ^*(r^2\wh{D_{\ubar g}\Ric}(0)) \in {}^\vee\Diff_{\bop,\rm I}^2({}^+N\cC;\pi^*S^2 T^*_\cC M).
\end{equation}
(This also follows from the explicit expression in Corollary~\ref{CorMkDiffDRic}.) In summary:

\begin{lemma}[b-normal operator at $\pa M_\circ$]
\label{LemmaAcBdyNormOp}
  We have
  \[
    \upbeta_\circ^*(D_g\Ric-\Lambda) \in r^{-2}\Diffe^2(M_\circ;\upbeta_\circ^*S^2 T^*M) \subset r^{-2}\Diffb^2(M_\circ;\upbeta_\circ^*S^2 T^*M).
  \]
  Moreover, if we identify $S^2 T^* M\cong \pi^*S^2 T^*_\cC M$ over $\cN$, then
  \begin{equation}
  \label{EqAcBdyNormOpDiff}
    \upbeta_\circ^*(D_g\Ric-\Lambda) - \hat\chi \upbeta_\circ^*\wh{D_{\ubar g}\Ric}(0) \hat\chi \in r^{-1}\Diffb^2(M_\circ;\upbeta_\circ^*S^2 T^*M).
  \end{equation}
\end{lemma}

Identifying the fibers of ${}^+N\cC$ over different points on $\cC$ by means of their trivializations given by Fermi normal coordinates, and similarly for the fibers of $T^*_\cC M$, the restriction of the operator~\eqref{EqAcBdyNormOp} to ${}^+N_{c(t_0)}\cC$ is independent of $t_0\in I$. We thus mainly need to study
\begin{equation}
\label{EqAcBdy0OpMink}
  \ubar\upbeta^*\bigl(r^2\wh{D_{\ubar g}\Ric}(0)\bigr) \in \Diff_{\bop,\rm I}^2\bigl([\R^3;\{0\}];\pi_\ff^*\ubar\upbeta^*S^2 T^*_{(0,0)}\R^4\bigr),
\end{equation}
where $\ubar\upbeta\colon[\R^3;\{0\}]=[0,\infty)_r\times\Sph^2_\omega\to\R^3$ is the blow-down map, $\pi_\ff(r,\omega)=\omega$ is as in Theorem~\ref{ThmRicSolv}, and where, as before, we regard $\R^3=t^{-1}(0)\subset\R^4$. The solvability theory for this operator, for quasi-homogeneous right hand sides, was analyzed in Theorem~\ref{ThmRicSolv}.

\begin{lemma}[Formal solution at $\cC$ with restricted polyhomogeneous right hand side]
\label{LemmaAcBdyFormal}
  Let $z\in\C\setminus(-2-\N_0)$, and let $\hat\cF\subset(z+\N_0)\times\N_0$ be an index set. Let $\hat\cE=\{(w+j+2,l) \colon (w,k)\in\hat\cF,\ j\in\N_0,\ l\leq k+j+1\}$. Let
  \[
    f \in \cA_\phg^{\hat\cF}(M_\circ;\upbeta_\circ^*S^2 T^*M),\qquad \delta_g\sfG_g f\in\CIdot(M_\circ;\upbeta_\circ^*T^*M),
  \]
  i.e.\ $\delta_g\sfG_g f$ vanishes to infinite order at $\pa M_\circ$. Then there exists $h\in\cA_\phg^{\hat\cE}(M_\circ;\upbeta_\circ^*S^2 T^*M)$, with support contained in any fixed neighborhood of $\pa M_\circ$, so that
  \begin{equation}
  \label{EqAcBdyFormal}
    (D_g\Ric-\Lambda)(h)=f + f_\flat,\qquad f_\flat\in\CIdot(M_\circ;\upbeta_\circ^*S^2 T^*M).
  \end{equation}
\end{lemma}

The equations $\delta_g\sfG_g f=0$ and $D_g\Ric(h)=f$ here are to be understood as equalities in $M\setminus\cC$ (or more precisely as equalities of extendible distributions on $M_\circ$).

\begin{proof}[Proof of Lemma~\usref{LemmaAcBdyFormal}]
  We identify the collar neighborhood $\cN$ using the polar coordinates $(t,r,\omega)$ associated with Fermi normal coordinates $(t,x)$ with a neighborhood of $I\times\{0\}\times\Sph^2\subset I\times[\R^3;\{0\}]=I\times[0,\infty)\times\Sph^2$. Let $\hat\chi\in\CIc(\cN)$ be a cutoff as in~\eqref{EqAcBdyN}--\eqref{EqAcBdyNCutoff}. We use the notation introduced in~\eqref{EqAcBdy0OpMink}. If $k\in\N_0$ is such that $(z,k)\in\hat\cF$ but $(z,k+1)\notin\hat\cF$, there exist $f_0,\ldots,f_k\in\CI(I;\CI(\Sph^2;\ubar\upbeta^*S^2 T^*_{(0,0)}\R^4))$ so that
  \[
    f - \hat\chi f^{(0)} \in \cA_\phg^{\hat\cF_0}(M_\circ;\upbeta_\circ^*S^2 T^*M),\qquad
    f^{(0)}:=\sum_{j=0}^k r^z(\log r)^j\pi_\ff^*f_j,
  \]
  where $\hat\cF_0=\hat\cF\setminus\{(z,l)\in\hat\cF\}=\{(w,l)\in\hat\cF\colon w\in z+1+\N_0\}\subset(z+1+\N_0)\times\N_0$. We now apply the operator $\upbeta_\circ^*(\delta_g\sfG_g)$ to $f-\hat\chi f^{(0)}\in\cA_\phg^{\hat\cF_0}$ and use Lemma~\ref{LemmaBgEb}; this implies that
  \[
    \hat\chi \wh{\ubar\delta\ubar G}(0)f^{(0)} \in \cA_\phg^{\hat\cF_0-1}(M_\circ;\upbeta_\circ^*T^*M). 
  \]
  Since $\wh{\ubar\delta\ubar\sfG}(0)$ is homogeneous of degree $-1$, this implies $\wh{\ubar\delta\ubar\sfG}(0)f^{(0)}=0$. Therefore, we are in the setting of Theorem~\ref{ThmRicSolv} with smooth parametric dependence on $t\in I$. This gives $k'\in\{0,1\}$ and $h_j(t)\in\CI(\Sph^2;\ubar\upbeta^*S^2 T^*_{(0,0)}\R^4)$, $j=0,\ldots,k+k'$, for each $t\in I$ so that
  \[
    \wh{D_{\ubar g}\Ric}(0)(h^{(0)}(t))=f^{(0)}(t),\qquad h^{(0)}(t)=\sum_{j=0}^{k+k'} r^{z+2}(\log r)^j\pi_\ff^*h_j(t).
  \]
  In view of the continuous and linear dependence of $(h_0(t),\ldots,h_{k+k'}(t))$ on $(f_0(t),\ldots,f_k(t))$, we first conclude that the $h_j(t)$ are continuous in $t$. Moreover, the solution of
  \begin{equation}
  \label{EqAcBdyFormalC1}
    \wh{D_{\ubar g}\Ric}(0)(h^{(1),\eps}(t))=\eps^{-1}\bigl(f^{(0)}(t+\eps)-f^{(0)}(t)\bigr)
  \end{equation}
  is on the one hand given by $\eps^{-1}(h^{(0)}(t+\eps)-h^{(0)}(t))$, and on the other hand converges as $\eps\searrow 0$ to the solution (as produced by Theorem~\ref{ThmRicSolv}) of $\wh{D_{\ubar g}\Ric}(0)(h^{(1)}(t))=\pa_t f^{(0)}(t)$. Since $h^{(1)}(t)$ is continuous in $t$, this shows that $h^{(0)}$ is $\cC^1$ in $t$. Iterating this argument implies $h_j\in\CI(I;\CI(\Sph^2;\ubar\upbeta^*S^2 T^*_{(0,0)}\R^4))$, $j=0,\ldots,k+k'$.

  By Lemma~\ref{LemmaAcBdyNormOp} and recalling the notation for index sets from~\S\ref{SsBgPhg}, we conclude that
  \[
    f' := f - (D_g\Ric-\Lambda)(\hat\chi h^{(0)}) \in \cA_\phg^{\hat\cF_0\cup(z+1,k+k')}(M_\circ;\upbeta_\circ^*T^*M).
  \]
  (Note that the cosmological constant produces terms which are of lower order even compared to the terms arising from~\eqref{EqAcBdyNormOpDiff}.) That is, we have solved away $f$ to leading order at $\pa M_\circ$, at the expense of (possibly) an additional logarithmic factor of the error $f'$. Note furthermore that the linearized second Bianchi identity ensures that
  \[
    \delta_g\sfG_g f' = \delta_g\sfG_g f - \delta_g\sfG_g(D_g\Ric-\Lambda)(\hat\chi h^{(0)}) \in \CIdot(M_\circ;\upbeta_\circ^*T^*M)
  \]
  still. (We use here the validity of the nonlinear Einstein vacuum equations $\Ric(g)-\Lambda g=0$ in Taylor series at $\cC$.)

  Since $z+1\in\C\setminus(-\N_0-2)$ still, we may repeat this procedure. Proceeding inductively, we thus obtain a sequence $\hat\cE_\ell\subset((z+2)+\ell+\N_0)\times\N_0$, $\ell\in\N_0$, of index sets (with $\hat\cE_0=(z+2,k+k')$ in the above notation) and $h^{(\ell)}\in\cA_\phg^{\hat\cE_\ell}(M_\circ;\upbeta_\circ^*S^2 T^*M)$ so that
  \[
    f - (D_g\Ric-\Lambda)\left(\hat\chi\sum_{\ell=0}^m h^{(\ell)}\right) \in \cA_\phg^{\hat\cF_m}(M_\circ;\upbeta_\circ^*S^2 T^*M),
  \]
  where $\hat\cF_m$ is an index set with $\hat\cF_m\subset(z+(m+1)+\N_0)\times\N_0$ (and with $\hat\cF_0$ defined previously). More precisely, the construction gives
  \begin{align*}
    \hat\cE_{m+1}&\subset\{((z+2)+(m+1)+j,l)\colon j\in\N_0,\ l\leq\max\{k+1\colon (z+(m+1),k)\in\hat\cF_{m+1}\}\}, \\
    \hat\cF_{m+1}&\subset\{(z+(m+2)+j,l) \colon j\in\N_0,\ l\leq\max\{k\colon((z+2)+(m+1),k)\in\hat\cE_{m+1}\}\} \\
      &\qquad \cup \{ (z+j,l)\in\hat\cF_m\colon j\geq m+2 \}.
  \end{align*}
  Thus, $\hat\cF_{m+1}$ removes the leading order terms $(z+(m+1),k)$ of $\hat\cF_m$ but adds the contribution from the leading order terms $((z+2)+(m+1),k+k')$ of $\hat\cE_{m+1}$ (with $k'\leq 1$ here replaced by $1$ simply, which may be lossy but acceptably so), while $\hat\cE_{m+1}$ picks up the leading order terms of $\hat\cF_{m+1}$ and (when $k'=1$) adds a logarithm. Setting $\hat\cE=\bigcup_{\ell\in\N_0}\hat\cE_\ell$ (which has the stated description), we may thus take $h\in\cA_\phg^{\hat\cE}$ to be an asymptotic sum of the $\hat\chi h^{(\ell)}$ over $\ell\in\N_0$. We can arrange for the desired support property of $h$ since multiplying $h$ by a smooth function which is equal to $1$ in a neighborhood of $\pa M_\circ$ preserves the conclusion~\eqref{EqAcBdyFormal}.
\end{proof}

\begin{rmk}[Control of logarithmic terms]
\label{RmkAcBdyFormalLog}
  Since the first step of the proof is an application of Theorem~\ref{ThmRicSolv}, the refined statements made there imply also refinements about $h$ and its index set: if all leading order terms of $f$, corresponding to elements $(z,k)\in\hat\cF$ with $\Re z=\min\Re\hat\cF$, are of a particular type (depending on the value of $z$), then one may take $\hat\cE$ to have the same leading order part as $\hat\cF$; more precisely, $\hat\cE=\{(w+2,k)\in\hat\cF\}\cup\{(w+2+j,l)\colon(w,k)\in\hat\cF,\ w\in z+1+\N_0,\ j\in\N_0,\ l\leq k+j+1\}$.
\end{rmk}

\begin{prop}[Formal solution at $\cC$]
\label{PropAcBdy}
  Under the assumptions of Theorem~\usref{ThmAc}, and using the notation of the Theorem, there exists $h\in\cA_\phg^{\hat\cE}(M_\circ;\upbeta_\circ^*S^2 T^*M)$, with support contained in any fixed neighborhood of $\pa M_\circ$, so that
  \begin{equation}
  \label{EqAcBdy}
    (D_g\Ric-\Lambda)(h) = f+f_\flat,\qquad f_\flat\in\CIdot(M_\circ;\upbeta_\circ^*S^2 T^*M).
  \end{equation}
\end{prop}
\begin{proof}
  Write $\hat\cF$ as the disjoint union of index sets $\hat\cF_j\subset(z_j+\N_0)\times\N_0$, $j\in J$, with the property that $z_i-z_j\notin\Z$ whenever $i\neq j$; necessarily $\min\Re\hat\cF_j\nearrow\infty$ as $j\to\infty$ (in case the set $J$ is infinite). This decomposition can be effected by taking $z_0\in\C$ so that $(z_0,0)\in\hat\cF$ and $\Re z_0=\min\Re\hat\cF$ and then defining $\hat\cF_0=\{(z,k)\in\hat\cF\colon z-z_0\in\N_0\}$; thus $\hat\cF\setminus\hat\cF_0$ is still an index set, and if it is non-empty we may repeat this process. The fact that $\hat\cF$ is an index set ensures that the complex numbers $z_0,z_1,\ldots$ one selects in this process satisfy $\Re z_j\nearrow\infty$, unless the process stops after finitely many steps.

  We can then write $f\in\cA_\phg^{\hat\cF}(M_\circ;\upbeta_\circ^*S^2 T^*M)$ as an asymptotic sum of $f_j\in\cA_\phg^{\hat\cF_j}$, $j\in J$. Lemma~\ref{LemmaBgE} gives
  \[
    \delta_g\sfG_g\colon\cA_\phg^{\hat\cF_j}(M_\circ;\upbeta_\circ^*S^2 T^*M)\to\cA_\phg^{\hat\cF_j-1}(M_\circ;\upbeta_\circ^*T^*M).
  \]
  Thus $\delta_g\sfG_g f_j$ is polyhomogeneous with index set $\hat\cF_j-1$; but $\delta_g\sfG_g f_j$ is at the same time also an asymptotic sum of polyhomogeneous distributions with index sets $\hat\cF_i$, $i\neq j$, since $\delta_g\sfG_g f_j\sim\delta_g\sfG_g f-\sum_{i\neq j}\delta_g\sfG_g f_i\sim-\sum_{i\neq j}\delta_g\sfG_g f_i$. Therefore,
  \[
    \delta_g\sfG_g f_j\in\CIdot(M_\circ;\upbeta_\circ^*T^*M),\qquad j\in J.
  \]

  We apply Lemma~\ref{LemmaAcBdyFormal} to each $f_j$, $j\in J$, separately. We obtain index sets $\hat\cE_j\subset(z_j+2+\N_0)\times\N_0$ and symmetric 2-tensors $h_j\in\cA_\phg^{\hat\cE_j}(M_\circ;\upbeta^*S^2 T^*M)$ so that $(D_g\Ric-\Lambda)h_j-f_j\in\CIdot$. Let now $\hat\cE=\bigcup_{j\in J}\hat\cE_j$, and take $h\in\cA_\phg^{\hat\cE}(M_\circ;\upbeta_\circ^*S^2 T^*M)$ to be an asymptotic sum of all $h_j$. Then~\eqref{EqAcBdy} holds. Multiplying $h$ by a smooth function which is identically $1$ near $\pa M_\circ$ and supported in the desired neighborhood of $\pa M_\circ$ furthermore ensures the desired support property of $h$.
\end{proof}

Proposition~\ref{PropAcBdy} remains valid (and Theorem~\ref{ThmAc} can be similarly extended), with the same proof, under the weaker assumption that $\hat\cF\cap((-2-\N_0)\times\N_0)=\emptyset$. We shall however only use the stated version in the solution of the gluing problem.

\subsection{True solution; proof of Theorem~\usref{ThmAc}}
\label{SsAcTrue}

With $\hat\cF$, $\hat\cE$, and $f\in\cA_\phg^{\hat\cF}(M_\circ;\upbeta_\circ^*S^2 T^*M)$ as in Theorem~\ref{ThmAc}, we now denote by $h_\sharp\in\cA_\phg^{\hat\cE}(M_\circ;\upbeta_\circ^*S^2 T^*M)$ the formal solution of~\eqref{EqAc} given by Proposition~\ref{PropAcBdy}, which we arrange to satisfy $\supp h_\sharp\cap\upbeta_\circ^{-1}(X)\Subset\upbeta_\circ^{-1}(\cU^\circ)$. Thus,
\begin{equation}
\label{EqAcTruefflat}
  (D_g\Ric-\Lambda)h_\sharp = f + f_\flat,\qquad f_\flat\in\CIdot(M_\circ;\upbeta_\circ^*S^2 T^*M).
\end{equation}
Due to the support properties of $f$ and $h_\sharp$, we also have $\supp f_\flat\cap\upbeta_\circ^{-1}(X)\Subset\upbeta_\circ^{-1}(\cU^\circ)$. Now, for any bundle $E\to M$, the space $\CIdot(M_\circ;\upbeta_\circ^*E)$ is equal to the space of lifts under $\upbeta_\circ$ of all smooth sections of $E\to M$ which vanish to infinite order at $\cC$. Thus, we may `blow down' $\pa M_\circ$ and regard $f_\flat$ as an element
\[
  f_\flat \in \CI(M;S^2 T^*M),\qquad \supp f_\flat\cap X\Subset\cU^\circ,
\]
which vanishes to infinite order at $\cC$ (although we will not use this final property). Applying the linearized second Bianchi identity to~\eqref{EqAcTruefflat} (and using that $\Ric(g)-\Lambda g=0$) implies that $\delta_g\sfG_g f_\flat=0$. The next result uses the full set of assumptions spelled out before the statement of Theorem~\ref{ThmAc} (in particular the absence of KIDs):

\begin{prop}[Solving away the trivial error]
\label{PropAcTrue}
  Suppose $f_\flat\in\CI(M;S^2 T^*M)$ satisfies $\delta_g\sfG_g f_\flat=0$, and $\supp f_\flat$ is contained in the domain of influence of a compact subset of $\cU^\circ$. Then there exists $h_\flat\in\CI(M;S^2 T^*M)$ with
  \begin{equation}
  \label{EqAcTrue}
    (D_g\Ric-\Lambda)h_\flat=f_\flat
  \end{equation}
  and so that $\supp h_\flat$ is contained in the domain of influence of a compact subset of $\cU^\circ$.
\end{prop}

We shall need:

\begin{prop}[Solving the linearized constraints equations]
\label{PropAcConstr}
  Let $X,\gamma,k,\cU^\circ$ be as above; that is, $P(\gamma,k;\Lambda)=0$ (see~\eqref{EqMConstraints}), and $(D_{(\gamma,k)}P)^*$ has trivial kernel on the space of sections of\footnote{The isomorphism is given by $T^*_X M\ni\omega\mapsto(\omega(\nu),\omega|_{T X})$ where $\nu$ is the future unit normal of $X$.} $\ul\R\oplus T^*X\cong T^*_X M$ over the non-empty smoothly bounded connected precompact open subset $\cU^\circ$. Then for all $\omega\in\CIc(\cU^\circ;T_X^*M)$, there exist $\dot\gamma,\dot k\in\CIc(\cU^\circ;S^2 T^*X)$ so that $D_{(\gamma,k)}P(\dot\gamma,\dot k;\Lambda)=\omega$.
\end{prop}
\begin{proof}
  This is a standard result, see e.g.\ \cite[\S3]{ChruscielDelayMapping}. Coercive a priori estimates for $(D_{(\gamma,k)}P)^*$ on function spaces allowing for exponential growth at the boundary of the complement $\cU^\circ_\delta\subset\cU^\circ$ of a $\delta$-neighborhood of $\pa\ol{\cU^\circ}$, with $\delta>0$ so small that $\omega\in\CIc(\cU^\circ_\delta)$ and $(X,\gamma,k)$ has no KIDs on $\cU^\circ_\delta$ still, imply the solvability of $D_{(\gamma,k)}P(\dot\gamma,\dot k)=\omega$ with smooth $\dot\gamma,\dot k$ which are exponentially decaying at $\pa\cU^\circ_\delta$; the extension of $\dot\gamma,\dot k$ by $0$ to $X\setminus\cU^\circ_\delta$ furnishes the desired solution.
\end{proof}

\begin{proof}[Proof of Proposition~\usref{PropAcTrue}]
  This is a variant of the main result of \cite{HintzLinEin}; we recall the argument for the sake of completeness. Following the strategy outlined in~\S\ref{SsELin}, we seek $h_\flat$ as the solution of an initial value problem for the gauge-fixed linearized Einstein equations. Thus, if $\nu\in\CI(X;T_X M)$ denotes the future unit normal at $X$, we solve
  \begin{equation}
  \label{EqAcTrueIVP}
    \begin{cases}
      (D_g\Ric-\Lambda)h_\flat + \delta_g^*(\delta_g\sfG_g h_\flat-\theta) = f_\flat & \text{in}\ M, \\
      (h_\flat,\nabla_\nu h_\flat) = (h_0,h_1) & \text{at}\ X,
    \end{cases}
  \end{equation}
  for carefully chosen $\theta\in\CIc(M;T^*M)$ and $h_0,h_1\in\CIc(\cU^\circ;S^2 T^*_X M)$; the task is to select $\theta,h_0,h_1$ so that $h_\flat$ also solves~\eqref{EqAcTrue}. As demonstrated in~\S\ref{SsELin}, it suffices to arrange
  \begin{equation}
  \label{EqAcTrueID}
    D_{(\gamma,k)}P(\dot\gamma,\dot k;\Lambda) = (\sfG_g f_\flat)(\nu,-)
  \end{equation}
  at $X$, where $(\dot\gamma,\dot k)$ are the linearized initial data induced by $(h_0,h_1)$. Since $(\sfG_g f_\flat)(\nu,-)\in\CIc(\cU^\circ;T_X^*M)$, and due to the absence of KIDs in $\cU^\circ$, we may apply Proposition~\ref{PropAcConstr} to obtain the existence of $\dot\gamma,\dot k\in\CIc(\cU^\circ;S^2 T^*X)$ solving~\eqref{EqAcTrueID}. Choosing $h_0,h_1$ as in Remark~\ref{RmkEIDh0h1}, we subsequently let $\theta$ be any smooth extension of $(\delta_g\sfG_g(h_0+s h_1))|_X\in\CIc(\cU^\circ;T^*_X M)$ to a 1-form on $M$ (see~\eqref{EqEIDeta}). With these data in place, we solve~\eqref{EqAcTrueIVP}. Then $h_\flat$ satisfies~\eqref{EqAcTrue}. (We recall the argument: since $\eta:=\delta_g\sfG_g h_\flat-\theta\in\CI(M;T^*M)$ vanishes at $X$, so does its normal derivative since $(\sfG_g\delta_g^*\eta)(\nu,-)=0$ due to~\eqref{EqAcTrueID}, and thus we have $\eta=0$ everywhere since $\delta_g\sfG_g\delta_g^*\eta=0$.) The support property of $h$ follows from finite speed of propagation. The proof is complete.
\end{proof}

\begin{proof}[Proof of Theorem~\usref{ThmAc}]
  In view of equations~\eqref{EqAcTruefflat} and~\eqref{EqAcTrue}, the tensor $h:=h_\sharp-\upbeta_\circ^*h_\flat$ solves the equation~\eqref{EqAc} and is indeed of the form~\eqref{EqAch}.

  When $\cC$ is a geodesic, it remains to arrange property~\eqref{ItAchflat} by exploiting the fact that we may add to $h_\flat$ any linearized pure gauge term, i.e.\ a tensor of the form $\delta_g^*\omega$ where we will be able to take $\omega\in\CI(M;T^*M)$ to be supported in an arbitrarily small neighborhood of $\cC$, so $\delta_g^*\omega$ satisfies the same support properties as $h_\flat$ in Proposition~\ref{PropAcTrue}. To arrange for simple vanishing, note that in Fermi normal coordinates $(t,x)$ around $\cC$, the restriction of $h_\flat$ to $\cC$ takes the form
  \[
    h_\flat(t,0) = a_{0 0}(t)\,\dd t^2 + 2 a_{0 k}(t)\,\dd t\,\dd x^k + a_{j k}(t)\,\dd x^j\,\dd x^k,\qquad a_{\mu\nu}\in\CI(I),
  \]
  where $a_{j k}=a_{k j}$. Pick $A_{0 0}\in\CI(I)$ with $A_{0 0}'=a_{0 0}$. Since $g$ agrees with the Minkowski metric $\ubar g$ at $\cC$ modulo error terms vanishing quadratically at $x=0$, we have
  \[
    \delta_g^*\dd t\equiv 0,\qquad \delta_g^*\dd x^j\equiv 0
  \]
  modulo tensors vanishing simply at $x=0$. Therefore,\footnote{Compare this with Theorem~\ref{ThmRicUniq}: the scalar type $1$ tensor $\dd t\,\dd x^k$, and the sum $\dd x^j\,\dd x^k$ of scalar type $0$ and $2$ tensors both lie in $\ker\wh{D_{\ubar g}\Ric}(0)$ or indeed in $\ker N(D\Ric,0)$ in the notation of~\S\ref{SRic}, and are thus pure gauge. Lastly, while $\dd t^2\in\ker N_{\rms 0}(D\Ric,0)$, this is \emph{not} in the range of $N_{\rms 0}(\delta^*,1)$ (cf.\ Proposition~\ref{PropRicMS}\eqref{ItRicMSs0}). Rather, it is the symmetric gradient of a \emph{non-stationary} 1-form, $\dd t^2=\delta_{\ubar g}^*(t\,\dd t)$. This explains the appearance of $A_{0 0}=\int a_{0 0}\,\dd t$.}
  \begin{alignat*}{2}
    \delta_g^*\bigl(A_{0 0}(t)\,\dd t\bigr) &= a_{0 0}(t)\,\dd t^2 + A_{0 0}(t)\delta_g^*\dd t &&\equiv a_{0 0}(t)\,\dd t^2, \\
    \delta_g^*\bigl(a_{0 k}(t)x^k\,\dd t\bigr) &= a_{0 k}(t)\,\dd t\,\dd x^k + a_{0 k}'(t)x^k\,\dd t^2+a_{0 k}(t)x^k\delta_g^*\dd t &&\equiv a_{0 k}(t)\,\dd t\,\dd x^k, \\
    \delta_g^*\bigl(a_{j k}(t)x^j\,\dd x^k\bigr) &= a_{j k}(t)\,\dd x^j\,\dd x^k + a_{j k}'(t)x^j\,\dd t\,\dd x^k + a_{j k}(t)x^j\delta_g^*\dd x^k &&\equiv  a_{j k}(t)\,\dd x^j\,\dd x^k.
  \end{alignat*}
  If we subtract from $h_\flat$ the tensor $\delta_g^*(\hat\chi\omega)$ where $\hat\chi\in\CI(M)$ is equal to $1$ near $\cC$ and supported in a small neighborhood of $\cC$, and with $\omega=A_{0 0}(t)\,\dd t+2 A_{0 k}(t)\,\dd x^k+a_{j k}(t)x^j\,\dd x^k$, then the new $h_\flat$ vanishes simply at $\cC$ and still satisfies the support condition of Proposition~\ref{PropAcTrue}.

  This can be improved to quadratic vanishing by similarly explicit means: modulo tensors vanishing quadratically at $x=0$, we have
  \begin{align*}
    a(t)x^j\,\dd t^2&\equiv\delta_g^*\bigl(A_1(t)x^j\,\dd t-A_2(t)\,\dd x^j\bigr), \\
    2 a(t)x^j\,\dd t\,\dd x^k&\equiv\delta_g^*\bigl(a(t)x^j x^k\,\dd t+A_1(t)(x^j\,\dd x^k-x^k\,\dd x^j)\bigr), \\
    2 a(t)x^j\,\dd x^k\,\dd x^\ell&\equiv\delta_g^*\bigl(a(t) x^j x^k\,\dd x^\ell+a(t)x^\ell(x^j\,\dd x^k-x^k\,\dd x^j)\bigr),
  \end{align*}
  where $A_1(t)=\int a(t)\,\dd t$ and $A_2(t)=\int A_1(t)\,\dd t$.

  An alternative, and conceptually cleaner, argument proceeds as follows. Consider for $s\in\R$ the symmetric 2-tensor $g+s h_\flat$; on any fixed compact subset of $M$, this is a Lorentzian metric when $s$ is sufficiently small. Fix $p\in\cC$ and a future timelike vector $v\in T_p\cC$. Denote by $\cC_s$ the geodesic for
  \[
    g_s:=g+s h_\flat
  \]
  with initial condition $(p,\frac{v}{\|v\|_{g_s}})$. Then we have Fermi normal coordinates $(t_s,x_s)$ around $\cC_s$ which are equal to $(t,x)$ for $s=0$; by the construction in the proof of Lemma~\ref{LemmaGLFermi}, these coordinates can be chosen to depend smoothly on $s$ near $0$ and are defined near any fixed precompact subset of $\cC$ when $s$ is sufficiently small. Denote by $\phi_s\colon\cD_s\subset\R\times\R^3\to M$ the corresponding coordinate chart $(T,X)\mapsto(t_s,x_s)=(T,X)$; here, if $J\Subset I$, then the domain $\cD_s$ contains a neighborhood of $J\times\{0\}$ when $s$ is small. By definition,
  \[
    \phi_s^*g_s = -\dd T^2 + \dd X^2 + G'_s(T,X;\dd T,\dd X),
  \]
  where $G'_s$ depends smoothly on $s$ and vanishes quadratically at $X=0$. Consider now the map
  \[
    \Phi_s := \phi_s\circ\phi_0^{-1} \colon \cD_0\cap\phi_0(\cD_s) \to M
  \]
  which is a diffeomorphism onto its image near $\phi_0(J\times\{0\})\subset M$ for small enough $s$, and which satisfies $\Phi_0=\Id$. Then $\Phi_0^*g_0=\Id^*g=g$ and
  \[
    \Phi_s^*(g+s h_\flat) = \Phi_s^*g_s=(\phi_0^{-1})^*(-\dd T^2+\dd X^2+G'_s)=-\dd t^2+\dd x^2+g'_s(t,x;\dd t,\dd x)
  \]
  (using the Fermi normal coordinates around $\cC$ for the expression on the right), where $g'_s$ is smooth in $s$ and vanishes quadratically at $x=0$; we express this as $g'_s=\cO(|x|^2)$. Differentiating at $s=0$ gives
  \[
    \cO(|x|^2) = \frac{\dd}{\dd s}(\Phi_s^*g)\Big|_{s=0} + \Phi_0^*h_\flat = \cL_V g + h_\flat,
  \]
  where the smooth vector field $V$ is defined in a neighborhood $\cN\subset\cC$ by $V(q)=\frac{\dd}{\dd s}\Phi_s(q)|_{s=0}$, $q\in\cN$. If $\hat\chi\in\CI(M)$ is supported in $\cN$, equal to $1$ near $\cC$, and has support sufficiently close to $\cC$, then we may replace $h_\flat$ by $h_\flat+\cL_{\hat\chi V}g$ (which is $\cO(|x|^2)$ near $\cC=\{x=0\}$) without affecting properties~\eqref{ItAcSolv}--\eqref{ItAcSupp} in Theorem~\ref{ThmAc}. This completes the proof.
\end{proof}

\subsection{Necessary conditions for gluing: I. A result of Gralla--Wald revisited}
\label{SsAcGW}

The material in this section is of supplementary character: it is not used in the rest of the paper and may therefore be skipped at first reading, but it sheds further light on the analysis in~\S\S\ref{SsAcBdy}--\ref{SsAcTrue}.

We shall show that the gluing construction (for polyhomogeneous total families $\wt g$ on $\wt M$, with the $\hat M$-model being a Kerr metric with nonzero mass) in this paper is only possible if $\cC$ is a geodesic; in this sense we recover, in our setting, the geodesic hypothesis. Furthermore, the rescaled mass of the glued black hole must be constant along $\cC$. Our arguments are closely related to those given by Gralla--Wald in \cite[\S{III}]{GrallaWaldSelfForce}; in present terminology, they exploit an obstruction to solving the linearized Einstein equations in the polyhomogeneous category with right hand sides having borderline $r^{-2}$ growth (which is precisely what is \emph{excluded} in Theorem~\ref{ThmAc}). Further necessary conditions are derived in~\S\ref{SsAhNec}.

Let $\cC=c(I)$, $I\subseteq\R$, be \emph{any} inextendible embedded smooth timelike curve (not necessarily a geodesic). Fix a smooth family of mass parameters $0<\bhm\in\CI(I)$, and consider the naively glued metric $\wt g_0$ produced by Lemma~\ref{LemmaGKNaive} for any choice of $\bha\in\CI(I;\R^3)$. Consider the leading order error of the Einstein vacuum equations for $\wt g_0$ at $M_\circ$, defined by
\begin{equation}
\label{EqAcGWError}
  f:=\Bigl(\eps^{-1}\bigl(\Ric(\wt g_0)-\Lambda\wt g_0\bigr)\Bigr)\Big|_{M_\circ}\in r^{-2}\CI(M_\circ;\upbeta_\circ^*S^2 T^*M).
\end{equation}
While by virtue of the (nonlinear) second Bianchi identity for $\wt g_0$ we have $\delta_g\sfG_g f=0$ in $(M_\circ)^\circ=M\setminus\cC$, it is no longer necessarily the case that for the unique $L^1_\loc$-extension $E f\in L^1_\loc(M;S^2 T^*M)$ of $f$ we still have $\delta_g\sfG_g(E f)=0$ in the distributional sense on $M$; rather, $\delta_g\sfG_g(E f)$ may be a $\delta$-distribution with support in $\cC$. As discussed in Remark~\ref{RmkAcWeight}, the nonvanishing of $\delta_g\sfG_g(E f)$ \emph{at $\pa M_\circ$} would render the solvability of $(D_g\Ric-\Lambda)h=f$ \emph{in the set $(M_\circ)^\circ$ with $|h|\lesssim|x|^{-1+\delta}$} impossible. (The upper bound on $h$ demanded here ensures that the correction $\eps h$ to $\wt g_0$ near $M_\circ$ does not change the $\hat M$-model metrics.) We proceed to make this explicit.

Working in Fermi normal coordinates $(t,x)$ around $\cC$ and with $t$, $r=|x|$, $\omega=\frac{x}{|x|}\in\Sph^2$ on $M_\circ$ near $\pa M_\circ$, we write $\wt g_0$ in the region $\rho_\circ:=\frac{\eps}{r}\lesssim 1$ (i.e.\ in a collar neighborhood of $M_\circ\subset\wt M$) as
\[
  \wt g_0(\rho_\circ,t,r,\omega) \equiv g(t,r,\omega) + \rho_\circ g_1(t,r,\omega) \bmod \rho_\circ^2\CI(\wt M;S^2\wt T^*\wt M),
\]
where $\Ric(g)-\Lambda g=0$ and $g_1(t,r,\omega)\in\CI(M_\circ;\upbeta_\circ^*S^2 T^*M)$. The restriction to $\hat M$, locally given by $r=0$, is $g(t,0,\omega)+\rho_\circ g_1(t,0,\omega)\bmod\rho_\circ^2\CI(\hat M;S^2\wt T^*_{\hat M}\wt M)$. Applying the map $\sfe$ from~\eqref{EqGffStatExt} to this gives the Kerr metric, which (independently of the choice of $\bha$ in Lemma~\ref{LemmaGKNaive}) modulo $\cO(\rho_\circ^2)$-terms is equal to the mass $\bhm(t)$ Schwarzschild metric
\[
  -\dd\hat t^2 + \dd\hat r^2 + \hat r^2\slg + \frac{2\bhm(t)}{\hat r}(\dd\hat t^2+\dd\hat r^2) + \cO(\hat r^{-2}),
\]
where $\hat r=\rho_\circ^{-1}$, in view of~\eqref{EqGKLot}. Since $\sfe(g(t,0,\omega))=-\dd\hat t^2+\dd\hat r^2+\hat r^2\slg$, we thus find
\begin{equation}
\label{EqAcGWhm}
  \rho_\circ g_1(t,0,\omega) = 2\bhm(t)\rho_\circ(\dd t^2+\dd r^2) = \eps g_{1,\bhm(t)},\qquad g_{1,\bhm(t)}:=\frac{2\bhm(t)}{r}(\dd t^2+\dd r^2).
\end{equation}
Since the error $\rho_\circ g_1(t,r,\omega)-\rho_\circ g_1(t,0,\omega)$ is of class $\rho_\circ r\CI=\eps\CI$, we conclude that
\begin{equation}
\label{EqAcGWMetric}
  \wt g_0(\rho_\circ,t,r,\omega)=g(t,r,\omega)+\eps g_{1,\bhm(t)}+\eps g_1'\bmod\rho_\circ^2\CI,\qquad g_1'\in\CI(M_\circ;\upbeta_\circ^*S^2 T^*M).
\end{equation}

\begin{rmk}[Leading order behavior]
\label{RmkAcGWLot}
  For $\wt g_0$ as in~\eqref{EqAcGWMetric}, but for general $g_{1,\bhm}\sim r^{-1}$, one has $f\sim r^{-3}$; the fact that $f\sim r^{-2}$ for the specific $g_{1,\bhm}$ in~\eqref{EqAcGWhm} is due to $g_{1,\bhm}\in\ker\upbeta_\circ^*\wh{D_{\ubar g}\Ric}(0)$, where we use the notation of Corollary~\ref{CorMkDiffDRic} and Lemma~\ref{LemmaAcBdyNormOp}.
\end{rmk}

Plugging~\eqref{EqAcGWMetric} into~\eqref{EqAcGWError}, we have
\[
  f = D_g(\Ric-\Lambda)(g_{1,\bhm(t)}+g_1')\quad\text{in}\ (M_\circ)^\circ.
\]
Denoting by $E$ the extension operator $E\colon L^1_\loc(M_\circ;|\dd g|)\to L^1_\loc(M;|\dd g|)$, the linearized second Bianchi identity implies
\begin{equation}
\label{EqAcGWdelGEf}
  \delta_g\sfG_g(E f) = \delta_g\sfG_g\Bigl(E D_g(\Ric-\Lambda)(g_{1,\bhm(t)}+g_1') - D_g(\Ric-\Lambda)(E g_{1,\bhm(t)}+E g_1')\Bigr).
\end{equation}
The big parenthesis is supported in $\cC$. Since $g_1'\in\cA^0(M_\circ)$ and therefore $E g_1'\in H_\loc^{\frac32-}(M)$, the term $E D_g(\Ric-\Lambda)g_1'-D_g(\Ric-\Lambda)(E g_1')\in H_\loc^{-\frac12-}(M)$, which is supported in $\cC$ (with $\codim\cC=3$), must vanish. For the same reason, in the evaluation of $E D_g(\Ric-\Lambda)(g_{1,\bhm})-D_g(\Ric-\Lambda)(E g_{1,\bhm})$ we may drop $\Lambda$ and replace $g$ by its restriction to $\cC$, i.e.\ by the Minkowski metric; moreover, since $t$-derivatives preserve the conormal order (unlike $r$-derivatives, which increase the strength of the singularity at $r=0$ by $1$ order), we may fix $\bhm$ to be equal to the constant value $\bhm(t_0)$ when computing~\eqref{EqAcGWdelGEf} at $t^{-1}(t_0)$. Since $\wh{D_{\ubar g}\Ric}(0)g_{1,\bhm(t_0)}=0$ in $r\neq 0$, only the term $D_{\ubar g}\Ric(E g_{1,\bhm(t_0)})$ in parentheses in~\eqref{EqAcGWdelGEf} is possibly nonzero. The key calculation is thus:

\begin{lemma}[Calculation for the linearized Einstein equation]
\label{LemmaAcGWSchw}
  Let $\ubar g=-\dd t^2+\dd x^2$, and write $r=|x|$. Recall $\ubar\sfG\colon h\mapsto h-\frac12\ubar g\tr_{\ubar g}h$. In the sense of distributional sections of $S^2 T^*\R^4\to\R^4$, we then have
  \begin{equation}
  \label{EqAcGWSchw}
    \ubar\sfG D_{\ubar g}\Ric\Bigl(\frac{2\bhm}{r}(\dd t^2+\dd r^2)\Bigr) = 8\bhm\pi\delta(x)\,\dd t^2.
  \end{equation}
\end{lemma}
\begin{proof}
  It suffices to prove this for $2\bhm=1$. Denote the expression on the left in~\eqref{EqAcGWSchw} by $T$. Note that $\ubar\sfG(\dd t^2+\dd r^2)=\dd t^2+\dd r^2$; since $\dd r=\frac{1}{|x|}r\,\dd r=\sum_{j=1}^3\frac{x^j}{|x|}\dd x^j$, we thus have
  \[
    \ubar\sfG T = D_{\ubar g}\Ric\Bigl(\frac{1}{|x|}(\dd t^2+\dd r^2)\Bigr)= \Bigl(\frac12(\pa_t^2-\pa_x^2)-\ubar\delta^*\ubar\delta\Bigr)\Bigl(\frac{1}{|x|}\dd t^2+\sum_{j,k}\frac{x^j x^k}{|x|^3}\dd x^j\,\dd x^k\Bigr)
  \]
  in view of the formula~\eqref{EqELinDRic} for $D_{\ubar g}\Ric$. Using $\pa_x^2\frac{1}{|x|}=-4\pi\delta(x)$ and $\ubar\delta(|x|^{-1}\dd t^2)=0$, the contribution of $\frac{1}{|x|}\dd t^2$ to $\ubar\sfG T$ equals $2\pi\delta(x)\,\dd t^2$. We further compute
  \begin{align*}
    &\Bigl(-\frac12\sum_\ell\pa_\ell^2 - \ubar\delta^*\ubar\delta\Bigr)\Bigl(\sum_{j,k}\frac{x^j x^k}{|x|^3}\dd x^j\,\dd x^k\Bigr) \\
    &\qquad = \frac12\sum_{j,k,\ell} \biggl[\pa_\ell^2\Bigl(\pa_j\frac{x^k}{|x|}-\frac{\delta_j^k}{|x|}\Bigr) + 2\pa_j\pa_\ell\frac{x^\ell x^k}{|x|^3} \biggr]\dd x^j\,\dd x^k \\
    &\qquad= \sum_j 2\pi\delta(x)\,(\dd x^j)^2 + \frac12\sum_{j,k,\ell} \biggl[ \pa_\ell^2\pa_j\frac{x^k}{|x|} - 2\pa_j\pa_\ell\Bigl(\pa_\ell\frac{x^k}{|x|}-\frac{\delta_\ell^k}{|x|}\Bigr) \biggr]\dd x^j\,\dd x^k \\
    &\qquad= 2\pi\delta(x)\,\dd x^2 + \frac12\sum_{j,k,\ell}\Bigl[ \bigl(-\pa_\ell^2\pa_j x^k + 2\delta_\ell^k\pa_\ell\pa_j\bigr)|x|^{-1}\Bigr]\dd x^j\,\dd x^k.
  \end{align*}
  But $\sum_\ell(-\pa_\ell^2\pa_j x^k+2\delta_\ell^k\pa_\ell\pa_j)=\sum_\ell\pa_j(-\pa_\ell^2 x^k+2\pa_k)=-\pa_j x^k\sum_\ell\pa_\ell^2$ annihilates $|x|^{-1}$. We conclude that $\ubar\sfG T=2\pi\delta(x)(\dd t^2+\dd x^2)$. Applying $\ubar\sfG$ to this and using $\ubar\sfG(\dd t^2+\dd x^2)=2\,\dd t^2$ gives $T=4\pi\delta(x)\,\dd t^2$, as claimed.
\end{proof}

\begin{cor}[Equations of motion for the stress-energy tensor]
\label{CorAcGWEOM}
  The $L^1_\loc(M)$-extension $E f$ of the tensor $f$ defined by~\eqref{EqAcGWError} satisfies
  \begin{equation}
  \label{EqAcGWEOM}
    \delta_g\sfG_g(E f)=8\pi\bigl(-\bhm(t)(\nabla_{\pa_t}\pa_t)^\flat + \bhm'(t)\pa_t^\flat\bigr)\delta(x).
  \end{equation}
  (Cf.\ \cite[Equations~(47) and (50)]{GrallaWaldSelfForce}.)
\end{cor}
\begin{proof}
  Lemma~\ref{LemmaAcGWSchw} implies (see also Remark~\ref{RmkAcGWLot}), via~\eqref{EqAcGWdelGEf}, that
  \[
    \delta_g\sfG_g(E f) = 8\pi\delta_g\bigl(\bhm(t)\delta(x)\,\dd t^2\bigr) = 8\pi\bhm(t)\delta_g\bigl(\delta(x)\,\dd t^2\bigr) - 8\pi\bhm'(t)\delta(x)\iota_{\nabla t}(\dd t^2).
  \]
  But since in the Fermi normal coordinates $(t,x)$ all Christoffel symbols vanish at $x=0$ except possibly for $\Gamma_{0 0}^j=\Gamma_{0 j}^0$ ($j=1,2,3$), we have $\delta_g(\delta(x)\,\dd t^2)=-\Gamma^0_{0 j}\delta(x)\,\dd x^j=-\delta(x)(\nabla_{\pa_t}\pa_t)^\flat$.
\end{proof}

To conclude, if it is possible to find $\wt h\in\CI(\wt M;S^2\wt T^*\wt M)$ (or more generally $\wt h\in\cA_\phg^{\hat\cE,\cE}$ where $\Re\hat\cE>-1$ and $\min\Re(\cE\setminus\{(0,0)\})>0$, so that $\eps\wt h$ has vanishing restriction to $\hat M$ and vanishes simply at $M_\circ$) with the property that for $\wt g_1=\wt g_0+\eps\wt h$, the error $\Ric(\wt g_1)-\Lambda\wt g_1$ vanishes to \emph{more than one order} at $M_\circ$, then we showed above that necessarily $\delta_g\sfG_g(E f)=0$; and Corollary~\ref{CorAcGWEOM} shows that this is equivalent to the requirement that $\cC$ be a geodesic, and $\bhm(t)$ must be a constant. (Indeed, since $t\mapsto c(t)$ is an arc-length parameterization of $\cC$, we have $\nabla_{\pa_t}\pa_t\perp\pa_t$, and therefore the vanishing of~\eqref{EqAcGWEOM} requires the vanishing of $\nabla_{\pa_t}\pa_t$ and $\bhm'(t)$ separately.)

We finally note that if $\cC$ is not a geodesic, then the homogeneous $r^{-2}$ leading order term of $f$ in~\eqref{EqAcGWError} at $\pa M_\circ$ gives rise to the example in Remark~\ref{RmkAcBdy0Opm2} as follows: if $g=(-1-2\Gamma_{j 0 0}x^j)\dd t^2+\sum_{j=1}^3(\dd x^j)^2$, say, with $\scal:=\sum_{j=1}^3\Gamma_{j 0 0}\frac{x^j}{|x|}\neq 0$ (cf.\ \eqref{EqGLFermi}), then for $f=D_g\Ric(g_{1,\bhm})$ with $\bhm\neq 0$, a calculation gives $\bhm^{-1}r^2 f=(-5\scal\,\dd t^2-3\scal\,\dd r^2+2\,\dd r\otimes_s r\sld\scal+4\scal r^2\slg)+o(1)$ as $r\to 0$, and this is not equal to $r^2(D_{\ubar g}\Ric-\Lambda)\tilde h$ in $r>0$ for any $\tilde h=\cO(|x|^{-1+\delta})$. (The expression in Remark~\ref{RmkAcBdy0Opm2} arises by expressing this in terms of~\eqref{EqMkBundleSplit} and \eqref{EqMkYSplits1}.) On the other hand, when $\cC$ is a geodesic, the fact that $g$ agrees to second order with the Minkowski metric along $\cC$ in Fermi normal coordinates lets one pick $\wt g_0$ in such a way that the $r^{-2}$-leading order term of $f$ vanishes at $\pa M_\circ$, i.e.\ one has $f\in r^{-1}\CI$ in~\eqref{EqAcGWError} (see Lemma~\ref{LemmaFErr0} below); and then $\delta_g\sfG_g f=0$ automatically implies $\delta_g\sfG_g(E f)=0$ by homogeneity considerations.

\section{Linear analysis on \texorpdfstring{$\hat M$}{the front face of the total gluing spacetime}}
\label{SAh}

We continue using the setup and notation of~\S\ref{SM}. We fix subextremal Kerr parameters
\[
  b=(\bhm,\bha)\in\R\times\R^3,
\]
and recall the Kerr metric $\hat g_b$ from Definition~\ref{DefGK}, defined on the compactified spacetime manifold $\hat M_b$ with compactified Cauchy surface $\hat X_b$ (see Definition~\ref{DefGKCpt}). Recall also the linearized Kerr metrics $\hat g'_b(\dot b)$ from Definition~\ref{DefGKLin}. The interior $\hat M_b^\circ$ of $\hat M_b$ is a subset of $\R^4=\R_{\hat t}\times\R^3_{\hat x}$, and indeed is the complement of $\{|\hat x|\leq\bhm\}$, while $\hat X_b^\circ=\hat t^{-1}(0)\subset\hat M_b^\circ$. Writing $|\dd\hat g_b|=|\dd\hat t||\dd\hat g_b|_{\hat X_b}|$, we have
\begin{equation}
\label{EqAh0Density}
  |\dd\hat g_b|_{\hat X_b}| \in \CI(\hat X_b;\Omegasc\hat X_b).
\end{equation}
With respect to this density, we have $\cA^\alpha(\hat X_b)\subset L^1(\hat X_b)$ if and only if $\alpha>3$. We write $\la\cdot,\cdot\ra$ for the $L^2$-inner product of sections of tensor bundles (arising as restrictions of tensor bundles on $\hat M$ to $\hat X_b$) on $\hat X_b^\circ$ with respect to the volume density $|\dd\hat g_b|_{\hat X_b}|$ and the fiber inner product induced by $\hat g_b$. We denote by
\begin{equation}
\label{EqAhTstar}
  \hat t_* := \hat t-T_*(\hat x),\qquad
  T_*\in\CI(\R^3),\quad T_*(\hat x)=\hat r=|\hat x|,\ \hat r\geq 10\bhm,
\end{equation}
a function which for $\hat r\geq 10\bhm$ is equal to the null coordinate $\hat t-|\hat x|$ on Minkowski space which was used in~\S\ref{SMk}. Finally, write
\begin{equation}
\label{EqAhKSchwParam}
  b_0=(\bhm,0)
\end{equation}
for the parameters of a Schwarzschild black hole with the same mass.

Motivated by~\S\ref{SssIPfN} (and also Proposition~\ref{PropELTAcc}\eqref{ItELTAccMcirc}), we shall study the solvability properties of the equation $D_{\hat g_b}\Ric(h)=f$ for stationary $h,f$; thus, we study
\begin{equation}
\label{EqAhKEqn}
  \wh{D_{\hat g_b}\Ric}(0)(h) = f,\qquad \wh{\delta_{\hat g_b}\sfG_{\hat g_b}}(0)f=0,
\end{equation}
on $\hat X_b^\circ$.

\subsection{Cokernel of the zero energy operator}
\label{SsAh0}

In equation~\eqref{EqAhKEqn}, we shall only consider $f\in\cA^2(\hat X_b;S^2\,\Ttsc^*_{\hat X_b}\hat M_b)=\rho_\circ^2\cA^0(\hat X_b;S^2\,\Ttsc^*_{\hat X_b}\hat M_b)$, i.e.\ $f$ has at least inverse quadratic decay as $\rho_\circ:=\la\hat x\ra^{-1}\searrow 0$, where we recall that $\rho_\circ$ is a boundary defining function of $\hat X_b$ (and also of the lift of the boundary of $\ol{\R^4_{\hat t,\hat x}}$ to $\hat M_b$). We first note:

\begin{lemma}[Necessary condition for solvability]
\label{LemmaAh0Nec}
  Let $f\in\cA^\alpha(\hat X_b;S^2\,\Ttsc^*_{\hat X_b}\hat M_b)$, $\alpha\geq 2$. If there exists a solution $h\in\cA^{-1+\delta}(\hat X_b;S^2\,\Ttsc^*_{\hat X_b}\hat M_b)$, $\delta>0$, of $\wh{D_{\hat g_b}\Ric}(0)(h)=f$, then
  \begin{equation}
  \label{EqAh0Nec}
    \la f,\wh{\sfG_{\hat g_b}\delta_{\hat g_b}^*}(0)\omega\ra = 0
  \end{equation}
  for all $\omega\in\sD'(\hat X_b^\circ)$ so that $\wh{\sfG_{\hat g_b}\delta_{\hat g_b}^*}(0)\omega\in\cA^{2-\delta+}(\hat X_b)+\sE'(\hat X_b^\circ)$ has support disjoint from $\hat r=\bhm$.
\end{lemma}

We allow for $\omega$ to be a distribution since the 1-forms $\omega$ which arise in the solvability theory for $\wh{D_{\hat g_b}\Ric}(0)$ are singular at the event horizon; see \cite[Proposition~9.1]{HaefnerHintzVasyKerr}, \cite[Theorem~7.5]{AnderssonHaefnerWhitingMode}, and Theorem~\ref{ThmAhPhgSolv} below. (By elliptic regularity, $\omega$ is smooth where $\wh{\sfG_{\hat g_b}\delta_{\hat g_b}^*}(0)\omega$ is.) We will show in Theorem~\ref{ThmAhPhg} below that~\eqref{EqAh0Nec} is also \emph{sufficient} for the solvability in the stated function space.

\begin{proof}[Proof of Lemma~\usref{LemmaAh0Nec}]
  Since $\wh{D_{\hat g_b}\Ric}(0)\in\rho_\circ^2\Diffb^2(\hat X_b;S^2\,\Ttsc^*_{\hat X_b}\hat M_b)$, integration by parts in
  \[
    \big\la\wh{D_{\hat g_b}\Ric}(0)h,\wh{\sfG_{\hat g_b}\delta_{\hat g_b}^*}(0)\omega\big\ra=\big\la h,\wh{D_{\hat g_b}\Ric}(0)^*\wh{\sfG_{\hat g_b}\delta_{\hat g_b}^*}(0)\omega\big\ra
  \]
  is justified since $(-1+\delta)+2+(2-\delta+\eps)=3+\eps>3$ when $\eps>0$. It then remains to recall that the kernel of $(D_{\hat g_b}\Ric)^*=\sfG_{\hat g_b}\circ D_{\hat g_b}\Ric\circ\sfG_{\hat g_b}$ contains all symmetric 2-tensors in the range of $\sfG_{\hat g_b}\delta_{\hat g_b}^*$.
\end{proof}

On the other hand, if $f\in\cA^\alpha(\hat X_b)$, $\alpha\geq 2$, satisfies $\wh{\delta_{\hat g_b}\sfG_{\hat g_b}}(0)f=0$ instead of~\eqref{EqAh0Nec}, then the integration by parts in
\begin{equation}
\label{EqAhKIBP}
  0 = \big\la\wh{\delta_{\hat g_b}\sfG_{\hat g_b}}(0)f,\omega\big\ra = \big\la f,\wh{\sfG_{\hat g_b}\delta_{\hat g_b}^*}(0)\omega\big\ra
\end{equation}
is only justified when $\omega\in\cA^\beta(\hat X_b)+\sE'(\hat X_b^\circ)$, $\supp\omega\cap\hat r^{-1}(\bhm)=\emptyset$, with $\alpha+1+\beta>3$, i.e.\ $\beta>2-\alpha$, which for $\alpha=2$ requires $\beta>0$ and thus $\omega$ to decay (as a 3sc-1-form, i.e.\ when expressed in the frame $\dd\hat t$, $\dd\hat x$) as $|\hat x|\to\infty$. But since the Kerr metric is asymptotically flat, it possesses approximate Killing 1-forms $\omega$ which do not decay (i.e.\ $\beta\leq 0$) and nonetheless fit the assumptions of Lemma~\ref{LemmaAh0Nec} since their symmetric gradients do decay, as we discuss below. Thus, the necessary condition~\eqref{EqAh0Nec} for solvability of $\wh{D_{\hat g_b}\Ric}(0)h=f$ with $h\in\cA^{-1+\delta}$ is strictly stronger than $\wh{\delta_{\hat g_b}\sfG_{\hat g_b}}(0)f=0$.\footnote{The equation $\wh{\delta_{\hat g_b}\sfG_{\hat g_b}}(0)f=0$ satisfied by $f$ imposes restrictions on the behavior of $f$ as $|\hat x|\to\infty$ which may cause the boundary terms at infinity in the integration by parts in~\eqref{EqAhKIBP} to vanish even when $\omega$ does not decay. That such a cancellation typically does not occur is a consequence of Theorem~\ref{ThmAhKCoker} below: the (stationary) tensors in~\eqref{EqAhKModulation} lie in the kernel of $\wh{\delta_{\hat g_b}\sfG_{\hat g_b}}(0)$ but are not orthogonal to certain (quadratically decaying) tensors $\sfG_{\hat g_b}\delta_{\hat g_b}^*\omega$ by Theorem~\ref{ThmAhKCoker}.}

\begin{definition}[Approximate Killing 1-forms]
\label{DefAhK}
  On $\R^4=\R_{\hat t}\times\R^3_{\hat x}$, $\hat x=(\hat x^1,\hat x^2,\hat x^3)$, we define the 1-forms
  \[
    \omega_0 := \dd\hat t, \qquad
    \omega_j := \dd\hat x^j, \qquad
    \omega_{j k} := \hat x^j\,\dd\hat x^k-\hat x^k\,\dd\hat x^j,
  \]
  where $j,k=1,2,3$ and $j\neq k$.
\end{definition}

Thus, $\omega_\mu\in\CI(\hat X_b;\Ttsc^*_{\hat X_b}\hat M_b)$ for $\mu=0,\ldots,3$, and $\omega_{j k}\in\rho_\circ^{-1}\CI(\hat X_b;\Ttsc^*_{\hat X_b}\hat M_b)$. Moreover, $\omega_0$ is of scalar type $0$, $\omega_j$ is of scalar type $1$, and $\omega_{j k}$ is of vector type $1$.

\begin{lemma}[Deformation tensors]
\label{LemmaAhKdef}
  For $\omega=\omega_0,\omega_j,\omega_{j k}$, we have
  \begin{equation}
  \label{EqAhKdef}
    \delta_{\hat g_b}^*\omega\in\rho_\circ^2\CI(\hat X_b;S^2\,\Ttsc^*_{\hat X_b}\hat M_b).
  \end{equation}
\end{lemma}
\begin{proof}
  First, note that $\wh{\delta_{\hat g_b}^*}(0)\in\rho_\circ\Diffb^1(\hat X_b;\Ttsc^*_{\hat X_b}\hat M_b,S^2\,\Ttsc^*_{\hat X_b}\hat M_b)$ differs from $\wh{\delta_{\ubar g}^*}(0)$ by a term of class $\rho_\circ^2\Diffb^0$ and from $\wh{\delta_{\hat g_{b_0}}^*}(0)$ by a term of class $\rho_\circ^3\Diffb^0$ since the metrics $\hat g_b$ and $\ubar g$, resp.\ $\hat g_{b_0}$ differ by $\rho_\circ\CI$, resp.\ $\rho_\circ^2\CI$ as sections of $S^2\,\Ttsc^*_{\hat X_b}\hat M_b$, see~\eqref{EqGKCpt3sc}. (Indeed, this implies that the Christoffel symbols in the $(\hat t,\hat x)$-coordinates differ by $\rho_\circ^2\CI$, resp.\ $\rho_\circ^3\CI$.) Since $\omega\in\ker\delta_{\ubar g}^*$, the membership~\eqref{EqAhKdef} follows directly for $\omega=\omega_0,\omega_j$, while for the more growing 1-forms $\omega=\omega_{j k}$ one notes the stronger vanishing $\delta_{\hat g_{b_0}}^*\omega=0$.
\end{proof}

Lemma~\ref{LemmaAhKdef} provides us with a $7$-dimensional space of 1-forms to which Lemma~\ref{LemmaAh0Nec} applies, but not~\eqref{EqAhKIBP} when $f\in\cA^2$; in this sense, the cokernel has (at least) $7$ dimensions more than what the condition $\wh{\delta_{\hat g_b}\sfG_{\hat g_b}}(0)f=0$ imposes.

We next discuss Lorentz boosts. Rather than considering $\hat t\,\dd\hat x^j-\hat x^j\,\dd\hat t$, we will use 1-forms which are better adapted to the Kerr metric:

\begin{definition}[Translations and boosts on Kerr]
\label{DefAhKBoosts}
  Set\footnote{Thus, $\omega_{b,j}\equiv\omega_j\bmod\rho_\circ\CI(\hat X_b;\Ttsc^*_{\hat X_b}\hat M_b)$ in view of~\eqref{EqGKCpt3sc}.}
  \[
    \omega_{b,j} := \pa_{\hat x^j}^\flat = \hat g_b(\pa_{\hat x^j},-).
  \]
  We define $h_{b,j}$, $\breve h_{b,j}$, and $\breve h_{*,b,j}$ using the $\hat t_*$-coordinate from~\eqref{EqAhTstar} by\footnote{See~\eqref{EqAhKBoostTerm1} and \eqref{EqAhKBoostTerm2} below for the explicit expressions.}
  \[
    h_{b,j} := \delta_{\hat g_b}^*\omega_{b,j},\qquad
    \delta_{\hat g_b}^*(\hat t\pa_{\hat x^j}^\flat+\hat x^j\pa_{\hat t}^\flat) = \hat t h_{b,j}+\breve h_{b,j} = \hat t_* h_{b,j} + \breve h_{*,b,j}.
  \]
  For $\hat c\in\R^3$, we moreover let
  \[
    \omega_{b,\hat c} = \sum_{j=1}^3 \hat c_j\omega_{b,j},\qquad
    h_{b,\hat c} = \sum_{j=1}^3 \hat c_j h_{b,j},\qquad
    \breve h_{b,\hat c} = \sum_{j=1}^3 \hat c_j \breve h_{b,j},\qquad
    \breve h_{*,b,\hat c} = \sum_{j=1}^3 \hat c_j \breve h_{*,b,j}.
  \]
\end{definition}

Writing $e_1=(1,0,0)^T$, $e_2=(0,1,0)^T$, $e_3=(0,0,1)^T$ for the standard basis of $\R^3$, we have
\begin{equation}
\label{EqAhKTrans}
  h_{b,j}=\frac12\cL_{\pa_{\hat x^j}}\hat g_b,\qquad
  h_{b,\hat c} = \frac12\cL_{\hat c\cdot\pa_{\hat x}}\hat g_b.
\end{equation}

\begin{lemma}[Leading order terms of boosts]
\label{LemmaAhKBoostLot}
  For $\hat c\in\R^3$, set $\scal(\hat c):=\hat c\cdot\frac{\hat x}{|\hat x|}\in\scalspace_1\subset\CI(\Sph^2)$. We have $h_{b,\hat c}\in\rho_\circ^2\CI(\hat X_b;\Ttsc^*_{\hat X_b}\hat M_b)$ and $\breve h_{b,\hat c},\breve h_{*,b,\hat c}\in\rho_\circ\CI(\hat X_b;\Ttsc^*_{\hat X_b}\hat M_b)$, and
  \begin{alignat}{2}
  \label{EqAhKBoostLot}
    h_{b,\hat c} &\equiv  \frac{\bhm}{\hat r^2}\Bigl( -\scal(\hat c)(\dd\hat t^2+\dd\hat r^2) + 2\,\dd\hat r\otimes_s\hat r\,\sld\scal(\hat c) \Bigr) &&\bmod \rho_\circ^3\CI, \\
  \label{EqAhKBoostLotBreve}
    \breve h_{*,b,\hat c} &\equiv \frac{\bhm}{\hat r}\Bigl( \scal(\hat c)(-\dd\hat t^2+4\,\dd\hat t\,\dd\hat r-\dd\hat r^2) + 2(\dd\hat t+\dd\hat r)\otimes_s\hat r\,\sld\scal(\hat c) \Bigr) &&\bmod \rho_\circ^2\CI.
  \end{alignat}
\end{lemma}
\begin{proof}
  In the computation of $h_{b,j}=\frac12\cL_{\pa_{\hat x^j}}\hat g_b$ modulo $\rho_\circ^3\CI$, contributions to $\hat g_b$ of class $\rho_\circ^2\CI$ do not matter since $\pa_{\hat x^j}\in\Vsc(\hat X_b)=\rho_\circ\Vb(\hat X_b)$. Since $\hat{\ubar g}=-\dd\hat t^2+\dd\hat x^2$ (see Definition~\ref{DefGKModel}) is translation-invariant, we thus conclude from~\eqref{EqGKLot} that
  \[
    h_{b,j} \equiv \frac12\cL_{\pa_{\hat x^j}}\Bigl(\frac{2\bhm}{\hat r}(\dd\hat t^2+\dd\hat r^2)\Bigr) = \frac{\bhm}{\hat r^2}\Bigl( -(\pa_{\hat x^j}\hat r)(\dd\hat t^2+\dd\hat r^2) + 2\,\dd\hat r\otimes_s(\hat r\cL_{\pa_{\hat x^j}}\dd\hat r) \Bigr)
  \]
  modulo $\rho_\circ^3\CI$. Since $\cL_{\pa_{\hat x^j}}\dd\hat r=\dd(\pa_{\hat x^j}\hat r)=\dd\frac{\hat x^j}{\hat r}$, this gives~\eqref{EqAhKBoostLot}.

  For the computation of $\breve h_{b,j}$ and $\breve h_{*,b,j}$, we use the formula
  \begin{equation}
  \label{EqAhKBoostTerm1}
    \delta_{\hat g_b}^*(\hat t\pa_{\hat x^j}^\flat+\hat x^j\pa_{\hat t}^\flat) = \hat t h_{b,j} + \breve h_{b,j},\qquad \breve h_{b,j}=\dd\hat t\otimes_s(\omega_{b,j}-\dd\hat x^j) + (\pa_{\hat t}^\flat+\dd\hat t)\otimes_s\dd\hat x^j;
  \end{equation}
  the $\hat t$-independent term $\breve h_{b,j}$ thus lies in $\rho_\circ\CI$. Using the $\hat t_*$-coordinate, we similarly have
  \begin{equation}
  \label{EqAhKBoostTerm2}
    \breve h_{*,b,j}=T_* h_{b,j} + \breve h_{b,j} \in\rho_\circ\CI;
  \end{equation}
  recall then that $T_*=\hat r$ for large $\hat r$. More precisely,
  \[
    \omega_{b,j}-\dd\hat x^j \equiv \iota_{\pa_{\hat x^j}}\Bigl(\frac{2\bhm}{\hat r}(\dd\hat t^2+\dd\hat r^2)\Bigr) = \frac{2\bhm}{\hat r}\frac{\hat x^j}{\hat r}\dd\hat r
  \]
  modulo $\rho_\circ^2\CI$, further
  \begin{equation}
  \label{EqAhKBoostpat}
    -\pa_{\hat t}^\flat\equiv\Bigl(1-\frac{2\bhm}{\hat r}\Bigr)\dd\hat t
  \end{equation}
  modulo $\rho_\circ^2\CI$ by~\eqref{EqGKLot}, and thus $\pa_{\hat t}^\flat+\dd\hat t\equiv\frac{2\bhm}{\hat r}\dd\hat t$. Plugging this into~\eqref{EqAhKBoostTerm1}--\eqref{EqAhKBoostTerm2} gives~\eqref{EqAhKBoostLotBreve}.
\end{proof}

The next result demonstrates that there is indeed a 7-dimensional space of obstructions for the solvability of $\wh{D_{\hat g_b}\Ric}(0)h=f\in\cA^2\cap\ker\wh{\delta_{\hat g_b}\sfG_{\hat g_b}}(0)$, and at the same time gives us the means to project off this 7-dimensional space (which is thus a subspace of the quotient of $\ker_{\cA^\alpha}\wh{\delta_{\hat g_b}\sfG_{\hat g_b}}(0)$, $\alpha\geq 2$, by the subspace of elements in the range of $\wh{D_{\hat g_b}\Ric}(0)$ on $\cA^{-1+\delta}$) via modulation of the Kerr and center of mass parameters. Write\footnote{The scaling is chosen such that the coefficients of $\vect(\fq)$ in the splitting~\eqref{EqMkYSplit}, using $\hat r=|\hat x|$ in place of $r$, are independent of $\hat r$.}
\begin{equation}
\label{EqAhKVect1form}
  \vect(\fq):=\Bigl(\fq\times\frac{\hat x}{|\hat x|}\Bigr)\cdot\frac{\dd\hat x}{|\hat x|}\in\vectspace_1,\qquad \fq\in\R^3,
\end{equation}
and recall $\scal(\fq)=\fq\cdot\frac{\hat x}{|\hat x|}\in\scalspace_1$ from Lemma~\usref{LemmaAhKBoostLot}.

\begin{thm}[Eliminating the cokernel]
\label{ThmAhKCoker}
  Suppose $\omega_0^*,\omega_j^*,\omega_{j k}^*$ are stationary 1-forms on $\hat M_b$ so that, for some $\delta>0$, the differences $\tilde\omega_\mu=\omega_\mu^*-\omega_\mu$, $\mu=0,1,2,3$, and $\tilde\omega_{j k}=\omega_{j k}^*-\omega_{j k}$, $1\leq j\neq k\leq 3$, lie in the space $\cA^\delta(\hat X_b;\Ttsc^*_{\hat X_b}\hat M_b) + \sE'(\hat X_b^\circ;T_{\hat X_b^\circ}\hat M_b^\circ)$, and so that $\supp\omega_\mu^*$, $\supp\omega_{j k}^*$ are disjoint from $\hat r^{-1}(\bhm)\subset\hat X_b$. Let
  \begin{align*}
    \cK^*_{b,\rm COM}&:=\mathspan\,\bigl\{\sfG_{\hat g_b}\delta_{\hat g_b}^*\omega_j^* \colon j=1,2,3 \bigr\}, \\
    \cK^*_{b,\rm Kerr}&:=\mathspan\,\bigl\{\sfG_{\hat g_b}\delta_{\hat g_b}^*\omega_0^*,\ \ \sfG_{\hat g_b}\delta_{\hat g_b}^*\omega_{j k}\colon 1\leq j<k\leq 3\bigr\}, \\
    \cK^*_{b,\rm tot}&:=\cK^*_{b,\rm COM}\oplus\cK^*_{b,\rm Kerr}.
  \end{align*}
  For $\bullet\in\{{\rm COM},{\rm Kerr},{\rm tot}\}$, write $(\cK^*_{b,\bullet})^*:=\cL(\cK^*_{b,\bullet},\R)$ for the dual space. Define the linear maps\footnote{The first arguments of the $L^2(\hat X_b^\circ;S^2 T^*_{\hat X_b^\circ}\hat M_b^\circ)$-pairings here are described in Lemma~\ref{LemmaAhKModulation} below.}
  \begin{alignat}{2}
  \label{EqAhKCokerC}
    \ell_{b,\rm COM} &\colon \R^3 \to (\cK^*_{b,\rm tot})^*,&\qquad
    \ell_{b,\rm COM}(\hat c) &\colon h^* \mapsto \bigg\la D_{\hat g_b}\Ric\biggl(\frac{\hat t_*^2}{2}h_{b,\hat c}+\hat t_*\breve h_{*,b,\hat c}\biggr), h^*\bigg\ra, \\
  \label{EqAhKCokerK}
    \ell_{b,\rm Kerr} &\colon \R^4 \to (\cK^*_{b,\rm tot})^*,&\qquad
    \ell_{b,\rm Kerr}(\dot b) &\colon h^* \mapsto \big\la D_{\hat g_b}\Ric(\hat t_*\hat g_b'(\dot b)),h^*\big\ra.
  \end{alignat}
  Denote by $\pi_{b,\rm COM}\colon(\cK^*_{b,\rm tot})^*\to(\cK^*_{b,\rm COM})^*$ and $\pi_{b,\rm Kerr}\colon(\cK^*_{b,\rm tot})^*\to(\cK^*_{b,\rm Kerr})^*$ the projection maps.\footnote{That is, $\pi_{b,\rm COM}$ restricts the domain of definition of a linear functional to $\cK^*_{b,\rm COM}$; this is the adjoint of the inclusion $\cK_{b,\rm COM}^*\hra\cK_{b,\rm tot}^*$. Likewise for $\pi_{b,\rm Kerr}$.} Then:
  \begin{enumerate}
  \item\label{ItAhKCokerIndep} the maps $\ell_{b,\rm COM}$ and $\ell_{b,\rm Kerr}$ are independent of the lower order terms $\tilde\omega_\mu$, $\tilde\omega_{j k}$;
  \item\label{ItAhKCokerKerrPK} $\pi_{b,\rm Kerr}\circ\ell_{b,\rm Kerr}\colon\R^4\to(\cK^*_{b,\rm Kerr})^*$ is an isomorphism;
  \item\label{ItAhKCokerKerrPC} $\pi_{b,\rm COM}\circ\ell_{b,\rm Kerr}=0$;
  \item\label{ItAhKCokerCOM} $\pi_{b,\rm COM}\circ\ell_{b,\rm COM}\colon\R^3\to(\cK_{b,\rm COM}^*)^*$ is an isomorphism.
  \item\label{ItAhCokerCOMKerr} $\pi_{b,\rm Kerr}\circ\ell_{b,\rm COM}=0$.
  \end{enumerate}
  Explicitly, writing\footnote{We commit this abuse of notation only for better readability of the table below. We consider it an acceptable abuse due to part~\eqref{ItAhKCokerIndep}.} $\dd\hat t$, $\hat r^2\vect(\fq)$, $\sld(\hat r\scal(\fq))$ for the 1-form in $\mathspan\{\omega_\mu^*,\omega_{j k}^*\}$ with leading order term $\dd\hat t$, $\hat r^2\vect(\fq)$, $\sld(\hat r\scal(\fq))=\fq\cdot\dd\hat x$, respectively, we have, with $b=(\bhm,\bha)$,
  \begin{align*}
    &\hspace{4.5em}\hspace{4em}\begin{minipage}{10.5em}\centering $\cK_{b,\rm Kerr}^*$ \end{minipage} \hspace{0.8em} \begin{minipage}{6.7em}\centering $\cK_{b,\rm COM}^*$\end{minipage} \\[-0.2em]
    &\hspace{4.5em}\hspace{4em}\overbrace{\hspace{10.5em}} \hspace{0.8em} \overbrace{\hspace{6.7em}} \\[-0.5em]
    &\begin{minipage}{4.5em}{
      \vspace{1.3em}
      \begin{alignat*}{2}
        &\ell_{b,\rm Kerr} &\ &\Big\{ \\[0.1em]
        &\ell_{b,\rm COM} &\ &\,\{
      \end{alignat*}
    }\end{minipage}
    \begin{array}{c||c|c|c}
        & \sfG_{\hat g_b}\delta_{\hat g_b}^*\dd\hat t & \sfG_{\hat g_b}\delta_{\hat g_b}^*(\hat r^2\vect(\fq)) & \sfG_{\hat g_b}\delta_{\hat g_b}^*(\sld(\hat r\scal(\fq))) \\
      \hline
      \hline
      (\dot\bhm,0) & -8\pi\dot\bhm & -8\pi(\fq\cdot\bha)\dot\bhm & 0 \\
      (0,\dot\bha) & 0 & -8\pi\bhm(\fq\cdot\dot\bha) & 0 \\
      \hline
      \hat c & 0 & 0 & 4\pi\bhm(\fq\cdot\hat c).
    \end{array}
  \end{align*}
  Finally:
  \begin{enumerate}
  \setcounter{enumi}{5}
  \item\label{ItAhKCokerNoC} $\la D_{\hat g_b}\Ric(\hat t_*h_{b,\hat c}),h^*\ra=0$ for all $h^*\in\cK_{b,\rm tot}^*$;
  \item\label{ItAhCokerTstar} the map $\ell_{b,\rm Kerr}$ is unchanged when replacing $\hat t_*$ by $\hat t$, or indeed by any function $\hat t'$ for which $\hat t_*-\hat t'\in\cA^{-1}(\hat X_b)$. The same is true for the map $\ell_{b,\rm COM}$ if in addition one replaces $\breve h_{*,b,\hat c}$ by $\breve h'_{b,\hat c}=\breve h_{*,b,\hat c}+(\hat t_*-\hat t')h_{b,\hat c}$ (so that $\hat t_* h_{b,\hat c}+\breve h_{*,b,\hat c}=\hat t'h_{b,\hat c}+\breve h'_{b,\hat c}$), so in particular when using $\hat t$ and $\breve h_{b,\hat c}$ in place of $\hat t_*$ and $\breve h_{*,b,\hat c}$.
  \end{enumerate}
\end{thm}

\begin{lemma}[Modulation terms]
\label{LemmaAhKModulation}
  We have
  \begin{equation}
  \label{EqAhKModulation}
    D_{\hat g_b}\Ric\biggl(\frac{\hat t_*^2}{2}h_{b,\hat c}+\hat t_*\breve h_{*,b,\hat c}\biggr),\quad
    D_{\hat g_b}\Ric(\hat t_*\hat g_b'(\dot b)) \in \rho_\circ^2\CI\bigl(\hat X_b;S^2\,\Ttsc^*_{\hat X_b}\hat M_b\bigr).
  \end{equation}
\end{lemma}
\begin{proof}
  Since $h_{b,\hat c}$ and $\hat t_*h_{b,\hat c}+\breve h_{*,b,\hat c}$ lie in $\ker D_{\hat g_b}\Ric$, we have
  \begin{align}
    D_{\hat g_b}\Ric\Bigl(\frac{\hat t_*^2}{2}h_{b,\hat c}+\hat t_*\breve h_{*,b,\hat c}\Bigr) &= \frac{\hat t_*^2}{2}D_{\hat g_b}\Ric(h_{b,\hat c}) + \frac12\hat t_*[D_{\hat g_b}\Ric,\hat t_*]h_{b,\hat c} + \frac12[D_{\hat g_b}\Ric,\hat t_*](\hat t_* h_{b,\hat c}) \nonumber\\
      &\quad\qquad + \hat t_* D_{\hat g_b}\Ric(\breve h_{*,b,\hat c}) + [D_{\hat g_b}\Ric,\hat t_*]\breve h_{*,b,\hat c} \nonumber\\
      &= \hat t_*\bigl( [D_{\hat g_b}\Ric,\hat t_*]h_{b,\hat c} + D_{\hat g_b}\Ric(\breve h_{*,b,\hat c})\bigr) + \frac12\bigl[ [D_{\hat g_b}\Ric,\hat t_*], \hat t_*\bigr]h_{b,\hat c} \nonumber\\
      &\quad\qquad + [D_{\hat g_b}\Ric,\hat t_*]\breve h_{*,b,\hat c} \nonumber\\
  \label{EqAhKCoker2}
      &= \frac12\bigl[ [ D_{\hat g_b}\Ric, \hat t_* ], \hat t_* \bigr] h_{b,\hat c} + [D_{\hat g_b}\Ric,\hat t_*]\breve h_{*,b,\hat c}.
  \end{align}
  Since $h_{b,\hat c}$ and $\breve h_{*,b,\hat c}$ are stationary, we can replace the operators acting on them by their zero energy operators. Since $\hat t_*\equiv\hat t\bmod\rho_\circ^{-1}\CI$, the (zero energy operator of the) double commutator lies in $\Diffb^0(\hat X_b)$ and the zero energy operator of the commutator lies in $\rho_\circ\Diffb^1$; this uses~\eqref{EqGVfhatMSpecn0Diff}. Since $h_{b,\hat c}\in\rho_\circ^2\CI$ and $\breve h_{*,b,\hat c}\in\rho_\circ\CI$ by Lemma~\ref{LemmaAhKBoostLot}, this verifies the first membership in~\eqref{EqAhKModulation}.

  Similarly, in view of~\eqref{EqGKLinSize} we have
  \begin{equation}
  \label{EqAhKCoker}
    D_{\hat g_b}\Ric(\hat t_*\hat g_b'(\dot b)) = [D_{\hat g_b}\Ric,\hat t_*]\hat g_b'(\dot b) \in \rho_\circ^2\CI.\qedhere
  \end{equation}
\end{proof}

\begin{proof}[Proof of Theorem~\usref{ThmAhKCoker}]
  Since the first argument in the pairing in~\eqref{EqAhKCokerC} lies in $\rho_\circ^2\CI$ by Lemma~\ref{LemmaAhKModulation}, its inner product with the tensors in $\cK_{b,\rm COM}^*$, which have at least inverse quadratic decay as $\hat r\to\infty$, are well-defined; this shows that $\ell_{b,\rm COM}$ is well-defined. Similarly, the well-definedness of $\ell_{b,\rm Kerr}$ follows from the second membership in~\eqref{EqAhKModulation}.

  \pfstep{Part~\eqref{ItAhKCokerIndep}} will be proved in the course of the computations for parts~\eqref{ItAhKCokerKerrPK}--\eqref{ItAhCokerCOMKerr}, and will be seen to be due to the fact that all pairings can be rewritten as boundary pairings (i.e.\ $L^2$-inner products only involving the leading order terms of various tensors at $\pa\hat X_b$).

  \pfstep{Parts~\eqref{ItAhKCokerKerrPK}--\eqref{ItAhKCokerKerrPC}: computation of $\ell_{b,\rm Kerr}$.} For each $h^*$, we will rewrite the inner product~\eqref{EqAhKCokerK} as a boundary pairing. We begin with preliminary computations. Let $\chi\in\CI([0,\infty))$ be identically $0$ on $[0,1]$ and identically $1$ on $[2,\infty)$, and let $\chi_\eps=\chi(\rho_\circ/\eps)$, $\rho_\circ=\la\hat x\ra^{-1}$. Then, by~\eqref{EqAhKCoker}, writing $h^*=\sfG_{\hat g_b}\delta_{\hat g_b}^*\omega^*$, and recalling~\eqref{EqELinDRicAdj}, we have
  \begin{align}
    \la D_{\hat g_b}\Ric(\hat t_*\hat g'_b(\dot b)),h^*\ra &= \big\la [ D_{\hat g_b}\Ric, \hat t_*]\hat g'_b(\dot b), h^* \big\ra \nonumber\\
      &=-\big\la \hat g'_b(\dot b), [(D_{\hat g_b}\Ric)^*,\hat t_*]\sfG_{\hat g_b}\delta_{\hat g_b}^*\omega^* \big\ra \nonumber\\
  \label{EqAhKCokerCalc}
      &= -\big\la \hat g_b'(\dot b), \sfG_{\hat g_b}[D_{\hat g_b}\Ric,\hat t_*]\delta_{\hat g_b}^*\omega^* \big\ra.
  \end{align}
  The integration by parts does not produce any boundary terms: at infinity, this is due to the sufficient decay of the terms involved (namely, $[D_{\hat g_b}\Ric,\hat t_*]\ftrans(0)\in\rho_\circ\Diffb^1$, $\hat g'_b(\dot b)\in\rho_\circ\CI$, and $h^*\in\rho_\circ^2\CI(\hat X_b)+\cE'(\hat X_b^\circ)$), while near the inner boundary $\hat r=\hat\bhm$ of $\hat X_b$, we use the vanishing of $\omega^*$. Since $D_{\hat g_b}\Ric\circ\delta_{\hat g_b}^*=0$, we further have
  \[
    [D_{\hat g_b}\Ric,\hat t_*]\delta_{\hat g_b}^*\omega^* = D_{\hat g_b}\Ric(\hat t_*\delta_{\hat g_b}^*\omega^*) = -D_{\hat g_b}\Ric([\delta_{\hat g_b}^*,\hat t_*]\omega^*) = -D_{\hat g_b}\Ric(\dd\hat t_*\otimes_s\omega^*).
  \]
  We can thus further rewrite~\eqref{EqAhKCokerCalc} as
  \begin{equation}
  \label{EqAhCokerCalc2}
    \big\la\hat g_b'(\dot b),\sfG_{\hat g_b}\wh{D_{\hat g_b}\Ric}(0)(\dd\hat t_*\otimes_s\omega^*)\big\ra = \big\la\hat g_b'(\dot b),\wh{D_{\hat g_b}\Ric}(0)^*\sfG_{\hat g_b}(\dd\hat t_*\otimes_s\omega^*)\big\ra.
  \end{equation}
  (Retracing the calculations thus far back to~\eqref{EqAhKCokerCalc}, the second argument here has $\cO(\rho_\circ^3)$ decay, and thus the inner product is well-defined.) Inserting $\chi_\eps$ in the left factor and taking the limit $\eps\searrow 0$, we can integrate by parts and obtain
  \begin{equation}
  \label{EqAhCokerCalc3}
  \begin{split}
    \ell_{b,\rm Kerr}(\dot b)(h^*) &= \lim_{\eps\searrow 0} \big\la \chi_\eps\hat g'_b(\dot b),\wh{D_{\hat g_b}\Ric}(0)^*\sfG_{\hat g_b}(\dd\hat t_*\otimes_s\omega^*)\big\ra \\
      &= \lim_{\eps\searrow 0}\big\la \bigl[ \wh{D_{\hat g_b}\Ric}(0),\chi_\eps \bigr] \hat g'_b(\dot b), \sfG_{\hat g_b}(\dd\hat t_*\otimes_s\omega^*) \big\ra.
  \end{split}
  \end{equation}

  \pfsubstep{(i)}{Pairings with deformation tensors of (approximate) spacetime translations $\omega_0^*,\omega_j^*$.} If $\hat g'_b(\dot b)\in\rho_\circ^2\CI$---which is the case if and only if $\dot b=(0,\dot\bha)$---and $\omega^*=\cO(1)$---which holds for $\omega_\mu^*$, $\mu=0,1,2,3$---then we can integrate by parts directly in~\eqref{EqAhCokerCalc2} (cf.\ the discussion following~\eqref{EqAh0Density}), and thus the inner product evaluates to $0$. This shows that
  \begin{equation}
  \label{EqAhKCokerKPf1}
    \ell_{b,\rm Kerr}(0,\dot\bha) ( h^* ) = 0,\qquad \dot\bha\in\R^3,\ h^*\in\cK_{b,\rm COM}^*\oplus \mathspan\{\sfG_{\hat g_b}\delta_{\hat g_b}^*\omega_0^*\}.
  \end{equation}

  For all other combinations of $\dot b$ and $h^*$, we use the expression~\eqref{EqAhCokerCalc3}. Modifications of $\omega^*$ which have compact support or lie in $\cA^\delta$, $\delta>0$, do not affect the limit (since they do not affect~\eqref{EqAhCokerCalc2}), and therefore we may now replace $\omega^*$ by one of the 1-forms $\omega_0$, $\omega_j$, $\omega_{j k}$. Consider first $h^*=\sfG_{\hat g_b}\delta_{\hat g_b}^*\omega^*$ with $\omega^*\in\mathspan\{\omega_\mu^*\colon\mu=0,1,2,3\}$ as in~\eqref{EqAhKCokerKPf1} and $\dot b=(\dot\bhm,0)$, thus $\omega^*\in\CI$ and $\hat g'_b(\dot b)\in\rho_\circ\CI$. Then we may replace $\hat g_b$ by the Minkowski metric $\hat{\ubar g}\equiv\hat g_b\bmod\rho_\circ\CI$, and Lemma~\ref{LemmaBgBdyPair} thus gives
  \begin{equation}
  \label{EqAhKCokerKPf2}
  \begin{split}
    &\ell_{b,\rm Kerr}(\dot\bhm,0)(h^*) = \big\la\pa_\lambda N(\rho_\circ^{-2}\wh{D_{\hat{\ubar g}}\Ric}(0),1)h_{(1)},h^*_{(0)}\big\ra_{L^2(\pa\hat X_b)}, \\
    &\qquad
      h_{(1)}:=\bigl(\rho_\circ^{-1}\hat g_b'(\dot\bhm,0)\bigr)|_{\pa\hat X_b}=2\dot\bhm(\dd\hat t^2+\dd\hat r^2),\quad
      h^*_{(0)}:=\bigl(\sfG_{\hat{\ubar g}}(\dd\hat t_*\otimes_s\omega^*)\bigr)|_{\pa\hat X_b}.
  \end{split}
  \end{equation}
  Here, we use the standard volume density on $\pa\hat X_b\cong\Sph^2$ and the fiber inner product on $S^2\,\Ttsc^*_{\pa\hat X_b}\hat M_b$ induced by $\hat{\ubar g}$; and we used~\eqref{EqGKLot}. Note that $h_{(1)}$ is of scalar type $0$. Thus, if $\omega^*=\omega_j$ for $j=1,2,3$, then $\omega^*$ is of scalar type $1$, and thus so is $h^*_{(0)}$; therefore, $\ell_{b,\rm Kerr}(\dot b)(h^*)=0$ in this case. Together with~\eqref{EqAhKCokerKPf1}, this proves part~\eqref{ItAhKCokerKerrPC}. In the remaining case $\omega^*=\omega_0=\dd\hat t$, we need to perform a calculation. Write $\hat x^0=\hat t+\hat r$ and $\hat x^1=\hat t-\hat r$ (which equals $\hat t_*$ near $\hat r=\infty$); in the bundle splittings~\eqref{EqMkBundleSplit} and using Lemma~\ref{LemmaMk}, we then find
  \begin{alignat*}{2}
    h_{(1)} &= \dot\bhm\bigl((\dd\hat x^0)^2+(\dd\hat x^1)^2\bigr) &&= \dot\bhm(1,0,0,1,0,0)^T, \\
    h^*_{(0)} &= \sfG_{\hat{\ubar g}}\Bigl(\dd\hat x^1\otimes_s\frac12(\dd\hat x^0+\dd\hat x^1)\Bigr) &&= \frac12(0,0,0,1,0,\slg)^T.
  \end{alignat*}
  Corollary~\ref{CorMkDiffDRic} with $\lambda=1$ gives $\pa_\lambda N(\rho_\circ^{-2}\wh{D_{\hat{\ubar g}}\Ric}(0),1)h_{(1)}=\dot\bhm(0,-\half,0,0,0,-2\slg)^T$.\footnote{This is $-\dd\hat x^0\,\dd\hat x^1-2\hat r^2\slg=-\dd t^2+\dd\hat r^2-2\hat r^2\slg=-\dd t^2-2\,\dd\hat x^2+3\,\dd\hat r^2$. Lemma~\ref{LemmaAcGWSchw} is an averaged version of this computation: the distributional pairing computed there is, for $h^*\in\CIc(\R^3;S^2 T^*\R^4)$, given by $\la D_{\ubar g}\Ric(\frac{1}{r}(\dd t^2+\dd r^2)),h^*\ra=\la\frac{1}{r}(\dd t^2+\dd r^2)h,(D_{\ubar g}\Ric)^*h^*\ra=\lim_{\eps\searrow 0}\la[D_{\ubar g}\Ric,\chi(\frac{r}{\eps})]\frac{1}{r}(\dd t^2+\dd r^2),h^*\ra=\int_{\Sph^2}\la\pa_\lambda N(\wh{D_{\ubar g}\Ric}(0),-1)(\dd t^2+\dd r^2)(\omega),h^*(0)\ra\,\dd\slg(\omega)$, where the indicial operator is defined in terms of $r$ (as compared to $\hat r^{-1}$ in the current section). The first argument is $\half(\dd t^2+2\,\dd x^2-3\,\dd r^2)$, as computed above. (Here, the 1-form $\dd r$ depends on $\omega$, and should correctly be regarded as a section of the pullback of $T^*\R^4$ to $[\R^4;x^{-1}(0)]$ over the front face $r^{-1}(0)$.) Note that the test `function' (symmetric 2-tensor) $h^*$ here is smooth across $x=0$, and thus we may integrate $\half(\dd t^2+2\,\dd x^2-3\,\dd r^2)$ in $\omega$, which gives $2\pi(\dd t^2+\dd x^2)$ since $\vol(\Sph^2)=4\pi$, $\dd r=\sum_{j=1}^3\frac{x^j}{|x|}\dd x^j$, and $\int_{\Sph^2} (\dd r^2)\,\dd\omega=\sum_{j,k=1}^3(\int\frac{x^j x^k}{|x|^2}\,\dd\omega)\,\dd x^j\,\dd x^k=\frac{4\pi}{3}\sum_{j=1}^3(\dd x^j)^2=\frac{4\pi}{3}\dd x^2$.} Upon using~\eqref{EqMkMinkInner}, we finally obtain from~\eqref{EqAhKCokerKPf2} the result
  \[
    \ell_{b,\rm Kerr}(\dot\bhm,0)\bigl(\sfG_{\hat g_b}\delta_{g_b}^*\omega_0^*\bigr) = -4\pi\dot\bhm\la\slg,\slg\ra_{\slg^{-1}_2} = -8\pi\dot\bhm.
  \]

  \pfsubstep{(ii)}{Pairings with deformation tensors of (approximate) spatial rotations $\omega_{j k}^*$.} We now consider $\omega^*\in\mathspan\{\omega^*_{j k}\}$ and $h^*=\sfG_{\hat g_b}\delta_{\hat g_b}^*\omega^*$. We first evaluate~\eqref{EqAhCokerCalc3} for $\dot b=(0,\dot\bha)$. Since $g'_b(\dot b)\in\rho_\circ^2\CI$ and $\omega^*\in\rho_\circ^{-1}\CI$ near $\hat r=\infty$, this can again be written as a boundary pairing; thus, only the leading order term of $\omega^*$ at $\pa\hat X_b$ matters. Recalling~\eqref{EqAhKVect1form}, let $\fq\in\R^3$ be such that
  \begin{equation}
  \label{EqAhKCokerAxis}
    \omega^* - \hat r^2\vect(\fq)\in\cA^\delta(\hat X_b)+\sE'(\hat X_b^\circ).
  \end{equation}
  Then
  \begin{equation}
  \label{EqAhKCokerKPf3}
    \ell_{b,\rm Kerr}(0,\dot\bha)(h^*) = \big\la\pa_\lambda N(\rho_\circ^{-2}\wh{D_{\hat{\ubar g}}\Ric}(0),2)h_{(2)},h^*_{(-1)}\big\ra_{L^2(\pa\hat X_b)},
  \end{equation}
  where now
  \begin{alignat*}{2}
    h_{(2)}&:=\bigl(\rho_\circ^{-2}\hat g_b'(0,\dot\bha)\bigr)|_{\pa\hat X_b}, \\
    h^*_{(-1)}&:=\bigl(\rho_\circ\sfG_{\hat{\ubar g}}(\dd\hat t_*\otimes_s\omega^*)\bigr)|_{\pa\hat X_b}&&=\dd\hat x^1\otimes_s \hat r\vect(\fq) = \half(0,0,0,0,\vect(\fq),0)^T
  \end{alignat*}
  in the splitting~\eqref{EqMkBundleSplit}. One can compute $h_{(2)}$ using~\eqref{EqGKLot}. Since $h_{(-1)}^*$ is of vector type $1$, we may replace $h_{(2)}$ by its vector type $1$ part $h_{(2),\rmv 1}$. Now, $\bha\cdot\frac{\hat x}{|\hat x|}=\scal(\bha)$ is a function of scalar type $1$, and hence its square is a sum of scalar type $0$ and $2$ functions; therefore, the same is true for its linearization in $\bha$. Similarly, $(\bha\times\frac{\hat x}{|\hat x|})\cdot\dd\hat x=\hat r\vect(\bha)$ is a 1-form of vector type $1$, and hence (the linearization of) its symmetric square is a sum of symmetric 2-tensors of scalar types 0 and 2, as discussed around~\eqref{EqMkYV1V1}. Thus, only the final term in the large square brackets in~\eqref{EqGKLot} contributes to $h_{(2),\rmv 1}$; we therefore obtain
  \begin{align*}
    h_{(2),\rmv 1}&=-4\bhm\,\dd\hat t\otimes_s\Bigl(\Bigl(\dot\bha\times\frac{\hat x}{|\hat x|}\Bigr)\cdot\dd\hat x\Bigr) = -2\bhm(\dd\hat x^0+\dd\hat x^1)\otimes_s\hat r\vect(\dot\bha) \\
      &= -\bhm(0,0,\vect(\dot\bha),0,\vect(\dot\bha),0)^T.
  \end{align*}
  Corollary~\ref{CorMkDiffDRic} gives $\pa_\lambda N(\rho_\circ^{-2}\wh{D_{\hat{\ubar g}}\Ric}(0),2)h_{(2),\rmv 1}=\frac{3\bhm}{2}(0,0,\vect(\dot\bha),0,\vect(\dot\bha),0)^T$, and therefore\footnote{This is most easily computed by writing $\vect(\fq)=\slstar\sld\scal(\fq)$ and noting that $\la\vect(\dot\bha),\vect(\fq)\ra=\la\sld\scal(\dot\bha),\sld\scal(\fq)\ra=\la\sldelta\sld\scal(\dot\bha),\scal(\fq)\ra=2\int_{\Sph^2}(\dot\bha\cdot\hat x)(\fq\cdot\hat x)\,\dd\slg=\frac{8\pi}{3}(\dot\bha\cdot\fq)$.}
  \[
    \ell_{b,\rm Kerr}(0,\dot\bha)\bigl(\sfG_{\hat g_b}\delta_{\hat g_b}^*\omega^*\bigr) = \int_{\Sph^2} -4\Big\la\frac{3\bhm}{2}\vect(\dot\bha),\frac12\vect(\fq)\Big\ra_{\slg^{-1}}\,\dd\slg = -8\pi\bhm(\fq\cdot\dot\bha).
  \]

  The computations thus far already imply part~\eqref{ItAhKCokerKerrPK}. In order to finish the proof of part~\eqref{ItAhKCokerIndep} (and to finish the evaluation of the first two rows of the table in the statement of the Theorem), we still need to compute $\ell_{b,\rm Kerr}(\dot\bhm,0)(h^*)$ for $h^*=\sfG_{\hat g_b}\delta_{\hat g_b}^*\omega^*$, with $\omega^*$ as in~\eqref{EqAhKCokerAxis}. In~\eqref{EqAhCokerCalc3}, we may replace $\omega^*$ by $\hat r^2\vect(\fq)$; since $\hat g'_b(\dot\bhm,0)\in\rho_\circ\CI$ and $\hat r^2\vect(\fq)\in\rho_\circ^{-1}\CI$, the limit~\eqref{EqAhCokerCalc3} is sensitive also to subleading order terms. We may thus only replace $\sfG_{\hat g_b}$ by $\sfG_{\hat g_{b_0}}$ (see~\eqref{EqAhKSchwParam}), and $\wh{D_{\hat g_b}\Ric}(0)$ by $\wh{D_{\hat g_{b_0}}\Ric}(0)$. But then $\sfG_{\hat g_{b_0}}(\dd\hat t_*\otimes\hat r^2\vect(\fq))$ is of vector type $1$, and therefore we only need to compute the vector type $1$ part of $\hat g'_b(\dot\bhm,0)$ modulo $\rho_\circ^3\CI$, which in view of~\eqref{EqGKLot} is given by
  \[
    -4\hat r^{-2}\dot\bhm\,\dd\hat t\otimes_s\hat r\vect(\bha),
  \]
  and in particular has an additional order of vanishing relative to $\hat g'_b((\dot\bhm,0))$ itself. The calculation is now the same as the one above, with $\bhm,\dot\bha$ replaced by $\dot\bhm,\bha$. Therefore,
  \[
    \ell_{b,\rm Kerr}(\dot\bhm,0)\bigl(\sfG_{\hat g_b}\delta_{\hat g_b}^*\omega^*\bigr) = -8\pi\dot\bhm(\fq\cdot\bha).
  \]

  \pfstep{Part~\eqref{ItAhKCokerCOM}: computation of $\pi_{b,\rm COM}\circ\ell_{b,\rm COM}$.} We record the identity
  \begin{align*}
    \frac12\bigl[ [D_{\hat g_b}\Ric,\hat t_*],\hat t_*\bigr]\circ\delta_{\hat g_b}^* &= \frac12 D_{\hat g_b}\Ric\circ\hat t_*^2\circ\delta_{\hat g_b}^* - \hat t_* D_{\hat g_b}\Ric\circ\hat t_*\circ\delta_{\hat g_b}^* \\
      &=-\frac12 D_{\hat g_b}\Ric\circ[\delta_{\hat g_b}^*,\hat t_*^2] + \hat t_* D_{\hat g_b}\Ric\circ[\delta_{\hat g_b}^*,\hat t_*] \\
      &= -D_{\hat g_b}\Ric\circ\hat t_*\circ[\delta_{\hat g_b}^*,\hat t_*] + \hat t_* D_{\hat g_b}\Ric\circ[\delta_{\hat g_b}^*,\hat t_*] \\ 
      &= -\bigl[D_{\hat g_b}\Ric,\hat t_*\bigr]\circ[\delta_{\hat g_b}^*,\hat t_*].
  \end{align*}
  For $h_{b,\hat c}$ and $\breve h_{*,b,\hat c}$ as in Lemma~\ref{LemmaAhKBoostLot}, and for $h^*=\sfG_{\hat g_b}\delta_{\hat g_b}^*\omega^*$ with $\omega^*=\sum_{j=1}^3\fq_j\omega_j^*$, $\fq\in\R^3$, we then compute, using $[D_{\hat g_b}\Ric,\hat t_*]\ftrans(0)^*=-[(D_{\hat g_b}\Ric)^*,\hat t_*]\ftrans(0)=-\sfG_{\hat g_b}[D_{\hat g_b}\Ric,\hat t_*]\ftrans(0)\sfG_{\hat g_b}$:
  \begin{align*}
    &\ell_{b,\rm COM}(\hat c)(h^*) \\
    &\quad = \bigg\la \frac12\bigl[ [ D_{\hat g_b}\Ric,\hat t_*], \hat t_*\bigr]h_{b,\hat c} + [D_{\hat g_b}\Ric,\hat t_*]\ftrans(0)\breve h_{*,b,\hat c}, \sfG_{\hat g_b}\delta_{\hat g_b}^*\omega^* \bigg\ra \\
    &\quad= \bigg\la h_{b,\hat c}, \frac12\sfG_{\hat g_b}\bigl[ [D_{\hat g_b}\Ric,\hat t_*], \hat t_*\bigr]\delta_{\hat g_b}^*\omega^* \bigg\ra - \Big\la \breve h_{*,b,\hat c}, \sfG_{\hat g_b}[D_{\hat g_b}\Ric,\hat t_*]\ftrans(0)\delta_{\hat g_b}^*\omega^*\Big\ra \\
    &\quad= -\Big\la h_{b,\hat c},\sfG_{\hat g_b}[ D_{\hat g_b}\Ric,\hat t_* ]\ftrans(0)\bigl( [\delta_{\hat g_b}^*,\hat t_*]\omega^*\bigr) \Big\ra + \Big\la \breve h_{*,b,\hat c},\sfG_{\hat g_b} \wh{D_{\hat g_b}\Ric}(0)\bigl([\delta_{\hat g_b}^*,\hat t_*]\omega^*\bigr) \Big\ra \\
    &\quad= \lim_{\eps\searrow 0}\Big\la [D_{\hat g_b}\Ric,\hat t_*]\ftrans(0)(\chi_\eps h_{b,\hat c}) + \wh{D_{\hat g_b}\Ric}(0)(\chi_\eps\breve h_{*,b,\hat c}), \sfG_{\hat g_b}[\delta_{\hat g_b}^*,\hat t_*]\omega^* \Big\ra \\
    &\quad= \lim_{\eps\searrow 0}\Big\la \bigl[ [D_{\hat g_b}\Ric,\hat t_*]\ftrans(0),\chi_\eps \bigr]h_{b,\hat c} + [\wh{D_{\hat g_b}\Ric}(0),\chi_\eps]\breve h_{*,b,\hat c}, \sfG_{\hat g_b}(\dd\hat t_*\otimes_s\omega^*) \Big\ra.
  \end{align*}
  Here, we used that $0=D_{\hat g_b}\Ric(\hat t_* h_{b,\hat c}+\breve h_{*,b,\hat c})=[D_{\hat g_b}\Ric,\hat t_*]h_{b,\hat c}+D_{\hat g_b}\Ric(\breve h_{*,b,\hat c})$ to obtain the final expression. Since the second argument in this pairing lies in $\CI(\hat X_b;S^2\,\Ttsc^*_{\hat X_b}\hat M_b)$, whereas $h_{b,\hat c},\wh{D_{\hat g_b}\Ric}(0)$ have weight $\rho_\circ^2$ and $[D_{\hat g_b}\Ric,\hat t_*]\ftrans(0)$, $\breve h_{*,b,\hat c}$ have weight $\rho_\circ$, this is a boundary pairing which only depends on the leading order terms of the operators and tensors involved; thus, we can replace $\hat g_b$ and $\omega^*$ by $\hat{\ubar g}$ and $\dd(\hat r\scal(\fq))=(\half\scal(\fq),-\half\scal(\fq),\sld\scal(\fq))^T$ in the splitting~\eqref{EqMkBundleSplit}. By Lemma~\ref{LemmaAhKBoostLot} and using Lemma~\ref{LemmaMk}, we then have
  \begin{alignat}{2}
  \label{EqAhKCokerCOMh2}
    h_{(2)} &:= (\rho_\circ^{-2}h_{b,\hat c})|_{\pa\hat X_b} &&=  \frac{\bhm}{2} \bigl( -\scal(\hat c), 0, \sld\scal(\hat c), -\scal(\hat c), -\sld\scal(\hat c), 0 \bigr)^T, \\
    \breve h_{(1)} &:= (\rho_\circ^{-1}\breve h_{*,b,\hat c})|_{\pa\hat X_b} &&= \frac{\bhm}{2} \bigl( \scal(\hat c), 0, 2\,\sld\scal(\hat c), -3\scal(\hat c), 0, 0 \bigr)^T, \nonumber\\
    h^*_{(0)} &:= \bigl(\sfG_{\hat g_b}(\dd\hat t_*\otimes_s\omega^*)\bigr)|_{\pa\hat X_b} &&=\Bigl( 0, 0, 0, -\frac12\scal(\fq), \frac12\sld\scal(\fq), \frac12\scal(\fq)\slg \Bigr)^T. \nonumber
  \end{alignat}
  Using Corollary~\ref{CorMkDiffDRic} and Lemma~\ref{LemmaMkYId} as well as the fiber inner product~\eqref{EqMkMinkInner}, we thus find (again using $\la\scal(\hat c),\scal(\fq)\ra=\frac{4\pi}{3}(\hat c\cdot\fq)$)
  \begin{align*}
    &\ell_{b,\rm COM}(\hat c)\bigl(\sfG_{\hat g_b}\delta_{\hat g_b}^*\omega^*\bigr) \\
    &\quad = \Big\la\pa_\lambda N\bigl(\rho_\circ^{-1}[D_{\hat{\ubar g}}\Ric,\hat t_*\bigr]\ftrans(0),2\bigr)h_{(2)} + \pa_\lambda N\bigl(\rho_\circ^{-2}\wh{D_{\hat{\ubar g}}\Ric}(0),1\bigr)\breve h_{(1)},h^*_{(0)}\Big\ra_{L^2(\pa\hat X_b)} \\
    &\quad = \frac{\bhm}{2}\Big\la \Bigl(0,\scal(\hat c),-\frac12\sld\scal(\hat c),0,-\frac12\sld\scal(\hat c),0\Bigr)^T \\
    &\quad\hspace{5em} + \Bigl(2\scal(\hat c),-\frac12\scal(\hat c),-\sld\scal(\hat c),0,-2\,\sld\scal(\hat c),4\scal(\hat c)\slg\Bigr)^T, \\
    &\quad\hspace{15em} \Bigl(0,0,0,-\frac12\scal(\fq),\frac12\sld\scal(\fq),\frac12\scal(\fq)\slg\Bigr)^T \Big\ra_{L^2(\Sph^2)} \\
    &\quad = \frac{\bhm}{2}\Bigl(-4\la\scal(\hat c),\scal(\fq)\ra + 3\la\sld\scal(\hat c),\sld\scal(\hat q)\ra + 2\la\scal(\hat c)\slg,\scal(\fq)\slg\ra\Bigr) \\
    &\quad = 4\pi\bhm(\hat c\cdot\fq).
  \end{align*}

  \pfstep{Part~\eqref{ItAhCokerCOMKerr}.} When $\omega^*=\omega_0^*$, then the above calculation shows that $\ell_{b,\rm COM}(\sfG_{\hat g_b}\delta_{\hat g_b}^*\omega^*)=0$ since the leading order term of $\omega_0^*$ is of scalar type $0$ and thus orthogonal to scalar type $1$ tensors on $\pa\hat X_b$. An alternative argument proceeds as follows. Write
  \begin{equation}
  \label{EqAhKCokerwtcL}
    \wt\cL(V) := \cL_V(\cdot);
  \end{equation}
  acting on a fixed tensor, this is a first order differential operator acting on $V$. We have $2 h_{b,j}=\wt\cL(\pa_{\hat x^j})\hat g_b$ and $2\breve h_{*,b,j}=([\wt\cL,\hat t_*](\pa_{\hat x^j})+\wt\cL(\hat x^j\pa_{\hat t_*}))\hat g_b$ (so that $\wt\cL(\hat t_*\pa_{\hat x_j}+\hat x^j\pa_{\hat t_*})=2(\hat t_* h_{b,j}+\breve h_{*,b,j})$). We then compute
  \begin{align*}
    \wt\cL\Bigl(\frac{\hat t_*^2}{2}\pa_{\hat x^j} + \hat t_*\hat x^j\pa_{\hat t_*}\Bigr) &= \frac{\hat t_*^2}{2}\wt\cL(\pa_{\hat x^j}) + \hat t_*[\wt\cL,\hat t_*](\pa_{\hat x^j}) + \frac12\bigl[ [\wt\cL,\hat t_*], \hat t_*\bigr](\pa_{\hat x^j}) \\
      &\quad\qquad + \hat t_*\wt\cL(\hat x^j\pa_{\hat t_*}) + [\wt\cL,\hat t_*](\hat x^j\pa_{\hat t_*});
  \end{align*}
  the third term is a double commutator and thus vanishes. Applying this operator to $\hat g_b$ gives
  \[
    \cL_{\frac{\hat t_*^2}{2}\pa_{\hat x^j}+\hat t_*\hat x^j\pa_{\hat t_*}}\hat g_b = 2\Bigl(\frac{\hat t_*^2}{2}h_{b,j} + \hat t_*\breve h_{b,j}\Bigr) + [\wt\cL,\hat t_*](\hat x^j\pa_{\hat t_*})\hat g_b,
  \]
  and therefore
  \begin{align*}
    &D_{\hat g_b}\Ric\Bigl(\frac{\hat t_*^2}{2}h_{b,j} + \hat t_*\breve h_{b,j}\Bigr) = \wh{D_{\hat g_b}\Ric}(0)\mathring{h}_{b,j}, \\
    &\qquad \mathring{h}_{b,j}:=-\frac12[\wt\cL,\hat t_*](\hat x^j\pa_{\hat t_*})\hat g_b = -\hat g_b(\hat x^j\pa_{\hat t_*},-)\otimes_s\dd\hat t_* \in \rho_\circ^{-1}\CI(\hat X_b;S^2\,\Ttsc^*_{\hat X_b}\hat M_b).
  \end{align*}
  Here, we use that $[\wt\cL,f](V)h=2 h(V,-)\otimes_s\dd f$ when $f$ is a smooth function, $V$ a vector field, and $h$ a symmetric 2-tensor; and since $\mathring{h}_{b,j}$ is stationary, the action of $D_{\hat g_b}\Ric$ on it is given by $\wh{D_{\hat g_b}\Ric}(0)\mathring{h}_{b,j}$ indeed. Thus for $h^*\in\ker\wh{D_{\hat g_b}\Ric}(0)^*$ which near $\pa\hat X_b$ lie in $\cA^{2+\delta}$ where $\delta>0$, integration by parts shows that
  \begin{equation}
  \label{EqAhKCokerInnerRing}
    \big\la \wh{D_{\hat g_b}\Ric}(0)\mathring{h}_{b,j},h^*\big\ra_{L^2(\hat X_b)}
  \end{equation}
  vanishes. This applies in particular to $h^*=\sfG_{\hat g_b}\delta_{\hat g_b}\omega^*$ when $\omega^*=\omega_0^*$ and also when $\omega^*$ has leading order term $\hat r^2\vect(\fq)$ where $\fq\in\R^3$, with the restriction $\fq\in\R\bha$ when $\bha\neq 0$, since in these cases $h^*$ is compactly supported in $\hat X_b^\circ$; thus $\ell_{b,\rm COM}(h^*)=0$.

  When $h^*$ has a $\rho_\circ^2$ leading order term, then the inner product~\eqref{EqAhKCokerInnerRing} can instead be rewritten as a boundary pairing in the usual manner; since the leading order term of $\mathring{h}_{b,j}$ is of scalar type $1$, this pairing only involves the scalar type $1$ part of the $\rho_\circ^2$ leading order term of $h^*$. (When $h^*\in\cK_{b,\rm COM}^*$, one can use this to compute $\pi_{b,\rm COM}\circ\ell_{b,\rm COM}$ again.) For $h^*=\sfG_{\hat g_b}\delta_{\hat g_b}^*\omega^*\in\cK_{b,\rm Kerr}^*$ where $\omega^*$ has leading order term $\hat r^2\vect(\fq)$, we claim that the scalar type $1$ part of $(\rho_\circ^{-2}h^*)_{\pa\hat X_b}$ vanishes. Since $h^*$ now depends linearly on $\fq$, we only need to consider the case $\fq\perp\bha\neq 0$ (the cases $\bha=0$ or $\bha\neq 0$, $\fq\in\R\bha$ having been discussed above). In this case, $\delta_{\hat g_b}^*(\hat r^2\vect(\fq))$ is the Lie derivative of $\hat g_b$ along a rotation vector field along an axis perpendicular to $\bha$, and thus it is a linearized Kerr metric $\hat g_b'(0,\dot\bha)$ for some $\dot\bha$. (See also Lemma~\ref{LemmaFhM1AxisLie} below.) But the $\hat r^{-2}$ part of~\eqref{EqGKLot} is the sum of scalar type $0$ and $2$ and vector type $1$ tensors; its scalar type $1$ part vanishes. Thus,~\eqref{EqAhKCokerInnerRing} vanishes also in this case.

  \pfstep{Part~\eqref{ItAhKCokerNoC}.} Repeating the calculations leading to~\eqref{EqAhCokerCalc3} for $h_{b,\hat c}\in\rho_\circ^2\CI$ in place of $g'_b(\dot b)$, the conclusion is clear when $\omega^*\in\CI$ near infinity (as in this case one can drop the regularizer $\chi_\eps$ in~\eqref{EqAhCokerCalc3}). When $\omega^*$ is an asymptotic rotation, then the limit in~\eqref{EqAhCokerCalc3} is a boundary pairing which evaluates to $0$ since the $\rho_\circ^2$ leading order term of $h_{b,\hat c}$ is of scalar type $1$, whereas the $\rho_\circ^{-1}$ leading order term of $\omega^*$ is of vector type $1$.

  \pfstep{Part~\eqref{ItAhCokerTstar}.} Recall that $h^*$, near $\pa\hat X_b$, is conormal with weight $2$ at $\pa\hat X_b$. If we set $T':=\hat t_*-\hat t'\in\cA^{-1}$, then $T'\hat g'_b(\dot b)\in\cA^0$, and hence we can integrate by parts to conclude that
  \[
    \la D_{\hat g_b}\Ric(T'\hat g'_b(\dot b)),h^*\ra = \big\la \wh{D_{\hat g_b}\Ric}(0)(T'\hat g'_b(\dot b)),h^*\big\ra = \big\la T'\hat g'_b(\dot b), \wh{D_{\hat g_b}\Ric}(0)^*h^* \big\ra = 0.
  \]
  Similarly, setting $\breve h'_{b,\hat c}=T'h_{b,\hat c}+\breve h_{*,b,\hat c}$ and using $\hat t_*=\hat t'+T'$, we have
  \[
    \frac{\hat t_*^2}{2}h_{b,\hat c}+\hat t_*\breve h_{*,b,\hat c} = \frac{\hat t'{}^2}{2}h_{b,\hat c} + \hat t'\breve h'_{b,\hat c} + \breve h'',\qquad
    \breve h'' := \frac{T'{}^2}{2}h_{b,\hat c} + T'\breve h_{*,b,\hat c} \in \cA^0.
  \]
  Thus, $\la D_{\hat g_b}\Ric(\breve h''),h^*\ra=0$ as before. This completes the proof.
\end{proof}

\subsection{Polyhomogeneous solutions of the zero energy problem}
\label{SsAhPhg}

The necessary condition for the solvability of the zero energy problem for the linearized Ricci operator of Lemma~\ref{LemmaAh0Nec} is also sufficient:

\begin{thm}[Polyhomogeneous solutions of the zero energy problem]
\label{ThmAhPhg}
  Let $\omega_0^*,\omega_j^*,\omega_{j k}^*$ be stationary 1-forms on $\hat M_b$ as in Theorem~\usref{ThmAhKCoker}, i.e.\ their supports are disjoint from $\hat r^{-1}(\bhm)\subset\hat X_b$ and $\omega_\mu^*-\omega_\mu,\omega_{j k}^*-\omega_{j k}\in\cA^\delta(\hat X_b;\Ttsc^*_{\hat X_b}\hat M_b)+\cE'(\hat X_b^\circ;T_{\hat X_b^\circ}\hat M_b^\circ)$ in the notation of Definition~\usref{DefAhK}. Let $\cF\subset\C\times\N_0$ be an index set with $\Re\cF>1$. Set
  \[
    \cE=\bigl\{ (z+j-2,l) \colon (z,k)\in\cF,\ j\in\N_0,\ l\leq k+j+1 \bigr\}.
  \]
  Suppose $f\in\cA_\phg^\cF(\hat X_b;S^2\,\Ttsc^*_{\hat X_b}\hat M_b)$ satisfies $\wh{\delta_{\hat g_b}\sfG_{\hat g_b}}(0)f=0$ and
  \begin{equation}
  \label{EqAhPhgCoker}
    \la f, \wh{\sfG_{\hat g_b}\delta_{\hat g_b}^*}(0)\omega^*\ra_{L^2(\hat X_b)} = 0,\qquad \omega^*\in\{\omega_\mu^*,\ \omega_{j k}^*\colon 0\leq \mu\leq 3,\ 1\leq j\neq k\leq 3 \}.
  \end{equation}
  Then there exists a solution $h\in\cA_\phg^\cE(\hat X_b;S^2\,\Ttsc^*_{\hat X_b}\hat M_b)$ of
  \begin{equation}
  \label{EqAhPhg}
    \wh{D_{\hat g_b}\Ric}(0)(h)=f.
  \end{equation}
\end{thm}

In the notation of Theorem~\ref{ThmAhKCoker}, condition~\eqref{EqAhPhgCoker} is equivalent to the requirement that
\begin{equation}
\label{EqAhPhgEquiv}
  \la f,-\ra_{L^2(\hat X_b)} \in (\cK_{b,\rm tot}^*)^*
\end{equation}
be equal to $0$, or equivalently $\la f,-\ra\in(\cK_{b,\rm COM}^*)^*=0$ and $\la f,-\ra_{L^2(\hat X_b)} \in (\cK_{b,\rm Kerr}^*)^*=0$.

\begin{rmk}[Non-uniqueness]
\label{RmkAhPhgNonuniq}
  One can of course add to $h$ any pure gauge solution $\wh{\delta_{\hat g_b}^*}(0)\omega$ without invalidating~\eqref{EqAhPhg}. In our gluing construction, we will only exploit the existence of a 7-dimensional kernel spanned by
  \[
    \hat g'_b(\dot b),\quad \dot b=(\dot\bhm,\dot\bha)\in\R\times\R^3; \qquad
    h_{b,\hat c},\quad \hat c\in\R^3.
  \]
\end{rmk}

The proof of Theorem~\ref{ThmAhPhg} proceeds in two steps: the construction of a formal solution near infinity, and solving away the remaining error.

\begin{lemma}[Step 1: formal solution at infinity]
\label{LemmaAhPhgFormal}
  Under the assumptions and using the notation of Theorem~\usref{ThmAhPhg}, there exists $h\in\cA_\phg^\cE(\hat X_b;S^2\,\Ttsc^*_{\hat X_b}\hat M_b)$ so that
  \[
    \wh{D_{\hat g_b}\Ric}(0)h = f + f_\flat,\qquad f_\flat\in\CIdot(\hat X_b;S^2\,\Ttsc^*_{\hat X_b}\hat M_b).
  \]
  If $f$ depends smoothly on a parameter lying in an open subset of $\R^k$, $k\in\N$,\footnote{That is, each term in its polyhomogeneous expansion at $\pa\hat X_b$ is a smooth function of the product of $\pa\hat X_b$ and the parameter space.} then we can find $h$ depending smoothly on this parameter as well.
\end{lemma}
\begin{proof}
  Fix $\chi_\circ\in\CI(\hat X_b)$ to be equal to $1$ near $\pa\hat X_b$ and $0$ near $\hat r=\bhm$. Let $(2-z,k)\in\cF$ be such that $\Re(2-z)=\min\Re\cF$ and $(2-z,k+1)\notin\cF$. (Since $\Re(2-z)>1$, we have $\Re z<1$.) Then for some $f_j\in\CI(\pa\hat X_b;S^2\,\Ttsc^*_{\pa\hat X_b}\hat M_b)$, $j=0,\ldots,k$,
  \[
    f - \chi_\circ\sum_{j=0}^k \hat r^{z-2}(\log\hat r)^j f_j(\omega) \in \cA_\phg^{\cF'}(\hat X_b),\qquad \cF':=\cF\setminus\{(2-z,j)\colon j=0,\ldots,k\},
  \]
  where $\hat r=|\hat x|$ and $\omega=\frac{\hat x}{|\hat x|}$. Applying the operator
  \[
    \wh{\delta_{\hat g_b}\sfG_{\hat g_b}}(0)\equiv\wh{\delta_{\hat{\ubar g}}\sfG_{\hat{\ubar g}}}(0)\bmod\la\hat r\ra^{-2}\Diffb^1(\hat X_b;S^2\,\Ttsc^*_{\hat X_b}\hat M_b,\Ttsc^*_{\hat X_b}\hat M_b)
  \]
  to this equation, and identifying $f_j$ with an element of $\CI(\ff;\ubar\upbeta^*S^2 T^*_{(0,0)}\R^4)$ in the notation of Theorem~\ref{ThmRicSolv} (where we recall that $\ff=\Sph^2$ is the front face of $[\R^3;\{0\}]$), we obtain
  \[
    \wh{\ubar\delta\ubar\sfG}(0)\ubar f = 0,\qquad \ubar f = \ubar f(r,\omega) := \sum_{j=0}^k r^{z-2}(\log r)^j f_j(\omega).
  \]

  Consider assumption~\eqref{EqAhPhgCoker}. Let $\chi\in\CI([0,\infty))$ be equal to $0$ on $[0,1]$ and equal to $1$ near $\infty$, and set $\chi_\eps=\chi(\frac{\hat r^{-1}}{\eps})$ for $0<\eps\ll 1$. Then
  \[
    \big\la f,\wh{\sfG_{\hat g_b}\delta_{\hat g_b}^*}(0)\omega^* \big\ra = \lim_{\eps\searrow 0} \big\la \bigl[ \wh{\delta_{\hat g_b}\sfG_{\hat g_b}}(0), \chi_\eps \bigr]f, \omega^* \big\ra.
  \]
  When $z=0$ and $\omega^*=\omega_\mu^*$, or $z=-1$ and $\omega^*=\omega_{j k}^*$, we get a nontrivial boundary pairing, and indeed we may replace $\hat g_b$ by $\hat{\ubar g}$, and $f$ and $\omega^*$ by their respective leading order terms $\sum_{j=0}^k \hat r^{z-2}(\log\hat r)^j f_j$ and $\omega_\mu$ or $\omega_{j k}$. We thus deduce that when $z=-1$, resp.\ $z=0$, then~\eqref{EqRicSolvv1Pair}, resp.\ \eqref{EqRicSolvs0Pair}--\eqref{EqRicSolvs1Pair} holds for $\ubar f$ in place of $f$. We now distinguish four cases:

  \begin{enumerate}
  \item If $z\notin\Z$, then Theorem~\ref{ThmRicSolv} produces $h_0,\ldots,h_k\in\CI(\ff;\ubar\upbeta^*S^2 T^*_{(0,0)}\R^4)$ so that
    \[
      \wh{D_{\ubar g}\Ric}(0)\Biggl(\sum_{j=0}^k r^z(\log r)^j h_j\Biggr) = \ubar f,
    \]
    and therefore
    \begin{equation}
    \label{EqAhPhgFormalStep}
      f' := f - \wh{D_{g_b}\Ric}(0)\Biggl(\chi_\circ\sum_{j=0}^k\hat r^z(\log\hat r)^j h_j\Biggr) \in \cA_\phg^{\cF'}(\hat X_b).
    \end{equation}
    Furthermore, $f'$ satisfies the same assumptions as $f$, with~\eqref{EqAhPhgCoker} following via integration by parts as in~\eqref{EqAhKIBP}.
  \item When $z=0$, then as argued above, Theorem~\ref{ThmRicSolv}\eqref{ItRicSolv0} applies and produces $h_0$, $\ldots$, $h_{k+1}\in\CI(\ff;\ubar\upbeta^*S^2 T^*_{(0,0)}\R^4)$ with $\wh{D_{\ubar g}\Ric}(0)\sum_{j=0}^{k+1}r^z(\log r)^j h_j=\ubar f$; analogously to~\eqref{EqAhPhgFormalStep}, we can use this to solve away the term in the asymptotic expansion of $f$ corresponding to $(2-z,f)\in\cF$.
  \item The case $z=-1$ is analogous to the case $z=0$: we now use Theorem~\ref{ThmRicSolv}\eqref{ItRicSolvm1}.
  \item Finally, when $z\in\Z$, $z\leq -2$, we can solve away the term in the asymptotic expansion of $f$ corresponding to $(2-z,f)\in\cF$ using Theorem~\ref{ThmRicSolv}\eqref{ItRicSolvl}.
  \end{enumerate}

  Repeating this procedure a finite number of times, we can, for any $N$, find $h_{(N)}\in\cA_\phg^\cE(\hat X_b;S^2\,\Ttsc^*_{\hat X_b}\hat M_b)$ so that $h_{(N+1)}-h_{(N)}\in\cA^{N-1}$ and
  \[
    \wh{D_{\hat g_b}\Ric}(0)h_{(N)}=f+f_{\flat,(N)},\qquad
    f_{\flat,(N)}\in(\cA_\phg^\cF\cap\cA^N)(\hat X_b;S^2\,\Ttsc^*_{\hat X_b}\hat M_b).
  \]
  Taking $h$ to be an asymptotic sum $h\sim h_{(1)}+\sum_{N\in\N} (h_{(N+1)}-h_{(N)})$ near $\pa\hat X_b$ finishes the proof for a single $f$.

  When $f=f(q)$ depends smoothly on a parameter $q\in U\subset\R^k$, then the above construction produces $h_{(N)}$ which depend smoothly on $q$ (due to the linear and continuous dependence of the solutions in Theorem~\ref{ThmRicSolv}, see also the analogous arguments around~\eqref{EqAcBdyFormalC1}). Asymptotically summing $h_{(1)}+\sum_{N\in\N}(h_{(N+1)}-h_{(N)})$ on $U\times\hat X_b$ produces the desired $h=h(q)$.
\end{proof}

\begin{lemma}[Step 2: polyhomogeneous solutions for Schwartz forcing]
\label{LemmaAhPhgSchwartz}
  For all $f_\flat\in\CIdot(\hat X_b;S^2\,\Ttsc^*_{\hat X_b}\hat M_b)$ with $\wh{\delta_{\hat g_b}\sfG_{\hat g_b}}(0)f_\flat=0$, there exists $h_\flat\in\rho_\circ\CI(\hat X_b;S^2\,\Ttsc^*_{\hat X_b}\hat M_b)$ (where $\rho_\circ=\la\hat x\ra^{-1}$) solving $\wh{D_{\hat g_b}\Ric}(0)h_\flat=f_\flat$. If $f_\flat$ depends smoothly on a finite-dimensional parameter, then we can find $h_\flat$ depending smoothly on this parameter as well.
\end{lemma}

Note that for $f_\flat$ as in this Lemma, condition~\eqref{EqAhPhgCoker} follows from $\wh{\delta_{\hat g_b}\sfG_{\hat g_b}}(0)f_\flat=0$. We will deduce Lemma~\ref{LemmaAhPhgSchwartz} from the existence of a solution $h_\flat$ which is conormal at $\hat X_b$ and decays at \emph{some} positive rate; the latter is a nontrivial statement about the perturbation theory of subextremal Kerr black holes.

\begin{thm}[Solvability on spaces of conormal tensors]
\label{ThmAhPhgSolv}
  Recall $b=(\bhm,\bha)$. There exists $\eta>0$ so that the following statements hold.
  \begin{enumerate}
  \item\label{ItAhPhgSolvRic}{\rm (Linearized Einstein equation.)} If\footnote{That is, $f$ is a stationary symmetric 2-tensor on the complement of the spatial $\bhm$-ball in $\R_{\hat t}\times\R_{\hat x}^3$ whose coefficients with respect to the standard coordinate differentials are rapidly vanishing as $|\hat x|\to\infty$.} $f\in\CIdot(\hat X_b;S^2\,\Ttsc^*_{\hat X_b}\hat M_b)$ satisfies $\delta_{\hat g_b}\sfG_{\hat g_b}f=0$, then there exists a solution $h\in\cA^\eta(\hat X_b;S^2\,\Ttsc^*_{\hat X_b}\hat M_b)$, of\footnote{In other words, $h$ is a stationary solution of $D_{\hat g_b}\Ric(h)=f$.} $\wh{D_{\hat g_b}\Ric}(0)(h)=f$; and $h$ can be taken to depend continuously and linearly on $f$.
  \item\label{ItAhPhgSolvGauge}{\rm (Gauge potential wave operator.)} There exists\footnote{This means that $E^\Ups$, expressed in the frame $\dd\hat t$, $\dd\hat x$, is a $10\times 4$ matrix whose entries are smooth functions of $|\hat x|^{-1}$, $\frac{\hat x}{|\hat x|}$ which are independent of $\hat t$ and vanish at $|\hat x|^{-1}=0$.}
  \[
    E^\Ups\in\rho_\circ\CI\bigl(\hat X_b;\Hom(S^2\,\Ttsc^*_{\hat X_b}\hat M_b,\Ttsc^*_{\hat X_b}\hat M_b)\bigr)
  \]
  so that for all $\theta\in\cA^{1+\eta}(\hat X_b;\Ttsc^*_{\hat X_b}\hat M_b)$ there exists $\omega\in\cA^{-1+\eta'}(\hat X_b;\Ttsc^*_{\hat X_b}\hat M_b)$, $\eta'\in(0,\eta)$, depending linearly and continuously on $\theta$, with
  \[
    (\delta_{\hat g_b}+E^\Ups)\sfG_{\hat g_b}\delta_{\hat g_b}^*\omega=\theta.
  \]
  \end{enumerate}
\end{thm}
\begin{proof}
  Define weighted b-Sobolev spaces $\Hb^{s,\alpha}(\ol{\R^3})=\la\hat x\ra^{-\alpha}\Hb^s(\ol{\R^3})$ on $\ol{\R_{\hat x}^3}$ using the Euclidean volume density $|\dd\hat x|$, and write $\Hbext^{s,\alpha}(\hat X_b)$ for the space of restrictions of elements of $\Hb^{s,\alpha}(\ol{\R^3})$ to $\hat X_b$. We have $\Hbext^{\infty,-\frac32+\alpha}(\hat X_b)\subset\cA^\alpha(\hat X_b)\subset\bigcap_{\alpha'<\alpha}\Hbext^{\infty,-\frac32+\alpha'}(\hat X_b)$.

  Part~\eqref{ItAhPhgSolvGauge} is stated in \cite[Remark~10.14]{HaefnerHintzVasyKerr} in the slowly rotating case $|\frac{\bha}{\bhm}|\ll 1$; using \cite[Theorem~5.1]{AnderssonHaefnerWhitingMode}, the same argument applies in the general subextremal range. In brief, for $E^\Ups=0$, the zero energy operator of the tensor wave operator $2\delta_{\hat g_b}\sfG_{\hat g_b}\delta_{\hat g_b}^*=\Box_{\hat g_b}$ acting on 1-forms is Fredholm of index $0$ as a map
  \begin{equation}
  \label{EqAhPhgHypoBox1}
  \begin{split}
    \wh{\Box_{\hat g_b}}(0) &\colon \bigl\{ \omega\in\Hbext^{s,\ell}(\hat X_b;\Ttsc^*_{\hat X_b}\hat M_b) \colon \wh{\Box_{\hat g_b}}(0)\omega\in\Hbext^{s-1,\ell+2}(\hat X_b;\Ttsc^*_{\hat X_b}\hat M_b) \bigr\} \\
      &\quad\hspace{15em} \to \Hbext^{s-1,\ell+2}(\hat X_b;\Ttsc^*_{\hat X_b}\hat M_b)
  \end{split}
  \end{equation}
  for all sufficiently large $s$ and for $\ell\in(-\frac32,-\frac12)$, with 1-dimensional kernel and cokernel spanned by explicit 1-forms $\omega_b$, $\omega_b^*=\delta(\hat r-\hat r_{\bhm,\bha})\,\dd\hat r$, respectively; see \cite[Theorem~7.1]{HaefnerHintzVasyKerr} for the very slowly rotating case $|\frac{\bha}{\bhm}|\ll 1$ and \cite[Theorem~5.1]{AnderssonHaefnerWhitingMode} for the full subextremal range. Since, by inspection, $\omega_b$ is not a linear combination of the 1-forms dual to time translations and rotations (around the axis of symmetry when $\bha\neq 0$), we have $\omega_b\notin\ker\delta_{\hat g_b}^*$. Therefore, we can select $E_1^\Ups$ so that $\la E_1^\Ups\sfG_{\hat g_b}\delta_{\hat g_b}^*\omega_b,\omega_b^*\ra_{L^2(\hat X_b)}\neq 0$. For the choice $E^\Ups:=c E_1^\Ups$ then, with $c\neq 0$ sufficiently small, the operator $((\delta_{\hat g_b}+E^\Ups)\sfG_{\hat g_b}\delta_{\hat g_b}^*)\ftrans(0)$ is invertible between the spaces in~\eqref{EqAhPhgHypoBox1}. By elliptic b-theory near $\pa\hat X_b$, this operator remains surjective for all $\ell\leq-\frac32$ which avoid a discrete set.

  We now turn to part~\eqref{ItAhPhgSolvRic}. Define the gauge-fixed linearized Einstein operator
  \[
    L_b:=2(D_{\hat g_b}\Ric+\delta_{\hat g_b}^*\delta_{\hat g_b}\sfG_{\hat g_b}),
  \]
  which is equal to $\Box_{\hat g_b}$ modulo lower order terms. By \cite[Theorem~7.5]{AnderssonHaefnerWhitingMode} (see \cite[Proposition~9.1]{HaefnerHintzVasyKerr} and the description \cite[(9.3a)--(9.3c)]{HaefnerHintzVasyKerr} for the slowly rotating case), all elements of the cokernel of $\wh{L_b}(0)$ are of the form $\sfG_{\hat g_b}\delta_{\hat g_b}^*\omega^*$ where $\omega^*$ is stationary, conormal of class $\cA^{-1}$ near $\pa\hat X_b$, and vanishes in the black hole interior (so in particular near $\hat r=\bhm$). Since $f$ is rapidly vanishing, we have
  \[
    \big\la f,\wh{\sfG_{\hat g_b}\delta_{\hat g_b}^*}(0)\omega^*\big\ra = \big\la \wh{\delta_{\hat g_b}\sfG_{\hat g_b}}(0)f,\omega^*\big\ra = 0
  \]
  via integration by parts. Therefore, there exists
  \[
    h\in\bigcap_{s\in\R,\ \ell\in(-\frac32,-\frac12)}\Hbext^{s,\ell}(\hat X_b;S^2\,\Ttsc^*_{\hat X_b}\hat M_b) \subset \bigcap_{\alpha<1}\cA^\alpha(\hat X_b;S^2\,\Ttsc^*_{\hat X_b}\hat M_b)
  \]
  solving $\wh{L_b}(0)h=f$. Applying $\wh{\delta_{\hat g_b}\sfG_{\hat g_b}}(0)$ to this implies
  \[
    \wh{\Box_{\hat g_b}}(0)\theta=0,\qquad \theta := \wh{\delta_{\hat g_b}\sfG_{\hat g_b}}(0)h \in \bigcap_{\alpha<1}\cA^{\alpha+1}(\hat X_b;\Ttsc^*_{\hat X_b}\hat M_b).
  \]
  The description of the kernel of $\wh{\Box_{\hat g_b}}(0)$ in \cite[Theorem~5.1]{AnderssonHaefnerWhitingMode}---concretely, the fact that $\omega_b$, having a nonzero $\hat r^{-1}$ leading order term at $\hat r=\infty$, does not lie in $\cA^{\alpha+1}(\hat X_b;\Ttsc^*_{\hat X_b}\hat M_b)$ for $\alpha>0$---implies that $\theta=0$, and therefore we in fact have $\wh{D_{\hat g_b}\Ric}(0)h=f$, as desired.
\end{proof}

\begin{proof}[Proof of Lemma~\usref{LemmaAhPhgSchwartz}]
  Let $h\in\cA^\eta(\hat X_b;S^2\,\Ttsc^*_{\hat X_b}\hat M_b)$, $\eta>0$, be the solution provided by Theorem~\ref{ThmAhPhgSolv}\eqref{ItAhPhgSolvRic}. We first find $\omega\in\cA^{-1+\eta/2}(\hat X_b;\Ttsc^*_{\hat X_b}\hat M_b)$ to arrange the gauge condition
  \[
    h + \delta_{\hat g_b}^*\omega \in \ker\bigl((\delta_{\hat g_b}+E^\Ups)\sfG_{\hat g_b}\bigr);
  \]
  this is possible by Theorem~\ref{ThmAhPhgSolv}\eqref{ItAhPhgSolvGauge}. Replacing $h$ by $h+\delta_{\hat g_b}^*\omega$, we therefore have
  \[
    \wh{D_{\hat g_b}\Ric}(0)h=f_\flat,\qquad
    (\delta_{\hat g_b}+E^\Ups)\sfG_{\hat g_b}h=0,
  \]
  and therefore
  \[
    \hat L(0)h=0,\qquad L := 2\bigl(D_{\hat g_b}\Ric + \delta_{\hat g_b}^*(\delta_{\hat g_b}+E^\Ups)\sfG_{\hat g_b}\bigr).
  \]
  But $L\in\rho_\circ^2\Difftb^2(\hat M_b;S^2\,\Ttsc^*\hat M_b)$ is equal to $\Box_{\hat g_b}$ modulo $\rho_\circ^3\Difftb^2$, and therefore $\hat L(0)\in\rho_\circ^2\Diffb^2(\hat X_b;S^2\,\Ttsc^*_{\hat X_b}\hat M_b)$ is elliptic near $\pa\hat X_b$ as a b-differential operator. By standard elliptic b-theory (locally near $\pa\hat X_b$), $h$ is necessarily polyhomogeneous with some index set $\cE'\subset\C\times\N_0$ satisfying $\Re\cE'>0$.

  We now show how to add to $h\in\cA_\phg^{\cE'}(\hat X_b;S^2\,\Ttsc^*_{\hat X_b}\hat M_b)$ a further pure gauge solution $\delta_{\hat g_b}^*\omega$, with $\omega=o(|\hat x|)$ stationary and polyhomogeneous, so as to obtain the desired solution $h_\flat=h+\delta_{\hat g_b}^*\omega\in\rho_\circ\CI(\hat X_b;S^2\,\Ttsc^*_{\hat X_b}\hat M_b)$. To this end, consider $(z,k)\in\cE'$ with $\Re z=\min\Re\cE'$ and $(z,k+1)\notin\cE'$; then the term
  \[
    \sum_{j=0}^k \hat r^{-z}(\log\hat r)^j h_j
  \]
  in the polyhomogeneous expansion of $h$ at $\pa\hat X_b=\{\hat r^{-1}=0\}$ lies in $\ker\wh{D_{\hat{\ubar g}}\Ric}(0)$. Theorem~\ref{ThmRicUniq} thus produces a stationary 1-form
  \[
    \omega_{(z)}=\chi_\circ\sum_{j=0}^{k+k'}\hat r^{-z}(\log\hat r)^j \omega_j,\qquad k'\in\{0,1\},\quad \omega_j\in\CI(\pa\hat X_b;\Ttsc^*_{\pa\hat X_b}\hat M_b),
  \]
  so that the only term in the polyhomogeneous expansion of $h-\delta_{\hat g_b}^*\omega_{(z)}$ involving $\hat r^{-z}$ is $\hat r^{-z}h'$ where, in fact, $h'=0$ when $z\notin\Z$. An asymptotic summation argument as in the proof of Lemma~\ref{LemmaAhPhgFormal} produces the desired $\omega$.
\end{proof}

\begin{proof}[Proof of Theorem~\usref{ThmAhPhg}]
  Apply Lemma~\ref{LemmaAhPhgFormal} to $f$ to get a formal solution $h_\sharp$. Then apply Lemma~\ref{LemmaAhPhgSchwartz} to $-f_\flat$ to produce $h_\flat$. The desired solution is $h=h_\sharp+h_\flat$.
\end{proof}

\subsection{Necessary conditions for gluing: II. A result of D'Eath revisited}
\label{SsAhNec}

In this supplementary section, we show how the results of~\S\ref{SsAh0} imply additional necessary conditions, beyond those of~\S\ref{SsAcGW}, for a total spacetime family with Kerr models along $\cC$ to solve the Einstein vacuum equations.

\begin{prop}[Necessary conditions on $\cC$ and the Kerr parameters]
\label{PropAhNec}
  Suppose that $\wt g$ is a $(\hat\cE,\cE)$-smooth total family (relative to $(M,g,\cC)$) with $\Re\hat\cE>1$, and the $\hat M_p$-model of $\wt g$ is a Kerr metric with parameters $b=(\bhm,\bha)\in\CI(I;\R\times\R^3)$ so that $|\bha(p)|<\bhm(p)$ for all $p\in\cC$, where $\bha$ is defined relative to a fixed choice of Fermi normal coordinates along $\cC$. If $\Ric(\wt g)-\Lambda\wt g$ vanishes to leading and subleading order at $\hat M$ (i.e.\ has weight $>-1$ at $\hat M$ as a conormal section of $S^2\wt T^*\wt M$), then
  \[
    \text{$\cC$ is a geodesic, $\bhm$ is constant, and $\bha$ is parallel along $\cC$ (i.e.\ constant).}
  \]
\end{prop}

This explains why in Lemma~\ref{LemmaGKNaive} we do not need to consider maps of the form $\Phi_p(\hat t,\hat x)=(\hat t,A(p)\hat x)$ where $A(p)\in O(3)$ is not the identity. The essence of the statement of Proposition~\ref{PropAhNec} appears already in work of D'Eath \cite{DEathSmallBHDynamics}.

\begin{proof}[Proof of Proposition~\usref{PropAhNec}]
  Let $\delta>0$ be such that $\min\Re\hat\cE>1+\delta$ and $\min\Re\cE>\delta$. We denote the Fermi normal coordinates around $\cC$ by $(t,x)$, and denote by $I\ni t\mapsto c(t)\in\cC$ an arc-length parameterization of $\cC$. Regarding a neighborhood of $\hat M\subset\wt M$ as a subset of the lift of $\eps=0$ in $[[0,1)\times\hat M;\{0\}\times\pa\hat M]$ via Lemma~\ref{LemmaGRel}, and using the fiber coordinates $\hat t=\dd t(-)$ and $\hat x=\dd x(-)$ on $T_\cC M$, we thus have
  \begin{equation}
  \label{EqAhNecMetric}
    (\sfe\wt g)(\eps,t,\hat x) = \hat g_{b(t)}(\hat x) + \eps\hat h(t,\hat x) + \cO(\eps^{1+\delta})
  \end{equation}
  for bounded $\hat x$, where we regard $\hat g_{b(t)}$ and $\hat h$ as symmetric 2-tensors in $\dd\hat t$, $\dd\hat x$, and the $\cO(\eps^{1+\delta})$ error term is such a symmetric 2-tensor whose coefficients are conormal on $\wt M$ with weight $1+\delta$ at $\hat M$ (and with weight $0$ at $M_\circ$); here, writing $\rho_\circ=\la\hat x\ra^{-1}$, the subleading term is
  \begin{equation}
  \label{EqAhNechat}
    \hat h\in \cA_{\phg,\rm I}^{(\N_0\cup\cE)-1}(\breve T_\cC M;S^2\,{}^{\tscop,\vee}T^*\breve T_\cC M)
  \end{equation}
  in the notation of Lemma~\ref{LemmaGffStatExt}, i.e.\ its coefficients lie in $\CI(I;\cA_{\phg,\rm I}^{(\N_0\cup\cE)-1}(\ol{\R^3_{\hat x}}))$. Moreover, since $\frac{\rho_\circ}{\eps}=\hat\rho^{-1}$ on $\wt M$, with $\hat\rho|_{M_\circ}=r$, we have
  \begin{equation}
  \label{EqAhNechm1}
    \hat h_{(-1)} := (\rho_\circ\hat h)|_{\pa\hat M}=\sfe(\hat\rho^{-1}(\wt g-\hat g_{b(t)}))|_{\pa M_\circ} = -2\sum_{j=1}^3\Gamma_{j 0 0}(t,0)\frac{\hat x^j}{|\hat x|}\,\dd\hat t^2
  \end{equation}
  by Lemma~\ref{LemmaGLFermi}. Therefore, $\hat h_{(-1)}=0$ if and only if $\cC$ is a geodesic.
  
  We now expand~\eqref{EqAhNecMetric} around $t=t_0$, and identify $\hat t=\frac{t-t_0}{\eps}$ at $T_{c(t_0)}M$; this gives
  \begin{equation}
  \label{EqAhNecMetric2}
    \hat g_b(\hat x) + \eps\bigl( \hat t\hat g_b'(\dot b) + \hat h(\hat x) \bigr) + \cO(\eps^2),\qquad b=b(t_0),\ \dot b=b'(t_0),\ \hat h(\hat x)=\hat h(t_0,\hat x).
  \end{equation}
  
  Since the cosmological constant $\Lambda$ only gives a $\cO(\eps^2)$ contribution to $\sfe(\eps^2(\Ric(\wt g)-\Lambda\wt g))=\Ric(\hat g)+\cO(\eps)$ at $\hat M$, the validity of the Einstein vacuum equations $\Ric(\wt g)-\Lambda\wt g$ also to subleading order at $\hat M$ is equivalent to~\eqref{EqAhNecMetric2} being Ricci-flat modulo $\cO(\eps^{1+\delta})$ errors, and thus to the validity of the equation
  \begin{equation}
  \label{EqAhNecLinRic}
    D_{\hat g_b}\Ric\bigl(\hat t\hat g'_b(\dot b)+\hat h\bigr) = \bigl[D_{\hat g_b}\Ric,\hat t\bigr]\hat g'_b(\dot b) + \wh{D_{\hat g_b}\Ric}(0)\hat h = 0
  \end{equation}
  for $\hat h(t_0,\hat x)$ with~\eqref{EqAhNechat}--\eqref{EqAhNechm1}. Since $\hat x^j\,\dd\hat t^2\in\ker D_{\hat{\ubar g}}\Ric$, we have $\wh{D_{\hat g_b}\Ric}(0)(\hat h)\in\cA^{1+\delta}$.
  
  Let $\psi\in\CI(\hat X_b)$ be equal to $1$ near $\pa\hat X_b$ and equal to $0$ near $\hat r^{-1}(\bhm)$. For $\hat c\in\R^3$, set $\omega^*=\psi\omega_{b,\hat c}$ in the notation of Definition~\ref{DefAhKBoosts}. We then integrate~\eqref{EqAhNecLinRic} against the tensor
  \[
    h^*=\sfG_{\hat g_b}\delta_{\hat g_b}^*\omega^*\in\rho_\circ^2\CI;
  \]
  the leading order term of $h^*$ at $\pa\hat X_b$ is $\hat r^{-2}h_{(2)}$, where $h_{(2)}$ is given by~\eqref{EqAhKCokerCOMh2}. The contribution of $\hat g'_b(\dot b)$ vanishes by Theorem~\ref{ThmAhKCoker}\eqref{ItAhKCokerKerrPC}. Furthermore, since for $\hat h'\in\cA^{-1+\delta}$, $\delta>0$, we can integrate in parts to obtain $\la\wh{D_{\hat g_b}\Ric}(0)\hat h',h^*\ra=0$, only the leading order term $\hat r\hat h_{(-1)}$ of $\hat h$ enters in the calculation of $\la\wh{D_{\hat g_b}\Ric}(0)\hat h,h^*\ra$. In the splitting~\eqref{EqMkBundleSplit}, we have
  \[
    \hat h_{(-1)}=(\scal(\fq),\scal(\fq),0,\scal(\fq),0,0)^T
  \]
  where $\fq:=-\half(\Gamma_{j 0 0}(t_0,0))_{j=1,2,3}$. Therefore, inserting a localizer $\chi_\eps=\chi(\frac{\rho_\circ}{\eps})$, where $\chi\in\CI([0,\infty))$ vanishes near $0$ and is equal to $1$ on $[1,\infty)$, we have, using Corollary~\ref{CorMkDiffDRic}, formula~\eqref{EqMkMinkInner}, equation~\eqref{EqAhKCokerCOMh2}, and $h^*\in\ker\wh{D_{\hat g_b}\Ric}(0)^*$,
  \begin{align*}
    \big\la\wh{D_{\hat g_b}\Ric}(0)\hat h,h^*\big\ra &= -\lim_{\eps\searrow 0} \big\la \bigl[\wh{D_{\hat g_b}\Ric}(0),\chi_\eps\bigr]\hat h,h^*\big\ra \\
      &= -\big\la\pa_\lambda N\bigl(\rho_\circ^{-2}\wh{D_{\hat{\ubar g}}\Ric}(0),-1\bigr)\hat h_{(-1)},h_{(2)}\big\ra_{L^2(\Sph^2)} \\
      &= -\frac{\bhm}{2}\big\la \bigl(\scal(\fq),2\scal(\fq),-\sld\scal(\fq),\scal(\fq),\sld\scal(\fq),-2\scal(\fq)\slg\bigr)^T, \\
      &\quad\hspace{8em} \bigl(-\scal(\hat c),0,\sld\scal(\hat c),-\scal(\hat c),-\sld\scal(\hat c),0\bigr)^T \big\ra \\
      &= 16\pi\bhm(\fq\cdot\hat c).
  \end{align*}
  We conclude that equation~\eqref{EqAhNecLinRic} can hold only if this pairing vanishes for all $\hat c$, which forces $\fq=0$ and thus $\Gamma_{j 0 0}(t_0,0)=0$. Since $t_0\in I$ was arbitrary, we conclude that $\hat h_{(-1)}=0$, and thus $\cC$ must be a geodesic; this re-proves, from the perspective of $\hat M$, what we had demonstrated already in~\S\ref{SsAcGW} following \cite{GrallaWaldSelfForce}.
  
  Next, from equation~\eqref{EqAhNecLinRic}, where we now have $\hat h\in\cA^{-1+\delta}$ for $0<\delta<\min\Re\cE$, we can extract information also about the derivative
  \[
    \dot b=(\dot\bhm,\dot\bha)
  \]
  of the Kerr parameters along $\cC$. First, we integrate against $\sfG_{\hat g_b}\delta_{\hat g_b}^*\omega^*$ with $\omega^*=\psi\,\dd t$; the contribution from $\hat h$ vanishes upon integrating by parts, whereas by Theorem~\ref{ThmAhKCoker}, the term involving $\hat g'_b(\dot b)$ contributes $-8\pi\dot\bhm$. Therefore, $\dot\bhm=0$, and we conclude that the black hole mass must be constant along $\cC$; we had previously deduced this from the perspective of $M_\circ$ using Corollary~\ref{CorAcGWEOM}.
  
  Finally, we constrain $\dot\bha$ by integrating~\eqref{EqAhNecLinRic} against $h^*(\fq)=\sfG_{\hat g_b}\delta_{\hat g_b}^*\omega^*\in\rho_\circ^2\CI$ where $\omega^*(\fq)=\psi\hat r^2\vect(\fq)$ in the notation of~\eqref{EqAhKVect1form}. By Theorem~\ref{ThmAhKCoker}, the contribution from $\hat g_b'(\dot b)$ is $-8\pi\bhm(\fq\cdot\dot\bha)$. On the other hand, we can integrate by parts to conclude that $\la\wh{D_{\hat g_b}\Ric}(0)\hat h,h^*(\fq)\ra=0$ (see Lemma~\ref{LemmaAh0Nec}). Therefore, we must have $\fq\cdot\dot\bha=0$ for all $\fq\in\R^3$, and thus $\dot\bha=0$, as claimed.
\end{proof}

\section{Construction of the formal solution \texorpdfstring{at $\eps=0$}{near vanishing gluing parameter}}
\label{SF}

In this section, we complete step~\eqref{ItMFormalI} of the proof of Theorem~\ref{ThmM}. We use the notation of the Theorem as introduced in Definition~\ref{DefMData} and the constructions following it. In particular, $(M,g)$ is globally hyperbolic, the metric $g$ solves the Einstein vacuum equations
\begin{equation}
\label{EqFMcEinstein}
  \Ric(g)=\Lambda g,
\end{equation}
the curve $\cC\subset M$ is a timelike geodesic arc-length parameterized by $c\colon I\subset\R\to M$, and we fix subextremal black hole parameters
\[
  b=(\bhm,\bha),\qquad \bhm>0.
\]
Without loss of generality (due to the $O(3)$ freedom in Lemma~\ref{LemmaGLFermi}), we may assume that
\begin{equation}
\label{EqFAngMomNorm}
  \bha=a e_3,\qquad a:=|\bha|,
\end{equation}
where $e_3\in\R^3$ is the third standard basis vector. We fix Fermi normal coordinates $(t,x)$ around $\cC$, and write $r=|x|$. We moreover set $\hat t=\dd t$, $\hat x=\dd x$ on $T_\cC M$, which we moreover identify with the coordinates $\hat t=\frac{t-t_0}{\eps}$ and $\hat x=\frac{x}{\eps}$ on the front face of $[\wt M;\hat M_p]$. Fix cutoff functions
\begin{equation}
\label{EqFCutoffs}
  \hat\chi,\chi_\circ\in\CI(\wt M)
\end{equation}
to collar neighborhoods of $\hat M,M_\circ$, respectively, as in~\eqref{EqGRelCutoffs}; we demand that $\supp\hat\chi$ is contained in the domain of influence of a compact subset of $\cU^\circ$.

With $\wt K$ defined near $\hat M$ by $\wt K=\{(\eps,t,x)\colon|\frac{x}{\eps}|\leq\bhm\}$ (cf.\ item~\eqref{ItMFamily} in~\S\ref{SM}), recall the $(\emptyset,\N_0+1)$-smooth total family
\[
  \wt g_0 \in \wt\upbeta^*\CI(M;S^2 T^*M) + \cA_\phg^{\N_0,\N_0+1}(\wt M\setminus\wt K^\circ;S^2\wt T^*\wt M)
\]
defined by Lemma~\ref{LemmaGKNaive}\eqref{ItGKNaiveGeod}. Concretely, we define $\wt g$ near $\cC$ as in the proof of Lemma~\ref{LemmaGKNaive} (albeit using slightly different notation here) by
\begin{equation}
\label{EqFwtgNaive}
  \wt g_0(t_0,x) = -\dd t^2+\dd x^2 + \wt\upbeta^*g'(t_0,x;\dd t,\dd x) + \hat\chi(t_0,x)\hat g_{1,b}\Bigl(\frac{|x|}{\eps},\omega;\dd t,\dd x\Bigr),
\end{equation}
where $g'\in\CI(M;S^2 T^*M)$ vanishes quadratically at $x=0$, and $\hat g_{1,b}=\hat g_b-\hat{\ubar g}$ in the notation of Definition~\ref{DefGKModel}, so
\begin{equation}
\label{EqFwtgNaiveg1b}
  \hat g_{1,b}(\hat r,\omega;\dd\hat t,\dd\hat x)=\frac{2\bhm}{\hat r}(\dd\hat t^2+\dd\hat r^2)+\cO(\hat r^{-2}).
\end{equation}

By Corollary~\ref{CorGLCurvature}, we have
\begin{equation}
\label{EqFErr0}
  \Err_0 := \Ric(\wt g_0) - \Lambda\wt g_0 \in \hat\rho^{-1}\rho_\circ\CI(\wt M\setminus\wt K^\circ;S^2\wt T^*\wt M),
\end{equation}
and $\supp\Err_0$ is contained in the domain of influence of a compact subset of $\cU^\circ$. In fact, $\Err_0$ has an additional order of vanishing at $\hat M$, as we proceed to show; that is, $\Err_0\in\rho_\circ\CI$. The following result is well-known, see e.g.\ \cite[\S8.5]{PoissonPointParticleReview}; we phrase it in general spacetime dimension $1+n$ and give a proof for completeness. (We do not need to require $g$ to solve the Einstein equations for this result.)

\begin{lemma}[Metric in Fermi normal coordinates]
\label{LemmaFMc1Metric}
  Use Fermi normal coordinates $(t,x)=z=(z^0,\ldots,z^n)$ around the timelike geodesic $\cC\subset M$; write $i,j,k,l$ for indices between $1$ and $n$. Then the metric $g$ on $M$ takes the form $g=-\dd t^2+\dd x^2+r^2 g_{(2)}+\cO(|x|^3)$, where $\cO(|x|^3)$ stands for a smooth symmetric 2-tensors whose coefficients in the frame $\pa_t,\pa_{x^j}$ ($j=1,\ldots,n$) vanish cubically at $x=0$, and where
  \begin{align*}
    (g_{(2)})_{0 0} &= -R_{0\ell 0 m}|_{(t,0)}\frac{x^\ell}{r}\frac{x^m}{r}, \\
    (g_{(2)})_{j 0} &= -\frac23 R_{j\ell 0 m}|_{(t,0)}\frac{x^\ell}{r}\frac{x^m}{r}, \\
    (g_{(2)})_{j k} &= -\frac13 R_{j\ell k m}|_{(t,0)}\frac{x^\ell}{r}\frac{x^m}{r}.
  \end{align*}
  Here, $R_{\kappa\lambda\mu\nu}|_{(t,0)}$ denotes the coefficients of the Riemann curvature tensor of $g$ at $(t,0)\in\cC$. That is,
  \begin{equation}
  \label{EqFMc1Metric}
    g_{(2)} = -\frac{x^\ell}{r}\frac{x^m}{r}\Bigl( R_{0\ell 0 m}\dd t^2 + \frac43 R_{j\ell 0 m}\dd t\,\dd x^j + \frac13 R_{j\ell k m}\dd x^j\,\dd x^k\Bigr).
  \end{equation}
\end{lemma}
\begin{proof}
  Consider a Jacobi field $J=J(s)$ along a geodesic $\gamma\colon s\mapsto\gamma(s)$. Denoting by $R$ the endomorphism of $T_\gamma M$ given by $R V=\Riem(g)(\gamma',V)\gamma'$, the Jacobi equation reads $J''=R J$, where we denote covariant differentiation along $\gamma$ by a prime. Note that $\la R V,W\ra=\la V,R W\ra$ for all $V,W$, and therefore $R$ and all its derivatives along $\gamma$ are symmetric. Define now the function $f(s)=\la J(s),J(s)\ra$, where $\la-,-\ra$ is the inner product given by $g(\gamma(s))$. Then
  \begin{align*}
    f&=|J|^2, \\
    f'&=2\la J,J'\ra, \\
    f''&=2|J'|^2+2\la J,R J\ra, \\
    f'''&=8\la J',R J\ra+2\la J,R'J\ra, \\
    f''''&=8|R J|^2+12\la J',R'J\ra+8\la J',R J'\ra+2\la J,R''J\ra.
  \end{align*}

  If $J(0)=0$ and $J'(0)=v\in T_{c(t)}M$, we have
  \[
    f(0)=f'(0)=0,\quad
    f''(0)=2|v|^2,\quad
    f'''(0)=0,\quad
    f''''(0)=8\la v,R v\ra,
  \]
  which implies $f(s)=s^2|v|^2+\frac13\la\Riem(g)(\gamma',v)\gamma',v\ra s^4+\cO(s^5)$. In Fermi normal coordinates $(t,x)$, the curves $\gamma_q(s)=(t,s(x+q v))$ (for fixed $t\in\R$ and $x\in\R^n$, and defined for small $s$) are geodesics for all $q$; therefore, $J(s)=\pa_q\gamma_q(s)|_{q=0}=s v$ is a Jacobi field along $\gamma_0(s)=(t,s x)$, and we obtain
  \[
    g(t,s x)_{j k}v^j v^k = \delta_{j k}v^j v^k+\frac{s^2}{3}\big\la\Riem(g)|_{(t,0)}(x^m\pa_m,v^k\pa_k)x^\ell\pa_\ell,v^j\pa_j\big\ra + \cO(s^3),
  \]
  which gives $r^2(g_{(2)})_{j k}(t,x)=\frac13 R_{j\ell m k}|_{(t,0)}x^\ell x^m+\cO(|x|^3)$.

  Next, we consider the family of geodesics $\gamma_q(s)=(t+q,s x)$, whose variation vector field $J(s)=\pa_t$ satisfies $J(0)=\pa_t$ and $J'(0)=\nabla_{\gamma_0'(0)}\pa_t=0$. Therefore,
  \[
    f(0)=-1,\quad
    f'(0)=0,\quad
    f''(0)=2\la\pa_t,R\pa_t\ra,
  \]
  which implies $f(s)=g(t,s x)_{0 0}=-1+s^2\la\Riem(g)|_{(t,0)}(x^m\pa_m,\pa_t)x^\ell\pa_\ell,\pa_t\ra+\cO(s^3)$ and thus the stated expression for $(g_{(2)})_{0 0}$.

  Finally, for $\gamma_q(s)=(t+q,s(x+q v))$, we have $J(s)=\pa_t+s v$, so $J(0)=\pa_t$ and $J'(0)=v$, and we compute
  \[
    f(0)=-1,\quad f'(0)=0,\quad f''(0)=2|v|^2+2\la\pa_t,R\pa_t\ra,\quad f'''(0)=8\la v,R\pa_t\ra + 2\la\pa_t,R'\pa_t\ra.
  \]
  In the resulting Taylor expansion of $g(t,s x)(\pa_t+s v,\pa_t+s v)$, we consider the coefficient of $s^3$, which is the sum of a term of schematic form $x^2 v$ and another term of the form $x^3$. The $x^2 v$ term is $2 r^2(g_{(2)})_{0 j}v^j$, which must equal the $x^2 v$ term of the $s^3$-coefficient of $f(s)$; the latter is $\frac{8}{3!}s^3\la v,R\pa_t\ra=\frac{4}{3}s^3\la\Riem(g)|_{(t,0)}(x^m\pa_m,\pa_t)x^\ell\pa_\ell,v^j\pa_j\ra$. This completes the proof.
\end{proof}

The components of the tensor $g_{(2)}$ in~\eqref{EqFMc1Metric} are dilation-invariant in $x$. In order to capture the fact that $g_{(2)}$ should rightfully live at $r=|x|=0$, we consider the blow-up
\[
  [\R_t\times\R^n_x;\R_t\times\{0\}] = \R_t\times\bigl([0,\infty)_r\times\Sph^{n-1}_\omega\bigr),
\]
with blow-down map $\upbeta\colon(t,r,\omega)\mapsto(t,r\omega)$, and identify $T^*_{(t,x)}\R^4\cong T^*_{(t,0)}\R^4$ via the trivialization by coordinate differentials. Write $\ff:=\R_t\times\{0\}\times\Sph^{n-1}_\omega\subset\R\times[0,\infty)\times\Sph^{n-1}$ for the front face and $\cC:=\R_t\times\{0\}\subset\R\times\R^n$ for the curve $x=0$. We can then regard
\begin{equation}
\label{EqFMc1MetricSpace}
  g_{(2)} \in \CI(\ff;\upbeta^*S^2 T^*_\cC\R^{n+1}).
\end{equation}

Specializing now the case $n=3$, we may split $\upbeta^*T_{(t,0)}\R^4$ and its symmetric second tensor power as in~\eqref{EqMkBundleSplit}, where $\dd x^0=\dd t+\dd r$ and $\dd x^1=\dd t-\dd r$. (Concretely, $g_{(2)}$ is now a smooth (in $t$) family of sections of a bundle over $\Sph^2$ which is isomorphic to $\ul\R\oplus\ul\R\oplus T^*\Sph^2\oplus\ul\R\oplus T^*\Sph^2\oplus S^2 T^*\Sph^2$, equipped with the fiber inner product induced by the Minkowski metric.) In its spherical harmonic decomposition, we next determine the projections to various pure types. (The metric $g$ does not need to satisfy the Einstein equations in the following result either.)

\begin{lemma}[Pure type components of $g_{(2)}$]
\label{LemmaFMc1Pure}
  We assume that $\dim M=n+1=4$ and define $g_{(2)}$ by~\eqref{EqFMc1Metric}, regarded as an element of~\eqref{EqFMc1MetricSpace}. Defining curvature components in the coordinates $(t,x)\in\R\times\R^3$, we have
  \begin{subequations}
  \begin{align}
  \label{ItFMc1Pures1}
    \pi_{\rms 1}(g_{(2)}) &= -\frac23\Ric(g)_{0 j}\,\dd t\,r\sld\omega^j; \\
  \label{ItFMc1Purev1}
    \pi_{\rmv 1}(g_{(2)}) &= 0.
  \end{align}
  \end{subequations}
  Here $R_g=\tr_g\Ric(g)$ is the scalar curvature of $g$ at $(t,0)\in\cC$. In particular, if $\Ric(g)=\Lambda g$, then $\pi_{\rms 1}(g_{(2)})=0$.
\end{lemma}
\begin{proof}
  Write $\omega^j=\frac{x^j}{r}$, $j=1,2,3$. Since $\omega^j\in\scalspace_1$, we have $\omega^j\omega^m\in\scalspace_0\oplus\scalspace_2$ since the $O(3)$-representation $\scalspace_1\otimes_s\scalspace_1$ is isomorphic to $\scalspace_0\oplus\scalspace_2$. Furthermore, $\dd t^2$ is of scalar type $0$, the 1-form $\dd x^j=r\,\sld\omega^j+\omega^j\,\dd r$ is of scalar type $1$, and $\dd x^j\,\dd x^k$ is a sum of scalar type $0$ and $2$ tensors. Therefore, writing $\rms l$ and $\rmv l$ for tensors of scalar type $l$ and vector type $l$, respectively, the three terms in~\eqref{EqFMc1Metric} are sums of tensors of the following pure types:
  \begin{subequations}
  \begin{align}
  \label{EqFMc1Pure1}
    -\omega^\ell\omega^m\,\dd t^2&:\quad \rms 0,\ \rms 2; \\
  \label{EqFMc1Pure2}
    -\frac43\omega^\ell\omega^m\,\dd t\,\dd(r\omega^j)&:\quad \rms 1,\ \rmv 2,\ \rms 3; \\
  \label{EqFMc1Pure3}
    -\frac13\omega^\ell\omega^m\,\dd(r\omega^j)\,\dd(r\omega^k)&:\quad \rms 0,\ \rmv 1,\ \rms 2,\ \rmv 3,\ \rms 4.
  \end{align}
  \end{subequations}
  We use here isomorphisms of $O(3)$-representations such as $\scalspace_2\otimes\scalspace_2\cong\scalspace_0\oplus\vectspace_1\oplus\scalspace_2\oplus\vectspace_3\oplus\scalspace_4$.

  \pfstep{Scalar type $1$ part.} The only contribution comes from~\eqref{EqFMc1Pure2}. Expanding $\dd(r\omega^j)=\omega^j\,\dd r+r\,\sld\omega^j$, we then note that the contraction of $\omega^\ell\omega^m\omega^j\,\dd r$ (which is symmetric in $j,\ell$) with $R_{j\ell 0 m}$ (which is antisymmetric in $j,\ell$) vanishes. Next, we shall need\footnote{This integral can be evaluated easily via polarization from $\int_{\Sph^2}(\fq\cdot x)^4\,\dd\slg=\frac{4\pi}{5}|\fq|^4$ for $\fq\in\R^3$, which in turn follows by spherical symmetry and scaling from the special case $\fq=(0,0,1)$. In polar coordinates, one finds $\int_0^{2\pi}\int_0^\pi(\cos\theta)^4\sin\theta\,\dd\theta\,\dd\phi=-\frac{2\pi}{5}(\cos^5\theta)|^\pi_0=\frac{4\pi}{5}$.}
  \[
    \frac{1}{4\pi}\int_{\Sph^2} \omega^\ell\omega^m\omega^j\omega^k\,\dd\slg = \frac{1}{15}(\delta^{\ell m}\delta^{j k}+\delta^{\ell j}\delta^{m k}+\delta^{\ell k}\delta^{m j});
  \]
  and we moreover compute $\slg^{-1}(\sld\omega^j,\sld\omega^k)=\la e_j-(e_j\cdot\omega)\omega,e_k-(e_k\cdot\omega)\omega\ra_{\R^3}$ at $\omega\in\Sph^2\subset\R^3$ (where $e_1,e_2,e_3$ is the standard basis of $\R^3$), which equals $\delta^{j k}-\omega^j\omega^k$. Since $\frac{1}{4\pi}\int_{\Sph^2}\omega^\ell\omega^m\,\dd\slg=\frac13\delta^{\ell m}$, the $\rms 1$ part of $\omega^\ell\omega^m\,r\,\sld\omega^j$ is therefore
  \[
    \sum_{k=1}^3 \left(\frac{1}{4\pi}\int_{\Sph^2} \omega^\ell\omega^m\slg^{-1}(\sld\omega^j,\sld\omega^k)\,\dd\slg\right)\,\frac32 r\,\sld\omega^k = \frac{1}{10}\bigl(4\delta^{\ell m}\,r\,\sld\omega^j-\delta^{\ell j}\,r\,\sld\omega^m-\delta^{m j}\,r\,\sld\omega^\ell\bigr),
  \]
  cf.\ Example~\ref{ExMkYProj}. Contracting with $-\frac43 R_{j\ell 0 m}$ gives the stated result.

  \pfstep{Vector type $1$ part.} The vector type $1$ components of~\eqref{EqFMc1Pure3} are necessarily of the form $\dd r\otimes_s r\vect$ for $\vect\in\vectspace_1$, cf.\ \eqref{EqMkYSplitv1}. Contracting
  \[
    \dd(r\omega^j)\,\dd(r\omega^k)=\omega^j\omega^k\,\dd r^2+\omega^k\,r\,\sld\omega^j\,\dd r+\omega^j\,r\,\sld\omega^k\,\dd r+r\,\sld\omega^j\,r\,\sld\omega^k
  \]
  with $R_{j\ell k m}$, we must thus find the vector type $1$ part of $R_{j\ell k m}\omega^\ell\omega^m(\omega^k\,\sld\omega^j+\omega^j\,\sld\omega^k)$. But this 1-form vanishes since $R_{j\ell k m}$ is odd in $k,m$, resp.\ $j,\ell$ while $\omega^\ell\omega^m\cdot\omega^k$, resp.\ $\omega^\ell\omega^m\cdot\omega^j$ is even.
\end{proof}

We now return to the specific setting of the black hole gluing problem.

\begin{lemma}[Error term \#0]
\label{LemmaFErr0}
  The total family $\wt g_0$, defined by~\eqref{EqFwtgNaive}, solves the Einstein vacuum equations to leading order at $M_\circ$ and to leading and subleading order at $\hat M$, in the sense that
  \[
    \Err_0 := \Ric(\wt g_0) - \Lambda\wt g_0 \in \rho_\circ\CI(\wt M\setminus\wt K^\circ;S^2\wt T^*\wt M).
  \]
  Moreover, $\supp\Err_0$ is contained in the domain of influence of a compact subset of $\cU^\circ$, and
  \begin{align}
  \label{EqFErr0MhatLot}
    \sfe\bigl(\Err_0|_{\hat M_{c(t)}}\bigr) &= \wh{D_{\hat g_b}\Ric}(0)\bigl(\hat r^2\sfe g_{(2)}(t)\bigr) - \Lambda\hat g_b, \\
  \label{EqFErr0McircLot}
    (\eps^{-1}\Err_0)|_{M_\circ} &\equiv (D_g\Ric-\Lambda)\Bigl(\frac{2\bhm}{r}(\dd t^2+\dd r^2)\Bigr) \bmod \CIdot(M_\circ;\upbeta_\circ^*S^2 T^*M).
  \end{align}
  These leading order terms lie in $\rho_\circ\CI(\hat M;\hat\upbeta^*S^2\wt T^*M)$ and $\hat\rho^{-1}\CI(M_\circ;\upbeta_\circ^*S^2 T^*M)$, respectively. Furthermore, the scalar type $1$ and vector type $1$ components of the common boundary values $(\hat r\Err_0)|_{\pa\hat M}=(r\eps^{-1}\Err_0)|_{\pa M_\circ}\in\CI(\pa M_\circ;\upbeta_\circ^*S^2 T^*_\cC M)$ vanish.
\end{lemma}
\begin{proof}
  We need to show that the $\cO(\eps)$ leading order term of $\eps^2\Err_0$ at $\hat M$ vanishes at $\hat M_{c(t_0)}^\circ$ for all $t_0\in I$. But since $\sfe\wt g_0|_{\{t=t_0\}}=\hat g_b+\cO(\eps^2)$ (regarded as an $\eps$-dependent family of symmetric 2-tensors on $T_{c(t_0)}M$), the desired conclusion follows from $\sfe(\eps^2(\Ric(\wt g_0)-\Lambda\wt g_0))=\Ric(\sfe\wt g_0)+\cO(\eps^2)$ near $\hat M^\circ$.

  In order to verify~\eqref{EqFErr0MhatLot}, we may replace $g'$ in~\eqref{EqFwtgNaive} by its leading order term $r^2 g_{(2)}$, since the lower order terms of $g'$ vanish cubically at $\hat M$ and thus do not contribute to $\Err_0|_{\hat M}$. Working near $\hat M^\circ$ and in the coordinates $\hat t,\hat x$ and writing $\omega=\frac{\hat x}{|\hat x|}$, we have\footnote{One may want to multiply the term $\wt\upbeta^*g'$ in~\eqref{EqFwtgNaive} by a cutoff $\chi_\circ$ vanishing near $\hat x=0$ to stress its origin as a correction term of the metric $g$ at $M_\circ$. We do not do this here for notational brevity.}
  \begin{align*}
    \sfe(\Err_0|_{\hat M}) &= \Bigl(\eps^{-2}\Ric\bigl(\hat g_b + \eps^2\hat r^2 g_{(2)}(t,\omega;\dd\hat t,\dd\hat x)\bigr)\Bigr)\Big|_{\eps=0} - \Lambda\sfe(\wt g_0|_{\hat M}) \\
      &= \wh{D_{\hat g_b}\Ric}(0)(\hat r^2\sfe g_{(2)}) - \Lambda\hat g_b.
  \end{align*}
  Near $(M_\circ)^\circ$, we compute
  \[
    \Ric(\wt g_0)-\Lambda\wt g_0 \equiv \eps (D_g\Ric-\Lambda)\Bigl(\bigl( \eps^{-1}(\wt g_0-\wt\upbeta^* g)\bigr)\big|_{M_\circ}\Bigr)
  \]
  modulo $\eps^2\CI(\wt M\setminus(\hat M\cup\wt K^\circ);S^2\wt T^*\wt M)$. This gives~\eqref{EqFErr0McircLot} upon using~\eqref{EqFwtgNaiveg1b}.

  We prove the final claim by considering the $\rho_\circ$-leading order term of $\sfe(\Err_0|_{\hat M_{c(t)}})$ at $\hat M_{c(t)}$. We need to use the expression~\eqref{EqFErr0MhatLot} and the following facts: the scalar type $1$ and vector type $1$ components of $\hat r^2\sfe g_{(2)}(t)$ vanish (Lemma~\ref{LemmaFMc1Pure}); the operator $\wh{D_{\hat g_b}\Ric}(0)\equiv\wh{D_{\hat g_{b_0}}\Ric}(0)\bmod\rho_\circ^4\Diffb^2$ maps $\hat r^2\sfe g_{(2)}(t)$ into a tensor with vanishing scalar and vector type $1$ components modulo $\rho_\circ^4\hat r^2\CI=\rho_\circ^2\CI$; and $\hat g_b$ is of scalar type $0$ modulo $\rho_\circ^2\CI$.
\end{proof}

The plan for the remainder for this section is as follows.
\begin{enumerate}
\item In~\S\ref{SsFMc1}, we solve away the leading order term of $\Err_0$ at $M_\circ$ using Theorem~\ref{ThmAc}.
\item In~\S\ref{SsFhM1}, we solve away the leading order term of the remaining error $\Err^1$ at $\hat M$. This involves the inversion of the zero energy operator $\wh{D_{\hat g_b}\Ric}(0)$, which however has a cokernel (see Theorem~\ref{ThmAhPhg}). We show, roughly speaking, that if one modulates the center of mass and the axis of rotation of the small black hole by amounts depending on the `slow' time variable $t$ along $\cC$, one can eliminate the cokernel. This is the most delicate part of the construction: an $\cO(\eps^k)$ error at $\hat M$ requires a modulation of the center of mass, resp.\ axis of rotation of the small black hole by an amount of size $\cO(\eps^k)$, resp.\ $\cO(\eps^{k+1})$. We subsequently show that by pulling back the updated metric by a suitable diffeomorphism on the total gluing spacetime, one can re-normalize the center of mass and axis of rotation of the small black hole.
\item In~\S\ref{SsFMc2}, we solve away the leading order term of the remaining error $\Err_2$ at $M_\circ$; we take care in keeping track of certain error terms produced at $\hat M$.
\item In~\S\ref{SsFhM2}, we again turn to $\hat M$, where the remaining error has an additional order of vanishing at $\hat M$, and indeed is of size $\cO(\eps\log\eps)$. By modulating the center of mass and axis of rotation, we again succeed in solving this error away, and afterwards re-center the black hole. Here it is important that the error terms from previous steps have a particular form, as this prevents modulations of the black hole mass or the magnitude of the angular momentum from becoming necessary at this stage.
\item Following another solution step at $M_\circ$ in~\S\ref{SsFMc3}, we turn again to $\hat M$ in~\S\ref{SsFhM3}, where we now face an error of size $\cO(\eps^2(\log\eps)^m)$; this can be solved away by modulating the center of mass of the small black hole by an amount of size $\cO(\eps^2(\log\eps)^m)$, and by also modulating the black hole parameters by amounts of size $\cO(\eps^3(\log\eps)^m)$ which is now better than $\eps^2$ and thus acceptable without the need for further re-centering.
\item The remainder of the construction consists of solving away errors in turn at $M_\circ$ and $\hat M$; at this point no further care needs to be taken as regards the particular nature of error terms. See~\S\ref{SsFRest}.
\end{enumerate}

See Theorem~\ref{ThmFh} for the final result of this section.

\subsection{First correction at \texorpdfstring{$M_\circ$}{the blown-up background spacetime}}
\label{SsFMc1}

We begin by solving away the error term $\Err_0$ from Lemma~\ref{LemmaFErr0} at $M_\circ$. Recall the notation~\eqref{EqBgPhgIndexSets}.

\begin{prop}[First correction at $M_\circ$]
\label{PropFMc1}
  There exist
  \[
    h^1_\sharp \in \cA_\phg^{(1,0)_+}(M_\circ;\upbeta_\circ^*S^2 T^*M),\qquad
    h^1_\flat \in \CI(M;S^2 T^*M),
  \]
  with the following properties:
  \begin{enumerate}
  \item\label{ItFMc1h1sharp} using Fermi normal coordinates near $\pa M_\circ$, we have
  \begin{align*}
    &h^1_\sharp=\hat\chi r h^1_{\sharp,(1,0)}+\tilde h^1_\sharp, \\
    &\qquad h^1_{\sharp,(1,0)}\in\CI(\pa M_\circ;\upbeta_\circ^*S^2 T^*_\cC M),\quad
            \tilde h^1_\sharp\in\cA_\phg^{(2,1)_+}(M_\circ;\upbeta_\circ^*S^2 T^*M),
  \end{align*}
  where $h^1_{\sharp,(1,0)}$ has vanishing scalar type $1$ and vector type $1$ components (on each fiber $\hat M_p\cap\pa M_\circ=\Sph^2$ of $\pa M_\circ$);
  \item\label{ItFMc1h1flat} $h^1_\flat$ vanishes quadratically at $\cC$;
  \item\label{ItFMc1Err1} if we set
    \begin{equation}
    \label{EqFMc1g1}
      \wt g^1 := \wt g_0 + \eps h^1,\qquad
      h^1:=\chi_\circ(h^1_\sharp+\upbeta_\circ^*h^1_\flat),
    \end{equation}
    then we have
    \begin{equation}
    \label{EqFMc1Err1}
      \Err^1 := \Ric(\wt g^1) - \Lambda\wt g^1 \in \cA_\phg^{(0,0)_+,(2,0)}(\wt M\setminus\wt K^\circ;S^2\wt T^*\wt M);
    \end{equation}
  \item we have $\wt g^1=g$ outside (the lift to $M_\circ$ of) the domain of influence $U$ of a compact subset of $\cU^\circ$ in $X$, and thus also $\supp\Err^1\cap M_\circ\subset\upbeta_\circ^*U$;
  \item the leading order term of $\Err^1$ at $\hat M$ is
    \begin{equation}
    \label{EqFMc1hat}
    \begin{split}
      \sfe(\Err^1|_{\hat M_{c(t)}}) &= \sfe(\Err_0|_{\hat M_{c(t)}}) + \wh{D_{\hat g_b}\Ric}(0)\bigl(\chi_\circ\hat r\sfe(h_{\sharp,(1,0)}^1(t))\bigr) \\
        &= \wh{D_{\hat g_b}\Ric}(0)\Bigl[\sfe\bigl( \hat r^2 g_{(2)}(t)+\chi_\circ\hat r h_{\sharp,(1,0)}^1(t) \bigr)\Bigr] - \Lambda\hat g_b;
    \end{split}
    \end{equation}
  \item the $\eps(\log\hat\rho)$ term of $\Err^1$ at $\hat M$ is
    \begin{equation}
    \label{EqFMc1hatLog}
      \sfe^{-1}\bigl(\eps\log(\eps\hat r)\Err^1_{(1,1)}\bigr),\qquad \Err^1_{(1,1)} := \wh{D_{\hat g_b}\Ric}(0)\bigl(\chi_\circ\hat r^2\sfe(h^1_{\sharp,(2,1)})\bigr),
    \end{equation}
    and satisfies $\sfe^{-1}\Err^1_{(1,1)}\in\rho_\circ\CI(\hat M;S^2\wt T^*\wt M)$; here\footnote{The notation reflects the fact that $h^1_{\sharp,(2,1)}$ is the $r^2\log r$ term in the polyhomogeneous expansion of $\tilde h^1_\sharp$ at $r=0$.} $h^1_{\sharp,(2,1)}\in\CI(\pa M_\circ;\upbeta_\circ^*S^2 T^*_\cC M)$ is a sum of scalar and vector type $2$ tensors, as is $(\rho_\circ^{-1}\Err^1_{(1,1)})|_{\pa\hat M}$. That is, if one extends $\Err^1_{(1,1)}$ to an $\eps$-independent tensor near $\hat M^\circ$ in the coordinates $t,\hat x$, then
    \[
      \Err^1-\hat\chi\sfe^{-1}\bigl(\eps\log(\eps\hat r)\Err^1_{(1,1)}\bigr)\in\cA_\phg^{(0,0)_+\setminus\{(1,1)\},(2,0)}(\wt M\setminus\wt K^\circ;S^2\wt T^*\wt M).
    \]
  \end{enumerate}
\end{prop}
\begin{proof}
  We obtain $h^1$ as the solution of
  \begin{equation}
  \label{EqFMc1hatEqh1}
    (D_g\Ric-\Lambda)h^1 = f_0 := -(\eps^{-1}\Err_0)|_{M_\circ} \in \hat\rho^{-1}\CI(M_\circ;\upbeta^*S^2 T^*M)
  \end{equation}
  by means of Theorem~\ref{ThmAc} with $\hat\cF=(-1+\N_0)\times\{0\}$. We have $\delta_g\sfG_g f_0=0$ since the second Bianchi identity, away from $\hat M$, gives
  \[
    0 = \delta_{\wt g_0}\sfG_{\wt g_0}\Err_0 \equiv \delta_g\sfG_g(\eps f_0) \bmod \eps^2\CI(\wt M\setminus\hat M;S^2\wt T^*\wt M).
  \]
  We need to be a bit more precise (cf.\ Remark~\ref{RmkAcBdyFormalLog}): near $\pa M_\circ$ and in Fermi normal coordinates, the forcing term is $f_0=r^{-1}f_{0,(-1,0)}+\tilde f_0$ where $f_{0,(-1,0)}\in\CI(\pa M_\circ;\upbeta_\circ^*S^2 T^*_\cC M)$ has vanishing scalar type $1$ and vector type $1$ components by Lemma~\ref{LemmaFErr0}, and $\tilde f_0\in\CI(M_\circ;\upbeta^*S^2 T^*M)$. Theorem~\ref{ThmRicSolv}\eqref{ItRicSolvl} then produces $h^1_{\sharp,(1,0)}$ with $\wh{D_{\ubar g}\Ric}(0)h^1_{\sharp,(1,0)}=f_{0,(-1,0)}$ (at each fiber of $\pa M_\circ\to\cC$), and
  \[
    f_0 - (D_g\Ric-\Lambda)\bigl(\hat\chi r h^1_{\sharp,(1,0)}\bigr) \in \CI(M_\circ;\upbeta^*S^2 T^*M)
  \]
  is one order better at $\pa M_\circ$ than $f_0$ and still lies in $\ker\delta_g\sfG_g$. Applying Theorem~\ref{ThmAc} to this error term produces $\tilde h_\sharp^1$ and $h_\flat^1$ satisfying properties~\eqref{ItFMc1h1sharp} and \eqref{ItFMc1h1flat}. In view of Theorem~\ref{ThmRicSolv}\eqref{ItRicSolvl} with $z=l=2$, one does expect there to be a nontrivial leading order logarithmic term $r^2(\log r)h^1_{\sharp,(2,1)}$ in the expansion of $\tilde h^1_\sharp$, with $h^1_{\sharp,(2,1)}$ necessarily being the sum of scalar and vector type $2$ tensors lying in $\ker N(r^2\wh{D_{\ubar g}\Ric}(0),2)$.

  Since $\eps h^1\in\cA_\phg^{(2,0)_+,(1,0)}(\wt M\setminus\wt K^\circ;S^2\wt T^*\wt M)$, with $(2,0)_+-2=(0,0)_+$ nonlinearly closed, property~\eqref{ItFMc1Err1} and the leading order description~\eqref{EqFMc1hat} at $\hat M$ follow from Proposition~\ref{PropELTAcc}\eqref{ItELTAccMcirc}. The leading order logarithmic contribution to $\wt g^1$ at $\hat M$ comes from the term
  \begin{equation}
  \label{EqFMc1Eps3LogEps}
    \eps r^2(\log r)\chi_\circ h^1_{\sharp,(2,1)} = \eps^3\log(\eps\hat r)\chi_\circ \hat r^2 h^1_{\sharp,(2,1)},
  \end{equation}
  and thus, following the arguments following~\eqref{EqELTAccFastCoord}, the leading order logarithmic term of $\Err^1$ is given by
  \[
    \sfe^{-1}\eps^{-2}\wh{D_{\hat g_b}\Ric}(0)\bigl(\eps^3(\log\eps)\chi_\circ\hat r^2\sfe(h^1_{\sharp,(2,1)})\bigr),
  \]
  as claimed in~\eqref{EqFMc1hatLog}. The fact that $\sfe^{-1}\Err^1_{(1,1)}\in\rho_\circ\CI$ follows from the earlier observation $h^1_{\sharp,(2,1)}\in\ker N(r^2\wh{D_{\ubar g}\Ric}(0),2)$.
\end{proof}

\subsection{First correction at \texorpdfstring{$\hat M$}{the front face of the total gluing spacetime}: modulation and re-centering}
\label{SsFhM1}

From~\eqref{EqFMc1Err1}, we deduce that
\[
  \sfe(\Err^1|_{\hat M})\in\rho_\circ^2\CI_{\rm I}(\breve T_\cC M\setminus\breve K^\circ;S^2\,{}^{\tscop,\vee}T(\breve T_\cC M)),
\]
with the explicit expression given by~\eqref{EqFMc1hat}; here we write $\breve K^\circ=\bigsqcup_{p\in\cC}\breve K_p^\circ$ where $\breve K_p\subset\breve T_p M$ is the closure of $\{|\hat x|\leq\bhm\}$. Restricted to a single fiber $\breve T_p M$ of $\breve T_\cC M$, we can identify the stationary tensor $\sfe(\Err^1|_{\hat M})$ with an element of $\rho_\circ^2\CI(\hat X_b;S^2\,\Ttsc^*_{\hat X_b}\hat M_b)$, and thus\footnote{This simply means that $f^1$ is a smooth family (in $t\in I$) of symmetric 2-tensors $f^1(t,\hat x)_{\mu\nu}\dd\hat z^\mu\,\dd\hat z^\nu$ where $\hat z=(\hat t,\hat x)\in\R^{1+3}$ and $f^1(t,-)_{\mu\nu}\in\CI(\{\hat x\in\R^3\colon|\hat x|\geq\bhm\})$ is smooth in $\rho_\circ=|\hat x|^{-1}$, $\frac{\hat x}{|\hat x|}$, and vanishes quadratically at $\rho_\circ=0$.}
\begin{equation}
\label{EqFhM1f1}
  f^1 := \sfe(\Err^1|_{\hat M}) \in \CI\bigl(I;\rho_\circ^2\CI(\hat X_b;S^2\,\Ttsc^*_{\hat X_b}\hat M_b)\bigr).
\end{equation}
Moreover, the identity $\delta_{\wt g^1}\sfG_{\wt g^1}\Err^1=0$ implies (in view of the stationarity of $f^1(t)$ for all $t\in I$) that
\[
  \wh{\delta_{\hat g_b}\sfG_{\hat g_b}}(0)f^1(t)=0\qquad\text{for all}\ t\in I.
\]
We investigate the condition~\eqref{EqAhPhgEquiv} required for an application of Theorem~\ref{ThmAhPhg} and thus set
\begin{align*}
  f^1_{\rm Kerr} &:= \la f^1,-\ra_{L^2(\hat X_b)} \in \CI(I;(\cK_{b,\rm Kerr}^*)^*), \\
  f^1_{\rm COM} &:= \la f^1,-\ra_{L^2(\hat X_b)} \in \CI(I;(\cK_{b,\rm COM}^*)^*)
\end{align*}
in the notation of Theorem~\ref{ThmAhKCoker}.

\begin{lemma}[$f^1$ and the cokernel of linearized Ricci]
\label{LemmaFhM1Coker}
  We use the notation of Theorem~\usref{ThmAhPhg}. For $\omega^*=\omega_0^*$ as well as for $\omega^*=\omega_{1 2}^*$ when $\bha\neq 0$,\footnote{Recall here the normalization~\eqref{EqFAngMomNorm}.} and for $\omega^*=\omega_{j k}^*$ for all $1\leq j<k\leq 3$ when $\bha=0$, we have
  \[
    \la f^1(t),\wh{\sfG_{\hat g_b}\delta_{\hat g_b}^*}(0)\omega^*\ra_{L^2(\hat X_b)} = f^1_{\rm Kerr}\bigl(\wh{\sfG_{\hat g_b}\delta_{\hat g_b}^*}(0)\omega^*\bigr) = 0
  \]
  for all $t\in I$.
\end{lemma}
\begin{proof}
  Let $\chi=\chi(\hat r)\in\CI(\hat X_b)$ be a radial function which is equal to $1$ near $\pa\hat X_b$ and supported in $\hat r=|\hat x|>\bhm$. In view of the discussion around~\eqref{EqAhKIBP}, we may take $\omega_0^*=\chi\pa_{\hat t}^\flat$ and $\omega_{1 2}^*=\chi V^\flat$ where $V=(e\times\hat x)\cdot\pa_{\hat x}$ is the rotation vector field around the axis $e=(0,0,1)^T$, which is the axis of rotation of $\hat g_b$ when $\bha\neq 0$, or $e\in\R^3$ is an arbitrary unit vector when $\bha=0$. Note then that $W=\pa_{\hat t},V$ is a Killing vector field, and therefore we have
  \[
    \wh{\sfG_{\hat g_b}\delta_{\hat g_b}^*}(0)\omega^* = \sfG_{\hat g_b}[\delta_{\hat g_b}^*,\chi]W^\flat\in\CIc(\hat X_b^\circ;\Ttsc^*_{\hat X_b}\hat M_b).
  \]
  Since this moreover vanishes near $\hat r=\hat\bhm$ by definition of $\chi$, we conclude that for any $h\in\sD'(\hat X_b^\circ;S^2 T^*_{\hat X_b^\circ}\hat M_b^\circ)$ we have
  \begin{equation}
  \label{EqFhM1CokerIBP}
    \la \wh{D_{\hat g_b}\Ric}(0)h, \wh{\sfG_{\hat g_b}\delta_{\hat g_b}^*}(0)\omega^*\ra = \la h, \wh{D_{\hat g_b}\Ric}(0)^*\wh{\sfG_{\hat g_b}\delta_{\hat g_b}^*}(0)\omega^*\ra = 0.
  \end{equation}
  Recalling the expression~\eqref{EqFMc1hat} for $f^1(t)$, it remains to note that
  \[
    \la\hat g_b,\wh{\sfG_{\hat g_b}\delta_{\hat g_b}^*}(0)\omega^*\ra = \int_{\hat X_b} \delta_{\hat g_b}(\omega^*)\,\dd\hat g_b|_{\hat X_b} = 0
  \]
  since, for $W=\pa_{\hat t},V$, we have
  \[
    \delta_{\hat g_b}(\chi W^\flat)=-\chi' W^\flat(\nabla\hat r)=-\chi'\hat g_b(W,\nabla\hat r)=-\chi'\dd\hat r(W)=0.\qedhere
  \]
\end{proof}

\begin{cor}[Parameters for elimination of the cokernel]
\label{CorFhM1CokerParam}
  We use the notation of Theorem~\usref{ThmAhKCoker}. There exist unique functions
  \[
    \hat c^1 \in \CI(I;\R^3),\qquad
    \begin{cases} \bha=0: & \dot\bha^1=0, \\ \bha\neq 0: & \dot\bha^1 \in \CI(I;(\R\bha)^\perp), \end{cases}
  \]
  so that
  \[
    f^1_{\rm Kerr}=\ell_{b,\rm Kerr}(0,\dot\bha^1),\qquad
    f^1_{\rm COM}=\ell_{b,\rm COM}(\hat c^1).
  \]
\end{cor}
\begin{proof}
  The existence and uniqueness of $\hat c^1(t)$, $t\in I$, follows from the fact that $\ell_{b,\rm COM}\colon\R^3\to(\cK_{b,\rm COM}^*)^*$ is an isomorphism. When $\bha=0$, then $f^1_{\rm Kerr}=0$ by Lemma~\ref{LemmaFhM1Coker}, which thus equals $\ell_{b,\rm Kerr}(0,0)$ indeed. When $\bha\neq 0$ on the other hand, then the explicit expression for $\ell_{b,\rm Kerr}$ in Theorem~\ref{ThmAhKCoker} shows that the space $\ell_{b,\rm Kerr}(0,(\R\bha)^\perp)\subset(\cK_{b,\rm Kerr}^*)^*$ consists of all linear functionals on $\cK_{b,\rm Kerr}^*$ which annihilate $\sfG_{\hat g_b}\delta_{\hat g_b}^*\omega^*$ for $\omega^*=\omega_0^*$ and $\omega_{1 2}^*$; this space thus contains $f^1_{\rm Kerr}(t)$ for each $t\in I$. The smoothness of $\hat c^1$ and $\dot\bha^1$ along $I$ follows from their uniqueness and the linearity of their construction.
\end{proof}

We shall account for $\hat c^1$ by modulating the center of mass of the small black hole, and for $\dot\bha^1$ in the case $\bha\neq 0$ by modulating its axis of rotation. We first motivate our strategy for the two modulations separately in~\S\S\ref{SssFhM1COM}--\ref{SssFhM1Axis}. In~\S\ref{SssFhM1Comb}, we combine the two modulations, and in~\S\ref{SssFhM1Mod} we graft them into the total gluing spacetime.

\subsubsection{Center of mass}
\label{SssFhM1COM}

In Fermi normal coordinates $(t,x)$, and with $\hat x=\frac{x}{\eps}$, and for $\hat c\in\CI(I;\R^3)$, define \emph{for fixed $t\in I$} the diffeomorphism
\[
  \Phi_{1,\hat c(t)} \colon (\eps,\hat x) \mapsto (\eps,\hat x+\hat c(t))
\]
on $[0,1)_\eps\times\{t\}\times\R^3_{\hat x}$. Here and in~\S\ref{SssFhM1Axis} below, we consider a simplified setting in which instead of $\wt g^1$ we work on
\[
  [0,1)_\eps\times\hat M^\circ = [0,1)_\eps \times I_t \times \R^3_{\hat x}
\]
(which is diffeomorphic to a neighborhood of $\hat M^\circ\subset\wt M$) with a section of the pullback of $S^2 T^*_\cC M$ which on a level set of $t$ we define by $\Phi_{1,\hat c(t)}^*\hat g_b$; we denote this tensor by $(\Phi_{1,\hat c(t)}^*\hat g_b)_{t\in I}$. We emphasize that $\hat g_b=\hat g_b(\hat x;\dd\hat t,\dd\hat x)$; thus,
\[
  (\Phi_{1,\hat c(t)}^*\hat g_b)_{t\in I}=\hat g_b(\hat x+\hat c(t);\dd\hat t,\dd\hat x).
\]
Working near the interior of the front face of $[\wt M;\hat M_{c(t_0)}]$ (cf.\ \eqref{EqGffBlowup}) with coordinates
\begin{equation}
\label{EqFhM1hatCoord}
  \eps,\quad
  \hat t=\frac{t-t_0}{\eps},\quad
  \hat x,
\end{equation}
we have
\begin{equation}
\label{EqFhM1COMFiberwise}
  (\Phi_{1,\hat c(t)}^*\hat g_b)_{t\in I} = \hat g_b(\hat x+\hat c(t_0+\eps\hat t);\dd\hat t,\dd\hat x).
\end{equation}
Taylor expanding $\hat c$ around $t_0$ gives $\hat c(t)=\hat c(t_0)+(t-t_0)\hat c'(t_0)+\frac{(t-t_0)^2}{2}\hat c''(t_0)+(t-t_0)^3\CI$ and thus (regarding $t$ as a parameter, and subsequently expressing it in terms of $\hat t$), defining $\Phi_{\hat c(t_0)}\colon(\hat t,\hat x)\mapsto(\hat t,\hat x+\hat c(t_0))$,
\begin{equation}
\label{EqFhM1COMTaylor}
  (\Phi_{1,\hat c(t)}^*\hat g_b)_{t\in I} \equiv \Phi_{\hat c(t_0)}^*\Bigl(\hat g_b + \eps\hat t\cL_{\hat c'(t_0)\cdot\pa_{\hat x}}\hat g_b + \frac{\eps^2\hat t^2}{2}\bigl( \cL_{\hat c'(t_0)\cdot\pa_{\hat x}}^2 + \cL_{\hat c''(t_0)\cdot\pa_{\hat x}}\bigr)\hat g_b \Bigr) \bmod \eps^3\CI
\end{equation}
for bounded $\hat t,\hat x$; here, we write $\CI=\CI([0,1)_\eps\times\R^4_{\hat t,\hat x};S^2 T^*\R^4)$. This gives
\begin{equation}
\label{EqFhM1COMRicFiberwise}
\begin{split}
  \Ric((\Phi_{1,\hat c(t)}^*\hat g_b)_{t\in I}) &\equiv \Phi_{\hat c(t_0)}^*\Bigl( \Ric(\hat g_b) + \eps D_{\hat g_b}\Ric\bigl(\hat t\cL_{\hat c'(t_0)\cdot\pa_{\hat x}}\hat g_b\bigr)\Bigr) \\
    &\equiv 2\eps \Phi_{\hat c(t_0)}^*\bigl( [D_{\hat g_b}\Ric,\hat t] h_{b,\hat c'(t_0)}\bigr) \bmod \eps^2\CI,
\end{split}
\end{equation}
where we used~\eqref{EqAhKTrans}. (An explicit calculation using Corollary~\ref{CorMkDiffDRic} shows that this does not vanish unless $\hat c'(t_0)=0$.)

\begin{rmk}[Comparison with total pullback]
\label{RmkFhM1COMTotal}
  The pullback of $(t,\hat x)\mapsto\hat g_b(\hat x;\dd\hat t,\dd\hat x)$ along
  \begin{equation}
  \label{EqFhM1TotalDiffeo}
    \Phi_{\hat c}\colon(\eps,t,\hat x)\mapsto(\eps,t,\hat x+\hat c(t)),
  \end{equation}
  which in the coordinates~\eqref{EqFhM1hatCoord} is
  \begin{equation}
  \label{EqFhM1COMTotal}
    \Phi_{\hat c} \colon (\eps,\hat t,\hat x)\mapsto(\eps,\hat t,\hat x+\hat c(t_0+\eps\hat t)) = \bigl(\eps,\hat t,\hat x+\hat c(t_0) + \eps\hat t\hat c'(t_0) + \cO(\eps^2)\bigr),
  \end{equation}
  is equal to
  \begin{equation}
  \label{EqFhM1COMTotalPb}
    \Phi_{\hat c}^*\hat g_b = \hat g_b\bigl(\hat x+\hat c(t_0+\eps\hat t),\dd\hat t,\dd\hat x+\eps\hat c'(t_0+\eps\hat t)\dd\hat t\bigr).
  \end{equation}
  Since $\hat g_b$ is Ricci-flat, this pullback is also Ricci-flat. Note that this pullback metric differs from~\eqref{EqFhM1COMFiberwise} when $\hat c'\neq 0$. The Taylor expansion of $\Phi_{\hat c}^*\hat g_b$ in~\eqref{EqFhM1COMTotalPb} at $\eps=0$ (and for bounded $\hat t$) is
  \[
    \Phi_{\hat c}^*\hat g_b \equiv \Phi_{\hat c(t_0)}^*\bigl(\hat g_b + \eps\cL_{\hat t\hat c'(t_0)\cdot\pa_{\hat x}}\hat g_b\bigr) \bmod \eps^2\CI.
  \]
  Compared with~\eqref{EqFhM1COMTaylor}, the factor $\hat t$ is in the argument of the Lie derivative here.
\end{rmk}

Returning to~\eqref{EqFhM1COMRicFiberwise}, recall now from Definition~\ref{DefAhKBoosts} that
\begin{equation}
\label{EqFhM1COMBoost}
\begin{split}
  &\hat t h_{b,\hat c'(t_0)}+\breve h_{b,\hat c'(t_0)} = \frac12\cL_{\hat t\hat c'(t_0)\cdot\pa_{\hat x}+(\hat c'(t_0)\cdot\hat x)\pa_{\hat t}}\hat g_b \in\ker D_{\hat g_b}\Ric \\
  &\qquad \implies \breve h_{b,\hat c'(t_0)} = \frac12\bigl(\cL_{\hat t\hat c'(t_0)\cdot\pa_{\hat x}}-\hat t\cL_{\hat c'(t_0)\cdot\pa_{\hat x}} + \cL_{(\hat c'(t_0)\cdot\hat x)\pa_{\hat t}}\bigr)\hat g_b.
\end{split}
\end{equation}
We thus deduce that $\Ric((\Phi_{1,\hat c(t)}^*(\hat g_b+2\eps\breve h_{b,\hat c'(t)}))_{t\in I})\in\eps^2\CI$. We proceed to compute the $\eps^2$ leading order term, which is related to the tensors in~\eqref{EqAhKCokerC} and~\eqref{EqAhKCoker2}. We first record the following general result:

\begin{lemma}[Lie derivatives and the Taylor expansion of Ricci]
\label{LemmaFhM1Lie}
  If $(\cM,g)$ is Ricci-flat and $V\in\cV(\cM)$, then
  \begin{align*}
    D_g^2\Ric(\cL_V g,\cL_V g) &= -D_g\Ric(\cL_V^2 g), \\
    D_g^3\Ric(\cL_V g,\cL_V g,\cL_V g)&=-3 D_g^2\Ric(\cL_V g,\cL_V^2 g) - D_g\Ric(\cL_V^3 g).
  \end{align*}
  For $V,W\in\cV(\cM)$, we furthermore have
  \begin{equation}
  \label{EqFhM1LiePolarized}
    D_g^2\Ric(\cL_V g,\cL_W g) = -\frac12\bigl(D_g\Ric(\cL_V\cL_W g) + D_g\Ric(\cL_W\cL_V g)\bigr).
  \end{equation}
\end{lemma}
\begin{proof}
  Denote by $\Psi_s$ the time $s$ flow of $V$, which is defined on any fixed precompact open subset $U\subset\cM$ when $s$ is sufficiently small (depending on $U$). We compute the Taylor expansion of $\Ric(\Psi_s^*g)=0$ around $s=0$ using
  \begin{equation}
  \label{EqFhM1LieMetric}
    \Psi_s^*g = g + s\cL_V g + \frac{s^2}{2}\cL_V^2 g + \frac{s^3}{6}\cL_V^3 g + \cO(s^4)
  \end{equation}
  where $\cO(s^4)$ denotes a smooth symmetric 2-tensor on $U$ which depends smoothly on $s$ when $s$ is sufficiently small and vanishes to fourth order at $s=0$, to be
  \begin{align*}
    0 &= \Ric(g) + s D_g\Ric(\cL_V g) + \frac{s^2}{2}\bigl(D_g\Ric(\cL_V^2 g)+D_g^2\Ric(\cL_V g,\cL_V g)\bigr) \\
      &\qquad + \frac{s^3}{6}\bigl(D_g\Ric(\cL_V^3 g) + 3 D_g^2\Ric(\cL_V g,\cL_V^2 g) + D_g^3\Ric(\cL_V g,\cL_V g,\cL_V g) \bigr) + \cO(s^4).
  \end{align*}
  The vanishing of the coefficients of $s^2$ and $s^3$ gives the desired result. The identity~\eqref{EqFhM1LiePolarized} follows by polarization.
\end{proof}

\begin{lemma}[Ricci tensor of fiberwise pullback: center of mass]
\label{LemmaFhM1COMRic}
  Consider on $[0,1)_\eps\times\hat M^\circ$ the symmetric 2-tensor (in $\dd\hat t$, $\dd\hat x$, with smooth dependence on $\eps,t,\hat x$) which for fixed $t\in I$ is given by $\Phi_{1,\hat c(t)}^*(\hat g_b+2\eps\breve h_{b,\hat c'(t)})$. For $t_0\in I$, set $V(t_0)=\hat t\hat c'(t_0)\cdot\pa_{\hat x}+(\hat c'(t_0)\cdot\hat x)\pa_{\hat t}$ (generator of a Lorentz boost). In terms of $t=t_0+\eps\hat t$, we then have
  \begin{equation}
  \label{EqFhM1COMRicMet}
    \bigl(\Phi_{1,\hat c(t)}^*(\hat g_b+2\eps\breve h_{b,\hat c'(t)})\bigr) \equiv \Phi_{\hat c(t_0)}^*\bigl(\hat g_b + \eps\cL_{V(t_0)}\hat g_b\bigr) \bmod \eps^2\CI,
  \end{equation}
  and furthermore
  \begin{equation}
  \label{EqFhM1Ric}
  \begin{split}
    &\Ric\bigl( (\Phi_{1,\hat c(t)}^*(\hat g_b+2\eps\breve h_{b,\hat c'(t)}))_{t\in I} \bigr) \\
    &\qquad \equiv \eps^2\Phi_{\hat c(t_0)}^*\Bigl[ 2 D_{\hat g_b}\Ric\Bigl(\frac{\hat t^2}{2}h_{b,\hat c''(t_0)}+\hat t\breve h_{b,\hat c''(t_0)}\Bigr) + D_{\hat g_b}\Ric(f(t_0)) \Bigr] \bmod \eps^3\CI
  \end{split}
  \end{equation}
  where $f\in\CI(I;\CI(\hat X_b;S^2\,\Ttsc^*_{\hat X_b}\hat M_b))$.\footnote{That is, the components of $f$ in the frame $\dd\hat t,\dd\hat x$ are smooth on $I_t\times(\ol{\R^3_{\hat x}}\setminus\{|\hat x|<\bhm\})$.}
\end{lemma}
\begin{proof}
  Using~\eqref{EqFhM1COMTaylor}, we have, modulo $\eps^3\CI$,
  \begin{equation}
  \label{EqFhM1RicgTaylor}
  \begin{split}
    &(\Phi_{1,\hat c(t)}^*(\hat g_b+2\eps\breve h_{b,\hat c'(t)}))_{t\in I} \\
    &\quad \equiv \Phi_{\hat c(t_0)}^*\Bigl[ \hat g_b + \eps\Bigl( \hat t\cL_{\hat c'(t_0)\cdot\pa_{\hat x}}\hat g_b + 2\breve h_{b,\hat c'(t_0)}\Bigr) \\
    &\quad\quad\hspace{4em} + \eps^2\Bigl(\frac{\hat t^2}{2}\cL_{\hat c''(t_0)\cdot\pa_{\hat x}}\hat g_b + \frac{\hat t^2}{2}\cL_{\hat c'(t_0)\cdot\pa_{\hat x}}^2\hat g_b + 2\hat t\breve h_{b,\hat c''(t_0)} + 2\hat t\cL_{\hat c'(t_0)\cdot\pa_{\hat x}}\breve h_{b,\hat c'(t_0)}\Bigr) \Bigr] \\
    &\quad = \Phi_{\hat c(t_0)}^*\Bigl[ \hat g_b + \eps \cL_{\hat t\hat c'(t_0)\cdot\pa_{\hat x}+(\hat c'(t_0)\cdot\hat x)\pa_{\hat t}}\hat g_b + 2\eps^2\Bigl(\frac{\hat t^2}{2}h_{b,\hat c''(t_0)}+\hat t\breve h_{b,\hat c''(t_0)}\Bigr) \\
    &\quad\quad\hspace{4em} + \eps^2\Bigl(\frac{\hat t^2}{2}\cL_{\hat c'(t_0)\cdot\pa_{\hat x}}^2\hat g_b + \hat t\cL_{\hat c'(t_0)\cdot\pa_{\hat x}}\bigl(\cL_{\hat t\hat c'(t_0)\cdot\pa_{\hat x}}-\hat t\cL_{\hat c'(t_0)\cdot\pa_{\hat x}} + \cL_{(\hat c'(t_0)\cdot\hat x)\pa_{\hat t}}\bigr)\hat g_b \Bigr) \Bigr],
  \end{split}
  \end{equation}
  where we used~\eqref{EqFhM1COMBoost} for the second equality. Motivated by Lemma~\ref{LemmaFhM1Lie}, we rewrite the final line. To simplify the notation, we use the notation $\wt\cL(V):=\cL_V(\cdot)$ from~\eqref{EqAhKCokerwtcL}. Writing $\hat c'=\hat c'(t_0)$, we then compute
  \begin{align*}
    &\frac12\cL_{\hat t\hat c'\cdot\pa_{\hat x}+(\hat c'\cdot\hat x)\pa_{\hat t}}^2\hat g_b = \frac12\wt\cL(\hat t\hat c'\cdot\pa_{\hat x}+(\hat c'\cdot\hat x)\pa_{\hat t})^2\hat g_b = \frac12(\hat t A_1 + A_0)^2\hat g_b, \\
    &\qquad A_1 := \wt\cL(\hat c'\cdot\pa_{\hat x}),\quad
            A_0 := [\wt\cL,\hat t](\hat c'\cdot\pa_{\hat x}) + \wt\cL((\hat c'\cdot\hat x)\pa_{\hat t}).
  \end{align*}
  Since $[A_1,\hat t]=0$, this is further equal to
  \[
    \Bigl(\frac{\hat t^2}{2}A_1^2 + \hat t A_1 A_0\Bigr)\hat g_b + \Bigl(\frac12[A_0,\hat t A_1] + \frac12 A_0^2\Bigr)\hat g_b.
  \]
  The first parenthesis precisely matches the $\eps^2$-coefficient in the final line of~\eqref{EqFhM1RicgTaylor}. Furthermore, $\frac12 A_0^2\hat g_b=A_0\breve h_{b,\hat c'}\in\rho_\circ\CI(\hat X_b;S^2\,\Ttsc^*_{\hat X_b}\hat M_b)$ (where $\rho_\circ=\la\hat x\ra^{-1}$) since $\breve h_{b,\hat c'}$ is of this class (and stationary) and $[\wt\cL,\hat t](\hat c'\cdot\pa_{\hat x})\in\Diffb^0(\hat X_b;S^2\,\Ttsc^*_{\hat X_b}\hat M_b)$ is stationary as well, and $\wt\cL((\hat c'\cdot\hat x)\pa_{\hat t})\breve h_{b,\hat c'}=[\wt\cL,\hat c'\cdot\hat x](\pa_{\hat t})\breve h_{b,\hat c'}$ with $[\wt\cL,\hat c'\cdot\hat x](\pa_{\hat t})\in\Diffb^0(\hat X_b;S^2\,\Ttsc^*_{\hat X_b}\hat M_b)$.

  We proceed to rewrite $\frac12[A_0,\hat t A_1]\hat g_b$: this is $\frac12$ times
  \begin{align*}
    &\bigl[ [\wt\cL,\hat t](\hat c'\cdot\pa_{\hat x}) + \wt\cL((\hat c'\cdot\hat x)\pa_{\hat t}), \wt\cL(\hat t\hat c'\cdot\pa_{\hat x}) - [\wt\cL,\hat t](\hat c'\cdot\pa_{\hat x}) \bigr]\hat g_b \\
    &\qquad = \bigl[ [\wt\cL,\hat t](\hat c'\cdot\pa_{\hat x}), \wt\cL(\hat t\hat c'\cdot\pa_{\hat x}) \bigr]\hat g_b + \bigl[\wt\cL((\hat c'\cdot\hat x)\pa_{\hat t}),\wt\cL(\hat t\hat c'\cdot\pa_{\hat x})\bigr]\hat g_b \\
    &\qquad\qquad - \bigl[\wt\cL((\hat c'\cdot\hat x)\pa_{\hat t}),[\wt\cL,\hat t](\hat c'\cdot\pa_{\hat x})\bigr]\hat g_b.
  \end{align*}
  The second term is the Lie derivative of $\hat g_b$ along $[(\hat c'\cdot\hat x)\pa_{\hat t},\hat t\hat c'\cdot\pa_x]=(\hat c'\cdot\hat x)\hat c'\cdot\pa_{\hat x}-|\hat c'|^2\hat t\pa_{\hat t}$; the third term is a stationary tensor which lies in $\CI(\hat X_b;S^2\,\Ttsc^*_{\hat X_b}\hat M_b)$. For $f=\hat t$, $V=\hat c'\cdot\pa_{\hat x}$, the operator acting on $\hat g_b$ in the first term is
  \begin{align*}
    \bigl[ [\wt\cL,f](V), \wt\cL(f V) \bigr] &= [ \wt\cL(f V)-f\wt\cL(V),\wt\cL(f V) ] = \cL_{f V}f\cL_V - f\cL_V\cL_{f V} \\
      &= f V(f)\cL_V+ f[\cL_{f V},\cL_V] = f V(f)\cL_V + f\cL_{[f V,V]} \\
      &= f V(f)\cL_V - f\cL_{V(f)V} = 0
  \end{align*}
  since $V(f)=0$. In summary, we have shown that
  \begin{align*}
    &(\Phi_{1,\hat c(t)}^*(\hat g_b+2\eps\breve h_{b,\hat c'(t)}))_{t\in I} \\
    &\qquad \Phi_{\hat c(t_0)}^*\Bigl[ \hat g_b + \eps\cL_{V(t_0)}\hat g_b + \frac{\eps^2}{2}\cL_{V(t_0)}^2\hat g_b + 2\eps^2\Bigl(\frac{\hat t^2}{2}h_{b,\hat c''(t_0)}+\hat t\breve h_{b,\hat c''(t_0)}\Bigr) \\
    &\qquad\hspace{18em} + \eps^2\bigl(\cL_{W(t_0)}\hat g_b + f(t_0)\bigr) \Bigr] \bmod \eps^3\CI,
\end{align*}
where $f\in\CI(I;\CI(\hat X_b;S^2\,\Ttsc^*_{\hat X_b}\hat M_b))$, and $W\in\CI(I;\cV(\R^{1+3}_{\hat t,\hat x}))$. (The proof gives $W(t_0)=\frac12((\hat c'(t_0)\cdot\hat x)\hat c'(t_0)\cdot\pa_{\hat x}-|\hat c'(t_0)|^2\hat t\pa_{\hat t})$.)

  Next, note that pullback of tensors along $\Phi_{\hat c(t_0)}^*$, regarded as a $t$-independent diffeomorphism of $t$-level sets coincides with pullback along $(\eps,t,\hat x)\mapsto(\eps,t,\hat x+\hat c(t_0))$ (cf.\ Remark~\ref{RmkFhM1COMTotal}) and thus commutes with the computation of the Ricci tensor. We moreover have $\Ric(\hat g_b+\eps\cL_{V(t_0)}\hat g_b+\frac{\eps^2}{2}\cL_{V(t_0)}^2\hat g_b)\in\eps^3\CI$ by (the proof of) Lemma~\ref{LemmaFhM1Lie}. The expression~\eqref{EqFhM1Ric} thus follows from the fact that $D_{\hat g_b}\Ric(\cL_{W(t_0)}\hat g_b)=0$.
\end{proof}

\subsubsection{Axis of rotation}
\label{SssFhM1Axis}

For this part, we shall assume
\[
  \bha\neq 0.
\]
(Otherwise, there is nothing to do by Corollary~\ref{CorFhM1CokerParam}.) We need to modify the tensor $(\eps,t,\hat x)\mapsto\hat g_b(\hat x;\dd\hat t,\dd\hat x)$, $t=t_0+\eps\hat t$, in such a way as to produce a term $\hat t g_b'(0,\dot\bha^1)$ in the $\eps^2$-coefficient of the Ricci tensor, similarly to (but in fact simpler than)~\eqref{EqFhM1Ric}. We define the rotation vector field
\begin{equation}
\label{EqFhM1AxisVF}
  \Omega(\dot\bha) = \Omega_0(\dot\bha,\hat x)\cdot\pa_{\hat x},\qquad
  \Omega_0(\dot\bha,\hat x) = \Bigl(\dot\bha\times\frac{\bha}{|\bha|^2}\Bigr)\times\hat x.
\end{equation}

\begin{lemma}[Linearized Kerr metric as a Lie derivative]
\label{LemmaFhM1AxisLie}
  If $\R^3\ni\dot\bha\perp\bha$ (in terms of the Euclidean metric on $\R^3$), then
  \[
    g_{\bhm,\bha}'(0,\dot\bha) = \cL_{\Omega(\dot\bha)}\hat g_b.
  \]
\end{lemma}
\begin{proof}
  Denote the time $s$ flow of $\Omega(\dot\bha)$ by $\phi_s$. Then the derivative at $s=0$ of $\phi_s^*\hat g_{\bhm,\bha}=\hat g_{\bhm,\phi_s^{-1}\bha}$ reads
  \[
    \cL_{\Omega(\dot\bha)}\hat g_{\bhm,\bha} = \hat g'_{\bhm,\bha}(0,-\Omega_0(\dot\bha,\bha)).
  \]
  But $\Omega_0(\dot\bha,\bha)=-(\bha\cdot\frac{\bha}{|\bha|^2})\dot\bha=-\dot\bha$ since $\dot\bha\cdot\bha=0$.
\end{proof}

The analogue of Lemma~\ref{LemmaFhM1COMRic} is then:

\begin{lemma}[Ricci tensor of fiberwise pullback: axis of rotation]
\label{LemmaFhM1AxisRic}
  Let $\dot\bha\in\CI(I;(\R\bha)^\perp)$. Consider on $[0,1)_\eps\times\hat M^\circ$ the symmetric 2-tensor (in $\dd\hat t$, $\dd\hat x$, with smooth dependence on $\eps,t,\hat x$) which for fixed $t\in I$ is given by $\hat g_b+\eps\hat g'_b(0,\dot\bha(t))$. Let $t_0\in I$. In terms of $t=t_0+\eps\hat t$, we then have
  \begin{equation}
  \label{EqFhM1AxisRicMet}
    \bigl(\hat g_b+\eps\hat g'_b(0,\dot\bha(t))\bigr)_{t\in I} = \hat g_b + \eps\cL_{\Omega(\dot\bha(t_0))}\hat g_b \bmod \eps^2\CI.
  \end{equation}
  Furthermore,
  \begin{equation}
  \label{EqFhM1AxisRic}
    \Ric\bigl((\hat g_b+\eps\hat g'_b(0,\dot\bha(t)))_{t\in I}\bigr) \equiv \eps^2\Bigl[D_{\hat g_b}\Ric\bigl(\hat t\hat g'_b(0,\dot\bha'(t_0))\bigr)  + D_{\hat g_b}\Ric(f(t_0)) \Bigr] \bmod \eps^3\CI,
  \end{equation}
  where $f=-\frac{1}{2}\cL^2_{\Omega(\dot\bha(\cdot))}\hat g_b\in\CI(I;\rho_\circ^2\CI(\hat X_b;S^2\,\Ttsc^*_{\hat X_b}\hat M_b))$ (where $\rho_\circ=\la\hat x\ra^{-1}$).
\end{lemma}

The first term in~\eqref{EqFhM1AxisRic} is precisely the one appearing in~\eqref{EqAhKCokerK}.

\begin{proof}[Proof of Lemma~\usref{LemmaFhM1AxisRic}]
  Taylor expanding $\dot\bha(t)=\dot\bha(t_0)+(t-t_0)\dot\bha'(t_0)+\cO((t-t_0)^2)$ and then writing $t=t_0+\eps\hat t$ gives
  \begin{align*}
    (\hat g_b+\eps\hat g'_b(0,\dot\bha(t)))_{t\in I} &\equiv \hat g_b + \eps\cL_{\Omega(\dot\bha(t_0))}\hat g_b + \eps^2\hat t\hat g_b'(0,\dot\bha'(t_0)) \\
      &\equiv \hat g_b + \eps\cL_{\Omega(\dot\bha(t_0))}\hat g_b + \frac{\eps^2}{2}\cL^2_{\Omega(\dot\bha(t_0))}\hat g_b \\
      &\quad\qquad + \eps^2\Bigl(\hat t\hat g'_b(0,\dot\bha'(t_0))-\frac12\cL^2_{\Omega(\dot\bha(t_0))}\hat g_b\Bigr) \bmod \eps^3\CI
  \end{align*}
  in view of Lemma~\ref{LemmaFhM1AxisLie}. Lemma~\ref{LemmaFhM1Lie} then gives~\eqref{EqFhM1AxisRic}.
\end{proof}

\subsubsection{Combination; total pullback}
\label{SssFhM1Comb}

The combination of Lemmas~\ref{LemmaFhM1COMRic} and \ref{LemmaFhM1AxisRic} is:

\begin{prop}[Ricci tensor of fiberwise pullback: combination]
\label{PropFhM1CombRic}
  Let $\hat c\in\CI(I;\R^3)$. If $\bha\neq 0$ and $\dot\bha\in\CI(I;(\R\bha)^\perp)$, define
  \[
    \hat g_{1,b,\hat c,\dot\bha}=\hat g_{1,b,\hat c,\dot\bha}(\eps,t,\hat x;\dd\hat t,\dd\hat x)
  \]
  on $[0,1)_\eps\times\hat M^\circ$ to be equal to $\Phi_{1,\hat c(t)}^*(\hat g_b+2\eps\breve h_{b,\hat c'(t)}+\eps\hat g'_b(0,\dot\bha(t)))$ for fixed $t\in I$, where $\Phi_{1,\hat c(t)}(\eps,\hat x)=(\eps,\hat x+\hat c(t))$. Let $t_0\in I$ and set $V(t_0)=\hat t\hat c'(t_0)\cdot\pa_{\hat x}+(\hat c'(t_0)\cdot\hat x)\pa_{\hat t}+\Omega(\dot\bha(t_0))$. Then, in terms of $t=t_0+\eps\hat t$, and writing $\Phi_{\hat c(t_0)}\colon(\hat t,\hat x)\mapsto(\hat t,\hat x+\hat c(t_0))$,
  \begin{equation}
  \label{EqFhM1CombRicMet}
    \hat g_{1,b,\hat c,\dot\bha} = \Phi_{\hat c(t_0)}^*\bigl( \hat g_b + \eps\cL_{V(t_0)}\hat g_b \bigr) \bmod \eps^2\CI.
  \end{equation}
  Furthermore,
  \begin{equation}
  \label{EqFhM1CombRic}
  \begin{split}
    \Ric(\hat g_{1,b,\hat c,\dot\bha}) &\equiv \eps^2\Phi_{\hat c(t_0)}^*\Bigl[ D_{\hat g_b}\Ric\Bigl(\frac{\hat t^2}{2} h_{b,2\hat c''(t_0)} + \hat t\breve h_{b,2\hat c''(t_0)} + \hat t\hat g'_b(0,\dot\bha'(t_0))\Bigr) \\
      &\quad\hspace{13em} + D_{\hat g_b}\Ric(f(t_0)) \Bigr] \bmod \eps^3\CI,
  \end{split}
  \end{equation}
   where $f\in\CI(I;\CI(\hat X_b;S^2\,\Ttsc^*_{\hat X_b}\hat M_b))$. If $\bha=0$, define $\hat g_{b,\hat c}$ to be equal to $\Phi_{1,\hat c(t)}^*(\hat g_b+2\eps\breve h_{b,\hat c'(t)})$; then we have~\eqref{EqFhM1CombRicMet}--\eqref{EqFhM1CombRic} for $\hat g_{b,\hat c}$ if we drop all terms involving $\dot\bha$.
\end{prop}
\begin{proof}
  When $\bha=0$, this is Lemma~\ref{LemmaFhM1COMRic}; when $\bha\neq 0$, then in the proof of Lemma~\ref{LemmaFhM1COMRic}, we merely need to replace every occurrence of $\breve h_{b,\hat c'(t)}$ by $\breve h_{b,\hat c'(t)}+\frac12\hat g'_b(0,\dot\bha(t))=\breve h_{b,\hat c'(t)}+\frac12\cL_{\Omega(\dot\bha)}\hat g_b$.
\end{proof}

In order to prepare the grafting of $\hat g_{1,b,\hat c,\dot\bha}$ into the total gluing spacetime $\wt M$, we record:

\begin{lemma}[Further properties of $\hat g_{1,b,\hat c,\dot\bha}$]
\label{LemmaFhM1CombProp}
  Define the space
  \begin{equation}
  \label{EqFhM1CombPropSpace}
    \wt\cM := \bigl[\,[0,1)_\eps\times I_t\times\ol{\R^3_{\hat x}}; \{0\}\times I\times\pa\ol{\R^3}\,\bigr];
  \end{equation}
  denote its front face by $\cM_\circ$ and the lift of $\{0\}\times I\times\ol{\R^3}$ by $\hat\cM$, and write $\rho_\circ$, resp.\ $\hat\rho\in\CI(\wt\cM)$ for defining functions of $\cM_\circ$, resp.\ $\hat\cM$. Suppose first that $\bha\neq 0$.
  \begin{enumerate}
  \item\label{ItFhM1CombPropDecay}{\rm (Decay of coefficients.)} In the coordinates $\hat z=(\hat t,\hat x)$, the components of the tensor $\hat g_{1,b,\hat c,\dot\bha}(\eps,t,\hat x;\dd\hat t,\dd\hat x)$ (defined in Proposition~\usref{PropFhM1CombRic}), which are smooth functions on $[0,1)_\eps\times I_t\times\R^3_{\hat x}$, lift to smooth functions on\footnote{More precisely, they are defined on the subset where $|\hat x+\hat c(t)|\geq\bhm$; we shall not make this explicit in the notation here.} $\wt\cM$ which differ from the corresponding components of $\hat g_b$ by terms in $\rho_\circ^2\CI(\wt M)$.
  \item\label{ItFhM1CombPropPb}{\rm (Total pullback.)} In the notation of~\eqref{EqFhM1AxisVF}, set
  \[
    \Phi_{1,\hat c,\dot\bha} \colon (\eps,t,\hat x) \mapsto \bigl(\eps, t - \eps^2(\hat c'(t)\cdot\hat x), e^{-\eps\Omega(\dot\bha(t))}\hat x - \hat c(t) \bigr),
  \]
  where $e^{-\eps\Omega(\dot\bha(t))}$ is the time $\eps$ flow of $-\Omega_{\dot\bha(t)}$ (i.e.\ a rotation). Then $\Phi_{1,\hat c,\dot\bha}$ lifts to a diffeomorphism of $\wt M$ which is the identity on $\cM_\circ$ and the translation $(t,\hat x)\mapsto(t,\hat x-\hat c(t))$ on $\hat\cM^\circ$. Moreover, the components of
  \[
    \Phi_{1,\hat c,\dot\bha}^*\hat g_{1,b,\hat c,\dot\bha} - \hat g_b
  \]
  lie in $\hat\rho^2\rho_\circ\CI(\wt M)$.
  \end{enumerate}
  If $\bha=0$, the same conclusions hold upon dropping all terms and subscripts involving $\dot\bha$.
\end{lemma}

Recalling the second part of Lemma~\ref{LemmaGRel}, a neighborhood of $\hat\cM\subset\wt\cM$ is diffeomorphic to a neighborhood of $\hat M\subset\wt M$, with the diffeomorphism given by the identity map in the coordinates $(\eps,t,\hat x)$ on $\wt M^\circ$ induced by the Fermi normal coordinates $(t,x)$.

\begin{proof}[Proof of Lemma~\usref{LemmaFhM1CombProp}]
  We only consider the case $\bha\neq 0$, as the case $\bha=0$ follows from the same arguments but with $\dot\bha$ dropped.

  \pfstep{Part~\eqref{ItFhM1CombPropDecay}.} Lemma~\ref{LemmaAhKBoostLot} gives
  \[
    \breve h_{b,\hat c'(\cdot)}\in\CI\bigl(I;\rho_\circ\CI(\hat X_b;S^2\,\Ttsc^*_{\hat X_b}\hat M_b)\bigr),
  \]
  and hence the coefficients of $\eps\breve h_{b,\hat c'(\cdot)}$ are of class $\eps\la\hat x\ra^{-1}\CI$ on~\eqref{EqFhM1CombPropSpace}. The same is true, a fortiori, for the coefficients of $\eps\hat g'_b(0,\dot\bha)$ in view of the even better (namely, quadratic) decay of the coefficients of $\hat g'_b(0,\dot\bha)$ as $|\hat x|\to\infty$, cf.\ \eqref{EqGKLinSize}. Since $\eps\la\hat x\ra^{-1}$ is a smooth function vanishing quadratically at the front face $\cM_\circ$ of~\eqref{EqFhM1CombPropSpace}, this shows that the coefficients of $(2\eps\breve h_{b,\hat c'(t)}+\eps\hat g'_b(0,\dot\bha(t)))_{t\in I}$ vanish quadratically at $\cM_\circ$.

  To deal with the pullback, consider now the map $\Phi_{1,\hat c(t)}\colon(\eps,\hat x)\mapsto(\eps,\hat x+\hat c(t))$ from Proposition~\ref{PropFhM1CombRic}. Letting $x=\eps\hat x=(x^1,x^2,x^3)$, so $\Phi_{1,\hat c(t)}\colon(\eps,x)\mapsto(\eps,x+\eps\hat c(t))$, this is a smooth family (in $t$) of invertible linear maps in $(\eps,x)\in\R\times\R^3$ which preserve $\{0\}\times\R^3$. Therefore, this lifts to a smooth family of diffeomorphisms of $[[0,1)\times\ol{\R^3};\{0\}\times\pa\ol{\R^3}]$ and hence to a diffeomorphism of $\wt\cM$ which preserves all boundary hypersurfaces.

  To conclude the proof of the first part, note then that $\Phi_{1,\hat c(t_0)}$ is the time $1$ flow of $\hat c(t_0)\cdot\pa_{\hat x}$; but since $\cL_{\hat c(t_0)\cdot\pa_{\hat x}}\hat g_b=2 h_{b,\hat c'(t_0)}\in\rho_\circ^2\CI(\hat X_b;S^2\,\Ttsc^*_{\hat X_b}\hat M_b)$, we have $(\hat g_b-\Phi_{1,\hat c(t)}^*\hat g_b)_{t\in I}\in\CI(I;\rho_\circ^2\CI(\hat X_b;S^2\,\Ttsc^*_{\hat X_b}\hat M_b))$, i.e.\ upon lifting to $\wt\cM$ its coefficients vanish quadratically at $\cM_\circ$.

  \pfstep{Part~\eqref{ItFhM1CombPropPb}.} The smoothness properties of $\Phi_{1,\hat c,\dot\bha}$ are clear near the interior of $\hat\cM$, and also near the interior of $\cM_\circ$ where this map has the local coordinate expression
  \[
    (\eps,t,x) \mapsto \bigl(\eps, t-\eps(\hat c'(t)\cdot x), e^{-\eps\Omega(\dot\bha(t))}x-\eps\hat c(t)\bigr).
  \]
  At $\eps=0$, this is the identity map $(0,t,x)\mapsto(0,t,x)$. Furthermore, the projection of $\Phi_{1,\hat c,\dot\bha}$ to $I$ is clearly smooth, and the projection to $[[0,1)\times\ol{\R^3};\{0\}\times\pa\ol{\R^3}]$ is a smooth family (in $t\in I$) of smooth maps as well, as follows from the same arguments as in the proof of part~\eqref{ItFhM1CombPropDecay}.

  Let us now work near $t=t_0\in I$ and write $\hat t=\frac{t-t_0}{\eps}$; thus $\hat t,\hat x$ are linear coordinates in the interior of the front face of the blow-up of $\wt\cM$ at the lift of $\{0\}\times\{t_0\}\times\ol{\R^3}$ (cf.\ \eqref{EqGffBlowup}). The Taylor expansion of $\Phi_{\hat c(t_0)}\circ\Phi_{1,\hat c,\dot\bha}$ around $t=t_0$ expressed in the coordinates $\eps,\hat t,\hat x$ and truncated at order $\eps^2$ is
  \[
    (\Phi_{\hat c(t_0)}\circ\Phi_{1,\hat c,\dot\bha})(\eps,\hat t,\hat x) \equiv \bigl(\eps, \hat t-\eps(\hat c'(t_0)\cdot\hat x), \hat x - \eps\hat t\hat c'(t_0) - \eps\Omega_0(\dot\bha(t_0),\hat x) \bigr) \bmod \eps^2\CI.
  \]
  The quadratic vanishing of $\Phi_{1,\hat c,\dot\bha}^*\hat g_{1,b,\hat c,\dot\bha}-\hat g_b$ at $\hat\cM^\circ$ then follows from
  \[
    (\Phi_{\hat c(t_0)}\circ\Phi_{1,\hat c,\dot\bha})^*(\hat g_b+\eps\cL_{V(t_0)}\hat g_b) - \hat g_b \equiv (\hat g_b+\eps\cL_{V(t_0)}\hat g_b)-\eps\cL_{V(t_0)}(\hat g_b+\eps\cL_{V(t_0)}\hat g_b) - \hat g_b \equiv 0
  \]
  modulo terms vanishing quadratically at $\hat\cM^\circ$. Since $\hat g_{1,b,\hat c,\dot\bha}$ and $\hat g_b$ are equal at $\cM_\circ$ and $\Phi_{1,\hat c,\dot\bha}$ is the identity there, the simple vanishing at $\cM_\circ$ is immediate.
\end{proof}

\subsubsection{Full implementation}
\label{SssFhM1Mod}

We shall only consider the case
\[
  \bha\neq 0
\]
explicitly here; the treatment of the case $\bha=0$ is simpler, and obtained by omitting all terms involving $\dot\bha$.

We return to the total family $\wt g^1$ from Proposition~\ref{PropFMc1} and its error $f^1$ from~\eqref{EqFhM1f1}. Recall the functions $\hat c^1\in\CI(I;\R^3)$, $\dot\bha^1\in\CI(I;(\R\bha)^\perp)$ produced by Corollary~\ref{CorFhM1CokerParam}. In light of Theorem~\ref{ThmAhKCoker}, we can solve away $f^1$ to leading order at $\hat M$ only if we adjust $\wt g^1$ two orders (in powers of $\eps$) earlier.\footnote{In particular, the vanishing of the scalar type $1$ component of $g_{(2)}$ in Lemma~\ref{LemmaFMc1Pure} is critical: otherwise, the solution $h^1$ of equation~\eqref{EqFMc1hatEqh1} would have an $r(\log r)$ scalar type $1$ leading order term, which would lead to a scalar type $1$ error term at $\hat M^\circ$ of size $\eps\cdot\eps(\log\eps)$, which could not be solved away since this would require a modulation of the center of mass at order $\log\eps$---which blows up at $\hat M$.} Theorem~\ref{ThmAhKCoker} and Proposition~\ref{PropFhM1CombRic} suggest setting, for any fixed $\ubar t\in I$,
\begin{equation}
\label{EqFhM1ModChoices}
  \hat c(t) = -\frac{1}{2}\int_{\ubar t}^t\int_{\ubar t}^s \hat c^1(w)\,\dd w\,\dd s,\qquad
  \dot\bha(t) = -\int_{\ubar t}^t \dot\bha^1(s)\,\dd s,
\end{equation}
so that $2\hat c''=-\hat c^1$ and $\dot\bha'=-\dot\bha^1$. In view of Lemma~\ref{LemmaFhM1CombProp}\eqref{ItFhM1CombPropDecay}, if we set
\begin{equation}
\label{EqFhM1Modgu1}
  \wt g'_{1,\hat c,\dot\bha} := \wt g^1 + \hat\chi\sfe^{-1}(\hat g_{1,b,\hat c,\dot\bha}-\hat g_b),
\end{equation}
then $\wt g'_{1,\hat c,\dot\bha}-\wt g^1$ is a smooth section of $S^2\wt T^*\wt M$ near $M_\circ$ which vanishes quadratically at $M_\circ$. Furthermore, by~\eqref{EqFwtgNaive}, Lemma~\ref{LemmaFMc1Metric} (so $g'(t,x;\dd t,\dd x)=r^2 g_{(2)}(t,\frac{x}{|x|};\dd t,\dd x)$ plus a remainder with cubic decay at $r=0$), and~\eqref{EqFMc1g1}, we have
\[
  \sfe\wt g'_{1,\hat c,\dot\bha} \equiv \hat g_{1,b,\hat c,\dot\bha}(\eps,t,\hat x;\dd\hat t,\dd\hat x) + \eps^2\Bigl(\hat r^2 g_{(2)}\Bigl(t,\frac{x}{|x|};\dd\hat t,\dd\hat x\Bigr) + \chi_\circ\sfe(\hat r h^1_{\sharp,(1,0)})\Bigr)
\]
modulo terms which vanish more than quadratically at $\hat M$; this expression also defines $\wt g'_{1,\hat c,\dot\bha}$ on $\wt M\setminus\wt K_{\hat c}^\circ$, where $\wt K_{\hat c}^\circ=\{(\eps,t,\hat x)\colon(\eps,t,\hat x+\hat c(t))\in\wt K^\circ\}$ similarly to Proposition~\ref{PropFMc1}. The $\eps^3(\log\hat\rho)$-term of $\wt g'_{1,\hat c,\dot\bha}$, which is the leading logarithmic term at $\hat M$, is moreover given by~\eqref{EqFMc1Eps3LogEps} still. We then have
\begin{equation}
\label{EqFhM1ModErru1}
  \Err'_{1,\hat c,\dot\bha} := \Ric(\wt g'_{1,\hat c,\dot\bha}) - \Lambda\wt g'_{1,\hat c,\dot\bha} \in \cA_\phg^{(0,0)_+,(2,0)}(\wt M\setminus\wt K_{\hat c}^\circ;S^2\wt T^*\wt M),
\end{equation}
and indeed at $\hat M_{c(t)}$ we have
\begin{align*}
  f_{1,\hat c,\dot\bha}(t) &:= \sfe(\Err'_{1,\hat c,\dot\bha}|_{\hat M_{c(t)}}) \\
    &= \Bigl(\eps^{-2}\bigl(\Ric(\sfe\wt g'_{1,\hat c,\dot\bha}) - \Lambda\eps^2\sfe\wt g'_{1,\hat c,\dot\bha}\bigr)\Bigr)\Big|_{\eps=0} \\
    &= \Phi_{\hat c(t)}^*\biggl( \wh{D_{\hat g_b}\Ric}(0)\Bigl[ \sfe\bigl(\hat r^2 g_{(2)}(t)+\chi_\circ\hat r h^1_{\sharp,(1,0)}(t)\bigr) \\
    &\qquad \hspace{3em} - \Bigl( \frac{\hat t^2}{2}h_{b,\hat c^1(t)}+\hat t\breve h_{b,\hat c^1(t)} + \hat t\hat g'_b(0,\dot\bha^1(t))\Bigr) + \wh{D_{\hat g_b}\Ric}(0)(f(t)) \Bigr] - \Lambda\hat g_b \biggr)
\end{align*}
(with smooth dependence on $t$) by Proposition~\ref{PropFhM1CombRic}, where $\Phi_{\hat c(t)}\colon(\hat t,\hat x)\mapsto(\hat t,\hat x+\hat c(t))$ and $f(t)\in\CI(\hat X_b;S^2\,\Ttsc^*_{\hat X_b}\hat M_b)$. That is, compared to~\eqref{EqFMc1hat} we have the additional terms from~\eqref{EqFhM1CombRic}. Furthermore, the $\eps(\log\hat\rho)$-term of $\Err_{1,\hat c,\dot\bha}$ is given on a $t$-level set by the pullback of~\eqref{EqFMc1hatLog} along $\Phi_{\hat c(t)}$.

By Theorem~\ref{ThmAhKCoker} and Corollary~\ref{CorFhM1CokerParam}, and in view of the choices~\eqref{EqFhM1ModChoices}, the assumptions of Theorem~\ref{ThmAhPhg} (cf.\ \eqref{EqAhPhgEquiv}) are satisfied for $(\Phi_{\hat c(t)}^{-1})^*f_{1,\hat c,\dot\bha}(t)\in\rho_\circ^2\CI(\hat X_b;S^2\,\Ttsc^*_{\hat X_b}\hat M_b)$ for all $t\in I$, in that
\[
  \big\la (\Phi_{\hat c(t)}^{-1})^*f_{1,\hat c,\dot\bha}(t), - \big\ra_{L^2(\hat X_b)} = 0 \in (\cK_{b,\rm tot}^*)^*.
\]
Theorem~\ref{ThmAhPhg} produces
\[
  h\in \cA_\phg^{(0,1)_+}(\hat M\setminus\wt K^\circ;S^2\wt T^*_{\hat M}\wt M),
\]
so $\sfe h\in\CI(I;\cA_\phg^{(0,1)_+}(\hat X_b;S^2\,\Ttsc^*_{\hat X_b}\hat M_b))$, with $\Phi_{\hat c(t)}^*(\wh{D_{\hat g_b}\Ric}(0)(\sfe h(t)))=-f_{1,\hat c,\dot\bha}(t)$ for all $t\in I$. Since $\eps^2\hat\chi h\in\cA_\phg^{(2,0),(2,1)_+}$, we therefore find that\footnote{The error space here is contained in $\cA_\phg^{\N_0\cup(3,1)_+,(1,0)_+}$. But we keep the two summands separate here to record the fact that the coefficient of the first logarithmic term at $\hat M$, resp.\ $M_\circ$ does not have a logarithmic leading order term at the other boundary.}
\[
  \wt g_{1,\hat c,\dot\bha} := \wt g'_{1,\hat c,\dot\bha} + \eps^2\hat\chi h \equiv \wt\upbeta^*g \bmod \bigl(\cA_\phg^{(0,0)\cup(3,1)_+,(1,0)}+\cA_\phg^{(2,0),(2,1)_+}\bigr)(\wt M\setminus\wt K_{\hat c}^\circ;S^2\wt T^*\wt M)
\]
satisfies
\[
  \Err_{1,\hat c,\dot\bha} := \Ric(\wt g_{1,\hat c,\dot\bha}) - \Lambda\wt g_{1,\hat c,\dot\bha} \in \cA_\phg^{(1,1)_+,(2,1)_+}(\wt M\setminus\wt K_{\hat c}^\circ;S^2\wt T^*\wt M).
\]
(Thus, we have succeeded in solving away the leading order error at $\hat M$ at the expense of generating mild logarithmic terms at $M_\circ$ and shifting the small black hole.)

More precisely, recalling~\eqref{EqFhM1ModErru1}, the error $\Err_{1,\hat c,\dot\bha}$ is the sum of $\Err'_{1,\hat c,\dot\bha}$, the linear term $\eps^2 D_{\wt g'_{1,\hat c,\dot\bha}}\Ric(\hat\chi h)$, and quadratic and higher order error terms. But since $\wt g'_{1,\hat c,\dot\bha}\in\cA_\phg^{(0,0)\cup(3,1)_+,(0,0)}$ by Proposition~\ref{PropFMc1} and~\eqref{EqFhM1Modgu1}---i.e.\ this metric has smooth coefficients modulo $\cA_\phg^{(3,1)_+,(0,0)}$ (so modulo almost three orders down at $\hat M$)---the operator $D_{\wt g_{1,\hat c,\dot\bha}}\Ric\in\cA_\phg^{(-2,0)\cup(1,1)_+,(0,0)}\Diffse^2$ has smooth coefficients modulo $\cA_\phg^{(1,1)_+,(0,0)}\Diffse^2$ (so almost modulo three orders down at $\hat M$). As a consequence, we can separate the leading order logarithmic terms of $\Err_{1,\hat c,\dot\bha}$ at $\hat M$ (which is a pullback of the term~\eqref{EqFMc1hat}) and $M_\circ$ (arising from the $\log\hat r$ leading order term of $h$ at $\pa\hat M$), to wit,\footnote{The leading parts of the index sets of the first term are $(1,1)$ and $(2,0)$, and those of the second term are $(1,0)$ and $(2,1)$.}
\[
  \Err_{1,\hat c,\dot\bha} \in \bigl(\cA_\phg^{(1,1)_+,(2,0)\cup(3,2)_+} + \cA_\phg^{(1,0),(2,1)_+}\bigr)(\wt M\setminus\wt K_{\hat c}^\circ;S^2\wt T^*\wt M).
\]
We refine the description of the second summand here: the first step in the proof of Theorem~\ref{ThmAhPhg} is, in Lemma~\ref{LemmaAhPhgFormal}, the application of Theorem~\ref{ThmRicSolv}\eqref{ItRicSolv0}, and therefore the coefficient $h_{(0,1)}\in\CI(\pa M_\circ;\upbeta_\circ^* S^2 T^*_\cC M)$ of the leading order term $(\log\hat r)\sfe(h_{(0,1)})$ of $h$ is of scalar type $0$. The same is thus true for the coefficient of the logarithmic term of $\Err_{1,\hat c,\dot\bha}$ of the second summand here at $\pa M_\circ$, as this is $\eps^2(\log\hat r)D_g\Ric(h_{(0,1)})$ with $D_g\Ric(h_{(0,1)})\equiv D_{\ubar g}\Ric(h_{(0,1)})\bmod\CI(M_\circ;\upbeta_\circ^*S^2 T^*M)$ (which lies in $\hat\rho^{-1}\CI$ since $h_{(0,1)}\in\ker N(D\Ric,0)$, cf.\ the final part of Theorem~\ref{ThmRicSolv}).

Finally, we pull back $\wt g_{1,\hat c,\dot\bha}$ by $\Phi_{1,\hat c,\dot\bha}$ in the notation of Lemma~\ref{LemmaFhM1CombProp}\eqref{ItFhM1CombPropPb}; pullback by $\Phi_{1,\hat c,\dot\bha}$ preserves polyhomogeneity, and
\[
  \wt g_1 := \Phi_{1,\hat c,\dot\bha}^*\wt g_{1,\hat c,\dot\bha}
\]
equals $g$ at $M_\circ$ and $\hat g_b$ modulo $\hat\rho^2\CI$ near $\hat M$. In summary, we have shown:

\begin{prop}[First correction at $\hat M$]
\label{PropFhM1}
  There exists a $((3,1)_+,(1,0)_+)$-smooth total family $\wt g_1$ which is equal to\footnote{That is, its $\hat M$-model is constant along $\cC$ and equal to $\hat g_b$, while its $M_\circ$-model is $g$.} $\wt g_0$ at $\hat M$ and $M_\circ$ so that
  \[
    \Err_1 := \Ric(\wt g_1) - \Lambda\wt g_1 \in \bigl(\cA_\phg^{(1,1)_+,(2,0)\cup(3,2)_+} + \cA_\phg^{(1,0),(2,1)_+}\bigr)(\wt M\setminus\wt K^\circ;S^2\wt T^*\wt M),
  \]
  which moreover has the following properties.
  \begin{enumerate}
  \item\label{ItFhM1hatMQuadr} $\wt g_1$ is equal to the Kerr metric along $\hat M$ modulo $\cO(\eps^2)$ errors in the sense that in the coordinates $\eps,t,\hat x$ near $\hat M^\circ$ the coefficients of $\sfe(\wt g_1)(t,\hat x)-\hat g_b(\hat x)$ in the frame $\dd\hat t,\dd\hat x$ lie in $\eps^2\CI([0,1)_\eps\times I_t\times\R^3_{\hat x})$;
  \item $\wt g_1=g$ outside the domain of influence $U$ of a compact subset of $\cU^\circ$ in $X$, and thus also $\supp\Err_1\cap M_\circ\subset\upbeta_\circ^*U$;
  \item\label{ItFhM1hM} the leading order term of $\Err_1$ at $\hat M$ (in Fermi normal coordinates) is equal to the $\eps\log\hat\rho$ term $\sfe^{-1}(\eps\log(\eps\hat r)\Err^1_{(1,1)})$, $\Err^1_{(1,1)}=\wh{D_{\hat g_b}\Ric}(0)(\chi_\circ\hat r^2\sfe(h^1_{\sharp,(2,1)}))$, $h^1_{\sharp,(2,1)}\in\CI(\pa M_\circ;\upbeta_\circ^*S^2 T^*_\cC M)$, of $\Err^1$ in Proposition~\usref{PropFMc1}, in the sense described there, and $(\hat r\Err^1_{(1,1)})|_{\pa\hat M}$ is a sum of scalar and vector type $2$ terms;
  \item\label{ItFhM1Mc} the leading order term of $\Err_1$ at $M_\circ$ is equal to $\eps^2(\log\rho_\circ)\Err_{1,(2,1)}$ where $\Err_{1,(2,1)}\in\hat\rho^{-1}\CI(M_\circ;\upbeta_\circ^*S^2 T^*M)$ has scalar type $0$ leading order term $(r\Err_{1,(2,1)})|_{\pa M_\circ}$.
  \end{enumerate}
\end{prop}

\begin{rmk}[Re-interpretation of the update from $\wt g^1$ to $\wt g_1$]
\label{RmkFhM1Nature}
  Carefully note that $\Phi_{1,\hat c,\dot\bha}$ differs from the identity map near $(M_\circ)^\circ$ by a term of size $\eps$; thus, we typically only have $\wt g_1-\wt g^1=\cO(\eps)$ there, and in this sense the construction leading to Proposition~\ref{PropFhM1} in fact also entails a correction at $M_\circ$ at the $\cO(\eps)$ level. (On the other hand, $\wt g_1$ and $\wt g^1$ are equal modulo $\cO(\eps^2)$ corrections near $\hat M^\circ$.) An alternative approach to the proof of Proposition~\ref{PropFhM1} is thus to replace the correction term $\eps h^1$ from Proposition~\ref{PropFMc1} by
  \[
    (e^{\eps W})^*(g+\eps h^1)-g=\eps(h^1+\cL_W g)+\eps^2\Bigl(\frac12\cL_W^2 g+\cL_W h^1\Bigr)+\cO(\eps^3)
  \]
  for a suitable smooth vector field $W$ on $M$; the contributions of $\eps\cL_W g$, $\frac{\eps^2}{2}\cL_W^2 g$, and $\eps^2\cL_W h$ at $\hat M$ are all of size $\hat\rho^2$ (and sensitive only to the quadratic term $g_{(2)}$ in the Taylor expansion of $g$ at $\cC$ as well as the size $r$ leading order term of $h$ at $\pa M_\circ$). The advantage of the approach chosen above is that it will generalize easily to a unified construction of corrections at later stages of our argument which require infinitesimal changes of the parameters of the black hole which (unlike those encountered above) are not pure gauge.
\end{rmk}

\begin{rmk}[Less precise description of the leading order terms]
\label{RmkFhM1LessPrecise}
  We keep track of the nature of the leading order terms in parts~\eqref{ItFhM1hM}--\eqref{ItFhM1Mc} in order to limit the power of logarithms appearing below. However, this precision is not actually necessary for completing the construction of a formal solution if one is willing to acquire more logarithmic terms early on in the construction; and in later stages of the construction, we shall abandon this quest.
\end{rmk}

\subsection{Second correction at \texorpdfstring{$M_\circ$}{the blown-up background spacetime}}
\label{SsFMc2}

By Proposition~\ref{PropFhM1}, there exist
\begin{equation}
\label{EqFMc2f0f1}
\begin{split}
  f_0&\in\cA_\phg^{(-1,1)_+}(M_\circ;\upbeta_\circ^*S^2 T^*M), \\
  f_1&\in\hat\rho^{-1}\CI(M_\circ;\upbeta_\circ^*S^2 T^*M),
\end{split}
\end{equation}
so that\footnote{The term $f_1$ is independent of the choice of defining function $\rho_\circ$, whereas changing from $\rho_\circ$ to $a\rho_\circ$ where $0<a\in\CI(\wt M)$ changes $f_0$ to $f_0+(\log a|_{M_\circ})f_1$. This does not affect the memberships~\eqref{EqFMc2f0f1}.}
\[
  \Err_1 \equiv -\bigl( \eps^2 f_0 + \eps^2(\log\rho_\circ)f_1 \bigr) \bmod \cA_\phg^{(1,1)_+,(3,2)_+}(\wt M\setminus\wt K^\circ;S^2\wt T^*\wt M).
\]
We take $\hat\rho=r$ near $\pa M_\circ$ and set $\rho_\circ=\frac{\eps}{\hat\rho}$; then the coefficient of $r^{-1}\log r$ of $f_0$ is a sum of terms of scalar and vector type $2$, and the coefficient of $r^{-1}$ of $f_1$ is of scalar type $0$. The important common feature in the sequel is that \emph{the leading order terms of $f_0$ and $f_1$ at $\pa M_\circ$ have vanishing scalar and vector type $1$ components}.

\begin{prop}[Second correction at $M_\circ$]
\label{PropFMc2}
  There exists a symmetric 2-tensor $h\in\cA_\phg^{(3,1)_+\cup(4,3)_+,(2,1)_+}(\wt M;S^2\wt T^*\wt M)$ with support in the domain of influence of a compact subset of $\cU^\circ$ so that $\wt g^2:=\wt g_1+h$, which is a $((3,1)_+\cup(4,3)_+,(1,0)_+)$-smooth total family, satisfies
  \[
    \Err^2 \in \cA_\phg^{(1,1)_+\cup(2,3)_+,(3,2)_+}(\wt M\setminus\wt K^\circ;S^2\wt T^*\wt M).
  \]
  Moreover, the leading order term of $\Err^2$ at $\hat M$ is $\eps(\log\hat\rho)\sfe^{-1}\Err^2_{(1,1)}$ where
  \[
    \Err^2_{(1,1)}:=\wh{D_{\hat g_b}\Ric}(0)(\sfe \hat h)
  \]
  for some $\hat h\in\cA^{(-2,0)_+}(\hat M;S^2\wt T^*\wt M)$; and $\sfe^{-1}\Err^2_{(1,1)}\in\cA_\phg^{(2,2)_+}(\hat M\setminus\wt K^\circ;S^2\wt T^*\wt M)$.
\end{prop}
\begin{proof}
  The second Bianchi identity for $\wt g_1$, i.e.\ $\delta_{\wt g_1}\sfG_{\wt g_1}\Err_1=0$, implies
  \[
    0 = \delta_g\sfG_g \bigl( f_0 + (\log\rho_\circ)f_1 \bigr) = \delta_g\sfG_g\bigl[ (\log\eps)f_1 + (f_0-(\log\hat\rho)f_1) \bigr]
  \]
  for all $\eps>0$, and therefore
  \[
    \delta_g\sfG_g f_1 = 0,\qquad
    \delta_g\sfG_g\bigl(f_0-(\log\hat\rho)f_1\bigr) = 0.
  \]
  We may thus apply Theorem~\ref{ThmAc}, strengthened (following Remark~\ref{RmkAcBdyFormalLog}) as in the proof of Proposition~\ref{PropFMc1} using the absence of scalar and vector type $1$ components of the leading order term of $f_1$, to obtain
  \[
    h_1 \in \cA_\phg^{(1,0)_+}(M_\circ;\upbeta_\circ^*S^2 T^*M),\qquad
    (D_g\Ric-\Lambda)(h_1)=f_1.
  \]

  We then note that
  \begin{align*}
    f'_0 &:= f_0 - (\log\hat\rho)f_1 + (D_g\Ric-\Lambda)\bigl((\log\hat\rho)h_1\bigr) = f_0 + [D_g\Ric,\log\hat\rho]h_1 \\
      &\in \cA_\phg^{(-1,1)_+}(M_\circ;\upbeta_\circ^*S^2 T^*M)
  \end{align*}
  has the same $r^{-1}\log r$ leading order term as $f_0$, and still lies in $\delta_g\sfG_g$. Since application of Theorem~\ref{ThmRicSolv} (with $z=1$) to the $r^{-1}$ terms of $f'_0$ does not produce additional logarithmic terms to leading order, Theorem~\ref{ThmAc} produces
  \[
    h'_0 \in \cA_\phg^{(1,1)_+\cup(2,3)_+}(M_\circ;\upbeta_\circ^*S^2 T^*M),\qquad
    (D_g\Ric-\Lambda)(h'_0)=f'_0.
  \]

  Altogether then, we have
  \begin{align*}
    (D_g\Ric-\Lambda)\bigl(h'_0 + (\log\rho_\circ)h_1 \bigr) &= (D_g\Ric-\Lambda)\bigl(h'_0 - (\log\hat\rho)h_1\bigr) + (\log\eps)(D_g\Ric-\Lambda)(h_1) \\
      &= f_0 - (\log\hat\rho)f_1 + (\log\eps)f_1 \\
      &= f_0 + (\log\rho_\circ)f_1.
  \end{align*}
  We then set
  \begin{equation}
  \label{EqFMc2h}
    \wt g^2 := \wt g_1 + h,\qquad
    h := \chi_\circ\eps^2\bigl(h'_0 + (\log\rho_\circ)h_1\bigr) \in \cA_\phg^{(3,1)_+\cup(4,3)_+,(2,1)_+}(\wt M;S^2\wt T^*\wt M).
  \end{equation}
  Since $\wt g^1\equiv\wt\upbeta^*g\bmod\cA_\phg^{(0,0)\cup(2,0)_+,(1,0)_+}(\wt M\setminus\wt K^\circ;S^2\wt T^*\wt M)$ (from Proposition~\ref{PropFhM1}), this is a $((2,0)_+\cup(4,3)_+,(1,0)_+)$-smooth total family, and thus
  \[
    \Err^2 = \Ric(\wt g^2)-\Lambda\wt g^2 \in \cA_\phg^{(0,0)_+\cup(2,3)_+,(1,0)_+}(\wt M\setminus\wt K^\circ;S^2\wt T^*\wt M).
  \]
  Being a correction of $\wt g_1$ by the tensor $h$ which eliminates the leading order terms of $\Err_1$ at $M_\circ$ and which only contributes terms at $\hat M$ with index set $(1,1)_+\cup(2,3)_+$ (cf.\ \eqref{EqFMc2h}), we in fact have
  \[
    \Err^2 \in \cA_\phg^{(1,1)_+\cup(2,3)_+,(3,2)_+}(\wt M\setminus\wt K^\circ;S^2\wt T^*\wt M).
  \]

  The leading order logarithmic term of $\Err^2$ at $\hat M$ is the sum of that of $\Err_1$ (described in Proposition~\ref{PropFhM1}\eqref{ItFhM1hM}) and an additional contribution which is $\wh{D_{\hat g_b}\Ric}(0)$ applied to the $\eps^3\log\hat\rho$ leading order term of $h$ (from $h'_0$ in~\eqref{EqFMc2h}).
\end{proof}

\subsection{Second correction at \texorpdfstring{$\hat M$}{the front face of the total gluing spacetime}}
\label{SsFhM2}

In the notation of Proposition~\ref{PropFMc2}, we have
\[
  \Err^2 - \hat\chi\eps(\log\hat\rho)\sfe^{-1}(\Err^2_{(1,1)}) \in \cA_\phg^{(1,0)_+\cup(2,3)_+,(3,2)_+}(\wt M\setminus\wt K^\circ;S^2\wt T^*\wt M).
\]
We shall thus first solve away this logarithmic term in~\S\ref{SssFhM2Log} before dealing with the remaining $\cO(\eps)$ error at $\hat M^\circ$ in~\S\ref{SssFhM2Rem}.

\subsubsection{Solving away the logarithmic leading order term}
\label{SssFhM2Log}

We have the following analogue of Corollary~\ref{CorFhM1CokerParam}, which follows from the analogue of Lemma~\ref{LemmaFhM1Coker} for the error term $\Err^2_{(1,1)}|_{\hat M_{c(t)}}\in\wh{D_{\hat g_b}\Ric}(0)(\cA^{(-2,0)_+}(\hat M_{c(t)}))$, $t\in I$, the key calculation being~\eqref{EqFhM1CokerIBP}:

\begin{lemma}[Parameters for elimination of $\Err^2_{(1,1)}$]
\label{LemmaFhM2Coker}
  We use the notation of Theorem~\usref{ThmAhKCoker}. There exist unique functions
  \[
    \hat c^2 \in \CI(I;\R^3),\qquad
    \begin{cases} \bha=0: & \dot\bha^2=0, \\ \bha\neq 0: & \dot\bha^2 \in \CI(I;(\R\bha)^\perp), \end{cases}
  \]
  so that
  \begin{alignat*}{3}
    \la\Err^2_{(1,1)}(t),-\ra_{L^2(\hat X_b)} &= \ell_{b,\rm Kerr}(0,\dot\bha^2(t)) &&\in (\cK_{b,\rm Kerr}^*)^*, \\
    \la\Err^2_{(1,1)}(t),-\ra_{L^2(\hat X_b)} &= \ell_{b,\rm COM}(\hat c^2(t)) &&\in (\cK_{b,\rm COM}^*)^*.
  \end{alignat*}
\end{lemma}

Furthermore, the second Bianchi identity for $\wt g^2$ implies that
\[
  \wh{\delta_{\hat g_b}\sfG_{\hat g_b}}(0)\bigl(\Err^2_{(1,1)}(t)\bigr) = 0\qquad\text{for all}\ t\in I.
\]
Similarly to~\S\ref{SsFhM1}, the strategy is now roughly to modify $\wt g^2$ by pure gauge terms $\eps(\log\eps)h_{b,\hat c}$ and $\eps^2(\log\eps)g'_b(0,\dot\bha)$ for suitable functions $\hat c,\dot\bha$ so as to generate $\eps^3(\log\eps)$ corrections making $\Err^2_{(1,1)}$ orthogonal to the full cokernel (i.e.\ so that~\eqref{EqAhPhgEquiv} holds for the modified $\Err^2_{(1,1)}$). There are two main differences to~\S\ref{SsFhM1}: first, such modifications leave the $\hat M$-model unchanged; second, the nonlinear terms of the Ricci curvature operator interact in a different manner due to different powers of $\eps$ and $\log\eps$ appearing here.

\textbf{Calculations for the Kerr model.} Consider functions $\hat c\in\CI(I;\R^3)$, and $\dot\bha\in\CI(I;(\R\bha)^\perp)$ if $\bha=0$ and $\dot\bha\equiv 0$ otherwise; we shall determine them later. Consider then
\[
  \Phi_{2,\hat c(t)} \colon (\eps,\hat x) \mapsto \bigl(\eps,\hat x+\eps(\log\eps)\hat c(t)\bigr).
\]
Anticipating the need to produce Lorentz boosts and modulations of the angular momentum vector at the order $\eps^2(\log\eps)$, consider, near $t=t_0\in I$ and with $\hat t=\frac{t-t_0}{\eps}$ as in~\eqref{EqFhM1hatCoord}, the symmetric 2-tensor $\hat g_{2,b,\hat c,\dot\bha}$ defined by
\begin{align}
  \hat g_{2,b,\hat c,\dot\bha}&:=\Bigl(\Phi_{2,\hat c(t)}^*\bigl(\hat g_b+\eps^2(\log\eps)\bigl(\breve h_{b,2\hat c'(t)}+\cL_{\Omega(\dot\bha(t))}\hat g_b\bigr)\bigr) \Bigr)_{t\in I} \nonumber\\
  &\equiv \Phi_{2,\hat c(t_0)}^*\Bigl(\hat g_b+\eps^2(\log\eps)\bigl(\hat t h_{b,2\hat c'(t_0)}+\breve h_{b,2\hat c'(t_0)} + \cL_{\Omega(\dot\bha(t_0))}\hat g_b\bigr) \nonumber\\
  &\hspace{8em} + \eps^3(\log\eps)\Bigl(\frac{\hat t^2}{2}h_{b,2\hat c''(t_0)}+\hat t\breve h_{b,2\hat c''(t_0)} + \hat t\cL_{\Omega(\dot\bha'(t_0))}\hat g_b\Bigr)\Bigr) \nonumber\\
\label{EqFhM2gbPullback}
\begin{split}
  &\equiv \hat g_b + \eps(\log\eps)\cL_{\hat c(t_0)\cdot\pa_{\hat x}}\hat g_b + \eps^2(\log\eps)\cL_{\hat t\hat c'(t_0)\cdot\pa_{\hat x}+(\hat c'(t_0)\cdot\hat x)\pa_{\hat t}+\Omega(\dot\bha(t_0))}\hat g_b \\
  &\quad\qquad + \frac12\eps^2(\log\eps)^2\cL_{\hat c(t_0)\cdot\pa_{\hat x}}^2\hat g_b + \eps^3(\log\eps)\Bigl(\frac{\hat t^2}{2}h_{b,2\hat c''(t_0)} + \hat t\breve h_{b,2\hat c''(t_0)} + \hat t\hat g_b'(0,\dot\bha'(t_0))\Bigr) \\
  &\quad\qquad + \eps^3(\log\eps)^2\cL_{\hat c(t_0)\cdot\pa_{\hat x}}\cL_{\hat t\hat c'(t_0)\cdot\pa_{\hat x}+(\hat c'(t_0)\cdot\hat x)\pa_{\hat t}+\Omega(\dot\bha(t_0))}\hat g_b \\
  &\quad\qquad + \frac16 \eps^3(\log\eps)^3\cL_{\hat c(t_0)\cdot\pa_{\hat x}}^3\hat g_b \bmod \cA_\phg^{(4,4)_+};
\end{split}
\end{align}
here, we write $\cA_\phg^{(4,4)_+}=\cA_\phg^{(4,4)_+}([0,1)_\eps;\CI(\R^4_{\hat t,\hat x};S^2 T^*\R^4))$, and we used Lemma~\ref{LemmaFhM1AxisLie}. Using Lemma~\ref{LemmaFhM1Lie}, we then compute
\begin{subequations}
\begin{equation}
\label{EqFhM2gbRic1}
\begin{split}
  \Ric(\hat g_{2,b,\hat c,\dot\bha})&\equiv \eps^3(\log\eps)\Bigl[ D_{\hat g_b}\Ric\Bigl(\frac{\hat t^2}{2}h_{b,2\hat c''(t_0)} + \hat t\breve h_{b,2\hat c''(t_0)} + \hat t\hat g_b'(0,\dot\bha'(t_0))\Bigr)\Bigr] \\
  &\quad + \eps^3(\log\eps)^2\Bigl[ D_{\hat g_b}^2\Ric\bigl( \cL_{\hat c(t_0)\cdot\pa_{\hat x}}\hat g_b, \cL_{\hat t\hat c'(t_0)\cdot\pa_{\hat x}+(\hat c'(t_0)\cdot\hat x)\pa_{\hat t}+\Omega(\dot\bha(t_0))}\hat g_b\bigr) \\
  &\quad \hspace{5.7em} + D_{\hat g_b}\Ric\bigl( \cL_{\hat c(t_0)\cdot\pa_{\hat x}}\cL_{\hat t\hat c'(t_0)\cdot\pa_{\hat x}+(\hat c'(t_0)\cdot\hat x)\pa_{\hat t}+\Omega(\dot\bha(t_0))}\hat g_b\bigr) \Bigr] \bmod \cA_\phg^{(4,4)_+}.
\end{split}
\end{equation}
Writing $V=\hat c(t_0)\cdot\pa_{\hat x}$ and $W=\hat t\hat c'(t_0)\cdot\pa_{\hat x}+(\hat c'(t_0)\cdot\hat x)\pa_{\hat t}+\Omega(\dot\bha(t_0))$, the identity~\eqref{EqFhM1LiePolarized} implies that the coefficient of $\eps^3(\log\eps)^2$ is $\frac12$ times
\begin{equation}
\label{EqFhM2gbRic2}
  D_{\hat g_b}\Ric(\cL_V\cL_W\hat g_b - \cL_W\cL_V\hat g_b) = D_{\hat g_b}\Ric(\cL_{[V,W]}\hat g_b) = 0.
\end{equation}
\end{subequations}
Thus, the Ricci tensor of $\hat g_{2,b,\hat c,\dot\bha}$ produces $\eps^3(\log\eps)$ terms which we shall use below to eliminate the cokernel of $\wh{D_{\hat g_b}\Ric}(0)$, cf.\ Theorem~\ref{ThmAhKCoker}.

\begin{lemma}[Further properties of $\hat g_{2,b,\hat c,\dot\bha}$]
\label{LemmaFhM2Further}
  We use the notation of Lemma~\usref{LemmaFhM1CombProp}. Suppose first that $\bha\neq 0$.
  \begin{enumerate}
  \item{\rm (Decay of coefficients.)} The components of $\hat g_{2,b,\hat c,\dot\bha}(\eps,t,\hat x;\dd\hat t,\dd\hat x)$, which are smooth functions on $(0,1)_\eps\times I_t\times\R^3_{\hat x}$, lift to elements of $\cA_\phg^{(0,0)_+,(0,0)\cup(3,1)_+}(\wt M)$ (with index sets referring to $\hat\cM$ and $\cM_\circ$, in this order) which differ from the corresponding components of $\hat g_b$ by terms in $\cA_\phg^{(1,1)_+,(3,1)_+}(\wt M)$.
  \item{\rm (Total pullback.)} Set
    \[
      \Phi_{2,\hat c,\dot\bha} \colon (\eps,t,\hat x) \mapsto \bigl(\eps, t - \eps^3(\log\eps)\hat c'(t)\cdot\hat x, e^{-\eps^2(\log\eps)\Omega(\dot\bha(t))}\hat x-\eps(\log\eps)\hat c(t)\bigr).
    \]
    Then $\Phi_{2,\hat c,\dot\bha}$ lifts to a conormal diffeomorphism of $\wt M$ with the property that the components of $(\Phi_{2,\hat c,\dot\bha})^*\hat g_{2,b,\hat c,\dot\bha}-\hat g_b$ lie in $\cA_\phg^{(3,2)_+,(2,1)_+}(\wt M)$.
  \end{enumerate}
  If $\bha=0$, the same conclusions hold upon dropping all terms and subscripts involving $\dot\bha$.
\end{lemma}
\begin{proof}
  The polyhomogeneity of the components of $\hat g_{2,b,\hat c,\dot\bha}$ follows from the smoothness of the coefficients of $\hat g_b$. The index set of $\hat g_{2,b,\hat c,\dot\bha}$ at $\hat\cM$ is $(0,0)_+$ since $f(\hat x+\eps(\log\eps)\hat c(t))\in\cA_\phg^{(0,0)_+}(\wt\cM\setminus\cM_\circ)$ when $f$ is a smooth function on $\R^3$, as follows by Taylor expanding $f$ around $\hat x$. Near $(\cM_\circ)^\circ$, we can use the coordinates $\eps$, $t$, $x=\eps\hat x$, so that $\Phi_{2,\hat c(t)}\colon(\eps,x)\mapsto(\eps,x+\eps^2(\log\eps)\hat c(t))$; we then write
  \[
    \hat g_{2,b,\hat c,\dot\bha} - \hat g_b = \bigl(\Phi_{2,\hat c(t)}^*\hat g_b - \hat g_b\bigr) + \eps^2(\log\eps)\Phi_{2,\hat c(t)}^*\bigl(\breve h_{b,2\hat c'(t)} + \cL_{\Omega(\dot\bha(t))}\hat g_b\bigr).
  \]
  Since the components of $\breve h_{b,2\hat c'(t)}+\cL_{\Omega(\dot\bha(t))}\hat g_b$ lie in $\rho_\circ\CI(\wt M\setminus\hat\cM)$, the second summand here has index set $(3,1)_+$ at $\cM_\circ$. If in the first summand we replaced $\hat g_b$ by $\hat{\ubar g}$, it would vanish; therefore, we may replace $\hat g_b$ by $\hat g_b-\hat{\ubar g}$, whose components lie in $\rho_\circ\CI(\wt M\setminus\hat\cM)$, and thus the pullback along $\Phi_{2,\hat c(t)}$ (which equals the identity map at $\eps=0$ modulo a $\cA_\phg^{(2,1)}$ correction) lies in $\cA_\phg^{(3,1)_+}$ as well.

  For the second part, the index set at $\cM_\circ$ is a consequence of the first part and the fact that, in the coordinates $\eps,t,x$ near $(\cM_\circ)^\circ$, the map $\Phi_{2,\hat c,\dot\bha}$ is the identity map plus corrections of class $\cA^{(2,1)_+}_\phg$. Near a fiber of $\hat\cM^\circ$ and using $\hat t=\frac{t-t_0}{\eps}$, we have
  \[
    \Phi_{2,\hat c,\dot\bha}(\eps,\hat t,\hat x) \equiv \bigl(\eps, \hat t - \eps^2(\log\eps)\hat c'(t_0+\eps\hat t)\cdot\hat x, \hat x - \eps(\log\eps)\hat c(t_0+\eps\hat t) - \eps^2(\log\eps)\Omega(\dot\bha(t_0+\eps\hat t)\bigr)
  \]
  modulo corrections of class $\cA^{(4,2)_+}$; thus, pullback along $\Phi_{2,\hat c,\dot\bha}$ cancels all terms in~\eqref{EqFhM2gbPullback} except for the $\eps^3(\log\eps)$ and $\eps^3(\log\eps)^2$ terms. In fact, all terms $\frac{1}{j!}\eps^j(\log\eps)^j\cL^j_{\hat c(t_0)\cdot\pa_{\hat x}}\hat g_b$, $j\in\N$, are canceled in this manner, and hence the index set of $\Phi_{2,\hat c,\dot\bha}^*\hat g_{2,b,\hat c,\dot\bha}-\hat g_b$ at $\hat M$ does not contain $(j,j)$, $j\in\N$.
\end{proof}

\textbf{Grafting into the total gluing spacetime; computation of the error term.} In the notation of Proposition~\ref{PropFMc2}, set
\[
  \wt g'_{2,\hat c,\dot\bha} := \wt g^2 + \hat\chi\sfe^{-1}(\hat g_{2,b,\hat c,\dot\bha}-\hat g_b) \equiv \wt g^2 \bmod \cA_\phg^{(1,1)_+,(3,1)_+}(\wt M\setminus\wt K^\circ;S^2\wt T^*\wt M).
\]
By the definition of $\wt g^2$ (and $h$) in Proposition~\ref{PropFMc2}, and recalling Proposition~\ref{PropFhM1}\eqref{ItFhM1hatMQuadr}, we have
\begin{equation}
\label{EqFhM2wtg2MhatLot}
\begin{split}
  &\sfe\wt g'_{2,\hat c,\dot\bha}(\eps,t,\hat x;\dd\hat t,\dd\hat x) \\
  &\quad \equiv \hat g_{2,b,\hat c,\dot\bha}(\eps,t,\hat x;\dd\hat t,\dd\hat x) + \eps^2 h_{(2)}(t,\hat x;\dd\hat t,\dd\hat x) + h \bmod \cA_\phg^{(3,0)_+\cup(4,3)_+,(0,0)\cup(1,0)_+} \\
  &\quad\equiv \hat g_{2,b,\hat c,\dot\bha}(\eps,t,\hat x;\dd\hat t,\dd\hat x) + \eps^2 h_{(2)}(t,\hat x;\dd\hat t,\dd\hat x) \bmod \cA_\phg^{(3,1)_+\cup(4,3)_+,(0,0)\cup(1,0)_+}
\end{split}
\end{equation}
for some $h_{(2)}$ with\footnote{The index set can be sharpened to $(-2,0)\cup(-1,0)_+$, but we do not need this level of precision here.} $\sfe^{-1}h_{(2)}\in\cA_\phg^{(-2,0)_+}(\hat M\setminus\wt K^\circ;S^2\wt T^*\wt M)$. Combining these two facts and using~\eqref{EqFhM2gbRic1}--\eqref{EqFhM2gbRic2} shows that
\begin{equation}
\label{EqFhM2Err2ca}
  \Err'_{2,\hat c,\dot\bha} := \Ric(\wt g'_{2,\hat c,\dot\bha})-\Lambda\wt g'_{2,\hat c,\dot\bha} \in \cA_\phg^{(1,1)_+\cup(2,3)_+,(3,2)_+}(\wt M\setminus\wt K^\circ;S^2\wt T^*\wt M).
\end{equation}
Furthermore, using the description~\eqref{EqFhM2wtg2MhatLot} and taking into account the contribution from the coupling of $\eps^2 h_{(2)}$ with the $\eps(\log\eps)$ term $\eps(\log\eps)\cL_{\hat c(t_0)\cdot\pa_{\hat x}}\hat g_b=\eps(\log\eps)h_{b,2\hat c(t_0)}$ in~\eqref{EqFhM2gbPullback}, the coefficient of the $\eps(\log\eps)$ leading order term of $\Err'_{2,\hat c,\dot\bha}$ at $\hat M$ is, at the fiber $\hat M_{c(t)}$, given by $\sfe^{-1}$ applied to
\begin{equation}
\label{EqFhM2Errca2}
\begin{split}
  \Err'_{2,\hat c,\dot\bha,(1,1)} &:= \Err^2_{(1,1)} + \wh{D_{\hat g_b}\Ric}(0)\Bigl(\frac{\hat t^2}{2}h_{b,2\hat c''(t)} + \hat t\breve h_{b,2\hat c''(t)} + \hat t\hat g_b'(0,\dot\bha'(t))\Bigr) \\
    &\qquad - \Lambda h_{b,2\hat c(t)} + D_{\hat g_b}^2\Ric(h_{b,2\hat c(t)},h_{(2)}(t)) \\
  &\quad \in \cA_\phg^{(2,2)_+}(\hat X_b;S^2\,\Ttsc^*_{\hat X_b}\hat M_b).
\end{split}
\end{equation}
The membership here is a consequence of~\eqref{EqFhM2Err2ca}, and the term $-\Lambda h_{b,2\hat c(t)}$ arises from the $\eps(\log\eps)\cL_{\hat c(t_0)\cdot\pa_{\hat x}}\hat g_b$ term of $\hat g_{2,b,\hat c,\dot\bha}$ in~\eqref{EqFhM2gbPullback}. We can simplify this expression using the following identity:

\begin{lemma}[Lie derivative and linearized Ricci]
\label{LemmaFhM2LieLinRic}
  Let $g$ be a metric, $h$ a symmetric 2-tensor, and $V$ a vector field. Then
  \[
    \cL_V\bigl( D_g\Ric(h)\bigr) = D_g^2\Ric(\cL_V g,h) + D_g\Ric(\cL_V h).
  \]
\end{lemma}
\begin{proof}
  Writing $e^{t V}$ for the time $t$ flow of $V$, the left hand side of the claimed identity is
  \begin{align*}
    \frac{\dd}{\dd t}\Big|_0\Bigl[(e^{t V})^*\Bigl(\frac{\dd}{\dd s}\Big|_0 \Ric(g+s h)\Bigr)\Bigr] &= \frac{\dd}{\dd t}\Big|_0\,\frac{\dd}{\dd s}\Big|_0 \Ric\bigl((e^{t V})^*g + s(e^{t V})^*h\bigr) \\
      &= \frac{\dd}{\dd t}\Big|_0\,\frac{\dd}{\dd s}\Big|_0 \Bigl( (e^{t V})^*\Ric(g) + s D_{(e^{t V})^*g}\Ric((e^{t V})^*h) \Bigr) \\
      &= \frac{\dd}{\dd t}\Big|_0 D_{(e^{t V})^*g}\Ric(h+t\cL_V h) \\
      &= D^2_g\Ric(\cL_V g,h) + D_g\Ric(\cL_V h).\qedhere
  \end{align*}
\end{proof}

Since $\eps^2\sfe(\Err^2)$ vanishes more than quadratically at $\hat M$, we have
\[
  D_{\hat g_b}\Ric(h_{(2)}(t))-\Lambda\hat g_b = \wh{D_{\hat g_b}\Ric}(0)(h_{(2)}(t))-\Lambda\hat g_b=0
\]
for all $t\in I$. Taking the Lie derivative of this along $\hat c(t)\cdot\pa_{\hat x}$ (on each $t$-level set separately) gives the identity
\begin{equation}
\label{EqFhM2h2Identity}
  D_{\hat g_b}^2\Ric(h_{b,2\hat c(t)},h_{(2)}(t)) - \Lambda h_{b,2\hat c(t)} + \wh{D_{\hat g_b}\Ric}(0)\bigl(\cL_{\hat c(t)\cdot\pa_{\hat x}}h_{(2)}(t)\bigr) = 0.
\end{equation}
Therefore, we can rewrite~\eqref{EqFhM2Errca2} as
\begin{equation}
\label{EqFhM2Err2}
  \Err'_{2,\hat c,\dot\bha,(1,1)} = \Err^2_{(1,1)} + \wh{D_{\hat g_b}\Ric}(0)\Bigl(\frac{\hat t^2}{2}h_{b,2\hat c''(t)} + \hat t\breve h_{b,2\hat c''(t)} + \hat t\hat g_b'(0,\dot\bha'(t)) - \cL_{\hat c(t)\cdot\pa_{\hat x}}h_{(2)}(t)\Bigr).
\end{equation}
For later use, we note that, by definition of $h_{(2)}$, the leading order term $(\hat r^2\sfe(h_{(2)}))|_{\pa\hat M}$ is equal to $g_{(2)}$ from Lemma~\ref{LemmaFMc1Metric}; therefore, near $|\hat x|^{-1}=0$, the coefficients (in the frame $\dd\hat t,\dd\hat x$) of $h_{(2)}(t)$ are quadratic polynomials in $\hat x$ modulo $\cA_\phg^{(-1,1)_+}$ errors, and therefore those of $\cL_{\hat c(t)\cdot\pa_{\hat x}}h_{(2)}(t)$ are linear functions of $\hat x$ modulo $\cA_\phg^{(0,1)_+}$ errors; since $\hat z^\mu\,\dd\hat z^\kappa\,\dd\hat z^\lambda\in\ker\wh{D_{\ubar g}\Ric}(0)$ where $\hat z=(\hat t,\hat x)$, this implies
\begin{equation}
\label{EqFhM2DRicLh2}
  \wh{D_{\hat g_b}\Ric}(0)\bigl(\cL_{\hat c(t)\cdot\pa_{\hat x}}h_{(2)}(t)\bigr) \in \cA_\phg^{(2,1)_+}(\hat X_b;S^2\,\Ttsc^*_{\hat X_b}\hat M_b),
\end{equation}
with smooth dependence on $t\in I$.

\textbf{Eliminating the cokernel; conclusion.} In the notation of Theorem~\ref{ThmAhKCoker}, consider now the map
\[
  \ell_{2,\rm COM} \colon \CI(I;\R^3) \ni \hat c \mapsto \la\Err'_{2,\hat c,\dot\bha,(1,1)},-\ra_{L^2(\hat X_b)} \in \CI\bigl(I;(\cK_{b,\rm COM}^*)^*\bigr);
\]
in view of~Theorem~\ref{ThmAhKCoker}\eqref{ItAhKCokerKerrPC}, this map is independent of the choice of $\dot\bha$, and by Lemma~\ref{LemmaFhM2Coker} it evaluates at $t\in I$ to
\begin{align*}
  \ell_{2,\rm COM}(\hat c)(t) &= \ell_{b,\rm COM}(\hat c^2(t)) + \ell_{b,\rm COM}(2\hat c''(t)) + \lambda(t,\hat c(t)), \\
  &\qquad \lambda(t,\hat c):=-\big\la \wh{D_{\hat g_b}\Ric}(0)\bigl(\cL_{\hat c\cdot\pa_{\hat x}}h_{(2)}(t)\bigr),-\big\ra_{L^2(\hat X_b)}\in(\cK_{b,\rm COM}^*)^*.
\end{align*}
The functional $\lambda(t,\hat c)$ is well-defined since by~\eqref{EqFhM2DRicLh2} the first factor has almost a full order of decay more than necessary for integration against an element of $\cK_{b,\rm COM}^*$. Since $I\times\R^3\ni(t,\hat c)\mapsto\lambda(t,\hat c)$ defines an element of $\CI(I;(\R^3)^*)$, the equation
\[
  \ell_{2,\rm COM}(\hat c) = 0
\]
is a nondegenerate linear second order ODE for $\hat c$ and thus has a solution
\[
  \hat c \in \CI(I;\R^3).
\]

Next, when $\bha\neq 0$, we still need to find $\dot\bha\in\CI(I;(\R\bha)^\perp)$ so that also $\ell_{2,\rm Kerr}(\dot\bha)=0$, where we define
\[
  \ell_{2,\rm Kerr} \colon \CI(I;(\R\bha)^\perp) \ni \dot\bha \mapsto \la\Err'_{2,\hat c,\dot\bha,(1,1)},-\ra_{L^2(\hat X_b)} \in \CI\bigl(I;(\cK_{b,\rm Kerr}^*)^*\bigr).
\]
Using Theorem~\ref{ThmAhKCoker} and the expression~\eqref{EqFhM2Err2}, we find that
\begin{align*}
  \ell_{2,\rm Kerr}(\dot\bha)(t) &= \ell_{b,\rm Kerr}(0,\dot\bha^2(t)) + \ell_{b,\rm Kerr}(0,\dot\bha'(t)) + \mu(t), \\
    &\qquad \mu(t):=-\la\wh{D_{\hat g_b}\Ric}(0)(\cL_{\hat c(t)\cdot\pa_{\hat x}}h_{(2)}(t)),-\ra\in(\cK_{b,\rm Kerr}^*)^*.
\end{align*}
Arguing as in~\eqref{EqFhM1CokerIBP} and the proof of Corollary~\ref{CorFhM1CokerParam}, we deduce that one can write $\mu(t)=\ell_{b,\rm Kerr}(0,\fq(t))$ where $\fq\in\CI(I;(\R\bha)^\perp)$. Therefore, the equation $\ell_{2,\rm Kerr}(\dot\bha)=0$ is a nondegenerate linear first order ODE for $\dot\bha$ which thus has a global solution
\[
  \dot\bha \in \CI(I;(\R\bha)^\perp)\qquad (\bha\neq 0).
\]

For these choices of $\hat c,\dot\bha$, the assumptions of Theorem~\ref{ThmAhPhg} are satisfied for $\Err'_{2,\hat c,\dot\bha,(1,1)}(t)$ for all $t\in I$, and therefore there exists
\[
  h \in \cA_\phg^{(0,3)_+}(\hat M\setminus\wt\cK^\circ;S^2\wt T_{\hat M}^*\wt M),
\]
i.e.\ $\sfe h\in\CI(I;\cA_\phg^{(0,3)_+}(\hat X_b;S^2\,\Ttsc^*_{\hat X_b}\hat M_b))$, so that $\wh{D_{\hat g_b}\Ric}(0)(\sfe h(t))=-\Err'_{2,\hat c,\dot\bha,(1,1)}(t)$ for all $t\in I$. The metric
\[
  \wt g_{2,\hat c,\dot\bha} := \wt g'_{2,\hat c,\dot\bha}+\eps^3(\log\hat\rho)\hat\chi h,
\]
which is a $((1,1)_+,(1,0)_+\cup(3,3)_+)$-smooth total family, thus satisfies
\[
  \Ric(\wt g_{2,\hat c,\dot\bha}) - \Lambda\wt g_{2,\hat c,\dot\bha} \in \cA_\phg^{(1,0)\cup(2,3)_+,(3,3)_+}(\wt M\setminus\wt K^\circ;S^2\wt T^*\wt M).
\]
As desired, this eliminates the $\eps(\log\hat\rho)$ leading order term of $\Err^2$ at $\hat M$ from Proposition~\ref{PropFMc2}, at the expense of additional logarithmic factors at $M_\circ$. In a last step, we re-center the small black hole by passing from $\wt g_{2,\hat c,\dot\bha}$ to
\[
  \wt g'_2 := \Phi_{2,\hat c,\dot\bha}^*\wt g_{2,\hat c,\dot\bha} = \Phi_{2,\hat c,\dot\bha}^*\bigl(\hat\chi\sfe^{-1}\hat g_{2,b,\hat c,\dot\bha} + (\wt g^2-\hat\chi\sfe^{-1}\hat g_b) + \eps^3(\log\hat\rho)\hat\chi h\bigr).
\]
By Lemma~\ref{LemmaFhM2Further}, noting that $\eps^3(\log\hat\rho)\hat\chi h\in\cA_\phg^{(3,1),(3,3)_+}(\wt M\setminus\wt K^\circ;S^2\wt T^*\wt M)$, and recalling~\eqref{EqFhM2wtg2MhatLot}--\eqref{EqFhM2Err2ca}, we have proved the following result:

\begin{prop}[Second correction at $\hat M$: first step]
\label{PropFhM2first}
  There exists a total family $\wt g'_2$ which is $((3,2)_+,(1,0)_+\cup(3,3)_+)$-smooth, equal to $\wt g_0$ at $\hat M$ and $M_\circ$, and which satisfies
  \begin{equation}
  \label{EqFhM2Err}
    \Err_2' := \Ric(\wt g_2') - \Lambda\wt g_2' \in \cA_\phg^{(1,0)\cup(2,3)_+,(3,3)_+}(\wt M\setminus\wt K^\circ;S^2\wt T^*\wt M);
  \end{equation}
  moreover, $\wt g'_2=g$ outside the domain of influence of a compact subset of $\cU^\circ$. Furthermore, $\wt g'_2$ is equal to the Kerr metric $\hat g_b$ along $\hat M$ modulo $\cO(\eps^2)$ in the sense of Proposition~\usref{PropFhM1}\eqref{ItFhM1hatMQuadr}, and the leading order term $h_{(2)}(t):=\eps^{-2}(\sfe\wt g'_2(t)-\hat g_b)|_{\hat M(t)}\in\cA_\phg^{(-2,0)_+}(\hat X_b;S^2\,\Ttsc^*_{\hat X_b}\hat M_b)$ (with smooth dependence on $t\in I$) satisfies~\eqref{EqFhM2DRicLh2} for all $\hat c$.
\end{prop}

\subsubsection{Solving away the remaining term}
\label{SssFhM2Rem}

We next solve away the leading order term of $\Err_2'$ at $\hat M$, which is
\[
  \Err_{2,(1,0)}' := \sfe\bigl((\eps^{-1}\Err_2')|_{\hat M}\bigr) \in \CI\bigl(I;\cA_\phg^{(2,3)_+}(\hat X_b;S^2\,\Ttsc^*_{\hat X_b}\hat M_b)\bigr).
\]
Unlike in previous arguments at $\hat M$, we no longer record any special properties of
\[
  \la\Err'_{2,(1,0)},-\ra \in \CI\bigl(I;(\cK_{b,\rm tot}^*)^*\bigr),
\]
and thus need to allow for the full range of linearized Kerr metrics when eliminating the cokernel.

As usual, we begin with a computation for the Kerr spacetime. Let $\hat c\in\CI(I;\R^3)$ and $\dot b\in\CI(I;\R\times\R^3)$, and set $\Phi'_{2,\hat c(t)}\colon(\eps,\hat x)\mapsto(\eps,\hat x+\eps\hat c(t))$. Fix $t_0\in I$ and write $\hat t=\frac{t-t_0}{\eps}$. Then the tensor
\begin{align*}
  \hat g'_{2,b,\hat c,\dot b} & := \Bigl((\Phi'_{2,\hat c(t)})^*\bigl(\hat g_b+\eps^2\bigl(\breve h_{b,2\hat c'(t)} + \hat g'_b(\dot b(t))\bigr)\bigr)\Bigr)_{t\in I} \\
  &\qquad \equiv (\Phi'_{2,\hat c(t_0)})^*\Bigl(\hat g_b + \eps^2\bigl(\hat t h_{b,2\hat c'(t_0)}+\breve h_{b,2\hat c'(t_0)}+\hat g'_b(\dot b(t_0))\bigr) \\
  &\qquad \hspace{7em} + \eps^3\Bigl(\frac{\hat t^2}{2}h_{b,2\hat c''(t_0)}+\hat t\breve h_{b,2\hat c''(t_0)} + \hat t\hat g'_b(\dot b'(t_0))\Bigr) \Bigr)_{t\in I} \\
  &\qquad \equiv \hat g_b + \eps\cL_{\hat c(t_0)\cdot\pa_{\hat x}}\hat g_b + \eps^2\Bigl(\frac{1}{2}\cL_{\hat c(t_0)\cdot\pa_{\hat x}}^2\hat g_b + \cL_{\hat t\hat c'(t_0)\cdot\pa_{\hat x}+(\hat c'(t_0)\cdot\hat x)\pa_{\hat t}}\hat g_b + \hat g'_b(\dot b(t_0))\Bigr) \\
  &\qquad\hspace{3em} + \eps^3\Bigl(\frac{\hat t^2}{2}h_{b,2\hat c''(t_0)} + \hat t\breve h_{b,2\hat c''(t_0)} + \hat t\hat g'_b(\dot b'(t_0))\Bigr) \\
  &\qquad\hspace{3em} + \eps^3\Bigl(\frac{1}{6}\cL_{\hat c(t_0)\cdot\pa_{\hat x}}^3\hat g_b + \cL_{\hat c(t_0)\cdot\pa_{\hat x}}\bigl(\hat t h_{b,2\hat c'(t_0)}+\breve h_{b,2\hat c'(t_0)}+\hat g'_b(\dot b(t_0))\bigr)\Bigr) \bmod \eps^4\CI
\end{align*}
satisfies
\begin{align*}
  \Ric(\hat g'_{2,b,\hat c,\dot b}) &\equiv \eps^3 D_{\hat g_b}\Ric\Bigl(\frac{\hat t^2}{2}h_{b,2\hat c''(t_0)} + \hat t\breve h_{b,2\hat c''(t_0)} + \hat t\hat g'_b(\dot b'(t_0))\Bigr) \\
  & \quad + \eps^3\Bigl[ D_{\hat g_b}^2\Ric\bigl(\cL_{\hat c(t_0)\cdot\pa_{\hat x}}\hat g_b, \cL_{\hat t\hat c'(t_0)\cdot\pa_{\hat x}+(\hat c'(t_0)\cdot\hat x)\pa_{\hat t}}\hat g_b + \hat g'_b(\dot b(t_0))\bigr) \\
  &\quad \hspace{2.05em} + D_{\hat g_b}\Ric\Bigl(\cL_{\hat c(t_0)\cdot\pa_{\hat x}}\bigl(\cL_{\hat t\hat c'(t_0)\cdot\pa_{\hat x}+(\hat c'(t_0)\cdot\hat x)\pa_{\hat t}}\hat g_b+\hat g'_b(\dot b(t_0))\bigr)\Bigr)\Bigr] \bmod \eps^4\CI.
\end{align*}
Using the identity~\eqref{EqFhM1LiePolarized} as around~\eqref{EqFhM2gbRic2}, the expression in square brackets is
\[
  D_{\hat g_b}^2\Ric\bigl(\cL_{\hat c(t_0)\cdot\pa_{\hat x}}\hat g_b,\hat g'_b(\dot b(t_0))\bigr) + D_{\hat g_b}\Ric\bigl(\cL_{\hat c(t_0)\cdot\pa_{\hat x}}\hat g'_b(\dot b(t_0))\bigr);
\]
but by Lemma~\ref{LemmaFhM2LieLinRic}, this is further equal to $\cL_{\hat c(t_0)\cdot\pa_{\hat x}}\bigl[D_{\hat g_b}\Ric\bigl(\hat g'_b(\dot b(t_0))\bigr)\bigr]$ and thus vanishes in view of~\eqref{EqGKLinRic}. Therefore,
\[
  \Ric(\hat g'_{2,b,\hat c,\dot b}) \equiv \eps^3\Bigl[ D_{\hat g_b}\Ric\Bigl(\frac{\hat t^2}{2}h_{b,2\hat c''(t_0)} + \hat t\breve h_{b,2\hat c''(t_0)} + \hat t\hat g'_b(\dot b'(t_0))\Bigr)\Bigr] \bmod \eps^4\CI.
\]

We now let
\[
  \wt g'_{2,\hat c,\dot b} := \wt g'_2 + \hat\chi\sfe^{-1}(\hat g'_{2,b,\hat c,\dot b}-\hat g_b).
\]
The second summand here lies in $\cA_\phg^{(1,0),(3,0)}(\wt M\setminus\wt K^\circ;S^2\wt T^*\wt M)$, and we have
\[
  \Ric(\wt g'_{2,\hat c,\dot b})-\Lambda\wt g'_{2,\hat c,\dot b} \in \cA_\phg^{(1,0)\cup(2,3)_+,(3,3)_+}(\wt M\setminus\wt K^\circ;S^2\wt T^*\wt M),
\]
cf.\ \eqref{EqFhM2Err}, with leading order term $\eps$ times $\sfe^{-1}$ applied to
\begin{equation}
\label{EqFhM2Err2p}
\begin{split}
  \Err'_{2,\hat c,\dot b} &:= \Err'_{2,(1,0)} + D_{\hat g_b}\Ric\Bigl(\frac{\hat t^2}{2}h_{b,2\hat c''(t_0)} + \hat t\breve h_{b,2\hat c''(t_0)} + \hat t\hat g'_b(\dot b'(t_0)) - \cL_{\hat c(t)\cdot\pa_{\hat x}}h_{(2)}(t)\Bigr) \\
    &\in \CI\bigl(I;\cA_\phg^{(2,3)_+}(\hat X_b;S^2\,\Ttsc^*_{\hat X_b}\hat M_b)\bigr),
\end{split}
\end{equation}
as follows from the same computations as in~\eqref{EqFhM2Errca2}--\eqref{EqFhM2Err2}. We also recall the membership~\eqref{EqFhM2DRicLh2}. By Theorem~\ref{ThmAhKCoker}, the equation
\[
  \la\Err'_{2,\hat c,\dot b},-\ra_{L^2(\hat X_b)} = 0 \in \CI\bigl(I;(\cK_{b,\rm COM}^*)^*\bigr)
\]
is a nondegenerate linear second order ODE for $\hat c$ which thus also has a global solution $\hat c\in\CI(I;\R^3)$. With $\hat c$ now fixed, the equation
\[
  \la\Err'_{2,\hat c,\dot b},-\ra_{L^2(\hat X_b)} = 0 \in \CI\bigl(I;(\cK_{b,\rm Kerr}^*)^*\bigr)
\]
is a nondegenerate linear first order ODE for $\dot b$ which thus has a global solution $\dot b\in\CI(I;\R\times\R^3)$. For these choices of $\hat c,\dot b$, we can now apply Theorem~\ref{ThmAhPhg} and obtain
\[
  h \in \cA_\phg^{(0,4)_+}(\hat M\setminus\wt K^\circ;S^2\wt T^*_{\hat M}\wt M)
\]
so that
\[
  \wh{D_{\hat g_b}\Ric}(0)(\sfe h(t)) = -\Err'_{2,\hat c,\dot b}(t)\qquad\forall\,t\in I.
\]
Finally, we set $\wt g_2=(\Phi'_{2,\hat c})^*(\wt g'_{2,\hat c,\dot b}+\eps^3\hat\chi h)$ where the pullback by $\Phi'_{2,\hat c}\colon(\eps,t,\hat x)\mapsto(\eps,t,\hat x-\eps\hat c(t))$ applied to $\wt g'_{2,\hat c,\dot b}$ is equal to $\hat g_b$ along $\hat M$ up to $\cO(\eps^2)$ errors. In summary:

\begin{prop}[Second correction at $\hat M$: second step]
\label{PropFhM2}
  There exists a $((3,2)_+,(1,0)_+\cup(3,4)_+)$-smooth total family $\wt g_2$ which is equal to $\wt g_0$ at $\hat M$ and $M_\circ$ and which satisfies\footnote{Recall the notation~\eqref{EqBgPhgIndexSets2}.}
  \[
    \Err_2 := \Ric(\wt g_2) - \Lambda\wt g_2 \in \cA_\phg^{(2,3)_+,(3,4)_{+\times}}(\wt M\setminus\wt K^\circ;S^2\wt T^*\wt M);
  \]
  moreover, $\wt g_2=g$ outside the domain of influence of a compact subset of $\cU^\circ$. Furthermore, $\wt g_2$ is equal to the Kerr metric $\hat g_b$ along $\hat M$ modulo $\cO(\eps^2)$ in the sense of Proposition~\usref{PropFhM1}\eqref{ItFhM1hatMQuadr}, and the subleading term $h_{(2)}(t):=\eps^{-2}(\sfe\wt g_2(t)-\hat g_b)|_{\hat M(t)}\in\cA_\phg^{(-2,0)_+}(\hat X_b;S^2\,\Ttsc^*_{\hat X_b}\hat M_b)$ satisfies~\eqref{EqFhM2DRicLh2} for all $\hat c$.
\end{prop}

\begin{rmk}[Linearized Kerr contributions]
\label{RmkFhM2LinKerr}
  The term $h_{(2)}(t)$ in Proposition~\ref{PropFhM2} is not pure gauge when $g_{(2)}(t)\neq 0$, or equivalently when $\Riem(g)\neq 0$ at $c(t)\in\cC$; this is why we may as well admit correction terms $\eps^2\hat g'_b(\dot b(t))$ in $\hat g'_{2,b,\hat c,\dot b}$ which are not pure gauge either. By contrast, we carefully avoided $\cO(\eps^2\log\eps)$ metric perturbations which are not pure gauge in~\S\ref{SssFhM2Log}.
\end{rmk}

\subsection{Third correction at \texorpdfstring{$M_\circ$}{the blown-up background spacetime}}
\label{SsFMc3}

At this point, no special arguments are required at $M_\circ$ anymore. We thus record the following general result:

\begin{prop}[Correction at $M_\circ$]
\label{PropFMc3Gen}
  Let $k\in\N$.\footnote{At the current stage, we apply this result with $k=3$.} Suppose $\wt g_{k-1}$ is a $((1,*),(1,*))$-smooth total family with $\wt g_{k-1}=\wt g_0$ at $\hat M\cup M_\circ$. Suppose moreover that
  \[
    \Err_{k-1} := \Ric(\wt g_{k-1}) - \Lambda\wt g_{k-1} \in \cA_\phg^{(k-1,*),(k,m)\cup(k+1,*)}(\wt M\setminus\wt K^\circ;S^2\wt T^*\wt M)
  \]
  for some $m\in\N_0$, and $\supp\Err_{k-1}\subset U$ where $U$ is the domain of influence of a compact subset of $\cU^\circ$. Then there exists $h\in\cA_\phg^{(k+1,*),(k,m)}(\wt M\setminus\wt K^\circ;S^2\wt T^*\wt M)$ with $\supp h\cap M_\circ\subset\upbeta_\circ^*U$ so that for $\wt g^k:=\wt g_{k-1}+h$, we have
  \[
    \Err^k := \Ric(\wt g^k) - \Lambda\wt g^k \in \cA_\phg^{(k-1,*),(k+1,*)}(\wt M\setminus\wt K^\circ;S^2\wt T^*\wt M).
  \]
\end{prop}

In other words, we can solve away the error term $\Err_{k-1}$ to leading order at $M_\circ$. A careful accounting of index sets, similarly to (but combinatorially more involved than) the proof of Proposition~\ref{PropFMc2}, allows one to specify the index sets; we leave this to the interested reader.

\begin{proof}[Proof of Proposition~\usref{PropFMc3Gen}]
  There exist $f_0,\ldots,f_m\in\cA_\phg^{(-1,*)}(M_\circ;\upbeta_\circ^*S^2 T^*M)$ so that the index set of
  \[
    \Err_{k-1} - \chi_\circ\eps^k\sum_{j=0}^m(\log\rho_\circ)^j f_j
  \]
  at $M_\circ$ is $(k+1,*)$. Writing $\rho_\circ=\frac{\eps}{\hat\rho}$, the second Bianchi identity for $\wt g_{k-1}$ implies
  \begin{align*}
    &0 = \delta_g\sfG_g\biggl(\sum_{j=0}^m (\log\eps-\log\hat\rho)^j f_j\biggr) = \sum_{q=0}^m (\log\eps)^q \delta_g\sfG_g f'_q, \\
    &\hspace{10em} f'_q := \sum_{j=0}^{m-q}\begin{pmatrix}j+q\\q\end{pmatrix}(-1)^j(\log\hat\rho)^j f_{j+q} \subset \cA_\phg^{(-1,*)}(M_\circ;\upbeta_\circ^*S^2 T^*M).
  \end{align*}
  Therefore $\delta_g\sfG_g f'_q=0$ for all $q=0,\ldots,m$. Theorem~\ref{ThmAc} produces
  \[
    h_q \in \cA_\phg^{(1,*)}(M_\circ;\upbeta_\circ^*S^2 T^*M),\qquad (D_g\Ric-\Lambda)h_q=f'_q,
  \]
  with $\supp h_q$ contained in the domain of influence of a compact subset of $\cU^\circ$. Therefore,
  \[
    (D_g\Ric-\Lambda)\biggl(\sum_{q=0}^m(\log\eps)^q h_q\biggr) = \sum_{q=0}^m(\log\eps)^q f'_q = \sum_{j=0}^m (\log\eps-\log\hat\rho)^j f_j.
  \]
  Thus, the conclusions of the Proposition hold for $h=\chi_\circ\eps^k\sum_{q=0}^m(\log\eps)^q h_q$.
\end{proof}

For $k=3$, this produces a $((3,*),(1,0)_+\cup(3,*))$-smooth total family $\wt g^3$ with
\begin{equation}
\label{EqFMc3}
\begin{split}
  \Err^3 := \Ric(\wt g^3)-\Lambda\wt g^3 \in \cA_\phg^{(2,*),(4,*)}(\wt M\setminus\wt K^\circ;S^2\wt T^*\wt M).
\end{split}
\end{equation}
Note that the correction $\wt g^3-\wt g_2$ produced by Proposition~\ref{PropFMc3Gen} has index set $(4,*)$ at $\hat M$, so in particular the description of the leading order correction to the Kerr metric at $\hat M$ in Proposition~\ref{PropFhM2} remains valid for $\wt g^3$ as well.

\subsection{Third correction at \texorpdfstring{$\hat M$}{the front face of the total gluing spacetime}}
\label{SsFhM3}

Let $m\in\N_0$ be such that the index set of $\Err^3$ in~\eqref{EqFMc3} at $\hat M$ is $(2,m)\cup(3,*)$. Write the leading order part of $\Err^3$ at $\hat M$ as
\[
  \sum_{j=0}^m \eps^2(\log\eps)^j \sfe^{-1}f_j,\qquad f_j\in\CI\bigl(I;\cA_\phg^{(2,*)}(\hat X_b;S^2\,\Ttsc^*_{\hat X_b}\hat M_b)\bigr).
\]
In order to solve this away, we need to make corrections to $\wt g^3$ by deformation tensors of translations by amounts $\cO(\eps^2(\log\eps)^m)$; since such corrections are thus still larger (by logarithmic factors) at $\hat M$ than the $\eps^2$ deviation of the total family $\wt g^3$ from $\hat g_b$ at $\hat M$, we need to solve away this leading order part using a combination of a correction and a pullback much as in~\S\S\ref{SsFhM1} and \ref{SsFhM2}. (Only in later steps of the construction can we omit the pullback part; see~\S\ref{SsFRest}.)

Since the arguments are very similar to those in~\S\ref{SsFhM2}, we shall be brief. We begin by solving away the term $\eps^2(\log\eps)^m\sfe^{-1}f_m$. Setting
\[
  \Phi_{3,\hat c(t)} \colon (\eps,\hat x) \mapsto \bigl(\eps,\hat x+\eps^2(\log\eps)^m\hat c(t)\bigr)
\]
where $\hat c\in\CI(I;\R^3)$ is to be determined, we consider for $\dot b\in\CI(I;\R\times\R^3)$ the tensor
\[
  \hat g_{3,\hat c,\dot b}=\hat g_{3,\hat c,\dot b}(\eps,t,\hat x;\dd\hat t,\dd\hat x)
\]
defined on a $t$-level set by $\Phi_{3,\hat c(t)}^*(\hat g_b+\eps^3(\log\eps)^m(\breve h_{b,2\hat c'(t)}+\hat g'_b(\dot b(t))))$; thus, in the coordinates~\eqref{EqFhM1hatCoord} near $t=t_0\in I$, we have
\begin{align*}
  \hat g_{3,\hat c,\dot b} &\equiv \hat g_b + \eps^2(\log\eps)^m \cL_{\hat c(t_0)\cdot\pa_{\hat x}}\hat g_b + \eps^3(\log\eps)^m\bigl(\cL_{\hat t\hat c'(t_0)\cdot\pa_{\hat x}+(\hat c'(t_0)\cdot\hat x)\pa_{\hat t}}\hat g_b + \hat g_b'(\dot b(t_0))\bigr) \\
    &\quad + \eps^4(\log\eps)^m\Bigl(\frac{\hat t^2}{2}h_{b,2\hat c''(t_0)}+\hat t\breve h_{b,2\hat c''(t_0)} + \hat t\hat g'_b(\dot b'(t_0))\Bigr) + \eps^4(\log\eps)^{2 m}\cL_{\hat c(t_0)\cdot\pa_{\hat x}}^2\hat g_b \\
    &\quad \bmod \cA_\phg^{(5,*)}\bigl([0,1)_\eps;\CI(\R^4_{\hat t,\hat x};S^2 T^*\R^4)\bigr).
\end{align*}
(The interaction of $\cL_{\hat c(t_0)\cdot\pa_{\hat x}}\hat g_b$ and $\cL_{\hat t\hat c(t_0)\cdot\pa_{\hat x}+(\hat c(t_0)\cdot\hat x)\pa_{\hat t}}\hat g_b+\hat g'_b(\dot b(t_0))$ is now of order $\eps^5(\log\eps)^{2 m}$ and thus negligible.) Therefore,
\[
  \Ric(\hat g_{3,\hat c,\dot b}) \equiv \eps^4(\log\eps)^m D_{\hat g_b}\Ric\Bigl(\frac{\hat t^2}{2}h_{b,2\hat c''(t_0)}+\hat t\breve h_{b,2\hat c''(t_0)} + \hat t\hat g'_b(\dot b'(t_0))\Bigr) \bmod \cA_\phg^{(5,*)}.
\]

For the tensor $\wt g_{3,\hat c,\dot b}:=\wt g^3+\hat\chi\sfe^{-1}(\hat g_{3,\hat c,\dot b}-\hat g_b)$, define the error term
\[
  \Err_{3,\hat c,\dot b} := \Ric(\wt g_{3,\hat c,\dot b})-\Lambda\wt g_{3,\hat c,\dot b} \in \cA_\phg^{(2,m)\cup(3,*),(4,*)}(\wt M\setminus\wt K^\circ;S^2\wt T^*\wt M);
\]
then the $\eps^2(\log\eps)^m$ coefficient of $\sfe\Err_{3,\hat c,\dot b}$ is at $\hat M_{c(t)}$ given by
\[
  f_{3,\hat c,\dot b,m}(t) = f_m(t) + D_{\hat g_b}\Ric\Bigl(\frac{\hat t^2}{2}h_{b,2\hat c''(t)}+\hat t\breve h_{b,2\hat c''(t)} + \hat t\hat g'_b(\dot b'(t)) - \cL_{\hat c(t)\cdot\pa_{\hat x}}h_{(2)}(t)\Bigr)
\]
analogously to~\eqref{EqFhM2Err2p}, where, as in Proposition~\ref{PropFhM2}, the subleading term $h_{(2)}=\eps^{-2}(\sfe\wt g^3-\hat g_b)|_{\hat M(t)}\in\cA_\phg^{(-2,0)_+}(\hat X_b;S^2\,\Ttsc^*_{\hat X_b}\hat M_b)$ satisfies~\eqref{EqFhM2DRicLh2} for all $\hat c$. Choosing $\hat c$ and then $\dot b$ by solving nondegenerate linear second, resp.\ first order ODEs by means of Theorem~\ref{ThmAhKCoker}, we can arrange for this coefficient to satisfy the assumptions of Theorem~\ref{ThmAhPhg} for all $t\in I$. Therefore, we obtain $h_m\in\cA_\phg^{(0,*)}(\hat M\setminus\wt K^\circ;S^2\wt T^*_{\hat M}\wt M)$ so that
\[
  \wh{D_{\hat g_b}\Ric}(0)(\sfe h_m(t)) = -f_{3,\hat c,\dot b,m}(t),\qquad t\in I.
\]
The update to $\wt g^3$ is then
\[
  \wt g_{3,m} := \Phi_{3,\hat c,\dot b}^*\bigl(\wt g_{3,\hat c,\dot b}+\hat\chi\eps^4(\log\eps)^m h_m\bigr),\qquad
  \Phi_{3,\hat c,\dot b}(\eps,t,\hat x) = \bigl(\eps,t,\hat x+\eps^2(\log\eps)^m\hat c(t)\bigr).
\]
The error term $\Err_{3,m}:=\Ric(\wt g_{3,m})-\Lambda\wt g_{3,m}$ has index set $(2,m-1)\cup(3,*)$ at $\hat M$ and $(4,*)$ at $M_\circ$; and $\sfe\wt g_{3,m}$ equals $\hat g_b$ along $\hat M$ modulo quadratic error terms.

One can continue in this fashion to eliminate the $\eps^2(\log\eps)^{m-1}$ term of $\Err_{3,m}$, and so on; after $m+1$ steps, one obtains:

\begin{prop}[Third correction at $\hat M$]
\label{PropFhM3}
  There exists a $((3,*),(1,0)_+\cup(3,*))$-smooth total family $\wt g_3$ which is equal to $\wt g_0$ at $\hat M$ and $M_\circ$ and equal to $g$ near $X\setminus\cU^\circ$, with
  \[
    \Err_3 := \Ric(\wt g_3) - \Lambda\wt g_3 \in \cA_\phg^{(3,*),(4,*)}(\wt M\setminus\wt K^\circ;S^2\wt T^*\wt M)
  \]
  with $\supp\Err_3\cap M_\circ\subset\upbeta_\circ^*U$ where $U$ is the domain of influence of a compact subset of $\cU^\circ$, and so that $\sfe\wt g_3$ is equal to $\hat g_b$ at $\hat M$ modulo quadratically vanishing errors in the sense of Proposition~\usref{PropFhM1}\eqref{ItFhM1hatMQuadr}.
\end{prop}

\subsection{Completion of the construction}
\label{SsFRest}

Proposition~\ref{PropFMc3Gen} produces a correction $h\in\cA_\phg^{(5,*),(4,*)}(\wt M\setminus\wt K^\circ;S^2\wt T^*\wt M)$ to $\wt g_3$, with support in the domain of influence of a compact subset of $\cU^\circ$, so that
\[
  \Err^4 := \Ric(\wt g^4)-\Lambda\wt g^4 \in \cA_\phg^{(3,*),(5,*)}(\wt M\setminus\wt K^\circ;S^2\wt T^*\wt M),\qquad \wt g^4:=\wt g_3+h.
\]
Note that $\wt g^4$ is a $((3,*),(1,0)_+\cup(3,*))$-smooth total family. Solving away $\Err^4$ to leading order at $\hat M$ now relies on the following result:

\begin{prop}[Correction at $\hat M$]
\label{PropFResthM}
  Let $k\geq 4$. Suppose $\wt g^k$ is a $((3,*),(1,0)_+\cup(3,*))$-smooth total family so that $\sfe\wt g^k$ is equal to $\hat g_b$ at $\hat M$ modulo quadratically vanishing errors in the sense of Proposition~\usref{PropFhM1}\eqref{ItFhM1hatMQuadr}. Suppose that
  \[
    \Err^k = \Ric(\wt g^k)-\Lambda\wt g^k \in \cA_\phg^{(k-1,m)\cup(k,*),(k+1,*)}(\wt M\setminus\wt K^\circ;S^2\wt T^*\wt M).
  \]
  Then there exists $h\in\cA_\phg^{(k-1,m),(k+1,*)}(\wt M\setminus\wt K^\circ;S^2\wt T^*\wt M)$ so that for $\wt g_k:=\wt g^k+h$, we have
  \[
    \Err_k := \Ric(\wt g_k)-\Lambda\wt g_k \in \cA_\phg^{(k,*),(k+1,*)}(\wt M\setminus\wt K^\circ;S^2\wt T^*\wt M).
  \]
  In particular, $\sfe\wt g_k$ is equal to $\hat g_b$ at $\hat M$ modulo quadratically vanishing errors; and we have $\supp\Err_k\cap M_\circ\subset\upbeta_\circ^*U$ where $U$ is the domain of influence of a compact subset of $\cU^\circ$.
\end{prop}
\begin{proof}
  We shall solve away the leading order term of $\Err^k$ at $\hat M$, which is
  \[
    \eps^{k-1}(\log\eps)^m\sfe^{-1}f,\qquad
    f\in\CI\bigl(I;\cA_\phg^{(2,*)}(\hat X_b;S^2\,\Ttsc^*_{\hat X_b}\hat M_b)\bigr).
  \]
  Let $\hat c\in\CI(I;\R^3)$ and $\dot b\in\CI(I;\R\times\R^3)$; we will determine these functions momentarily. Consider $\Phi_{\hat c(t)}\colon(\eps,\hat x)\mapsto(\eps,\hat x+\eps^{k-1}(\log\eps)^m\hat c(t))$ and the tensor
  \[
    g_{\hat c,\dot b} = \Bigl( \Phi_{\hat c(t)}^*\bigl( \hat g_b + \eps^k(\log\eps)^m\bigl(\breve h_{b,\hat c'(t)} + \hat g'_b(\dot b(t))\bigr)\bigr) \bigr)_{t\in I}.
  \]
  Taylor expanding around $t=t_0\in I$ and using the coordinates $\eps$, $\hat t=\frac{t-t_0}{\eps}$, $\hat x$, this is
  \begin{align*}
    \hat g_{\hat c,\dot b} &\equiv \hat g_b + \eps^{k-1}(\log\eps)^m h_{b,2\hat c(t_0)} + \eps^k(\log\eps)^m\bigl(\cL_{\hat t\hat c'(t_0)\cdot\pa_{\hat x}+(\hat c'(t_0)\cdot\hat x)\pa_{\hat t}}\hat g_b + \hat g'_b(\dot b(t_0))\bigr) \\
      &\quad\qquad + \eps^{k+1}(\log\eps)^m\Bigl(\frac{\hat t^2}{2}h_{b,2\hat c''(t_0)} + \hat t\breve h_{b,2\hat c''(t_0)} + \hat t\hat g'_b(\dot b'(t_0))\Bigr) \\
      &\quad\qquad \hspace{8em} \bmod \cA_\phg^{(k+2,*)}\bigl([0,1)_\eps;\CI(\R^4_{\hat t,\hat x};S^2 T^*\R^4)\bigr),
  \end{align*}
  and therefore we have
  \[
    \Ric(\hat g_{\hat c,\dot b}) \equiv \eps^{k+1}(\log\eps)^m\wh{D_{\hat g_b}\Ric}(0)\Bigl(\frac{\hat t^2}{2}h_{b,2\hat c''(t_0)} + \hat t\breve h_{b,2\hat c''(t_0)} + \hat t\hat g'_b(\dot b'(t_0))\Bigr) \bmod \cA_\phg^{(k+2,*)}.
  \]

  We then consider
  \[
    \wt g_{\hat c,\dot b} := \wt g^k + \hat\chi\sfe^{-1}(\hat g_{\hat c,\dot b}-\hat g_b).
  \]
  Since $\sfe\wt g^k\equiv\hat g_b+\eps^2 h_{(2)}\bmod\cA_\phg^{(3,*)}$ at $\hat M^\circ$, we compute, as in~\eqref{EqFhM2Err2p} and earlier in~\eqref{EqFhM2Errca2}--\eqref{EqFhM2Err2}, the leading order term of $\Ric(\wt g_{\hat c,\dot b})-\Lambda\wt g_{\hat c,\dot b}$ at $\hat M$ to be $\eps^{k-1}(\log\eps)^m$ times $\sfe^{-1}$ applied to
  \begin{align*}
    \Err_{\hat c,\dot b}(t) &:= f(t) + D_{\hat g_b}\Ric\Bigl(\frac{\hat t^2}{2}h_{b,2\hat c''(t)} + \hat t\breve h_{b,2\hat c''(t)} + \hat t\hat g'_b(\dot b'(t)) - \cL_{\hat c(t)\cdot\pa_{\hat x}}h_{(2)}(t)\Bigr), \\
    &\qquad \Err_{\hat c,\dot b}\in \CI\bigl(I;\cA_\phg^{(2,*)}(\hat X_b;S^2\,\Ttsc^*_{\hat X_b}\hat M_b)\bigr).
  \end{align*}
  We then solve $\la\Err_{\hat c,\dot b},-\ra_{L^2(\hat X_b)}=0\in\CI(I;(\cK_{b,\rm COM}^*)^*)$, which is a nondegenerate linear second order ODE for $\hat c$ by Theorem~\ref{ThmAhKCoker}; and then we take $\dot b$ to be a solution of the nondegenerate first order ODE $\la\Err_{\hat c,\dot b},-\ra_{L^2(\hat X_b)}=0\in\CI(I;(\cK_{b,\rm Kerr}^*)^*)$.

  For these choices of $\hat c,\dot b$, we can then apply Theorem~\ref{ThmAhPhg} to find $h_m\in\cA_\phg^{(0,*)}(M_\circ\setminus\wt K^\circ;S^2\wt T^*_{\hat M}\wt M)$ with $\wh{D_{\hat g_b}\Ric}(0)(\sfe h_m(t))=-\Err_{\hat c,\dot b}(t)$ for all $t\in I$. We then set
  \[
    \wt g_{k,m} := \wt g_{\hat c,\dot b} + \hat\chi\eps^{k+1}(\log\eps)^m h_m;
  \]
  this is a $((3,*),(1,0)_+\cup(3,*))$-smooth total family with
  \[
    \Ric(\wt g_{k,m})-\Lambda\wt g_{k,m} \in \cA_\phg^{(k-1,m-1)\cup(k,*),(k+1,*)}(\wt M\setminus\wt K^\circ;S^2\wt T^*\wt M).
  \]
  Thus, we have removed the $\cO(\eps^{k-1}(\log\eps)^m)$ leading order term of $\Err^k$ using a correction $\hat\chi\eps^{k+1}(\log\eps)^m h_m$ with index sets $(k-1,m)\cup(k+1,*)$ at $\hat M$ and $(k+1,*)$ at $M_\circ$. Continuing in this fashion, we can successively remove all $\cO(\eps^{k-1}(\log\eps)^j)$ terms for $j=k-1,\ldots,0$. This finishes the proof.
\end{proof}

We are now ready to conclude the first part of the construction:

\begin{thm}[Formal solution at $\eps=0$]
\label{ThmFh}
  In the notation of Theorem~\usref{ThmM}, there exists a $((3,*),(1,0)_+\cup(3,*))$-smooth total family $\wt g_\infty$ over $\wt M\setminus\wt K^\circ$ with respect to $g$ with the following properties:
  \begin{enumerate}
  \item in the frame $\dd\hat t,\dd\hat x$ related to the fixed choice Fermi normal coordinates,\footnote{See also point~\eqref{ItMFamily} in~\S\ref{SM}.} the $\hat M_p$-model of $\hat g$ is equal to the Kerr metric $\hat g_{\bhm,\bha}$ for all $p\in\cC$;
  \item the tensor $\sfe\wt g_\infty$ (defined near $\hat M$) is equal to the Kerr metric $\hat g_b$ at $\hat M$ modulo quadratically vanishing errors in the sense of Theorem~\usref{ThmM} (or as in Proposition~\usref{PropFhM1}\eqref{ItFhM1hatMQuadr});
  \item we have
    \begin{equation}
    \label{EqFh}
      \Err_\infty := \Ric(\wt g_\infty) - \Lambda\wt g_\infty \in \CIdot(\wt M\setminus\wt K^\circ;S^2\wt T^*\wt M);
    \end{equation}
  \item $\wt g_\infty=g$ outside the domain of influence $U$ of a compact subset of $\cU^\circ$ (and in particular $\supp\Err_\infty\cap M_\circ\subset U$).
  \end{enumerate}
\end{thm}
\begin{proof}
  Using Proposition~\ref{PropFResthM} with $k=4$, we correct the $((3,*),(1,0)_+\cup(3,*))$-smooth total family $\wt g^4$ using $h_4\in\cA_\phg^{(3,*),(5,*)}(\wt M\setminus\wt K^\circ;S^2\wt T^*\wt M)$ to $\wt g_4=\wt g^4+h_4$, which satisfies
  \[
    \Err_4 := \Ric(\wt g_4) - \Lambda\wt g_4 \in \cA_\phg^{(4,*),(5,*)}(\wt M\setminus\wt K^\circ;S^2\wt T^*\wt M).
  \]
  Proposition~\ref{PropFMc3Gen}, with $k=5$, then produces $h^5\in\cA_\phg^{(6,*),(5,*)}(\wt M\setminus\wt K^\circ;S^2\wt T^*\wt M)$ so that $\wt g^5:=\wt g_4+h^5$ satisfies
  \[
    \Err^5 := \Ric(\wt g^5) - \Lambda\wt g^5 \in \cA_\phg^{(4,*),(6,*)}(\wt M\setminus\wt K^\circ;S^2\wt T^*M).
  \]
  Continuing in this fashion, we obtain sequences $h_k\in\cA_\phg^{(k-1,*),(k+1,*)}(\wt M\setminus\wt K^\circ;S^2\wt T^*\wt M)$ and $h^{k+1}\in\cA_\phg^{(k+2,*),(k+1,*)}(\wt M\setminus\wt K^\circ;S^2\wt T^*\wt M)$ so that for $\wt g_k=\wt g^k+h_k$ and $\wt g^{k+1}=\wt g_k+h^{k+1}$ we have
  \begin{align*}
    \Err_k&:=\Ric(\wt g_k)-\Lambda\wt g_k\in\cA_\phg^{(k,*),(k+1,*)}, \\
    \Err^{k+1}&:=\Ric(\wt g^{k+1})-\Lambda\wt g^{k+1}\in\cA_\phg^{(k-1,*),(k+1,*)}.
  \end{align*}
  Moreover, the supports of the correction terms $h_k,h^{k+1}$, and thus also of the error terms $\Err_k$, $\Err^{k+1}$, are contained in the domain of influence of a compact subset of $\cU^\circ$. More precisely, in the $k$-th step of the construction we can ensure that $\supp h_k$ and $\supp h^{k+1}$ are contained in the domain of influence of a compact subset of $\cU^\circ$ whose distance (with respect to any fixed Riemannian metric on $X$) to $\pa\cU^\circ$ is at least $\delta+\delta 2^{-k}$ for some fixed small $\delta>0$; for $h_k$ this is clear since we may cut off $h_k$ to any neighborhood of $\hat M$ (which is, indeed, how $h_k$ is constructed), while for $h^{k+1}$ this follows from the proof of Theorem~\ref{ThmAc}, specifically Proposition~\ref{PropAcTrue} which one simply applies to a $\delta$-shrinking of $\cU^\circ$, for a fixed small $\delta>0$, throughout the entire construction in this section. Define then $\wt g_\infty$ to be an asymptotic sum
  \[
    \wt g_\infty \sim \wt g^4+\sum_{k=4}^\infty(h_k+h^{k+1}),
  \]
  in the sense that the difference of $\wt g_\infty$ and the truncation of the series at $k=N$ (which gives $\wt g^N$) is polyhomogeneous on $\wt M$ with index sets $(N,*)$ and $(N+2,*)$ at $\hat M$ and $M_\circ$, respectively. In view of the support properties of $\wt g^4,h_k,h^{k+1}$, we can arrange that $\wt g_\infty=\wt g^4=g$ outside the domain of influence of a compact subset of $\cU^\circ$, as desired. Since $\wt g_\infty$ differs from $\wt g_k$ by error terms which have increasing orders of vanishing at $\hat M$ and $M_\circ$ as $k\to\infty$, the membership~\eqref{EqFh} follows.
\end{proof}

\begin{rmk}[Construction on a slightly larger manifold]
\label{RmkMLarger}
  While in the notation of Definition~\ref{DefGKModel}, the above construction takes place outside of $[0,1)_\eps\times I_t\times\hat K_{\bhm,\bha}^\delta$ where $\delta=0$, any other choice of $\delta\in(-\sqrt{\bhm^2-\bha^2},\sqrt{\bhm^2-\bha^2})$ works just as well. Thus, if we write
  \[
    \wt K^\delta=\{(\eps,t,x)\colon|\eps x|\leq\bhm-\delta\},
  \]
  then Theorem~\ref{ThmFh} remains valid if one replaces $\wt K$ by $\wt K^\delta$. By taking $\delta>0$, we can thus construct a formal solution $\wt g_\infty$ on the larger manifold $\wt M\setminus(\wt K^\delta)^\circ$.
\end{rmk}

\section{Formal solution of the initial value problem}
\label{SIVP}

We continue using the notation from~\S\ref{SM}. In this section, the focus is on the hypersurface $X\subset M$, which intersects the curve $\cC\subset M$ orthogonally at the point $\fp$. We now complete the proof of Theorem~\ref{ThmM} by adding to the total family $\wt g_\infty$ from Theorem~\ref{ThmFh} a further correction, supported near the total gluing space $\wt X\subset\wt M$ (see point~\eqref{ItMDataTilde} in~\S\ref{SM}), which vanishes to infinite order at $\hat M\cup M_\circ$. For technical reasons, we start with $\wt g_\infty$ defined on a larger manifold $\wt M\setminus(\wt K^\delta)^\circ$, $\delta\in(0,\sqrt{\bhm^2-|\bha|^2})$, as in Remark~\ref{RmkMLarger}.\footnote{Following Notation~\ref{NotGLDef}, all tensors are defined only in some open neighborhood of $\hat M\cup M_\circ$. Recall that $\wt g_\infty$ is a Lorentzian signature section of $S^2\wt T^*\wt M$ in some such open neighborhood by Corollary~\ref{CorGLInitial}\eqref{ItGLInitialSpace}.}

\begin{thm}[Formal solution at $\wt X$]
\label{ThmIVP}
  Let $\hat\cE,\cE\subset\C\times\N_0$ be two nonlinearly closed index sets with $\Re\hat\cE,\Re\cE>0$. Suppose $\wt g_\infty$ is a $(\hat\cE,\cE)$-smooth total family with respect to $g$ over $\wt M\setminus(\wt K^\delta)^\circ$ whose $\hat M$-model is $\hat g$, i.e.\ the $\hat M_p$-model is a fixed subextremal Kerr metric $\hat g_{\bhm,\bha}$ for all $p\in\cC$ (as in point~\eqref{ItMFamily} in~\textnormal{\S}\usref{SM}). Suppose that $\wt g_\infty=g$ near $X\setminus\cU^\circ$, and $\Ric(\wt g_\infty)-\Lambda\wt g_\infty\in\CIdot(\wt M\setminus\wt K^\circ;S^2\wt T^*\wt M)$.\footnote{For $\hat\cE=(3,*)$ and $\cE=(1,0)_+\cup(3,*)$, Theorem~\ref{ThmFh} produces such a $\wt g_\infty$.} Then there exists $\wt h\in\CIdot(\wt M\setminus\wt K^\circ;S^2\wt T^*\wt M)$ with support contained in any fixed open neighborhood of $\hat M_\fp\cup\upbeta_\circ^*\cU^\circ$, so that for the $(\hat\cE,\cE)$-smooth total family
  \[
    \wt g := \wt g_\infty + \wt h
  \]
  over $\wt M\setminus\wt K^\circ$, the error $\Ric(\wt g)-\Lambda\wt g$ vanishes to infinite order at $\hat M\cup M_\circ\cup\wt X$.
\end{thm}

The infinite order vanishing of $\Ric(\wt g_\infty+\wt h)-\Lambda(\wt g_\infty+\wt h)$ at $\hat M\cup M_\circ$ is automatic for $\wt h\in\CIdot(\wt M\setminus\wt K^\circ;S^2\wt T^*\wt M)$. Furthermore, the desired conclusion depends only on $\wt h$ in a neighborhood of $\wt X\subset\wt M$, and indeed only on the jet of $\wt h$ at $\wt X$; therefore, the support property of $\wt h$ can be arranged via multiplication with a smooth cutoff which equals $1$ near $\wt X$. The task is thus to construct $\wt h$ in Taylor series at $\wt X$.

We shall work in a $(3+1)$-decomposition near $\wt X$ (see~\S\ref{SsIVPFol}), and construct $\wt g$ with the additional requirement (which amounts to fixing a gauge) that the lapse and shift of $\wt g$ be equal to those of $\wt g_\infty$. The first step of the construction is to arrange for the constraint equations to hold at $\wt X$ (Proposition~\ref{PropIVPConst}); then we solve for the full Taylor series of $\wt h$ at $\wt X$ order by order (Proposition~\ref{PropIVPTay}). This gives Theorem~\ref{ThmIVP}.

\subsection{Foliations by spacelike hypersurfaces}
\label{SsIVPFol}

To fix notation and conventions and to illustrate the Taylor series construction, we first consider in~\S\ref{SssIVPn1} $(n+1)$-decompositions on general spacetimes. In~\S\ref{SssIVPwt}, we set up the construction on $\wt X$ which is then carried out in~\S\ref{SsIVPPf}.

\subsubsection{\texorpdfstring{$(n+1)$-}{(n+1)-}decomposition on a general manifold}
\label{SssIVPn1}

We first recall the $(n+1)$-de\-com\-po\-si\-tion near a spacelike hypersurface $X$ inside a smooth $(n+1)$-dimensional Lorentzian manifold $(M,g)$, following \cite[\S\S{VI.2--VI.3}]{ChoquetBruhatGR} (which uses different sign conventions). Thus, we identify an open neighborhood of $X\subset M$ with a neighborhood
\[
  O \subset \R\times X
\]
of $\{0\}\times X$, where $X$ is identified with $\{0\}\times X$. This induces splittings $T_O M=\pi_1^*T\R\oplus\pi_2^* T X$ and $T_O^*M=\pi_1^*T^*\R\oplus\pi_2^*T^*X$ of the tangent and cotangent bundles, where $\pi_1\colon O\to\R$, $\pi_2\colon O\to X$ are the projection maps which we henceforth drop from the notation. We denote the coordinate in the first factor of $\R\times X$ by $t$, and assume that $\dd t$ is past timelike, so in particular all $t$-level sets $X_t$ are spacelike. The future unit normals to the $X_t$ give the vector field
\[
  \nu := -\frac{\dd t^\sharp}{(-g^{-1}(\dd t,\dd t))^{1/2}}.
\]
The \emph{lapse} $0<N\in\CI(O)$ and the \emph{shift} $\beta\in\CI(O;T X)$ are uniquely determined by
\begin{equation}
\label{EqIVPn1LapseShift}
  \nu = N^{-1}(\pa_t-\beta),\qquad
  N := \frac{1}{\dd t(\nu)} = \bigl(-g^{-1}(\dd t,\dd t)\bigr)^{-1/2},\quad
  \beta := N\nu-\pa_t.
\end{equation}
Defining the vector field
\begin{equation}
\label{EqIVPn1e0}
  e_0 := \pa_t - \beta\in\CI(O;T_O M),\qquad g(e_0,e_0)=-N^2,
\end{equation}
we then have an orthogonal splitting
\[
  T_{(t,x)} M=\R e_0\oplus T_x X,\qquad (t,x)\in O.
\]
which induces an orthogonal splitting
\[
  T_{(t,x)}^*M=\R\dd t\oplus B(T_x^*X),\qquad B\colon\xi\mapsto\xi+\xi(\beta)\,\dd t,
\]
and correspondingly
\begin{equation}
\label{EqIVPn1S2Split}
\begin{split}
  S^2 T_{(t,x)} M &= \R e_0^2 \oplus \bigl(2 e_0 \otimes_s T_x X\bigr) \oplus S^2 T_x X, \\
  S^2 T^*_{(t,x)} M &= \R\dd t^2 \oplus \bigl(2\dd t\otimes_s B(T^*_x X)\bigr) \oplus B(S^2 T^*_x X),
\end{split}
\end{equation}
where we write $B$ also for the induced map on $S^2 T_x^*X$. In view of $e_0\perp T_x X$ and $\dd t\perp B(T^*_x X)$, we can thus write the metric $g$ and dual metric $g^{-1}$ as\footnote{In local coordinates $x^1,\ldots,x^n$ on $X$, this means $g=-N^2\dd t^2+\bar g_{i j}(\dd x^i+\beta^i\,\dd t)(\dd x^j+\beta^j\,\dd t)$ and $g^{-1}=-N^{-2}e_0^2+\bar g^{i j}\pa_{x^i}\otimes_s\pa_{x^j}$.}
\[
  g = \bigl( -N^2, 0, B(\bar g) \bigr),\quad
  g^{-1} = \bigl( -N^{-2}, 0, \bar g^{-1} \bigr),
  \qquad \bar g\in\CI(O;S^2 T^*X),
\]
where $\bar g$ is positive definite, and $\bar g^{-1}\in\CI(O;S^2 T X)$ is its dual. The first and second fundamental form of the $t$-level sets inside $O$ are\footnote{The sign convention for $k$ here is consistent with the one used in~\S\ref{SsECE}, but it is the opposite of \cite{ChoquetBruhatGR}.}
\[
  \gamma=\bar g,\qquad
  k(V,W):=g(\nabla_V\nu,W)\quad (V,W\in T X).
\]
Let us write $\bar\nabla$ for the Levi-Civita connection of $\bar g$ on each level set of $t$ inside $O$, and $\ol{\Ric}$ and $\bar R$ for the Ricci and scalar curvature of $\bar g$ on the $t$-level sets, respectively.

\begin{lemma}[Levi-Civita connection]
\label{LemmaIVPn1Nabla}
  For $V,W\in\CI(O;T X)$, we have
  \begin{alignat*}{2}
    \nabla_{e_0}e_0 &= N^{-1}e_0(N)e_0 + N\bar\nabla N, &\qquad
    \nabla_{e_0}W &= N^{-1}W(N)e_0 + N k(W,\cdot) + [e_0,W], \\
    \nabla_V e_0 &= N^{-1}V(N)e_0 + N k(V,\cdot), &\qquad
    \nabla_V W &= N^{-1}k(V,W)e_0 + \bar\nabla_V W.
  \end{alignat*}
  Furthermore, $k=\frac12 N^{-1}\cL_{e_0}\bar g$.
\end{lemma}
\begin{proof}
  We have
  \begin{equation}
  \label{EqIVPn1NableHor}
    [e_0,V]=[\pa_t,V]-[\beta,V]\in\CI(O;T X),
  \end{equation}
  so $[e_0,V]\perp e_0$. Writing $\la\cdot,\cdot\ra=g(\cdot,\cdot)$ (which on vector fields tangent to $X$ is equal to $\bar g(\cdot,\cdot)$), we can then use~\eqref{EqIVPn1e0}, the orthogonality $e_0\perp V,W$, and the torsion-free property $\nabla_{e_0}V-\nabla_V e_0=[e_0,V]$ to get
  \begin{align*}
    \la\nabla_{e_0}e_0,e_0\ra &= \frac12 e_0\la e_0,e_0\ra = -N e_0(N), \\
    \la\nabla_{e_0}e_0,V\ra &= -\la e_0,\nabla_{e_0}V\ra = -\la e_0,\nabla_V e_0\ra = -\frac12 V\la e_0,e_0\ra = N V(N);
  \end{align*}
  since $V(N)=\la\bar\nabla N,V\ra$, this gives the stated expression for $\nabla_{e_0}e_0$. Similarly, we have $\la\nabla_{e_0}W,e_0\ra=\frac12 W\la e_0,e_0\ra=-N W(N)$ and
  \begin{align*}
    \la\nabla_{e_0}W,V\ra &= \la\nabla_W e_0,V\ra + \la[e_0,W],V\ra \\
      &= \la\nabla_W(N\nu),V\ra + \la[e_0,W],V\ra \\
      &= N k(W,V) + \la[e_0,W],V\ra,
  \end{align*}
  which gives the stated expression for $\nabla_{e_0}W$. The expression for $\nabla_V e_0$ is obtained from this using $\nabla_V e_0=\nabla_{e_0}V+[V,e_0]$. For $\nabla_V W$, we have $\nabla_V W=\bar\nabla_V W+k(V,W)\nu$.

  Finally, the stated expression for $k$ arises from~\eqref{EqIVPn1NableHor} and
  \begin{align*}
    2 N k(V,W) &= g(\nabla_V e_0,W) + g(V,\nabla_W e_0) \\
      &= g(\nabla_{e_0}V,W)+g(V,\nabla_{e_0}W) - g([e_0,V],W)-g(V,[e_0,W]) \\
      &= e_0(g(V,W)) - g([e_0,V],W)-g(V,[e_0,W]) \\
      &= e_0(\bar g(V,W)) - \bar g([e_0,V],W)-\bar g(V,[e_0,W]) \\
      &= (\cL_{e_0}\bar g)(V,W).\qedhere
  \end{align*}
\end{proof}

\begin{cor}[Curvature]
\label{CorIVPn1Curv}
  In local coordinates $x^1,\ldots,x^n$ on $X$, the Ricci tensor $\Ric=\Ric(g)$ of $g$ satisfies
  \begin{align*}
    \Ric_{0 0} &= -N\bar\Delta N - N\tr_{\bar g}(\cL_{e_0}k) + N^2|k|_{\bar g}^2, \\
    \Ric_{0 i} &= -N(\delta_{\bar g}k + \dd\tr_{\bar g}k)_i, \\
    \Ric_{m i} &= N^{-1}(\cL_{e_0}k)_{m i} + \ol{\Ric}_{m i} + (\tr_{\bar g}k)k_{m i} - 2 k_{i q}k_m{}^q - N^{-1}(\bar\nabla^2 N)_{m i}.
  \end{align*}
  Here $\bar\Delta N=-\tr_{\bar g}(\bar\nabla^2 N)$, and the index {\rm`$0$'} stands for $e_0$. The scalar curvature is
  \[
    R_g = \bar R_g + 2 N^{-1}\tr_{\bar g}(\cL_{e_0}k) + (\tr_{\bar g}k)^2 - 3|k|_{\bar g}^2 + 2 N^{-1}\bar\Delta N.
  \]
\end{cor}
\begin{proof}
  Using abstract indices $i,j,l,m=1,\ldots,n$ for local coordinates on $X$, we have
  \begin{align*}
    \nabla_{\pa_i}\nabla_{\pa_j}\pa_l&=\nabla_{\pa_i}\bigl(\bar\nabla_{\pa_j}\pa_l + N^{-1}k_{j l} e_0\bigr) \\
      &= \bar\nabla_{\pa_i}\bar\nabla_{\pa_j}\pa_l + \bigl(N^{-1}k(\pa_i,\bar\nabla_{\pa_j}\pa_l) + \pa_i(N^{-1}k_{j l}) \bigr)e_0 \\
      &\quad\qquad + N^{-1}k_{j l}\bigl(N^{-1}(\pa_i N)e_0+N k(\pa_i,\cdot)\bigr) \\
      &= \bar\nabla_{\pa_i}\bar\nabla_{\pa_j}\pa_l + k_{j l}k(\pa_i,\cdot) \\
      &\quad\qquad + N^{-1}\bigl( (\bar\nabla_i k)(\pa_j,\pa_l) + k(\pa_i,\bar\nabla_{\pa_j}\pa_l)+k(\pa_j,\bar\nabla_{\pa_i}\pa_l) + k(\bar\nabla_{\pa_i}\pa_j,\pa_\ell)\bigr)e_0,
  \end{align*}
  and therefore
  \[
     [\nabla_{\pa_i},\nabla_{\pa_j}]\pa_l = [\bar\nabla_{\pa_i},\bar\nabla_{\pa_j}]\pa_l + N^{-1}\bigl((\bar\nabla_i k)_{j l}-(\bar\nabla_j k)_{i l}\bigr)e_0 + k_{j l}k(\pa_i,\cdot) - k_{i l}k(\pa_j,\cdot).
  \]
  This gives
  \begin{align*}
    R_{m l i j} &= \la[\nabla_{\pa_i},\nabla_{\pa_j}]\pa_l,\pa_m\ra = \bar R_{m l i j} + k_{i m}k_{j l} - k_{i l}k_{j m}, \\
    R_{0 l i j} &= \la[\nabla_{\pa_i},\nabla_{\pa_j}]\pa_l,e_0\ra = N\bigl((\bar\nabla_j k)_{i l}-(\bar\nabla_i k)_{j l}\bigr).
  \end{align*}
  Lastly,
  \begin{align*}
    R_{0 l 0 j} &= \big\la [\nabla_{e_0},\nabla_{\pa_j}]\pa_l - \nabla_{[e_0,\pa_j]}\pa_l,e_0\big\ra \\
      &=\big\la \nabla_{e_0}\bigl(\bar\nabla_{\pa_j}\pa_l + N^{-1}k_{j l}e_0\bigr) - \nabla_{\pa_j}\bigl(N^{-1}(\pa_l N)e_0+N k(\pa_l,\cdot)+[e_0,\pa_l]\bigr) - \nabla_{[e_0,\pa_j]}\pa_l, e_0 \big\ra \\
      &= -N(\bar\nabla_{\pa_j}\pa_l)(N)-N^2 e_0(N^{-1}k_{j l}) - k_{j l}e_0(N) + N^2\pa_j\bigl(N^{-1}(\pa_l N)\bigr) + (\pa_l N)(\pa_j N) \\
        &\quad\qquad + N^2 k(\pa_j,k(\pa_l,\cdot)) + N k(\pa_j,[e_0,\pa_l]) + N k([e_0,\pa_j],\pa_l) \\
      &= N(\bar\nabla^2 N)_{j l} - N(\cL_{e_0}k)_{j l} + N^2 k_{j q}k_l{}^q.
  \end{align*}
  This gives the stated expressions for $\Ric_{0 0}=\bar g^{l j}R_{0 l 0 j}$, $\Ric_{0 i}=\bar g^{l j}R_{0 l i j}$, and
  \[
    \Ric_{m i} = \bar g^{l j}R_{m l i j} - N^{-2} R_{m 0 i 0} = \bar g^{l j}R_{m l i j} - N^{-2} R_{0 m 0 i}.
  \]

  The scalar curvature can now be computed from $R_g=\bar g^{m i}\Ric_{m i}-N^{-2}\Ric_{0 0}$.
\end{proof}

For the Einstein tensor $\Ein(g)=\Ric(g)-\frac12 R_g g$ of $g$, we recover the constraints
\begin{alignat*}{2}
  (\Ein(g)+\Lambda g)(\nu,\nu) &= N^{-2}\Bigl(\Ric_{0 0} + \frac12 R_g N^2\Bigr) - \Lambda &&= \frac12\bigl(\bar R - |k|_{\bar g}^2 + (\tr_{\bar g}k)^2 - 2\Lambda \bigr), \\
  (\Ein(g)+\Lambda g)(\nu,\pa_j) &= N^{-1}\Ric_{0 j} &&= -(\delta_{\bar g}k + \dd\tr_{\bar g}k)_j,
\end{alignat*}
as in~\S\ref{SsECE}.

By Lemma~\ref{LemmaIVPn1Nabla}, $k$ is essentially the first $e_0$-derivative of $\bar g$, and therefore Corollary~\ref{CorIVPn1Curv} shows that $\Ric_{m i}-\Lambda\bar g_{m i}=0$ determines the second $e_0$-derivative of $\bar g$. We shall thus use the spatial part $(\Ric(g)-\Lambda g)_{m i}=0$ to construct $\bar g$ in Taylor series at $t=0$. This is sufficient to solve the full Einstein vacuum equations, provided the constraints are satisfied at $t=0$.

\begin{lemma}[Solution in Taylor series]
\label{LemmaIVPn1Taylor}
  Let $n\geq 3$. Suppose the constraint equations $(\Ein(g)+\Lambda g)_{0\mu}=0$, $\mu=0,\ldots,n$, hold for $g$ at $t=0$. (The index $0$ refers to $e_0$, and indices between $1$ and $n$ refer to local coordinates on $X$.) Suppose moreover that $(\Ric(g)-\frac{2\Lambda}{n-1}g)_{i j}$ vanishes to infinite order at $t=0$ for all spatial indices $1\leq i,j\leq n$. Then $\Ein(g)+\Lambda g$ (and thus also $\Ric(g)-\frac{2\Lambda}{n-1}g$) vanishes to infinite order at $t=0$.
\end{lemma}
\begin{proof}
  This is a Taylor series version of \cite[Chapter VI, Theorem~4.1]{ChoquetBruhatGR}. We consider the tensor $E:=\Ein(g)+\Lambda g$; so $E_{0\mu}=0$ at $t=0$ by assumption. Let us write `$\equiv$' for equality in Taylor series at $t=0$. Then
  \[
    R_g \equiv \bar g^{m i}\cdot\frac{2\Lambda}{n-1}\bar g_{m i} - N^{-2}\Ric_{0 0} = \frac{2\Lambda n}{n-1} - N^{-2}\Ric_{0 0}
  \]
  and
  \begin{align*}
    E_{0 0} &= \Ric_{0 0} + \Lambda g_{0 0} - \frac12 R_g g_{0 0} \\
      &\equiv \Ric_{0 0} - N^2\Lambda + \frac{N^2}{2}\Bigl(\frac{2\Lambda n}{n-1}-N^{-2}\Ric_{0 0}\Bigr) \\
      &= \frac12\Ric_{0 0} + \frac{N^2\Lambda}{n-1}
  \end{align*}
  imply $N^{-2}\Ric_{0 0}=2 N^{-2}E_{0 0}-\frac{2\Lambda}{n-1}$, so $R_g\equiv \frac{2(n+1)}{n-1}\Lambda-2 N^{-2} E_{0 0}$, and therefore
  \begin{equation}
  \label{EqIVP13Taylor}
    E_{i j} = (\Ric(g)+\Lambda g)_{i j} - \frac12 R_g\bar g_{i j} \equiv \Bigl(\frac{2\Lambda}{n-1}+\Lambda-\frac{\Lambda(n+1)}{n-1}\Bigr)\bar g_{i j} + N^{-2}E_{0 0}\bar g_{i j} = N^{-2}E_{0 0}\bar g_{i j}.
  \end{equation}
  Therefore, we also have $E_{i j}=0$ at $t=0$.

  Next, the second Bianchi identity $\delta_g E=0$ gives
  \[
    0 = (\delta_g E)(e_0) = N^{-2}(\nabla_{e_0}E)(e_0,e_0) - \bar g^{j k}(\nabla_{\pa_j}E)(e_0,\pa_k),
  \]
  so $e_0(E_{0 0})$ can be written in terms of $E_{\mu\nu}$ and their spatial derivatives at $t=0$, and therefore $e_0(E_{0 0})=0$ at $t=0$. Similarly,
  \[
    0 = (\delta_g E)(\pa_i) = N^{-2}(\nabla_{e_0}E)(\pa_i,e_0) - \bar g^{j k}(\nabla_{\pa_j}E)(\pa_i,\pa_k)
  \]
  allows one to express $e_0(E_{0 i})$ in terms of $E_{\mu\nu}$ and their spatial derivatives at $t=0$, so $e_0(E_{0 i})=0$ at $t=0$. Differentiating~\eqref{EqIVP13Taylor} along $e_0$, we thus obtain $e_0(E_{i j})=0$ at $t=0$.

  Differentiating the second Bianchi identity along $e_0$ at $t=0$ then implies $e_0^2(E_{0\mu})=0$ at $t=0$, and then the second derivative of~\eqref{EqIVP13Taylor} along $e_0$ gives $e_0^2(E_{i j})=0$. Continuing in this fashion, we deduce that all derivatives of $E_{\mu\nu}$ at $t=0$ vanish, as claimed.
\end{proof}

We end this section by explaining the procedure for constructing a formal solution of the Einstein vacuum equations at $\Sigma=t^{-1}(0)\cong X$. Suppose that $g$ satisfies the constraint equations at $\Sigma$. We shall leave the initial data $\gamma=\bar g$ and $k=\frac12 N^{-1}\cL_{e_0}\bar g$ unchanged at $\Sigma$. Let $\bar h_2\in\CI(O;S^2 T^*X)$ and consider
\[
  g_2 := g + t^2 B\bar h_2;
\]
then in the splitting~\eqref{EqIVPn1S2Split} we have $g_2=(-N^2,0,B(\bar g+t^2\bar h_2))$, and thus the lapse and shift of $g_2$ are equal to those of $g$. By Lemma~\ref{LemmaIVPn1Nabla}, the second fundamental form $k_2$ of $X_t$ with respect to $g_2$ is
\[
  k_2=\frac12 N^{-1}\cL_{e_0}\bigl(\bar g+t^2\bar h_2\bigr)=k + t N^{-1}\bar h_2 + \frac{t^2}{2}N^{-1}\cL_{e_0}\bar h_2,
\]
and therefore $\cL_{e_0}k_2=\cL_{e_0}k+N^{-1}\bar h_2$ at $\Sigma$. By Corollary~\ref{CorIVPn1Curv}, the equation
\[
  \Bigl(\Ric(g_2)_{m i}-\frac{2\Lambda}{n-1}(g_2)_{m i}\Bigr)\Big|_\Sigma=0
\]
gives an algebraic expression for $\bar h_2|_\Sigma$ in terms of $k|_\Sigma$ and $\bar g|_\Sigma$. With $\bar h_2|_\Sigma$ thus fixed and $\bar h_2$ being any smooth extension, one can then determine $\bar h_3|_\Sigma$ so that for $g_3=g_2+t^3 B\bar h_3$, the spatial coefficients $\Ric(g_3)_{m i}-\frac{2\Lambda}{n-1}(g_3)_{m i}$ vanish not just to first, but to second order at $t=0$: this uses the fact, again from Corollary~\ref{CorIVPn1Curv}, that the equation $(e_0(\Ric(g_3)_{m i}-\frac{2\Lambda}{n-1}(g_3)_{m i}))|_\Sigma=0$ produces an algebraic expression for $\bar h_3$. Proceeding in this manner, one can construct a full Taylor expansion $g'\sim g+\sum_{j\geq 2}t^j B\bar h_j$ which satisfies $\Ric(g')_{m i}-\frac{2\Lambda}{n-1}g'_{m i}\equiv 0$ (equality in Taylor series at $t=0$), and thus $\Ein(g')+\Lambda g'\equiv 0$ at $\Sigma$ by Lemma~\ref{LemmaIVPn1Taylor}.

\subsubsection{\texorpdfstring{$(3+1)$-}{(3+1)-}decomposition near \texorpdfstring{$\wt X$}{the total gluing space for initial data}}
\label{SssIVPwt}

Returning to the setup of Theorem~\ref{ThmIVP}, fix a smooth function $t\in\CI(M)$ so that $\dd t$ is past timelike on $X$, and fix a smooth vector field $V\in\CI(M;T M)$ near $X$ with $\dd t(V)=1$ which at $\cC\subset M$ is tangent to $\cC$. The flow of $V$ defines a diffeomorphism $\Psi$ from an open neighborhood $O\subset\R\times X$ of $\{0\}\times X$ to an open neighborhood $\Psi(O)\subset M$ of $X\subset M$; it has the property that $\Psi(0,x)=x$ for $x\in X$, and $\Psi(t,\fp)\in\cC$ for all $t$ with $(t,\fp)\in O$ (where $\{\fp\}=X\cap\cC$). Write
\[
  X_t := \Psi\bigl(O\cap(\{t\}\times X)\bigr)
\]
for the (images of the) level sets of $t$; upon shrinking $O$ if necessary, we may assume that they are spacelike hypersurfaces in $(M,g)$. Via the identification $O\cong\Psi(O)$ (which identifies $X_t$ with open submanifolds of $X$ containing $\fp$), the map $\Psi$ induces an embedding $T_{X_t}X=T X_t\hra T_{X_t}M$.

Consider now the map $\wt\Psi'\colon[0,1)\times O\to\wt M'$, $(\eps,t,x)\mapsto(\eps,\Psi(t,x))$, which is a diffeomorphism onto $[0,1)\times\Psi(O)$. Since $\wt\Psi'(\{0\}\times(O\cap(\R\times\{\fp\})))\subset\cC$, the map $\wt\Psi'$ lifts to a diffeomorphism
\[
  \wt\Psi \colon \wt O := \bigl[ [0,1)\times O; \{0\}\times(O\cap(\R\times\{\fp\})) \bigr] \to \wt\beta^*\bigl([0,1)\times\Psi(O)\bigr) \subset \wt M = [\wt M';\cC].
\]
This maps (a neighborhood of) the lift of $[0,1)\times(O\cap(\{t\}\times X))$ to (a neighborhood of) the lift $\wt X_t$ of $[0,1)\times X_t$ to $\wt M$. Moreover, we have $\wt X_0=\wt X$. We henceforth identify $\wt O\cong\wt\Psi(\wt O)$; thus, we write points in the open neighborhood $\wt O\subset\wt M$ of $\wt X$ as $(t,\wt x)\in\R\times\wt X$. We have subbundles $\wt T_{\wt X_t}\wt X\cong\wt T\wt X_t\hra\wt T_{\wt X_t}\wt M$ of corank $1$, and we write
\[
  \wt T_{\wt O}\wt X:=\bigsqcup_t\,\{t\}\times\wt T\wt X_t\subset\wt T_{\wt O}\wt M.
\]

We now define lapse $\wt N$ and shift $\wt\beta$ for the section $\wt g_\infty$ of $S^2\wt T^*\wt M$ over $\wt M\setminus(\wt K^\delta)^\circ$, which has Lorentzian signature in a neighborhood of $\hat M\cup\wt M$ over which we exclusively work (even though we do not make this explicit in the notation), relative to the foliation of $\wt O$ by $\{t\}\times\wt X_t$.\footnote{That is, the restrictions of $\wt N$ and $\wt\beta$ to an $\eps$-level set $\wt M_\eps$, $\eps>0$, of $\wt M$ are the lapse and shift of $\wt g_\infty|_{\wt M_\eps}$.} Note that (upon shrinking $O$ and thus $\wt O$ further, if necessary) $\dd t$ is past timelike by Corollary~\ref{CorGLInitial} (using Lemma~\ref{LemmaGKCoords} and the spacelike nature of $\dd\hat t$ for the Kerr metrics at $\hat M$). Thus, $-\wt g_\infty^{-1}(\dd t,\dd t)$ has a strictly positive lower bound near $(\hat X\setminus(\wt K^\delta)^\circ)\cup X_\circ$. The regularity of $\wt g_\infty$ in Theorem~\ref{ThmIVP} thus implies, via~\eqref{EqIVPn1LapseShift},\footnote{We have the more precise memberships in $\wt\upbeta^*\CI(O)+\cA_\phg^{\N_0\cup\hat\cE,\cE}(\wt O\setminus(\wt K^\delta)^\circ)$ here and below.}
\[
  \wt N,\ \wt N^{-1} \in \cA_\phg^{\hat\cG,\cG}(\wt O\setminus(\wt K^\delta)^\circ),\qquad
  \wt\beta \in \cA_\phg^{\hat\cG,\cG}(\wt O\setminus(\wt K^\delta)^\circ;\wt T_{\wt O}\wt X),\qquad
  \hat\cG:=\N_0\cup\hat\cE,\quad \cG:=\N_0\cup\cE.
\]
We moreover define the vector field
\[
  \wt e_0:=\pa_t-\wt\beta \in \cA_\phg^{\hat\cG,\cG}(\wt O\setminus(\wt K^\delta)^\circ;\wt T_{\wt O}\wt M),
\]
with $\wt\nu:=N^{-1}\wt e_0$ being the future pointing unit normal to all $\wt X_t$. The shift $\wt\beta$ gives rise to a bundle map $\wt B\in\cA_\phg^{\hat\cG,\cG}(\wt O\setminus(\wt K^\delta)^\circ;\Hom(\wt T^*_{\wt O}\wt X,\wt T^*_{\wt O}\wt M))$,
\begin{equation}
\label{EqIVPwtB}
  \wt B \colon \wt T^*_{\wt x}\wt X_t \to \wt T^*_{(t,\wt x)}\wt M,\qquad \xi \mapsto \xi + \wt\beta(\xi)\,\dd t,
\end{equation}
which induces a map on the symmetric second tensor power which we also denote $\wt B$. We can then write
\[
  \wt g_\infty = -N^2\,\dd t^2 + \wt B\bigl(\wt{\bar g}_\infty\bigr),\qquad \wt{\bar g}_\infty \in \cA_\phg^{\hat\cG,\cG}(\wt O\setminus(\wt K^\delta)^\circ;S^2\wt T^*\wt X).
\]
Since $\wt e_0$, as a vector field on $\wt O\subset\wt M$, is a vertical (i.e.\ tangent to $\eps$-level sets) vector field of class
\begin{equation}
\label{EqIVPwte0}
  \wt e_0 \in \cA_\phg^{\hat\cG-1,\cG}\Vb(\wt O\setminus(\wt K^\delta)^\circ),
\end{equation}
the second fundamental forms $\wt k_\infty=\frac12\wt N^{-1}\cL_{\wt e_0}\wt{\bar g}_\infty$ (see Lemma~\ref{LemmaIVPn1Nabla}) of the leaves $\wt X_t$ satisfy
\[
  \wt k_\infty \in \cA_\phg^{\hat\cG-1,\cG}(\wt O\setminus(\wt K^\delta)^\circ;S^2\wt T^*_{\wt O}\wt X).
\]

The initial data of $\wt X=\wt X_0$ are
\[
  \wt\gamma_{\infty,0} = \wt{\bar g}_\infty|_{\wt X}\in\cA_\phg^{\hat\cG,\cG}(\wt X\setminus(\wt K^\delta)^\circ;S^2\wt T^*\wt X),\qquad
  \wt k_{\infty,0} = \wt k_\infty|_{\wt X}\in\cA_\phg^{\hat\cG-1,\cG}(\wt X\setminus(\wt K^\delta)^\circ;S^2\wt T^*\wt X).
\]
Matching Corollary~\ref{CorGLInitial}, they satisfy $(\wt\gamma_{\infty,0},\wt k_{\infty,0})|_{X_\circ}=\upbeta_\circ^*(\gamma,k)$ where $(\gamma,k)$ are the initial data of $X$ inside $(M,g)$, and
\[
  \sfs\bigl(\wt\gamma_{\infty,0},\eps\wt k_{\infty,0}\bigr)\big|_{\hat X}=(\hat\gamma,\hat k)
\]
are the initial data of $\hat X_b$ inside the Kerr spacetime $(\hat M_b,\hat g_b)$; they are in fact a $(\hat\cE,\cE)$-smooth total family in the terminology of \cite[Definition~4.18]{HintzGlueID}. Furthermore, since $\Ein(\wt g_\infty)+\Lambda\wt g_\infty\in\CIdot(\wt M\setminus(\wt K^\delta)^\circ;S^2\wt T^*\wt M)$, the constraint equations are satisfied to infinite order at $\eps=0$ as well, in particular at $\wt X$; that is,
\begin{subequations}
\begin{equation}
\label{EqIVPwtConstraints}
  P(\wt\gamma_{\infty,0},\wt k_{\infty,0};\Lambda) \in \CIdot(\wt X\setminus(\wt K^\delta)^\circ;\ul\R\oplus\wt T^*\wt X)
\end{equation}
in the notation of~\eqref{EqMConstraints}. Since $\wt g_\infty=g$ near $X\setminus\cU^\circ$, the intersection of the support of~\eqref{EqIVPwtConstraints} with $\upbeta_\circ^*X$ is in fact contained in $\upbeta_\circ^*\cU^\circ$ and in particular compact; thus, there exists $\eps_0>0$ so that
\begin{equation}
\label{EqIVPwtSupp}
  \{\eps\leq\eps_0\} \cap \supp P(\wt\gamma_{\infty,0},\wt k_{\infty,0};\Lambda) \Subset \wt\upbeta^*([0,\eps_0]\times\cU^\circ).
\end{equation}
\end{subequations}

\subsection{Proof of Theorem~\ref{ThmIVP}}
\label{SsIVPPf}

We continue using the notation of~\S\ref{SssIVPwt}. We begin the proof of Theorem~\ref{ThmIVP} by removing the error term~\eqref{EqIVPwtConstraints}.

\begin{prop}[Solving the constraint equations]
\label{PropIVPConst}
  There exists a metric perturbation $\wt h_1\in\CIdot(\wt M\setminus\wt K^\circ;S^2\wt T^*\wt M)$ so that the initial data of
  \[
    \wt g_1 := \wt g_\infty + \wt h_1
  \]
  satisfy the constraint equations at $(\wt X\setminus\wt K^\circ)\cap\{\eps<\eps_0\}$ for some small $\eps_0>0$, and so that lapse and shift of $\wt g_1$ are equal to those of $\wt g_\infty$ in the $(3+1)$-decomposition fixed in~\textnormal{\S}\usref{SssIVPwt}.
\end{prop}
\begin{proof}
  With $\wt\gamma_{\infty,0},\wt k_{\infty,0}$ as in~\eqref{EqIVPwtConstraints}--\eqref{EqIVPwtSupp}, the key step is to find corrections
  \begin{equation}
  \label{EqIVPConsthq}
    \wt h,\wt q\in\CIdot(\wt X\setminus\wt K^\circ;S^2\wt T^*\wt X)
  \end{equation}
  with support compactly contained in $\wt\upbeta^*([0,\eps_0]\times\cU^\circ)$ so that
  \begin{equation}
  \label{EqIVPConst}
    P(\wt\gamma_{\infty,0}+\wt h,\wt k_{\infty,0}+\wt q;\Lambda) = 0
  \end{equation}
  on $\wt X\cap\{\eps\leq\eps_0\}$ for some small $\eps_0>0$. This is \emph{almost} the content of \cite[Proposition~5.6]{HintzGlueID}; the difference to the reference is that in the present paper the initial data $\sfs(\wt\gamma_{\infty,0},\eps\wt k_{\infty,0})$ near $\hat X^\circ$ are not equal to exact Kerr initial data in any set $\{(\eps,x)\colon|\eps x|<\hat R_0\}$, $\hat R_0>\bhm-\delta$, due to the global (in the fibers of $\hat M$) nature of our construction of $\wt g_\infty$. (Cf.\ \cite[Theorem~5.2, item (2)]{HintzGlueID}.)

  We can fix this in the following ad hoc manner, using that the initial data are equal to Kerr data $(\hat\gamma,\hat k)$ to leading order at $\hat X$. Fix a cutoff function $\chi\in\CI(\wt X)$ so that $\chi=1$ for $\hat r\leq\bhm-\frac23\delta$ and $\chi=0$ for $\hat r\geq\bhm-\frac13\delta$. Then
  \[
    (\wt\gamma'_{\infty,0},\wt k'_{\infty,0}) := (\wt\gamma_{\infty,0},\wt k_{\infty,0}) + \chi\Bigl(\bigl(\sfs^{-1}\hat\gamma,\eps^{-1}\sfs^{-1}\hat k\bigr) - (\wt\gamma_{\infty,0},\wt k_{\infty,0})\Bigr).
  \]
  Then $P(\wt\gamma'_{\infty,0},\wt k'_{\infty,0};\Lambda)\in\cA_\phg^{\hat\cF-2,\emptyset}(\wt X\setminus(\wt K^\delta)^\circ;\ul\R\oplus\wt T^*\wt X)$ where $\Re\hat\cF>0$ since the $\hat X$-model of $(\wt\gamma'_{\infty,0},\wt k'_{\infty,0})$ (i.e.\ the Kerr initial data) satisfies the constraint equations; and since in fact the constraint equations are violated, modulo errors vanishing to infinite order at $\eps=0$, only in the transition region $\supp\dd\chi$, we conclude that if $\chi'\in\CIc(\wt X)$ with $\supp\chi'\Subset\hat r^{-1}((\bhm-\delta,\bhm))$ equals $1$ near $\supp\dd\chi$, then we can write
  \[
    P(\wt\gamma'_{\infty,0},\wt k'_{\infty,0};\Lambda) = \chi'\Err_0 + \Err_1
  \]
  where $\Err_1\in\CIdot$ is supported in $\hat r>\bhm-\delta$, and $\Err_0\in\cA_\phg^{\hat\cF-2,\emptyset}$. Let $K'\Subset\hat r^{-1}((\bhm-\delta,\bhm))\subset\wt X$ be a smoothly bounded domain with $\supp\chi'\subset(K')^\circ$. Using arguments as in the proof of \cite[Theorem~6.2]{HintzGlueID}, we can then correct $(\wt\gamma'_{\infty,0},\wt k'_{\infty,0})$ in generalized Taylor series at $\hat X^\circ$ by tensors $(\wt h',\wt q')$ with support contained in $K'$ with vanishing restriction to $\hat X$, at the expense of admitting a failure of the constraints which lies in some fixed finite-dimensional space; that is,
  \[
    P(\wt\gamma'_{\infty,0}+\wt h',\wt k'_{\infty,0}+\wt q';\Lambda) - E_1(\wt c') = \Err'_1,
  \]
  where $\Err'_1\in\CIdot$ has support in $\hat r>\bhm-\delta$, while $\wt c'=\wt c'(\eps)$ with $\wt c'(0)=0$ is polyhomogeneous and $E_1$ is a linear map from $\C^N$ into $\CIc((K')^\circ;\ul\R\oplus\wt T^*\wt X)$ whose range spans the cokernel of the linearization of the constraints map around $(\hat\gamma,\hat k)$ on 00-Sobolev spaces on $(K')^\circ$ with exponentially decaying weights at $\pa K'$.

  One can then correct $(\wt\gamma'_{\infty,0}+\wt h',\wt k'_{\infty,0}+\wt q')$ further by tensors $(\wt h'',\wt q'')$ to a true solution of the constraints except for the presence of $E_1(\wt c'+\wt c'')$ on the right hand side (instead of $0$) by following the arguments of \cite[Theorem~6.2]{HintzGlueID}; here $\wt h'',\wt q''$ vanish identically near $\hat r=\bhm-\delta$, and $\wt h'',\wt q'',\wt c''$ vanish to infinite order at $\eps=0$. The restrictions
  \[
    (\wt\gamma_1,\wt k_1) := (\wt\gamma'_{\infty,0}+\wt h'+\wt h'',\wt k'_{\infty,0}+\wt q'+\wt q'')|_{\wt X\setminus\wt K^\circ}
  \]
  to the smaller domain $\wt X\setminus\wt K^\circ\subset\wt X\setminus(\wt K^\delta)^\circ$ (which is disjoint from $\supp E_1(\wt c'+\wt c'')$) thus satisfies the constraint equations for all small $\eps$, and by construction differ from $(\wt\gamma_{\infty,0},\wt k_{\infty,0})|_{\wt X\setminus\wt K^\circ}$ by tensors $\wt h,\wt q$ of class~\eqref{EqIVPConsthq}.

  Finally then, we shall take $\wt h_1=\wt B(\wt{\bar h}_1)$ (using the second symmetric tensor power of the map~\eqref{EqIVPwtB}) for a suitable choice of $\wt{\bar h}_1\in\CIdot(\cO\setminus\wt K^\circ;S^2\wt T^*\wt X)$. Writing $\wt{\bar h}_1(t)=\wt{\bar h}_1(0)+t\wt{\bar h}{}_1'(0)+\cO(t^2)$, we only need to specify $\wt{\bar h}_1(0)$ and $\wt{\bar h}{}_1'(0)$, and the remaining Taylor coefficients are arbitrary. The requirement that the first fundamental form $(\wt{\bar g}_\infty+\wt{\bar h}_1)|_{\wt X}=\wt\gamma_{\infty,0}+\wt{\bar h}_1(0)$ of $\wt g_\infty+\wt h_1$ be equal to $\wt\gamma_1=\wt\gamma_{\infty,0}+\wt h$ forces $\wt{\bar h}_1(0)=\wt h$. For the second fundamental form, we require
  \begin{align*}
    \wt k_1 = \wt k_{\infty,0}+\wt q &= \Bigl(\frac12\wt N^{-1}\cL_{\wt e_0}(\wt{\bar g}_\infty+\wt{\bar h}_1(0)+t\wt{\bar h}{}_1'(0))\Bigr)\Big|_{\wt X} \\
      &= \wt k_{\infty,0} + \frac12\wt N^{-1}\bigl( \wt{\bar h}{}'_1(0) + \cL_{\wt e_0}\wt{\bar h}_1(0)\bigr),
  \end{align*}
  with $\wt{\bar h}_1(0)$ on the right regarded as a $t$-independent section of $S^2\wt T^*\wt X$ over $\wt O$. This determines $\wt{\bar h}{}'_1(0)\in\CIdot(\wt X\setminus\wt K^\circ;S^2\wt T^*\wt X)$ indeed; this uses that the Lie derivative along $\wt e_0$ preserves the rapid vanishing of $\wt{\bar h}_1(0)$ at $\eps=0$ in view of~\eqref{EqIVPwte0}.
\end{proof}

We now improve the tensor $\wt g_1$ from Proposition~\ref{PropIVPConst} further.

\begin{prop}[Taylor series construction at $\wt X$]
\label{PropIVPTay}
  Let $\wt g_1$ be as in Proposition~\usref{PropIVPConst}; thus $\Ric(\wt g_1)-\Lambda\wt g_1\in\CIdot(\wt M\setminus\wt K^\circ;S^2\wt T^*\wt M)$ and the constraint equations are satisfied at $\wt X$. Then there exists $\wt h\in\CIdot(\wt M\setminus\wt K^\circ;S^2\wt T^*\wt M)$ so that for $\wt g:=\wt g_1+\wt h$, the error $\Ric(\wt g)-\Lambda\wt g\in\CIdot(\wt M\setminus\wt K^\circ;S^2\wt T^*\wt M)$ vanishes to infinite order at $\wt X$.
\end{prop}
\begin{proof}
  In view of Lemma~\ref{LemmaIVPn1Taylor}, applied to each level set of $\eps>0$, it suffices to construct $\wt h=\wt B(\wt{\bar h})$ (using~\eqref{EqIVPwtB}) where $\wt{\bar h}\in\CIdot(\wt O\setminus\wt K^\circ;S^2\wt T^*\wt X)$ so that the pullback of
  \begin{equation}
  \label{EqIVPTay}
    \Ric(\wt g_1+\wt h)-\Lambda(\wt g_1+\wt h)
  \end{equation}
  to $\wt X_t$ (i.e.\ the spatial components) vanishes to infinite order at $t=0$. This form of $\wt h$ ensures that lapse and shift of $\wt g_1+\wt h$ are equal to those of $\wt g_\infty$. Writing the Taylor series of $\wt{\bar h}(t)$ as $\sum_{j\geq 2}\wt{\bar h}_j t^j$ at $t=0$, we can then iteratively compute the coefficients $\wt{\bar h}_j$, $j=2,3,\ldots$, using the formula $\frac12\wt N^{-1}\cL_{\wt e_0}(\wt g_1+\wt h)$ for the second fundamental form and using the formula for the spatial components of~\eqref{EqIVPTay} from Corollary~\ref{CorIVPn1Curv}, by repeating the arguments at the end of~\S\ref{SssIVPwt}.
\end{proof}

This finishes the proof of Theorem~\ref{ThmIVP}, and in combination with Theorem~\ref{ThmFh} also the proof of the main result of this paper, Theorem~\ref{ThmM} (which is Theorem~\ref{ThmI} in the setting~\eqref{ItIOGeneric}).

\section{Extreme mass ratio mergers}
\label{SX}

While Theorem~\ref{ThmM} allows one to glue a subextremal Kerr black hole into a given spacetime $(M,g)$ under a genericity assumption (see Definition~\ref{DefMData}\eqref{ItMDataKID}), this assumption is not satisfied for explicit spacetimes of physical interest such as Kerr or Kerr--(anti) de~Sitter (K(A)dS) spacetimes. The situation for initial data gluing was discussed in \cite[\S6.2]{HintzGlueID}; here we describe the analogue for our formal spacetime gluing procedure (and thereby prove Theorem~\ref{ThmI} in the setting~\eqref{ItIOKdS}).

Recall that for $\Lambda\in\R$ and parameters $\bhm>0$, $a\in\R$, the Kerr, Kerr--de~Sitter, or Kerr--anti de~Sitter (depending on whether $\Lambda=0$, $\Lambda>0$, or $\Lambda<0$) metric $g_{\Lambda,\bhm,a}$ is given by
\begin{alignat*}{2}
  &g_{\Lambda,\bhm,a} = -\frac{\mu(r)}{b^2\varrho^2(r,\theta)}(\dd t-a\,\sin^2\theta\,\dd\phi)^2 + \varrho^2(r,\theta)\Bigl(\frac{\dd r^2}{\mu(r)}+\frac{\dd\theta^2}{c(\theta)}\Bigr) \hspace{-15em}&& \\
  &\hspace{14em} + \frac{c(\theta)\sin^2\theta}{b^2\varrho^2(r,\theta)}\bigl((r^2+a^2)\dd\phi-a\,\dd t)^2, \hspace{-15em}&& \\
  &\quad \mu(r)=(r^2+a^2)\Bigl(1-\frac{\Lambda r^2}{3}\Bigr)-2\bhm r, &\qquad \varrho^2(r,\theta)&=r^2+a^2\cos^2\theta, \\
  &\quad b=1+\frac{\Lambda a^2}{3},&\qquad c(\theta)&=1+\frac{\Lambda a^2}{3}\cos^2\theta.
\end{alignat*}
This solves the Einstein vacuum equations
\[
  \Ric(g_{\Lambda,\bhm,a}) - \Lambda g_{\Lambda,\bhm,a} = 0.
\]
We require the parameters $(\Lambda,\bhm,a)$ to be \emph{subextremal}; this means that $\mu(r)=0$ has four (when $\Lambda\neq 0$), resp.\ two (when $\Lambda=0$) distinct real roots. The second largest (when $\Lambda\neq 0$), resp.\ largest (when $\Lambda=0$) root $r=r_{\Lambda,\bhm,a}$ is the radius of the event horizon. For $\Lambda=0$, we recover the Kerr spacetime from Definition~\ref{DefGK} with different symbols for the coordinates.

Using a change of coordinates $\ft=t-T(r)$ and $\varphi=\phi-\Phi(r)$ for suitable functions $T,\Phi$ (see \cite[Equation~(1.5)]{HintzKdSMS} for the case $\Lambda>0$), the metric $g_{\Lambda,\bhm,a}$ extends analytically across the two largest roots $r_{\Lambda,\bhm,a}<r_{\Lambda,\bhm,a}^c$ (the event and cosmological horizons) when $\Lambda>0$, resp.\ the largest root $r_{\Lambda,\bhm,a}$ (the event horizon) when $\Lambda=0$, and indeed to the manifold
\[
  M = M_{2\eta},\qquad
  M_\delta := \R_\ft \times [r_{\Lambda,\bhm,a}-\delta,\infty) \times \Sph^2_{\theta,\varphi}
\]
for sufficiently small $\eta>0$; here $\delta\in(0,2\eta]$. We can select $T$ so that moreover the level sets of $\ft$ are spacelike; set then
\[
  X = X_{2\eta},\qquad X_\delta := \{0\} \times [r_{\Lambda,\bhm,a}-\delta,\infty) \times \Sph^2_{\theta,\varphi}.
\]
The \emph{black hole exterior} (or \emph{domain of outer communications}) is the subset $M_{\rm ext}$ of $M$ where $r\in(r_{\Lambda,\bhm,a},r_{\Lambda,\bhm,a}^c)$ when $\Lambda>0$, resp.\ $r>r_{\Lambda,\bhm,a}$ when $\Lambda=0$; similarly, we set
\[
  X_{\rm ext}=
    \begin{cases}
      \{0\}\times(r_{\Lambda,\bhm,a},r_{\Lambda,\bhm,a}^c)\times\Sph^2, & \Lambda>0, \\
      \{0\}\times(r_{0,\bhm,a},\infty)\times\Sph^2, & \Lambda\leq 0.
    \end{cases}
\]
The part of the \emph{black hole interior} contained in $M$ is the region where $r<r_{\Lambda,\bhm,a}$. In the case $\Lambda<0$, we replace $M$ by the domain of dependence of $\{\ft=0\}$ in order to avoid having to impose boundary conditions at the conformal boundary.

\begin{thm}[Gluing a small black hole along a timelike geodesic in a subextremal Kerr(--de~Sitter) spacetime]
\label{ThmXGlue}
  Let $\Lambda\in\R$, $\bhm>0$, $a\in\R$ be subextremal K((A)dS) parameters, and let $g=g_{\Lambda,\bhm,a}$. Let $\fp\in X_{\rm ext}$, let $v\in T_\fp M$ be a future timelike unit vector, and write $\cC\subset M$ for the maximally extended geodesic with initial conditions $\fp,v$, and let $\cU_M^\circ\subset M$ be a smoothly bounded precompact connected open neighborhood of $\fp$ so that $\cU_M^\circ\cap X$ contains a point in $X_{2\eta}\setminus X_{\frac32\eta}$. Let\footnote{We use hats here to notationally distinguish the small black hole parameters from those of $(M,g)$.} $\hat\bhm>0$ and $\hat\bha\in T_\fp M$, $\hat\bha\perp v$, with $|\hat\bha|<\hat\bhm$, be subextremal Kerr parameters. Then the conclusions of Theorem~\usref{ThmM} hold on $\wt M_\eta\setminus\wt K^\circ$ (with $\hat\bhm,\hat\bha$ in place of $\bhm,\bha$), with $\wt K$ defined as in point~\eqref{ItMFamily} in~\textnormal{\S}\usref{SM}. That is, there exists a $((3,*),(1,0)_+\cup(3,*))$-smooth total family $\wt g$ on $\wt M_\eta\setminus\wt K^\circ$ with respect to $g$ and with $\hat M_p$-model equal to the Kerr metric $\hat g_{\hat\bhm,\hat\bha}$ for all $p\in\cC$ (as in point~\eqref{ItMFamily} in~\textnormal{\S}\usref{SM}) so that:
  \begin{enumerate}
  \item $\Ric(\wt g)-\Lambda\wt g$ is a smooth section of $S^2\wt T^*\wt M$ over $\wt M_\eta\setminus\wt K^\circ$ which vanishes to infinite order at $\hat M$, $M_\circ$, and $\wt X$;
  \item $\wt g=g$ outside the union of the causal past and future of a compact subset of $\cU_M^\circ\cap X$;
  \item $\sfe\wt g$ is equal to the Kerr metric at $\hat M$ modulo quadratically vanishing error terms.
  \end{enumerate}
\end{thm}

With only minor notational modifications (concerning the definitions of $X,M$, specifically defining their extensions across the event horizon up until a spacelike boundary hypersurface), one can also treat the case of \emph{extremal} black holes.

Theorem~\ref{ThmXGlue} can be applied to produce formal solutions (as $\eps\searrow 0$) of the Einstein vacuum equations describing \emph{extreme mass ratio mergers}. See Figure~\ref{FigXGlue}. A fortiori, by definition of $\wt M_\eta$, the restriction of $\wt g$ to the $\eps$-level set of $\wt M_\eta$ is $\eps$-close to the K((A)dS) metric $g$ outside any fixed neighborhood of $\cC$, whereas an $\eps^{-1}$-rescaling of $\wt g$ near a point on $\cC$ is $\eps$-close to any fixed compact subset of a Kerr spacetime with parameters $\hat\bhm,\hat\bha$.

\begin{figure}[!ht]
\centering
\includegraphics{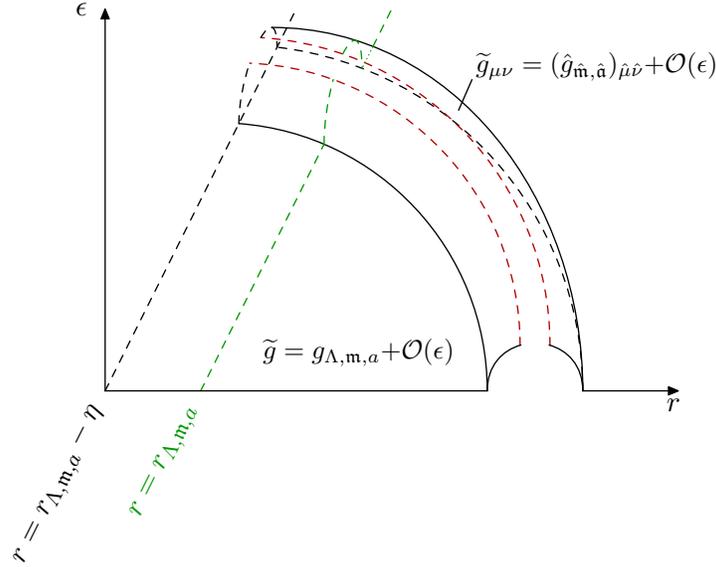}
\caption{Illustration of Theorem~\ref{ThmXGlue}. The front face arises by blowing up a timelike geodesic $\cC$---which here passes the event horizon of the ambient K((A)dS) black hole, indicated in green---inside $[0,1)_\eps\times M_\eta$ at $\eps=0$. The restriction of $\wt g$ to an $\eps$-level set is $\eps$-close to the K((A)dS) metric $g_{\Lambda,\bhm,a}$ away from the front face, whereas the metric coefficients (in Fermi normal coordinates $t,x$) near the front face, arising by blowing up the timelike geodesic $\cC$ in $[0,1)_\eps\times M_\eta$ at $\eps=0$, are $\eps$-close to those of a Kerr metric (with respect to the `fast' coordinates $\hat t=\frac{t-t_0}{\eps},\hat x=\frac{x}{\eps}$ near any point $(t_0,0)$ along $\cC$). The red dashed lines form a spacelike hypersurface bounding a coordinate sphere in the interior of the small black hole which we excise, much like we excise the ball $r<r_{\Lambda,\bhm,a}-\eta$ in the interior of the ambient K((A)dS) black hole.}
\label{FigXGlue}
\end{figure}

\begin{rmk}[Formal extreme mass ratio mergers: finite time theory]
\label{RmkXDomain}
  Consider the case $\Lambda\geq 0$. Suppose the geodesic $c(s)$ with $c(0)=\fp$, $c'(0)=v$ crosses the event horizon of $(M,g)$ at proper time $s=s_0>0$. Then since $\dd r$ is future timelike in $r<r_{\Lambda,\bhm,a}$, the function $r\circ c(s)$ is monotonically decreasing for $s>s_0$, and thus $r(c(s))<r-\frac32\eta$ for all $s>s_1>s_0$. Fixing $s_2>s_1$ and letting $\ft_2:=t(c(s_2))$, the total family $\wt g$ is then $\eps$-close to the ambient K(dS) metric $g$ on compact subsets of $M_\eta\cap\{\ft\geq\ft_2\}$. In the Kerr--de~Sitter case ($\Lambda>0$), this in particular applies to $[\ft_2,\ft_2+1]\times[r_{\Lambda,\bhm,a}-\eta,r_{\Lambda,\bhm,a}^c+\eta]\times\Sph^2$; in the Kerr case, $\wt g$ is \emph{equal} to $g$ on $[\ft_2,\ft_2+1]\times[r_2,\infty)\times\Sph^2$ for some $r_2<\infty$ (depending on $\ft_2$), and $\eps$-close to $g$ on $[\ft_2,\ft_2+1]\times[r_{\bhm,a}-\eta,r_2]\times\Sph^2$. This can be taken as the starting point for the evolution of initial data $\eps$-close to K(dS) data (but violating the constraint equations by an amount $\cO(\eps^\infty)$) in an $\eta$-neighborhood of the domain of outer communications. We stress that $\eta>0$ is fixed, and $\eps>0$ can be taken to be arbitrarily small, and therefore domain of dependence considerations imply that the initial data at $\ft_2$ are defined on a sufficiently large set so as to permit, in principle, the unique global future solvability of (quasi)linear wave equations.
\end{rmk}

\begin{rmk}[Formal extreme mass ratio mergers: nonlinear stability of the merger]
\label{RmkXStab}
  Continuing Remark~\ref{RmkXDomain} and focusing on the very slowly rotating Kerr--de~Sitter case $\Lambda>0$, $|\frac{a}{\bhm}|\ll 1$, note that, a fortiori, the initial data of $\wt g_\eps=\wt g|_{(\wt M_\eta)_\eps}$ at $\ft=\ft_2$ satisfy the constraint equations up to $\cO(\eps^\infty)$ error terms. Therefore, we can apply \cite[Theorem~11.2]{HintzVasyKdSStability}, with $(h,k)$ there equal to the initial data of $\wt g_\eps$ at $\ft=\ft_2$, to obtain a solution $\wt g'_\eps$ of the gauge-fixed Einstein vacuum equation \cite[(11.10)]{HintzVasyKdSStability} on
  \[
    \Omega:=[\ft_2,\infty)\times[r_{\Lambda,\bhm,a}-\eta,r_{\Lambda,\bhm,a}^c+\eta]\times\Sph^2
  \]
  with the following properties:
  \begin{enumerate}
  \item $\wt g'_\eps$ is equal to a Kerr--de~Sitter metric plus a tail whose coefficients in standard coordinates on $\Omega$ are bounded, together with all their coordinate derivatives, by $e^{-\alpha\ft}$ for some $\alpha>0$;
  \item $\Err_\eps:=\Ric(\wt g'_\eps)-\Lambda\wt g'_\eps$ (in the sign convention of the present paper) obeys the same $\cO(e^{-\alpha\ft})$ bound with constants of size $\cO(\eps^N)$ for all $N$, i.e.\ the components of $\Err_\eps$ are bounded by $C_N\eps^N e^{-\alpha\ft}$ together with all coordinate derivatives.
  \end{enumerate}
  The second property is a consequence of the fact that the gauge 1-form, denoted $\Ups(g)-\Ups(g_{b_0,b})-\theta$ in \cite{HintzVasyKdSStability}, has Cauchy data at $\ft=\ft_2$ of size $\cO(\eps^\infty)$; since it moreover lies in $\ker\wt\Box^{\rm CP}_g$, it can be seen to decay exponentially using the extension of \cite[Theorem~8.1]{HintzVasyKdSStability}---which in particular proves the absence of non-exponentially decaying mode solutions when $g$ is a subextremal Schwarzschild--de~Sitter metric---to metrics which exponentially decay to very slowly rotating Kerr--de~Sitter spacetimes; this extension follows from the methods of \cite[\S5]{HintzVasyKdSStability}. Thus, we can control a formal (i.e.\ up $\cO(\eps^\infty)$ errors) solution of the initial value problem for the Einstein vacuum equations with initial data given at $\ft=0$ \emph{globally} in forward time: Theorem~\ref{ThmXGlue} provides the part of the solution for $\ft\in[0,\ft_2]$, and \cite{HintzVasyKdSStability} provides the rest. See Figure~\ref{FigXStab}.
\end{rmk}

\begin{figure}[!ht]
\centering
\includegraphics{FigXStab}
\caption{Illustration of Remark~\ref{RmkXStab} for very small but fixed $\eps>0$: the metric $\wt g_\eps=\wt g|_{(\wt M_\eta)_\eta}$ is a formal solution of the Einstein vacuum equations and roughly speaking describes a small Kerr black hole with parameters $(\eps\hat\bhm,\eps\hat\bha)$ near $\cC$ which merges with a given Kerr--de~Sitter black hole $(M,g)$ with parameters $(\bhm,a)$. Having crossed the event horizon of the KdS black hole at time $\ft=\ft_2$, the metric $\wt g_\eps$ is $\eps$-close to $g$ in a fixed neighborhood of the domain of outer communications (the exterior of the gray cylinder, with the cosmological horizon not shown). Nonlinear stability results for the gauge-fixed Einstein vacuum equations, applied with parametric dependence in $\eps$, give a formal solution $\wt g'_\eps$ for all subsequent times, which decays exponentially fast to an exact Kerr--de~Sitter metric.}
\label{FigXStab}
\end{figure}

\begin{proof}[Proof of Theorem~\usref{ThmXGlue}]
  We indicate the modifications required in (the proofs of) Theorems~\ref{ThmFh} and \ref{ThmIVP}. Fix $\eta_0=2\eta$ and pick $\eta_j\in\N$ with $\eta_0>\eta_1>\eta_2>\ldots$, and $\ubar\eta:=\inf_j\eta_j>\eta$.

  The only place in the proof of Theorem~\ref{ThmFh} where the absence of Killing vector fields (or KIDs on the level of initial data sets) is used is the proof of Proposition~\ref{PropAcTrue} (which leads to Theorem~\ref{ThmAc}), specifically in the construction of the symmetric 2-tensors $\dot\gamma,\dot k$ in equation~\eqref{EqAcTrueID}. In the present setting, we cannot solve~\eqref{EqAcTrueID} directly due to the presence of a cokernel $K^*:=\ker_{\CI(\cU^\circ;\ul\R\oplus T^*X_{\eta_0})}D_{(\gamma,k)}P^*$ of dimension $N=\dim K^*\geq 1$; here $\cU^\circ=\cU_M^\circ\cap X$. Instead, as in \cite[\S6.2]{HintzGlueID}, we may fix an injective linear map
  \[
    E_1 \colon \R^N \to \CIc(\cU^\circ\setminus X_{\eta_1};\ul\R\oplus T^*X_{\eta_0})
  \]
  so that the $L^2$-pairing $K^*\times\ran E_1\to\R$ is nondegenerate; and then there exists a unique solution $c\in\R^N$, $\dot\gamma,\dot k\in\CIc(\cU^\circ;S^2 T^*X)$ of the equation
  \[
    D_{(\gamma,k)}P(\dot\gamma,\dot k;\Lambda) = (\sfG_g f_\flat)(\nu,-) + E_1(c).
  \]
  Proceeding as after~\eqref{EqAcTrueID}, the gauge 1-form $\eta=\delta_g\sfG_g h_\flat-\theta$ still vanishes at $X$; but now $(\sfG_g\delta_g^*\eta)(\nu,-)$ vanishes only outside of $\supp E_1(c)\subset X_{\eta_0}\setminus X_{\eta_1}$, and therefore $\eta\in\ker\delta_g\sfG_g\delta_g^*$ does not necessarily vanish. However, since $\dd r$ is future timelike in $M_{\rm int}$, we do have $\supp\eta\subset M_{\eta_0}\setminus M_{\eta_1}$ by finite speed of propagation. In conclusion, Theorem~\ref{ThmAc} remains valid if we replace $M$ by $M_{\eta_0}$ and allow for $(D_g\Ric-\Lambda)h-f$ in equation~\eqref{EqAc} to be nonzero outside $M_{\eta_1}$.

  After the first application of Theorem~\ref{ThmFh} in~\S\ref{SsFMc1}, we restrict to $M_{\eta_1}$. The subsequent solution step at $\hat M$ is unchanged. When applying Theorem~\ref{ThmFh} for the first of two times again in~\S\ref{SsFMc2}, we first work on $M_{\eta_1}$ and permit violations of the constraints, and thus of the linearized Einstein vacuum equations, outside $M_{\eta_2}$; and the second time we work on $M_{\eta_2}$ and permit violations outside $M_{\eta_3}$; and so on. In conclusion then, the analogue of Theorem~\ref{ThmFh} in the present setting remains valid on $\wt M_{\ubar\eta}$.

  In the proof of Theorem~\ref{ThmIVP}, the absence of KIDs is used in the (adapted) proof of \cite[Proposition~5.6]{HintzGlueID} in \cite[Propositions~4.15]{HintzGlueID} (via \cite[Proposition~4.21]{HintzGlueID}). In the present setting, where KIDs are present, one proceeds as in the proof of \cite[Theorem~6.2]{HintzGlueID} to solve the constraint equations up to an error which lies in a fixed finite-dimensional space of smooth tensors supported in $X_{\ubar\eta}\setminus X_\eta$. The remainder of the proof of Theorem~\ref{ThmIVP} is unchanged.
\end{proof}

\begin{rmk}[General spacetimes with noncompact Cauchy hypersurfaces: Theorem~\usref{ThmI2}]
\label{RmkXNc}
  We again indicate the modifications required in the proofs of Theorems~\ref{ThmFh} and \ref{ThmIVP}. First, in the construction of initial data $\dot\gamma,\dot k$ for equation~\eqref{EqAcTrueID} in the proof of Proposition~\ref{PropAcTrue}, we are now dropping the requirement that $\dot\gamma,\dot k\in\CI(X;S^2 T^*X)$ have compact support. The cokernel of $D_{(\gamma,k)}P$ on the corresponding dual space of compactly supported distributions (which thus vanish on some open set due to the noncompactness of $X$) is then trivial; see also \cite{HintzLinEin}. Therefore, we can indeed solve~\eqref{EqAcTrueID} for $\dot\gamma,\dot k\in\CI(X;S^2 T^*X)$. The proof of Theorem~\ref{ThmFh} then goes through without any further modifications. To deal with the possibility of KIDs in the proof of Theorem~\ref{ThmIVP}, fix a nonempty open set $W\subset X\setminus\bar V$; similarly to before, we can then solve the constraint equations up to an error which lies in a fixed finite-dimensional space of smooth tensors supported in $W$.
\end{rmk}

\bibliographystyle{alphaurl}
\bibliography{
/scratch/users/hintzp/ownCloud/research/bib/math,
/scratch/users/hintzp/ownCloud/research/bib/mathcheck,
/scratch/users/hintzp/ownCloud/research/bib/phys
}

\end{document}